\documentclass[12pt]{article}
\pdfoutput=1
\usepackage{amsfonts}
\usepackage{amsmath}
\usepackage{amssymb}
\usepackage{amsthm}
\usepackage{adjustbox}
\usepackage{graphicx}
\usepackage{natbib}
\usepackage{amsthm}
\usepackage{authblk}
\usepackage{array}

\usepackage{multirow}

\newtheorem{theorem}{Theorem}
\newtheorem{lemma}{Lemma}

\newtheorem{remark}{Remark}

\newcommand{\x}{\mathbf{x}}
\newcommand{\w}{\mathbf{w}}
\newcommand{\X}{\mathbf{z}}

\newcommand{\vv}{v}
\newcommand{\sign}{\operatorname{sign}}
\newcommand{\Beta}{\boldsymbol{\beta_1}}
\newcommand{\E}{\operatorname{\mathbb{E}}}
\newcommand{\var}{\operatorname{\mathbb{V}ar}}
\newcommand{\mginv}{mg^{-1}}
\newcommand{\penn}[3]{P_{#1, #3}\left( #2 \right)}

\begin{document}

\title{Robust Penalized Estimators for High--Dimensional Generalized Linear Models}

\author[1]{Marina Valdora}
\author[2]{Claudio Agostinelli}
\affil[1]{Instituto de C\'alculo and Department of Mathematics, University of Buenos Aires, Buenos Aires, Argentina \texttt{mvaldora@dm.uba.ar}}
\affil[2]{Department of Mathematics, University of Trento, Trento, Italy \texttt{claudio.agostinelli@unitn.it}}

\date{\today}

\maketitle

\begin{abstract}
  Robust estimators for generalized linear models (GLMs) are not easy to develop due to the nature of the distributions involved. Recently, there has been growing interest in robust estimation methods, particularly in contexts involving a potentially large number of explanatory variables. Transformed M-estimators (MT-estimators) provide a natural extension of M-estimation techniques to the GLM framework, offering robust methodologies. We propose a penalized variant of MT-estimators to address high-dimensional data scenarios. Under suitable assumptions, we demonstrate the consistency and asymptotic normality of this novel class of estimators. Our theoretical development focuses on redescending $\rho$-functions and penalization functions that satisfy specific regularity conditions. We present an Iterative Re-Weighted Least Squares algorithm, together with a deterministic initialization procedure, which is crucial since the estimating equations may have multiple solutions. We evaluate the finite sample performance of this method for Poisson distribution and well known penalization functions through Monte Carlo simulations that consider various types of contamination, as well as an empirical application using a real dataset.
  
  
\end{abstract}

\section{Introduction}
\label{sec:introduction}
GLMs are widely used tools in data analysis because of their ability to handle a variety of response distributions and link functions. However, traditional estimation methods for GLMs often fail in high-dimensional settings, where the number of explanatory variables exceeds the number of observations. The problem of high-dimensional data has been widely studied and penalized procedures have been proposed to address it. For linear models, notable contributions include the ridge regression proposed by \citet{hoerl1970Ridge}, the least absolute shrinkage and selection operator (lasso) by \citet{tibshirani1996regression}, and the elastic net by \citet{zou2005regularization}. These methods were extended to GLMs by \citet{friedman2010regularization}. Additional penalization techniques include the bridge estimator introduced by \citet{frank1993}, the smoothly clipped absolute deviation (SCAD) penalty by \citet{fan2001variable}, and the minimax concave penalty (MCP) by \citet{zhang2010}. More recently, \citet{bianco2022penalized} introduced the sign penalty for logistic regression. The theoretical aspects and performance of penalized estimators in the context of linear regression have been further explored by \citet{knight2000asymptotics} and \citet{zou2006adaptive}.

While most of these penalized methods perform well under the assumption that all data points follow the model, their performance is poor in the presence of outliers. A small proportion of atypical data can significantly affect the reliability of the estimators. In response to this limitation, robust estimators for high-dimensional linear models have been developed by \citet{maronna2011robust} and \citet{smucler2017robust}. \citet{avella2018robust} extended the ideas in \cite{cantoni2001robust} and proposed robust penalized estimation for GLMs,  aiming at effective variable selection.  However, their theoretical assumptions are too strong for our purposes. In fact, their Condition 1 excludes our proposed penalized estimators in the case of covariates with unbounded support, see Appendix \ref{sec:computational}, equation \eqref{MTA1}. In the context of logistic regression, \citet{bianco2022penalized} proposed penalized robust estimators. \citet{salamwade2021robust} further contributed to robust variable selection in GLMs by introducing penalized MT-estimators and proving their oracle properties. Their approach is limited to the SCAD penalty, does not cover robust and computational aspects such as the development of a specific algorithm or how to obtain robust initial values which are very important since the estimating equations might have more then one solution. Additionally, their numerical study does not address high-dimensional scenarios, which require specialized computational techniques, as highlighted by \citet{agostinelli2019initial}.

MT-estimators were originally introduced by \citet{ValdoraYohai2014} for classical dimensional settings and further studied in \citet{agostinelli2019initial}, where scenarios with a large number of covariates were explored. These studies have shown that MT-estimators outperform other robust methods in terms of efficiency and robustness in several simulation settings and two real data examples, specifically focusing on Poisson regression. \citet{cantoni2001robust} proposed robust tests based on robust quasi-likelihood (RQL) estimators. The advantage of MT-estimators over RQL-estimators is that they not only have the property of infinitesimal robustness but have a high breakdown point as well. In fact, the simulations in \cite{ValdoraYohai2014} and \cite{agostinelli2019initial} show the resistance of MT-estimators to contaminations of up to $15\%$. 
%%%On the other hand, RQL-estimators with a Huber-type $\rho$ function have a very low breakdown point, as suggested by the simulation study and the fact that the loss function is unbounded.}
The breakdown point is formally defined in Section \ref{sec:robustness}. Informally, the breakdown point represents the maximum proportion of outliers that the estimator can tolerate before its values become unbounded. It is clear that a high breakdown point does not ensure the robustness or the good performance of an estimator, especially in the case of non-equivariant methods; see \citet{daviesgather2005}. However, it is still a desirable property. % and estimators with a very low breakdown point are clearly not robust.  
%\cite{cantoni2001robust} developed the theory for robust inference in GLMs based on RQL-estimators with a general loss function and carried out their simulations using the Huber function. %, other loss functions could be used, including redescending functions. %had good results using the code provided in the supplement to \citet{avella2018robust} because of very frequent convergence problems and large computational times. 
Another advantage of MT-estimators is that, since they are defined using a bounded loss function, they are robust even in the absence of weights on the covariates. Weighted MT-estimators (WMT) were investigated in \cite{ValdoraYohai2014} and weights on the covariates have been used succesfully in many cases; see, for instance \cite{pensia2024robust}. In some cases, however, the use of weights may cause the estimator to lose robustness. This happens when there are ``good'' high leverage observations, meaning that they have a response that follows the nominal model. Down-weighting these observations may, not only cause a small loss of efficiency, but a loss of robustness as well, since penalizing high leverage non-outlying observations may increase the influence of low-leverage outliers. In fact, the use of hard-rejection weights can reduce the breakdown point since good high leverage point will have weights equal to zero. See the simulations reported in Section 6 in \cite{ValdoraYohai2014}.
Despite the above comments, adding weights on the covariates is still a useful technique that may increase robustness when high leverage outliers are present. The investigation of the theoretical properties of penalized  WMT-estimators is a challenging and interesting problem that may be the subject of further work.

%RQL-estimators need weights on the covariates in order to achieve robustness in the presence of high leverage outliers. 
In this paper we adapt MT-estimators to the high-dimensional setting. We introduce penalized MT-estimators that utilize redescending $\rho$-functions and general penalty functions. We investigate their theoretical properties, including consistency, robustness, asymptotic normality, and variable selection under suitable assumptions. Our theoretical approach follows closely that of \citet{bianco2022penalized} for logistic regression. This approach allows us to consider random penalties instead of only deterministic ones, providing a more realistic setting than those considered in \citet{avella2018robust} and \citet{salamwade2021robust} since penalties are usually selected by data-driven procedures.
Finally, we present a computational algorithm and a procedure for obtaining initial estimates. The performance of the proposed method is evaluated through extensive Monte Carlo simulations, which consider various types of contamination, and an empirical application to predict the length of hospital stay of patients using multiple covariates.

The rest of the paper is organized as follows. Section \ref{sec:MT} reviews the definition of MT-estimators for GLMs, while Section \ref{sec:penalized} introduces their penalized version. Section \ref{sec:penaltyselection} discusses the selection of penalty parameters by information criteria and cross-validation, while Section \ref{sec:asres} provides asymptotic and robustness results. Sections \ref{sec:montecarlo} and \ref{sec:rhc} present the Monte Carlo setting, summarize its results, and provide an application example. Section \ref{sec:conclusion} offers concluding remarks. Appendix \ref{sec:proofs} contains the complete proofs of the theorems given in Section \ref{sec:asres}, whereas Appendix \ref{sec:computational} provides details on obtaining a robust starting value and describes the iteratively reweighted least squares (IRWLS) algorithm. Finally, the Supplementary Material reports complete results of the Monte Carlo experiment.

\section{M-estimators based on transformations} \label{sec:MT}
Let $(y, \x)$ be a random vector of dimension $p+2$ and $F = F(\cdot,\mu, \phi)$ a distribution function depending on two parameters $\mu$ and $\phi$, with $\mu \in (\mu_1, \mu_2)$ where $\mu_i$ can be $\pm \infty$ and $\phi>0$. We say that $(y, \x)$ follows a GLM with parameter $\boldsymbol{\beta}$, link function $g$ and distribution function $F$ if 
\begin{equation}\label{eq:glm}
y \vert \x \sim F(\cdot,\mu, \phi), \qquad g(\mu) = \eta = \x^\top \boldsymbol{\beta}, 
\end{equation}
where $\boldsymbol{\beta} = (\beta_0, \Beta^\top)^\top$ is the parameter of the linear predictor $\eta$, $\beta_0 \in \mathbb{R}$, $\Beta = (\beta_1,\ldots, \beta_{p})^\top \in \mathbb{R}^p$, $\x = (1, \X^\top)^\top$, $\X=(x_1, \ldots, x_p)^\top \in \mathbb{R}^p$ and $\phi$ is a nuisance scale parameter which is assumed to be known. We assume that $g$ is defined on an open interval $(\mu_1, \mu_2)$, and that it is continuous and strictly increasing. We also assume that $\mathop{\lim_{\mu\rightarrow \mu_1}} g(\mu) = - \infty$ and $\mathop{\lim_{\mu\rightarrow \mu_2}} g(\mu) = + \infty$.

By $\E_{\boldsymbol{\beta}_*}(h(y,\x))$ we denote the expectation of $h(y,\x)$, when $(y, \x)$ follows a GLM with parameter $\boldsymbol{\beta}_*$, as defined in \eqref{eq:glm}. This notation is used in the definition of the objective function $L$, defined below, in order to emphasize the different roles of $\boldsymbol{\beta}$ and $\boldsymbol{\beta}_*$. In the sequel, we will denote it either by $\E_{\boldsymbol{\beta}_*}(h(y,\x))$ or simply by $\E(h(y,\x))$ when the distribution of $(y,\x)$ is clear.

Let $t$ be a variance stabilizing transformation, that is, such that $\var(t(y)) \approx a(\phi)$ is approximately constant and $a$ is a known function. We define
\begin{equation}\label{eq:L}
L(\boldsymbol{\beta}) = \E_{\boldsymbol{\beta}_*}\left(\rho\left(\frac{t(y)- m ( g^{-1}(\x^\top\boldsymbol{\beta}))}{\sqrt{a(\phi)}}\right)\right),
\end{equation}
where the $m$ function is given by
\begin{equation}\label{eq:m}
m(\mu) = \arg\min_\gamma \E_\mu\left( \rho\left(\frac{t(y)- \gamma}{\sqrt{a(\phi)}}\right)\right) ,
\end{equation}
$\rho$ is a bounded $\rho$-function, see Section \ref{sec:asres} for the definition, and $t$ is a variance stabilizing transformation; for instance, in the Poisson case $t(y) = \sqrt{y}$. The function $m$ is necessary to ensure the Fisher consistency of the estimator. We assume that this function is uniquely defined for all $\mu$. We denote by $\mginv$ the composition $m\circ g^{-1}$. 

Introduced by \citet{ValdoraYohai2014}, MT-estimators are the finite sample version of (\ref{eq:L}). Given an independent and identically distributed (i.i.d.) sample $(y_1, \x_1), \ldots, (y_n, \x_n)$ of size $n$, we define
\begin{equation}
 L_n(\boldsymbol{\beta}) = \frac{1}{n}\sum_{i=1}^n \rho \left( \frac{t(y_i) - \mginv \left(\x_{i}^\top \boldsymbol{\beta}\right)}{\sqrt{a(\phi)}}\right)
\label{eq:Ln}
\end{equation}
and the MT-estimator is $\tilde{\boldsymbol{\beta}}_n = \arg\min_{\boldsymbol{\beta}} L_n(\boldsymbol{\beta})$.
  
Taking $\rho(x)=x^2$ we get a least squares estimate based on transformations (LST) for GLMs:
\begin{equation}
\tilde{\boldsymbol{\beta}}_{LS} = \arg\min_{\boldsymbol{\beta}}\sum_{i=1}^{n} \left(t(y_{i})-m_{LS}\left(g^{-1}( \x_{i}^\top\boldsymbol{\beta})\right)\right)^2,
\label{eq:LSTest}
\end{equation}
where $m_{LS}(\mu) = \mathbb{E}_\mu (t(y))$. Hence, MT-estimators can be seen as a robustification of LST-estimators. LST- and MT-estimators are investigated in \citet{ValdoraYohai2014} and \citet{agostinelli2019initial}, where an initial estimator based on the idea of \citet{PenaYohai1999} is proposed and an IRWLS algorithm to solve the minimization problem is studied.
\section{Penalized  MT-estimators for Generalized Linear Models}
\label{sec:penalized}
Consider the following penalized objective function
\begin{align}\label{eq:lntilde}
\hat{L}_n	(\boldsymbol{\beta}) & =  L_n(\boldsymbol{\beta}) + \penn{\lambda_n}{\boldsymbol{\beta}}{\alpha_n} \\
& = \frac{1}{n}\sum_{i=1}^n \rho \left( \frac{t(y_i) - \mginv \left(\x_{i}^\top \boldsymbol{\beta}\right)}{\sqrt{a(\phi)}} \right) + \penn{\lambda_n}{\boldsymbol{\beta}}{\alpha_n} , \nonumber
\end{align}
where $\penn{\lambda}{\boldsymbol{\beta}}{\alpha}$ is a penalization term depending on the vector $\boldsymbol{\beta}$, and $\alpha_n$ and $\lambda_n$ are penalization parameters that may depend on the sample. We define penalized MT-estimators as
\begin{equation}
\label{eq:MTpenalized}
\hat{\boldsymbol{\beta}}_n = \hat{\boldsymbol{\beta}}_{n,\lambda,\alpha} = \arg\min_{\boldsymbol{\beta}} \hat{L}_n	(\boldsymbol{\beta}).
\end{equation}
For instance, as in \citet{zou2005regularization}, let 
\begin{align}\label{eq:elasticnet}
\penn{\lambda}{\boldsymbol{\beta}}{\alpha} & = \lambda\left((1-\alpha)\frac{1}{2} \| \Beta \|^2_2 + \alpha \|\Beta\|_1\right) \\
\nonumber & = \lambda \left( \sum_{j=1}^p \frac{1}{2} (1-\alpha) \beta_j^2 + \alpha |\beta_j| \right)
\end{align}
be the elastic net penalty function \citep{zou2005regularization}, which corresponds to the ridge penalty for $\alpha=0$ and to the lasso penalty for $\alpha=1$. Notice that we are not penalizing the intercept $\beta_0$. However, the theory we develop also applies to the case in which the intercept is penalized. Other relevant penalty functions are the bridge penalty \citep{frank1993} where $\penn{\lambda}{\boldsymbol{\beta}}{\alpha} = \lambda \|\boldsymbol{\beta}\|_\alpha^\alpha$ with $\alpha > 0$, the SCAD \citep{fan2001variable} and the MCP \citep{zhang2010} penalties. Recently \citet{bianco2022penalized} studied the sign penalty function, which is known as the $\ell_1/\ell_2$ penalty, for penalized robust estimators in logistic regression. It is defined as
\begin{equation*}
\penn{\lambda}{\boldsymbol{\beta}}{\alpha} = \lambda \frac{\|\boldsymbol{\beta}\|_1}{\|\boldsymbol{\beta}\|_2} \mathbb{I}_{\{\boldsymbol{\beta} \neq \mathbf{0} \}} = \lambda \|s(\boldsymbol{\beta})\|_1 \mathbb{I}_{\{\boldsymbol{\beta} \neq \mathbf{0} \}} ,
\end{equation*}
where $s(\boldsymbol{\beta}) = \boldsymbol{\beta}/\|\boldsymbol{\beta}\|_2$ is the sign function.

For the computation of these estimators we use an IRLWS algorithm which is described in Appendix \ref{sec:computational}. The use of a redescending $\rho$ function, essential to get a high breakdown point robust estimator (see Section \ref{sec:asres}), leads to a non convex optimization problem, in which the loss function may have several local minima. For this reason, the choice of the initial point for the IRWLS algorithm is crucial. In particular, we need to start the iterations at an estimator that is already close to the global minimum in order to avoid convergence to a different local minimum. In low dimensions, solutions to this problem are often obtained by the subsampling procedure e.g., the fast S-estimator for regression of \citet{salibian2006fast}. Penalized robust procedures can use the same approach e.g., in \citet{alfons2013sparse}. However, when the dimension of the problem becomes large, the number of sub-samples one needs to explore in order to find a good initial estimator soon becomes computationally infeasible. For this reason, we propose a deterministic algorithm, reported in Appendix \ref{sec:computational} and inspired by \citet{PenaYohai1999}. A similar algorithm is introduced in \citet{agostinelli2019initial} for the unpenalized MT-estimators, and it is proved to be highly robust and computationally efficient.

\section{Selection of the penalty parameters}
\label{sec:penaltyselection}
Robust selection of the penalty parameters can be performed using information like criteria, such as AIC or BIC, or by Cross-Validation. Information criteria are of the form
\begin{equation*}
IC(\lambda, \alpha) = \text{Goodness of the fit} + \text{Complexity Penalty} \times \text{Degree of Freedom}.
\end{equation*}
For penalized methods the degrees of freedom, which measure the complexity of the model, have been studied extensively; see e.g. \citep{friedman2001elements, zuo2007, TibshiraniTaylor2012} and the references therein. For instance, for the Elastic Net penalty the degrees of freedom are based on an ``equivalent projection'' matrix $\mathbf{H}$ that in our context is given by
\begin{equation}
\mathbf{H}=\sqrt{\mathbf{W}} \mathbf{X}_\mathcal{A}( \mathbf{X}_\mathcal{A}^\top \mathbf{W} \mathbf{X}_\mathcal{A} + \lambda (1-\alpha) \mathbf{I} )^{-1}\mathbf{X}_\mathcal{A}^\top\sqrt{\mathbf{W}}
\end{equation}
where $\mathbf{W}$ is a diagonal matrix of weights given by $\mathbf{w}^2(\x_i^\top \hat{\boldsymbol{\beta}}_n) \mathbf{w}^\ast(\x_i^\top \hat{\boldsymbol{\beta}}_n)$ ($i=1,\ldots,n$) defined in equation \eqref{equ:mtIRWLS} in Appendix \ref{sec:computational} and  $\mathbf{X}_\mathcal{A}$ is the matrix containing only the covariates corresponding to the active set $\mathcal{A}$, which is defined as the set of covariates with estimated coefficients different from $0$. 

Then, the effective degrees of freedom are given by 
\begin{equation*}
\text{df}(\lambda, \alpha) = \operatorname{tr}(\mathbf{H})
\end{equation*}
which leads to the definitions of the following robust information criterion
\begin{equation*}
RIC(\lambda, \alpha) =  \hat{L}_n(\hat{\boldsymbol{\beta}}_n) + \operatorname{C}(n,p) \ \operatorname{df}(\lambda,\alpha)
\end{equation*}
and to the choice of $\lambda$ and $\alpha$ that minimise these measure. $\operatorname{C}(n,p)$ is the complexity penalty which in general depends on the sample size $n$ and on the number of parameters $p$; in case $\operatorname{C}(n,p) = 2$ we obtain a robust version of the AIC, while for $\operatorname{C}(n,p) = \log(n)$ we have a robust version of the BIC. As suggested in \citet{avella2018robust}, see also \citet{fan2013} and the references therein, we can set $\operatorname{C}(n,p) = \log(n) + \gamma \log(p)$ and $0 \le \gamma \le 1$ is a constant to have a robust extended BIC. For other penalties, the degree of freedom can be evaluated in a similar way.

Alternatively, we can select the penalty parameters by cross-validation. We first randomly divide the data into $K$ disjoint subsets of approximately the same number of observations. Let $c_1, \ldots, c_K$ be the number of observations in each of the subsets.  Let $ \hat{\boldsymbol{\beta}}^{(j)}_{\lambda,\alpha}$ be an estimator of $\boldsymbol{\beta}$ computed with penalty parameters $(\lambda,\alpha)$ and without using the observations of the $j$-th subset. The robust cross-validation criterion chooses $\lambda,\alpha$ that minimise
\begin{equation*}
\operatorname{RCV}(\lambda,\alpha) = \frac{1}{n} \sum_{j=1}^K c_j \hat{L}_n^{(j)}(\hat{\boldsymbol{\beta}}^{(j)}_{\lambda,\alpha})
\end{equation*}
where $\hat{L}_n^{(j)}$ is the loss function defined in \eqref{eq:lntilde} for the observations belonging to the $j$-th cross-validation subset.

\section{Asymptotic and robustness results}\label{sec:asres}

In this section we discuss some asymptotic and robustness properties of the penalized MT-estimators. 
Throughout this and the following section, we assume $(y_1, \x_1), \ldots,$ $(y_n, \x_n)$ and $(y,\x)$ to be i.i.d. random vectors following a GLM  with parameter $\boldsymbol{\beta}_*$, distribution $F$ and continuous and strictly increasing link function $g$ and we consider the estimator $\hat{\boldsymbol{\beta}}_{n}$ defined in \eqref{eq:MTpenalized}. 
\subsection{Consistency}
The following assumptions are needed to prove consistency of the penalized MT-estimators, and they are the same used for the MT-estimators, see \citet{ValdoraYohai2014}.
\begin{itemize}
\item[A1] $\sup_{\mu}\var_{\mu}(t(y)) = A <\infty$.
\item[A2] $\ m(\mu)$ is univocally defined for all $\mu$ and
$\mu_{1} < \mu_{2}$ implies $m(\mu_{1}) < m(\mu_{2})$.
\item[A3] $F(\mu,\phi)$ is continuous in $\mu$ and $\phi$.
\item[A4] Suppose that $\mu_{1} < \mu_{2}$, $Y_{1} \sim F(\mu_{1},\phi)$ and $Y_{2} \sim F(\mu_{2},\phi)$ then $Y_{1}$ is stochastically smaller $Y_{2}$.
\item[A5] The function $t$ is strictly increasing and continuous.
\end{itemize}  
A function $\rho:\mathbb R \to \mathbb R$ is a $\rho$-function if it satisfies the following assumptions
\begin{itemize}
\item[B1] $\rho(u) \geq$ $0$, $\rho(0)=0$ and $\rho(u)=\rho(-u)$.
\item[B2] $\lim_{u\rightarrow\infty}\rho(u)=a<\infty$. Without loss of generality we will assume $a=1$.
\item[B3] $0\leq u<v$ implies $\rho(u)\leq\rho(v)$.
\item[B4] $0\leq u<v$ and $\rho(u)<1$ implies $\rho(u)<\rho(v)$.
\item[B5] $\rho$ is continuous.
\end{itemize}  
For our purposes we also need  $\rho$ to verify the following
\begin{itemize}
\item[B6] Let $A$ be as in A1, then there exists $\eta$ such that $\rho(A^{1/2}+\eta)<1$.
\end{itemize}  
We also consider the following assumption on the distribution of the covariate vector $\x$
\begin{itemize}
\item[B7] Let $S=\left\{\mathbf{t} \in \mathbb{R}^{p+1}: ||\mathbf{t}||=1\right\}$, then $\inf \left\{\mathbb{P} \left(\mathbf{t}^\top \mathbf{x} \neq \mathbf{0}\right), \mathbf{t} \in S \right\}>0$. 
\end{itemize}
The following lemma has been proved in \citet{ValdoraYohai2014}, as part of the proof of their Theorem 1. We include it here because it is needed to state the next assumption.
\begin{lemma} \label{lemma:tau} Assume A1-A5 and B1-B7, then
	\begin{equation}
		\label{eq:tau}  
		\tau= \inf_{\mathbf{t} \in S} \mathbb{E}_{\boldsymbol{\beta}_*}\left(\lim_{\zeta\rightarrow\infty} \rho\left( \frac{t(y) -\mginv(\zeta\x^\top \mathbf{t})}{\sqrt{a(\phi)}} \right)\right) - L(\boldsymbol{\beta}_*) >0.
	\end{equation}
\end{lemma}
The following assumption on the distribution of the covariate vector $\x$, which is a little stronger than B7,  is needed for consistency.
\begin{itemize}
\item[B8] Let $\tau$ be defined as in \eqref{eq:tau}, then $\inf \left\{\mathbb{P} \left(\mathbf{t}^\top \mathbf{x} \neq \mathbf{0}\right), \mathbf{t} \in S \right\}>1-\tau$.
\end{itemize}
Note that assumption B8 is trivially verified if $\x = (1, \X)$ and $\X$ has a density. For $\epsilon > 0$ and a vector $\boldsymbol{\beta}$ we denote by $B(\boldsymbol{\beta}, \epsilon)$ the $L_2$-ball centered in $\boldsymbol{\beta}$ and radius $\epsilon$. Depending on the investigated property some of these assumptions are required on the penalty function. 
\begin{itemize}
\item[P1] $\penn{\lambda}{\boldsymbol{\beta}}{\alpha}$ is Lipschitz in a neighborhood of $\boldsymbol{\beta}_*$, that is, there exits $\epsilon > 0$ and a constant $K$ that does not depend on $\alpha$, such that if $\boldsymbol{\beta}_1, \boldsymbol{\beta}_2 \in B(\boldsymbol{\beta}_*, \epsilon)$ then $| \penn{\lambda}{\boldsymbol{\beta}_1}{\alpha} - \penn{\lambda}{\boldsymbol{\beta}_2}{\alpha}| \le K \| \boldsymbol{\beta}_1 - \boldsymbol{\beta}_2 \|_1$.

\end{itemize}
Assumption P1 holds, e.g., for elastic net, SCAD and MCP, for the sign penalty when $\|\boldsymbol{\beta}_*\|_2 > 0$ and for the bridge when $\alpha \ge 1$. In what follows we assume we have a sequence $\{\lambda_n,\alpha_n\}_{n=1}^\infty$ indexed by the sample size $n$ of parameters for the penalty term. Such a sequence might be random as well. The following theorem establishes the strong consistency of penalized MT-estimators. All the proofs are deferred to Appendix \ref{sec:proofs}. 

\begin{theorem}[Consistency]\label{teo:consistency}
 Assume A1-A5 and B1-B6 and B8 hold, and that  $\penn{\lambda_n}{\boldsymbol{\beta}_*}{\alpha_n} \xrightarrow{a.s.} 0$ when $n \rightarrow \infty$. Then $\hat{\boldsymbol{\beta}}_n$ is strongly consistent for $\boldsymbol{\beta}_{*}$.
\end{theorem}
It is sufficient that $\lambda_n \xrightarrow{a.s.} 0$ for elastic net, bridge, sign, SCAD and MCP to fulfill the assumption in the above Theorem.

\subsection{Rate of convergence, variable selection and asymptotic normality}
The theorems in this subsection are analogue to Theorems 2 to 8 in \citet{bianco2022penalized}, where they were proved for logistic regression. Here we show that they are also valid for penalized MT-estimators for general GLMs and establish the necessary assumptions.
\begin{itemize}
	\item[C1]\label{cond:c1} $F(\cdot,\mu,\phi)$ has three continuous and bounded derivatives as a function of $\mu$ and the link function $g(\mu)$ is twice continuously differentiable.
	\item[C2]\label{cond:c2} $\rho$ has three continuous and bounded derivatives. We write $\psi=\rho^{\prime}$.
	\item[C3]\label{cond:c3} $\E_{\mu}(\psi^{\prime}(t(y)-m(\mu)))  \neq 0$ for all $\mu$.
  \item[C4]\label{cond:c4}  Let $\boldsymbol{\Psi}(y, \x, \boldsymbol{\beta})$ be the derivative of $\rho\left((t(y)-\mginv(\x^\top\boldsymbol{\beta}))/\sqrt{a(\phi)}\right)$ with respect to $\boldsymbol{\beta}$ and $\mathbf{J}(y,\x,\boldsymbol{\beta})$ be the Hessian matrix. There exists $\eta>0$ such that $\E\left(\sup_{||\boldsymbol{\beta}-\boldsymbol{\beta}_*||\leq\eta}\left\vert J^{j,k}(y,\x,\boldsymbol{\beta})\right\vert \right) < \infty$, for all $1 \leq j,k\leq p,$ where $||\cdot||$ denotes the $L_{2}$ norm, and $\E(\mathbf{J}(y,\x,\boldsymbol{\beta}_*))$ is non singular.
\end{itemize}
It is worth mentioning that Lemma 5 in \citet{ValdoraYohai2014} proves that the function $m$ is twice differentiable, which ensures the existence of the derivatives in Assumption C4.
\begin{remark}
Assumption C2 does not hold for the Tukey-Bisquare $\rho$-function. A family of functions satisfying B1-B6 and C2 is suggested in \citet{ValdoraYohai2014} as
\begin{equation*}
\rho_k(u) =
\begin{cases}
1 - \left( 1 - \frac{u}{k}^2 \right)^4 & |u| \le k , \\
1 & |u| > k .
\end{cases}
\end{equation*}
with $k > A^{1/2}$. An alternative solution is to use an approximated version of the Tukey-Bisquare function, e.g. by the use of a positive, symmetric mollifier. However, as it is shown in the Monte Carlo simulation (Section \ref{sec:montecarlo}) the fact that Tukey-Bisquare function does not verify C2 seems to have very little impact on the performance of the MT-estimators.
\end{remark}
In order to establish the asymptotic normality of penalized MT-estimators we define the expectation of the Hessian matrix $\mathbf{J}$ and the variance of the gradient vector $\boldsymbol{\Psi}$ together with their empirical versions
\begin{align}\label{eq:AB}
\mathbf{A} & = \E_{\boldsymbol{\beta}_{*}}(\mathbf{J}(y,\x ,\boldsymbol{\beta}_*)), & \mathbf{B} & = \E_{\boldsymbol{\beta}_*}(\boldsymbol{\Psi}(y,\x ,\boldsymbol{\beta}_*)\boldsymbol{\Psi}(y,\x,\boldsymbol{\beta}_*)^\top), \\
\nonumber \mathbf{A}_n(\boldsymbol{\beta}) & = \frac{1}{n}\sum_{i=1}^{n} \mathbf{J}(y_i,\x_i ,\boldsymbol{\beta}), & \mathbf{B}_n(\boldsymbol{\beta}) & = \frac{1}{n}\sum_{i=1}^{n}  \boldsymbol{\Psi}(y_i,\x_i,\boldsymbol{\beta})\boldsymbol{\Psi}(y_i,\x_i,\boldsymbol{\beta})^\top.
\end{align}

\begin{theorem}[Order of convergence] \label{teo:order}
Assume A2, C1-C4. Furthermore, assume $\hat{\boldsymbol{\beta}}_n \xrightarrow{p} \boldsymbol{\beta}_*$.
	\begin{enumerate}
		\item[a)] If assumption P1 holds, $\|\hat{\boldsymbol{\beta}}_n-\boldsymbol{\beta}_*\|=O_{\mathbb{P}}\left(\lambda_n+1 / \sqrt{n}\right)$. Hence, if $\lambda_n=$ $O_{\mathbb{P}}(1 / \sqrt{n})$, we have that $\|\hat{\boldsymbol{\beta}}_n-\boldsymbol{\beta}_0\|=O_{\mathbb{P}}(1 / \sqrt{n})$, while if $\lambda_n \sqrt{n} \rightarrow \infty$, $\|\tilde{\boldsymbol{\beta}}_n-\boldsymbol{\beta}_0\|=O_{\mathbb{P}}\left(\lambda_n\right)$.
		\item[b)] Suppose $\penn{\lambda_n}{\boldsymbol{\beta}}{\alpha}=\sum_{\ell=1}^p J_{\ell, \lambda_n}\left(\left|\beta_{\ell}\right|\right)$ where the functions $J_{\ell, \lambda_n}(\cdot)$ are twice continuously differentiable in $(0, \infty)$, take non-negative values, $J_{\ell, \lambda_n}^{\prime}\left(\left|\beta_{*, \ell}\right|\right) \geq 0$ and $J_{\ell, \lambda_n}(0)=0$, for all $1 \leq \ell \leq p$. Let
\begin{equation*}
a_n = \max \left\{J_{\ell, \lambda_n}^{\prime}\left(\left|\beta_{*, \ell}\right|\right): 1 \leq \ell \leq p \text{ and } \beta_{*, \ell} \neq 0\right\}
\end{equation*}
and let $b_n = \frac{1}{\sqrt{n}} + a_n$. In addition, assume that there exists a $\delta>0$ such that
\begin{equation*}
\sup \left\{\left|J_{\ell, \lambda_n}^{\prime \prime}\left(\left|\beta_{*, \ell}\right|+\tau \delta\right)\right|: \tau \in[-1,1], 1 \leq \ell \leq p \text { and } \beta_{*, \ell} \neq 0\right\} \xrightarrow{p} 0 .
\end{equation*}
Then, $\|\hat{\boldsymbol{\beta}}_n-\boldsymbol{\beta}_*\|=O_{\mathbb{P}}\left(b_n\right)$.
\end{enumerate}
\end{theorem}

\begin{remark}
Theorem \ref{teo:order} implies that, for SCAD and MCP penalties, $\sqrt{n}$-consistency can be achieved assuming only that $\lambda_n \xrightarrow{p} 0$, while for lasso and Ridge penalties the additional requirement that $\lambda_n=O_{\mathbb{P}}(1 / \sqrt{n})$ is necessary. See Remark 5 in \citet{bianco2022penalized}.
\end{remark}
From now on, let us assume there exists $k$ such that $\boldsymbol{\beta}_*=(\beta_{*0}, \boldsymbol{\beta}_{*A}^\top, \boldsymbol 0^\top_{p-k})^\top$, where $\beta_{*0}\in \mathbb R$ is the intercept, $\boldsymbol{\beta}_{*A} \in \mathbb R^k$ are the active entries of $\boldsymbol\beta_*$ (that is to say the entries that are different from $0$) and $\boldsymbol 0^\top_{p-k}$ is a zero vector of length $p-k$.
\begin{theorem}[Variable selection]\label{teo:varsel}
Let $\hat{\boldsymbol{\beta}}_n=(\hat\beta_{n0}, \hat{\boldsymbol{\beta}}_{nA}^\top, \hat{\boldsymbol{\beta}}_{nB}^\top)^\top$ be the estimator defined in \eqref{eq:lntilde}. Assume A2 and C1-C4. Moreover, assume that $\sqrt{n}\|\hat{\boldsymbol{\beta}}_n-\boldsymbol{\beta}_*\|=O_{\mathbb{P}}\left(1\right)$ and that for every $C>0$ and $\ell\in \{k+1, \ldots, p\}$ there exists constants $K_{C, \ell}>0$ and $N\in \mathbb N$ such that, if $||u||_2<C$ and $b\geq N$, then
\begin{equation}\label{eq:varsel}
\penn{\lambda_n}{\boldsymbol{\beta}_{*} + \frac{\mathbf u}{\sqrt{n}}}{\alpha_n} - \penn{\lambda_n}{\boldsymbol{\beta}_{*} + \frac{\mathbf u^{(-\ell)}}{\sqrt{n}}}{\alpha_n} \geq \frac{ \lambda_n K_{C, \ell}}{\sqrt{n}}|u_\ell|,  
\end{equation} 
where $\mathbf{u}^{(-\ell)}$ is obtained by replacing $u_{\ell}$, the $\ell$-th coordinate of $\mathbf{u}$, with zero. Then 
\begin{enumerate}
	\item[(a)] For every $\tau>0$, there exists $b>0$ and $n_0 \in \mathbb{N}$ such that if $\lambda_n=b / \sqrt{n}$, we have that, for any $n \geq n_0$,
	\begin{equation*}
		\mathbb{P}\left(\hat{\boldsymbol{\beta}}_{n, B}=\mathbf{0}_{p-k}\right) \geq 1-\tau .
		\end{equation*}
	\item[(b)] If $\lambda_n \sqrt{n} \rightarrow \infty$, then
		\begin{equation*}
			\mathbb{P}\left(\hat{\boldsymbol{\beta}}_{n, B}=\mathbf{0}_{p-k}\right) \rightarrow 1 .
	\end{equation*}
\end{enumerate}
\end{theorem}

\begin{remark} \label{remark:varsel}
\citet{bianco2022penalized} show that condition \eqref{eq:varsel} holds for several penalty functions, e.g., lasso, SCAD, MCP and sign. On the contrary, Ridge penalty and bridge penalty with $\alpha>1$ do not verify \eqref{eq:varsel}, because they have a derivative that verifies $P'_\alpha(0)=0$. Note that, as a consequence of Theorems \ref{teo:order} and \ref{teo:varsel}, penalties SCAD and MCP lead to estimators that are $\sqrt{n}-$consistent and perform variable selection. As remarked in \citet{bianco2022penalized}, this states a difference with
\citet{avella2018robust}, who require stronger rates on $\lambda_n$. On the other hand, lasso and sign penalties cannot achieve these two properties simultaneously. In these cases, if we want to have $\sqrt{n}-$consistency, we can only ensure the result in Theorem \ref{teo:varsel}(a). See also Remark 6 in \citet{bianco2022penalized}.
\end{remark}
The following theorem states the variable selection property for random penalties. The proof is very similar to the proof of Theorem \ref{teo:varsel} and is omitted.
\begin{theorem}
\label{teo:varselrandom}
Let $\hat{\boldsymbol{\beta}}_n=(\hat\beta_{n0}, \hat{\boldsymbol{\beta}}_{nA}^\top, \hat{\boldsymbol{\beta}}_{nB}^\top)^\top$ be the estimator defined in \eqref{eq:lntilde}. Assume A2 and C1-C4. Moreover, assume that $\sqrt{n}\|\hat{\boldsymbol{\beta}}_n-\boldsymbol{\beta}_*\|=O_{\mathbb{P}}\left(1\right)$  
and that for some $\gamma>0, n^{(1+\gamma) / 2} \lambda_n \rightarrow \infty$. Furthermore, assume that for every $C>0, \ell \in\{k+1, \ldots, p\}$ and $\tau>0$, there exist a constant $K_{C, \ell}$ and $N_{C, \ell} \in \mathbb{N}$ such that if $\|\mathbf{u}\|_2 \leq C$ and $n \geq N_{C, \ell}$, we have that
\begin{equation*}
\mathbb{P}\left( \penn{\lambda_n}{\boldsymbol{\beta}_* + \frac{\mathbf{u}}{\sqrt{n}}}{\alpha_n} - \penn{\lambda_n}{\boldsymbol{\beta}_* + \frac{\mathbf{u}^{(-\ell)}}{\sqrt{n}}}{\alpha_n} \geq \frac{ \lambda_n K_{C, \ell}}{\sqrt{n^{1-\gamma}}}\left|u_{\ell}\right|\right)>1-\tau,
\end{equation*}
where $\mathbf{u}^{(-\ell)}$ and $u_{\ell}$ are defined as in Theorem \ref{teo:varsel}. Then,
\begin{equation*}
\mathbb{P}\left(\hat{\boldsymbol{\beta}}_{n, B}=\mathbf{0}_{p-k}\right) \rightarrow 1.
\end{equation*}  
\end{theorem}
We now consider the asymptotic behavior of the penalized MT-estimator when $\sqrt{n} \lambda_{n} \xrightarrow{p} b$ for some constant $b$.
\begin{theorem}
\label{teo:asdist}
Assume A2, C1-C4 and P1 hold. Furthermore assume that $\sqrt{n}(\hat{\boldsymbol{\beta}}_{n}-\boldsymbol{\beta}_*)=O_{\mathbb P}(1)$ and that $\sqrt{n} \lambda_{n} \xrightarrow{p} b$, with $0\leq b < + \infty$. Then, if $\left\|\boldsymbol{\beta}_*\right\| \neq 0$,
\begin{equation*}
  \sqrt{n}(\hat{\boldsymbol{\beta}}_{n}-\boldsymbol{\beta}_*) \xrightarrow{d} \arg\min_{\mathbf{w}} R(\mathbf{w})
\end{equation*}  
where the process $R: \mathbb{R}^{p+1} \rightarrow \mathbb{R}$ is defined as
\begin{equation*}
R(\mathbf{w})=\mathbf{w}^\top \mathbf{w}_0+\frac{1}{2} \mathbf{w}^{\top} \mathbf{A} \mathbf{w}+ {b} \mathbf{w}^\top \mathbf{q}(\mathbf{w}).
\end{equation*}
with $\mathbf{w}_0 \sim N_{p+1}(\mathbf{0}, \mathbf{B})$, $\mathbf{q}(\mathbf{w})=(q_0(\mathbf w), q_{1}(\mathbf{w}), \ldots, q_{p}(\mathbf{w}))^\top$.
\begin{itemize}
\item[(i)] For the elastic net penalty function
\begin{equation*}
q_{\ell}(\mathbf{w}) = (1-\alpha)  \beta_{*, \ell} +  \alpha \left\{ \sign(\beta_{*, \ell}) \mathbb{I}_{\left\{\beta_{*, \ell} \neq 0\right\}} + \sign(w_{\ell}) \mathbb{I}_{\left\{\beta_{*, \ell}=0\right\}}\right\}.
\end{equation*}
and $\w = (w_0, w_1, \ldots, w_p)$.
\item[(ii)] For the bridge penalty function, with $\alpha > 1$ we have
\begin{equation*}
q_{\ell}(\mathbf{w}) = \sign(\beta_{*, \ell}) |\beta_{*, \ell}|^{\alpha-1}
\end{equation*}
$\alpha=1$ corresponds to the lasso, which is equivalent to the elastic net penalty function and $\alpha=1$. 
\item[(iii)] For the sign penalty function
\begin{equation*}
q_{\ell}(\mathbf{w}) = \nabla_\ell(\boldsymbol{\beta}_*) \mathbb{I}_{\left\{\beta_{*, \ell} \neq 0\right\}} + \sign(w_\ell)/\|\boldsymbol{\beta}_*\|_2 \mathbb{I}_{\left\{\beta_{*, \ell} = 0\right\}} \mathbf{e}_\ell
\end{equation*}
and $\nabla_\ell(\boldsymbol{\beta}) = - (|\beta_\ell| / \|\boldsymbol{\beta}_*\|_2^3) \boldsymbol{\beta} + \sign(\beta_\ell)/\|\boldsymbol{\beta}\|_2 \mathbf{e}_\ell$.
\end{itemize}
\end{theorem}

\begin{remark}
Note that, if $\sqrt{n} \lambda_n \rightarrow 0$, $R$ is minimized when $\w = - A^{-1} \w_0 \sim N_{p+1}(\mathbf{0}, A^{-1}\mathbf{B}A^{-1})$. Therefore, the penalized MT-estimator defined in \eqref{eq:MTpenalized} has the same asymptotic distribution as its unpenalized version introduced in \citet{ValdoraYohai2014}. Also note that, when $\penn{\lambda}{\boldsymbol\beta)}{\alpha} = \lambda||\boldsymbol\beta||_2^2$ (the Ridge estimator), then $\sqrt{n}(\hat{\boldsymbol{\beta}}_n-\boldsymbol{\beta}_*) \xrightarrow{d} -\mathbf{A}^{-1}(\mathbf{w}_0 + b \boldsymbol{\beta}_*) \sim N_{p+1}(-b \mathbf{A}^{-1} \boldsymbol{\beta}_*, \mathbf{A}^{-1} \mathbf{B} \mathbf{A}^{-1})$,  so in this case $\hat{\boldsymbol{\beta}}_n$ is asymptotically biased.
\end{remark}
In the next theorem we study the behaviour of $\hat{\boldsymbol{\beta}}_n$ when $\sqrt{n} \lambda_{n} \xrightarrow{p} \infty$.

\begin{theorem}\label{teo:asdist2}
Assume A2, C1-C4 and P1 hold. Furthermore assume that $\hat{\boldsymbol{\beta}}_{n}-\boldsymbol{\beta}_*=O_{\mathbb P}(\lambda_{n})$, $\lambda_{n} \xrightarrow{p} 0$, $\sqrt{n} \lambda_{n} \xrightarrow{p} \infty$ and $\left\|\boldsymbol{\beta}_*\right\| \neq 0$.
Then, $(1 / \lambda_{n})(\hat{\boldsymbol{\beta}}_{n}-\boldsymbol{\beta}_*) \xrightarrow{p} \arg\min_{\mathbf{w}} R(\mathbf{w})$, where the process $R: \mathbb{R}^{p+1} \rightarrow \mathbb{R}$ is defined as
\begin{equation*}
R(\mathbf{w})=\frac{1}{2} \mathbf{w}^\top \mathbf{A} \mathbf{w}+\mathbf{w}^\top \mathbf{q}(\mathbf{w}) .
\end{equation*}
with $\mathbf{q}(\mathbf{w})$ being the function defined in Theorem \ref{teo:asdist}.
\end{theorem}
The following theorem gives the asymptotic distribution of the vector of active entries, when the penalty is consistent for variable selection. Hence, assume there exists $k$ such that $\boldsymbol{\beta}_*=(\beta_{*0}, \boldsymbol{\beta}_{*A}^\top, \boldsymbol 0^\top_{p-k})^\top$. Let
\begin{equation*}
\nabla \penn{\lambda}{\mathbf{b}}{\alpha} = \frac{\partial \penn{\lambda}{\left(\mathbf{b}^\top, \mathbf{0}_{p-k}^\top\right)^\top}{\alpha}}{\partial \mathbf{b}}
\end{equation*}  
and let $\tilde{\mathbf{A}}$ and $\tilde{\mathbf{B}}$ be the $(k +1) \times (k + 1)$ submatrices of $\mathbf{A}$ and $\mathbf{B}$, respectively, corresponding to the first $k+1$ entries of $\boldsymbol{\beta}_*$, where $\mathbf{A}$ and $\mathbf{B}$ were defined in equation \eqref{eq:AB}.

\begin{theorem}\label{teo:oracle}
  Let $\hat{\boldsymbol{\beta}}_n=(\hat\beta_{n0}, \hat{\boldsymbol{\beta}}_{nA}^\top, \hat{\boldsymbol{\beta}}_{nB}^\top)^\top$ be the estimator defined in \eqref{eq:lntilde} and assume A2 and C1-C3 hold. Suppose that there exists some $\delta>0$ such that
\begin{equation} \label{eq:conteo11}
		\sup _{\|\boldsymbol\beta_A-\boldsymbol\beta_{* A}\| \leq \delta} \left\|\nabla \penn{\lambda_n}{\beta_0, \boldsymbol{\beta}_A}{\alpha_n} \right\|=o_{\mathbb{P}}\left(\frac{1}{\sqrt{n}}\right). 
\end{equation}	
Assume also that $\mathbb{P}\left(\hat{\boldsymbol{\beta}}_{n B}=\mathbf{0}_{p-k}\right) \rightarrow 1$ and $\hat{\boldsymbol{\beta}}_n \xrightarrow{p} \boldsymbol{\beta}_*$. Then, if $\tilde{\mathbf{A}}$ is invertible,
\begin{equation*}
\sqrt{n}\left( \left(\begin{matrix} \hat\beta_{n0} \\
\hat{\boldsymbol{\beta}}_{n A}
\end{matrix} \right)-  \left(\begin{matrix}
\beta_{*0}\\ \boldsymbol{\beta}_{* A}
\end{matrix}\right) \right) \xrightarrow{d} N_k\left(\mathbf{0}, \tilde{\mathbf{A}}^{-1} \tilde{\mathbf{B}} \tilde{\mathbf{A}}^{-1}\right).
\end{equation*}
\end{theorem}

\begin{remark} SCAD and MCP penalties verify \ref{eq:conteo11} (see Remark 9 in \citet{bianco2022penalized}) and $\mathbb{P}\left(\hat{\boldsymbol{\beta}}_{n B}=\mathbf{0}_{p-k}\right) \rightarrow 1$ because of Theorem \ref{teo:varsel} and Remark \ref{remark:varsel}. This means that the penalized MT-estimators proposed in this paper have the oracle property for SCAD and MCP penalizations.
\end{remark}
In the appendix we give the proofs of Theorems \ref{teo:order}, \ref{teo:asdist} and \ref{teo:asdist2}. Similar modifications to Theorems 2 and 8 in \citet{bianco2022penalized} allow to prove Theorems \ref{teo:varsel} and \ref{teo:oracle}.

\subsection{Robustness properties}\label{sec:robustness}
We introduce now some results about the robustness properties of the penalized MT-estimators. The robustness of an estimator is measured by its stability when a small fraction of the observations are arbitrarily replaced by outliers that may not follow the assumed model. We expect that a robust estimator should not be much affected by a small fraction of outliers. A classic quantitative measure of an estimator's robustness, introduced by \citet{donohohuber1983}, is the finite-sample replacement breakdown point. Loosely speaking, the finite-sample replacement breakdown point of an estimator is the minimum fraction of outliers that may take the estimator beyond any bound. In our case, given a sample $S_0 = \{(y_i, \x_i), 1 \le i \le n\}$ from model \eqref{eq:glm}, let $\mathcal{S}_m$ be the collection of samples that can be obtained by arbitrarily replacing $m \le n$ observations from $S_0$. Given an estimator $\hat{\boldsymbol{\beta}}(S_m)$ with $S_m \in \mathcal{S}_m$ we define its finite-sample breakdown point as 
\begin{equation*}
	\operatorname{FBP}(\hat{\boldsymbol{\beta}}) = \frac{m_*}{n} ,
\end{equation*}
where
\begin{equation*}
	m_* = \max\{ 0 \le m \le n: \| \hat{\boldsymbol{\beta}}(S_m) \| < \infty, \forall S_m \in \mathcal{S}_m \} . 
\end{equation*}
A breakdown point equal to $\varepsilon_*$ only guarantees that for any given contamination fraction $\varepsilon < \varepsilon_*$, there exists a compact set such that the estimator in question remains in that compact set whenever a fraction of $\varepsilon$ observations is arbitrarily modified. However, this compact set may be very large. Thus, although a high breakdown point is always a desirable property, an estimator that has a high breakdown point can still be largely affected by a small fraction of contaminated observations. We notice that, as also discussed in \citet{smucler2017robust} for linear models, our penalized MT-estimators for GLM are not equivariant, hence the concept of finite-sample breakdown point is rather vacuous \citep{daviesgather2005}. Consider the following additional assumptions on the penalty function
\begin{enumerate}
\item[P2] $\penn{\lambda}{\boldsymbol{\beta}}{\alpha}$ has a second order Taylor expansion around $\boldsymbol{\beta}_*$.
\item[P3] Given $\alpha$, let $\underline{\boldsymbol{\beta}} = \arg\min_{\boldsymbol{\beta}} \penn{\lambda}{\boldsymbol{\beta}}{\alpha}$, $\underline{b} = \penn{1}{\underline{\boldsymbol{\beta}}}{\alpha}$ and $\overline{b} = \lim_{\|\boldsymbol{\beta}\| \rightarrow \infty} \penn{1}{\boldsymbol{\beta})}{\alpha}$. Then $\overline{b} > \underline{b} \ge 0$ and $\lambda > a/(\overline{b} - \underline{b})$, where $a$ is as in B2. If $\overline{b}=\infty$, this assumption is simply that $\lambda > 0$.
\item[P3'] Using previous notation, let $\overline{b} = \lim\inf_{\|\boldsymbol{\beta}\| \rightarrow \infty} \penn{1}{\boldsymbol{\beta}}{\alpha}$. Then $\overline{b} > \underline{b} \ge 0$ and $\lambda > a/(\overline{b} - \underline{b})$.
\item[P4] For any given $\alpha$, $\lim_{\|\boldsymbol{\beta}\| \rightarrow \infty} \penn{1}{\boldsymbol{\beta}}{\alpha}$ exists and is equal to $\sup_{\boldsymbol{\beta}} \penn{1}{\boldsymbol{\beta}}{\alpha}$.  
\end{enumerate}
Assumption P2 holds, e.g., for bridge with $\alpha \ge 1$, SCAD and MCP. It does not hold for elastic net and for sign penalties. Assumption P3 holds for elastic net and bridge, while the more restrictive P3' holds for SCAD when $\lambda^3 > 2a/(\alpha+1)$, and for MCP when $\lambda^3 > 2 a/\alpha$. It does not hold for the sign penalty. Assumption P4 holds for elastic net and bridge. It does not hold for SCAD and MCP since $\sup_{\boldsymbol{\beta}} \penn{\lambda}{\boldsymbol{\beta})}{\alpha}$ is equal to $p \lambda^2 (\alpha+1)/2$ and $p \alpha \lambda^2 /2$ respectively.
P4 does not hold for the sign penalty either, since $\sup_{\boldsymbol{\beta}} \penn{\lambda}{\boldsymbol{\beta}}{\alpha}=\sup_{\boldsymbol{\beta}\neq 0} \penn{\lambda}{\boldsymbol{\beta}}{\alpha} = \sqrt{p}$.
\begin{theorem}[Finite-Sample Breakdown Point] \label{teo:finitebreakdown}
Assume A2, B1 and B2 and P3 or P3' hold. Let $\hat{\boldsymbol{\beta}}_n$ be the penalized MT-estimator given by \eqref{eq:MTpenalized}. Then
	\begin{equation*}
		\operatorname{FBP}(\hat{\boldsymbol{\beta}}_n) = 1 .
	\end{equation*}
\end{theorem}

\begin{remark}
For some penalty functions, e.g., elastic net and bridge, we have $\overline{b} = \infty$ and $\underline{b} = 0$ and hence the above theorem holds for all $\lambda_n > 0$. 
\end{remark}
We turn our attention to the stability of the penalized MT-estimators in an  $\varepsilon$-contamination neighborhood. We will obtain a result equivalent to Theorem 1 in \citet{avella2018robust}. Let $(y, \x)$ follow a GLM with parameter $\boldsymbol{\beta}$, link function $g$ and distribution function $F$ and let $\x \sim H(\cdot)$. We consider an $\varepsilon$-contamination neighborhood of $F(,\mu,\phi)$ for each fixed $\x$, that is, for any $0 \le \varepsilon < 0.5$,
\begin{equation*}
\mathcal{F}_{\varepsilon}(\x) = \{ (1 - \varepsilon) F(y,\mu,\phi) + \varepsilon G_\x(y|\x) \}
\end{equation*}
where $\{ G(y|\x): \x \in \mathbb{R}^p \}$ is a family of conditional distributions. Given $\lambda \ge 0$ and $\alpha$, let $\boldsymbol{\beta}_\varepsilon$ be the solution of the following optimization problem
\begin{equation} \label{equ:stability}
	\boldsymbol{\beta}_\varepsilon = \arg\min_{\boldsymbol{\beta}} \int \int \rho \left( \frac{t(y) - \mginv \left(\x^\top \boldsymbol{\beta}\right)}{\sqrt{a(\phi)}} \right) \ dF_{\varepsilon}(y|\x) \ dH(\x) + \penn{\lambda}{\boldsymbol{\beta})}{\alpha}
\end{equation}
where $F_{\varepsilon}(\cdot|\x) \in \mathcal{F}_{\varepsilon}(\x)$ and $\x \sim H(\cdot)$. Let $\boldsymbol{\beta}_0$ be the solution for $\varepsilon = 0$.
\begin{theorem}[Stability in a $\varepsilon$-contamination neighborhood] \label{teo:stability}
Assume A1-A5, B1-B6, B8, C1-C4, P2 and $\E(\|\x\|^2)<\infty$. Let $\lambda \ge 0$ and $\alpha$ be fixed, and assume
\begin{itemize}
\item[C5]\label{cond:c5} $\E(\mathbf{J}(y,\x,\boldsymbol{\beta}_0))$ is non singular.
\end{itemize}
Then, there exists an $\varepsilon_* > 0$ such that for all $0 \le \varepsilon < \varepsilon_*$ we have $\| \boldsymbol{\beta}_\varepsilon - \boldsymbol{\beta}_0 \| = O(\varepsilon)$.
\end{theorem}
The Asymptotic Breakdown Point (ABP) is a measure of robustness of an estimator introduced by \citet{hampel1971}. In this view, for the penalized MT-estimator we define its asymptotic breakdown point as
\begin{equation*}
	\operatorname{ABP}(\hat{\boldsymbol{\beta}}) = \sup_{\varepsilon} \left\{ 0 \le \varepsilon \le 1: \sup_{F_\varepsilon \in \mathcal{F}_\varepsilon} \| \boldsymbol{\beta}_\varepsilon \| < \infty \right\}
\end{equation*}
where $\boldsymbol{\beta}_\varepsilon$ is defined in \eqref{equ:stability}.
\begin{theorem}[Asymptotic Breakdown Point] \label{teo:abp}
	Assume A1-A5, B1-B6, B8 and P4 hold. Let
	\begin{equation*}
		\varepsilon_* = \frac{\E(\min(\rho(t(y) - m_1), \rho(t(y) - m_2))) - \E(\rho(t(y) - \mginv(\x^\top\boldsymbol{\beta}_*)))}{1 + \E(\min(\rho(t(y) - m_1), \rho(t(y) - m_2))) - \E(\rho(t(y) - \mginv(\x^\top\boldsymbol{\beta}_*)))}
	\end{equation*}
	where $m_1 = \lim_{\mu \rightarrow \mu_1} m(\mu)$, $m_2 = \lim_{\mu\rightarrow \mu_2} m(\mu)$ and $\mu_1$ and $\mu_2$ are the values defined in Section \ref{sec:MT}. Then
	\begin{equation*}
		\operatorname{ABP}(\hat{\boldsymbol{\beta}}) \ge \varepsilon_* . 
	\end{equation*}  
\end{theorem}

\begin{remark}
	As it is shown in \citet{ValdoraYohai2014}, for the Poisson case, $m_1=0$ and $m_2=\infty$ and hence
	\begin{equation*}
		\varepsilon_* = \frac{\E(\rho(t(y))) - \E(\rho(t(y) - \mginv(\x^\top\boldsymbol{\beta}_*)))}{1 + \E(\rho(t(y))) - \E(\rho(t(y) - \mginv(\x^\top\boldsymbol{\beta}_*)))}
	\end{equation*}
	so that $\varepsilon_*$ is small only when the probability that $\mginv(\x^\top\boldsymbol{\beta}_*))$ is close to zero is large. For the Poisson model this happens if $\x^\top\boldsymbol{\beta}_*$ is negative and has a large absolute value. Note that in this case $P(y=0)$ is large and a small fraction of observations equal to $0$ can make the fraction of observed zeros larger than $0.5$. Therefore the good non-null observations may be mistaken for outliers. Further examples and details are available in Section 5 of \citet{ValdoraYohai2014}.
\end{remark}

\section{Monte Carlo study}
\label{sec:montecarlo}
In this section we report the results of a Monte Carlo study where we compare the performance of different estimators for Poisson regression with Ridge and lasso penalizations. For the Ridge penalization we have two estimators: the estimator introduced in \citet{friedman2010regularization} and implemented in the package {\tt{glmnet}} (ML Ridge) and  the estimator introduced here (MT Ridge). For the lasso penalization we have three estimators: the estimator introduced in \citet{friedman2010regularization}  (ML lasso),  the estimator introduced here (MT Lasso) and the estimator introduced in \citet{avella2018robust} (RQL lasso).

We report in this section two simulation settings labelled as AVY and AMR respectively, a third setting is reported in the Supplementary Material. Setting AVY is similar to ``model 1'' considered in \citet{agostinelli2019initial}. The differences are that we increase the number of explanatory variables $p$ and that we introduce correlation among them. Setting AMR is the same as in \citet{avella2018robust}.

Let us first describe setting AVY. Let $\x = (1, \X)$ be a random vector in $\mathbb{R}^{p+1}$ such that $\X$ is distributed as $\mathcal{N}_{p}(\mathbf{0}, \mathbf{\Sigma})$, where $\boldsymbol\Sigma_{i,j} = \rho^{|i-j|}$, and $\rho = 0.5$. Let $\mathbf{e}_i$ be the vector of $\mathbb{R}^{p+1}$ with all entries equal to zero except for the $i$-th entry which is equal to one. Let $\boldsymbol{\beta}_* = \mathbf{e}_2$, $\mu_{\x} = \exp(\x^\top \boldsymbol{\beta}_*)$ and $y$ be a random variable such that $y | \x \sim \mathcal{P}(\mu_{\x})$. In this setting, we consider dimensions $p=10, 50$ and sample sizes $n=100, 400, 1000$. 

In the AMR setting, $\boldsymbol{\beta}_* = 1.8 \mathbf{e}_2 + \mathbf{e}_3 + 1.5 \mathbf{e}_6$, $\X=(x_1,\ldots, x_p)$ is such that $x_i \sim U(0,1)$ and $\operatorname{\mathbb{C}or}(x_i, x_j)=\rho^{|i-j|}$ with $\rho =0.5$ and $y \vert \x \sim \mathcal{P}(\mu_{\x})$. In this setting, we considered  $p=100, 400, 1600$ and sample sizes $n=50, 100, 200$.

In both settings, we generated $N=1000$ replications of $(y, \x)$ and computed ML Ridge, MT Ridge, ML lasso, MT Lasso and RQL lasso for each replication. 
We then contaminated the samples with a proportion $\epsilon$ of outliers. In setting AMR, we only contaminated the responses, which were generated following a distribution of the form $y \vert \x \sim (1-b) \mathcal{P}\left(\mu_{\x} \right)+b \mathcal{P}\left(y_0 \mu_{\x}\right)$, $b\sim \mathcal B(1, \epsilon)$ and $y_0 = 0, 5, 10, 100$. 
In setting AVY we also contaminated the covariates, replacing a proportion $\epsilon$ of the observations for $(y_0, \mathbf{x}_0)$, with $\mathbf{x}_0 = \mathbf{e}_1 + 3 \mathbf{e}_2$ and $y_0$ in a grid ranging from $0$ to $400$. In both settings, we considered contamination levels $\epsilon=0.05, 0.1, 0.15$. Note that in the AMR setting $E(y)=E\left(E(y |\x)\right) = E\left(exp(1.8 x_1 + x_2 + 1.5 x_5 )\right) \approx 3.71$, while in the AVY setting $E(y |\x_0) \approx \exp(3) \approx 20.08$, and that the contaminating observations we add are only harmful for the fit if the response is very different from its expectation. This is why, in Figures \ref{fig:mseRidge} and \ref{fig:mselasso} we see that, for the AMR setting, the MSE of all estimators is increasing, while for the AVY setting, it decreases when $y_0$ increases from $0$ to $20$ and then increases until it stabilizes.

As a performance measure we computed the mean squared error of each estimator as
\begin{equation*}
{\operatorname{MSE}(\hat{\boldsymbol{\beta}})}=\frac{1}{N} \sum_{j=1}^{N}||\hat{\boldsymbol{\beta}}_j - \boldsymbol{\beta}_*||^2 ,
\end{equation*}
where $\hat{\boldsymbol{\beta}}_j$ is the value of the estimator at the $j$-th replication and $N$ is the number of Monte Carlo replications.

Figure \ref{fig:mseRidge}  summarizes the results for MT Ridge and ML Ridge in AVY setting, while Figure \ref{fig:mselasso} summarizes the  results for MT Lasso, ML lasso and RQL lasso for AVY (left) and AMR (right) setting. In these figures we plot the MSE of the estimators as a function of the contamination $y_0$. 

 These results show that, in these scenarios, as the size of the outlying response increases, the mean squared errors of MT Ridge and MT Lasso remain bounded, while the mean squared errors of ML Ridge, ML lasso and RQL lasso seem to increase without bound. The results of the complete simulation study are given in the Supplementary Material.

\begin{figure}[tp]
\begin{center}
\includegraphics[width=0.45\textwidth]{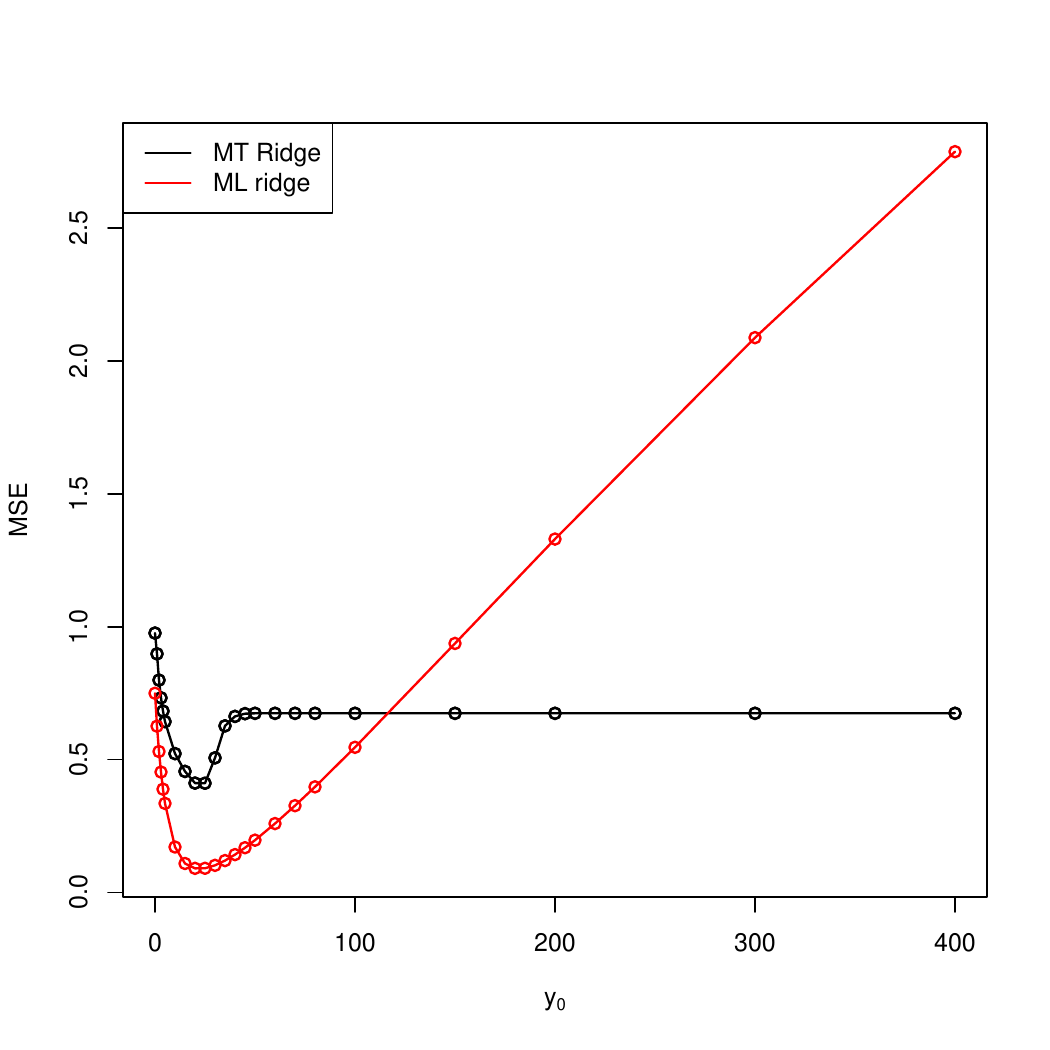}
\includegraphics[width=0.45\textwidth]{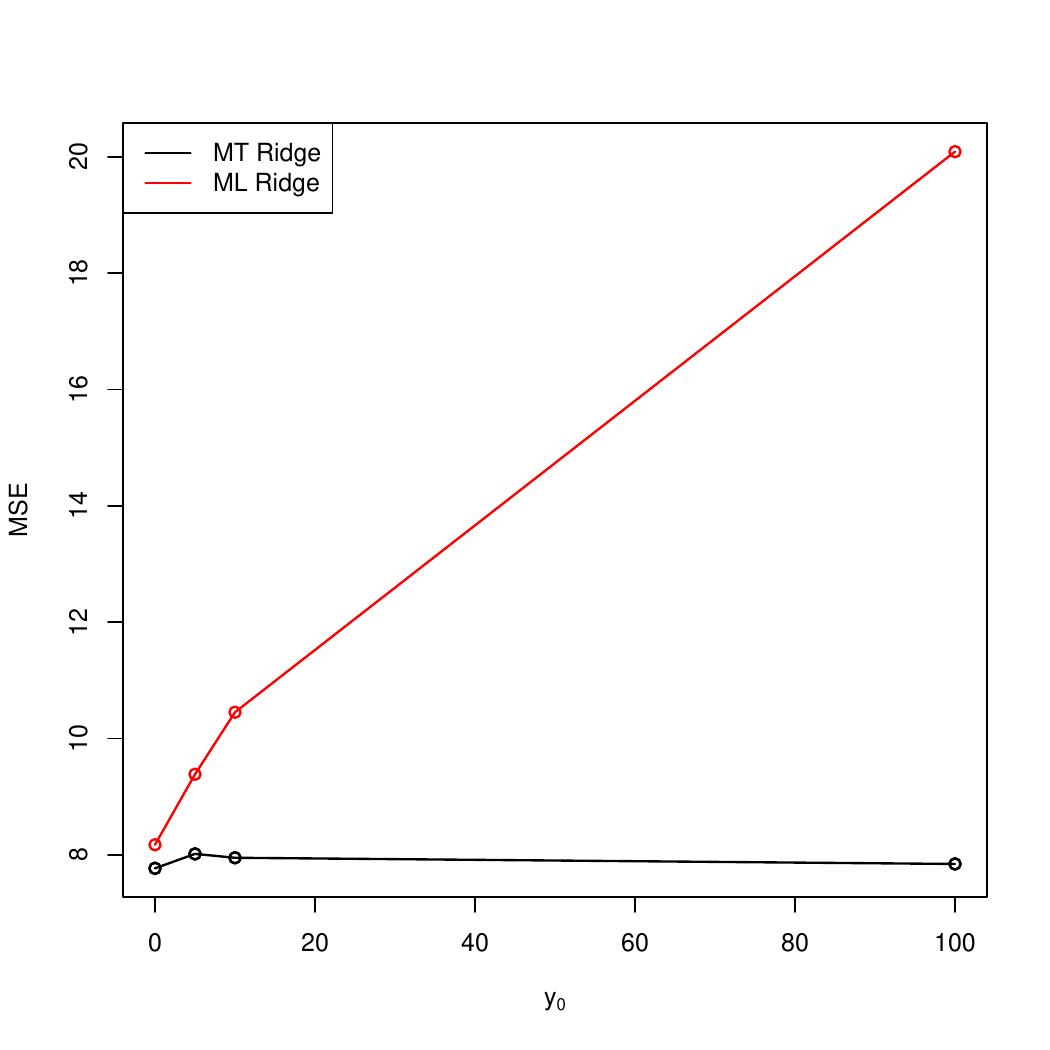}
\end{center}
\caption{Mean squared errors of MT Ridge and ML Ridge for AVY setting, with $p=50$, $n=100$ and $\epsilon=0.1$ (left) and for AMR setting with $p=1600$, $n=50$ and $\epsilon=0.1$ (right).}
\label{fig:mseRidge}\end{figure}

\begin{figure}[tp]
\begin{center}
\includegraphics[width=0.45\textwidth]{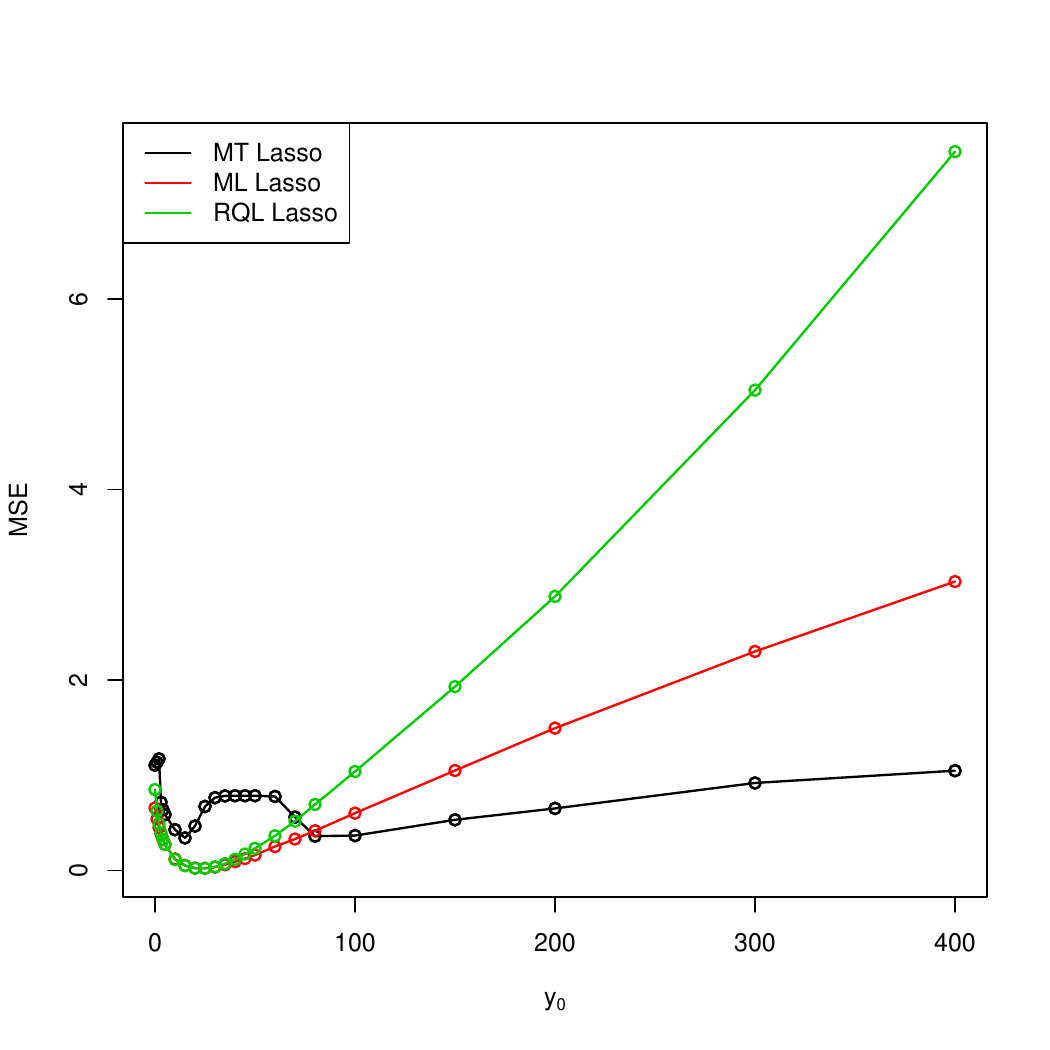}
\includegraphics[width=0.45\textwidth]{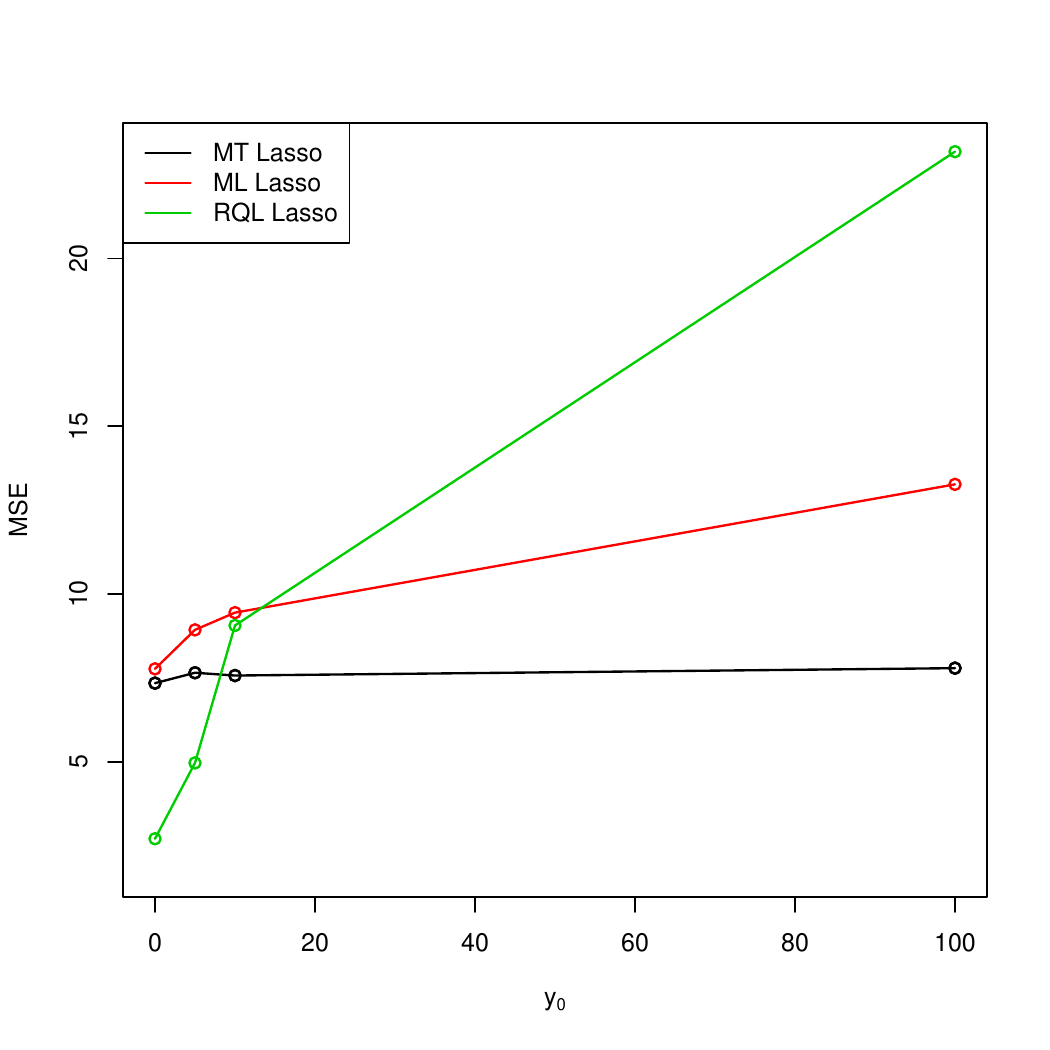}
\end{center}
\caption{Mean squared errors of MT Lasso, ML lasso and RQL lasso for AVY setting, with $p=50$, $n=100$ and $\epsilon=0.1$ (left) and for AMR setting with $p=1600$, $n=50$ and $\epsilon=0.1$ (right).}
\label{fig:mselasso}
\end{figure}

\section{Example: Right Heart Catheterization}
\label{sec:examples}
\label{sec:rhc}
This data set was used by \citet{Connors1996} to study the effect of right heart catheterization in critically ill patients. It contains data from 5735 patients from five medical centers in the USA between 1989 and 1994 on several variables. These variables include laboratory measurements taken on day one, dates of admission and discharge, category of the primary disease, and whether or not the right heart catheterization was performed, among other features.  
A detailed description of the covariates can be found in \citet{Connors1996}.
The data were downloaded from the repository at Vanderbilt University, specifically from
\begin{verbatim}
	http://biostat.mc.vanderbilt.edu/wiki/pub/Main/DataSets/rhc.csv
\end{verbatim}
We focus on patients with disease category COPD (chronic obstructive pulmonary disease).
We define the response variable as $y= \text{ length of hospital stay} - 2$, computed as discharge date minus admission date minus $2$. The matrix $\mathbf x$ of covariates contains information on $57$ variables for each of the 456 patients with COPD. We assume that $y|\mathbf{x}$ follows a Poisson distribution with mean $\mu = \exp(\boldsymbol{\beta}^\top \mathbf{x}) $ and we seek to estimate $\boldsymbol{\beta}$ and to study its usefulness to explain and predict the length of hospital stay.

To study the predictive power of the proposed method, we randomly divided the data set in a training set, composed of $75\%$ of the observations and a test set, composed of the remaining $25\%$. We then computed all the estimates based on the training set and their deviance residuals based on the test set. We repeated this for ten randomly chosen partitions. 

Figure \ref{fig:medianboxplots} shows boxplots of the median absolute deviance residuals for each method and for each train-test partition and different methods. These figures show that MT, MT Lasso and MT Ridge give a better prediction for at least $50\%$ of the test set than ML, ML lasso and MT Ridge, respectively. They also show the beneficial effect of penalization in the prediction error of the MT methods. 
It is worth mentioning that ML method only converged in 7 out of the 10 partitions. This is due to the large number of covariates and relatively small number of observations. The penalized ML methods achieve convergence for all 10 partitions.
\begin{figure}
\centering
\begin{tabular}{cc}
\includegraphics[scale=0.6]{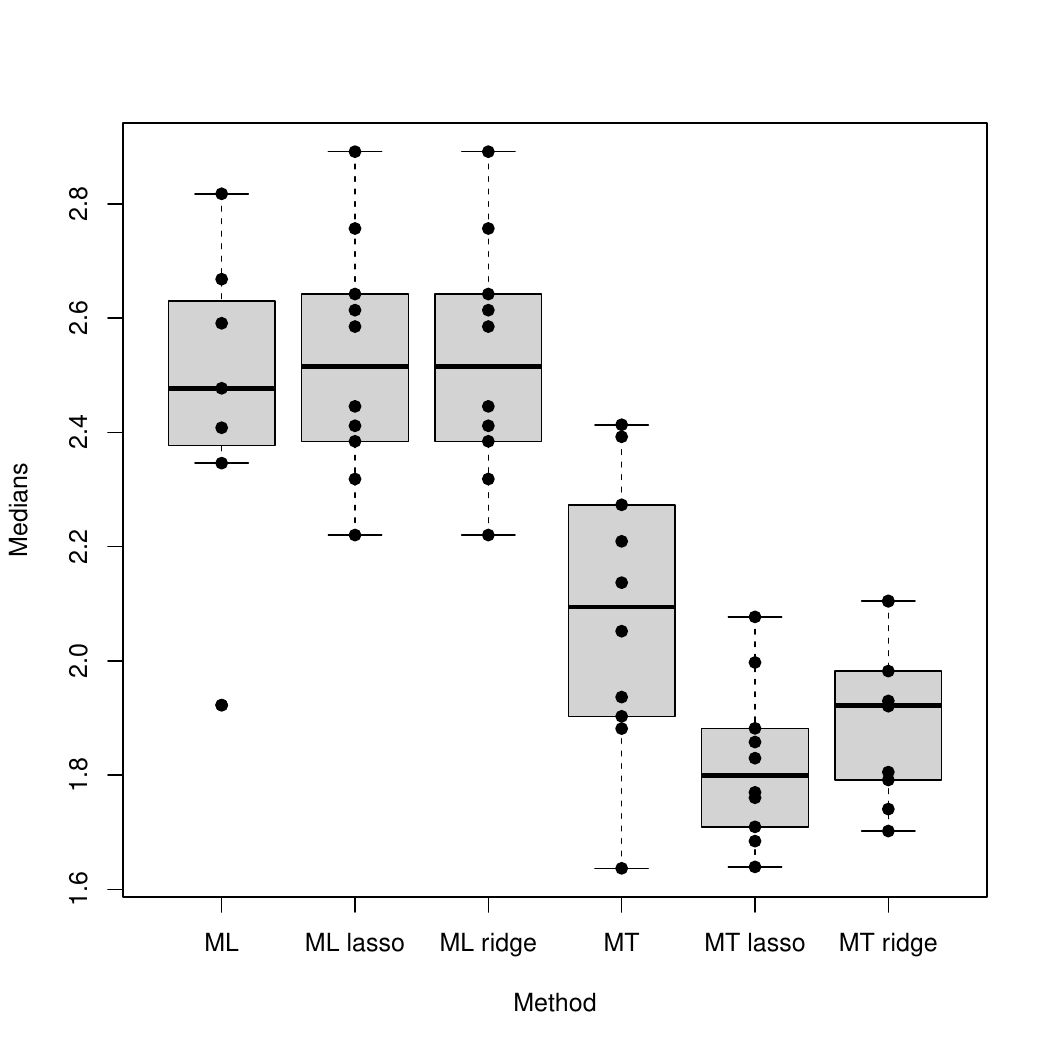}
\end{tabular}\caption{Median of the absolute deviance residuals for each train-test partition and different methods.}\label{fig:medianboxplots}
\end{figure}

In general, robust estimators can be used for outlier detection. With this aim, we compute the estimators using the complete sample. Then we generate a bootstrap sample of deviance residuals of size $B=100000$ in the following way. For  $k = 1, \ldots, B$, we randomly choose an index $i \in \{1, \ldots, n\}$ and we generate a response $y \sim \mathcal P(\mathbf{x}_i^\top \hat{\boldsymbol\beta})$, where $\mathbf{x}_i$ is the $i$-th row of the matrix of covariates and $\hat{\boldsymbol{\beta}}$ is a robust estimate (MT Lasso or MT Rigde) and  we compute the corresponding deviance residual. Let $q_1$ and $q_2$ be the minimum and the maximum  of the bootstrap sample of deviance residuals. We consider outliers observations with a deviance residual smaller than $q_1$ or larger than $q_2$.
We repeated this procedure for MT Lasso and MT Ridge with the same bootstrap samples and we obtained $q_1=-4.83$ and $q_2=4.16$ in both cases. Both methods detect $64$ outliers, aproximately $14\%$ of the observations. 
Finally, Table \ref{table:rhcMTlassocoef} shows the estimated coefficients by lasso methods. Entries for which all methods are $0$ are omitted.

ML lasso estimates as $0$ all the coefficients but the intercept. On the other hand, MT Lasso estimates as $0.0001$ the coefficient corresponding to the covariate {\tt scoma1}, that is a coma score, according to the description of the data given in \cite{sokbae2021ATbounds}.
\begin{table}[ht]
	\centering
	\begin{tabular}{rrr}
		\hline
		& ML lasso & MT Lasso \\ 
		\hline
		(Intercept) & 2.7846 & 2.1402 \\ 
%%%		surv2md1 & 0.0000 & 0.0000 \\ 
		scoma1 & 0.0000 & 0.0001 \\ 
		\hline
	\end{tabular}\caption{Estimated coefficients by lasso methods. Coefficients estimated as $0$ by both methods are omitted.}\label{table:rhcMTlassocoef}
\end{table}

\section{Conclusion}
\label{sec:conclusion}
This paper addresses the issue of robust estimation in the context of high-dimensional covariates for generalized linear models (GLMs). To this end, we introduce penalized MT-estimators. These estimators are constructed by incorporating a penalty term into the MT-estimators defined in \citet{ValdoraYohai2014} and further studied in \citet{agostinelli2019initial}. We focus on a general penalization term, and we discuss the behavior of several well known penalizations as elastic net, brige, sign, SCAD and MCP.

We give theoretic results regarding strong consistency, convergence rate, asymptotic normality and variable selection for the general class of GLMs as well as robustness properties. We also give a numerical algorithm which allows to efficiently compute the proposed estimators.

We study the performance of the proposed estimators in finite samples by a Monte Carlo study, for the case of Poisson response. This simulation study shows the good robustness properties of the proposed estimators. We also consider a real data set in which we seek to explain and predict the length of hospital stay of patients with  chronic obstructive pulmonary disease as primary disease category, using a large number of covariates. In this example, we see that both MT Lasso and MT Ridge give a better prediction for the majority of the data than their non-robust counterparts and that MT Lasso allows variable selection.
We also show how both MT Lasso and MT Ridge can be used for outlier detection.

To sum up, the proposed estimators constitute useful methods for the analysis of high-dimensional data, allowing to perform good predictions even in the presence of outliers, to perform robust variable selection and also to effectively detect outliers. 

\subsection*{Acknowledgments}
Claudio Agostinelli was partially funded by BaC INF-ACT S4 - BEHAVE-MOD PE00000007 PNRR M4C2 Inv. 1.3 - NextGenerationEU, CUP: I83C22001810007 and by the PRIN funding scheme of the Italian Ministry of University and Research (Grant No. P2022N5ZNP). The publication was produced with funding from the Italian Ministry of University and Research as part of the Call for Proposals for the scrolling of the final rankings of the PRIN 2022 - Project title ``MEMIMR: Measurement Errors and Missing Information in Meta-Regression'' - Project No. 2022FZY9PM - CUP C53C24000740006.

\clearpage

\appendix

\section{Proofs}\label{sec:proofs}
In this Appendix we provide auxiliary results and detailed proofs of the results presented in Section \ref{sec:asres}. Throughout this Appendix, {we assume $(y_1, \x_1),  \ldots (y_n, \x_n),$  and $(y,\x)$  follow a GLMs with parameter $\boldsymbol{\beta}_*$, link function $g$ and distribution function $F$}, that $\rho$ and $t$ are functions $\mathbb R \rightarrow \mathbb R$ and that $m$ is the function defined in \eqref{eq:m}. We use the notation introduced in Section \ref{sec:MT} to \ref{sec:asres}, furthermore to simplify the notation we assume without loss of generality that $a(\phi) \equiv 1$.

\subsection{Proofs of the consistency results in Section \ref{sec:asres}}
{The following Lemma states the Fisher consistency of MT-estimators and has already been proved in \citet{ValdoraYohai2014}.} We include it here for the sake of completeness.
\begin{lemma}[Fisher-Consistency]\label{lemma:FC}
Let $L(\boldsymbol{\beta})$ be the function defined in \eqref{eq:L} Under assumption A2, $L(\boldsymbol\beta)$ {has a unique minimum} at $\boldsymbol{\beta} = \boldsymbol{\beta}_*$. Therefore MT-estimators are Fisher consistent.
\end{lemma}
\begin{proof}
Let $y| \x \sim F(\mu_\x, \phi)$ and $\mu_\x = \exp(\x^\top \boldsymbol{\beta}_*)$, then 
\begin{equation*}
\E(\rho(t(y)-\mginv(\x^\top \boldsymbol{\beta}))) = \E\left[ \E(\rho(t(y)-\mginv(\x^\top \boldsymbol{\beta})) \mid \x)\right]
\end{equation*}  
The conditional expectation on the right is minimized in $\boldsymbol{\beta} = \boldsymbol{\beta}_*$ by definition of $m$, for all $\X$. Therefore, so is its expectation.
\end{proof}

The following lemma concerns the Vapnik-Chervonenkis (VC) dimension of a class of functions; see a definition in \citet{kosorok2008introduction}. {The proof is similar to Lemma S.2.1 in \citet{boente2020robust}; see also  Lemma 4.2.2 in  \citet{smucler2016estimadores}.
}
\begin{lemma}[VC-index] \label{lemma:vcsubgraph}
Assume A2, B1 and B3 hold. Then the class of functions
\begin{equation*}
	\mathcal{F} = \{f:\mathbb R\times \mathbb R^{p} \rightarrow \mathbb R, \,\, f(\vv, \w) = \rho(t(\vv) - \mginv(\beta_0 + \w \Beta)): \beta_0 \in \mathbb{R}, \ \Beta \in \mathbb{R}^p \} \ .
\end{equation*}  
is VC-subgraph with VC-index at most $2(p+3)-1$.
\end{lemma}
\begin{proof}
Consider the set of functions
\begin{equation*}
\mathcal{F}_1 = \{ f(\vv, \w) = \beta_0 + \w^\top \Beta: \beta_0 \in \mathbb{R}, \ \Beta \in \mathbb{R}^p \} \ .
\end{equation*}  
$\mathcal{F}_1$ is a subset of the vector space of all affine functions in $p+1$ variables, hence e.g. by Lemma 9.6 of \citet{kosorok2008introduction}, $\mathcal{F}_1$ is  VC-subgraph with VC-index at most $p+3$.  Using Lemma 9.9 (viii) of \citet{kosorok2008introduction}
we get that
\begin{equation*}
\mathcal{F}_2 = \{ f(\vv, \w) = \mginv(\beta_0 + \w^\top \Beta): \beta_0 \in \mathbb{R}, \ \Beta \in \mathbb{R}^p \}
\end{equation*}  
is also VC-subgraph with VC-index at most $p+3$ and, using item (v) of the same lemma we get  that the class of functions
\begin{equation*}
\mathcal{F}_3 = \{ f(\vv, \w) = t(\vv) - \mginv(\beta_0 + \w\top \Beta): \beta_0 \in \mathbb{R}, \ \Beta \in \mathbb{R}^p \}
\end{equation*}
is VC-subgraph with at most VC-index $p+3$. Consider the functions $\rho^+(s) = \rho(s) \mathbb{I}_{[0, \infty)}(s)$ and $\rho^-(s) = \rho(s) \mathbb{I}_{(-\infty,0)}(s)$ which are monotone by B1 and B3, and $\rho(s) = \max\{\rho^+(s), \rho^-(s)\}$. Hence, by appling Lemma 9.9 (viii) and (ii) of \citet{kosorok2008introduction} we have that both
\begin{align*}
\mathcal{F}^+ & = \{ f(\vv, \w) = \rho^+(t(\vv) - \mginv(\beta_0 + \w\top \Beta)): \beta_0 \in \mathbb{R}, \ \Beta \in \mathbb{R}^p \} \\
\mathcal{F}^- & = \{ f(\vv, \w) = \rho^-(t(\vv) - \mginv(\beta_0 + \w^\top \Beta))): \beta_0 \in \mathbb{R}, \ \Beta \in \mathbb{R}^p \}.
\end{align*}
are VC-subgraph and $\mathcal{F} = \mathcal{F}^+ \vee \mathcal{F}^-$ is VC-subgraph with VC-index at most $2(p+3) - 1$.
\end{proof}

The following four lemmas have been proved in \citet{ValdoraYohai2014}, as part of the proofs of Theorem 1 and Lemma 1. We give  statements and proofs here, for the sake of completeness. Let $B(\mathbf{t},\epsilon) = \left\{\mathbf s \in \mathbb R^{p+1}/||\mathbf s - \mathbf t||\leq \epsilon\right\}$.

\begin{lemma}\label{lemma:Ct}
	\begin{enumerate} 
		\item Let $\w$, $\mathbf t \in \mathbb R^{p+1}$ be such that $\w ^\top \mathbf t >0$, then  there exist $\epsilon>0$ and
		$C_{\mathbf t}$ a compact set in $\mathbb R^{p+1}$ such that for all $\mathbf{s} \in B(\mathbf{t}, \epsilon)$ and $\w \in C_{\mathbf{t}}$, 
		$\operatorname{sign}\left(\w^\top \mathbf{t}\right)=\operatorname{sign}\left(\w^\top \mathbf{s}\right)$ and  $\left|\w^\top \mathbf{s}\right|>\epsilon.$
		\item If $\x$ is a random vector such that $\mathbb P(\mathbf x^\top \mathbf t = 0 )<\delta$, then the set $C_\mathbf{t}$ can be chosen in such a way that $P(\mathbf x \in C_\mathbf t) >1-\delta$. 
	\end{enumerate}
\end{lemma}
\begin{proof}
	\begin{enumerate}
		\item Let $\varsigma$ be such that $\w ^\top \mathbf t > \varsigma$ and $K$ be such that $||\w|| \leq K$. Take $\epsilon=\varsigma/(2K)$ and $C_\mathbf{t} = \left\{\mathbf u \in \mathbb R^{p+1}:  |\mathbf u ^\top \mathbf t| \geq \varsigma, ||\mathbf u||\leq K \right\}$.
		\item Note that, since $\mathbb P(\mathbf x^\top \mathbf t = 0 )<\delta$, there exist $\xi$ and $\varsigma$ such that  $\mathbb P(\mathbf x^\top \mathbf t \geq \varsigma )>1-\delta-\xi$. Also note that there exists $K\geq 1$ such that $\mathbb P(||\mathbf x|| \leq K)>1-\xi$. Then $P\left(\left|\mathbf{x}^\top \mathbf{t}\right| \geq \varsigma, \|\mathbf{x}\| \leq K\right) \geq P\left(\left|\mathbf{x}^\top \mathbf{t}\right| \geq \varsigma\right)-P(\|\mathbf{x}\|>K)>1-\delta+\xi-\xi=1-\delta$.
	\end{enumerate}
\end{proof}
\begin{lemma}
	\label{lemma:phistar} Assume A1-A5 and B1-B6.
Then there exists a function $\Phi^{\ast}:\mathbb{R} \times \{-1, 0, 1\} \longrightarrow \mathbb{R}$ such that, for all $\mathbf{t}, \w \in \mathbb{R}^{p+1}$, $v\in\mathbb{R}$,
\begin{equation}\label{eq:phiestrella0}  
	\lim_{\zeta\rightarrow\infty}\rho(t(\vv) -\mginv(\zeta\w^\top   \mathbf{t}) ) = \Phi^{\ast}(\vv,\sign( \w^\top \mathbf{t})) 
\end{equation}
and if $\w^\top \mathbf{t} \neq0$ there exists a neighborhood of {$(\mathbf{t}, \w)$ where this convergence is uniform. }
\end{lemma}
\begin{proof}
Let $m_1 = \lim_{\mu \rightarrow \mu_1} m(\mu)$, $m_2 = \lim_{\mu\rightarrow \mu_2} m(\mu)$ and $m_0= \mginv(0)$, where $\mu_1$ and $\mu_2$ are the values defined in Section \ref{sec:MT}. These limits exist due to A2, though they may be $\infty$. Then \eqref{eq:phiestrella0} follows by taking 
	\begin{equation}
		\Phi^{\ast}(y,j) = \left\{ \begin{array} {l}
				\rho(t(\vv) - m_1) \text{ if } j=-1\\ 
						\rho(t(\vv) - m_0) \text{ if } j=0\\
				\rho(t(\vv) - m_2) \text{ if } j=1\\
		\end{array}\right. \label{eq:phistar}
	\end{equation} $\Phi^{\ast}(\vv,j)$ is well defined due to B2, where, if $m_i = \pm \infty$ we understand $\rho(t(\vv) - m_i)$ as $\lim_{m\rightarrow\pm\infty}\rho(t(\vv) - m) =1$ and it is continuous due to Lemma 3 in \citet{ValdoraYohai2014}.

Let $\epsilon>0$ and
 $C_{\mathbf t}$ be  those given in Lemma \ref{lemma:Ct}. We have that, for all $\mathbf s \in B(\mathbf{t}, \epsilon)$ and $\w \in C_{\mathbf t}$,
 $ \lim_{\zeta\rightarrow \infty} \rho(t(\vv) - \mginv(\zeta\w^\top   \mathbf{s}) )= \rho(t(\vv)-m_i)$ for $i=1$ or $2$. Using B4 and the fact that  $B(\mathbf{t}, \epsilon) \times C_\mathbf{t}$ is compact, we conclude that the convergence is uniform in $\mathbf{w}$ and $\mathbf{s}$; see for instance, Theorem 7.13 in \citet{rudin1953principles}.
\end{proof}

\begin{proof}[Proof of Lemma \ref{lemma:tau}]
 Note that assumption B3 together with \eqref{eq:phistar} imply that $\tau\geq 0$. 
By assumption B7 there exists $\delta>0$ such that $\inf _{||\mathbf t||=1} \mathbb P\left(\mathbf{t}^{\top} \x \neq \mathbf{0}\right)=\delta$.
Let $K_1, K_2 \in \mathbb R$ be such that $\mathbb P \left(\boldsymbol{\beta}_*^\top \x \in [K_1, K_2]\right) >1 - \delta/2$ and consider the event $V_{\mathbf t}=  \left\{ \mathbf t^{\top} \x \neq 0 , \boldsymbol{\beta}_*^{\top}\x \in [K_1, K_2]\right\}$. 
$$\mathbb P(V_{\mathbf{t}})\geq 
 \mathbb P(\mathbf t^{\top} \x \neq 0 ) - \mathbb P\left(\boldsymbol{\beta}_*^{\top}\x \notin [K_1, K_2]\right) \geq \delta - \delta/2 =\delta/2.$$ 
For each $i=1,2$, the function $C_i(\mu) = \mathbb E_{\mu}\left( \rho \left( t(y) - m_i \right)\right)-\mathbb E_{\mu}\left( \rho \left( t(y) - m(\mu) \right)\right)$ is positive
 and continuous by assumption B5. Then there exists $c_0>0$ such that, for each $i=1,2$,  $C_i(\mu) \geq c_0$ for all $\mu \in [g^{-1}(K_1), g^{-1}(K_2)]$. Denote $\mu_{\mathbf x} = g^{-1}\left(\boldsymbol{\beta}_*^{\top} \mathbf{x}\right)$. Then 
$$\begin{aligned}
	\mathbb E & \left(\Phi^*\left(y, \operatorname{sign}\left(\mathbf{t}^{\top}\x \right)\right)\right) - L(\boldsymbol{\beta}_*) \\
		& \geq \mathbb E\left(\mathbb E\left(\Phi^*\left(y,  \operatorname{sign}\left(\mathbf{t}^{\top} \x\right)\right)-\rho \left( t(y) - \mginv\left(\boldsymbol{\beta}_*^{\top} \mathbf{x}\right) \right)    \mid \mathbf{x} \right)  \right)
	\\
	& = \mathbb E\left(\mathbb E_{\mu_\mathbf{x}}\left(\Phi^*\left(y,  \operatorname{sign}\left(\mathbf{t}^{\top} \x\right)\right)-\rho \left( t(y) - m\left(\mu_{\mathbf x}\right) \right)    \mid \mathbf{x} \right)  \right)
	\\
	& \geq \mathbb E\left(C_i\left(\mu_\x \right) \mathbf{I}\left(\mathbf{x} \in V_{\mathbf{t}}\right) \right) \text{ for } i = 1 \text{ or } 2 \text{ since }  \mathbf{t}^{\top} \mathbf{x}\neq 0 \text{ for } \mathbf x \in V_{\mathbf{t}}\\ 
	& \geq c_0 \, \mathbb E\left( \mathbf{I}\left(\mathbf{x} \in V_{\mathbf{t}}\right)\right) \\	 &\geq c_0 \, \delta/2 > 0 .
\end{aligned}$$
 Since the bound is independent of $\mathbf t$, \eqref{eq:tau} follows.
\end{proof}

\begin{lemma} \label{lemma:liminf}
Assume conditions A1-A5 and B1-B6 and B8. Then, for all $\mathbf{t} \in S$, there exists $\varepsilon>0$ such that
\begin{equation}
\label{eq1.2}   
L(\boldsymbol{\beta}_*) < \E_{\boldsymbol{\beta}_{*}}\left(
\mathop{\underline{\lim}}_{\zeta\rightarrow\infty}\inf_{\mathbf{s}\in B(\mathbf{t},\varepsilon)}\rho(t(y) - \mginv(\zeta \x^\top  \boldsymbol{s}))\right) .
\end{equation}
\end{lemma}

\begin{proof} Because of B8, $\mathbb P(\x^\top \mathbf{t}=0) < 1-\tau$ for all $\mathbf{t} \in S$. Let $C_\mathbf t$ be the compact set given in Lemma \ref{lemma:Ct},then $P(C_\mathbf t) >\tau$. Given $\vv\in \mathbb R$ and $\w \in \mathbb R^{p+1}$, by Lemma \ref{lemma:phistar}, we have
\begin{align*}
\mathop{\underline{\lim}}_{\zeta\rightarrow\infty}\inf_{\mathbf{s}\in B(\mathbf{t},\varepsilon)}\rho(t(\vv) -\mginv(\zeta\w^\top  \boldsymbol{s}))
& \leq \mathop{\lim}_{\zeta\rightarrow\infty}\rho\left(t(\vv) -\mginv(\zeta\w^\top \mathbf{t})\right) \\
& = \Phi^{\ast}(\vv,\sign(\w^\top\mathbf{t})) . 
\end{align*}
Let us assume that the strict inequality holds for some point $\w \in C_{\mathbf{t}}$ and $\vv \in \mathbb{R}$, that is
\begin{equation*}
	\mathop{\underline{\lim}}_{\zeta\rightarrow\infty}
\inf_{\mathbf{s}\in B(\mathbf{t},\varepsilon)}\rho(t(\vv) -\mginv(\zeta\w^\top \boldsymbol{s})) < \Phi^{\ast}(\vv,\sign(\w^\top \mathbf{t})) ,
\end{equation*}
then, there exist $\varsigma>0$, a sequence of positive numbers $\zeta_{n} \rightarrow \infty$ and a sequence $\mathbf{s}_{n} \in B(\mathbf{t},\varepsilon)$ such that
\begin{equation}
\label{eq12}  
\rho(t(\vv) -\mginv(\zeta_n \w^\top  \boldsymbol{s}_n)) < \Phi^{\ast}(\vv,\sign(\w^\top \mathbf{t})) -\varsigma.
\end{equation}
We can assume without loss of generality that $\mathbf{s}_{n} \rightarrow \mathbf{s}_{0},$ where $\mathbf{s}_{0} \in B(\mathbf{t}, \varepsilon)$ and $\w^\top \mathbf{s}_{0} \neq 0$. Moreover the sign of $\w^\top \mathbf{s}_{n}$ is the same as the sign of $\w^\top \mathbf{s}_0$ and of $\w^\top \mathbf{t}$. By Lemma \ref{lemma:phistar}, we have
\begin{equation*}
	\lim_{n\rightarrow\infty}\rho(t(\vv) -\mginv (\zeta_n\w^\top  \boldsymbol{s}_n)) = \Phi^{\ast}(\vv,\sign(\w^\top \mathbf{s}_{0})) = \Phi^{\ast}(\vv,\sign(\w^\top \mathbf{t})).
\end{equation*}
contradicting \eqref{eq12}. Then
\begin{equation*}
\mathop{\underline{\lim}}_{\zeta\rightarrow\infty}\inf_{\mathbf{s}\in B(\mathbf{t},\varepsilon)}\rho(t(\vv) -\mginv( \zeta  \w^\top\boldsymbol{s})) = \Phi^{\ast}(\vv,\sign(\w^\top \mathbf{t})) .
\end{equation*}
Given a set $A,$ we denote $A^{c}$ its complement. Then since $P(C_{\mathbf{t}}^{c})<1-\tau$ and $\sup\Phi^{\ast}\leq 1$, because of Lemma \ref{lemma:phistar} we get 
\begin{align*}
\E_{\boldsymbol{\beta}_*}(\mathop{\underline{\lim}}_{\zeta \rightarrow\infty}\inf_{\mathbf{s}\in B(\mathbf{t},\varepsilon)} \rho(t(y) - \mginv( \zeta \x^\top \boldsymbol{s}))) & \geq \E_{\boldsymbol{\beta}_*}(\mathop{\underline{\lim}}_{\zeta \rightarrow\infty}\inf_{\mathbf{s}\in B(\mathbf{t},\varepsilon)} \rho(t(y) - \mginv(\zeta \x^\top \boldsymbol{s})\mathbb{I}_{C_{\mathbf{t}}}(\x))) \\
& = \E_{\boldsymbol{\beta}_*}(\Phi^{\ast}(y,\sign(\x^\top \mathbf{t})) \mathbb{I}_{C_{\mathbf{t}}}(\x)) \\ 
& = \E_{\boldsymbol{\beta}_*}\left(\Phi^{\ast}(y,\sign(\x^\top \mathbf{t})) - {\Phi}^{\ast}(y,\sign(\x^\top \mathbf{t})) \mathbb{I}_{C_{\mathbf{t}}^{c}}(\mathbf{x})\right) \\
& \geq \E_{\boldsymbol{\beta}_*}(\Phi^{\ast}(y,\sign(\x^\top \mathbf{t}))) - \mathbb P\left(C_{\mathbf{t}}^{c}\right)\\ 
& > \E_{\boldsymbol{\beta}_*}(\Phi^{\ast}(y,\sign(\x^\top \mathbf{t}))) - \tau \\
& \geq L(\boldsymbol{\beta}_*)
\end{align*}
This concludes the proof.
\end{proof}

The following is a general consistency theorem valid for a large class of estimators that includes penalized  MT-estimators. It has been stated and proved in Theorem S.2.1 in the Supplement to \citet{bianco2022penalized} for the particular case of logistic regression. However, the same proof is valid for general GLMs. We recall that in what follows we assume we have a sequence $\{\lambda_n,\alpha_n\}_{n=1}^\infty$ indexed by the sample size $n$ of parameters for the penalty term; such a sequence might be random as well.
\begin{theorem}\label{teo:gralcons}
Let $\hat{\boldsymbol{\beta}}_n$ be an estimator defined as
	 \begin{equation}\label{eq:gral_est}
	 	 \hat{\boldsymbol{\beta}}_n=\underset{\boldsymbol{\beta} \in \mathbb{R}^p}{\operatorname{argmin}} \frac{1}{n} \sum_{i=1}^n \Phi\left(y_i, \mathbf{x}_i^{\top} \boldsymbol{\beta}\right) + \penn{\lambda_n}{(\boldsymbol{\beta})}{\alpha_n},
   \end{equation}	 
   where $\Phi:\mathbb R^2 \longrightarrow \mathbb R$ and  $\penn{\lambda_n}{}{\alpha_n}: \mathbb R^{p+1} \longrightarrow \mathbb R$. Assume $L(\boldsymbol{\beta})=\mathbb{E} \left(\Phi(y, \mathrm{x}^\top \boldsymbol{\beta})\right)$ has a unique minimum at $\boldsymbol{\beta}=\boldsymbol{\beta}_*$ and that $\penn{\lambda_n}{\boldsymbol{\beta}_*}{\alpha_n} \xrightarrow{a.s.} 0$ when $n \rightarrow \infty$. Furthermore, assume that, for any $\epsilon>0$,
	\begin{equation}\label{eq:desigL}
		\inf _{\left\|\boldsymbol{\beta}-\boldsymbol{\beta}_*\right\|>\epsilon} L(\boldsymbol{\beta})>L\left(\boldsymbol{\beta}_*\right)
	\end{equation}	and that the following uniform Law of Large Numbers holds
	\begin{equation} \label{eq:unifLLN}
		\mathbb{P}\left(\lim _{n \rightarrow \infty} \sup _{\boldsymbol{\beta} \in \mathbb{R}^p}\left|\frac{1}{n} \sum_{i=1}^n \Phi\left(y_i, \mathbf{x}_i^{\top} \boldsymbol{\beta}\right) -\mathbb{E} \left(\Phi(y, \mathbf{x}^\top \boldsymbol{\beta})\right) \right|=0\right)=1 .
	\end{equation}
	Then, $\hat{\boldsymbol{\beta}}_n$ is strongly consistent for $\boldsymbol{\beta}_*$.
\end{theorem}
\begin{proof}
The proof is the same as the proof of {Theorem S.2.1} in \citet{bianco2022penalized}, with  $L_n(\boldsymbol{\beta}) =  \frac{1}{n} \sum_{i=1}^n \Phi\left(y_i, \mathbf{x}_i^{\top} \boldsymbol{\beta}\right),  \,\,\, L(\boldsymbol{\beta}) = \mathbb{E}\left(\Phi\left(y, \mathbf{x}^{\top} \boldsymbol{\beta}\right)\right)$ and $w(\x) = 1$.
\end{proof}

\begin{proof}[Proof of Theorem \ref{teo:consistency} ]
It is enough to show that the conditions in Theorem \ref{teo:gralcons} hold, with $\Phi(y, \x^{\top}\boldsymbol{\beta})=\rho(t(y) - m g^{-1} \left(\x^{\top}\boldsymbol{\beta}\right))$. First note that, because of Lemma \ref{lemma:FC}, $L(\boldsymbol{\beta})$ has a unique minimum at $\boldsymbol{\beta_*}$. Second, we prove \eqref{eq:unifLLN} by applying Corollary 3.12 in \citet{geer2000empirical}. To apply this corollary, note that $\left\vert \Phi(y, \x^{\top}\boldsymbol{\beta}) \right\vert = \left\vert \rho(t(y) - m g^{-1} \left(\x^{\top}\boldsymbol{\beta}\right)) \right\vert$ is uniformly bounded and that, by Lemma \ref{lemma:vcsubgraph}, the family 
$$\mathcal F = \left\{ f_{\boldsymbol{\beta}}(\vv, \w)=\Phi(\vv, \w^\top\boldsymbol{\beta}), \boldsymbol{\beta} \in \mathbb R^{p+1} \right\}$$ is VC-subgraph with envelope $F=1$. The mentioned corollary then implies that the family $\mathcal{F}$ satisfies the Uniform Law of Large Numbers and \eqref{eq:unifLLN} follows. It remains to prove \eqref{eq:desigL} holds for all $\epsilon>0$. Assume there exists $\epsilon>0$ such that \eqref{eq:desigL} does not hold. Then there exists a sequence $\left(\boldsymbol{\beta}_n\right)_n$ such that $\left\|\boldsymbol{\beta}_n-\boldsymbol{\beta}_*\right\| >\epsilon$ and  $\lim_{n_\rightarrow\infty} L(\boldsymbol{\beta}_n) = l  \leq L\left(\boldsymbol{\beta}_*\right)$. 
Assume first $\boldsymbol{\beta}_n$ is bounded, then it has a subsequence $\boldsymbol{\beta}_{n_k}$ such that $\lim_{n_\rightarrow\infty} L(\boldsymbol{\beta}_n) = L\left(\boldsymbol{\beta}_{**}\right)$ for some $\boldsymbol{\beta}_{**}$ with $\left\|\boldsymbol{\beta}_{**}-\boldsymbol{\beta}_*\right\|>\epsilon$. Since $L\left(\boldsymbol{\beta}_{*}\right) < L\left(\boldsymbol{\beta}_{**}\right)$ by Lemma \ref{lemma:FC} we arrive at a contradiction.
 This means that the sequence $\boldsymbol{\beta}_n$ must be unbounded.
Let $\boldsymbol{\gamma}_n = \boldsymbol{\beta}_n/||\boldsymbol{\beta}_n||$. We can assume without loss of generality that $ \lim_{n\rightarrow \infty}\boldsymbol{\gamma}_n =\boldsymbol{\gamma}_*$ with $\boldsymbol{\gamma}_* = 1$ and that $ \lim_{n\rightarrow \infty}||\boldsymbol{\beta}_n||$. 
By Lemma \ref{lemma:liminf}, there exists $\epsilon>0$ such that 
\begin{equation}\label{eq:lbetastar}
L(\boldsymbol{\beta}_*) < \E_{\boldsymbol{\beta}_{*}}(\mathop{\underline{\lim}}_{a\rightarrow\infty}\inf_{\boldsymbol{\gamma}\in B(\boldsymbol{\gamma}_*,\varepsilon)}\Phi(y, \x^\top a \boldsymbol{\gamma})) .	
\end{equation}	
For each $M>0$ we can choose $n_0 \in \mathbb N$ such that, for all $n\geq n_0$, $\left\vert\left\vert \boldsymbol{\beta}_n \right\vert\right\vert >M$  and $\boldsymbol{\gamma}_n \in B(\boldsymbol{\gamma}_*,\varepsilon)$. Then 
$$\Phi(y, \x^\top \boldsymbol{\beta}_n) = \Phi(y, M\x^\top \boldsymbol{\gamma}_n) \geq    \inf_{a > M}  \inf_{ \boldsymbol{\gamma} \in B(\boldsymbol{\gamma}_*,\varepsilon) }  \Phi(y, a \x^\top \boldsymbol{\gamma}). $$

$$ \lim_{n\rightarrow\infty} \Phi(y, \x^\top \boldsymbol{\beta}_n)  \geq  \mathop{\underline{\lim}}_{a\rightarrow\infty}    \inf_{\boldsymbol{\gamma} \in B(\boldsymbol{\gamma}_*,\varepsilon)}  \Phi(y, a \x^\top \boldsymbol{\gamma}) $$
By the dominated convergence theorem and \eqref{eq:lbetastar},
$$ \lim_{n\rightarrow\infty} L(\boldsymbol{\beta}_n) =  \E \left( \lim_{n\rightarrow\infty}\Phi(y, \x^\top \boldsymbol{\beta}_n) \right)   \geq \E \left( \mathop{\underline{\lim}}_{a\rightarrow\infty}  \inf_{\boldsymbol{\gamma} \in B(\boldsymbol{\gamma}_*,\varepsilon)}   \Phi(y, a \x^\top \boldsymbol{\gamma}) \right) > L(\boldsymbol{\beta_*}),$$ which is a contradiction. This implies \eqref{eq:desigL} and the theorem follows.
\end{proof}

The following lemma is a direct consequence of Lemma 4.2 in \citet{yohai1985high}. We include it here, together with its proof, for the sake of completeness.
\begin{lemma}
\label{lemma:victor}
Let $\tilde{\boldsymbol\beta}_{n}$ be a sequence of estimators such that $\tilde{\boldsymbol\beta}_{n} \rightarrow \boldsymbol{\beta}_*$ a.s.. Suppose {C1, C2 and C4} hold. Then
\begin{equation*}
\mathop{\lim _{n \rightarrow \infty}} \mathbf{A}_n(\tilde{\boldsymbol{\beta}}_{n}) = \mathbf{A} \quad \text {a.s.} ,
\end{equation*}
where $\mathbf{A}_n$ and $\mathbf{A}$ are given in \eqref{eq:AB}.
\end{lemma}
\begin{proof}
First note that it is enough to prove that, for all $\epsilon>0$
\begin{equation*}
\mathop{\overline{\lim}_{n\rightarrow \infty}} \mathbf{A}_n(\tilde{\boldsymbol{\beta}}_{n}) < \mathbf{A} + \epsilon \qquad \text{and} \qquad
\mathop{\underline{\lim}_{n\rightarrow \infty}} \mathbf{A}_n(\tilde{\boldsymbol{\beta}}_{n}) > \mathbf{A} - \epsilon
\end{equation*}
Then it is enough to show that for all $\epsilon>0$ there exists $\delta>0$ such that
\begin{equation}
\label{eq:limsups}  
\mathop{\overline{\lim}_{n\rightarrow \infty}}  \mathop{\sup_{||\boldsymbol{\beta} - \boldsymbol{\beta}_*||<\delta}} \mathbf{A}_n(\boldsymbol{\beta}) < \mathbf{A} + \epsilon \qquad \text{and} \qquad \mathop{\underline{\lim}_{n\rightarrow \infty}} \mathop{\inf_{||\boldsymbol{\beta} - \boldsymbol{\beta}_*||<\delta}} \mathbf{A}_n(\boldsymbol{\beta}) < \mathbf{A} - \epsilon 
\end{equation}
To prove the first inequality, note that, by the dominated convergence theorem and the continuity of $\mathbf{J}$, there exists $\delta$ such that
\begin{equation*}
\E\left(\mathop{\sup_{||\boldsymbol{\beta} - \boldsymbol{\beta}_*||<\delta}} \mathbf{J}(y, \x, \boldsymbol{\beta})\right) <  \E ( \mathbf{J}(y, \x, \boldsymbol{\beta}_*)) + \epsilon .
\end{equation*}
Using the Law of Large Numbers we get the first inequality in \eqref{eq:limsups}. The second inequality is proved similarly.
\end{proof}
\begin{lemma}\label{lemma:limits} Assume conditions A1-A5, B1-B7 and C1-C4. Then, 
\begin{enumerate}
\item[(a)] $\mathop{\lim}_{\boldsymbol{\beta}\rightarrow\boldsymbol{\beta}_*} \E(\mathbf{J}(y,\x,\boldsymbol{\beta})) = \E(\mathbf{J}(y,\x,\boldsymbol{\beta}_*))$
\item[(b)] Let $\eta$ be the constant given in assumption C4. Then,  for all $\delta \in (0, \eta)$,
\begin{equation}
\sup _{\| \boldsymbol{\beta}-\boldsymbol{\beta}_{*} \|<\delta} \frac{1}{n} \sum_{i=1}^n \mathbf{J}(y_i,\x_i,\boldsymbol{\beta})- \E ( \mathbf{J}(y,\x,\boldsymbol{\beta})) \xrightarrow{a.s.} 0.
\end{equation}
\end{enumerate}
\end{lemma}	
\begin{proof}
  Part (a) follows from the dominated convergence theorem and conditions C1, C2 and C4. To prove part (b) we follow the lines of the proof of Lemma 1 in \citet{bianco2002asymptotic}. We prove the result componentwise. Fix a component of $J^{j,l}$ of $\mathbf{J}$ and $\delta \in (0,\eta)$. From Theorem 2 in Chapter 2 in \citet{pollard1984convergence}, it is enough to show that for all $\epsilon > 0$, there exists a finite class of functions $\mathcal F_\epsilon$ such that for all $f$ in the class
\begin{equation*}
\mathcal{F} =\left\{f_{\boldsymbol{\beta}}(\vv,\w) = J^{k,l}(\vv,\w,\boldsymbol{\beta}), ||\boldsymbol{\beta}-\boldsymbol{\beta}_*|| \leq \delta \right\} ,
\end{equation*}
there exist  functions $f_{\epsilon}^-$, $f_{\epsilon}^+$ $\in \mathcal F_\epsilon$ such that
\begin{equation}
\label{eq:pollard}
f_{\epsilon}^- \leq f \leq f_{\epsilon}^+ \quad \text { and } \quad \E(f^+_{\epsilon}(y,\x)-f^-_{\epsilon}(y,\x)) < \epsilon .
\end{equation}
Let $\mathcal{A}_K=\left\{(\vv, \w)\in \mathbb{R}^{p+2} ||\w|| \leq K, |\vv| \leq K \right\}$, then, for all $\vv \in\mathbb{R}, \w\in \mathbb{R}^{p+1}$,  $\mathop{\lim}_{K\rightarrow\infty}\mathbb{I}_{\mathcal{A}_K^c}(\vv,\w) =0$, so
\begin{equation}
	\label{eq:sup}
	\mathop{\lim}_{K\rightarrow\infty}\mathop{\sup}_{ ||\boldsymbol{\beta}-\boldsymbol{\beta}_*||\leq \delta}\left\vert J^{k,l}(\vv,\w,\boldsymbol{\beta})\right\vert\mathbb{I}_{\mathcal{A}_K^c}(\vv,\w) =0 .
\end{equation} 
Because of C4 and the dominated convergence theorem, there exists $K_0$, such that for all $K\geq K_0$ and $\delta<\eta$,
\begin{equation}
\label{eq:expsup}
\E\left(\mathop{\sup}_{ ||\boldsymbol{\beta}-\boldsymbol{\beta}_*||\leq \delta}\left\vert J^{k,l}(y,\x,\boldsymbol{\beta})\right\vert\mathbb{I}_{\mathcal{A}_K^c}(y,\x)\right) < \epsilon/8 .
\end{equation} 
Because of C1, C2 and Lemmas 3 and 5 in \citet{ValdoraYohai2014},  $J^{k,l}$ is uniformly continuous in $\mathcal A_K \times \{||\boldsymbol{\beta}-\boldsymbol{\beta}_*||\leq \delta \}$, then there exists $\xi$ such that
\begin{equation}
\label{eq:unifcont}
|J^{k,l}(\vv_1,\w_1,\Beta) - J^{k,l}(\vv_2,\w_2,\boldsymbol{\beta}_2)| < \epsilon/4
\end{equation}
for all $(\vv_1, \w_1), (\vv_2, \w_2) \in \mathcal A_K$, $\Beta, \boldsymbol{\beta}_2 \in \{||\boldsymbol{\beta}-\boldsymbol{\beta}_*||\leq \delta \}$ such that $||\Beta-\boldsymbol{\beta}_2|| \leq \xi$, $||\w_1 -\w_2|| < \xi$ and $|\vv_1 -\vv_2| < \xi$. Since $\left\{||\boldsymbol{\beta}-\boldsymbol{\beta}_*|| \leq \delta \right\}$ is compact, there exist $B_1, \ldots, B_N$ balls with radius smaller than $\xi$  and centers at certain points $\boldsymbol{\beta}_1, \ldots, \boldsymbol{\beta}_N \in \{||\boldsymbol{\beta}-\boldsymbol{\beta}_*||\leq \delta \}$ such that $\left\{||\boldsymbol{\beta}-\boldsymbol{\beta}_*||\leq \delta \right\} \subset \bigcup_{i=1}^N B_i$. The class $\mathcal{F}_\epsilon$ is then the class of functions
\begin{equation*}
f^{\pm}_{\epsilon, i}(\vv, \w) = J^{k,l}(\vv, \w, \boldsymbol{\beta}_i) \pm \left( \frac{\epsilon}{4} + 2 \mathop{\sup}_{ ||\boldsymbol{\beta}-\boldsymbol{\beta}_*||\leq \delta}\left\vert J^{k,l}(\vv,\w,\boldsymbol{\beta}) \right\vert \mathbb{I}_{\mathcal{A}_K^c}(\vv, \w) \right) ,
\end{equation*}  
for $i \in 1, \ldots, N$. To show the first inequality in \eqref{eq:pollard}, let $f_{\boldsymbol{\beta}}(\vv, \w) = J^{k,l}(\vv, \w,\boldsymbol{\beta}) \in \mathcal{F}$. If $(\vv, \w) \in \mathcal{A}_K$, then there exists $i_0 \in \{1, \ldots, N\}$ such that $||\boldsymbol{\beta} - \boldsymbol{\beta}_{i_0}|| \leq \xi$ and
\begin{equation*}
|J^{k,l}(\vv,\w,\boldsymbol{\beta}) - J^{k,l}(\vv,\w,\boldsymbol{\beta}_{i_0})| < \epsilon/4 .
\end{equation*}  
Then
\begin{equation*}
f^-_{\epsilon, i_0}(\vv, \w) = J^{k,l}(\vv,\w,\boldsymbol{\beta}_{i_0}) - \epsilon/4 < J^{k,l}(\vv,\w,\boldsymbol{\beta}) < J^{k,l}(\vv,\w,\boldsymbol{\beta}_{i_0}) + \epsilon/4 = f^+_{\epsilon, i_0}(\vv, \w) .
\end{equation*}  
If $(\vv, \w) \in \mathcal{A}_K^c$, because of the triangular inequality,
\begin{equation*}
|J^{k,l}(\vv,\w,\boldsymbol{\beta}) - J^{k,l}(\vv,\w,\boldsymbol{\beta}_{i_0})| \leq 2 \sup_{||\boldsymbol{\beta}-\boldsymbol{\beta}_*||<\delta} |J^{k,l}(\vv,\w,\boldsymbol{\beta})| .
\end{equation*}
Then
\begin{equation*}
J^{k,l}(\vv,\w,\boldsymbol{\beta}) \leq J^{k,l}(\vv,\w,\boldsymbol{\beta}_{i_0})+  2 \mathop{\sup}_{ ||\boldsymbol{\beta}-\boldsymbol{\beta}_*||\leq \delta}\left\vert J^{k,l}(\vv,\w,\boldsymbol{\beta}) \right\vert \leq f^{+}_{\epsilon, i_0}(\vv, \w)
\end{equation*}  
and
\begin{equation*}
J^{k,l}(\vv,\w,\boldsymbol{\beta}) \geq J^{k,l}(\vv,\w,\boldsymbol{\beta}_{i_0})- 2 \mathop{\sup}_{ ||\boldsymbol{\beta}-\boldsymbol{\beta}_*||\leq \delta}\left\vert J^{k,l}(\vv,\w,\boldsymbol{\beta}) \right\vert \geq f^{-}_{\epsilon, i_0}(\vv, \w) .
\end{equation*}  
To show the second inequality in \eqref{eq:pollard}, note that, by definition of $f^\pm_{\epsilon, i}$ and \eqref{eq:expsup},
\begin{equation*}
\E( f^+_{\epsilon, i_0}(y, \x) - f^-_{\epsilon, i_0}(y, \x)) = \frac{\epsilon}{2} + 4 \E\left( \mathop{\sup}_{ ||\boldsymbol{\beta}-\boldsymbol{\beta}_*||\leq \delta}\left\vert J^{k,l}(y,\x,\boldsymbol{\beta})\right\vert \mathbb{I}_{\mathcal{A}_K^c}(y, \x)\right) < \epsilon .
\end{equation*}
\end{proof}

\begin{proof}[Proof of Theorem \ref{teo:order}]
Let $W_n(\boldsymbol{\beta}) = L_n(\boldsymbol{\beta}) + \penn{\lambda_{n}}{\boldsymbol{\beta}}{\alpha_n}$ and $\mathbf{A}$ and $\mathbf{A}_n(\boldsymbol{\beta})$ as defined in \eqref{eq:AB}. Because of the definition of $\hat{\boldsymbol{\beta}}_n$, $W_n(\hat{\boldsymbol{\beta}}_n) - W_n(\boldsymbol{\beta}_*) \leq 0$. 
Using a Taylor expansion of $L_n(\hat{\boldsymbol{\beta}}_n)$ about $\boldsymbol{\beta}_*$, we obtain that for a certain point $\boldsymbol{\xi}_{n}=\boldsymbol{\beta}_*+\tau_{n}(\hat{\boldsymbol{\beta}}_{n}-\boldsymbol{\beta}_*)$ with $\tau_{n} \in[0,1]$, 
\begin{align*}
0 & \geq \frac{1}{n}\sum_{i=1}^{n}\ \boldsymbol{\Psi}(y_{i},\x_{i},\boldsymbol{\beta}_{*})^\top (\hat{\boldsymbol{\beta}}_{n}-\boldsymbol{\beta}_*) + (\hat{\boldsymbol{\beta}}_{n}-\boldsymbol{\beta}_*) ^\top \mathbf{A}_n(\boldsymbol{\xi}_n)(\hat{\boldsymbol{\beta}}_{n}-\boldsymbol{\beta}_*) \\
& + \penn{\lambda_{n}}{\hat{\boldsymbol{\beta}}_{n}}{\alpha_n} - \penn{\lambda_{n}}{\boldsymbol{\beta}_{*}}{\alpha_n} \\
& = \frac{1}{n}\sum_{i=1}^{n}\ \boldsymbol{\Psi}(y_{i},\x_{i},\boldsymbol{\beta}_{*})^{\top} (\hat{\boldsymbol{\beta}}_{n}-\boldsymbol{\beta}_*) + (\hat{\boldsymbol{\beta}}_{n}-\boldsymbol{\beta}_*)^\top (\mathbf{A}_n(\boldsymbol{\xi}_n)- \mathbf{A}) (\hat{\boldsymbol{\beta}}_{n}-\boldsymbol{\beta}_*) \\
& + (\hat{\boldsymbol{\beta}}_{n}-\boldsymbol{\beta}_*)^\top \mathbf{A}(\hat{\boldsymbol{\beta}}_{n}-\boldsymbol{\beta}_*)+ \penn{\lambda_{n}}{\hat{\boldsymbol{\beta}}_{n}}{\alpha_n} - \penn{\lambda_{n}}{\boldsymbol{\beta}_{*}}{\alpha_n} \\
& \geq \frac{1}{n}\sum_{i=1}^{n}\boldsymbol{\Psi}(y_{i},\x_{i},\boldsymbol{\beta}_{*})^{\top} (\hat{\boldsymbol{\beta}}_{n}-\boldsymbol{\beta}_*) + (\hat{\boldsymbol{\beta}}_{n}-\boldsymbol{\beta}_*)^\top (\mathbf{A}_n(\boldsymbol{\xi}_n)- \mathbf{A})(\hat{\boldsymbol{\beta}}_{n}-\boldsymbol{\beta}_*) \\
& + \zeta_1||\hat{\boldsymbol{\beta}}_{n}-\boldsymbol{\beta}_*||^2 + \{ \penn{\lambda_{n}}{\hat{\boldsymbol{\beta}}_{n}}{\alpha} - \penn{\lambda_{n}}{\boldsymbol{\beta}_{*}}{\alpha} \} ,
\end{align*}
where $\zeta_1$ is the smallest eigenvalue of $\mathbf{A}$ which we know is positive by C4. Using that the sum of the first two terms is greater than or equal to minus its absolute value and applying the triangle inequality we get that the complete sum is 
\begin{align*}
& \geq - \frac{1}{n} \sum_{i=1}^{n} ||\boldsymbol{\Psi}(y_{i},\x_{i},\boldsymbol{\beta}_{*})|| ||\hat{\boldsymbol{\beta}}_{n}-\boldsymbol{\beta}_*|| - ||\hat{\boldsymbol{\beta}}_{n}-\boldsymbol{\beta}_*||^2 ||\mathbf{A}_n(\boldsymbol{\xi}_n) - \mathbf{A}|| \\
& + \zeta_1||\hat{\boldsymbol{\beta}}_{n}-\boldsymbol{\beta}_*||^2 + \{ \penn{\lambda_{n}}{\hat{\boldsymbol{\beta}}_{n}}{\alpha} - \penn{\lambda_{n}}{\boldsymbol{\beta}_{*}}{\alpha} \} \ .
\end{align*}
Let $\epsilon>0$. To bound the first term note that the Central Limit Theorem and Lemma \ref{lemma:FC}  imply that
\begin{equation*}
\sqrt{n} \frac{1}{n}\sum_{i=1}^{n} \boldsymbol{\Psi}(y_{i},\x_{i},\boldsymbol{\beta}_{*}) = O_{\mathbb P}(1)
\end{equation*}
Therefore there exists a constant $M_{1}$ such that, for all $n$, $P(\mathcal{C}_{n}) > 1-\varepsilon / 4$, where $\mathcal{C}_{n} = \left\{\left\|\sqrt{n} \frac{1}{n}\sum_{i=1}^{n} \boldsymbol{\Psi}(y_{i},\x_{i},\boldsymbol{\beta}_{*}) \right\| < M_{1} \right\}$.

To bound the second term, note that, since $\hat{\boldsymbol{\beta}}_{n} \xrightarrow{p} \boldsymbol{\beta}_{*}$, from Lemma \ref{lemma:limits}, we have that $\mathbf{A}_{n}(\hat{\boldsymbol{\beta}}_{n}) \xrightarrow{p} \mathbf{A}$, so there exists $n_{1}$ such that for every $n \geq n_{1}, P (\mathcal{B}_{n}) > 1-\varepsilon/4$, where $\mathcal{B}_{n} = \left\{\left\|\mathbf{A}_{n}(\boldsymbol{\xi}_n)-\mathbf{A}\right\|<\zeta_{1} / 2\right\}$. 

We then have:
\begin{equation}\label{eq:ineqcons}
0 \geq-\left\|\hat{\boldsymbol{\beta}}_{n}-\boldsymbol{\beta}_*\right\| \frac{1}{\sqrt{n}} M_{1}+\frac{\zeta_{1}}{2}\left\|\hat{\boldsymbol{\beta}}_{n}-\boldsymbol{\beta}_*\right\|^{2} + \{ \penn{\lambda_{n}}{\hat{\boldsymbol{\beta}}_{n}}{\alpha} - \penn{\lambda_{n}}{\boldsymbol{\beta}_{*}}{\alpha} \}
\end{equation}

Finally, to bound the third term, note that since P1, the function $\penn{\lambda}{\cdot}{\alpha}$ is Lipschitz with constant $K$ independent of $\alpha$ and $\lambda$, in a neighbourhood of $\boldsymbol{\beta}_{*}$ and define the event
\begin{equation*}
\mathcal{D}_{n}=\left\{\mathop{\sup_{\alpha\in [0,1]}} | \penn{1}{\hat{\boldsymbol{\beta}}_{n}}{\alpha} - \penn{1}{\boldsymbol{\beta}_{*}}{\alpha} | \leq K \left\|\hat{\boldsymbol{\beta}}_{n}-\boldsymbol{\beta}_{*} \right\| \right\} .
\end{equation*}  
Then, there exists $n_{2} \in \mathbb{N}$ such that for $n \geq n_{2}$, $P(\mathcal{D}_{n}) \geq 1-\epsilon/2$. Hence, for $n \geq \max \left\{n_{1}, n_{2}\right\}$ we have that $P(\mathcal{B}_{n} \cap \mathcal{C}_{n} \cap \mathcal{D}_{n})>1-\epsilon$. Besides, in $\mathcal{B}_{n} \cap \mathcal{C}_{n} \cap \mathcal{D}_{n}$ we have that
\begin{equation*}
0 \geq-\left\|\hat{\boldsymbol{\beta}}_{n}-\boldsymbol{\beta}_*\right\| \frac{1}{\sqrt{n}} M_{1}+\frac{\zeta_{1}}{2}\left\|\hat{\boldsymbol{\beta}}_{n}-\boldsymbol{\beta}_*\right\|^{2}-K \lambda_{n}\left\|\hat{\boldsymbol{\beta}}_{n}-\boldsymbol{\beta}_*\right\|
\end{equation*}
which implies
\begin{equation*}
\left\|\hat{\boldsymbol{\beta}}_{n}-\boldsymbol{\beta}_{*}\right\| \leq 2\left(\lambda_{n}+\frac{1}{\sqrt{n}}\right) \frac{M_{1}+K}{\zeta_{1}}.
\end{equation*}
Hence, $\left\|\hat{\boldsymbol{\beta}}_{n}-\boldsymbol{\beta}_{*}\right\|=O_{\mathbb P}(\lambda_{n}+1 / \sqrt{n})$, which completes the proof.
\end{proof}

\subsection{Proofs of the robustness results in Section \ref{sec:asres}}

\begin{proof}[Proof of Theorem \ref{teo:finitebreakdown}]
We follow \citet{smucler2017robust}. Given a sample size $n$ assume we can construct a sequence of samples $\{\{(y_{Ni}, \x_{Ni}): 1 \le i \le n \}\}_{N=1}^\infty$ such that for the corresponding sequence $\{\hat{\boldsymbol{\beta}}_N\}_{N=1}^\infty$ of penalized MT-estimators we have $\lim_{N \rightarrow \infty} \|\hat{\boldsymbol{\beta}}_N \| = \infty$. Since B1 and B2 hold, for sufficiently large $N$, by assumption P3 on $\lambda_n$, we have $\penn{\lambda_n}{\hat{\boldsymbol{\beta}}_N}{\alpha} > a + \penn{\lambda_n}{\underline{\boldsymbol{\beta}}}{\alpha}$ that implies $\hat{L}(\hat{\boldsymbol{\beta}}_N) > \hat{L}(\underline{\boldsymbol{\beta}})$ which contradicts the definition of $\hat{\boldsymbol{\beta}}_N$. The proof is similar under P3'.
\end{proof}

\begin{proof}[Proof of Theorem \ref{teo:stability}]
We follow the proof of Theorem 1 in \citet{avella2018robust}. Let $\boldsymbol{\beta}_\varepsilon$ be the solution of the optimization problem in \eqref{equ:stability} and $\boldsymbol{\beta}_0$ be the solution of the same problem when $\varepsilon = 0$. Let $\boldsymbol{\Psi}(y, \x, \boldsymbol{\beta})$ be the derivative of $\rho((t(y)-\mginv(\x^\top\boldsymbol{\beta}))/\sqrt{a(\phi)})$ with respect to $\boldsymbol{\beta}$ and $\mathbf{J}(y,\x,\boldsymbol{\beta})$ be the Hessian matrix as in C4. Let $J = \{j: \boldsymbol{\beta}_{\varepsilon,j} \neq 0\}$, then the Karush-Kuhn-Tucker conditions for the $|J|$ equations corresponding to the indices in $J$ must satisfy the following estimating equations, where we intend that any quantity that follows is restricted to indices in $J$,
\begin{equation*}
\int \int \boldsymbol{\Psi}(y, \x, \boldsymbol{\beta}_\varepsilon) \ dF_\varepsilon(y | \x) \ dH(\x) + \dot{P}_{\lambda,\alpha}(\boldsymbol{\beta}_\varepsilon) = \boldsymbol{0} ,
\end{equation*}
where $\dot{P}_{\lambda,\alpha}(\boldsymbol{\beta})$ is the gradient of $\penn{\lambda}{\cdot}{\alpha}$ at $\boldsymbol{\beta}$. Consider a first order Taylor expansion of $\boldsymbol{\Psi}(y, \x, \boldsymbol{\beta}_\varepsilon)$ around $\boldsymbol{\beta}_0$ we obtain
\begin{equation*}
\boldsymbol{\Psi}(y, \x, \boldsymbol{\beta}_\varepsilon) = \boldsymbol{\Psi}(y, \x, \boldsymbol{\beta}_0) + \mathbf{J}(y, \x, \tilde{\boldsymbol{\beta}}_\varepsilon) (\boldsymbol{\beta}_\varepsilon - \boldsymbol{\beta}_0)
\end{equation*}
with $\tilde{\boldsymbol{\beta}}_\varepsilon = \boldsymbol{\beta}_0 + \tau (\boldsymbol{\beta}_\varepsilon - \boldsymbol{\beta}_0)$ for some $\tau \in [0,1]$. Similarly, since P2, we have
\begin{equation*}
\dot{P}_{\lambda,\alpha}(\boldsymbol{\beta}_\varepsilon) = \dot{P}_{\lambda,\alpha}(\boldsymbol{\beta}_0) + \ddot{P}_{\lambda,\alpha}(\check{\boldsymbol{\beta}}_\varepsilon) (\boldsymbol{\beta}_\varepsilon - \boldsymbol{\beta}_0)
\end{equation*}
with $\check{\boldsymbol{\beta}}_\varepsilon = \boldsymbol{\beta}_0 + \tau (\boldsymbol{\beta}_\varepsilon - \boldsymbol{\beta}_0)$ for some $\tau \in [0,1]$, where $\ddot{P}_{\lambda,\alpha}(\cdot)$ is the Hessian matrix. Hence, we have
\begin{align*}
0 & = \varepsilon \E_{\boldsymbol{X}} \left( \int \boldsymbol{\Psi}(y, \x, \boldsymbol{\beta}_0) \ d(G(y|\x) - F(y|\x)) \right) \\
& + \E_{\boldsymbol{X}} \left( \int \boldsymbol{\Psi}(y, \x, \boldsymbol{\beta}_0) \ dF(y|\x) \right) + \dot{P}_{\lambda,\alpha}(\boldsymbol{\beta}_0) \\
& + \left[ \varepsilon \E_{\boldsymbol{X}} \left( \int \mathbf{J}(y, \x, \tilde{\boldsymbol{\beta}}_\varepsilon) \ d(G(y|\x) - F(y|\x)) \right) \right. \\
& \left. + \E_{\boldsymbol{X}} \left( \int \mathbf{J}(y, \x, \tilde{\boldsymbol{\beta}}_\varepsilon) \ dF(y|\x) \right) \right. \\
& + \left. \ddot{P}_{\lambda,\alpha}(\check{\boldsymbol{\beta}}_\varepsilon) \right] (\boldsymbol{\beta}_\varepsilon - \boldsymbol{\beta}_0) \ .
\end{align*}
We let 
\begin{align*}
S & = \varepsilon \E_{\boldsymbol{X}} \left( \int \mathbf{J}(y, \x, \tilde{\boldsymbol{\beta}}_\varepsilon) \ d(G(y|\x) - F(y|\x)) \right) \\
& + \E_{\boldsymbol{X}} \left( \int \mathbf{J}(y, \x, \tilde{\boldsymbol{\beta}}_\varepsilon) \ dF(y|\x) \right) \\
& + \ddot{P}_{\lambda,\alpha}(\check{\boldsymbol{\beta}}_\varepsilon)
\end{align*}
and
\begin{align*}
N_1 & = \E_{\boldsymbol{X}} \left( \int \boldsymbol{\Psi}(y, \x, \boldsymbol{\beta}_0) \ d(G(y|\x) - F(y|\x)) \right) \\
N_2 & = \E_{\boldsymbol{X}} \left( \int \boldsymbol{\Psi}(y, \x, \boldsymbol{\beta}_0) \ dF(y|\x) \right) + \dot{P}_{\lambda,\alpha}(\boldsymbol{\beta}_0)  
\end{align*}
and we obtain
\begin{equation*}
(\boldsymbol{\beta}_\varepsilon - \boldsymbol{\beta}_0) = - S^{-1} \left(\varepsilon  N_1 + N_2\right) \ . 
\end{equation*}
Since C2, $N_1$ and $N_2$ are bounded, while since C5 there exists $\varepsilon_* > 0$ small enough such that, for all $0 \le \varepsilon < \varepsilon_*$ the matrix $S$ is non singular.
\end{proof}

\begin{proof}[Proof of Theorem \ref{teo:abp}]
  Under the assumption of Theorem \ref{teo:abp}, as it is shown in Theorem 3 of \citet{ValdoraYohai2014} the unpenalized MT-estimator $\tilde{\boldsymbol{\beta}}_\varepsilon$, solution of \eqref{equ:stability} with $\lambda=0$, has an asymptotic breakdown point of at least $\varepsilon_*$ and hence $\|\tilde{\boldsymbol{\beta}}_\varepsilon \| < \infty$ for any $0 \le \varepsilon \le \varepsilon_*$. Given, $\lambda \ge 0$, assume that $\|\boldsymbol{\beta}_\varepsilon \|$, solution of \eqref{equ:stability}, is not finite. This means there exists a sequence $\{ \boldsymbol{\beta}_{N\varepsilon} \}_{N=1}^\infty$ such that $\lim_{N \rightarrow \infty}\boldsymbol{\beta}_{N\varepsilon} = \boldsymbol{\beta}_\varepsilon$ and $\lim_{N \rightarrow \infty} \| \boldsymbol{\beta}_{N\varepsilon} \| = \infty$. Since assumption P3 we have that $\lim_{N \rightarrow \infty} \penn{\lambda}{\boldsymbol{\beta}_{N\varepsilon}}{\alpha} = \sup_{\boldsymbol{\beta}} \penn{\lambda}{\boldsymbol{\beta}}{\alpha}$ and hence $\penn{\lambda}{\tilde{\boldsymbol{\beta}}_\varepsilon}{\alpha} \le \penn{\lambda}{\boldsymbol{\beta}_{N\varepsilon}}{\alpha}$ for a sufficiently large $N$. Since $\tilde{\boldsymbol{\beta}}_\varepsilon$ and $\boldsymbol{\beta}_\varepsilon$ cannot be the same, the last fact contradict the definition of $\boldsymbol{\beta}_\varepsilon$, which proves that $\|\boldsymbol{\beta}_\varepsilon \| < \infty$ as well for any $0 \le \varepsilon \le \varepsilon_*$. 
\end{proof}
  
\subsection{Proofs of the asymptotic normality results in Section \ref{sec:asres}}
We now turn to the proof of the asymptotic normality of the proposed estimators. To derive the proof of Theorem \ref{teo:asdist} we will need two lemmas.
\begin{lemma}
\label{lemma:rn1}
Assume A2, C1, C2 nd C4.  {For each $\mathbf w \in \mathbb{R}^{p+1}$, define} 
\begin{equation*}
	R_{1}(\mathbf{w})=\mathbf{w}_0^\top \mathbf{w}+\frac{1}{2} \mathbf{w}^\top \mathbf{A} \mathbf{w}
\end{equation*}
where $\mathbf{w}_0 \sim N_{p+1}(\mathbf{0}, \mathbf{B})$ with $\mathbf{B}$ given {in \eqref{eq:AB}}. Furthermore, let $R_{n, 1}(\mathbf{w})$ be
\begin{align*}
R_{n, 1}(\mathbf{w}) & = \sum_{i=1}^{n}\left\{\rho \left(t(y_{i})-mg^{-1}\left(\x_{i}^\top \left[\boldsymbol{\beta}_* + \frac{\mathbf{w}}{\sqrt{n}} \right]\right)\right) - \rho \left(t(y_{i})-mg^{-1}\left(\x_{i}^\top\boldsymbol{\beta}_*\right)\right) \right\}. 
\end{align*}
Then, the process $R_{n, 1}$ converges in distribution to $R_{1}$.
\end{lemma}

\begin{proof} The proof is {similar} to the proof of Lemma A.7 in \citet{bianco2019penalized}. According to Theorem 2.3 in \citet{kim1990cube}, it is enough to prove that
\begin{itemize}
\item[(a)] For any $\mathbf{w}_{1}, \cdots, \mathbf{w}_{s}$ we have $$(R_{n, 1}(\mathbf{w}_{1}), \cdots, R_{n, 1}(\mathbf{w}_{s})) \xrightarrow{d} (R_{1}(\mathbf{w}_{1}), \cdots, R_{1}(\mathbf{w}_{s}))$$
\item[(b)] Given $\epsilon>0, \eta>0$ and $M<\infty$ there exists $\delta>0$ such that
\begin{equation*}
\limsup _{n \rightarrow \infty} P^{*}\left(\sup_{\|\mathbf{u}\| \leq M,\|\mathbf{v}\| \leq M \atop\|\mathbf{u}-\mathbf{v}\|<\delta }\left|R_{n, 1}(\mathbf{u})-R_{n, 1}(\mathbf{v})\right|>\epsilon \right) < \eta
\end{equation*}
where $P^{*}$ stands for outer probability.	
\end{itemize}
First note that it is enough to consider $s=1$  since for any other $s$ the proof follows similarly using the Cramer-Wald device.
Hence, we fix {$\mathbf{w} \in \mathbb{R}^{p+1}$}. A Taylor expansion of $R_{n, 1}$ around $\mathbf{w} =\mathbf{0}$ yields
\begin{equation*}
R_{n, 1}(\mathbf{w}) = \sqrt{n} \mathbf{w}^\top \nabla L_{n}(\boldsymbol{\beta}_*)+\frac{1}{2} \mathbf{w}^\top \mathbf{A}_{n}(\tilde{\boldsymbol{\beta}}_{\mathbf{w}}) \mathbf{w},
\end{equation*}
with $\mathbf{A}_n(\boldsymbol{\beta})$ as defined in \eqref{eq:AB} and $\tilde{\boldsymbol{\beta}}_{\mathbf{w}}=\boldsymbol{\beta}_{*}+\tau_{n} \mathbf{w} / \sqrt{n}$ with $\tau_n \in [0,1]$. Note that $\nabla L_{n}(\boldsymbol{\beta})=\frac{1}{n}\sum_{i=1}^n \mathbf{\Psi}(y_i, \mathbf x_i,\boldsymbol{\beta})$.
The Multivariate Central Limit Theorem, A2 and Lemma \ref{lemma:FC} entail that $\sqrt{n} \nabla L_{n}(\boldsymbol{\beta}_{*}) \xrightarrow{d} N_{p+1}(\mathbf{0}, \mathbf{B})$. On the other hand, Lemma \ref{lemma:victor} implies that $\mathbf{A}_{n}(\tilde{\boldsymbol{\beta}}_{\mathbf{w}}) \xrightarrow{p} \mathbf{A}$, so using Slutsky's Theorem we obtain that $R_{n, 1} \xrightarrow{d} \mathbf{w}_0^\top \mathbf{w}+\frac{1}{2} \mathbf{w}^\top \mathbf{A} \mathbf{w}$, concluding the proof of (a).
 
To derive (b), we perform a first order Taylor expansion of $R_{n, 1}(\mathbf{u})$ and $R_{n, 1}(\mathbf{v})$ around $0$ obtaining
\begin{equation*}
R_{n, 1}(\mathbf{u})-R_{n, 1}(\mathbf{v})=\sqrt{n} \nabla L_{n}(\boldsymbol{\beta}_*)^\top(\mathbf{u}-\mathbf{v})+\frac{1}{2} \mathbf{u}^\top \mathbf{A}_{n}(\tilde{\boldsymbol{\beta}}_{\mathbf{u}}) \mathbf{u}-\frac{1}{2} \mathbf{v}^\top \mathbf{A}_{n}(\tilde{\boldsymbol{\beta}}_{\mathbf{v}}) \mathbf{v}
\end{equation*}
where
$\tilde{\boldsymbol{\beta}}_{\mathbf{v}}=\boldsymbol{\beta}_*+{\tau_{\mathbf{v}, n} \mathbf{v}}/{\sqrt{n}}$ and $\tilde{\boldsymbol{\beta}}_{\mathbf{u}}=\boldsymbol{\beta}_* + {\tau_{\mathbf{u}, n} \mathbf{u}}/{\sqrt{n}}$ with $\tau_{\mathbf{v}, n}, \tau_{\mathbf{u}, n} \in[0,1]$. Noting that $\sqrt{n} \nabla L_{n}(\boldsymbol{\beta}_*)^\top(\mathbf{u}-\mathbf{v}) \leq O_{\mathbb P}(1)\|\mathbf{u}-\mathbf{v}\|$ and
\begin{align*}
\mathbf{u}^\top \mathbf{A}_{n}(\tilde{\boldsymbol{\beta}}_{\mathbf{u}}) \mathbf{u}-\mathbf{v}^\top \mathbf{A}_{n}(\tilde{\boldsymbol{\beta}}_{\mathbf{v}}) \mathbf{v} &=\mathbf{u}^\top \mathbf{A}_{n}(\tilde{\boldsymbol{\beta}}_{\mathbf{u}}) \mathbf{u}-\mathbf{u}^\top \mathbf{A}_{n}(\tilde{\boldsymbol{\beta}}_{\mathbf{v}}) \mathbf{u}+\mathbf{u}^\top \mathbf{A}_{n}(\tilde{\boldsymbol{\beta}}_{\mathbf{v}}) \mathbf{u} \\
 &-\mathbf{u}^\top \mathbf{A}_{n}(\tilde{\boldsymbol{\beta}}_{\mathbf{v}}) \mathbf{v}+\mathbf{u}^\top \mathbf{A}_{n}(\tilde{\boldsymbol{\beta}}_{\mathbf{v}}) \mathbf{v}-\mathbf{v}^\top \mathbf{A}_{n}(\tilde{\boldsymbol{\beta}}_{\mathbf{v}}) \mathbf{v}
\end{align*}
we obtain that, if $\|\mathbf{u}\|,\|\mathbf{v}\| \leq M$,
\begin{align*}
\left|R_{n, 1}(\mathbf{u})-R_{n, 1}(\mathbf{v})\right| \leq & O_{p}(1)\|\mathbf{u}-\mathbf{v}\|+M^{2}\left\|\mathbf{A}_{n}(\tilde{\boldsymbol{\beta}}_{\mathbf{u}})-\mathbf{A}_{n}(\tilde{\boldsymbol{\beta}}_{\mathbf{v}})\right\| \\
& + 2\|\mathbf{u}-\mathbf{v}\| M\left\|\mathbf{A}_{n}(\tilde{\boldsymbol{\beta}}_{\mathbf{v}})\right\|
\end{align*}
where $\| \mathbf C  \|$ stands for the Frobenius norm if $\mathbf{C}$ is a matrix. Lemma \ref{lemma:limits} entails that $\mathbf{A}_{n}(\tilde{\boldsymbol{\beta}}_{\mathbf{u}})- \mathbf{A}_{n}(\tilde{\boldsymbol{\beta}}_{\mathbf{v}}) \xrightarrow{p} 0$ and $\mathbf{A}_{n}(\tilde{\boldsymbol{\beta}}_{\mathbf{v}}) \xrightarrow{p} \mathbf{A}$, uniformly over all $\left\{\mathbf{u}, \mathbf{v} \in \mathbb{R}^{p+1}: \max \{\|\mathbf{u}\|,\|\mathbf{v}\|\} \leq M\right\}$ and $(\mathrm{b})$ follows, concluding the proof.
\end{proof}

In the next Lemma and in Theorem \ref{teo:asdist} $\lambda_{n}$ is allowed to be random.
\begin{lemma}\label{lemma:rn2} Let assume that P1 holds for $\penn{\lambda}{\cdot}{\alpha}$ and $\sqrt{n} \lambda_{n}=O_{\mathbb P}(1)$. Let $\mathbf{w}=(w_0, \mathbf{w}_1) \in \mathbb{R}^{p+1}$ and 
\begin{equation*}
R_{n,2}(\mathbf{w}) = n \left\{ \penn{\lambda_n}{\boldsymbol{\beta}_* + \frac{\mathbf{w}}{\sqrt{n}}}{\alpha} - \penn{\lambda_n}{\boldsymbol{\beta}_*}{\alpha} \right\}.
\end{equation*}  
Then, the process $R_{n, 2}$ is equicontinuous, i.e., for any $\epsilon>0, \eta>0$ and $M<\infty$ there exists $\delta>0$ such that
\begin{equation*}
\limsup_{n \rightarrow \infty} P^{*}\left(\sup _{\|\mathbf{u}|| \leq M,\|\mathbf{v}\|\leq M \atop\| \mathbf{u}-\mathbf{v} \|<\delta}\left|R_{n, 2}(\mathbf{u})-R_{n, 2}(\mathbf{v})\right|>\epsilon\right)<\eta.
\end{equation*}
\end{lemma}
\begin{proof} Because of P1, the function $\penn{\lambda}{\cdot}{\alpha}$ is Lipschitz with constant $K$ independent of $\alpha$. Therefore, in a neighbourhood of $\boldsymbol{\beta}_*$, there exists $M$ such that, if $||\mathbf{u}||< M$ and $||\mathbf{v}||<M$,
\begin{equation*}
\left|R_{n, 2}(\mathbf{u})-R_{n, 2}(\mathbf{v})\right| = n \left| \penn{\lambda_n}{\boldsymbol{\beta}_* + \frac{\mathbf{u}}{\sqrt{n}}}{\alpha} - \penn{\lambda_n}{\boldsymbol{\beta}_* + \frac{\mathbf{v}}{\sqrt{n}}}{\alpha} \right|\leq K \sqrt n\lambda_n \left\|\mathbf{u}-\mathbf{v}  \right\| .
\end{equation*}
Since $\sqrt n\lambda_n=O_P(1)$, there exists $C$ such that $P(\sqrt n \lambda_n \geq C)<\eta$, then, if $||\mathbf{u} - \mathbf{v}||<\epsilon/(CK)$, $||\mathbf{u}||<M$ and $||\mathbf{v}||<M$,
\begin{equation*}
P^*\left(\left|R_{n, 2}(\mathbf{u})-R_{n, 2}(\mathbf{v})\right|>\epsilon\right) \leq P\left(\sqrt n \lambda_n > \frac{\epsilon}{K ||\mathbf{u} - \mathbf{v}|| } \right) \leq P \left(\sqrt n \lambda_n> M\right) < \eta .
\end{equation*}
The result follows taking $\delta =\epsilon/(CK)$.
\end{proof}

\begin{proof}[Proof of Theorem \ref{teo:asdist}]
The proof is similar to the proof of Theorem {5.3} in \citet{bianco2019penalized}. Let $R_n(\mathbf{w}) = R_{n,1}(\mathbf{w}) + R_{n,2}(\mathbf{w})$, with $R_{n,1}(\mathbf{w})$ and $R_{n,2}(\mathbf{w})$ defined as in Lemmas \ref{lemma:rn1} and \ref{lemma:rn2} respectively and note that $\arg\min_{\mathbf{w}} R_{n}(\mathbf{w}) = \sqrt{n}(\hat{\boldsymbol{\beta}}_{n}-\boldsymbol{\beta}_*)$. So what we need to prove is that 
{
\begin{equation*}
\operatorname{argmin}_{\mathbf{w}} R_{n}(\mathbf{w}) \rightarrow \operatorname{argmin}_{\mathbf{w}}R(\mathbf{w}).\end{equation*}
We will use Theorem 2.7 in \citet{kim1990cube} with $t_n=-\sqrt{n}(\hat{\boldsymbol{\beta}}_{n}-\boldsymbol{\beta}_*)$, $Z=-R(\mathbf w)$ and $Z_n= R_n$.}
Condition (iii) is verified since $R_n( \sqrt{n}(\hat{\boldsymbol{\beta}}_{n}-\boldsymbol{\beta}_*)) \leq \inf_{\mathbf{t}} R_n(\mathbf{t})$ while condition (ii) follows from Theorem \ref{teo:order}. To verify (i) we need to prove that $R_{n}(\mathbf{w}) \xrightarrow{d} R(\mathbf{w})$. As explained in the proof of Lemma \ref{lemma:rn1}, it is enough to show that
\begin{itemize}
\item[(a)] For any $\mathbf{w}_{1}, \cdots, \mathbf{w}_{s}$ we have $(R_{n}(\mathbf{w}_{1}), \cdots, R_{n}(\mathbf{w}_{s})) \xrightarrow{d} (R(\mathbf{w}_{1}), \cdots, R(\mathbf{w}_{s}))$.
\item[(b)] Given $\epsilon>0, \eta>0$ and $M<\infty$ there exists $\delta>0$ such that
\begin{equation*}
\limsup_{n \rightarrow \infty} P^{*} \left( \sup _{\|\mathbf{u}\|\leq M,\|\mathbf{v}\|\leq M \atop\|\mathbf{u}-\mathbf{v}\|<\delta}\left|R_{n}(\mathbf{u})-R_{n}(\mathbf{v})\right|>\epsilon \right) < \eta
\end{equation*}
where $P^{*}$ stands for outer probability.
\end{itemize}
First note that (b) follows easily from Lemmas \ref{lemma:rn1} and \ref{lemma:rn2}. To prove (a), recall from the proof of Lemma \ref{lemma:rn1} that it is enough to consider the case $s=1$ and fix $\mathbf{w} \in \mathbb{R}^{p+1}$. We already know $R_{n, 1}(\mathbf{w}) \xrightarrow{d} \mathbf{w}^\top \mathbf{w}_0+\frac{1}{2} \mathbf{w}^\top \mathbf{A} \mathbf{w}$, so we only have to study the convergence of $R_{n, 2}(\mathbf{w})$. 

\noindent (i) For the elastic net penalty function we have that $R_{n, 2}(\mathbf{w})=R_{n, 2,1}(\mathbf{w})+R_{n, 2,2}(\mathbf{w})$ where $\mathbf{w} = (w_0, w_1, \ldots, w_p)$ and 
\begin{align*}
R_{n, 2,1}(\mathbf{w}) & = n \lambda_{n}\frac{1-\alpha}{2}\left\{\sum_{\ell=1}^{p} \left(\beta_{*, \ell}+\frac{w_{\ell}}{\sqrt{n}}\right)^2-(\beta_{*, \ell})^2\right\} \\
R_{n, 2,2}(\mathbf{w}) & = n \lambda_{n} \alpha\left\{\sum_{\ell=1}^{p}\left|\beta_{*, \ell}+\frac{w_{\ell}}{\sqrt{n}}\right|-\left|\beta_{*, \ell}\right|\right\}
\end{align*}
\begin{equation*}
R_{n, 2,1}(\mathbf{w}) = n \lambda_n\frac{1-\alpha}{2} \sum_{\ell=1}^{p} \left\{ 2 \beta_{*, \ell} \frac{w_{\ell}}{\sqrt n} + \frac{w_{\ell}^2}{n} \right\}=\sqrt n \lambda_n (1-\alpha) \sum_{\ell=1}^{p} \left\{ \beta_{*, \ell} w_{\ell} + \frac{w_{\ell}^2}{2\sqrt n}\right\}
\end{equation*}
{For each $\ell = 1, \ldots, p$}, if $\beta_{*,\ell} > 0$ and $n$ is large enough, then $\beta_{*, \ell}+\frac{w_{\ell}}{\sqrt{n}} >0$ as well, and the $\ell$-th term of $R_{n, 2,2}(\mathbf{w})$ is  $\sqrt{n} \lambda_n \alpha  w_{\ell}$. If $\beta_{*,\ell} < 0$ and $n$ is large enough, then the $\ell$-th term of $R_{n, 2,2}(\mathbf{w})$  is $-\sqrt{n} \lambda_n \alpha {w_{\ell}}$, while, if $\beta_{*,\ell} = 0$, the $\ell$-th term in $R_{n, 2,2}(\mathbf{w})$ is $\sqrt{n} \lambda_n \alpha  {|w_{\ell}|}$. So we have that
\begin{equation*}
R_{n, 2,2}(\mathbf{w}) = \sqrt{n} \lambda_n \alpha \sum_{\ell=1}^{p} w_\ell \sign(\beta_{*,\ell}) \mathbb{I}_{\{\beta_{*,\ell}\neq 0\}} +\left|w_{\ell}\right| \mathbb{I}_{\{{\beta_{*,\ell}=0} \}} \ .
\end{equation*}  
And then, 
\begin{align*}
R_{n, 2,1}(\mathbf{w}) \xrightarrow{p} & b(1-\alpha) \sum_{\ell=1}^{p} \beta_{*, \ell} w_{\ell} \\
R_{n, 2,2}(\mathbf{w}) \xrightarrow{p} & b \alpha \sum_{\ell=1}^{p}\left\{w_{\ell} \sign(\beta_{*, \ell}) \mathbb{I}_{\left\{\beta_{*, \ell} \neq 0\right\}}+|w_{\ell}| \mathbb{I}_{\left\{{\beta_{*,\ell}=0}  \right\}}\right\}.
\end{align*}
Therefore, $R_{n, 2}(\mathbf{w}) \xrightarrow{p} b \mathbf{w}^\top \mathbf{q}(\mathbf{w})$ and the result follows from Slutzky's lemma.

\noindent (ii) For the bridge penalty function and $\alpha > 1$, Theorem 7 in \citet{smucler2015} shows that
\begin{equation*}
R_{n,2} (\mathbf{w}) = \sqrt{n} \lambda_n \left( \sum_{\ell=1}^p |\beta_{*, \ell} + \frac{w_\ell}{\sqrt{n}}|^\alpha - |\beta_{*, \ell}|^\alpha \right) \xrightarrow{p} b \alpha \sum_{\ell=1}^p w_{\ell} \sign(\beta_{*, \ell}) |\beta_{*, \ell}|^{\alpha-1}   
\end{equation*}
uniformly over compact sets, while for $\alpha=1$ (lasso) it is equivalent to the Elastic net penalty and $\alpha=1$ in the previous point (i).
For $\alpha < 1$ and $\lambda/n^{\alpha/2} \rightarrow b$ Theorem 8 in \citet{smucler2015} shows that
\begin{equation*}
R_{n,2} (\mathbf{w}) = \sqrt{n} \lambda_n \left( \sum_{\ell=1}^p |\beta_{*, \ell} + \lambda_n w_\ell|^\alpha - |\beta_{*, \ell}|^\alpha \right) \xrightarrow{p} \sum_{\ell=1}^p \mathbb{I}_{\left\{\beta_{*, \ell} = 0\right\}} |w_\ell|^{\alpha} 
\end{equation*}
uniformly over compact sets.

\noindent (iii) For the sign penalty function Theorem 5 of \citet{bianco2022penalized} shows that, given $J(\boldsymbol{\beta}) = \sum_{\ell=1}^p J_\ell(\boldsymbol{\beta})$, $J_\ell(\boldsymbol{\beta}) = |\beta_\ell|/\|\boldsymbol{\beta}\|_2$ and $\tau_{n,\ell} \in [0,1]$,
\begin{align*}
  R_{n,2} (\mathbf{w}) & = \sqrt{n} \lambda_n \left( \sum_{\ell=1}^p \left[\nabla J_\ell\left(\boldsymbol{\beta}_{*} + \tau_{n,\ell} \frac{\mathbf{w}}{\sqrt{n}}\right)\right]^\top \mathbf{w} \mathbb{I}_{\left\{\beta_{*, \ell} \neq 0\right\}} + \sum_{\ell=1}^p \frac{|w_\ell|}{\| \boldsymbol{\beta}_{*} + \frac{\mathbf{w}}{\sqrt{n}} \|} \mathbb{I}_{\left\{\beta_{*, \ell} = 0\right\}} \right) \\
& \xrightarrow{p} b \sum_{\ell=1}^p  {\nabla_\ell(\boldsymbol{\beta}_*)^\top\mathbf{w}} \mathbb{I}_{\left\{\beta_{*, \ell} \neq 0\right\}} + \sign(w_\ell)/\|\boldsymbol{\beta}_*\|_2 \mathbb{I}_{\left\{\beta_{*, \ell} = 0\right\}} \mathbf{e}_\ell
\end{align*}
where $\mathbf{e}_\ell$ stands for the $\ell$th canonical vector and $\nabla_\ell(\boldsymbol{\beta}) = - (|\beta_\ell| / \|\boldsymbol{\beta}_*\|_2^3) \boldsymbol{\beta} + \sign(\beta_\ell)/\|\boldsymbol{\beta}\|_2 \mathbf{e}_\ell$.
\end{proof}

\begin{proof}[Proof of Theorem \ref{teo:asdist2}]
Let
\begin{align*}
R_{n, 1}(\mathbf{w})& = \frac{1}{n \lambda_{n}^{2}} \sum_{i=1}^{n}\left\{\rho(t(y_{i})- \mginv(\x_{i}^\top(\boldsymbol{\beta}_*+\lambda_{n} \mathbf{w})) )-\rho(t(y_{i})- \mginv(\x_{i}^\top\boldsymbol{\beta}_*) )\right\}\\
R_{n, 2}(\mathbf{w})& = \frac{1}{\lambda_{n}}\left\{ \penn{1}{\boldsymbol{\beta}_* + \lambda_{n} \mathbf{w}}{\alpha} - \penn{1}{\boldsymbol{\beta}_*}{\alpha} \right\}
\end{align*}
and $R_n(\mathbf{w})=R_{n, 1}(\mathbf{w}) + R_{n, 2}(\mathbf{w})$.

Note that $\arg\min_{\mathbf{w}} R_{n}(\mathbf{w})=(1 / \lambda_{n})(\hat{\boldsymbol{\beta}}_{n}-\boldsymbol{\beta}_*)$. As in the proof of Theorem \ref{teo:asdist}, we begin by showing that given $\epsilon>0, \eta>0$ and $M<\infty$ there exists $\delta>0$ such that
\begin{equation}
\label{eq:equicont}
\limsup_{n \rightarrow \infty} P^{*}\left(\sup_{\|\mathbf{u}\| \leq M,\|\mathbf{v}\| \leq M \atop\|\mathbf{u}-\mathbf{v}\|<\delta }\left|R_{n, l}(\mathbf{u})-R_{n, l}(\mathbf{v})\right|>\epsilon\right)<\eta
\text{ for } l=1,2.
\end{equation}
For $l=2$ \eqref{eq:equicont} follows from the fact that $\penn{\lambda}{\cdot}{\alpha}$ is locally Lipschitz since P1.

To prove \eqref{eq:equicont} for $l=1$, we use Taylor expansions of $R_{n, 1}(\mathbf{u})$ and $R_{n, 1}(\mathbf{v})$ around $\boldsymbol{\beta}_*$ and we get
\begin{equation*}
R_{n, 1}(\mathbf{u})-R_{n, 1}(\mathbf{v})=\frac{1}{\lambda_{n} \sqrt{n}} \sqrt{n} \nabla L_{n}(\boldsymbol{\beta}_*)^\top(\mathbf{u}-\mathbf{v})+\frac{1}{2} \mathbf{u}^\top \mathbf{A}_{n}(\tilde{\boldsymbol{\beta}}_{\mathbf{u}}) \mathbf{u}-\frac{1}{2} \mathbf{v}^\top \mathbf{A}_{n}(\tilde{\boldsymbol{\beta}}_{\mathbf{v}}) \mathbf{v}
\end{equation*}  
where $\tilde{\boldsymbol{\beta}}_{\mathbf{u}}$ and $\tilde{\boldsymbol{\beta}}_{\mathbf{v}}$ are intermediate points defined as
\begin{equation*}
\tilde{\boldsymbol{\beta}}_{\mathbf{u}}=\boldsymbol{\beta}_*+\lambda_{n} \tau_{\mathbf{u}, n} \mathbf{u}
\quad \text{ and } \quad
\tilde{\boldsymbol{\beta}}_{\mathbf{v}}=\boldsymbol{\beta}_*+\lambda_{n} \tau_{\mathbf{v}, n} \mathbf{v}
\end{equation*}
with $\tau_{\mathbf{u}, n}, \tau_{\mathbf{v}, n} \in [0,1]$. As in the proof of Lemma \ref{lemma:rn1}, using that $\sqrt{n} \nabla L_{n}(\boldsymbol{\beta}_*) \xrightarrow{d} N_{p+1}(\mathbf{0}, \mathbf{B})$, we conclude that $\sqrt{n} \nabla L_{n}(\boldsymbol{\beta}_*)^\top(\mathbf{u}-\mathbf{v}) \leq O_{\mathbb P}(1)\|\mathbf{u}-\mathbf{v}\| .$ On the other hand,
\begin{align*}
\mathbf{u}^\top \mathbf{A}_{n}(\tilde{\boldsymbol{\beta}}_{\mathbf{u}}) \mathbf{u}-\mathbf{v}^\top \mathbf{A}_{n}(\tilde{\boldsymbol{\beta}}_{\mathbf{v}}) \mathbf{v} & = \mathbf{u}^\top \mathbf{A}_{n}(\tilde{\boldsymbol{\beta}}_{\mathbf{u}}) \mathbf{u}-\mathbf{u}^\top \mathbf{A}_{n}(\tilde{\boldsymbol{\beta}}_{\mathbf{v}}) \mathbf{u}+\mathbf{u}^\top \mathbf{A}_{n}(\tilde{\boldsymbol{\beta}}_{\mathbf{v}}) \mathbf{u} \\
& - \mathbf{u}^\top \mathbf{A}_{n}(\tilde{\boldsymbol{\beta}}_{\mathbf{v}}) \mathbf{v}+\mathbf{u}^\top \mathbf{A}_{n}(\tilde{\boldsymbol{\beta}}_{\mathbf{v}}) \mathbf{v}-\mathbf{v}^\top \mathbf{A}_{n}(\tilde{\boldsymbol{\beta}}_{\mathbf{v}}) \mathbf{v}
\end{align*}
so, if $\|\mathbf{u}\|,\|\mathbf{v}\| \leq M$
\begin{align*}
\left|R_{n, 1}(\mathbf{u})-R_{n, 1}(\mathbf{v})\right| \leq & \frac{1}{\lambda_{n} \sqrt{n}} O_{p}(1){\|\mathbf{u}-\mathbf{v}\|}+M^{2}\left\|\mathbf{A}_{n}(\tilde{\boldsymbol{\beta}}_{\mathbf{u}})-\mathbf{A}_{n}(\tilde{\boldsymbol{\beta}}_{\mathbf{v}})\right\| \\
&+2\|\mathbf{u}-\mathbf{v}\| M\left\|\mathbf{A}_{n}(\tilde{\boldsymbol{\beta}}_{\mathbf{v}})\right\|.
\end{align*}
Using that $\lambda_{n} \sqrt{n} \rightarrow \infty$ and  that $\mathbf{A}_{n}(\tilde{\boldsymbol{\beta}}_{\mathbf{u}})-\mathbf{A}_{n}(\tilde{\boldsymbol{\beta}}_{\mathbf{v}}) \xrightarrow{p} 0$ and $\mathbf{A}_{n}(\tilde{\boldsymbol{\beta}}_{\mathbf{v}}) \xrightarrow{p} \mathbf{A}$ uniformly over $\left\{\mathbf{u}, \mathbf{v} \in \mathbb{R}^{p}: \max \{\|\mathbf{u}\|,\|\mathbf{v}\|\} \leq M\right\}$ by Lemma \ref{lemma:limits}, we get that \eqref{eq:equicont} holds for $\ell=1$.
It remains to see that given $\mathbf{w}_{1}, \ldots, \mathbf{w}_{s} \in \mathbb{R}^{p+1}$ we have $(R_{n}(\mathbf{w}_{1}), \ldots, R_{n}(\mathbf{w}_{s})) \xrightarrow{d} (R(\mathbf{w}_{1}), \ldots, R(\mathbf{w}_{s}))$, where $R(\mathbf{w})=(1 / 2) \mathbf{w}^\top \mathbf{A} \mathbf{w}+\mathbf{w}^\top \mathbf{q}(\mathbf{w})$. As in the proof of Theorem \ref{teo:asdist}, it is enough to show the result when $s=1$. Fix $\mathbf{w} \in \mathbb{R}^{p+1}$. Using again a Taylor expansion of order one, we obtain
\begin{equation}
\label{eq:rn1}
R_{n, 1}(\mathbf{w})=\frac{1}{\lambda_{n} \sqrt{n}} \sqrt{n} \nabla L_{n}(\boldsymbol{\beta}_*)^\top \mathbf{z}
+ \frac{1}{2} \mathbf{z}^\top \mathbf{A}_{n}(\tilde{\boldsymbol{\beta}}_{\mathbf{w}}) \mathbf{w}
\end{equation}
where $\tilde{\boldsymbol{\beta}}_{\mathbf{w}}=\boldsymbol{\beta}_*+\lambda_{n} \tau_{n} \mathbf{w}$, with $\tau_{n} \in [0,1]$. {Since}
\begin{equation*}
\sqrt{n} \nabla L_{n}(\boldsymbol{\beta}_*) \xrightarrow{d} N_{p+1}(\mathbf{0}, \mathbf{B}) ,
\end{equation*}
The fact that $\lambda_{n} \sqrt{n} \rightarrow \infty$, implies that
\begin{equation*}
\frac{1}{\lambda_{n} \sqrt{n}} \sqrt{n} \nabla L_{n}(\boldsymbol{\beta}_*)^\top \mathbf{w} \xrightarrow{p} 0.
\end{equation*}
The second term in \eqref{eq:rn1} converges to $(1 / 2) \mathbf{w}^\top \mathbf{A} \mathbf{w}$ in probability by Lemma \ref{lemma:limits}. This proves
\begin{equation*}
R_{n,1}(\mathbf{w}) \xrightarrow{p} (1 / 2) \mathbf{w}^\top \mathbf{A} \mathbf{w}.
\end{equation*}  
We now show $R_{n, 2}(\mathbf{w}) \xrightarrow{p} \mathbf{w}^\top \mathbf{q}(\mathbf{w})$.

\noindent (i) For the elastic net penalty function, first note that
\begin{equation*}
R_{n, 2}(\mathbf{w})= \frac{1}{\lambda_n} \sum_{\ell=1}^{p} \frac{1-\alpha}{2} \left( 2 \beta_{*, \ell} \lambda_n w_{\ell} + \lambda_n^2w_{\ell}^2) + \alpha (|\beta_{*, \ell}+\lambda_n w_{\ell}|-|\beta_{*, \ell}| \right) .
\end{equation*}
Note that, for $\ell > 0$ if $n$ is large enough, then $\beta_{*, \ell}+\lambda_n w_{\ell}$ has the same sign as $\beta_{*, \ell}$, then, as in the proof of Theorem \ref{teo:asdist}, we get
\begin{align*}
R_{n, 2}(\mathbf{w}) & = \sum_{\ell=1}^{p} \frac{1-\alpha}{2} \left( 2 \beta_{*, \ell} w_{\ell} + \lambda_n w_{\ell}^2\right) + \alpha \left( w_\ell \sign(\beta_{*,\ell}) \mathbb{I}_{\{\beta_{*,\ell}\neq 0\}} +\left|w_{\ell}\right| \mathbb{I}_{{\{\beta_{*,\ell}= 0 \}}}\right) \\
& \xrightarrow{p} \mathbf{w}^\top \mathbf{q}(\mathbf{w}) .
\end{align*}

\noindent (ii) For the bridge penalty function, with $\alpha > 1$ we have
\begin{equation*}
R_{n,2} (\mathbf{w}) = \frac{1}{\lambda_n} \left( \sum_{\ell=1}^p |\beta_{*, \ell} + \lambda_n w_\ell|^\alpha \right) \xrightarrow{p} \alpha \sum_{\ell=1}^p w_{\ell} \sign(\beta_{*, \ell}) |\beta_{*, \ell}|^{\alpha-1}   
\end{equation*}
uniformly over compact sets, while for $\alpha=1$ (lasso) it is equivalent to the Elastic net penalty and $\alpha=1$ in the previous point (i). For $\alpha < 1$ and $\lambda n^{\alpha/2} \rightarrow \infty$
\begin{equation*}
R_{n,2} (\mathbf{w}) = \frac{1}{\lambda_n} \left( \sum_{\ell=1}^p |\beta_{*, \ell} + \lambda_n w_\ell|^\alpha - |\beta_{*, \ell}|^\alpha \right) \xrightarrow{p} \sum_{\ell=1}^p \mathbb{I}_{\left\{\beta_{*, \ell} = 0\right\}} |w_\ell|^{\alpha} 
\end{equation*}
uniformly over compact sets.

\noindent (iii) For the sign penalty function Theorem 7 of \citet{bianco2022penalized} shows that, given $J(\boldsymbol{\beta}) = \sum_{\ell=1}^p J_\ell(\boldsymbol{\beta})$, $J_\ell(\boldsymbol{\beta}) = |\beta_\ell|/\|\boldsymbol{\beta}\|_2$ and $\tau_{n,\ell} \in [0,1]$,
\begin{align*}
R_{n,2} (\mathbf{w}) & = \left( \sum_{\ell=1}^p \left[\nabla J_\ell\left(\boldsymbol{\beta}_{*} + \tau_{n,\ell} \lambda_n \mathbf{w} \right)\right]^\top \mathbf{w} \mathbb{I}_{\left\{\beta_{*, \ell} \neq 0\right\}} + \sum_{\ell=1}^p \frac{|w_\ell|}{\| \boldsymbol{\beta}_{*} + \lambda_n \mathbf{w} \|} \mathbb{I}_{\left\{\beta_{*, \ell} = 0\right\}} \right) \\
& \xrightarrow{p} \sum_{\ell=1}^p {\nabla_\ell(\boldsymbol{\beta}_*)^\top\mathbf{w}} \mathbb{I}_{\left\{\beta_{*, \ell} \neq 0\right\}} + \sign(w_\ell)/\|\boldsymbol{\beta}_*\|_2 \mathbb{I}_{\left\{\beta_{*, \ell} = 0\right\}} \mathbf{e}_\ell
\end{align*}
where $\mathbf{e}_\ell$ stands for the $\ell$th canonical vector and $\nabla_\ell(\boldsymbol{\beta}) = - (|\beta_\ell| / \|\boldsymbol{\beta}_*\|_2^3) \boldsymbol{\beta} + \sign(\beta_\ell)/\|\boldsymbol{\beta}\|_2 \mathbf{e}_\ell$.
\end{proof}

\section{Computational details and algorithm}
\label{sec:computational}
We introduce an Iteratively Re-Weighted Least Square (IRWLS) algorithm in Subsection \ref{app:irwls} while in subsection \ref{app:initial} we discuss how to obtain feasible starting values.

\subsection{Iteratively Re-Weighted Least Square algorithm}
\label{app:irwls}
An IRWLS procedure to compute the MT-estimator is described in \citet{agostinelli2019initial}, where we show that the solution to the minimization of \eqref{eq:Ln} can be approximated by the following 
\begin{equation}
\label{eq:appMT}
\min_{\boldsymbol{\beta}} \sum_{i=1}^{n} \rho( t(y_i) - s(\x_i^\top \boldsymbol{\beta}^k) - s^{\prime} (\x_i^\top \boldsymbol{\beta}^k ) \x_i^\top (\boldsymbol{\beta} - \boldsymbol{\beta}^k) ,
\end{equation}
where $s(\x_i^\top \boldsymbol{\beta}^k) = \mginv(\x_i^\top \boldsymbol{\beta}^k)$.

Since the elastic net penalty is not differentiable, we wish to keep the minimization problem instead of replacing it by the estimating equations. The iterative procedure developed in \citet{agostinelli2019initial} is based on the estimating equations. However, the same iterative method can be written in the following way
\begin{equation}
\label{eq:appMT2}
\min_{\boldsymbol{\beta}} \sum_{i=1}^{n} ( t(y_i) - s(\x_i^\top \boldsymbol{\beta}^k) - s^{\prime} (\x_i^\top \boldsymbol{\beta}^k) \x_i^\top (\boldsymbol{\beta} - \boldsymbol{\beta}^k))^2 w_i^\ast( \x_i^\top \boldsymbol{\beta}^k)
\end{equation}
where $w_i^\ast(v) = \psi(t(y_i) - s(v))/(t(y_i)-s(v))$ and $\psi=\rho^{\prime}$. In fact, the estimating equations corresponding to problem \eqref{eq:appMT2} are
\begin{equation}
\label{MTA1}
\sum_{i=1}^n ( t(y_i) - s(\x_i^\top \boldsymbol{\beta}^k ) - s^{\prime} (\x_i^\top \boldsymbol{\beta}^k ) \x_i^\top (\boldsymbol{\beta} - \boldsymbol{\beta}^k) ) w_i^\ast(\x_i^\top \boldsymbol{\beta}^k) s^{\prime}(\x_i^\top \boldsymbol{\beta}^k ) \x_i = \mathbf{0}. 
\end{equation}
The IRWLS introduced in \citet{agostinelli2019initial} to solve the equations above is the following. If the solution on step $k$ is $\boldsymbol{\beta}^k$, the solution on step $k+1$ is given by 
\begin{equation}
\label{equ:mtIRWLS}  
\boldsymbol{\beta}^{k+1} = \boldsymbol{\beta}^{k} + ( \mathbf{X}^\top \mathbf{W}^{2}(\mathbf{X} \boldsymbol{\beta}^{k})^\top \mathbf{W}^\ast(\mathbf{X} \boldsymbol{\beta}^k ) \mathbf{X} )^{-1} \mathbf{X}^\top \mathbf{W} (\mathbf{X} \boldsymbol{\beta}^k ) \mathbf{W}^\ast( \mathbf{X} \boldsymbol{\beta}^k)(\mathbf{T} - \mathbf{s}( \mathbf{X} \boldsymbol{\beta}^k))
\end{equation}
where $\mathbf{X}$ is the $n \times (p+1)$ matrix whose $i$-th row is $\x_i$, $\mathbf{W}$ is a diagonal matrix whose elements are $s^{\prime}(\x_i^\top \boldsymbol{\beta}^k) \x_i^\top \boldsymbol{\beta}^k$ and $\mathbf{W}^\ast$ is a diagonal matrix of robust weights $(w_1^\ast, \ldots, w_n^\ast)$ as defined below equation \eqref{eq:appMT2}. Now we turn to the problem of solving the penalized optimization problem
\begin{equation*}
\hat{\boldsymbol{\beta}} = \arg\min_{\boldsymbol{\beta}} \sum_{i=1}^n \rho ( t(y_i) - \mginv(\x_i^\top \boldsymbol{\beta})) + \penn{\lambda}{\boldsymbol{\beta}}{\alpha} .
\end{equation*}
At each step of an IRWLS algorithm we can approximate our optimization problem by
\begin{equation}\label{eq:app:elasticnet1}
\boldsymbol{\beta}^{k+1} = \arg\min_{\boldsymbol{\beta}} \frac{1}{2n}\sum_{i=1}^{n}( t(y_i) - s(\x_i^\top \boldsymbol{\beta}^k) - s^{\prime} (\x_i^\top \boldsymbol{\beta}^k) \x_i^\top (\boldsymbol{\beta} - \boldsymbol{\beta}^k))^2 w_i^\ast(\x_i^\top \boldsymbol{\beta}^k) + \penn{\lambda}{\boldsymbol{\beta}}{\alpha}
\end{equation}
which is a form of weighted elastic net with weights evaluated at a previous step $k$. This problem can be written as 
\begin{equation}
\label{eq:app:elasticnet}
\boldsymbol{\beta}^{k+1} = \arg\min_{\boldsymbol{\beta}} \frac{1}{2n} \sum_{i=1}^{n}( z_i^k - \mathbf{v}_i^{k\top} (\boldsymbol{\beta} - \boldsymbol{\beta}^k))^2 + \penn{\lambda}{\boldsymbol{\beta}}{\alpha}
\end{equation}
with $z_i^k = (t(y_i) - s(\x_i^\top \boldsymbol{\beta}^k)) \sqrt{w_i^\ast( \x_i^\top \boldsymbol{\beta}^k)}$ and $\mathbf{v}_i^k = s^{\prime} (\x_i^\top \boldsymbol{\beta}^k) \x_i \sqrt{w_i^\ast( \x_i^\top \boldsymbol{\beta}^k)}$. Solutions to \eqref{eq:app:elasticnet} can be obtained with several algorithms depending on $\alpha$ and on the dimension of the problem. Here we follow the approach in \citet{friedman2007pathwise} which is specialized to the GLMs case in \citet{friedman2010regularization} using a coordinate descendent algorithm. At step $k$ we can define $\hat{z}_i^{k(j)} = \beta_0^k + \sum_{l\neq j} v_{il} \beta_l^k$ as the fitted value excluding the contribution from $x_{ij}$, and $z_i^k - \hat{z}_i^{k(j)}$ the partial residual for fitting $\beta_j$. Hence, as explained in \citet{friedman2010regularization}, an update for $\beta_j$ can be obtained as
\begin{equation}
\label{eq:app:updatej}
\beta_j^{k+1} = \frac{ S\left( \frac{1}{n} \sum_{i=1}^n  v_{ij} (z_i^k - \hat{z}_i^{k(j)}) , \lambda \alpha \right)}{\frac{1}{n} \sum_{i=1}^n v_{ij}^2 + \lambda (1 - \alpha)} \qquad j=1,\ldots,p.
\end{equation}
where $S(z, \gamma) = \sign(z) (|z| - \gamma)_+$ is the soft-thresholding operator, while for $j=0$,
\begin{equation}
\label{eq:app:update0}
\beta_0^{k+1} = \frac{1}{n} \sum_{i=1}^n (z_i^k - \hat{z}_i^{k(0)})
\end{equation}
where $\hat{z}_i^{k(0)} =  \sum_{l} v_{il} \beta_l^k$. In the special case $\alpha=0$ we have the Ridge penalty and instead of the coordinate descent algorithm we can solve \eqref{eq:app:elasticnet} directly to obtain the Ridge normal equations 
\begin{align*}
\beta_0^{k+1} & = \bar{z}^k - \bar{\mathbf{v}}^{k\top} \boldsymbol{\beta}^{k+1} \\
\Beta^{k+1} & = ( \mathbf{V}^{k\top} \mathbf{V}^{k} + \lambda \mathbf{I} )^{-1} \mathbf{V}^{k\top} (\mathbf{z}^{k} - \beta_0^{k+1} \mathbf{1}_n)
\end{align*}
where $\Beta^\top = (\beta_1, \ldots, \beta_p)$, $\mathbf{V}^k$ is the matrix with rows $\mathbf{v}_i^k$, $\mathbf{z}^k = (z_1^k, \ldots, z_n^k)^\top$, $\bar{\mathbf{v}}^k$ is the vector of the means of the columns of $\mathbf{V}^k$ and $\bar{z}^k$ is the mean of $\mathbf{z}^k$.

On the other hand, if $\alpha=1$, we have the lasso penalty and \eqref{eq:app:elasticnet} can also be solved with a weighted lasso algorithm based on least angle regression as done in \citet{alfons2013sparse} and \citet{smucler2017robust}. This was implemented using the function \texttt{fastLasso} from package \texttt{robustHD}, see \citet{alfons2019robustHD}.

\subsection{Procedure for obtaining a robust initial estimate}
\label{app:initial}

In our penalized version of MT-estimators we consider redescending $\rho$-functions and hence the goal function $\hat{L}_n(\boldsymbol{\beta})$ might have several local minima. As a consequence, it might happen that the IRWLS algorithm converges to a solution of the estimating equations that is not a solution of the optimization problem. In practice, to avoid this, one must begin the iterative algorithm at an initial estimator which is a very good approximation of the absolute minimum of $\hat{L}_n(\boldsymbol{\beta})$. If $p$ is small, an approximate solution may be obtained by the subsampling method which consists in computing a finite set $A$ of candidate solutions and then replace the minimization over $\mathbb{R}^{p+1}$ by a minimization over $A$. The set $A$ is obtained by randomly drawing subsamples of size $p+1$ and computing the maximum likelihood estimator based on the subsample. Assume the original sample contains a proportion $\epsilon$ of outliers, then the probability of having at least one subsample free of outliers is $1-(1-(1-\epsilon)^p)^N$ where $N$ is the number of subsamples drawn. So, for a given probability $\alpha$ such that $1-(1-(1-\epsilon)^p)^N > \alpha$ we need
\begin{equation*}
N > \frac{\log(\alpha)}{\log(1-(1-\varepsilon)^p)}\underline{\sim}\left\vert
\frac{\log(\alpha)}{(1-\varepsilon)^p}\right\vert.
\end{equation*}
This makes the algorithm based on subsampling infeasible for large $p$. We instead propose an adaptation of the procedure in \citet{agostinelli2019initial} which is a deterministic algorithm.

Consider a random sample following a generalized linear model. The following procedure computes an approximation of $\boldsymbol{\beta}_*$ which will be used as an initial estimator in the IRWLS algorithm for the estimating equation describe in Subsection \ref{app:irwls}. The procedure has two stages. Stage 1 aims at finding a highly robust but possibly inefficient estimate and stage 2 aims at increasing its efficiency.

\noindent{\itshape Stage 1.} In this stage, the idea is to find a robust, but possibly inefficient, estimate of $\boldsymbol{\beta}_*$ by an iterative procedure. In each iteration $k \geq 1$ we get
\begin{equation}
\label{eq:optLikAk}
\hat{\boldsymbol{\beta}}^{(k)} = \arg\min_{\boldsymbol{\beta} \in A_{k}} \hat{L}_n(\boldsymbol{\beta}) .
\end{equation}
In the first iteration ($k=1$) the set $A_{1}$ is constructed as follows. We begin by computing the penalized LST estimate $\hat{\boldsymbol{\beta}}_{LS}^{(1)}$ with the complete sample and the principal sensitivity components \citep{PenaYohai1999} obtained as follows. We define the fitted values $\hat{\boldsymbol{\mu}}^{(1)} = g^{-1}(\mathbf{X}\hat{\boldsymbol{\beta}}^{(1)})$ and for each index $j$ the fitted values $\hat{\boldsymbol{\mu}}_{(j)}^{(1)} = g^{-1}(\mathbf{X}\hat{\boldsymbol{\beta}}_{LS(j)}^{(1)})$ obtained by computing the penalized LST estimate $\hat{\boldsymbol{\beta}}_{LS(j)}^{(1)}$ with the sample without using observation with $j$ index. We compute the sensitivity vector $\mathbf{r}_{(j)} = \hat{\mathbf{t}} - \hat{\mathbf{t}}_{(j)}$ which is the difference between the predicted value $\hat{\mathbf{t}} = m(\hat{\boldsymbol{\mu}}^{(1)})$ based on the complete sample and $\hat{\mathbf{t}}_{(j)} = m(\hat{\boldsymbol{\mu}}_{(j)}^{(1)})$ based on the sample without the observation with $j$ index. The sensitivity matrix $\mathbf{R}$ is built from the sensitivity vectors $\mathbf{r}_{(1)}, \ldots, \mathbf{r}_{(n)}$. We obtain the direction $\mathbf{v}_{1}$ in which the projections of the sensitivity vectors is largest, i.e., 
\begin{equation*}
\mathbf{v}_{1}=\mbox{argmax}_{||\mathbf{v}||=1}\sum_{j=1}^{n}\left(
\mathbf{v}^{\top}\mathbf{r}_{(j)}\right)^{2}
\end{equation*}
and $\mathbf{z}_{1} = \mathbf{R}\mathbf{v}_{1}$ where the largest entries in $\mathbf{z}_{1}$ correspond to the largest terms in the sum $\sum_{j=1}^{n}\left(\mathbf{v}^\top \mathbf{r}_{(j)} \right)^2$, which in turn correspond to the observations that have the largest projected sensitivity in the direction $\mathbf{v}_{1}$. Recursively we compute $\mathbf{v}_{i},$ $2 \leq i \leq n$ as the solution of
\begin{align*}
\mathbf{v}_{i} & = \arg\max_{||\mathbf{v}||=1}\sum_{j=1}^{n}\left(\mathbf{v}^{\top}\mathbf{r}_{(j)}\right)^{2} \\
& \text{ subject to } \mathbf{v}_{i}\mathbf{v}_{j} =0 \text{ for all } 1 \leq j < i
\end{align*}
and the corresponding principal sensitivity components $\mathbf{z}_{i} = \mathbf{R}\mathbf{v}_{i}$. Large entries are considered potential outliers. For each principal sensitivity component $\mathbf{z}_{i}$ we compute three estimates by the
penalized LST method. The first eliminating the half of the observations corresponding to the smallest entries in $\mathbf{z}_{i}$, the second eliminating the half corresponding to the largest entries in $\mathbf{z}_{i}$ and the third eliminating the half corresponding to the largest absolute values. To these $3p$ initial candidates we add the penalized LST estimate computed using the complete sample, obtaining a set of $3p+1$ elements. Once we have $A_{1}$ we obtain $\hat{\boldsymbol{\beta}}^{(1)}$ by minimizing $\hat{L}_n(\boldsymbol{\beta})$ over the elements of $A_{1}$.

Suppose now that we are on stage $k$. Let $0 < \alpha < 0.5$ be a trimming constant, in all our applications we set $\alpha=0.05$. Then, for $k>1$, we first delete the observations $(i=1,\cdots,n)$ such that $y_{i} > F_{\hat{\mu}_{i}^{(k-1)}}^{-1}(1-\alpha/2)$ or $y_{i}<F_{\hat{\mu}^{(k-1)}_{i}}^{-1}(\alpha/2)$ and, with the remaining observations, we re-compute the penalized LST estimator $\hat{\boldsymbol{\beta}}_{LS}^{(k)}$ and the principal sensitivity components. Let us remark that, for the computation of $\hat{\boldsymbol{\beta}}_{LS}^{(k)}$ we have deleted the observations that have large residuals, since $\hat{\boldsymbol{\mu}}^{(k-1)}$ are the fitted values obtained using $\hat{\boldsymbol{\beta}}_{LS}^{(k-1)}$. In this way, while candidates on the first step of the iteration are protected from high leverage outliers, candidate $\hat{\boldsymbol{\beta}}_{LS}^{(k)}$ is protected from low leverage outliers, which may not be extreme entries of the $\mathbf{z}_{i}$.

The set $A_{k}$ will contain $\hat{\boldsymbol{\beta}}_{LS}^{(k)}$, $\hat{\boldsymbol{\beta}}^{(k-1)}$ and the $3p$ penalized LST estimates computed deleting extreme values according to the principal sensitivity components as in the first iteration. $\hat{\boldsymbol{\beta}}^{(k)}$ is the element of $A_k$ minimizing $\hat{L}_n(\boldsymbol{\beta})$.

The iterations will continue until $\hat{\boldsymbol{\beta}}^{(k)} \approx \hat{\boldsymbol{\beta}}^{(k-1)}$. Let $\tilde{\boldsymbol{\beta}}^{(1)}$ be the final estimate obtained at this stage.

\noindent{\itshape Stage 2.} We first delete the observations $y_{i}$ ($i=1,\cdots,n$) such that $y_{i}>F_{\tilde{\mu}_{i}^{(1)}}^{-1}(1-\alpha/2)$ or $y_{i}<F_{\tilde{\mu}_{i}^{(1)}}^{-1}(\alpha/2)$, where $\tilde{\boldsymbol{\mu}}^{(1)}=g^{-1}\left( \mathbf{X}\tilde{\boldsymbol{\beta}}^{(1)} \right)$ and compute the penalized LST estimate $\tilde{\boldsymbol{\beta}}^{(2)}$ with the reduced sample. Then for each of the deleted observations we check whether $y_{i}>F_{\tilde{\mu}_{i}^{(2)}}^{-1}(1-\alpha/2)$ or $y_{i}<F_{\tilde{\mu}_{i}^{(2)}}^{-1}(\alpha/2)$, where $\tilde{\boldsymbol{\mu}}^{(2)}=g^{-1}\left( \mathbf{X}\tilde{\boldsymbol{\beta}}^{(2)}\right)$. Observations which are not within these bounds are finally eliminated and those which are, are restored to the sample. With the resulting set of observations we compute the penalized LST estimate $\tilde{\boldsymbol{\beta}}$ which is our proposal as a starting value for solving the estimating equations of the MT-estimates.

\subsection{Asymptotic variance}

Let $\eta = \x^\top \boldsymbol{\beta}$, $\dot{m}$, $\dot{(g^{-1})}$ and $\dot{\psi}$ be the derivatives with respect to their arguments and $z(y;\eta,\phi) = (t(y) - \mginv(\eta))/\sqrt{a(\phi)}$.
\begin{align*}
\Psi(y, \x, \boldsymbol{\beta}) & = \frac{\partial}{\partial \eta} \rho(z(y;\eta,\phi)) \nabla_{\boldsymbol{\beta}} \eta \\
& = \psi(z(y;\eta,\phi)) \dot{m}(g^{-1}(\eta)) \dot{(g^{-1})}(\eta) a(\phi)^{-1/2} \x \\
& = \psi(z(y;\eta,\phi)) K_1(\eta) a(\phi)^{-1/2} \x
\end{align*}
where we let $K_1(\eta) = \dot{m}(g^{-1}(\eta)) \dot{(g^{-1})}(\eta)$. Hence
\begin{align*}
\mathbf{B}(\boldsymbol{\beta}) & = \E(\Psi(y, \x, \boldsymbol{\beta}) \Psi(y, \x, \boldsymbol{\beta})^\top) \\
& = \var\left(\psi\left(\frac{t(y) - \mginv(\eta)}{\sqrt{a(\phi)}}\right)\right) K_1(\eta)^2  a(\phi)^{-1} \x \x^\top \\
& = B(\eta, \phi) K_1^2(\eta)  a(\phi)^{-1} \x \x^\top
\end{align*}
and since
\begin{equation*}
\frac{\partial}{\partial \eta} K_1(\eta) = \ddot{m}(g^{-1})\dot{(g^{-1})} + \dot{m}(g^{-1})\ddot{(g^{-1})} = K_2(\eta)
\end{equation*}
we have
\begin{align*}
\mathbf{J}(y,\x ,\boldsymbol{\beta}) & = \nabla \Psi(y, \x, \boldsymbol{\beta}) \\
& = \frac{\partial}{\partial \eta} \psi\left(\frac{t(y) - \mginv(\eta)}{\sqrt{a(\phi)}}\right) K_1(\eta) a(\phi)^{-1/2} \x \x^\top \\
& + \psi(z(y;\eta,\phi)) \frac{\partial}{\partial \eta} K_1(\eta) a(\phi)^{-1/2} \x \x^\top \\
& = \dot{\psi}(z(y;\eta,\phi)) K_1^2(\eta) a(\phi)^{-1} \x \x^\top \\
& + \psi(z(y;\eta,\phi)) K_2(\eta) a(\phi)^{-1/2} \x \x^\top \\
& = \left[ \dot{\psi}(z(y;\eta,\phi)) a(\phi)^{-1/2} K_1^2(\eta)  + \psi(z(y;\eta,\phi)) K_2(\eta) \right] a(\phi)^{-1/2} \x \x^\top 
\end{align*}
we have that
\begin{align*}
\mathbf{A}(\boldsymbol{\beta}) & = \E(\mathbf{J}(y,\x ,\boldsymbol{\beta})) \\
& = \left[ \E(\dot{\psi}(z(y;\eta,\phi))) a(\phi)^{-1/2} K_1^2(\eta) + \E(\psi(z(y;\eta,\phi))) K_2(\eta) \right] a(\phi)^{-1/2} \x \x^\top \\
& = \left[ A_1(\eta, \phi) K_1^2(\eta) + A_2(\eta, \phi) K_2(\eta) \right] a(\phi)^{-1/2} \x \x^\top
\end{align*}
So that, the asymptotic variance is given by
\begin{equation*}
\mathbf{A}^{-1}(\boldsymbol{\beta})\mathbf{B}(\boldsymbol{\beta})\mathbf{A}^{-1}(\boldsymbol{\beta}) = B(\eta, \phi) K_1^2(\eta) \left[ A_1(\eta, \phi) K_1^2(\eta) + A_2(\eta, \phi) K_2(\eta) \right]^{-2} (\x \x^\top)^{-1}
\end{equation*}

\clearpage

\section{Supplementary Material: Robust Penalized Estimators for High--Dimensional Generalized Linear Models}

\subsection{Complete results of the Monte Carlo study}
\label{app:montecarlo}
In this Section we report  full results for simulation settings AVY and AMR together with the results of a third setting namely AVY2 which is similar to ``model 2'' considered in \cite{agostinelli2019initial}. The difference with AVY described in Section 6 ``Monte Carlo study'' of the main document is in the value of parameters, which in AVY2 are given by $\boldsymbol{\beta}_* = 2 \mathbf{e}_1 + \mathbf{e}_2$.

\subsubsection{False Negative and False Positive for Lasso Methods}
\label{app:FN}
In Tables \ref{table:fp1} to \ref{table:fp3} we report a summary  of the performance of Lasso methods for variable selection for the AMR setting. These tables are similar to Table 1 in \cite{avella2018robust}. {For each measure, we give the median over 100 simulations and its median absolute deviation in parentheses}. Size is the number of selected variables in the final model, $\#FN$ is the number of parameters that are incorrectly estimated as $0$, $\#FP$ is the number of parameters that are zero but their estimates are not. Tables \ref{table:fp4} to \ref{table:fp6} give the selected value of the penalty parameter, the BIC  and the degrees of freedom for this value.
\begin{table}[ht]
\resizebox{\textwidth}{!}{
\begin{tabular}{lllllllllllllll}
  \hline
  \hline
& $p=100$ & & \multicolumn{3}{c}{$\nu=0$} & \multicolumn{3}{c}{$\nu=5$} & \multicolumn{3}{c}{$\nu=10$} & \multicolumn{3}{c}{$\nu=100$} \\
  $n$ & $\varepsilon$ & Method & Size & $\#FN$ & $\#FP$ & Size & $\#FN$ & $\#FP$ & Size & $\#FN$ & $\#FP$ & Size & $\#FN$ & $\#FP$ \\
\hline
50 & 0.05 & MT  & 1 & 3 & 0 & 1.03 & 2.97 & 0.03 & 1 & 3 & 0 & 1 & 3 & 0 \\ 
    &   &   & (0) & (0) & (0) & (0.3) & (0.3) & (0.3) & (0) & (0) & (0) & (0) & (0) & (0) \\ 
    &   & ML  & 13.39 & 0 & 12.39 & 21.01 & 0.11 & 20.01 & 22.57 & 0.49 & 21.57 & 14.64 & 2.23 & 13.64 \\ 
    &   &   & (6.3) & (0) & (6.3) & (11.16) & (0.35) & (11.16) & (12.93) & (0.61) & (12.93) & (11.19) & (0.92) & (11.19) \\ 
    &   & RQL  & 12.31 & 0 & 11.31 & 13.45 & 0.02 & 12.45 & 13.71 & 0.02 & 12.71 & 13.82 & 0.06 & 12.82 \\ 
    &   &   & (3.66) & (0) & (3.66) & (4.39) & (0.14) & (4.39) & (5.07) & (0.14) & (5.07) & (5.43) & (0.28) & (5.43) \\
\hline\hline  
\end{tabular}
}\caption{Summary of results for Lasso methods for AMR setting and $p=100$}\label{table:fp1}
\end{table}

\begin{table}[ht]
\resizebox{\textwidth}{!}{
\begin{tabular}{lllllllllllllll}
  \hline
  \hline
& $p=400$ &        & \multicolumn{3}{c}{$\nu=0$} & \multicolumn{3}{c}{$\nu=5$} & \multicolumn{3}{c}{$\nu=10$} & \multicolumn{3}{c}{$\nu=100$} \\
  $n$ & $\varepsilon$ & Method & Size & $\#FN$ & $\#FP$ & Size & $\#FN$ & $\#FP$ & Size & $\#FN$ & $\#FP$ & Size & $\#FN$ & $\#FP$ \\
\hline  
50 & 0.05 & MT  & 1 & 3 & 0 & 1 & 3 & 0 & 1 & 3 & 0 & 1 & 3 & 0 \\ 
    &   &   & (0) & (0) & (0) & (0) & (0) & (0) & (0) & (0) & (0) & (0) & (0) & (0) \\ 
    &   & ML  & 16.66 & 0.04 & 15.66 & 29.9 & 0.16 & 28.9 & 30.91 & 0.57 & 29.91 & 19.43 & 2.46 & 18.43 \\ 
    &   &   & (7.82) & (0.2) & (7.82) & (13.23) & (0.42) & (13.23) & (15.41) & (0.79) & (15.41) & (12.08) & (0.87) & (12.08) \\ 
    &   & RQL  & 13.81 & 0 & 12.81 & 14.59 & 0.01 & 13.59 & 14.58 & 0.02 & 13.58 & 14.84 & 0.02 & 13.84 \\ 
    &   &   & (3.3) & (0) & (3.3) & (3.91) & (0.1) & (3.91) & (3.84) & (0.2) & (3.84) & (4.75) & (0.2) & (4.75) \\
\hline\hline  
\end{tabular}
}\caption{Summary of results for Lasso methods for AMR setting and $p=400$}\label{table:fp2}
\end{table}

\begin{table}[ht]
\resizebox{\textwidth}{!}{  
\begin{tabular}{lllllllllllllll}
  \hline
  \hline
& $p=1600$ & & \multicolumn{3}{c}{$\nu=0$} & \multicolumn{3}{c}{$\nu=5$} & \multicolumn{3}{c}{$\nu=10$} & \multicolumn{3}{c}{$\nu=100$} \\
  $n$ & $\varepsilon$ & Method & Size & $\#FN$ & $\#FP$ & Size & $\#FN$ & $\#FP$ & Size & $\#FN$ & $\#FP$ & Size & $\#FN$ & $\#FP$ \\
\hline  
  50 & 0.05 & MT  & 1 & 3 & 0 & 1 & 3 & 0 & 1 & 3 & 0 & 1 & 3 & 0 \\ 
    &   &   & (0) & (0) & (0) & (0) & (0) & (0) & (0) & (0) & (0) & (0) & (0) & (0) \\ 
    &   & ML  & 19.56 & 0.08 & 18.56 & 33.71 & 0.27 & 32.71 & 34.62 & 0.8 & 33.62 & 23.48 & 2.79 & 22.48 \\ 
    &   &   & (8.09) & (0.31) & (8.09) & (13.93) & (0.49) & (13.93) & (15.58) & (0.89) & (15.58) & (13.39) & (0.46) & (13.39) \\ 
    &   & RQL  & 15.71 & 0.04 & 14.71 & 16.62 & 0.05 & 15.62 & 16.86 & 0.07 & 15.86 & 17.03 & 0.11 & 16.03 \\ 
    &   &   & (2.85) & (0.2) & (2.85) & (4.13) & (0.26) & (4.13) & (4.56) & (0.33) & (4.56) & (5.3) & (0.4) & (5.3) \\
\hline\hline 
\end{tabular}
}\caption{Summary of results for Lasso methods for AMR setting and $p=1600$}\label{table:fp3}
\end{table}

\begin{table}[ht]
\resizebox{\textwidth}{!}{   
\begin{tabular}{lllllllllllllll}
  \hline
  \hline
& $p=100$ & & \multicolumn{3}{c}{$\nu=0$} & \multicolumn{3}{c}{$\nu=5$} & \multicolumn{3}{c}{$\nu=10$} & \multicolumn{3}{c}{$\nu=100$} \\
$n$ & $\varepsilon$ & Method & $\hat{\lambda}$ & BIC & df & $\hat{\lambda}$ & BIC & df & $\hat{\lambda}$ & BIC & df & $\hat{\lambda}$ & BIC & df \\
\hline  
50 & 0.05 & MT  & 798.38 & 28.81 & 1 & 805.04 & 28.71 & 1.03 & 830.67 & 28.97 & 1 & 816.45 & 28.99 & 1 \\ 
    &   &   & (152.3) & (2.5) & (0) & (167.26) & (2.55) & (0.3) & (144.79) & (2.59) & (0) & (166.17) & (2.59) & (0) \\ 
    &   & ML  & 0.25 & 2.71 & 12.39 & 0.29 & 6.53 & 20.01 & 0.49 & 17.21 & 21.57 & 7.11 & 305.33 & 13.64 \\ 
    &   &   & (0.12) & (0.94) & (6.3) & (0.21) & (3.54) & (11.16) & (0.5) & (10.48) & (12.93) & (8.02) & (213.27) & (11.19) \\ 
    &   & RQL  & 17.13 & 171.79 & 11.31 & 17.77 & 194.61 & 12.45 & 17.72 & 260.66 & 12.71 & 17.71 & 743.19 & 12.82 \\ 
    &   &   & (2.63) & (27.79) & (3.66) & (2.61) & (30.6) & (4.39) & (2.61) & (49.07) & (5.07) & (2.63) & (216.81) & (5.43) \\
\hline\hline  
\end{tabular}
}\caption{Selected value of tuning parameter and the corresponding BIC and df for Lasso methods for AMR seting and $p=100$}\label{table:fp4}
\end{table}

\begin{table}[ht]
\resizebox{\textwidth}{!}{   
\begin{tabular}{lllllllllllllll}
  \hline
  \hline
& $p=400$ & & \multicolumn{3}{c}{$\nu=0$} & \multicolumn{3}{c}{$\nu=5$} & \multicolumn{3}{c}{$\nu=10$} & \multicolumn{3}{c}{$\nu=100$} \\
$n$ & $\varepsilon$ & Method & $\hat{\lambda}$ & BIC & df & $\hat{\lambda}$ & BIC & df & $\hat{\lambda}$ & BIC & df & $\hat{\lambda}$ & BIC & df \\
\hline  
50 & 0.05 & MT  & 81.92 & 28.79 & 1 & 82.85 & 28.73 & 1 & 82.77 & 28.93 & 1 & 83.48 & 28.97 & 1 \\ 
    &   &   & (10.85) & (2.59) & (0) & (14.47) & (2.72) & (0) & (14.18) & (2.71) & (0) & (13.46) & (2.69) & (0) \\ 
    &   & ML  & 0.31 & 3.02 & 15.66 & 0.27 & 7.22 & 28.9 & 0.49 & 18.21 & 29.91 & 6.68 & 311.27 & 18.43 \\ 
    &   &   & (0.15) & (1.14) & (7.82) & (0.27) & (3.56) & (13.23) & (0.66) & (10.69) & (15.41) & (7.68) & (237.85) & (12.08) \\ 
    &   & RQL  & 13.58 & 178.86 & 12.81 & 13.9 & 208.17 & 13.59 & 13.85 & 275.06 & 13.58 & 13.84 & 776.34 & 13.84 \\ 
    &   &   & (2.5) & (26.43) & (3.3) & (2.24) & (29.95) & (3.91) & (2.24) & (50.39) & (3.84) & (2.23) & (217.52) & (4.75) \\
\hline\hline  
\end{tabular}
}\caption{Selected value of tuning parameter and the corresponding BIC and df for Lasso methods for AMR seting and $p=400$}\label{table:fp5}
\end{table}

\begin{table}[ht]
\resizebox{\textwidth}{!}{   
\begin{tabular}{lllllllllllllll}
  \hline
  \hline
& $p=1600$ & & \multicolumn{3}{c}{$\nu=0$} & \multicolumn{3}{c}{$\nu=5$} & \multicolumn{3}{c}{$\nu=10$} & \multicolumn{3}{c}{$\nu=100$} \\  
$n$ & $\varepsilon$ & Method & $\hat{\lambda}$ & BIC & df & $\hat{\lambda}$ & BIC & df & $\hat{\lambda}$ & BIC & df & $\hat{\lambda}$ & BIC & df \\
\hline  
  50 & 0.05 & MT  & 80.9 & 28.71 & 1 & 84.74 & 28.6 & 1 & 83.1 & 28.79 & 1 & 83.93 & 28.84 & 1 \\ 
    &   &   & (11.28) & (2.31) & (0) & (9.02) & (2.32) & (0) & (12.43) & (2.34) & (0) & (11.04) & (2.32) & (0) \\ 
    &   & ML  & 0.37 & 4 & 18.56 & 0.3 & 8.22 & 32.71 & 0.54 & 20.26 & 33.62 & 7.47 & 328.33 & 22.48 \\ 
    &   &   & (0.17) & (1.38) & (8.09) & (0.31) & (4.28) & (13.93) & (0.67) & (12.07) & (15.58) & (8.13) & (223.86) & (13.39) \\ 
    &   & RQL  & 10.94 & 202.15 & 14.71 & 11.4 & 225.11 & 15.62 & 11.26 & 294.19 & 15.86 & 11.24 & 803.11 & 16.03 \\ 
    &   &   & (2.34) & (32.74) & (2.85) & (1.86) & (34.99) & (4.13) & (1.77) & (53.43) & (4.56) & (1.8) & (215.95) & (5.3) \\
\hline\hline  
\end{tabular}
}\caption{Selected value of tuning parameter and the corresponding BIC and df for Lasso methods for AMR seting and $p=1600$}\label{table:fp6}
\end{table}

\clearpage

\subsubsection{Complete results of the Monte Carlo study}
\label{app:Ridge}

In this section we report the complete results of the Monte Carlo Study. We analyze the behaviour of the estimators under contamination in models AVY, AVY2 and AMR. We also present a table with the MSE for samples without contamination. The models and contamination schemes are described in Section 6 ``Monte Carlo study'' of the main document.

In Figures \ref{fig:msesimu1p10Ridge} to \ref{fig:msesimuRQLp1600Ridge}  we plot the MSE of MT Ridge and ML Ridge as a function of the contamination $y_0$ for the different models, while in Figures \ref{fig:msesimu1p10Lasso} to \ref{fig:msesimuRQLp1600Lasso}  we plot the MSE of MT Lasso, RQL Lasso and ML Lasso. 
We note the very high robustness of MT Ridge and  MT Lasso in all scenarios.

Table \ref{table:msenoout} gives mean squared errors for clean samples. We note that MT Ridge gives smaller MSE than ML Ridge in most situations. In the case of Lasso methods, we note that the performance of MT Lasso improves when the sample size increases, while RQL and ML give relatively small MSE for all sample sizes. 
 
 Figures \ref{fig:msesimuRQLp100Lasso} to \ref{fig:msesimuRQLp1600Lasso} show that in the presence of outliers, the MSE of ML Lasso and RQL Lasso increases without bound while the MSE of MT Lasso remains bounded. 

 In some cases, the boundedness of the mean squared errors of MT Lasso is not apparent. Nevertheless, we can see that they increase at a much slower rate than the mean squared errors of ML Lasso and RQL Lasso.

\begin{table}[ht]
\footnotesize
\begin{tabular}{llllllllll}
		\hline
		&&& \multicolumn{3}{c}{Ridge methods} & \multicolumn{3}{c}{Lasso methods} \\ 
				& &$p$ & $n=$100 & $n=$400 & $n=$1000 & $n=$100 & $n=$400 & $n=$1000 \\ 
		\hline
		\multirow{6}{*}{AVY} & MT& 10 & 0.2704 & 0.0452 & 0.0115 &  0.0773 & 0.0204 & 0.0074 \\ 
		&ML&10 & 0.0856 & 0.03 & 0.0206 &  0.0469 & 0.0104 & 0.0045 \\ 
		&RQL&10 & - & - & - & 0.047 & 0.0113 & 0.0042 \\ 
		\cline{2-9}
		& MT&50 & 0.4529 & 0.1243 & 0.0434 &  0.151 & 0.0327 & 0.0126 \\ 
		& ML&50 & 0.4208 & 0.0976 & 0.0454 &  0.0794 & 0.0188 & 0.0072 \\ 
		& RQL&50 & - & - & - &  0.0767 & 0.0213 & 0.0085 \\ 
\hline
		\multirow{6}{*}{AVY2} 
		 & MT&10 & 0.0212 & 0.0035 & 0.0011 & 0.0103 & 0.0026 & 0.0009 \\ 
		& ML&10 & 0.0292 & 0.0209 & 0.0172  & 0.0054 & 0.0018 & 0.0006 \\ 
		& RQL&10 & - & - & -  & 0.0057 & 0.002 & 0.0006 \\ 
	\cline{2-9}
		 & MT&50 & 0.1258 & 0.0156 & 0.005  & 0.0202 & 0.004 & 0.0014 \\ 
		& ML&50 & 0.13 & 0.0372 & 0.0241  & 0.0123 & 0.0026 & 0.001 \\ 
	&	RQL&50 & - & - & - & 0.0119 & 0.0027 & 0.0011 \\ \hline 
	&&	$p$ & $n=$50 & $n=$100 & $n=$400   & $n=$50 & $n=$100 & $n=$400 \\ 
		\hline 
				\multirow{9}{*}{AMR} 
		 &MT&50 & 7.3578 & 7.0669 & 6.7142  & 4.9747 & 0.7999 & 0.554 \\ 
		&ML&50 & 7.7215 & 7.4187 & 1.6  & 0.995 & 0.8351 & 0.7575 \\ 
		&RQL&50 & - & - & -  & 0.8315 & 0.7613 & 0.6966 \\ 
		\cline{2-9}
		 &MT & 400 &7.4941 & 7.2262 & 6.8469  & 6.602 & 1.4071 & 0.5189 \\ 
		&ML&400 & 8.0409 & 7.7678 & 7.692 &  0.9115 & 0.7091 & 0.6964 \\ 
		&RQL&400 & - & - & -  & 0.8262 & 0.6268 & 0.629 \\ 
		\cline{2-9}
		& MT &1600& 7.9792 & 7.6376 & 7.285  & 7.0163 & 2.4811 & 0.4944 \\ 
		&ML&1600 & 8.3787 & 8.0553 & 7.9132 & 1.0702 & 0.7063 & 0.6321 \\ 
		&RQL&1600 & - & - & -  & 0.9806 & 0.5868 & 0.5616 \\ 
		\hline
	\end{tabular}
	\caption{Mean squared errors in samples with no outliers} \label{table:msenoout}
\end{table}

\begin{figure}[tp]
	\begin{center}
		\footnotesize
		\renewcommand{\arraystretch}{0.22}
		\newcolumntype{M}{>{\centering\arraybackslash}m{\dimexpr.01\linewidth-1\tabcolsep}}
		\newcolumntype{G}{>{\centering\arraybackslash}m{\dimexpr.38\linewidth-1\tabcolsep}}
		\begin{tabular}{GGG}
			$n=100,\, \epsilon=0.05$ & $n=100,\, \epsilon=0.1$  & $n=100,\, \epsilon=0.15$ \\
			\includegraphics[scale=0.29]{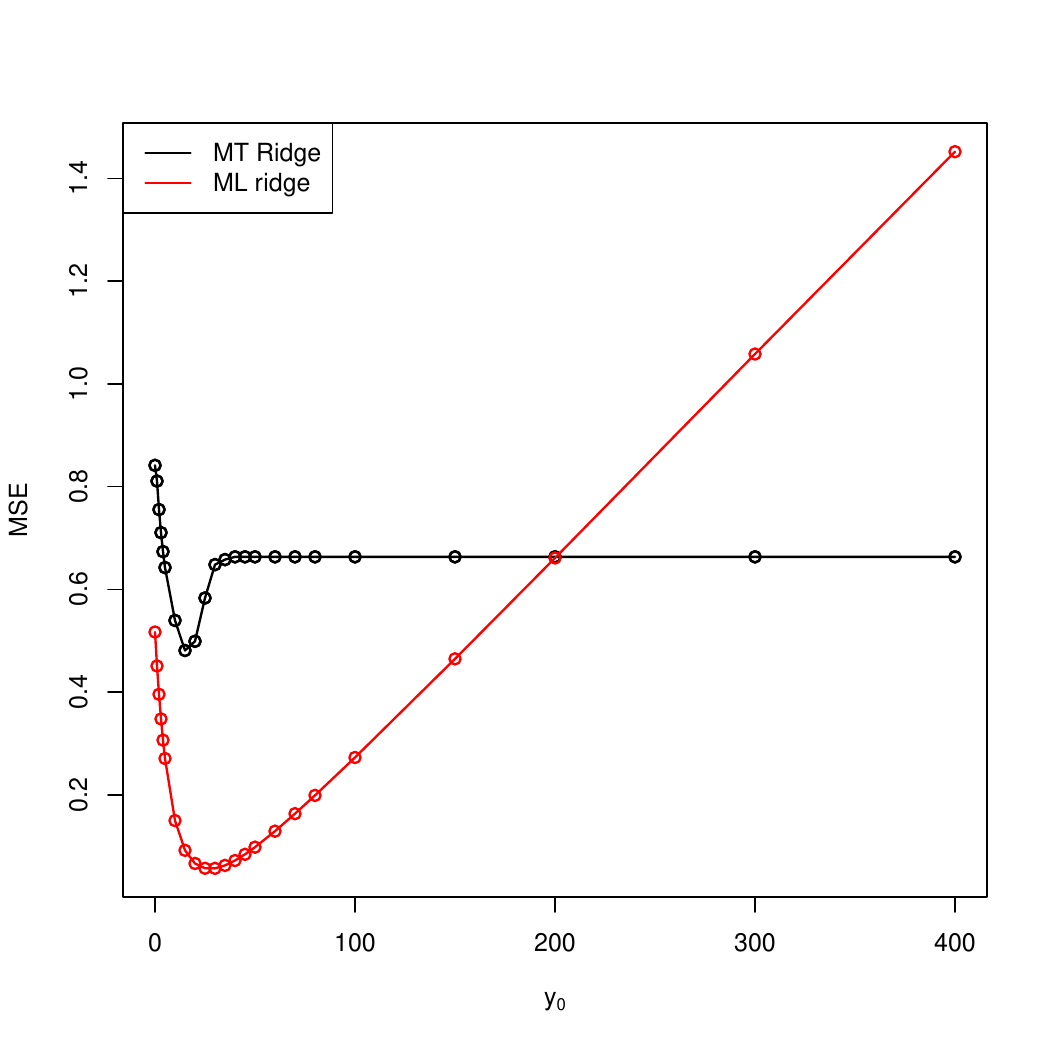} 
			&  \includegraphics[scale=0.29]{simu1p10-mse-ridge-100-0-1}
			&   \includegraphics[scale=0.29]{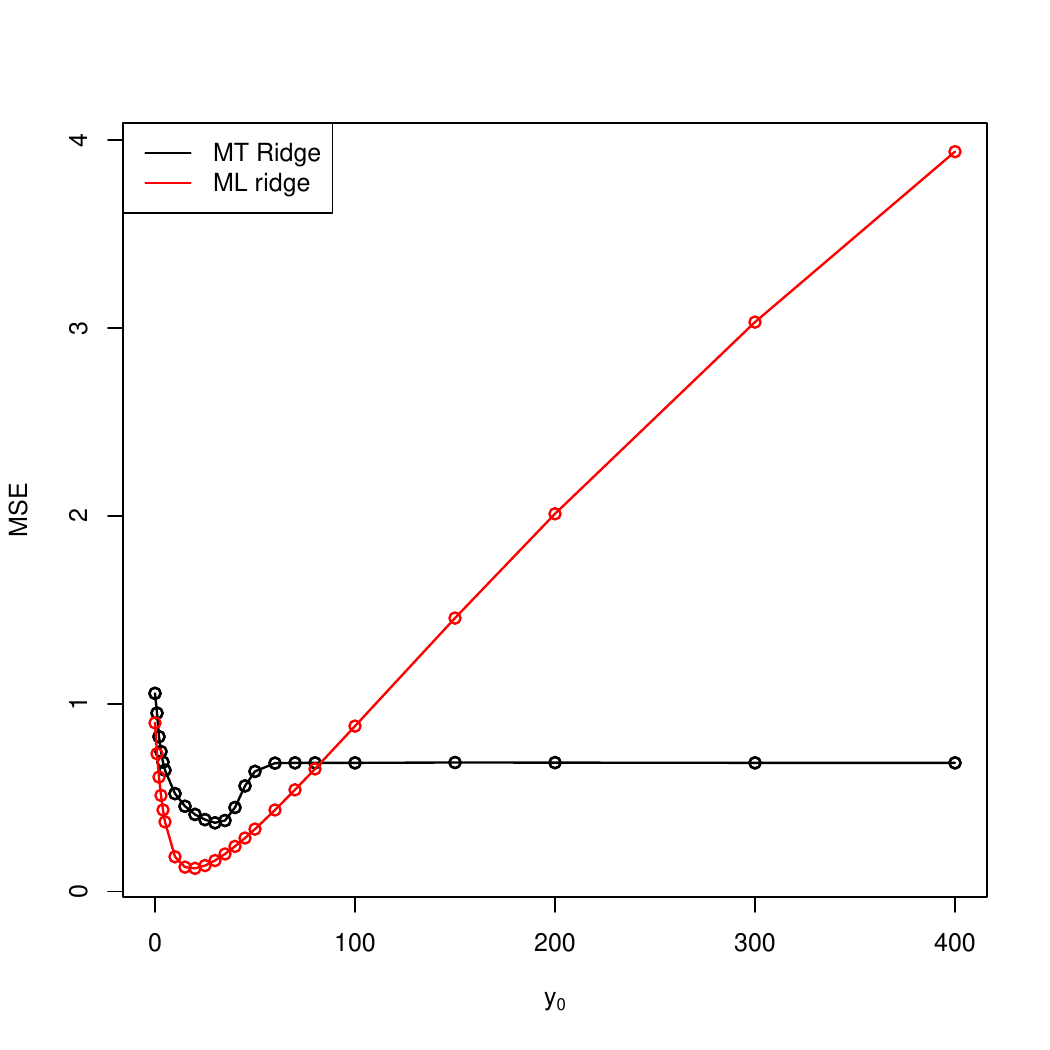}
			\\
			$n=400,\, \epsilon=0.05$ & $n=400,\, \epsilon=0.1$  & $n=400,\, \epsilon=0.15$ \\
			\includegraphics[scale=0.3]{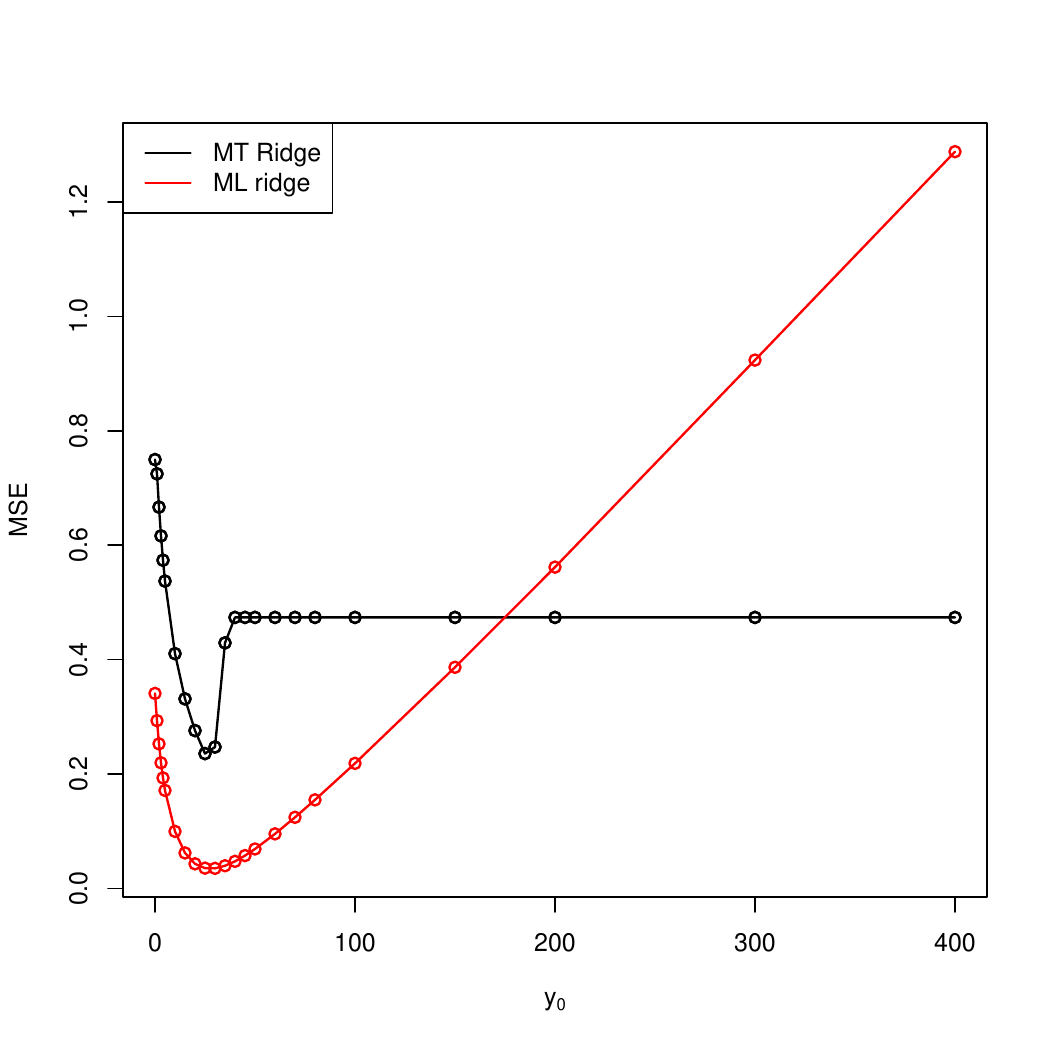}
			&   \includegraphics[scale=0.3]{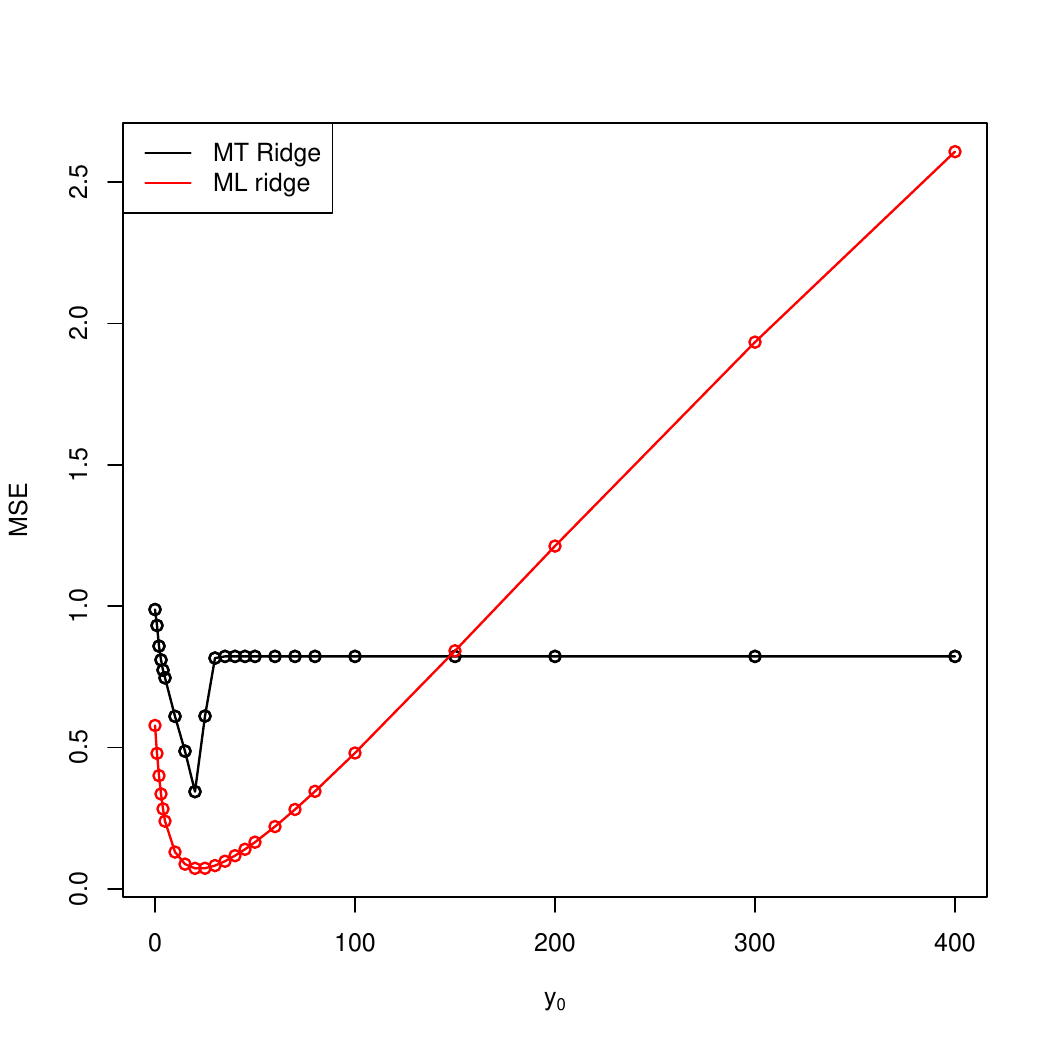} 
			&   \includegraphics[scale=0.3]{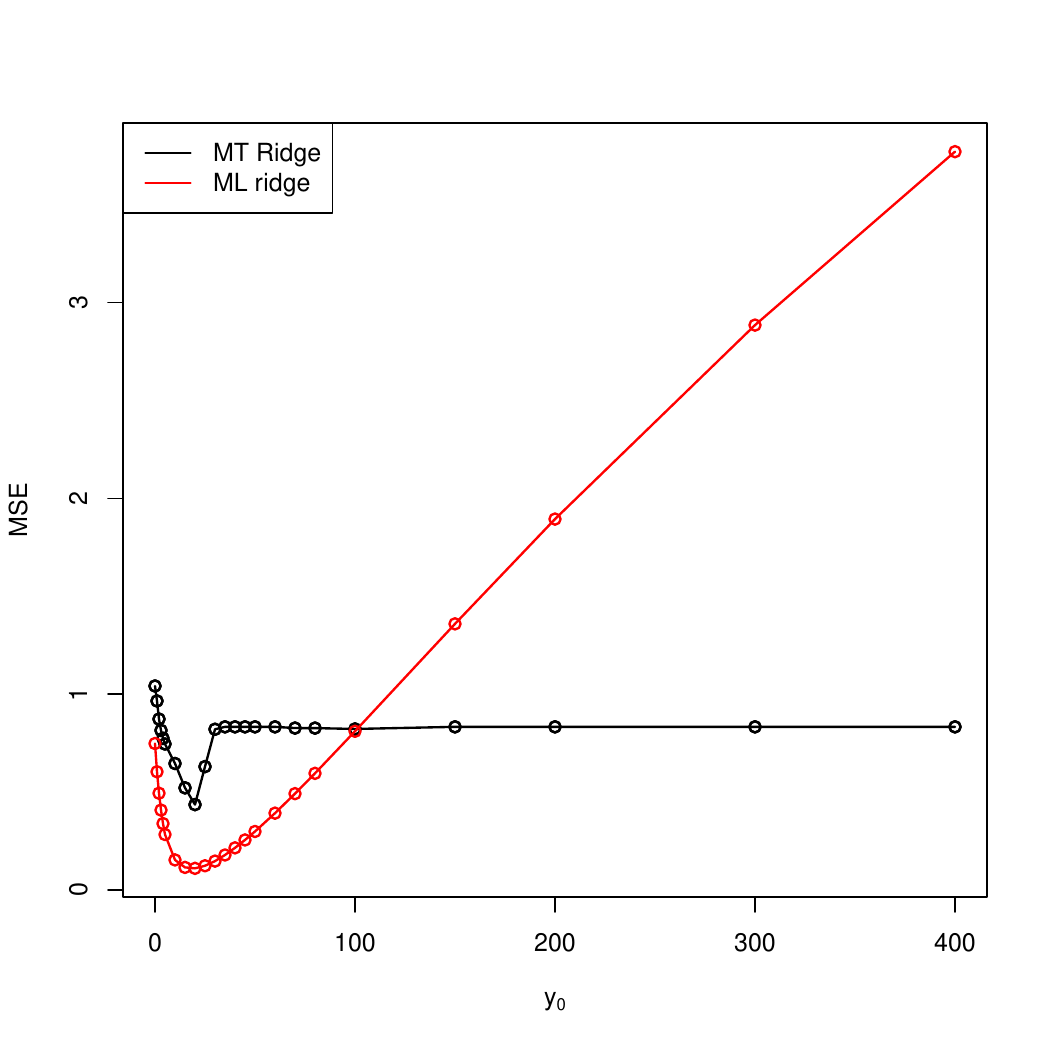} 
			\\
			$n=1000,\, \epsilon=0.05$ & $n=1000,\, \epsilon=0.1$  & $n=1000,\, \epsilon=0.15$ \\
			\includegraphics[scale=0.3]{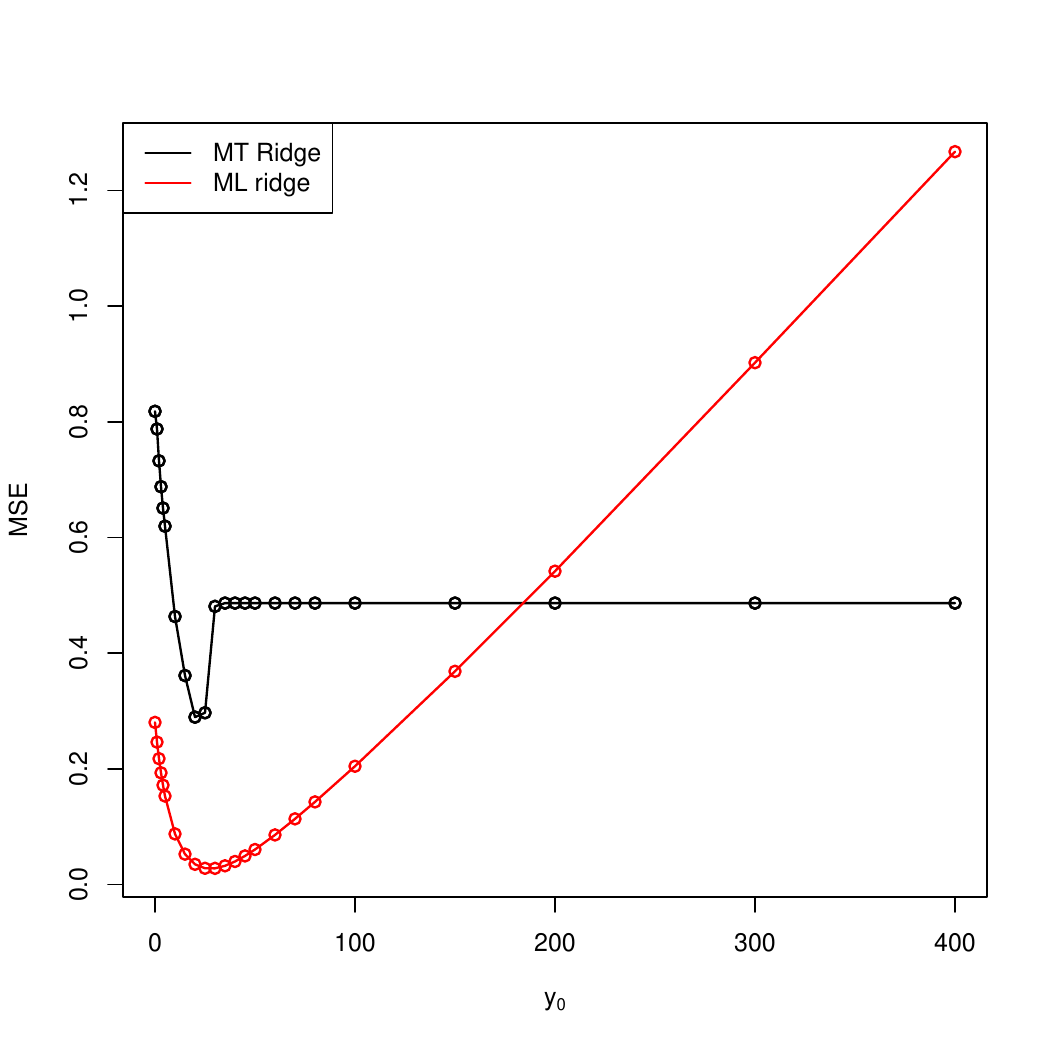}
			&   \includegraphics[scale=0.3]{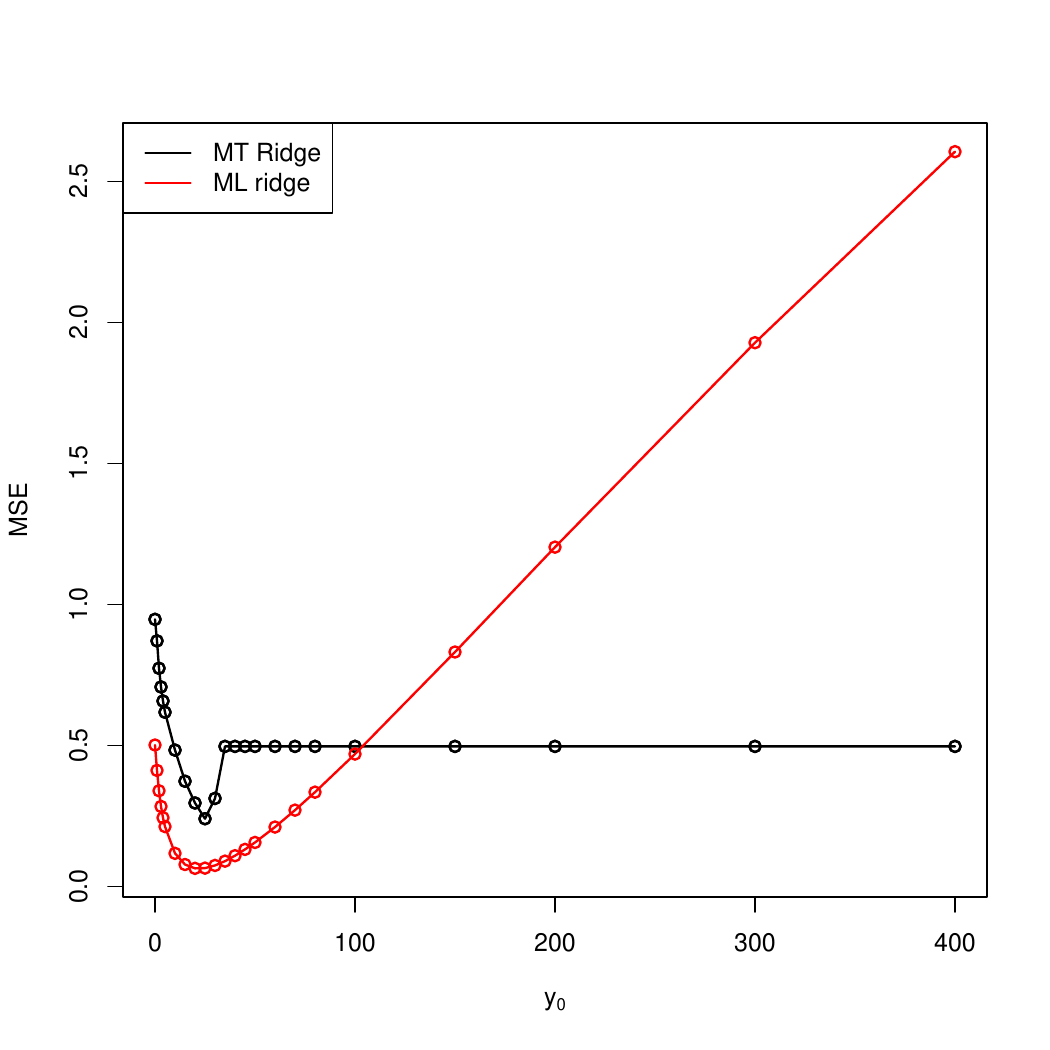} 
			&   \includegraphics[scale=0.3]{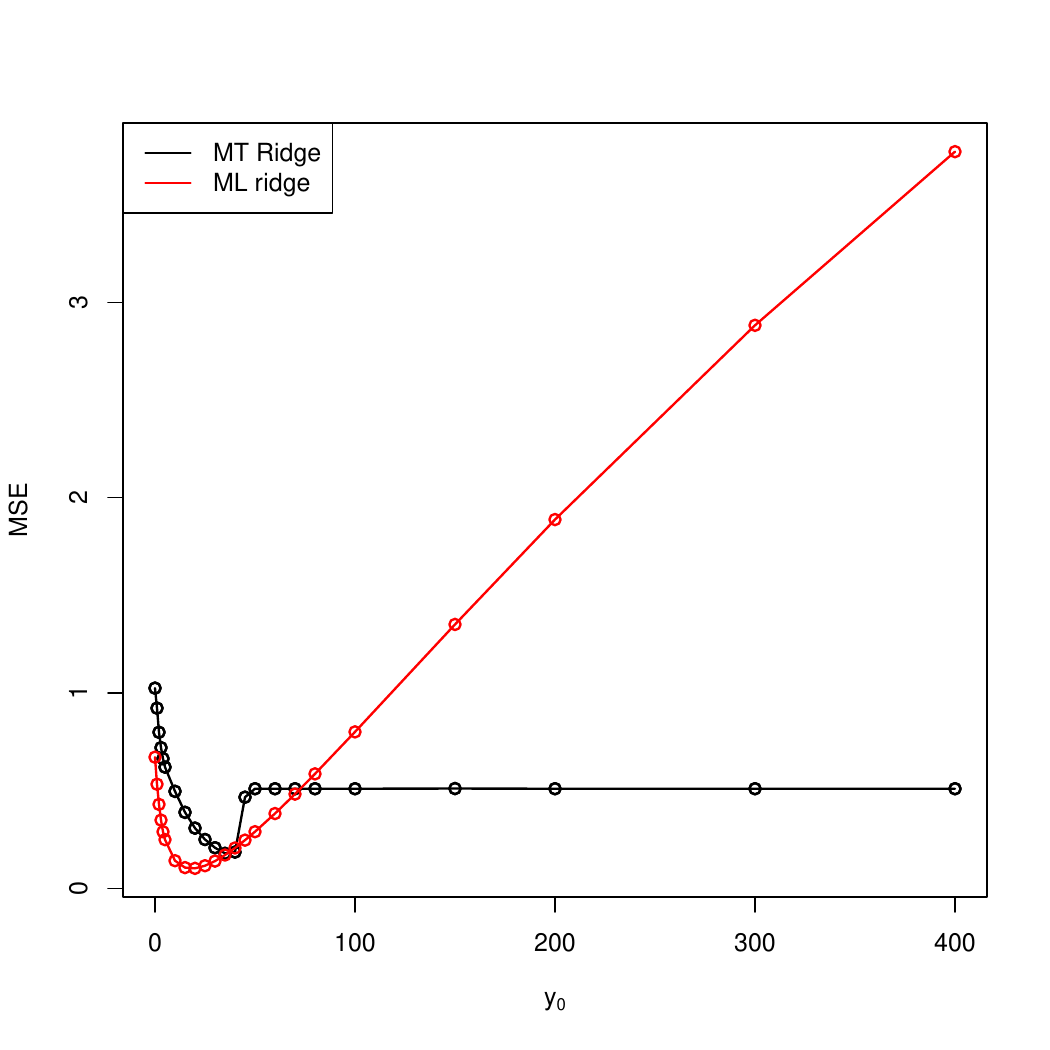} 			
		\end{tabular}
		\caption{\small \label{fig:msesimu1p10Ridge}  Mean squared errors for model AVY and Ridge methods, $p=10$.}
	\end{center} 
\end{figure}

\begin{figure}[tp]
	\begin{center}
		\footnotesize
		\renewcommand{\arraystretch}{0.22}
		\newcolumntype{M}{>{\centering\arraybackslash}m{\dimexpr.01\linewidth-1\tabcolsep}}
		\newcolumntype{G}{>{\centering\arraybackslash}m{\dimexpr.38\linewidth-1\tabcolsep}}
		\begin{tabular}{GGG}
			$n=100,\, \epsilon=0.05$ & $n=100,\, \epsilon=0.1$  & $n=100,\, \epsilon=0.15$ \\
			\includegraphics[scale=0.29]{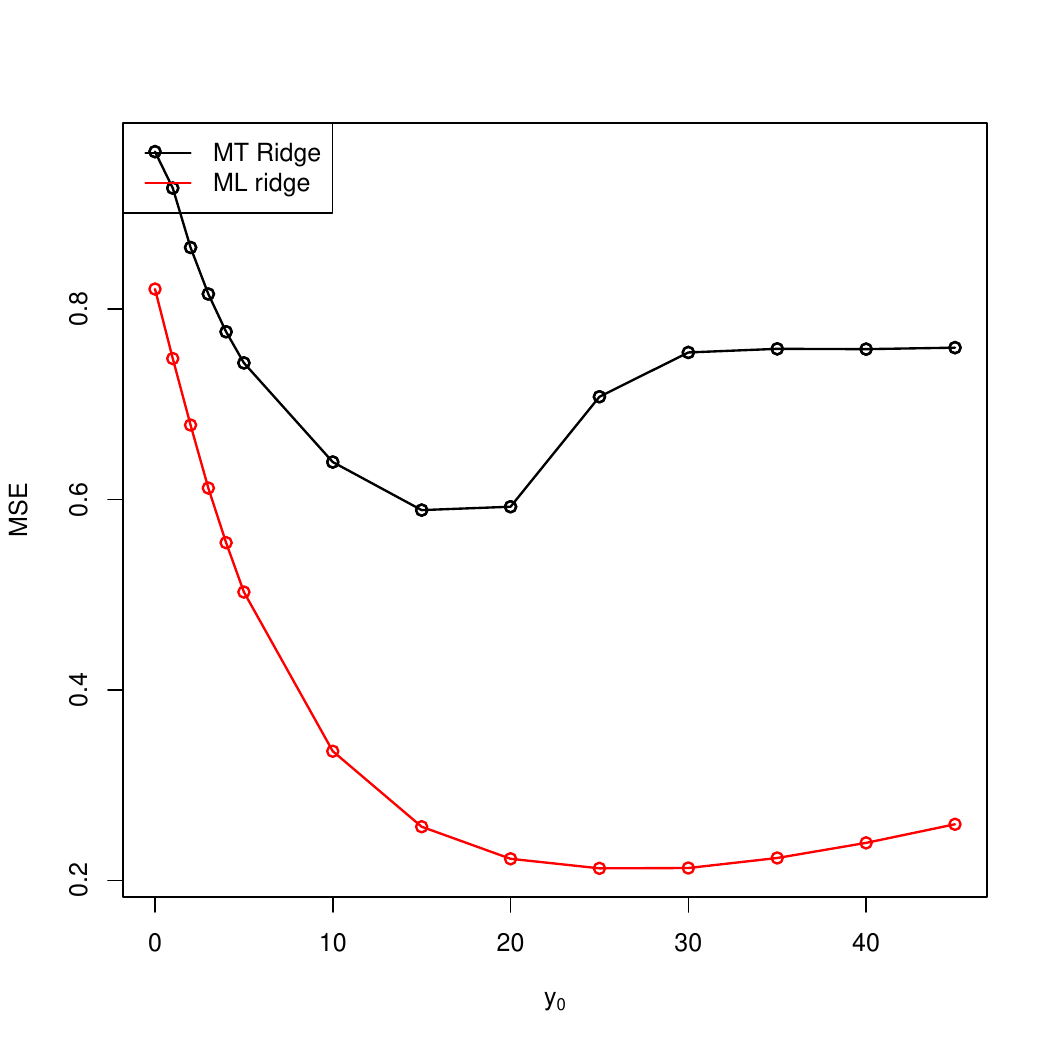} 
			&  \includegraphics[scale=0.29]{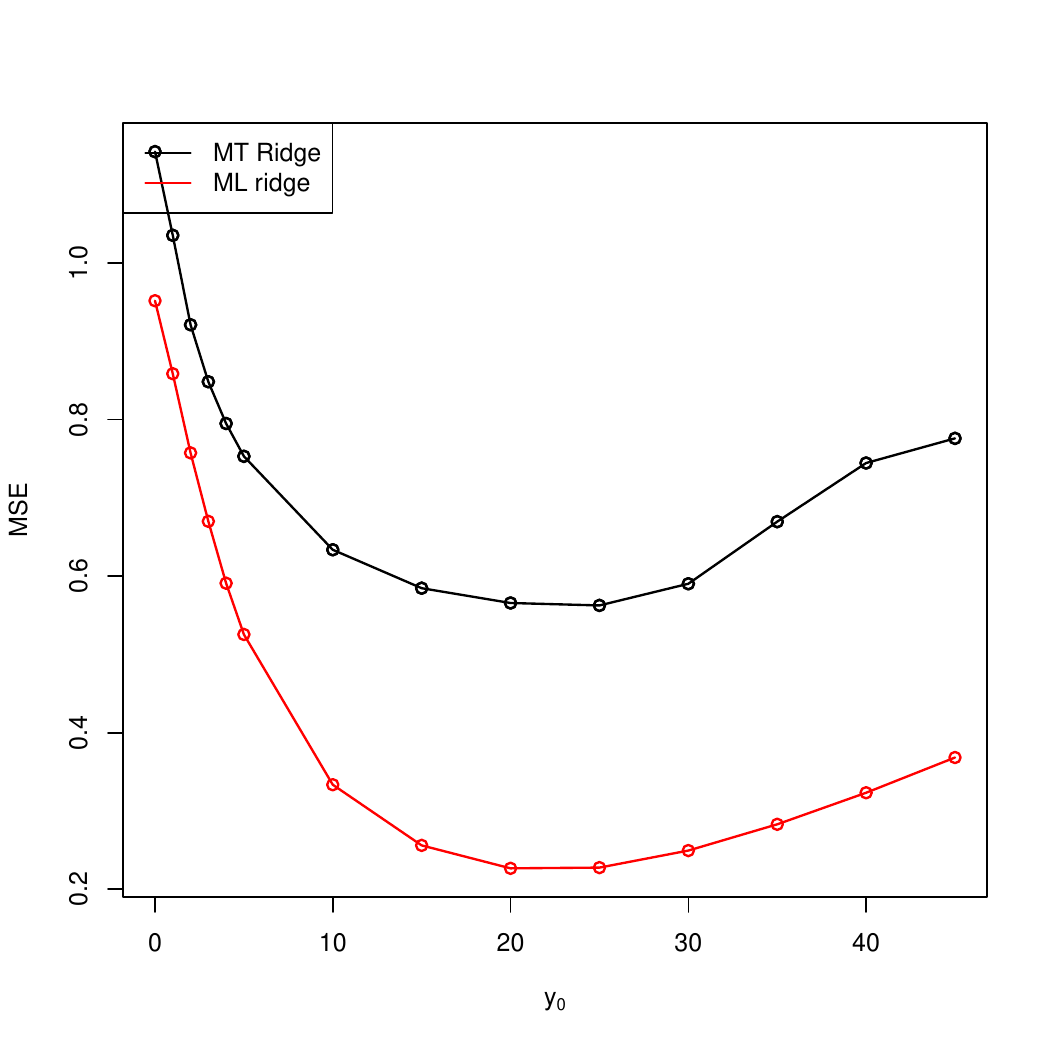}
			&   \includegraphics[scale=0.29]{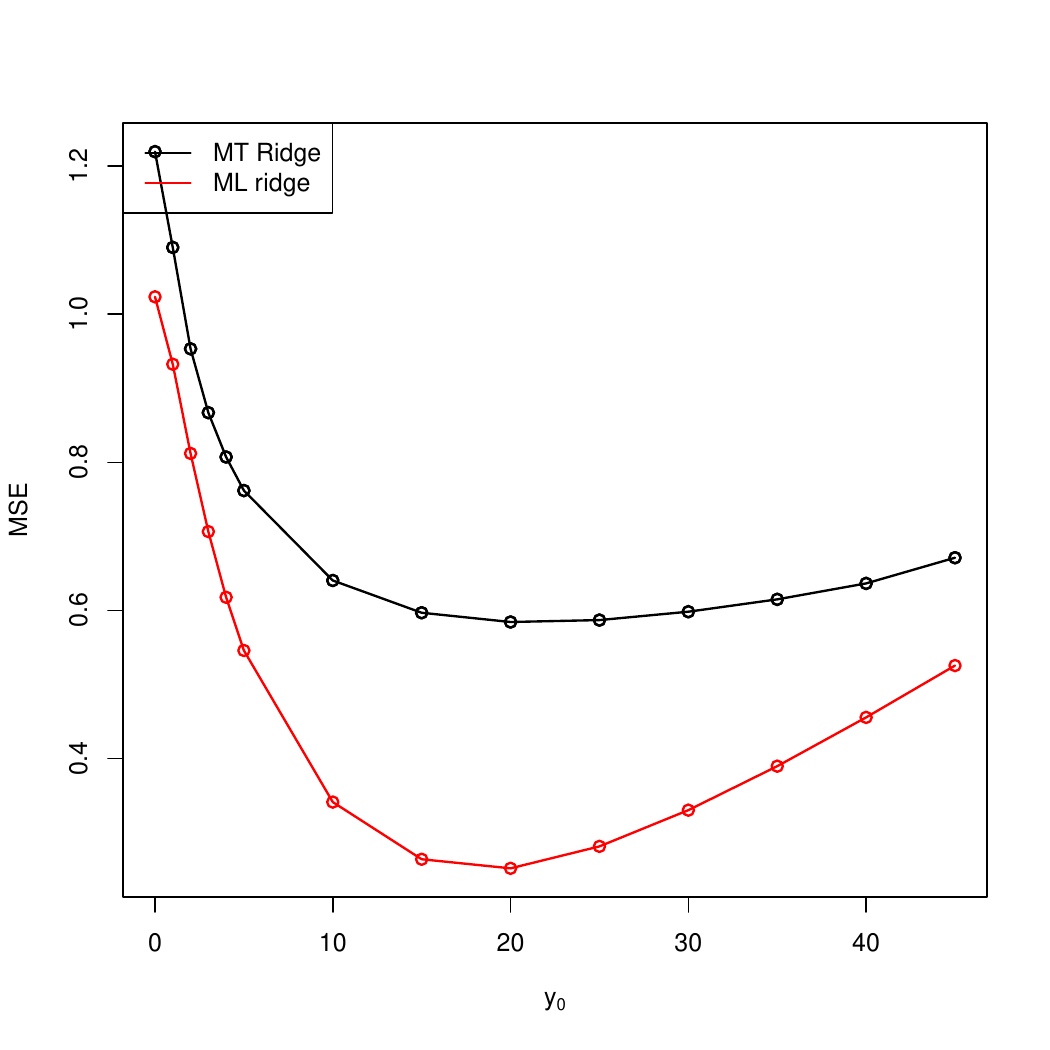}
			\\
			$n=400,\, \epsilon=0.05$ & $n=400,\, \epsilon=0.1$  & $n=400,\, \epsilon=0.15$ \\
			\includegraphics[scale=0.3]{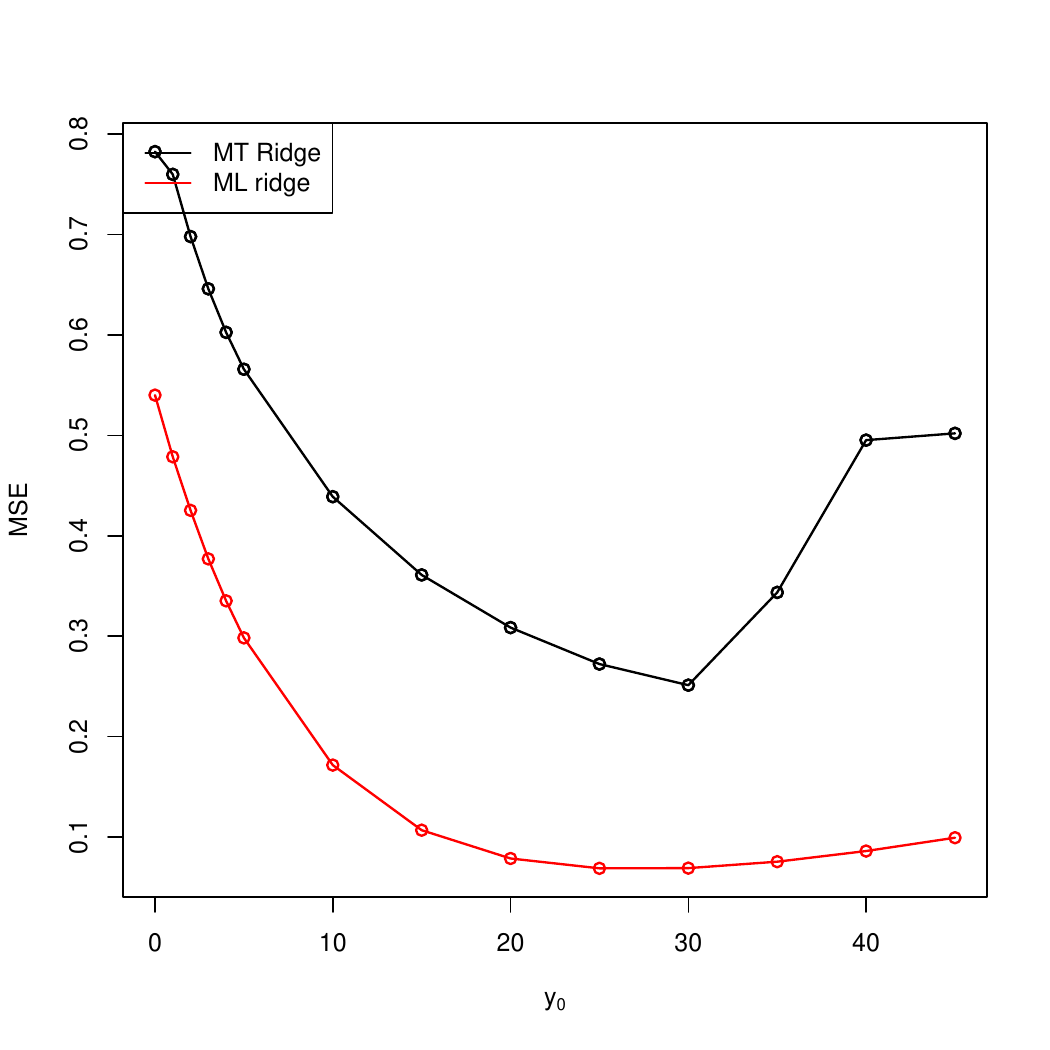}
			&   \includegraphics[scale=0.3]{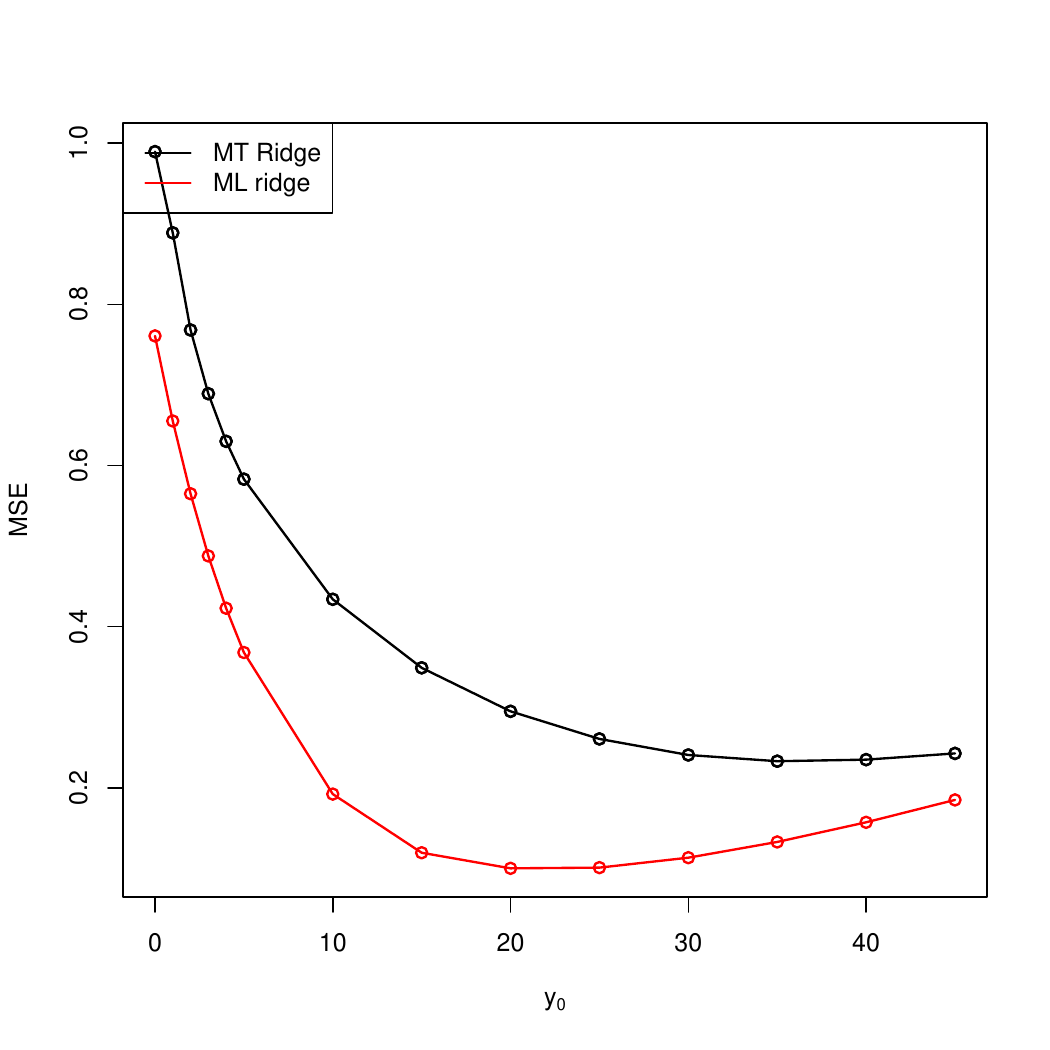} 
			&   \includegraphics[scale=0.3]{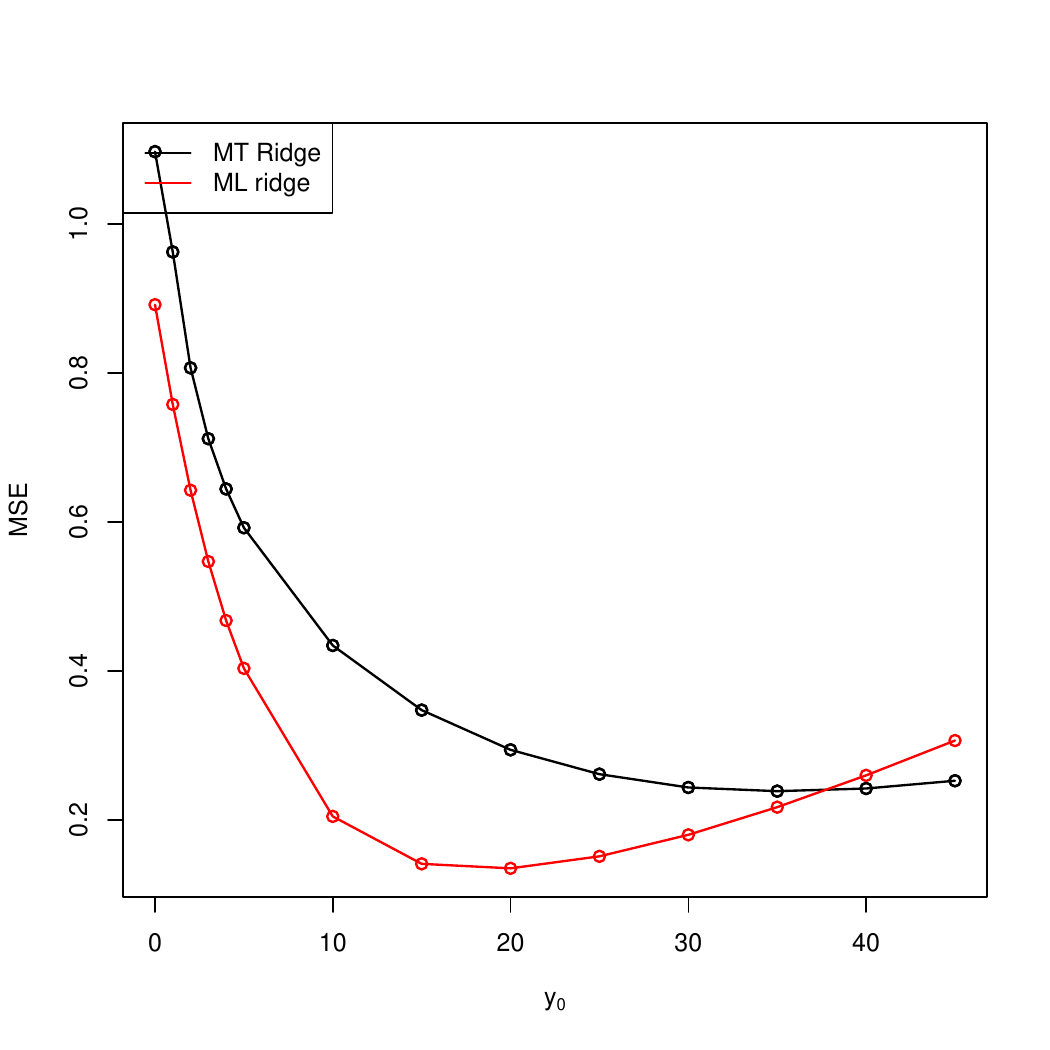} 
			\\
			$n=1000,\, \epsilon=0.05$ & $n=1000,\, \epsilon=0.1$  & $n=1000,\, \epsilon=0.15$ \\
			\includegraphics[scale=0.3]{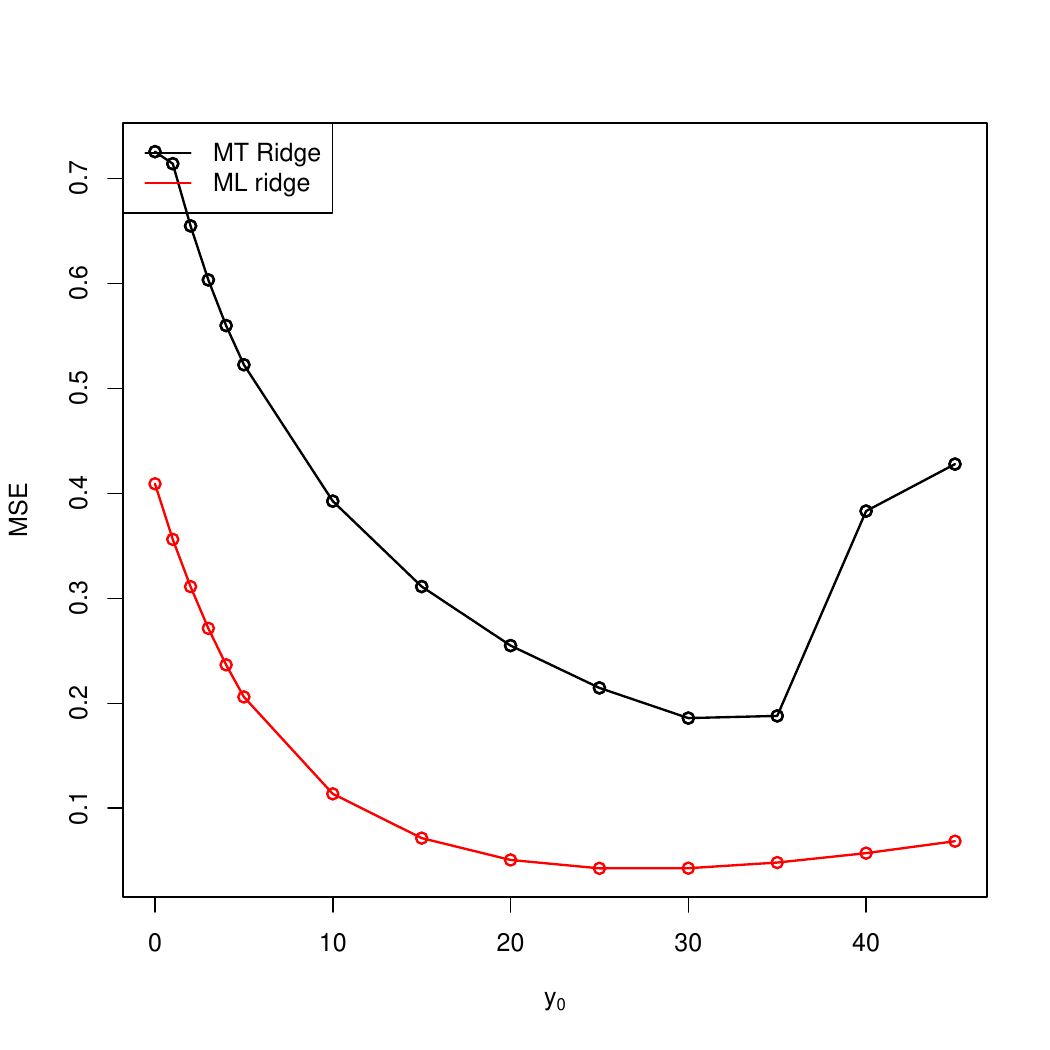}
			&   \includegraphics[scale=0.3]{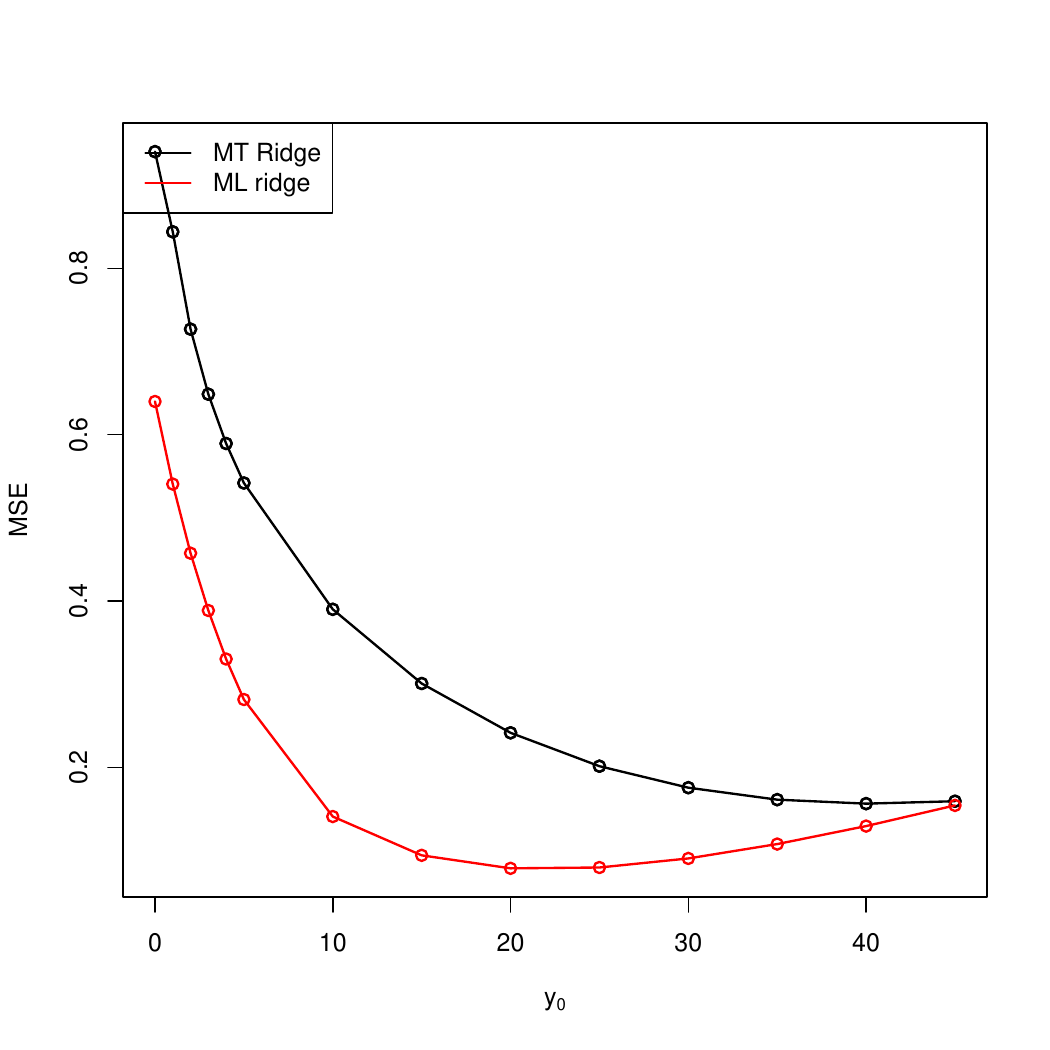} 
			&   \includegraphics[scale=0.3]{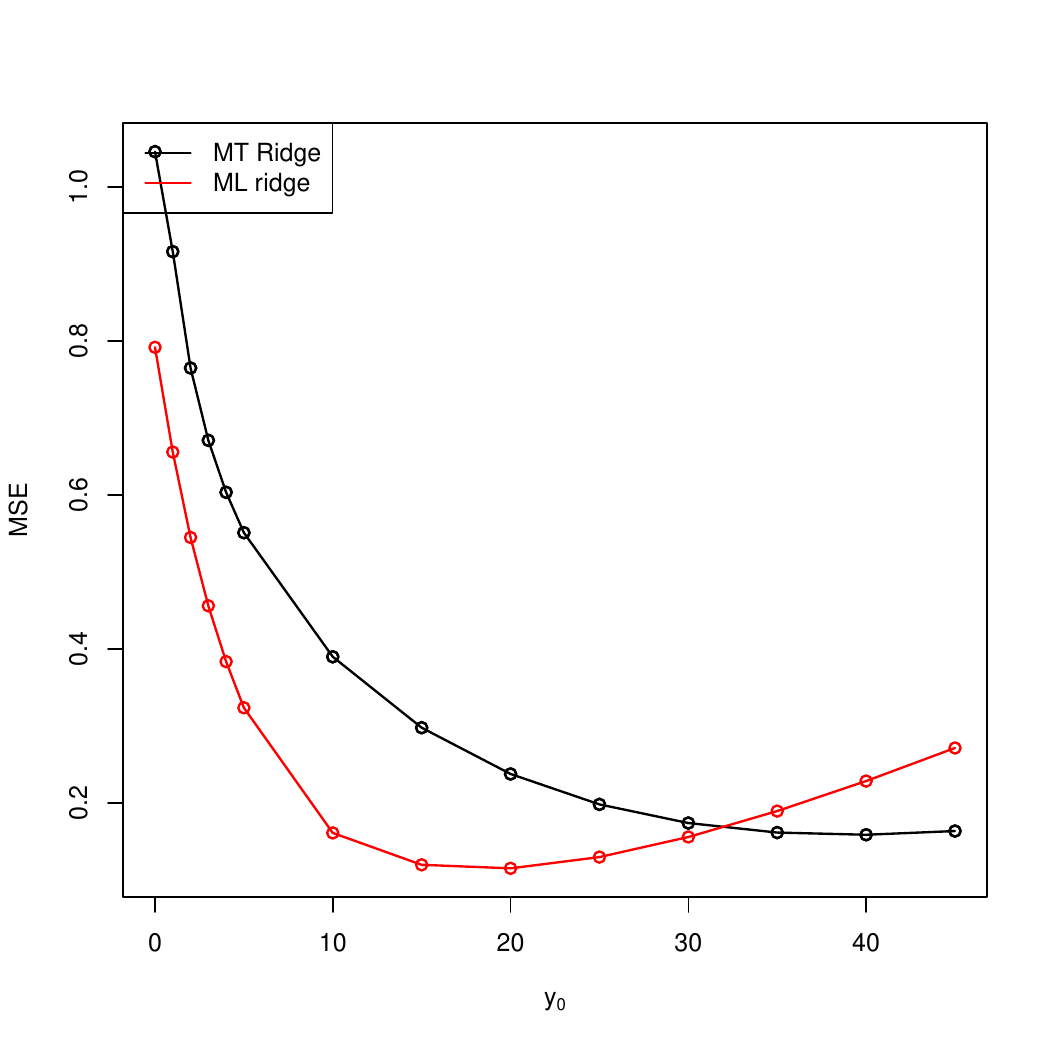}
		\end{tabular}
		\caption{\small \label{fig:msesimu1p50Ridge}  Mean squared errors for model AVY  and Ridge methods, $p=50$.}
	\end{center} 
\end{figure}
	 
%%%% Ridge MODEL 2 %%%%%%%%%%%%%%%%%%%%%%%%%%%%%%%%%%%%%%%%%%%%%	 

\begin{figure}[tp]
	\begin{center}
		\footnotesize
		\renewcommand{\arraystretch}{0.22}
		\newcolumntype{M}{>{\centering\arraybackslash}m{\dimexpr.01\linewidth-1\tabcolsep}}
		\newcolumntype{G}{>{\centering\arraybackslash}m{\dimexpr.38\linewidth-1\tabcolsep}}
		\begin{tabular}{GGG}
			$n=100,\, \epsilon=0.05$ & $n=100,\, \epsilon=0.1$  & $n=100,\, \epsilon=0.15$ \\
			\includegraphics[scale=0.29]{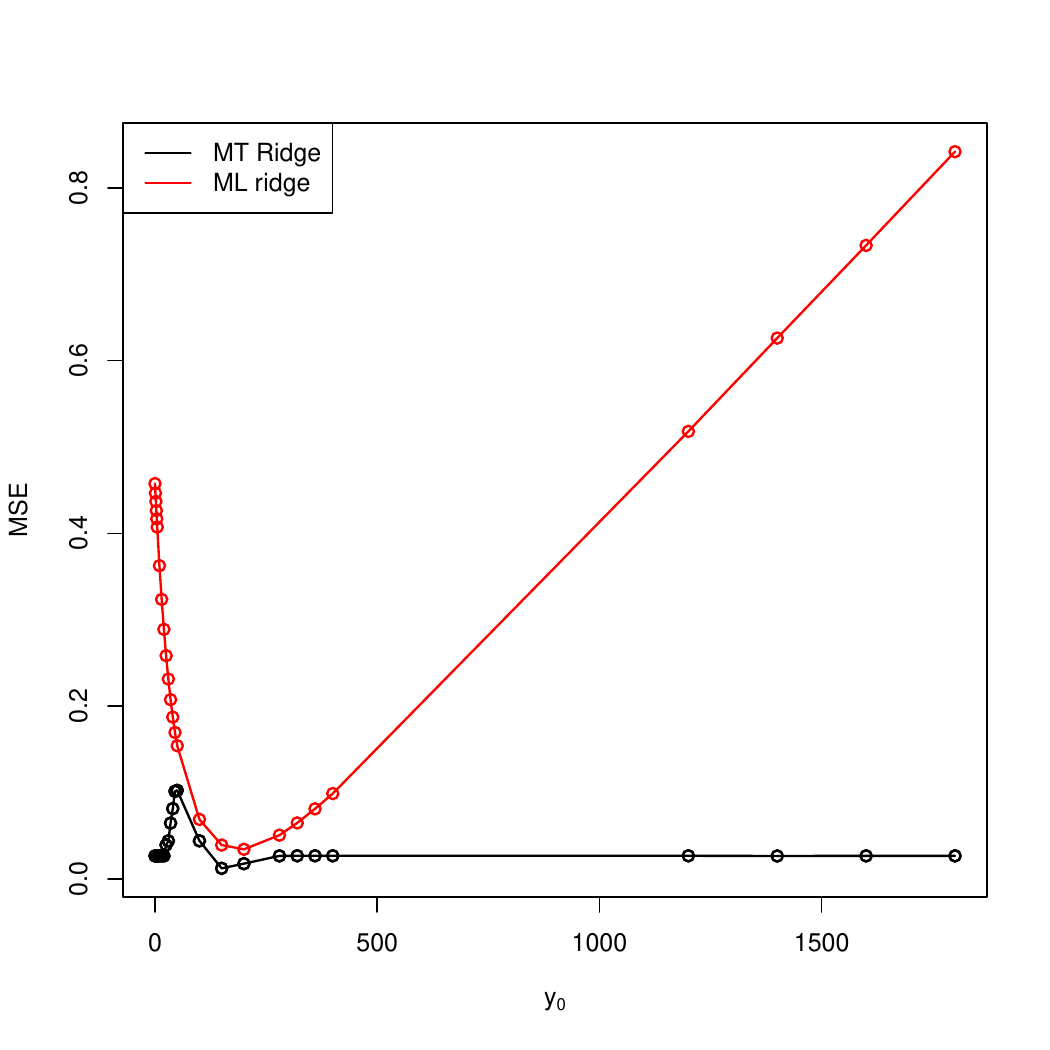} 
			&  \includegraphics[scale=0.29]{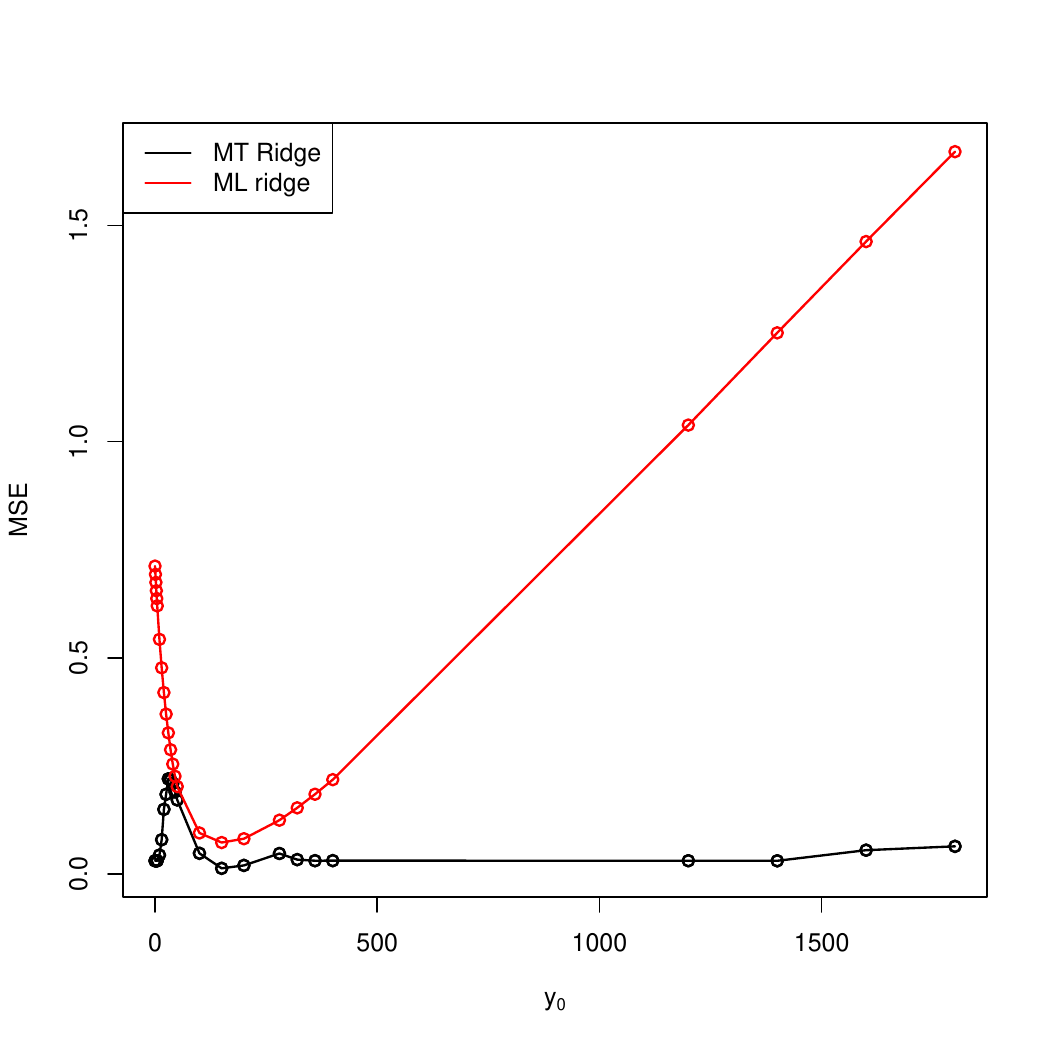}
			&   \includegraphics[scale=0.29]{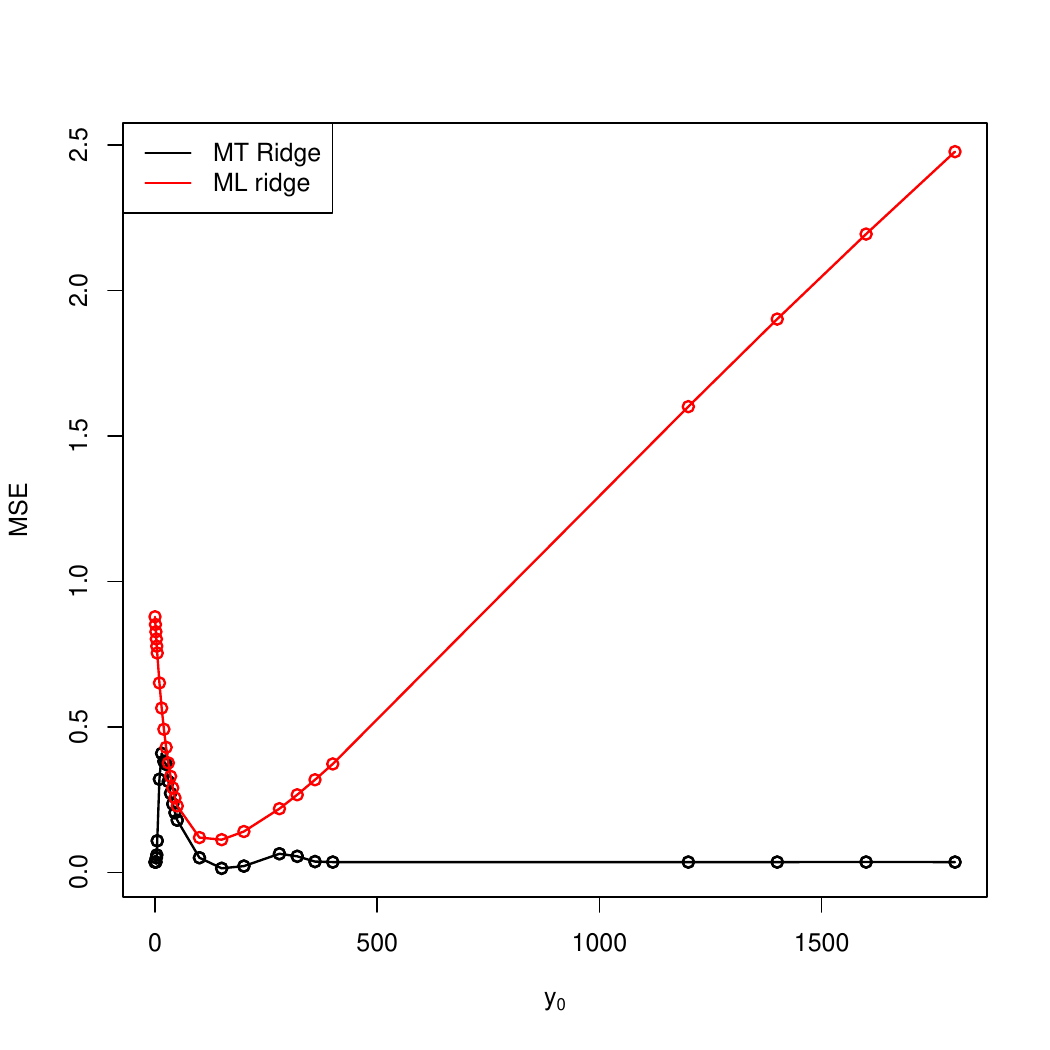}
			\\
			$n=400,\, \epsilon=0.05$ & $n=400,\, \epsilon=0.1$  & $n=400,\, \epsilon=0.15$ \\
			\includegraphics[scale=0.3]{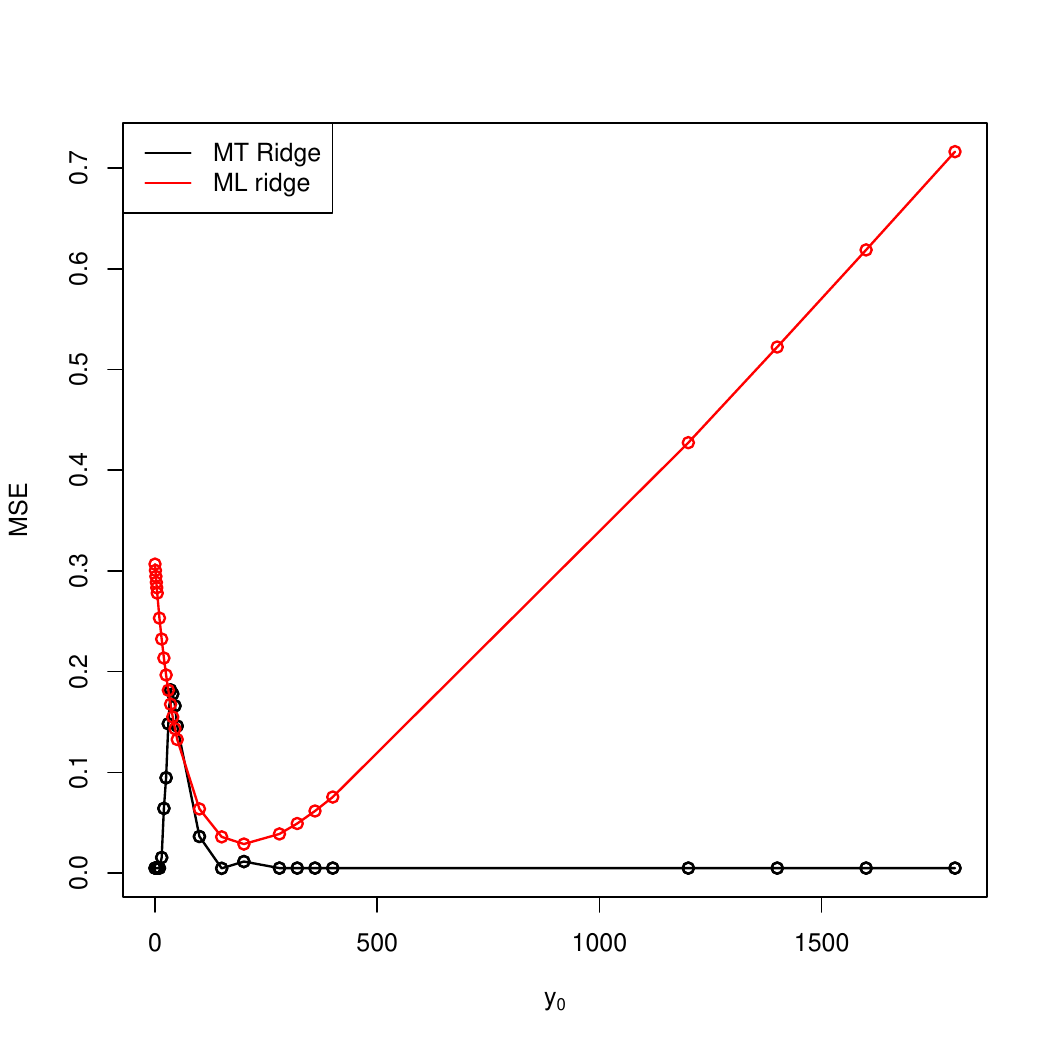}
			&   \includegraphics[scale=0.3]{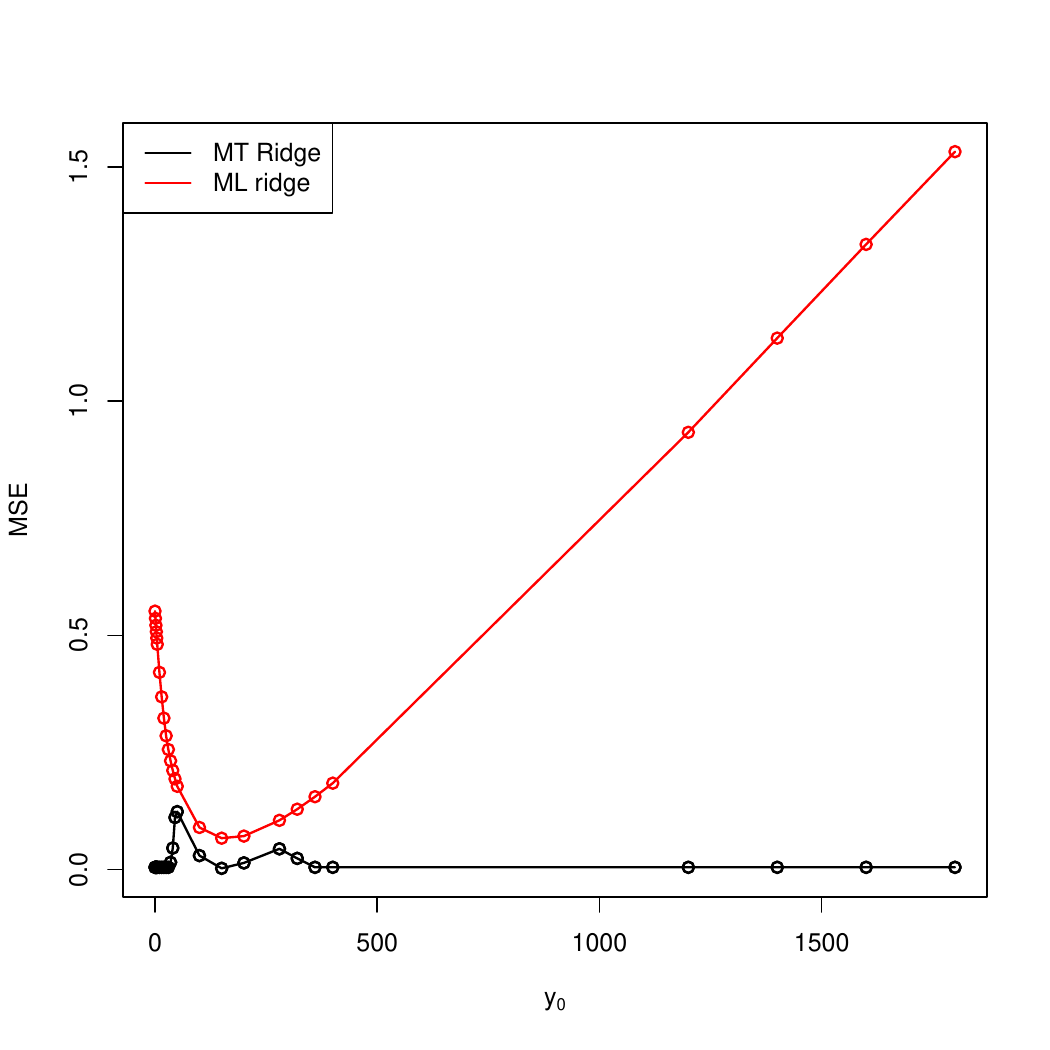} 
			&   \includegraphics[scale=0.3]{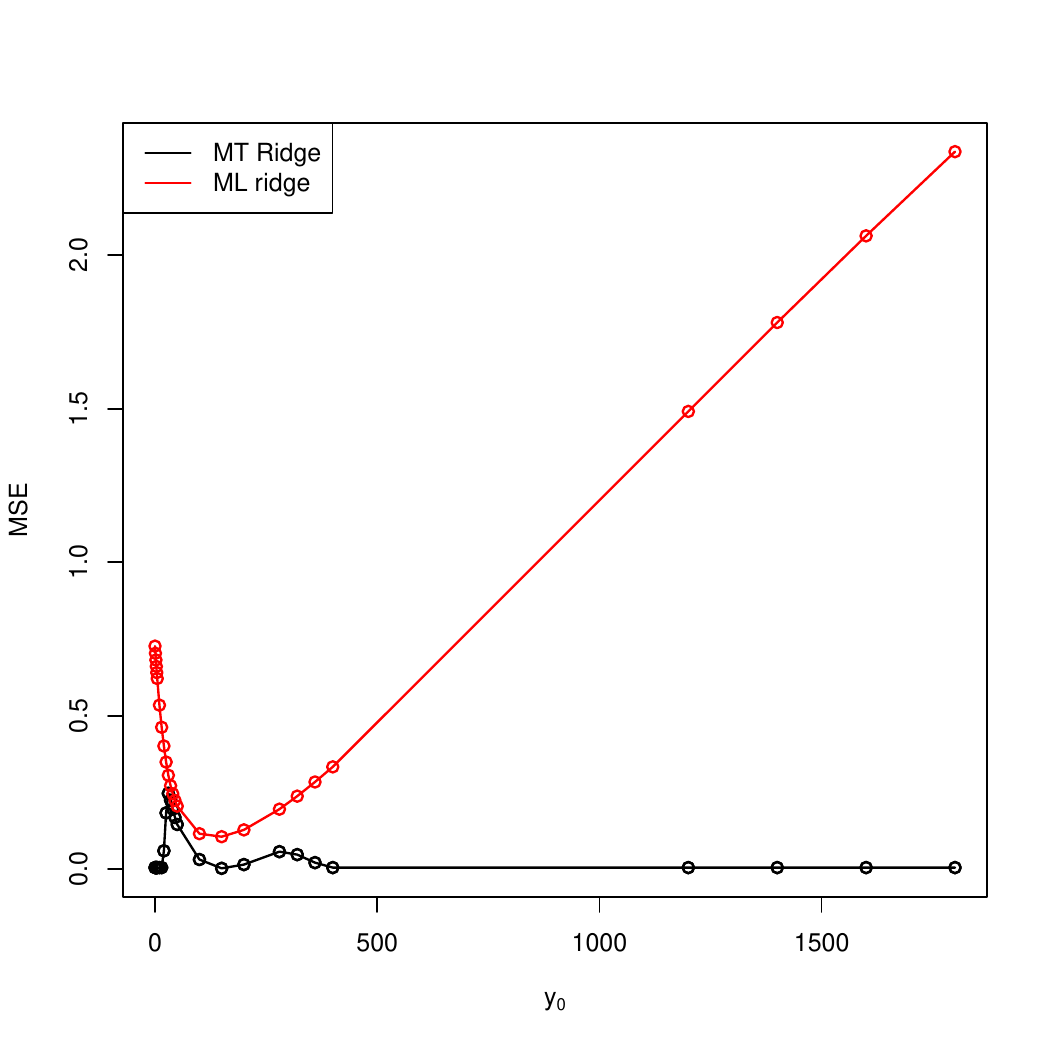} 
			\\
			$n=1000,\, \epsilon=0.05$ & $n=1000,\, \epsilon=0.1$  & $n=1000,\, \epsilon=0.15$ \\
			\includegraphics[scale=0.3]{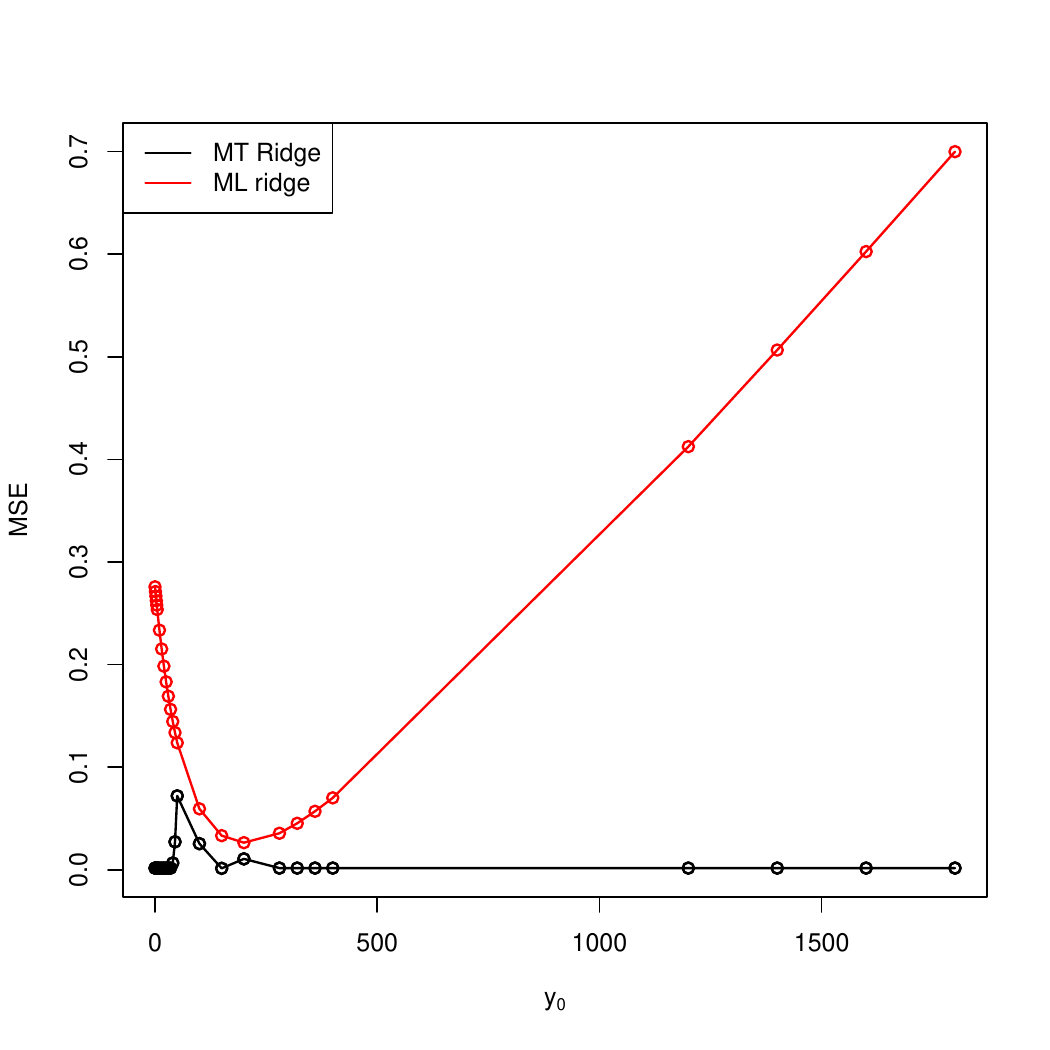}
			&   \includegraphics[scale=0.3]{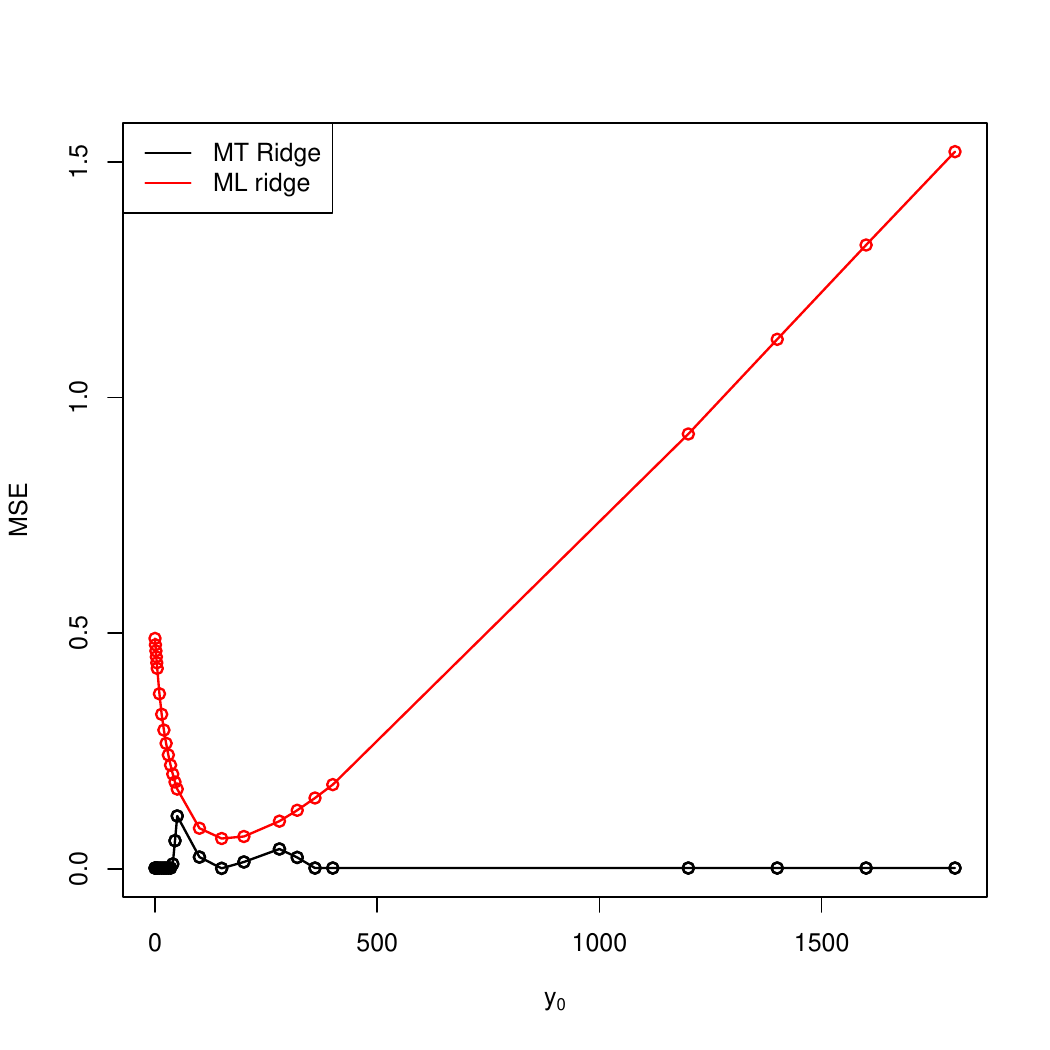} 
			&   \includegraphics[scale=0.3]{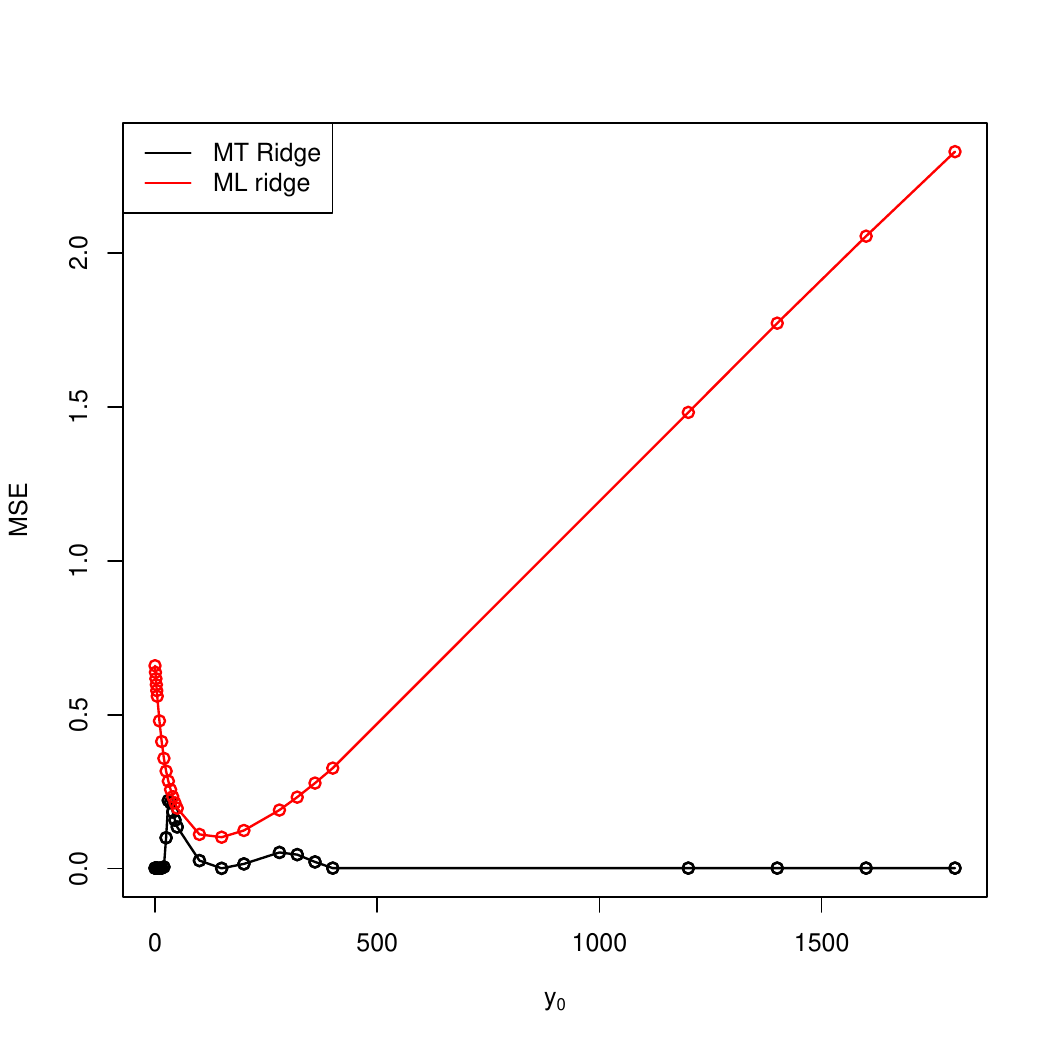}
		\end{tabular}
		\caption{\small \label{fig:msesimu2p10Ridge}  Mean squared errors for model AVY2  and Ridge methods, $p=10$.}
	\end{center} 
\end{figure}

\begin{figure}[tp]
	\begin{center}
		\footnotesize
		\renewcommand{\arraystretch}{0.22}
		\newcolumntype{M}{>{\centering\arraybackslash}m{\dimexpr.01\linewidth-1\tabcolsep}}
		\newcolumntype{G}{>{\centering\arraybackslash}m{\dimexpr.38\linewidth-1\tabcolsep}}
		\begin{tabular}{GGG}
			$n=100,\, \epsilon=0.05$ & $n=100,\, \epsilon=0.1$  & $n=100,\, \epsilon=0.15$ \\
			\includegraphics[scale=0.29]{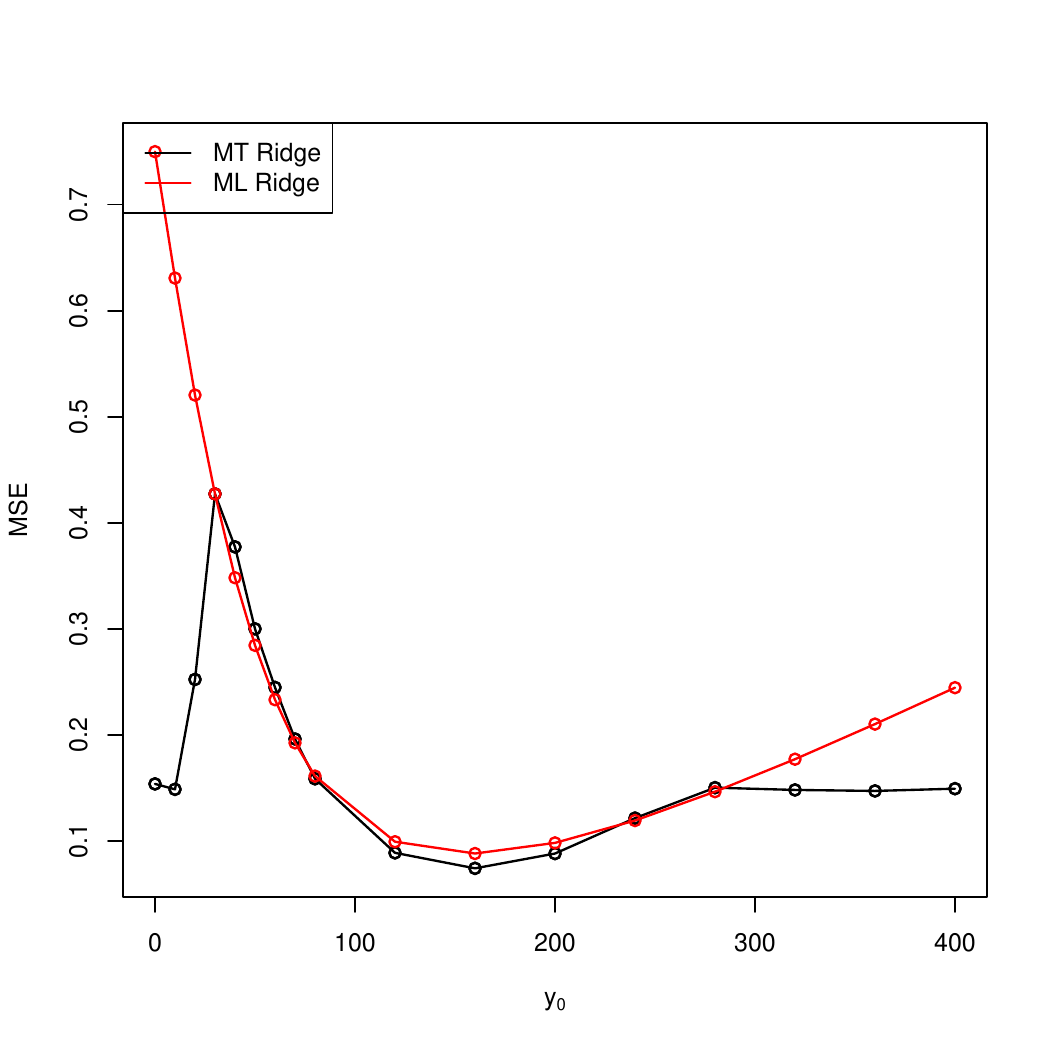} 
			&  \includegraphics[scale=0.29]{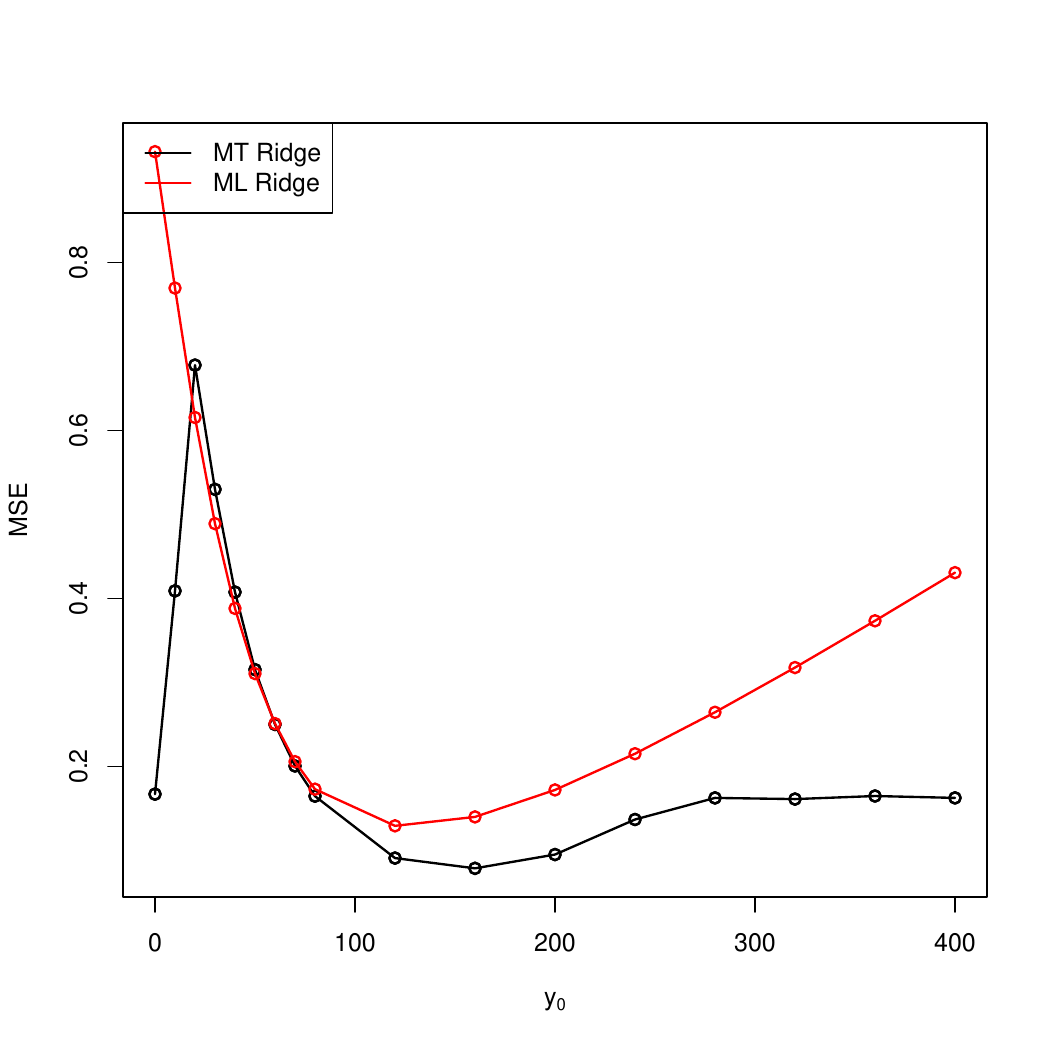}
			&   \includegraphics[scale=0.29]{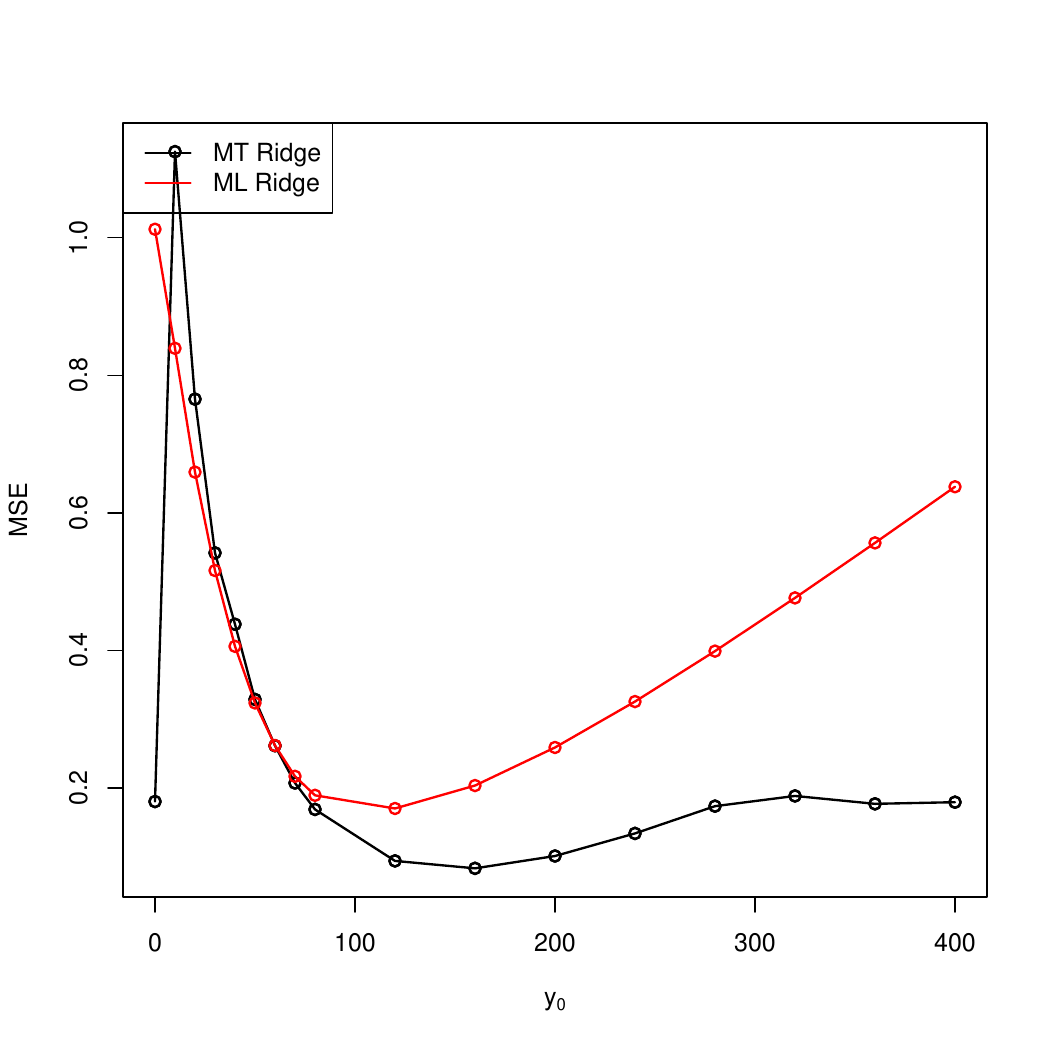}
			\\
			$n=400,\, \epsilon=0.05$ & $n=400,\, \epsilon=0.1$  & $n=400,\, \epsilon=0.15$ \\
			\includegraphics[scale=0.3]{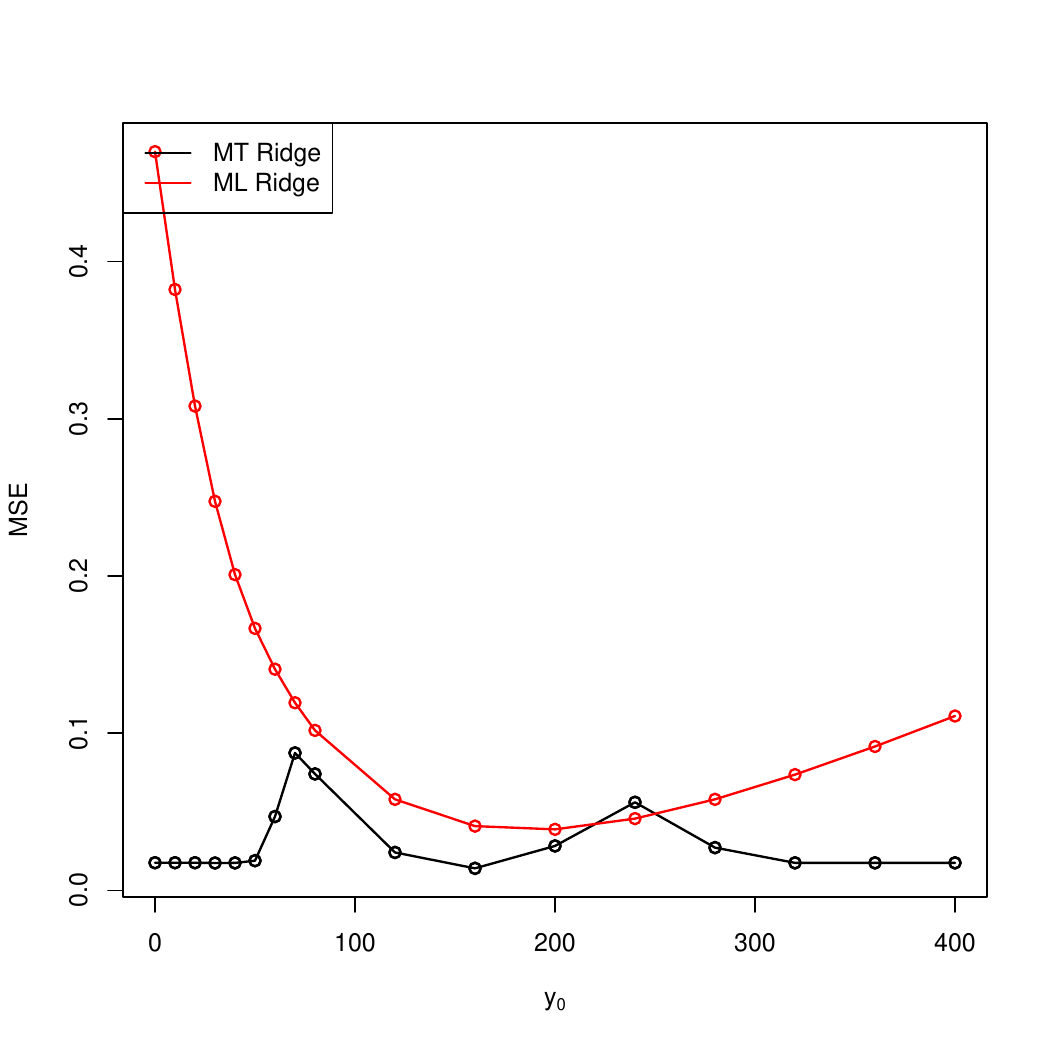}
			&   \includegraphics[scale=0.3]{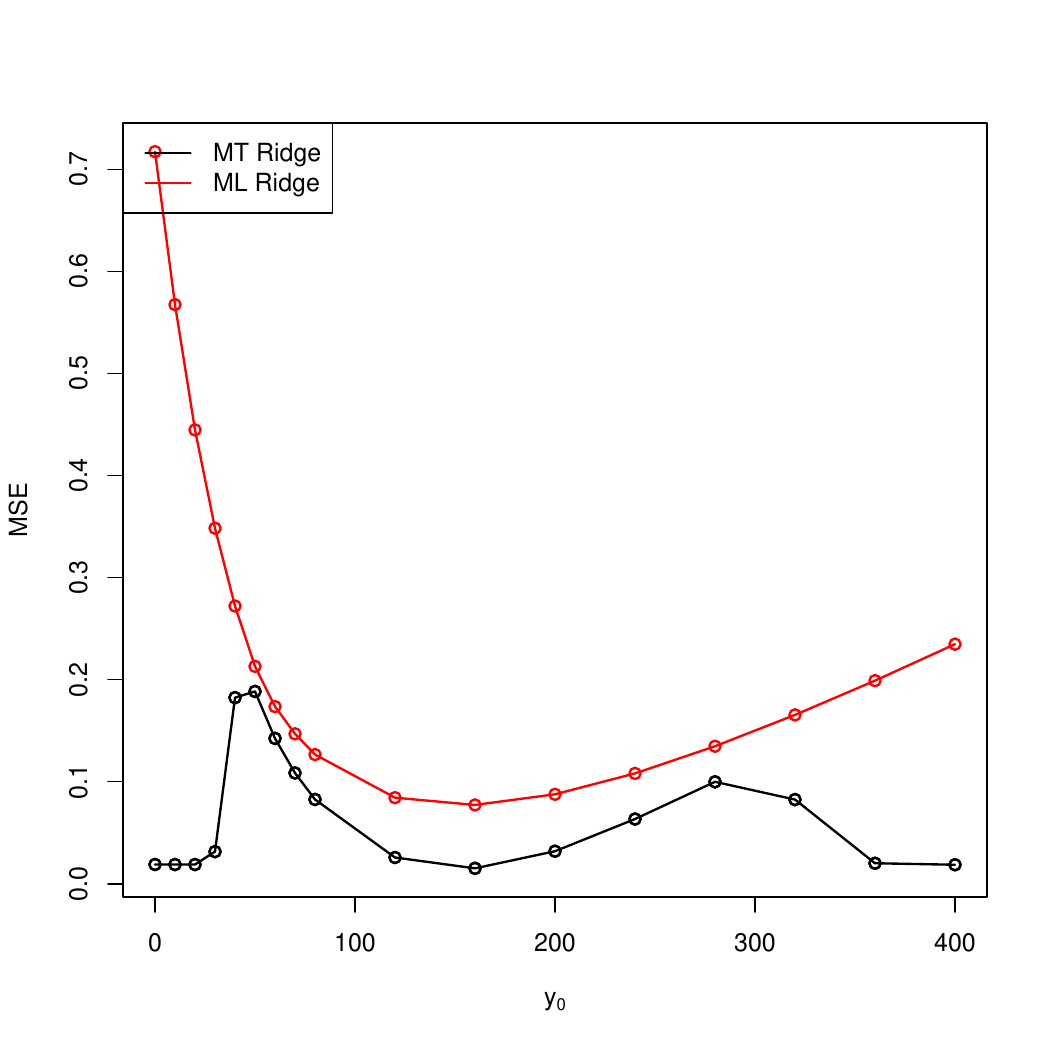} 
			&   \includegraphics[scale=0.3]{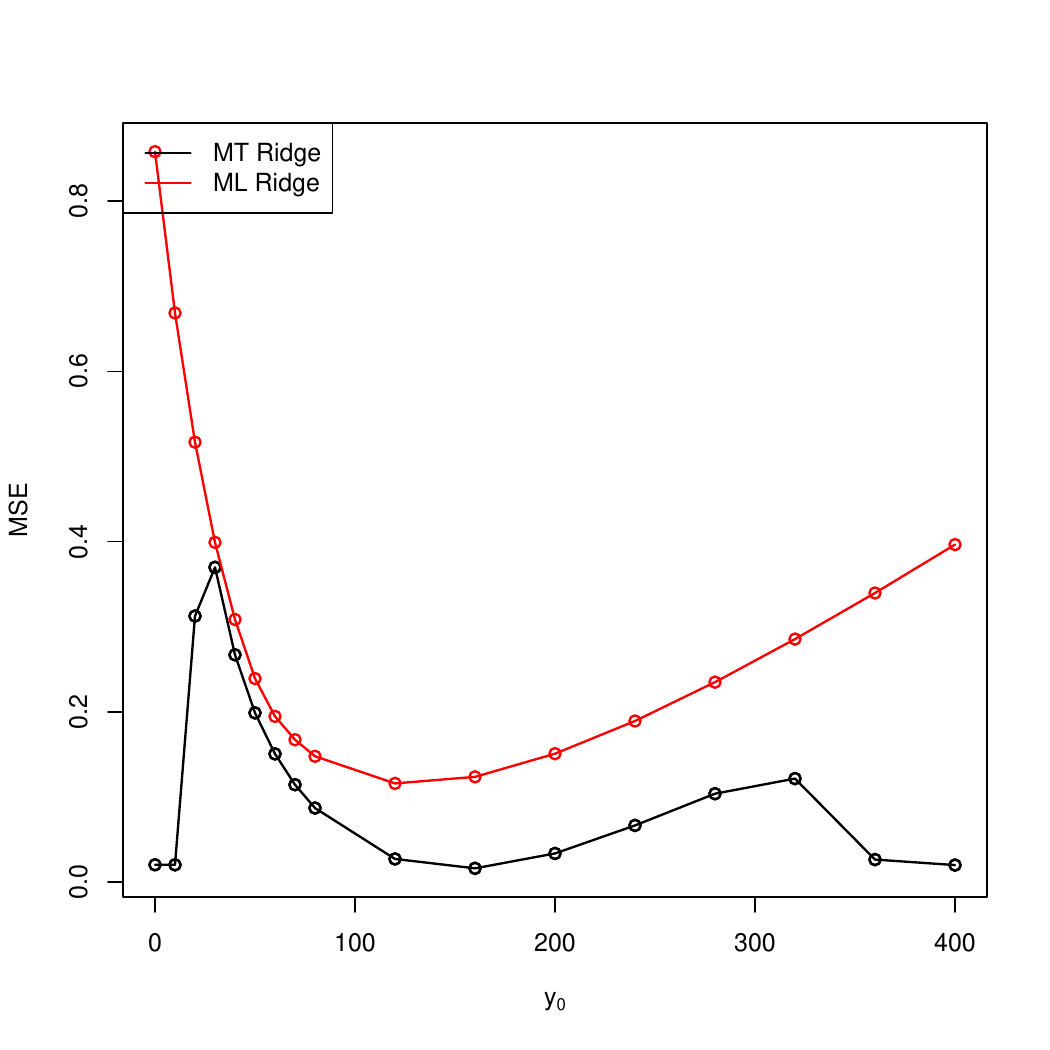} 
			\\
			$n=1000,\, \epsilon=0.05$ & $n=1000,\, \epsilon=0.1$  & $n=1000,\, \epsilon=0.15$ \\
			\includegraphics[scale=0.3]{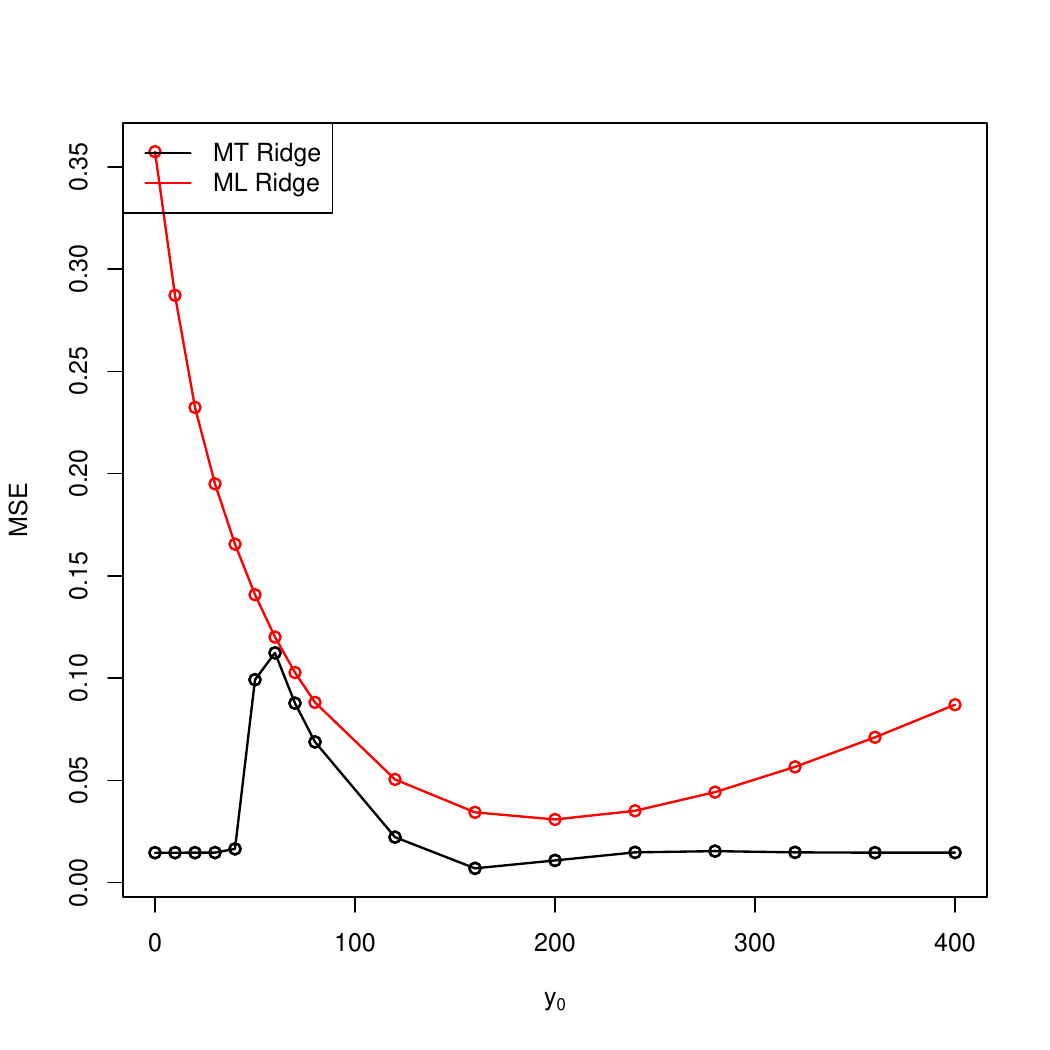}
			&   \includegraphics[scale=0.3]{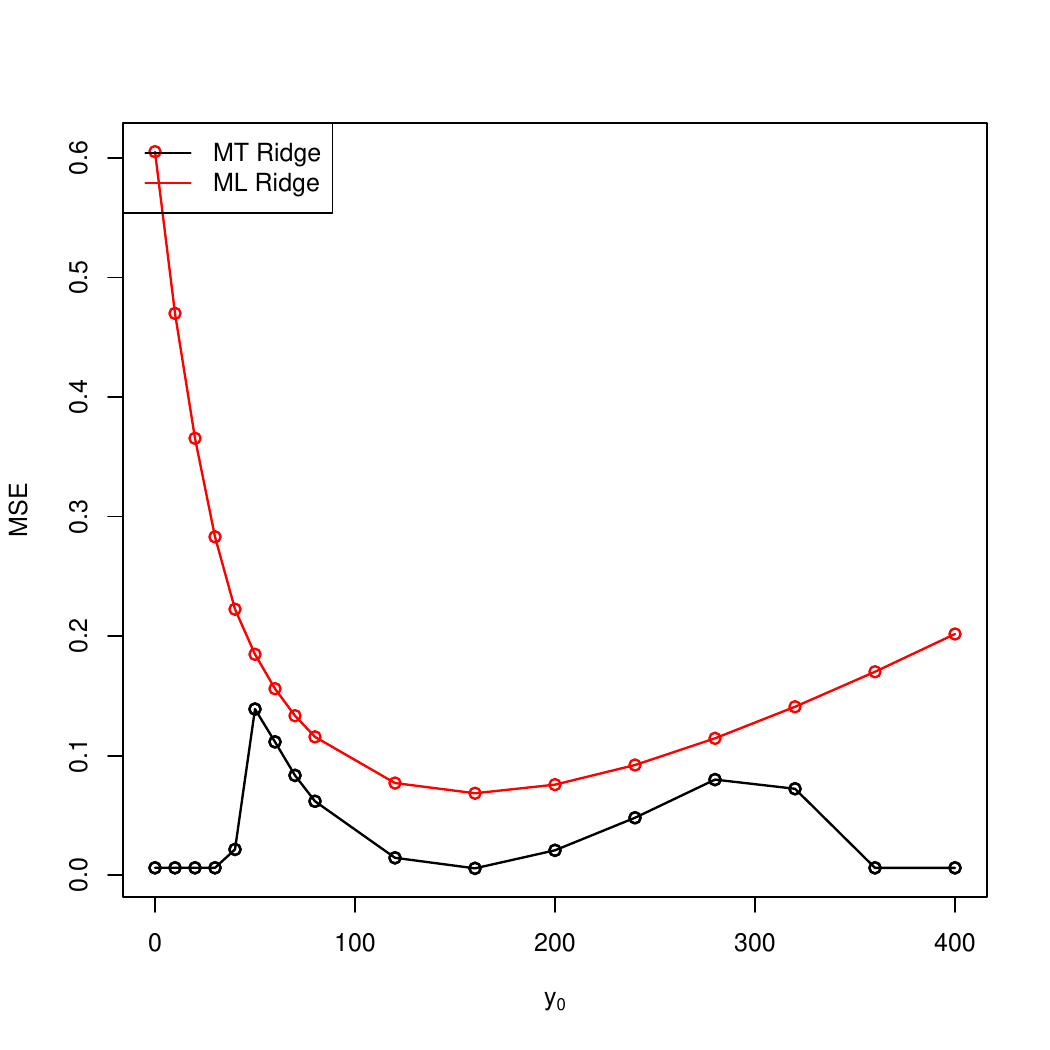} 
			&   \includegraphics[scale=0.3]{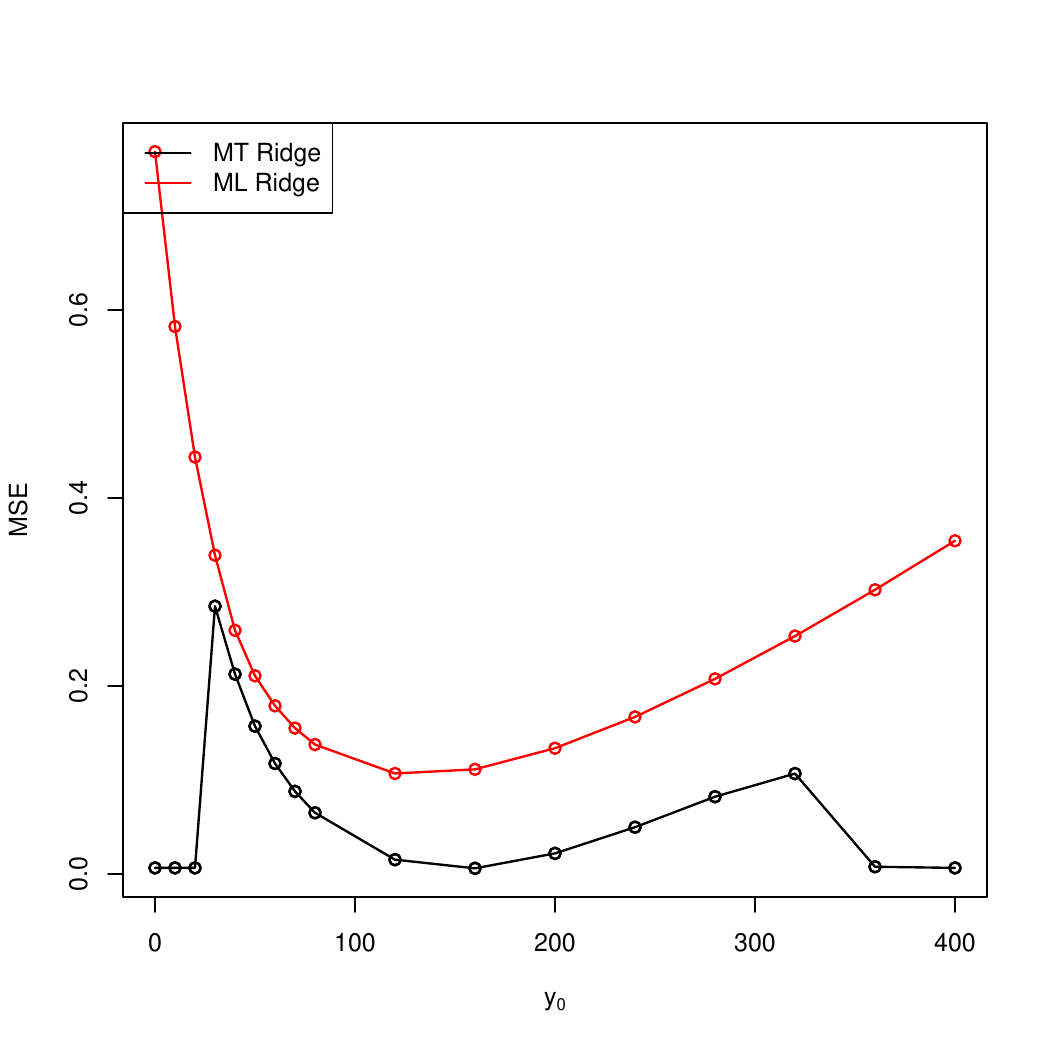}
		\end{tabular}
		\caption{\small \label{fig:msesimu2p50Ridge}  Mean squared errors for model AVY2 and Ridge methods, $p=50$.}
	\end{center} 
\end{figure}

%%%%% Ridge for avella medina setting %%%%%%%%%%%%%%%%%

\begin{figure}[tp]
	\begin{center}
		\footnotesize
		\renewcommand{\arraystretch}{0.22}
		\newcolumntype{M}{>{\centering\arraybackslash}m{\dimexpr.01\linewidth-1\tabcolsep}}
		\newcolumntype{G}{>{\centering\arraybackslash}m{\dimexpr.38\linewidth-1\tabcolsep}}
		\begin{tabular}{GGG}
			$n=50,\, \epsilon=0.05$ & $n=100,\, \epsilon=0.1$  & $n=50,\, \epsilon=0.15$ \\
			\includegraphics[scale=0.29]{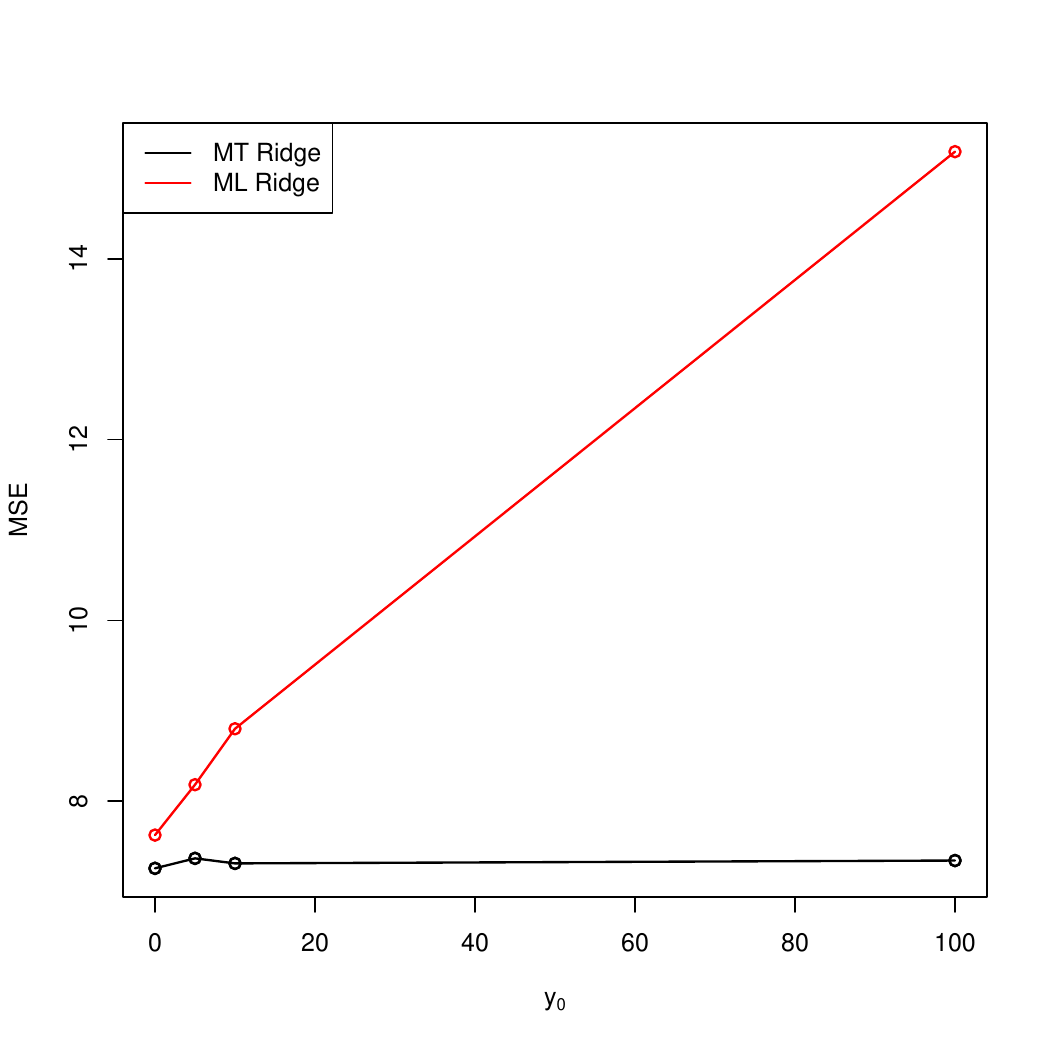} 
			&  \includegraphics[scale=0.29]{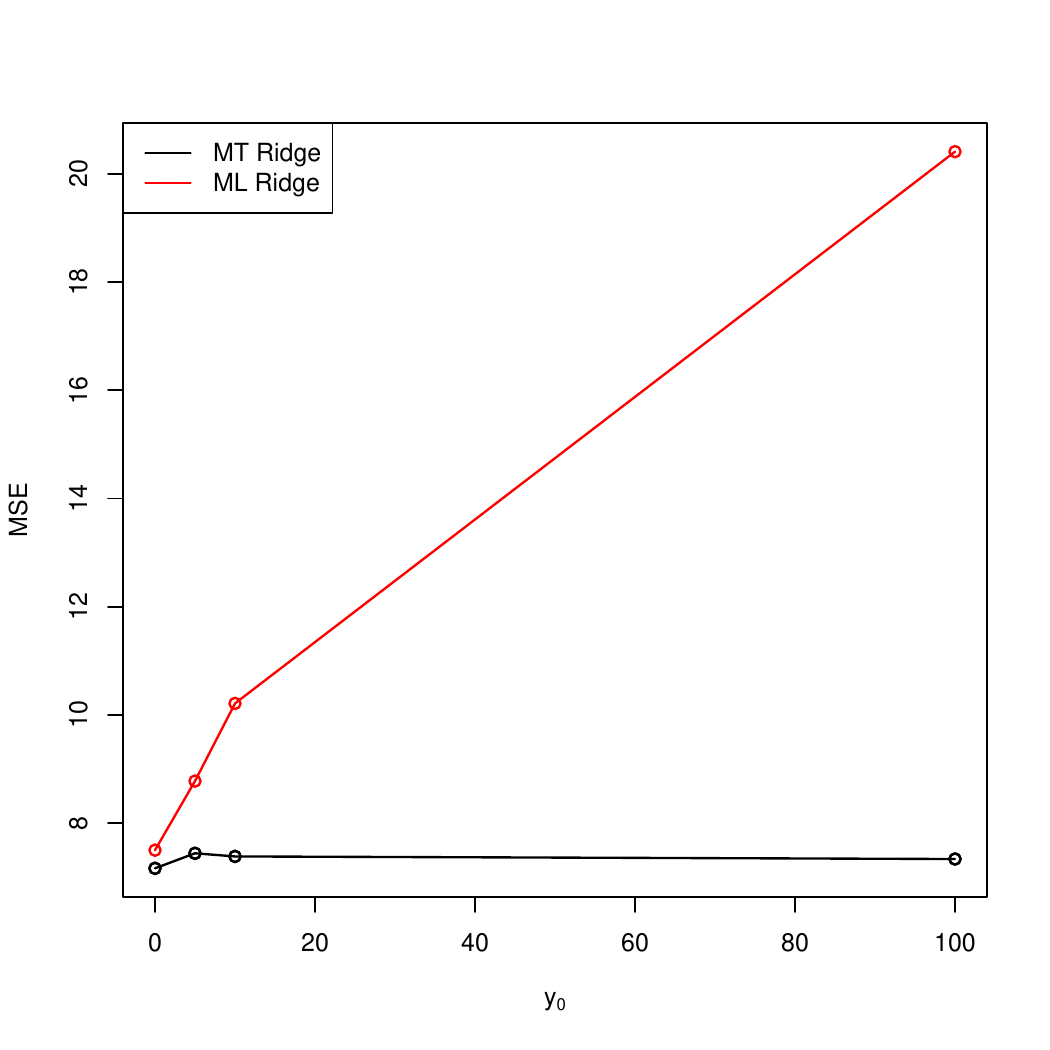}
			&   \includegraphics[scale=0.29]{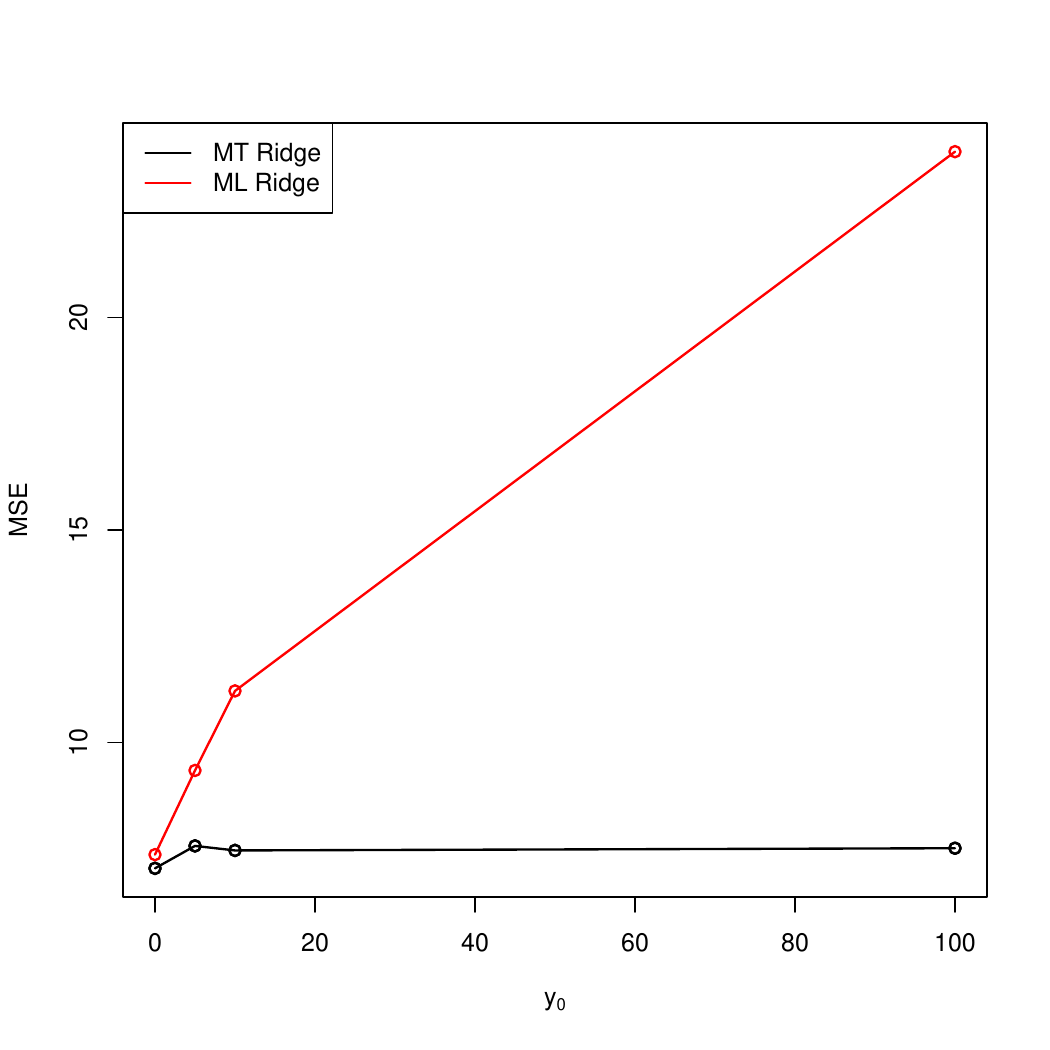}
			\\
			$n=100,\, \epsilon=0.05$ & $n=100,\, \epsilon=0.1$  & $n=100,\, \epsilon=0.15$ \\
			\includegraphics[scale=0.3]{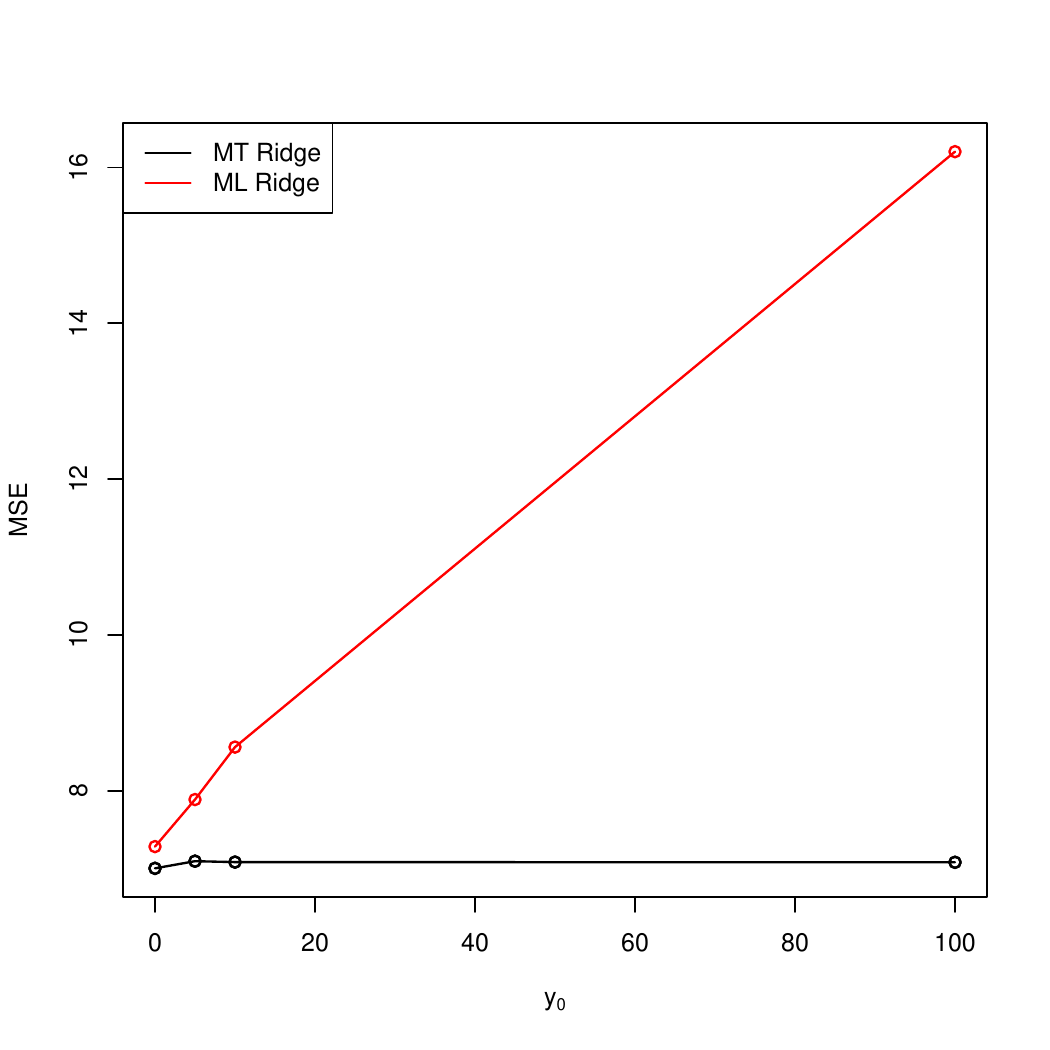}
			&   \includegraphics[scale=0.3]{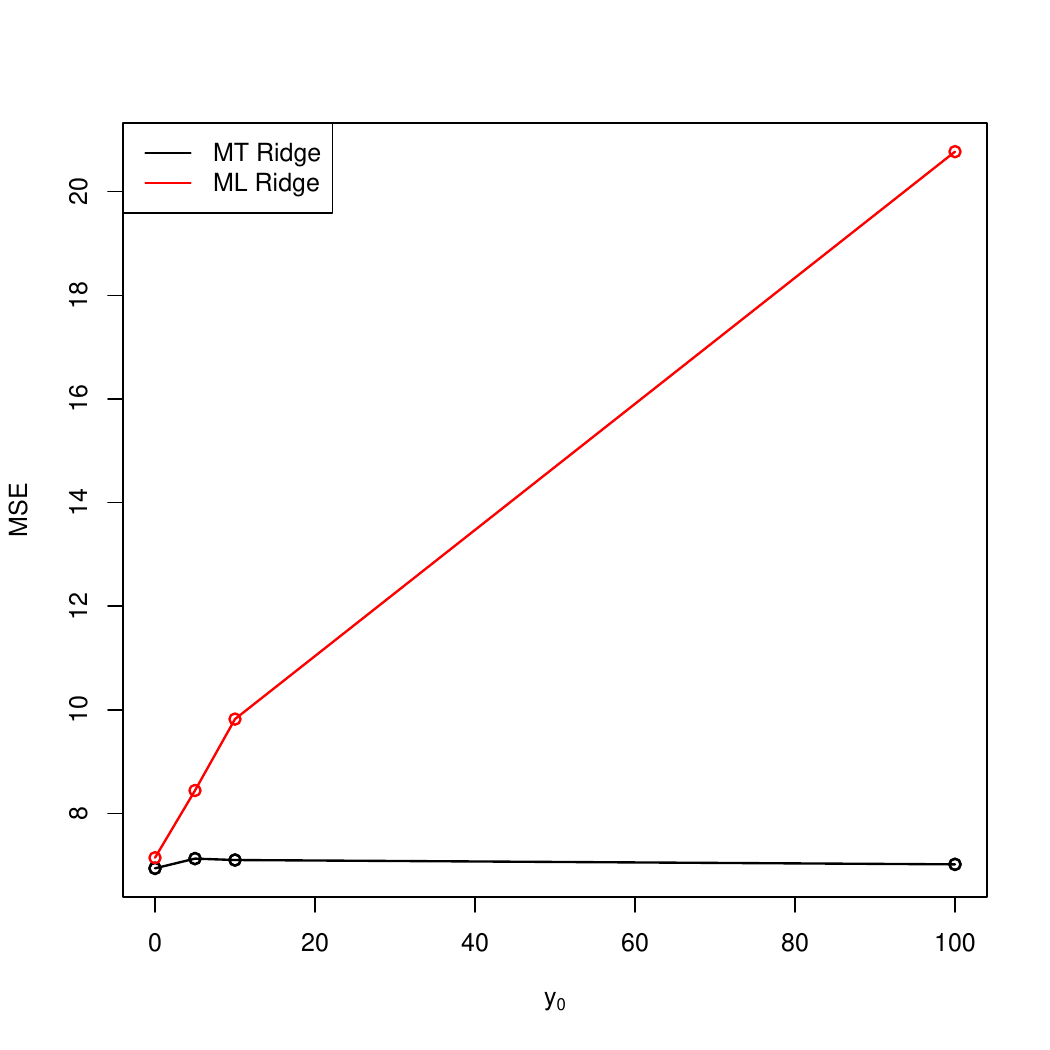} 
			&   \includegraphics[scale=0.3]{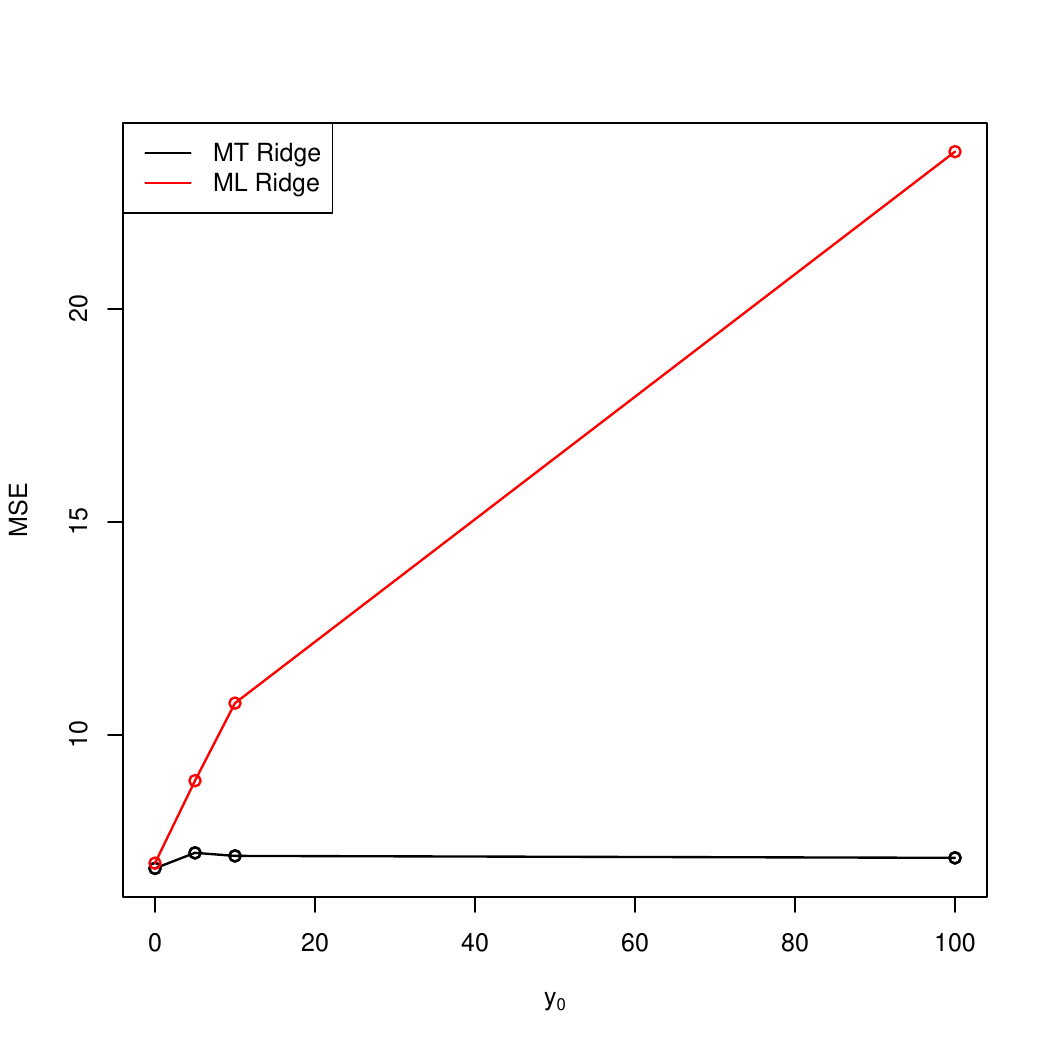} 
			\\
			$n=200,\, \epsilon=0.05$ & $n=200,\, \epsilon=0.1$  & $n=200,\, \epsilon=0.15$ \\
			\includegraphics[scale=0.3]{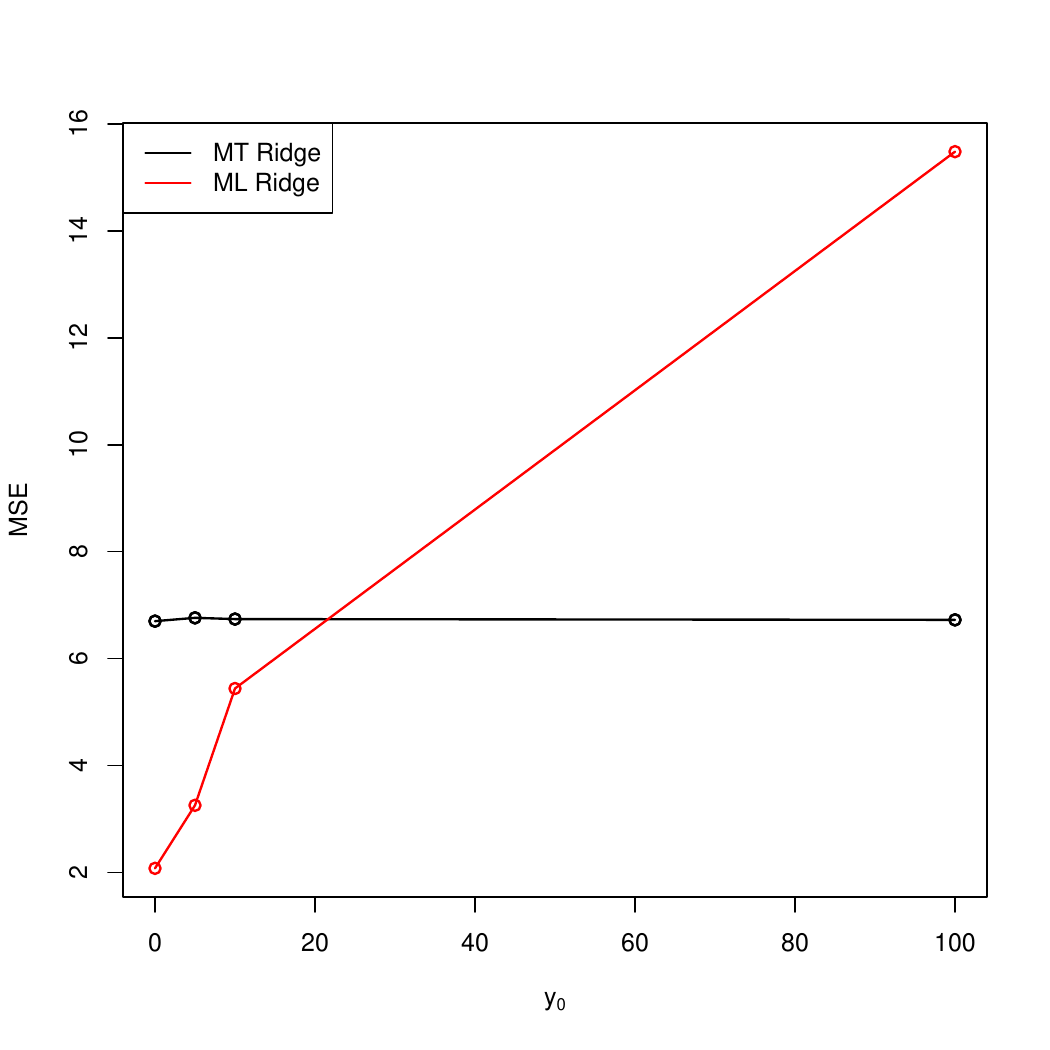}
			&   \includegraphics[scale=0.3]{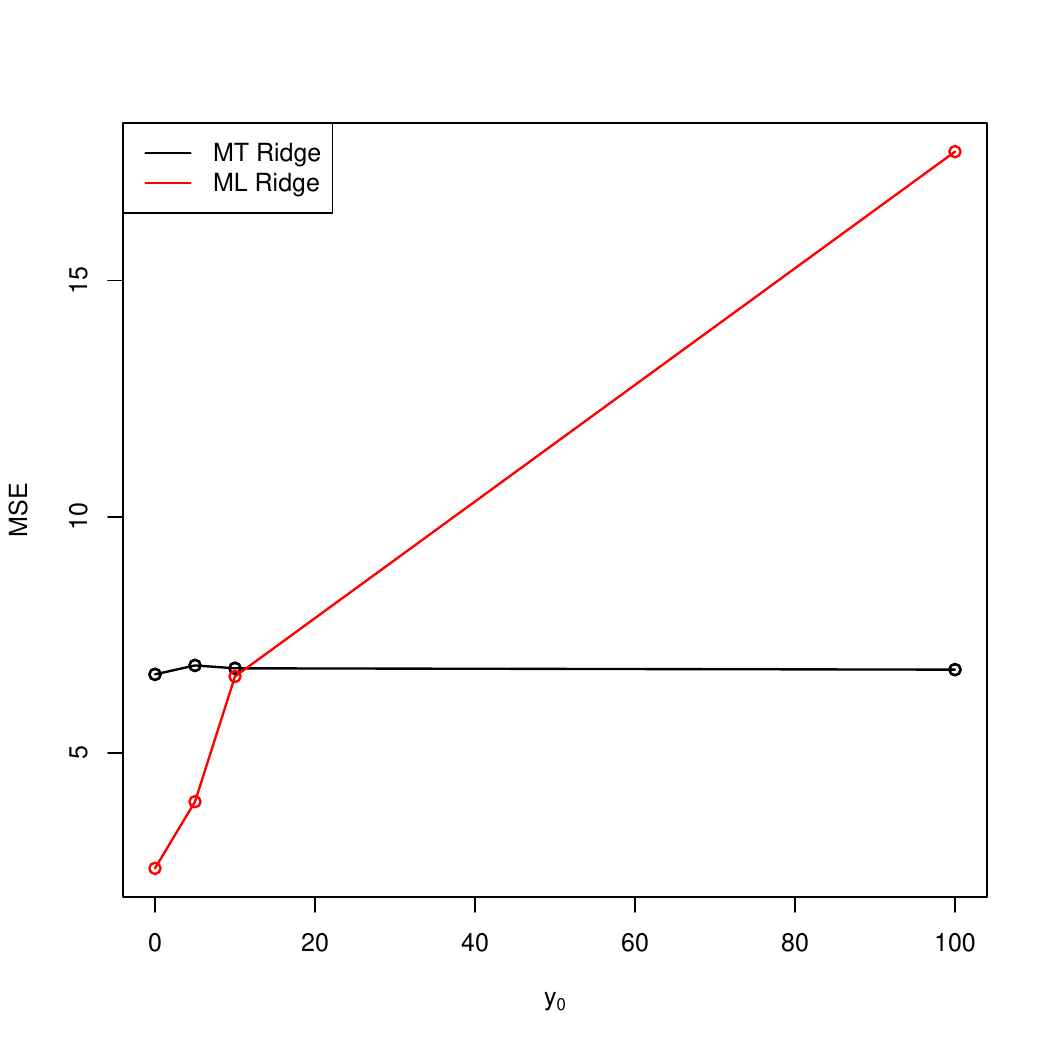} 
			&   \includegraphics[scale=0.3]{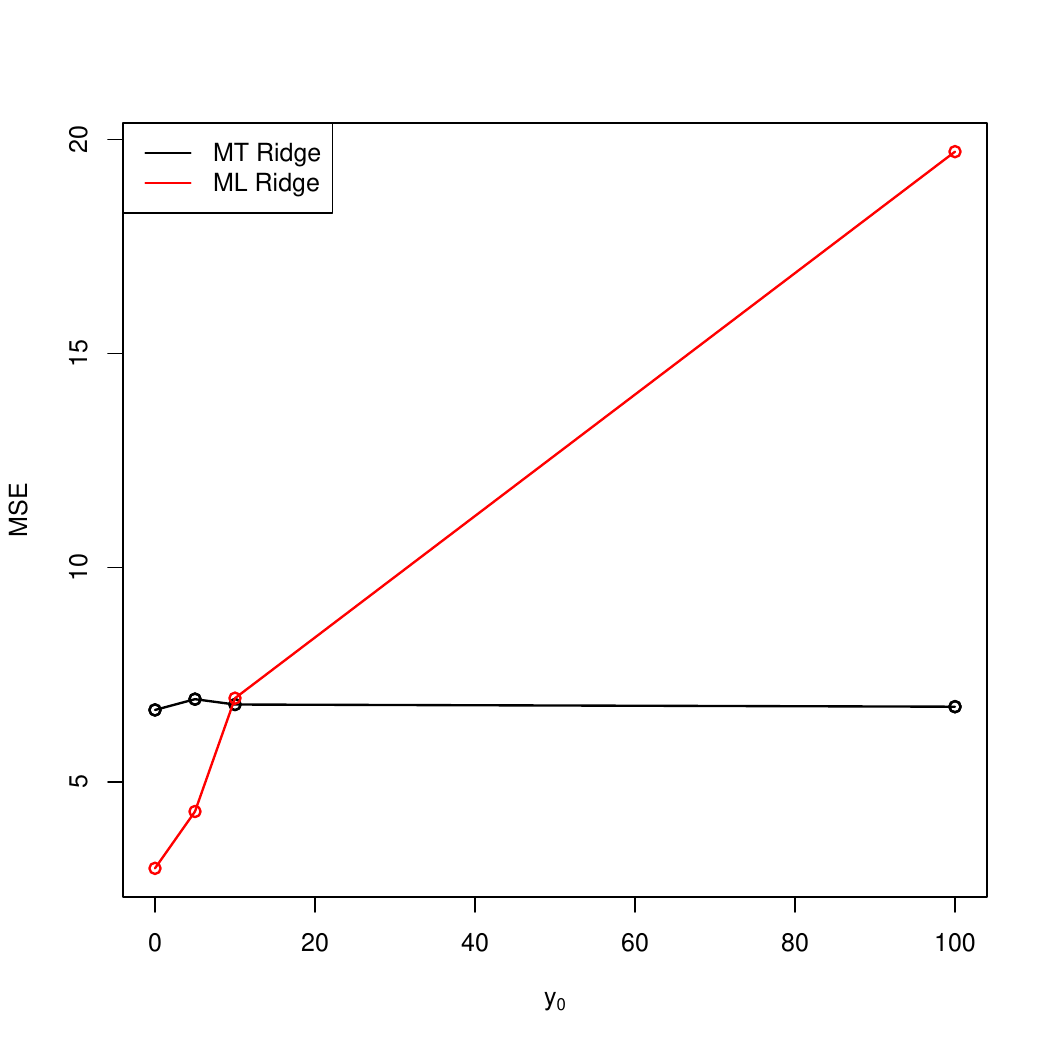}
		\end{tabular}
		\caption{\small \label{fig:msesimuRQLp100Ridge}  Mean squared errors for model AMR and Ridge methods, $p=100$.}
	\end{center} 
\end{figure}

\begin{figure}[tp]
	\begin{center}
		\footnotesize
		\renewcommand{\arraystretch}{0.22}
		\newcolumntype{M}{>{\centering\arraybackslash}m{\dimexpr.01\linewidth-1\tabcolsep}}
		\newcolumntype{G}{>{\centering\arraybackslash}m{\dimexpr.38\linewidth-1\tabcolsep}}
		\begin{tabular}{GGG}
			$n=50,\, \epsilon=0.05$ & $n=50,\, \epsilon=0.1$  & $n=50,\, \epsilon=0.15$ \\
			\includegraphics[scale=0.29]{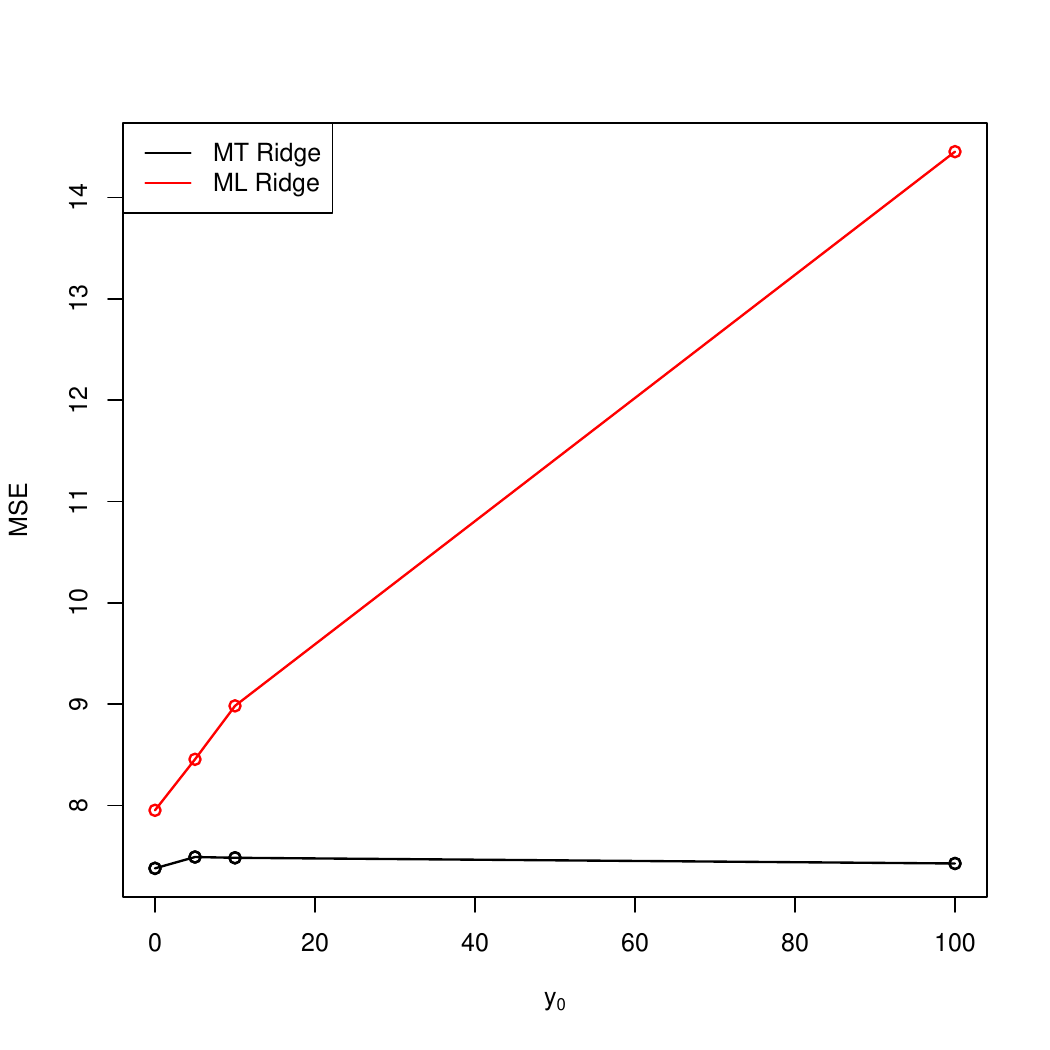} 
			&  \includegraphics[scale=0.29]{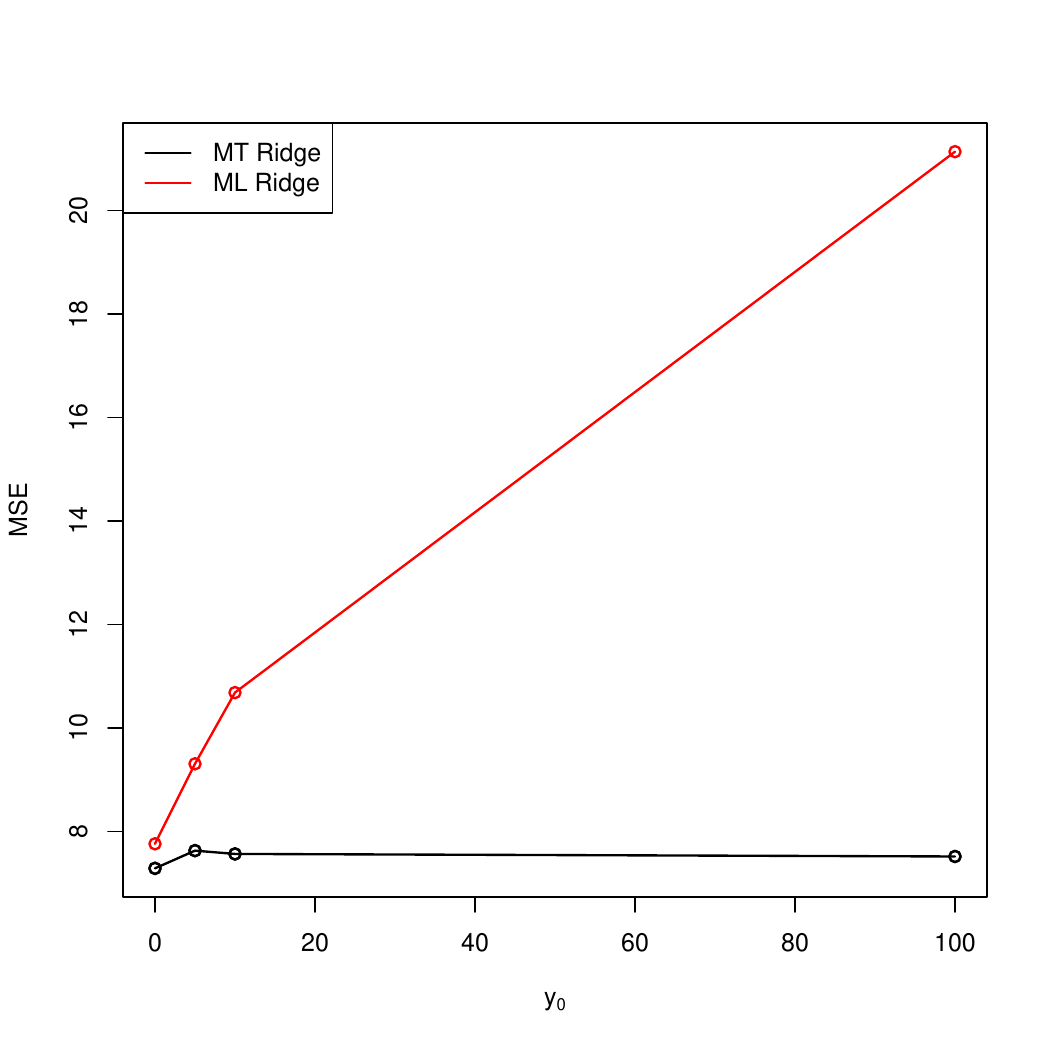}
			&   \includegraphics[scale=0.29]{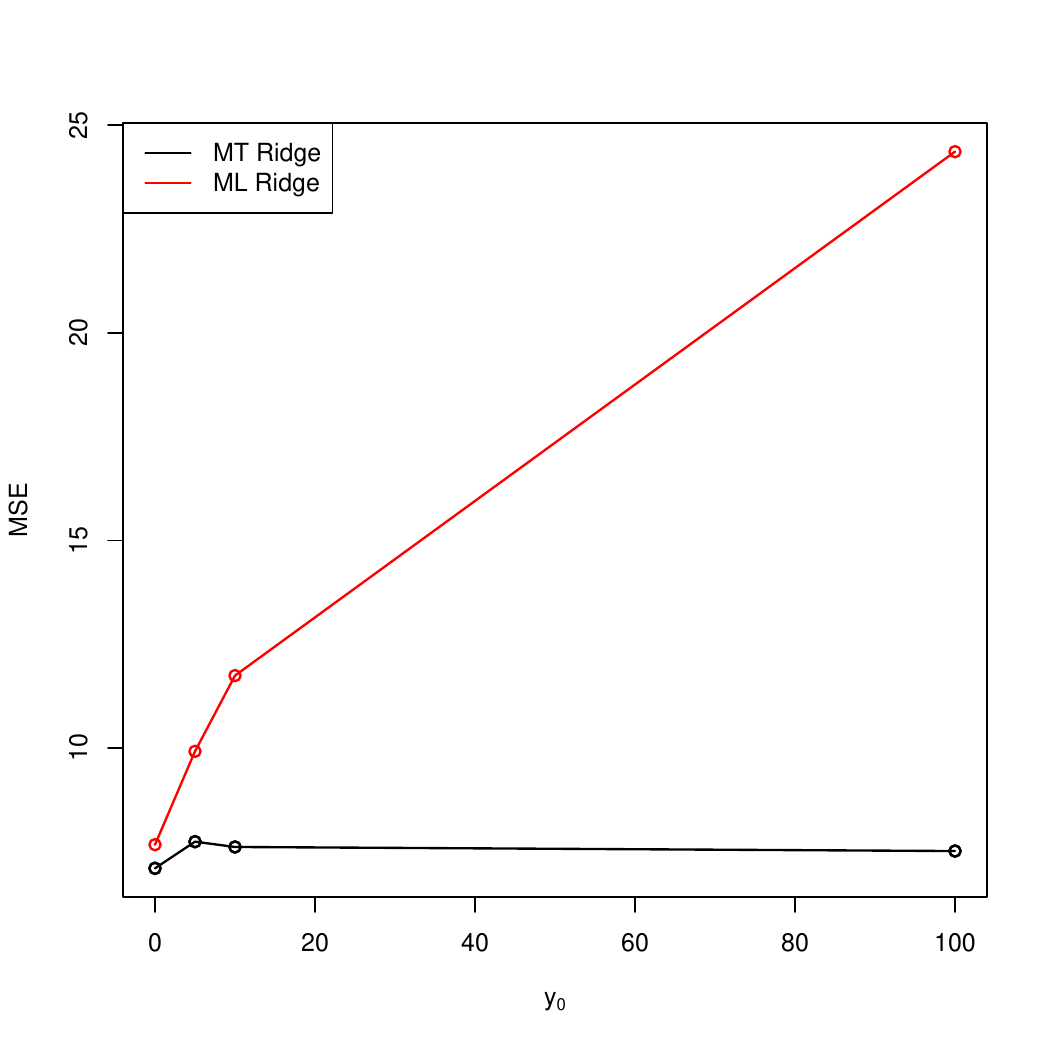}
			\\
			$n=100,\, \epsilon=0.05$ & $n=100,\, \epsilon=0.1$  & $n=100,\, \epsilon=0.15$ \\
			\includegraphics[scale=0.3]{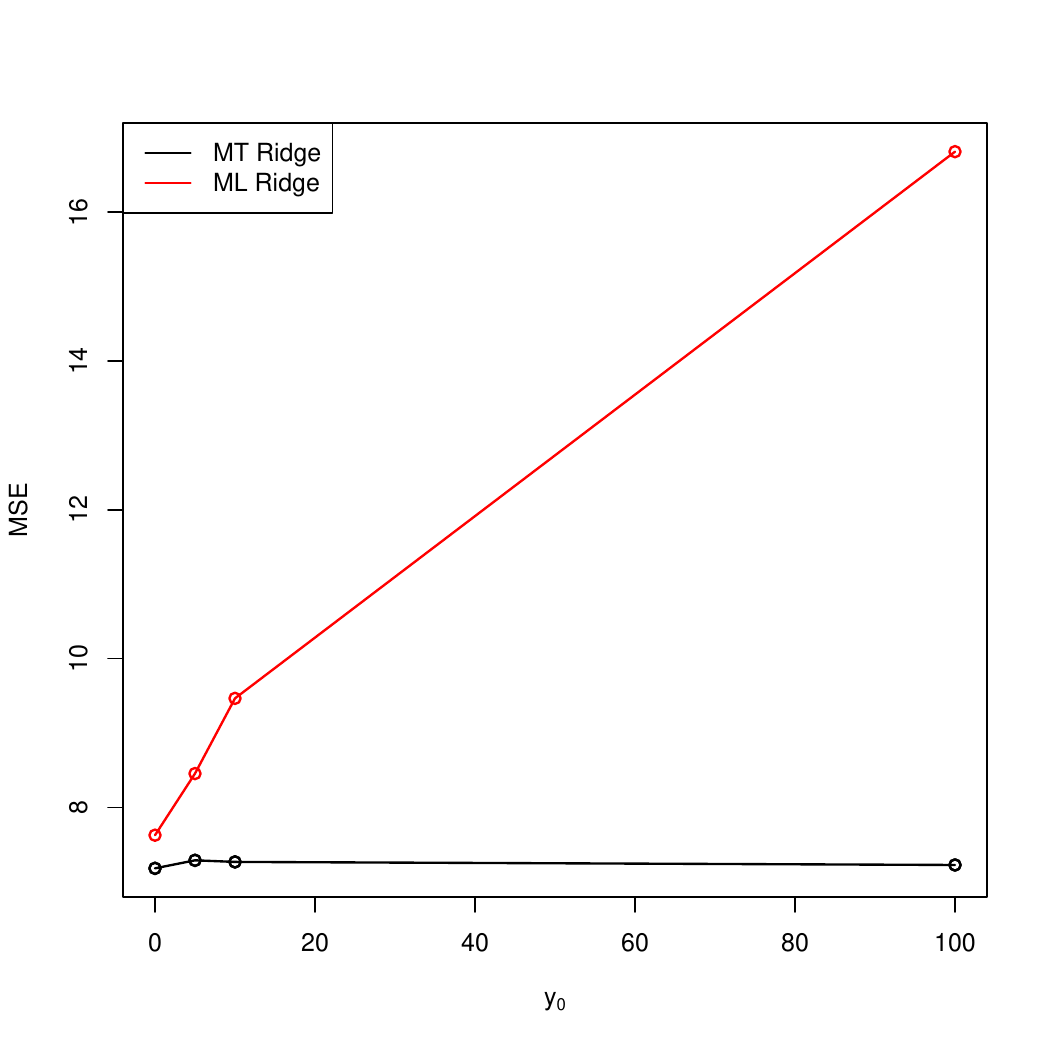}
			&   \includegraphics[scale=0.3]{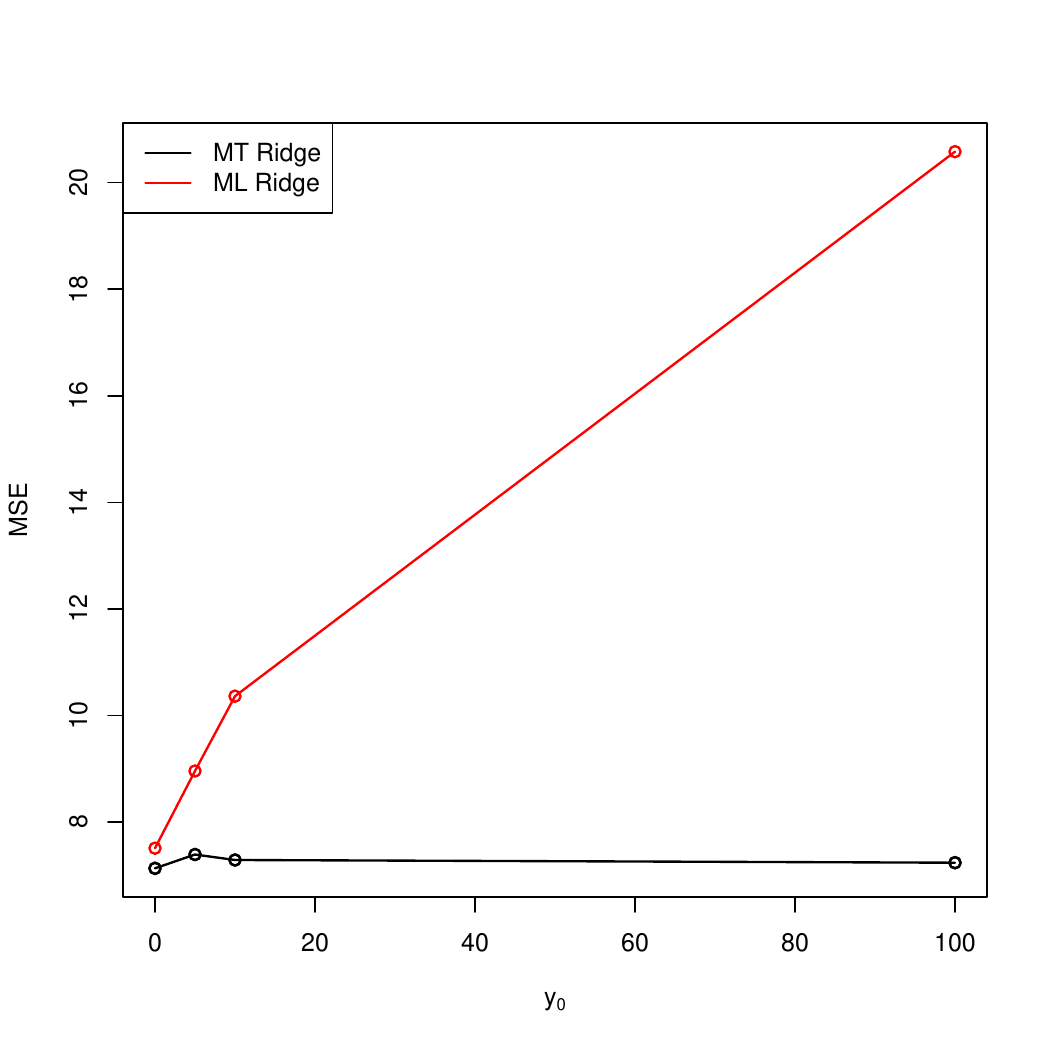} 
			&   \includegraphics[scale=0.3]{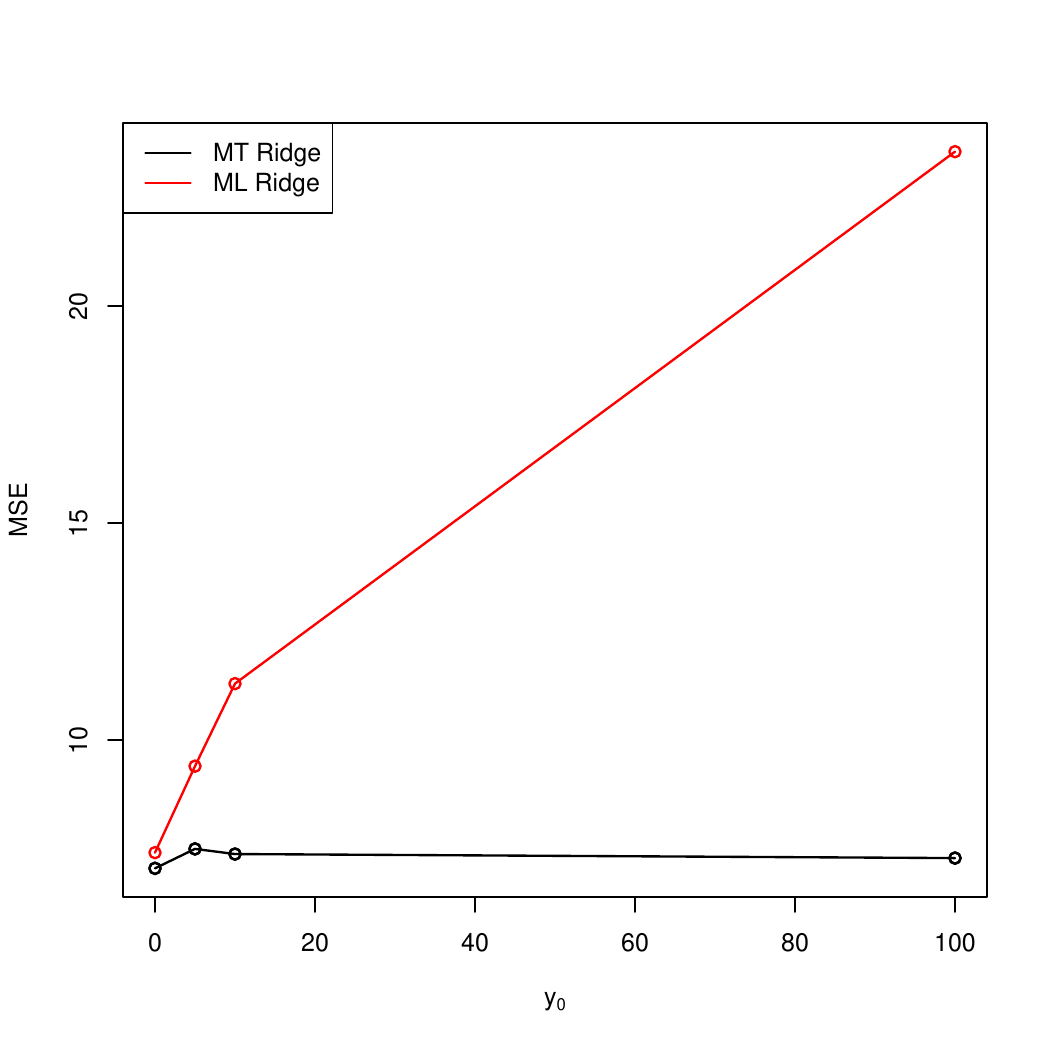} 
			\\
			$n=200,\, \epsilon=0.05$ & $n=200,\, \epsilon=0.1$  & $n=200,\, \epsilon=0.15$ \\
			\includegraphics[scale=0.3]{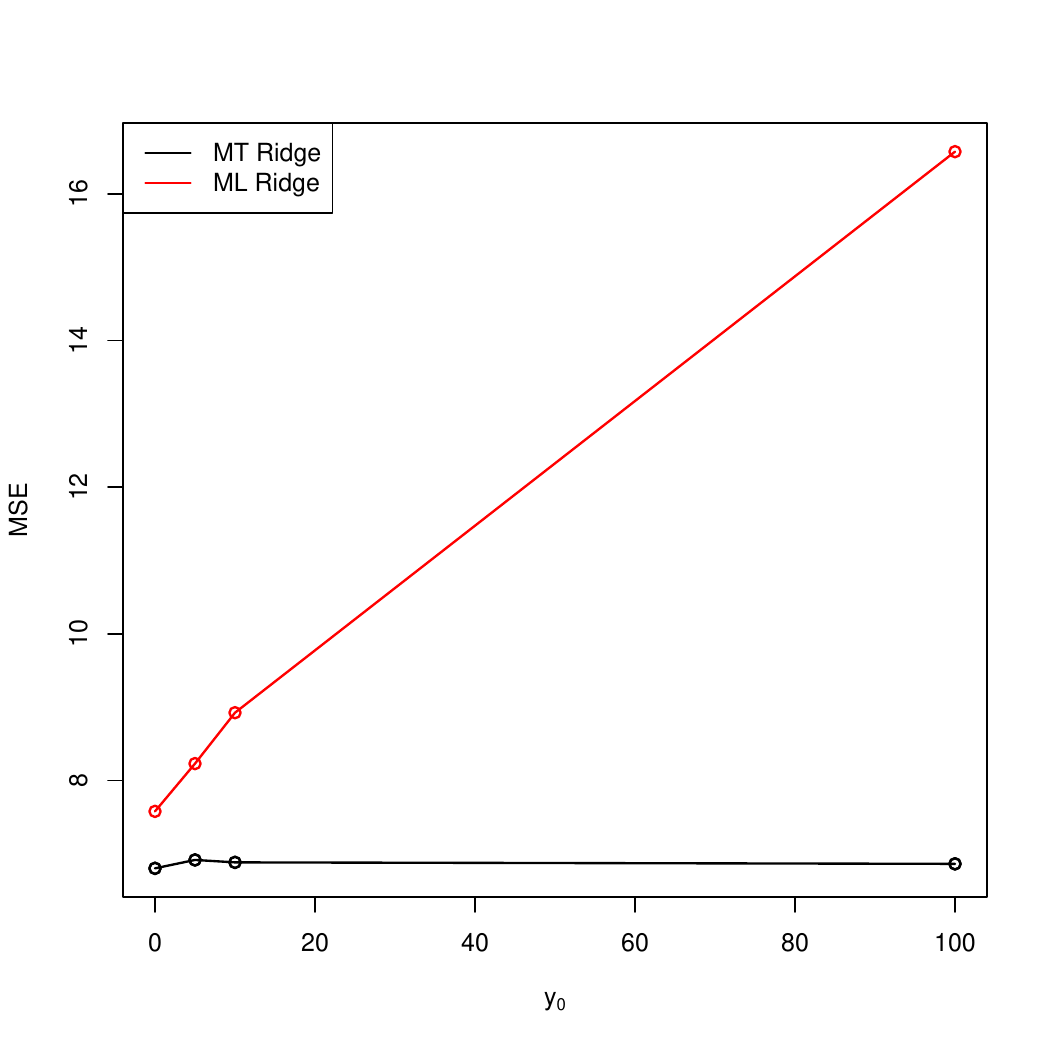}
			&   \includegraphics[scale=0.3]{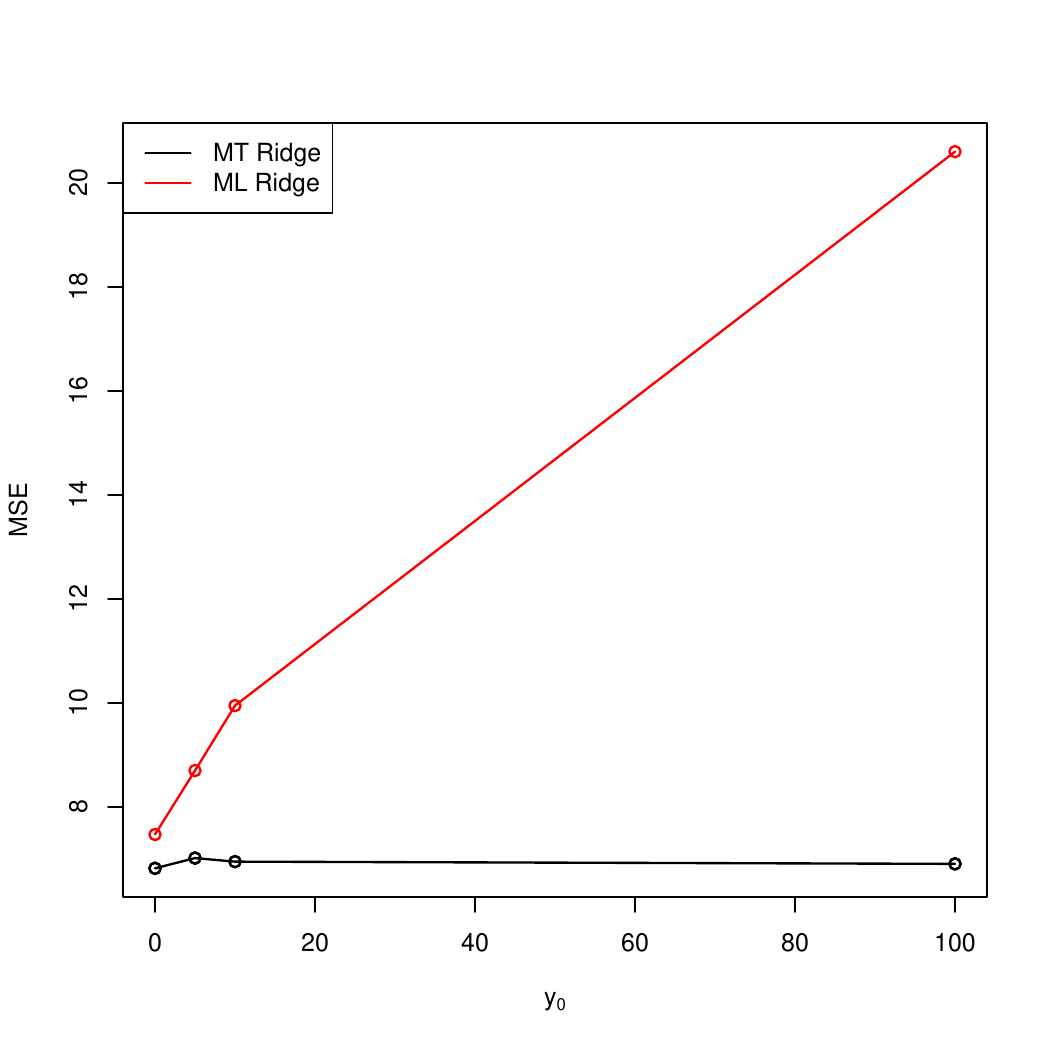} 
			&   \includegraphics[scale=0.3]{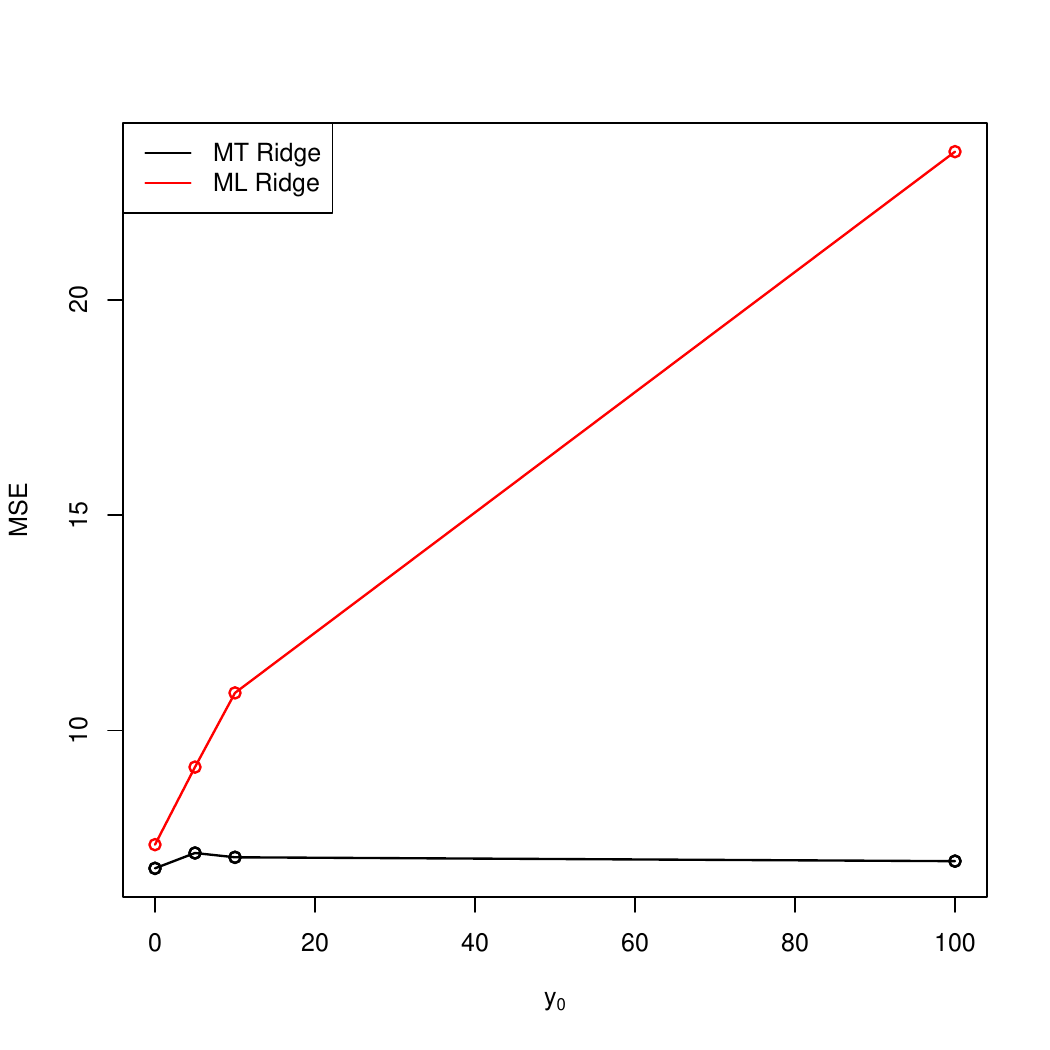}
		\end{tabular}
		\caption{\small \label{fig:msesimuRQLp400Ridge}  Mean squared errors for model AMR and Ridge methods, $p=400$.}
	\end{center} 
\end{figure}

\begin{figure}[tp]
	\begin{center}
		\footnotesize
		\renewcommand{\arraystretch}{0.22}
		\newcolumntype{M}{>{\centering\arraybackslash}m{\dimexpr.01\linewidth-1\tabcolsep}}
		\newcolumntype{G}{>{\centering\arraybackslash}m{\dimexpr.38\linewidth-1\tabcolsep}}
		\begin{tabular}{GGG}
			$n=50,\, \epsilon=0.05$ & $n=50,\, \epsilon=0.1$  & $n=50,\, \epsilon=0.15$ \\
			\includegraphics[scale=0.29]{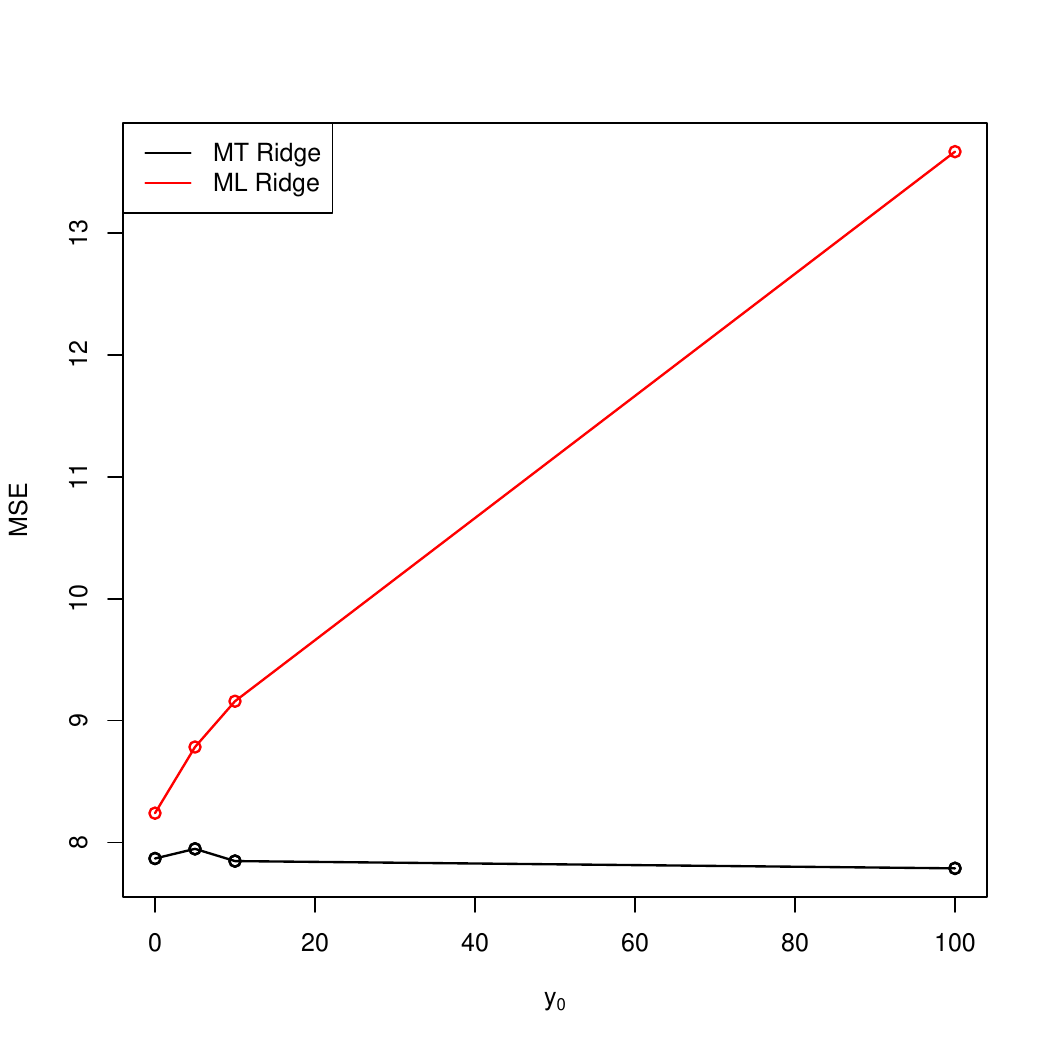} 
			&  \includegraphics[scale=0.29]{simuAMR-mse-ridge-1600-50-0-1}
			&   \includegraphics[scale=0.29]{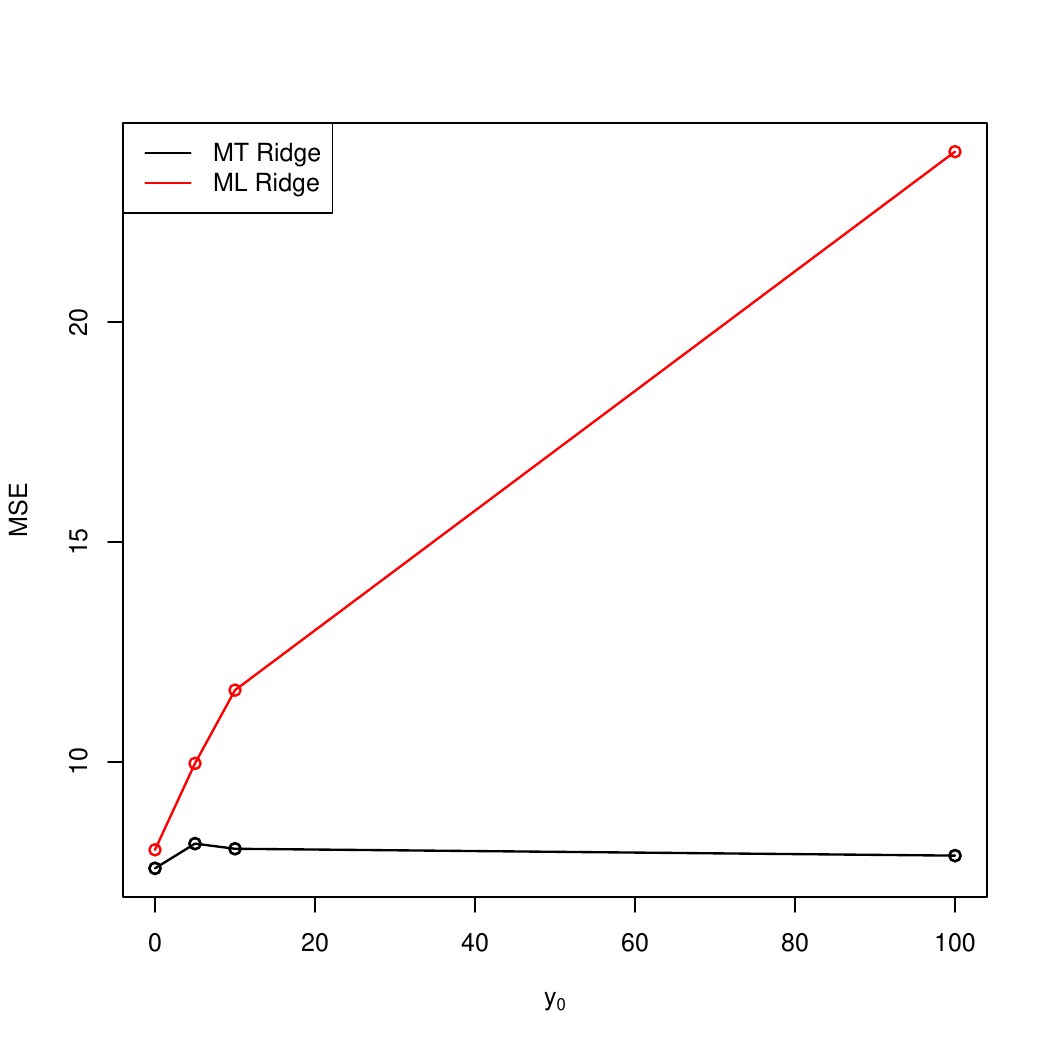}
			\\
			$n=100,\, \epsilon=0.05$ & $n=100,\, \epsilon=0.1$  & $n=100,\, \epsilon=0.15$ \\
			\includegraphics[scale=0.3]{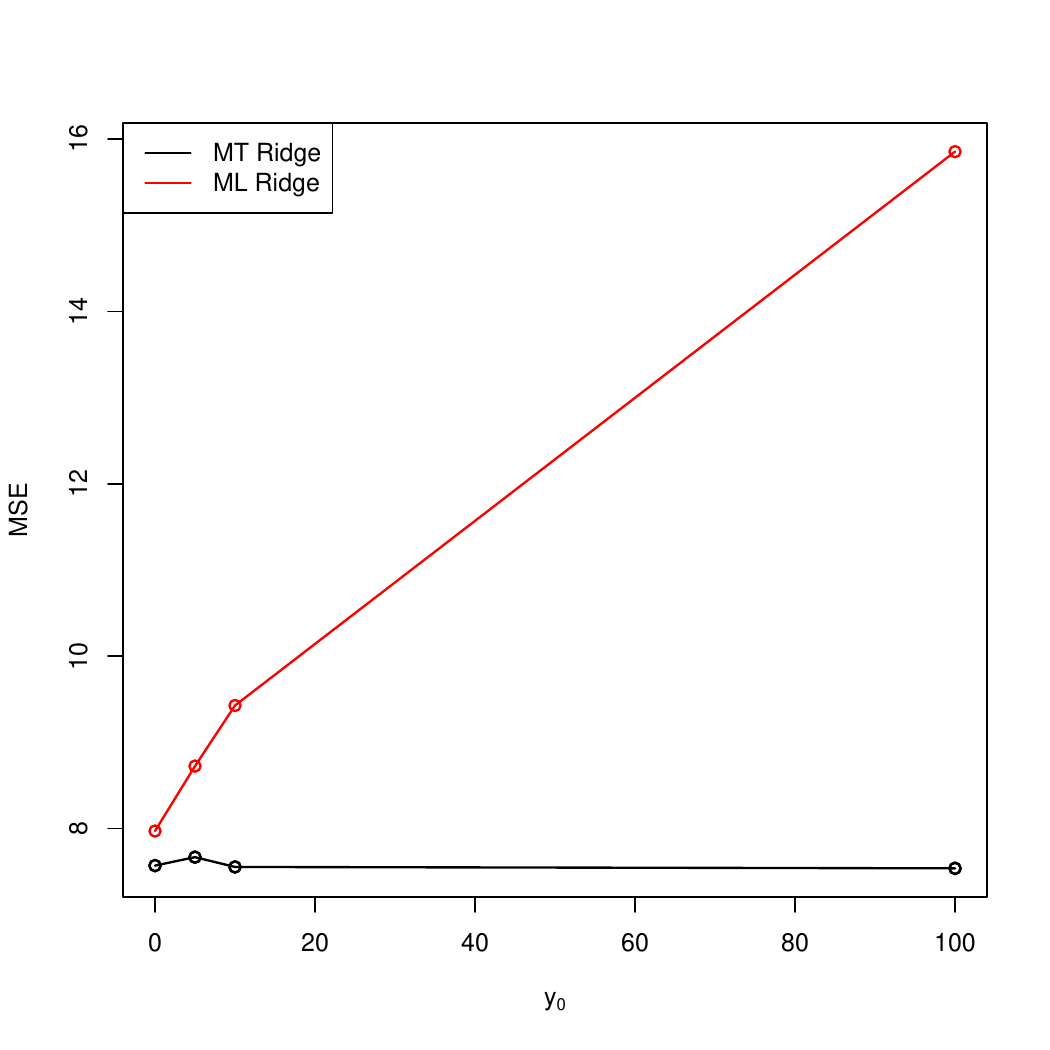}
			&   \includegraphics[scale=0.3]{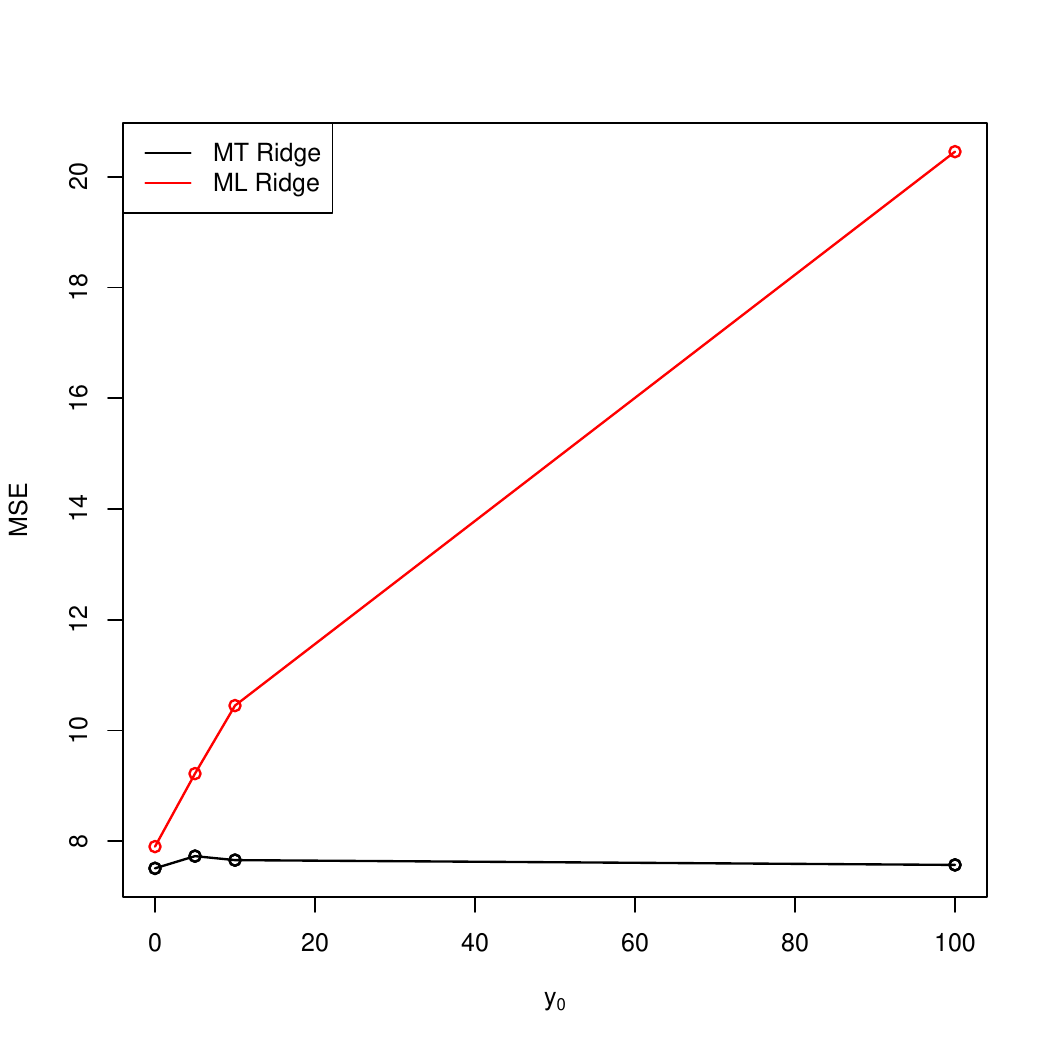} 
			&   \includegraphics[scale=0.3]{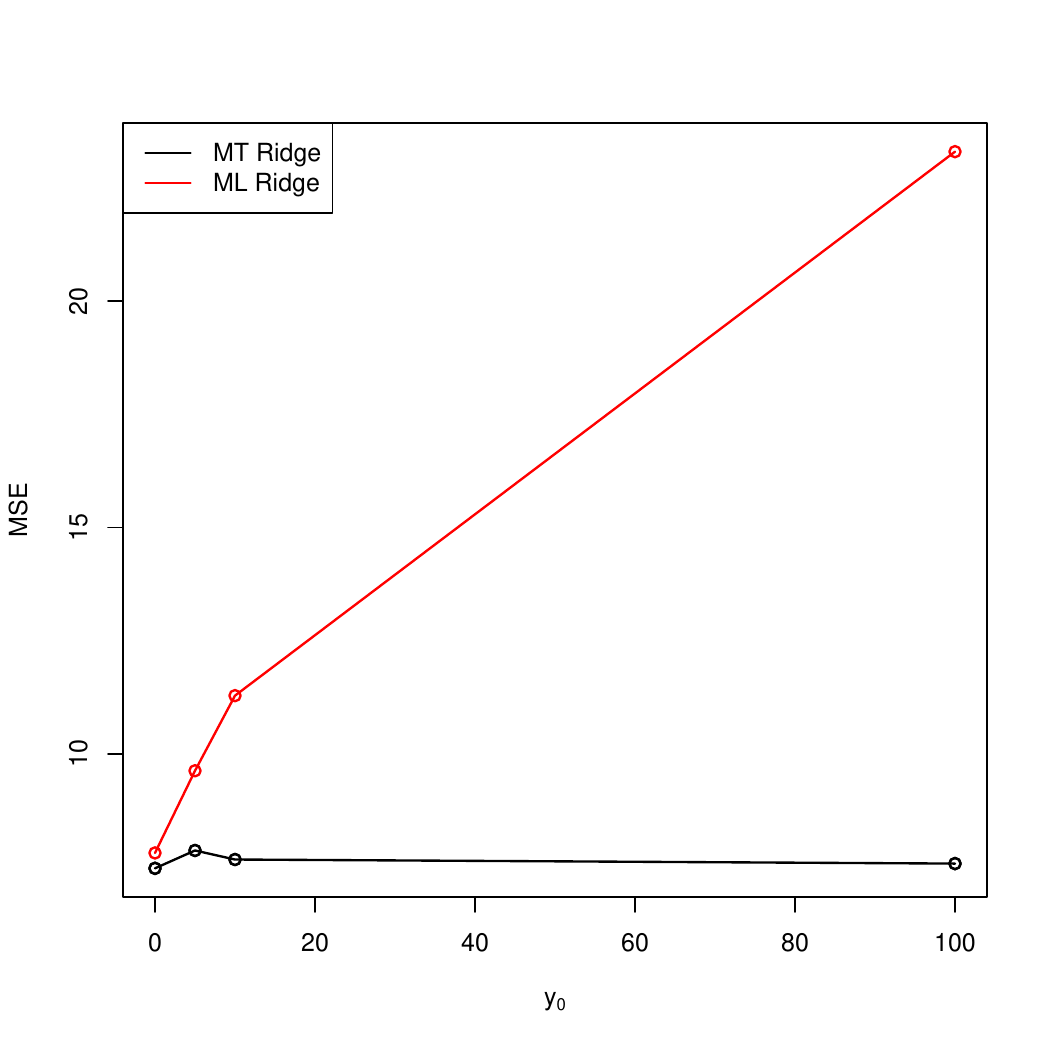} 
			\\
			$n=200,\, \epsilon=0.05$ & $n=200,\, \epsilon=0.1$  & $n=200,\, \epsilon=0.15$ \\
			\includegraphics[scale=0.3]{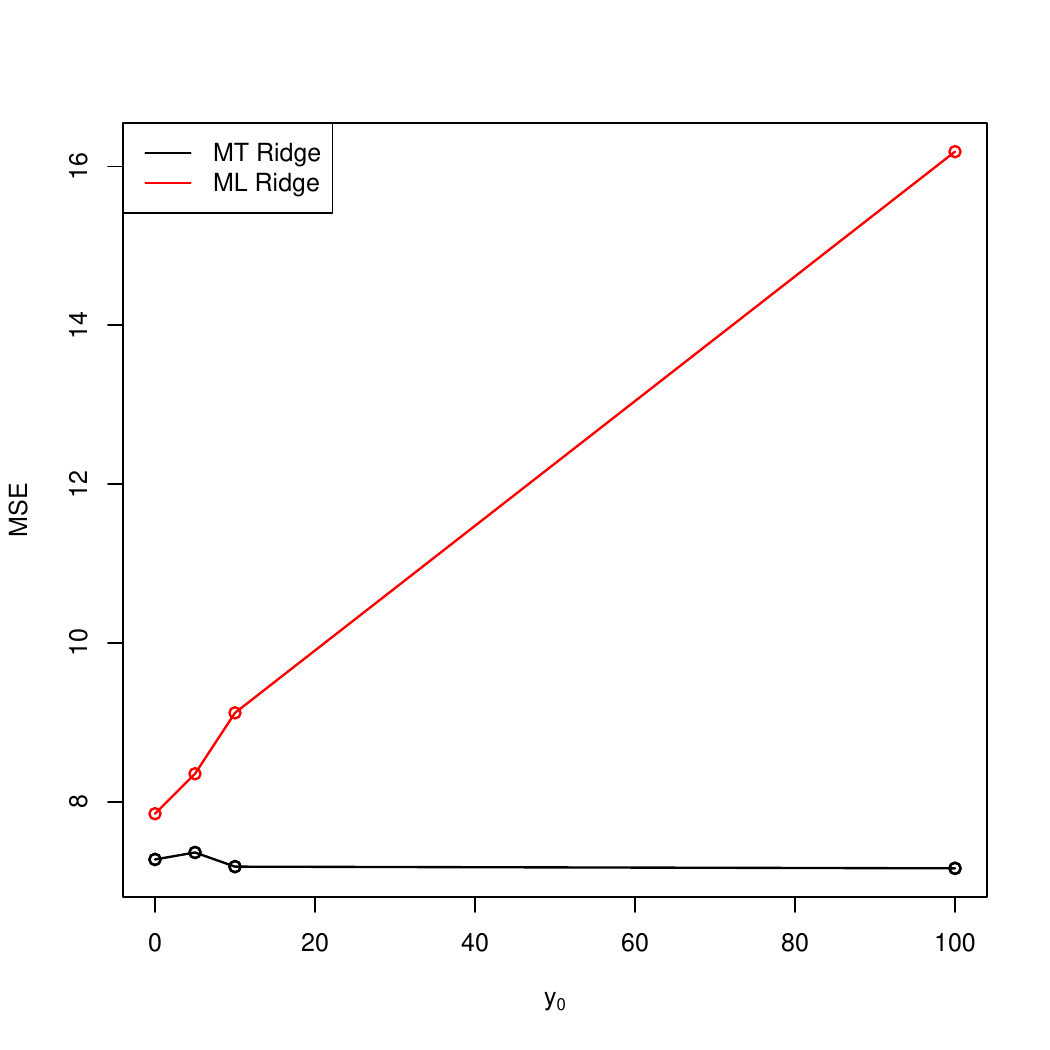}
			&   \includegraphics[scale=0.3]{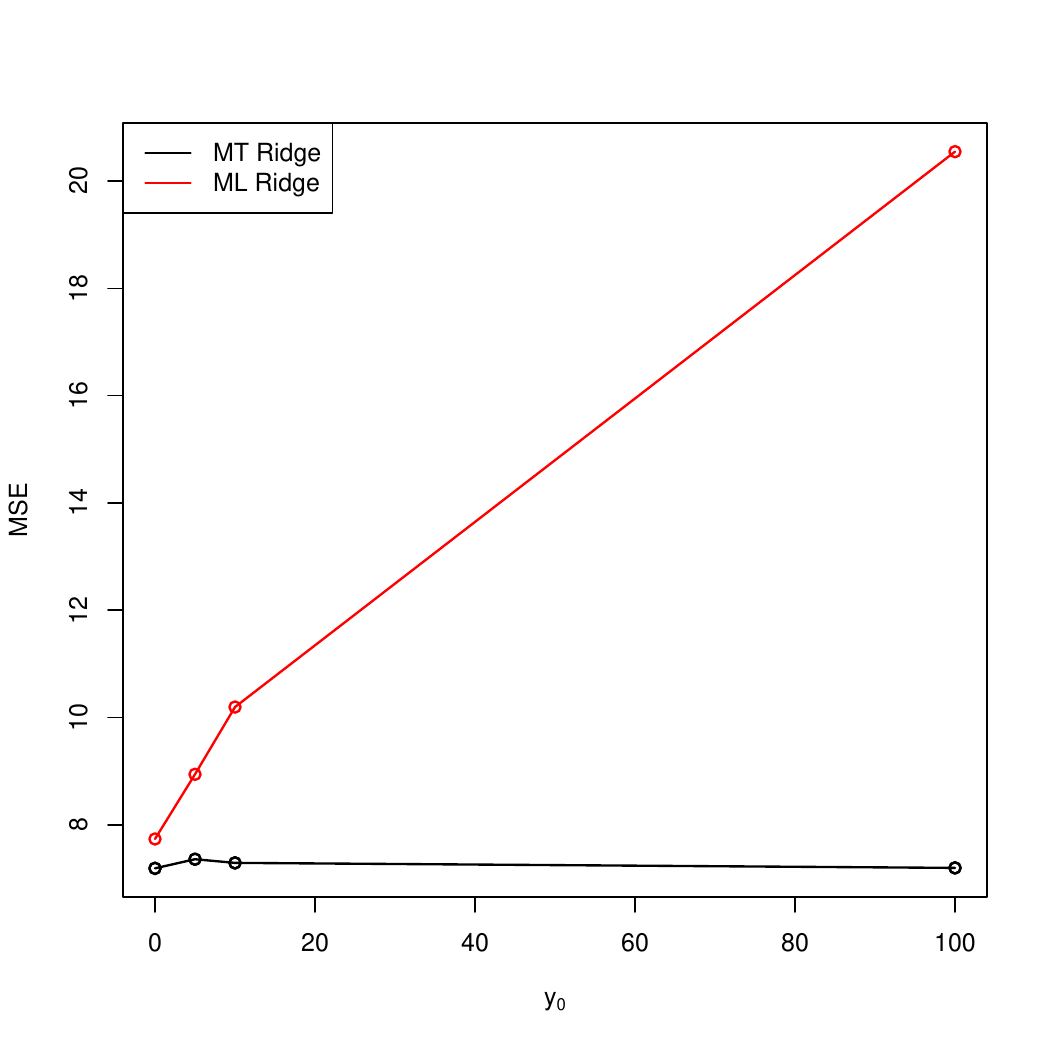} 
			&   \includegraphics[scale=0.3]{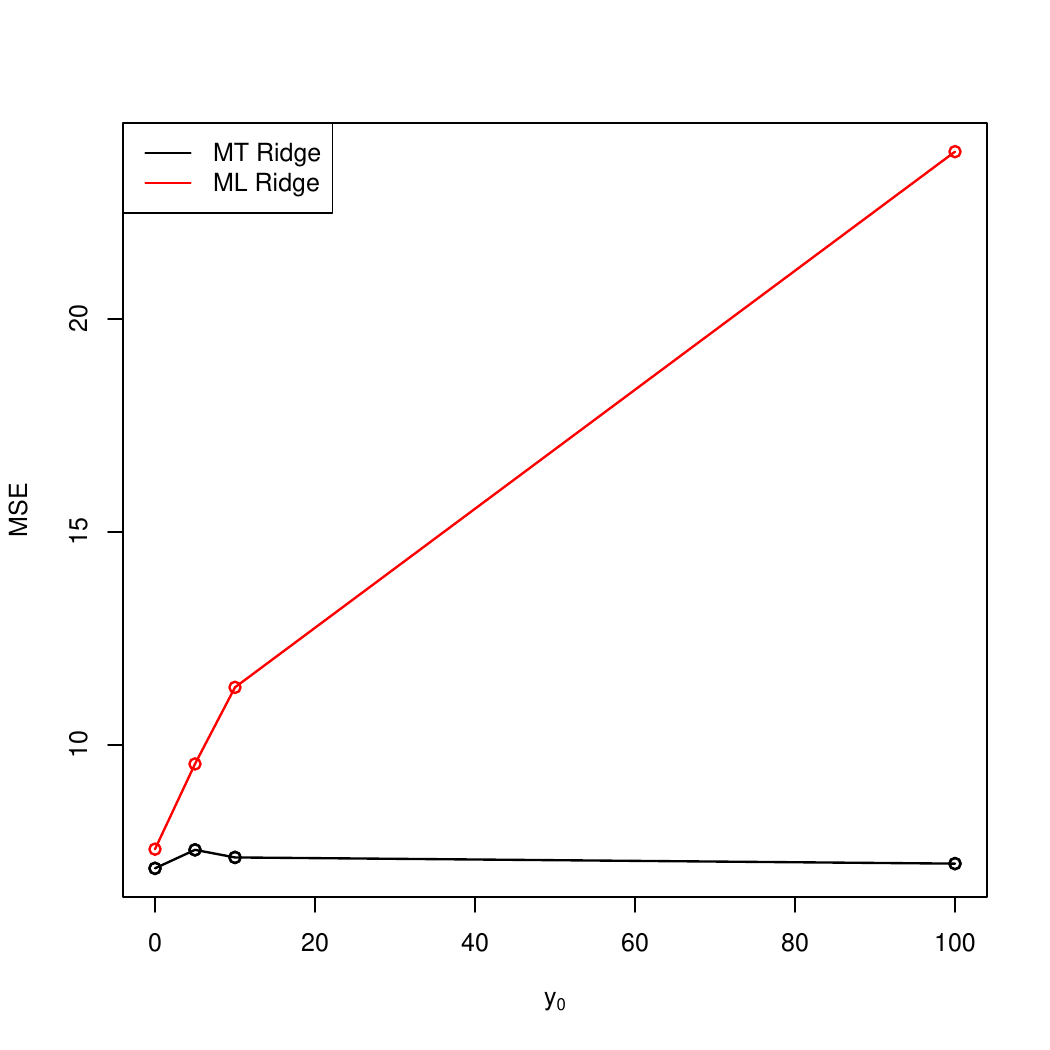}
		\end{tabular}
		\caption{\small \label{fig:msesimuRQLp1600Ridge}  Mean squared errors for model AMR and Ridge methods, $p=1600$.}
	\end{center} 
\end{figure}

%%%%%%%%%%%%%%%%%%%%%%%%%%% Lasso %%%%%%%%%%%%%%%%%%%%%%%%%%%%%%%%%%%%%%%%%%

\begin{figure}[tp]
 \begin{center}
 \footnotesize
 \renewcommand{\arraystretch}{0.22}
 \newcolumntype{M}{>{\centering\arraybackslash}m{\dimexpr.01\linewidth-1\tabcolsep}}
   \newcolumntype{G}{>{\centering\arraybackslash}m{\dimexpr.38\linewidth-1\tabcolsep}}
\begin{tabular}{GGG}
  $n=100,\, \epsilon=0.05$ & $n=100,\, \epsilon=0.1$  & $n=100,\, \epsilon=0.15$ \\
\includegraphics[scale=0.29]{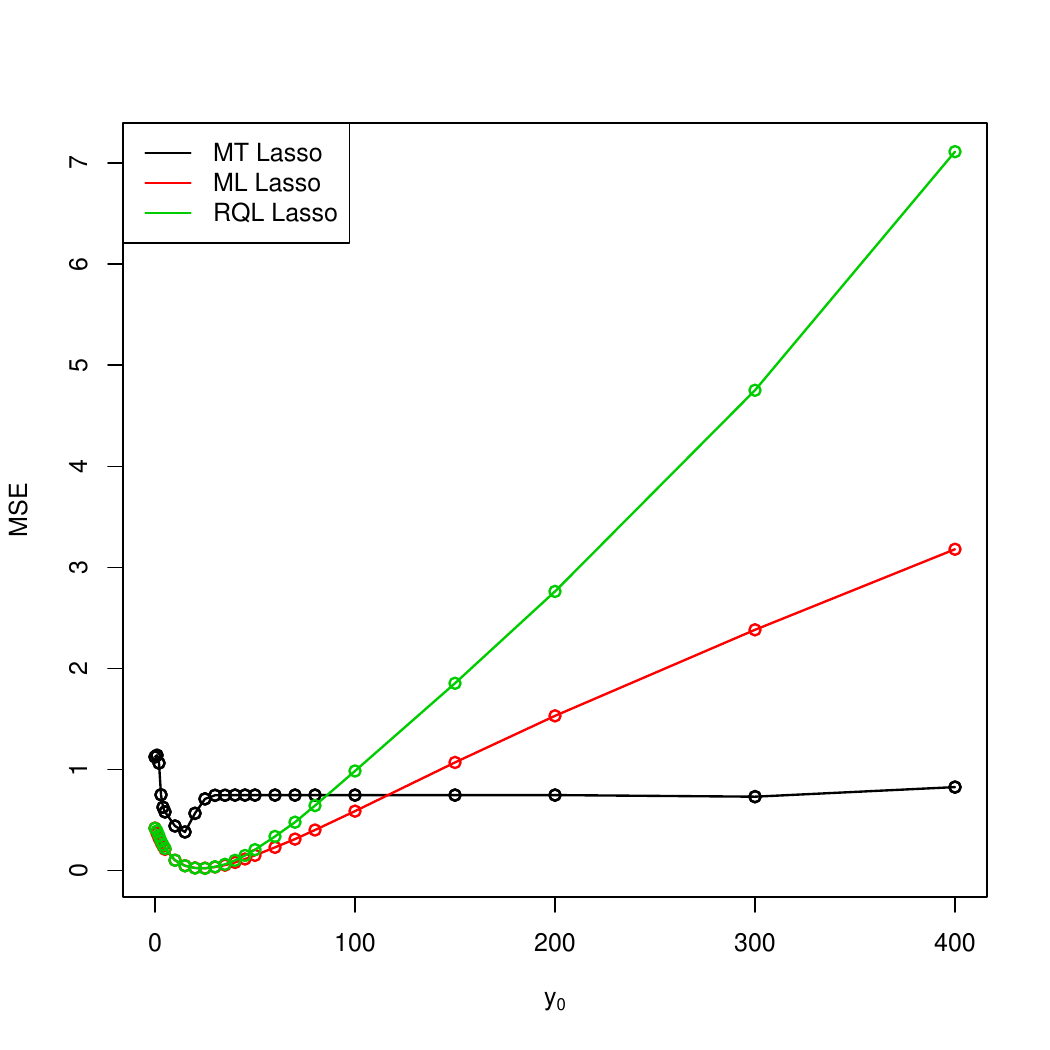} 
&  \includegraphics[scale=0.29]{simu1p10-mse-lasso-100-0-1}
&   \includegraphics[scale=0.29]{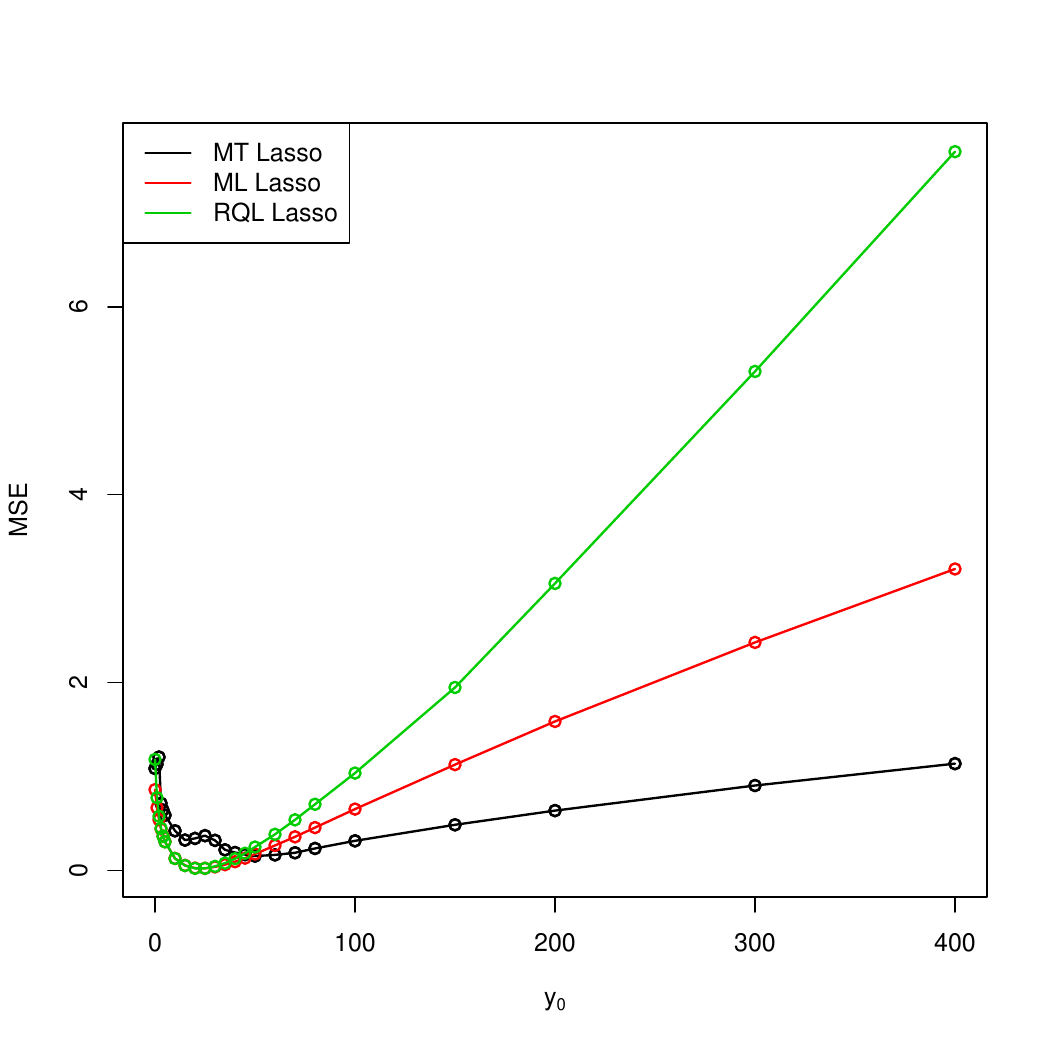}
\\
  $n=400,\, \epsilon=0.05$ & $n=400,\, \epsilon=0.1$  & $n=400,\, \epsilon=0.15$ \\
\includegraphics[scale=0.3]{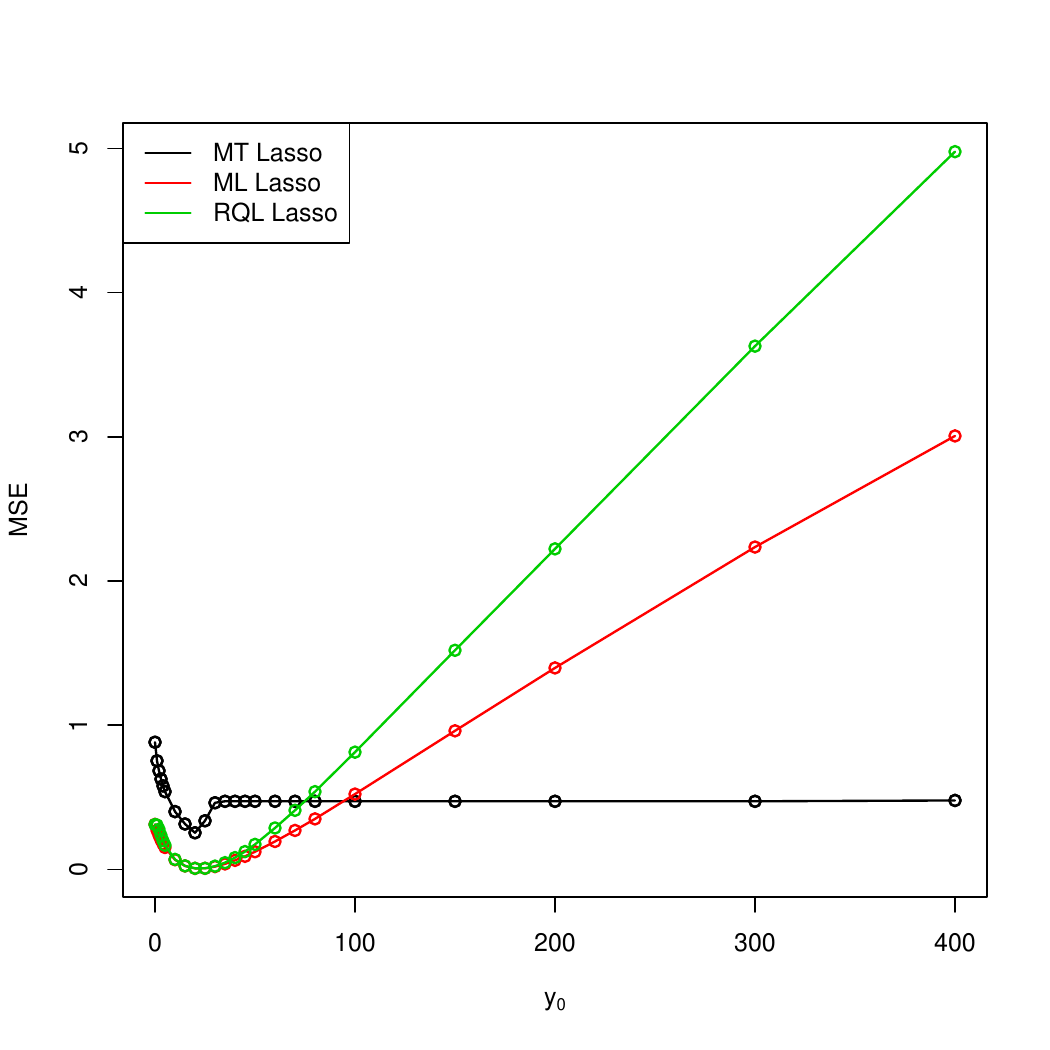}
&   \includegraphics[scale=0.3]{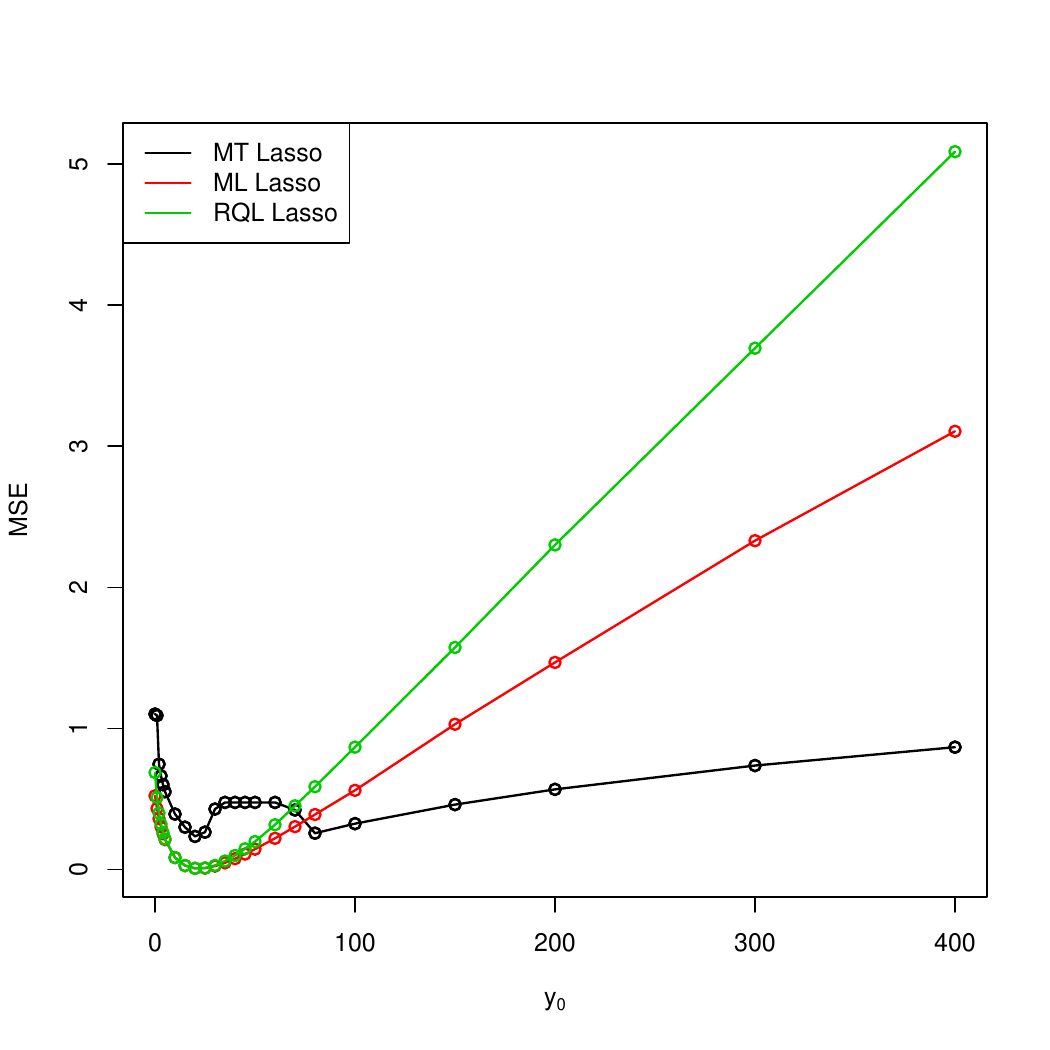} 
&   \includegraphics[scale=0.3]{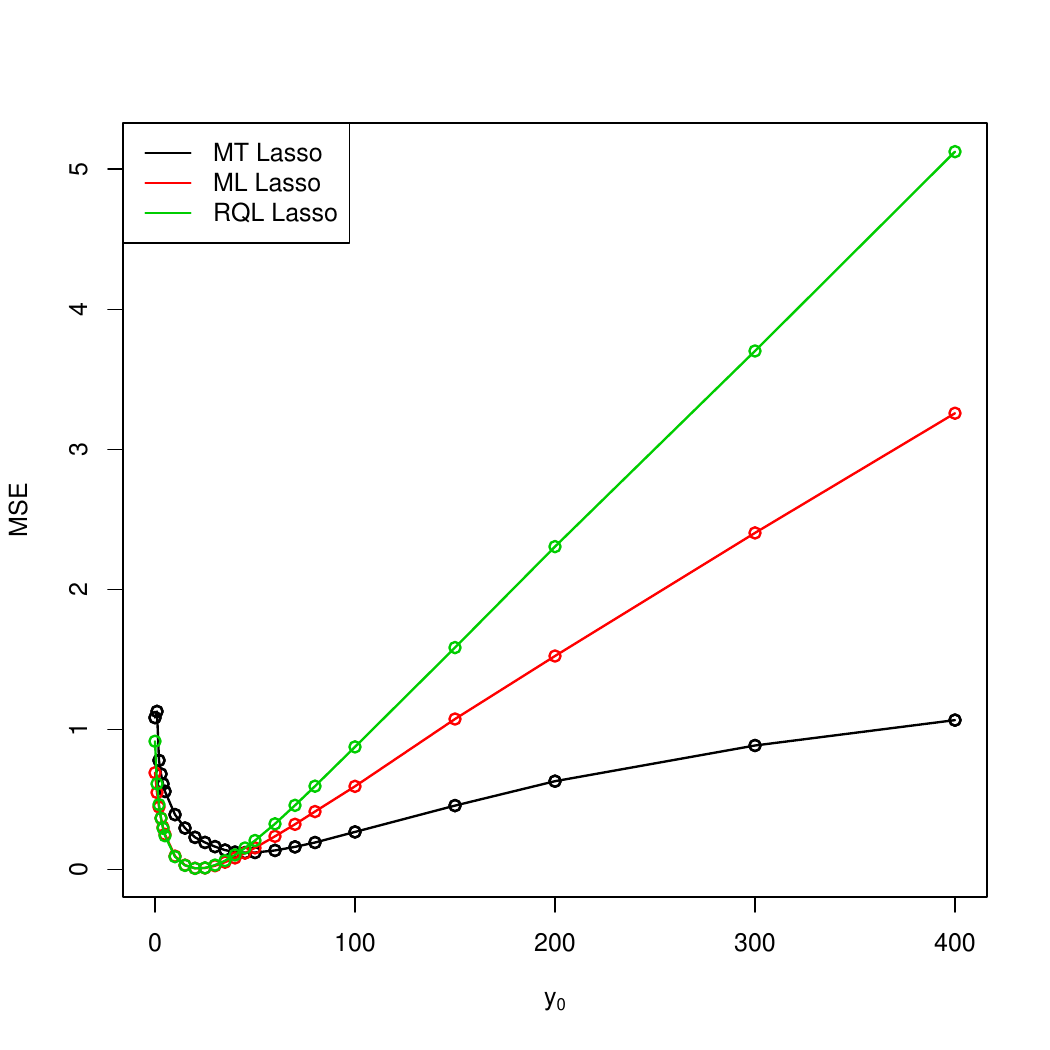} 
\\
$n=1000,\, \epsilon=0.05$ & $n=1000,\, \epsilon=0.1$  & $n=1000,\, \epsilon=0.15$ \\
\includegraphics[scale=0.3]{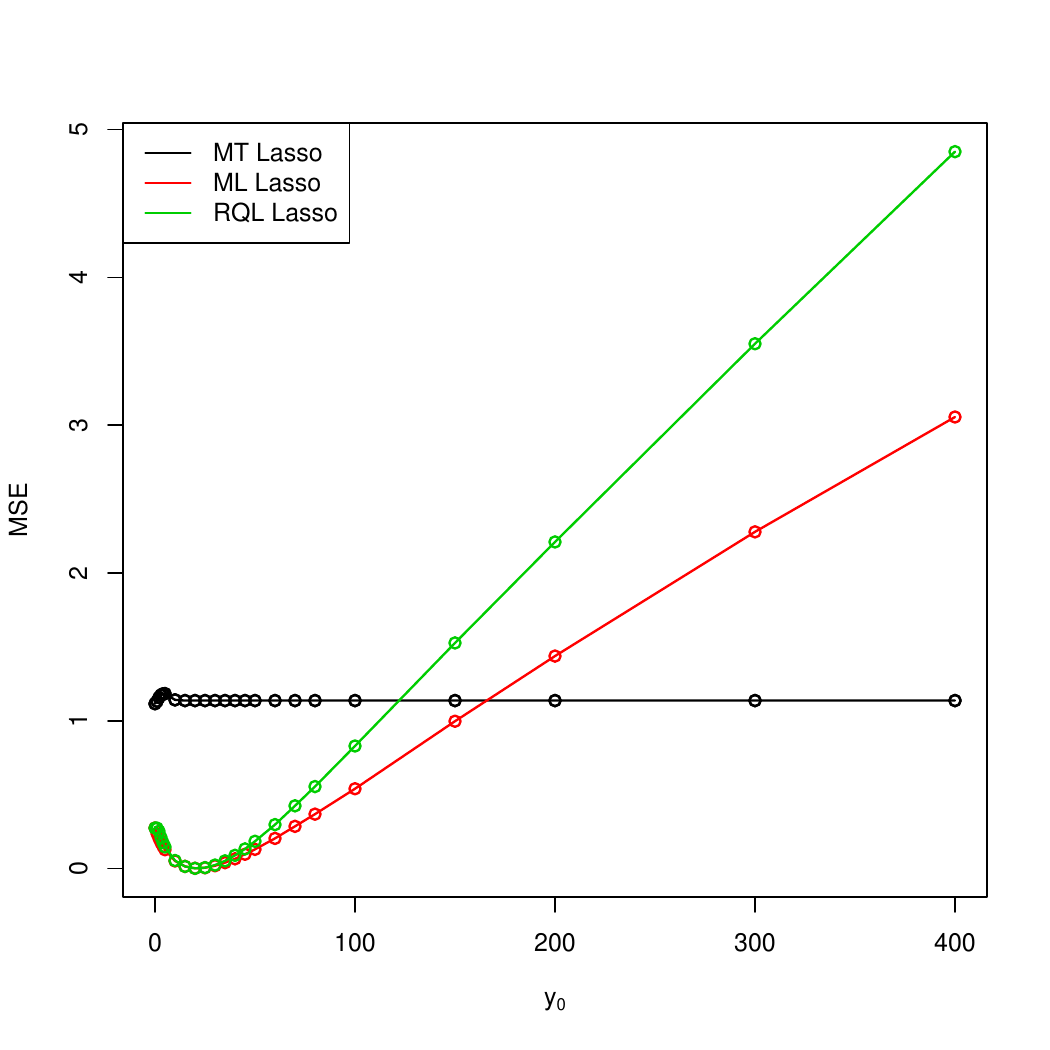}
&   \includegraphics[scale=0.3]{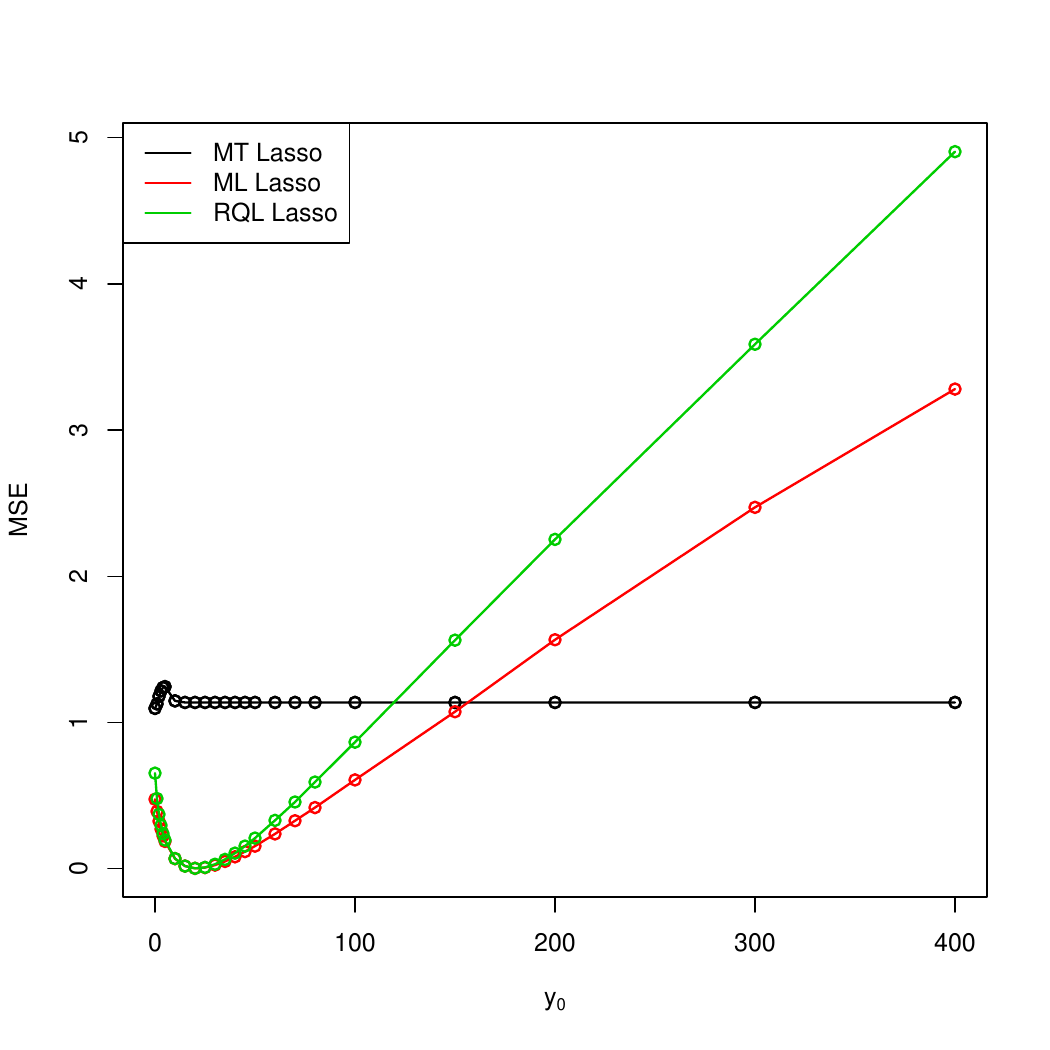} 
&   \includegraphics[scale=0.3]{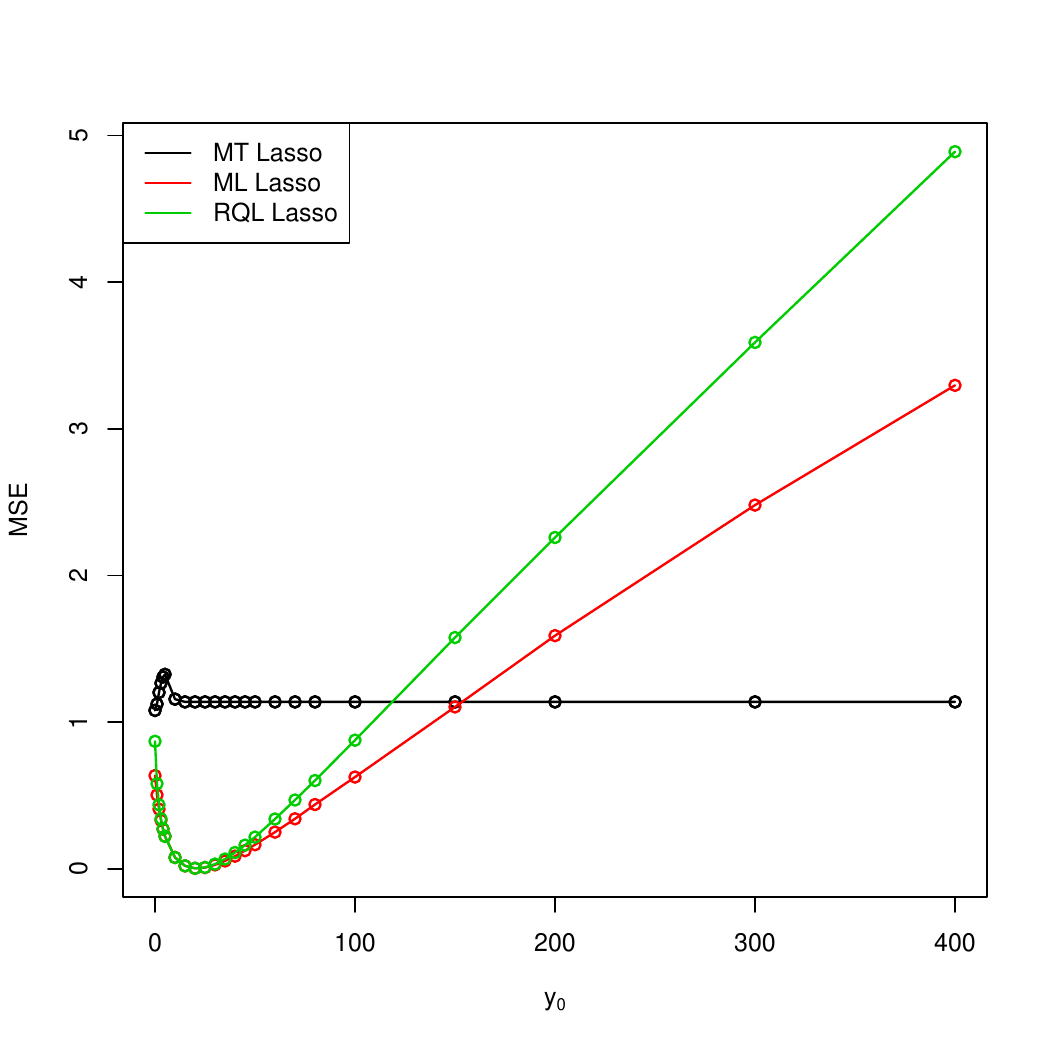}
 \end{tabular}
\caption{\small \label{fig:msesimu1p10Lasso}  Mean squared errors for model AVY and Lasso methods, $p=10$.}
\end{center} 
\end{figure}

\begin{figure}[tp]
	\begin{center}
		\footnotesize
		\renewcommand{\arraystretch}{0.22}
		\newcolumntype{M}{>{\centering\arraybackslash}m{\dimexpr.01\linewidth-1\tabcolsep}}
		\newcolumntype{G}{>{\centering\arraybackslash}m{\dimexpr.38\linewidth-1\tabcolsep}}
		\begin{tabular}{GGG}
			$n=100,\, \epsilon=0.05$ & $n=100,\, \epsilon=0.1$  & $n=100,\, \epsilon=0.15$ \\
			\includegraphics[scale=0.29]{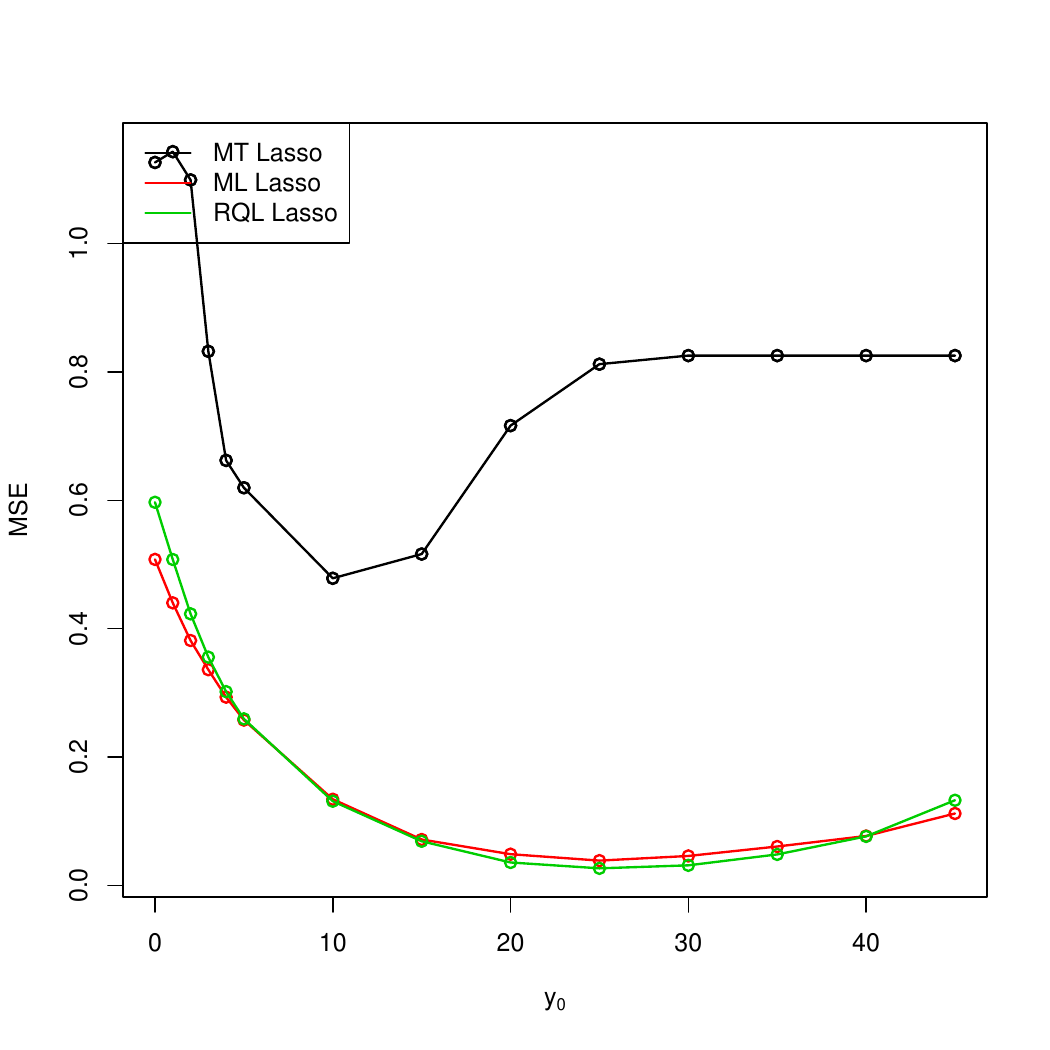} 
			&  \includegraphics[scale=0.29]{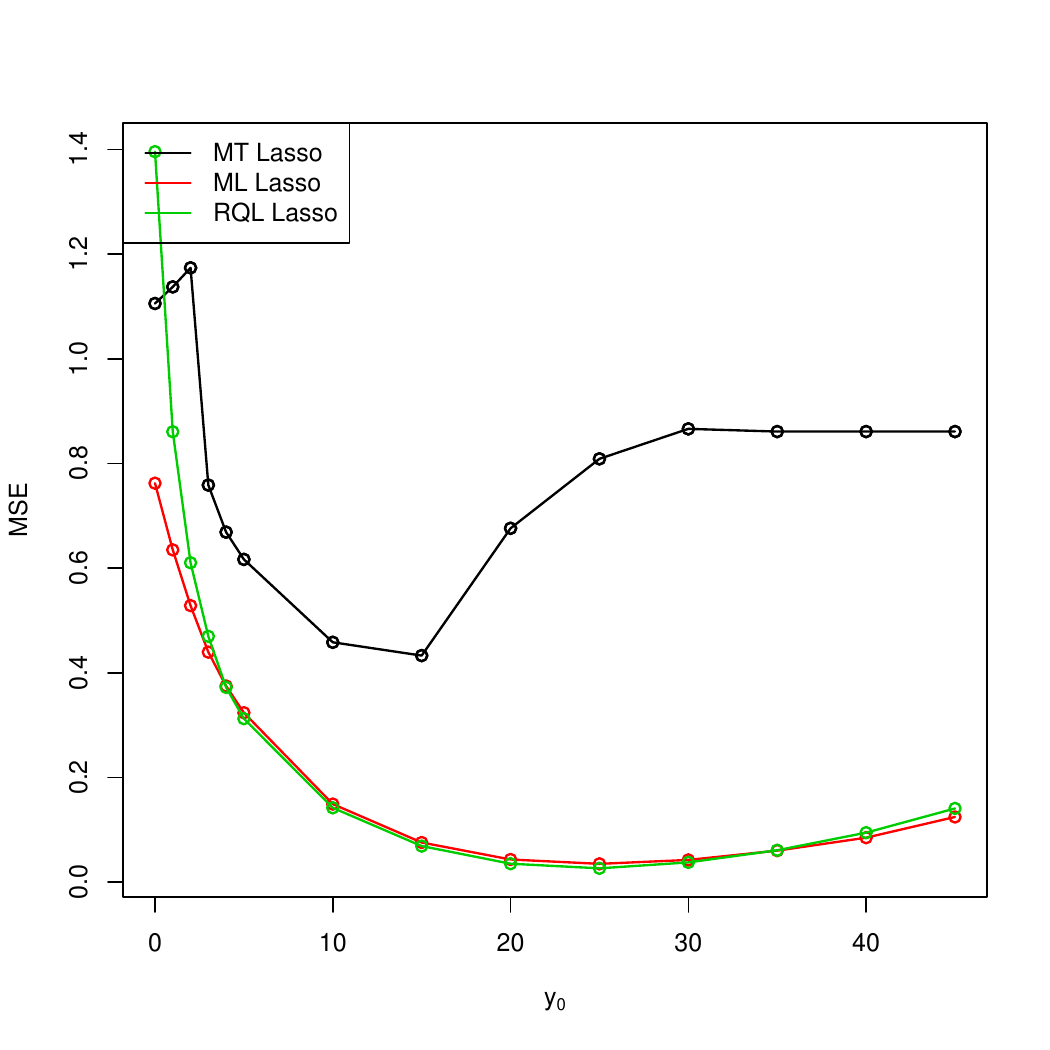}
			&   \includegraphics[scale=0.29]{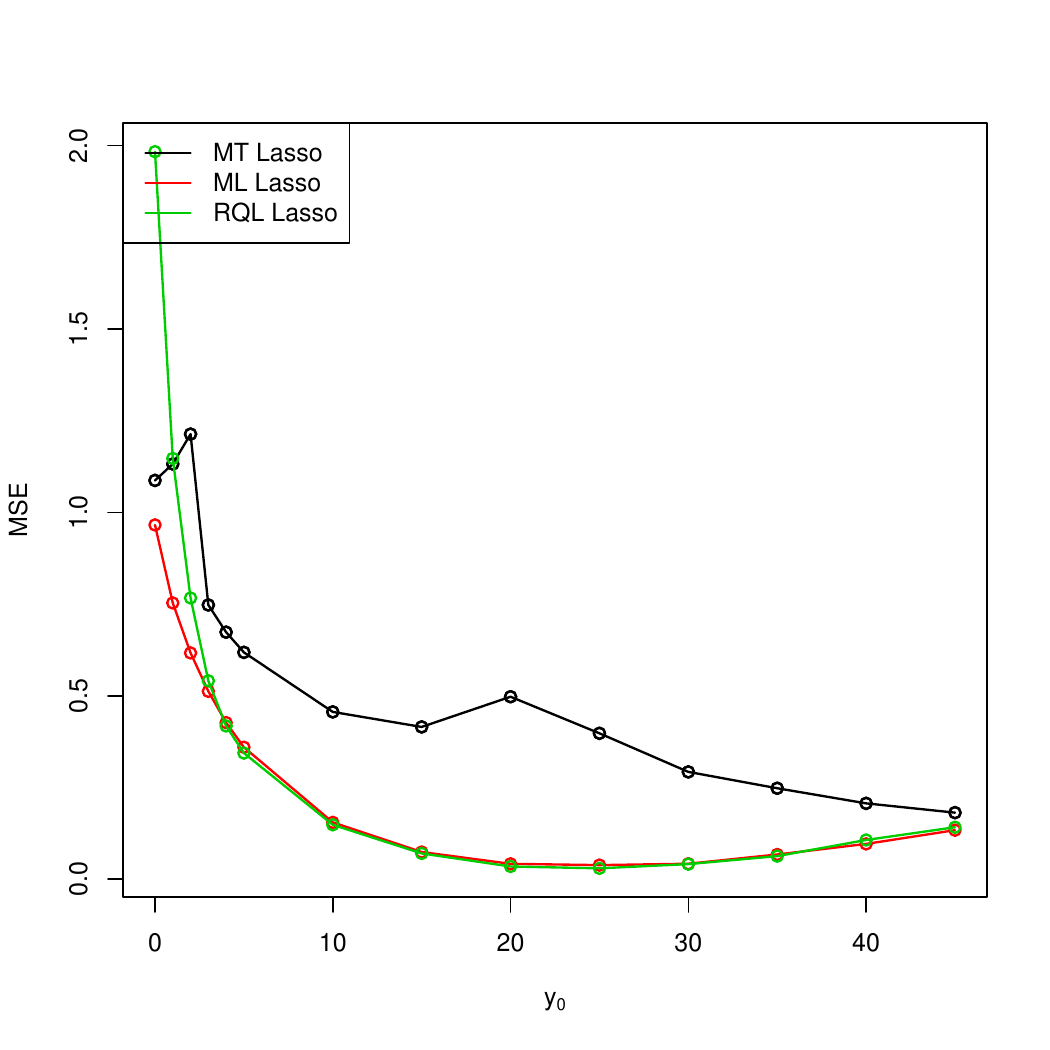}
			\\
			$n=400,\, \epsilon=0.05$ & $n=400,\, \epsilon=0.1$  & $n=400,\, \epsilon=0.15$ \\
			\includegraphics[scale=0.3]{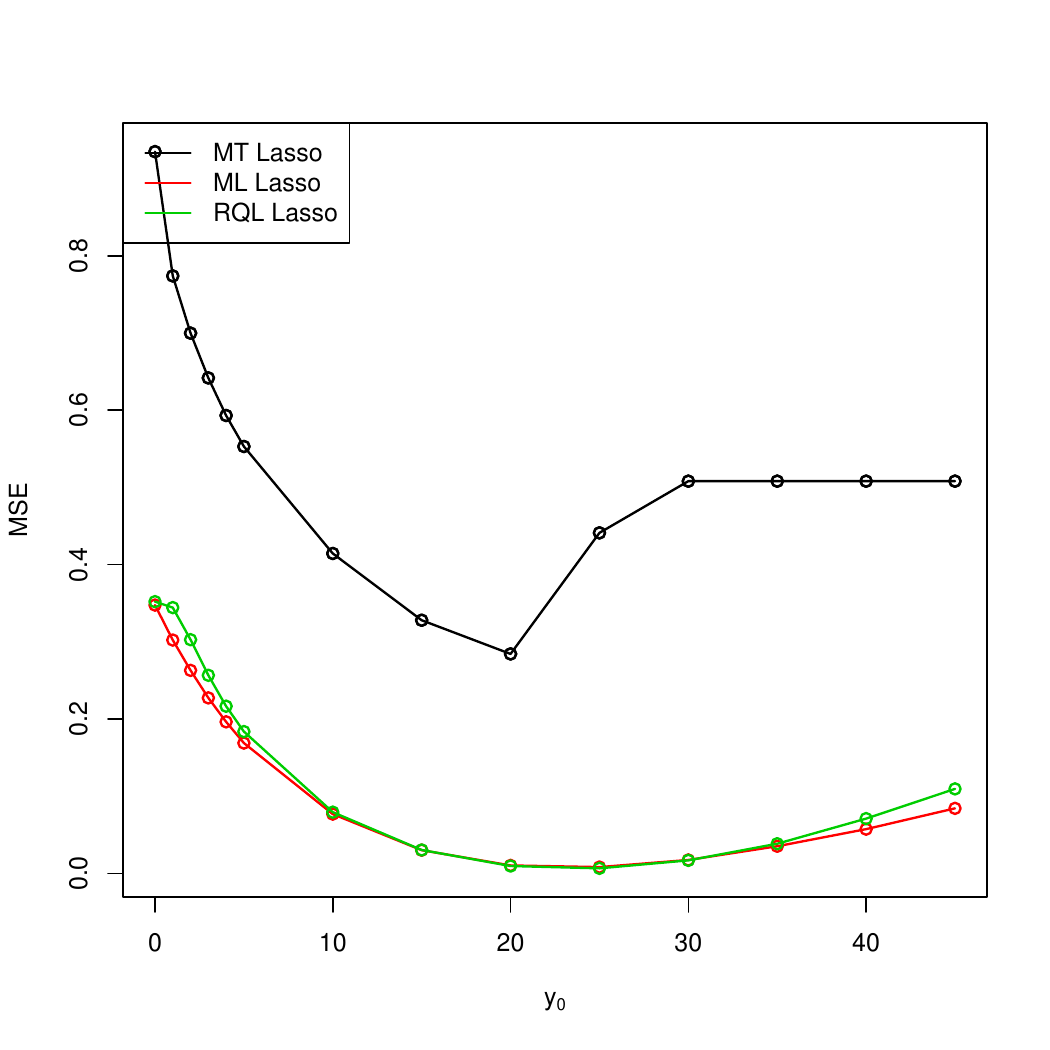}
			&   \includegraphics[scale=0.3]{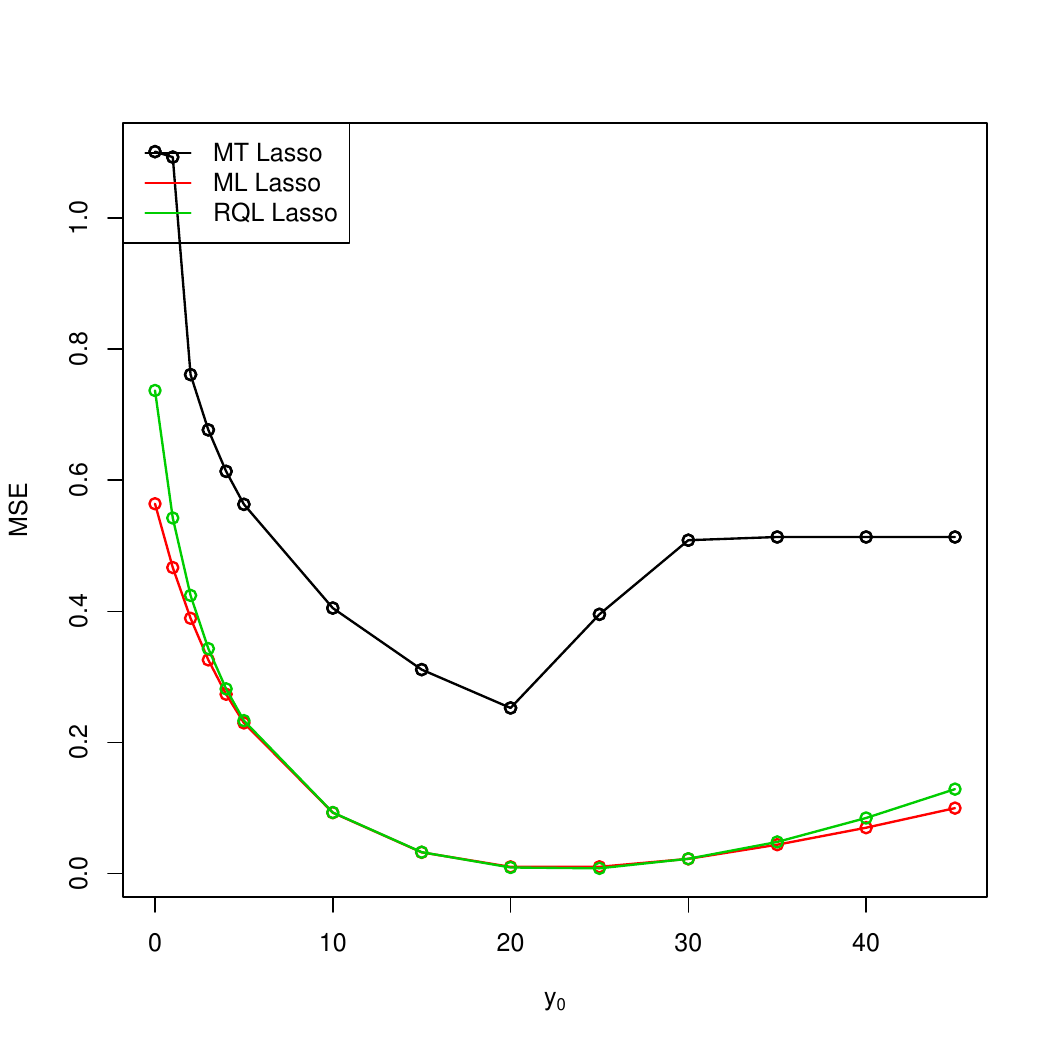} 
			&   \includegraphics[scale=0.3]{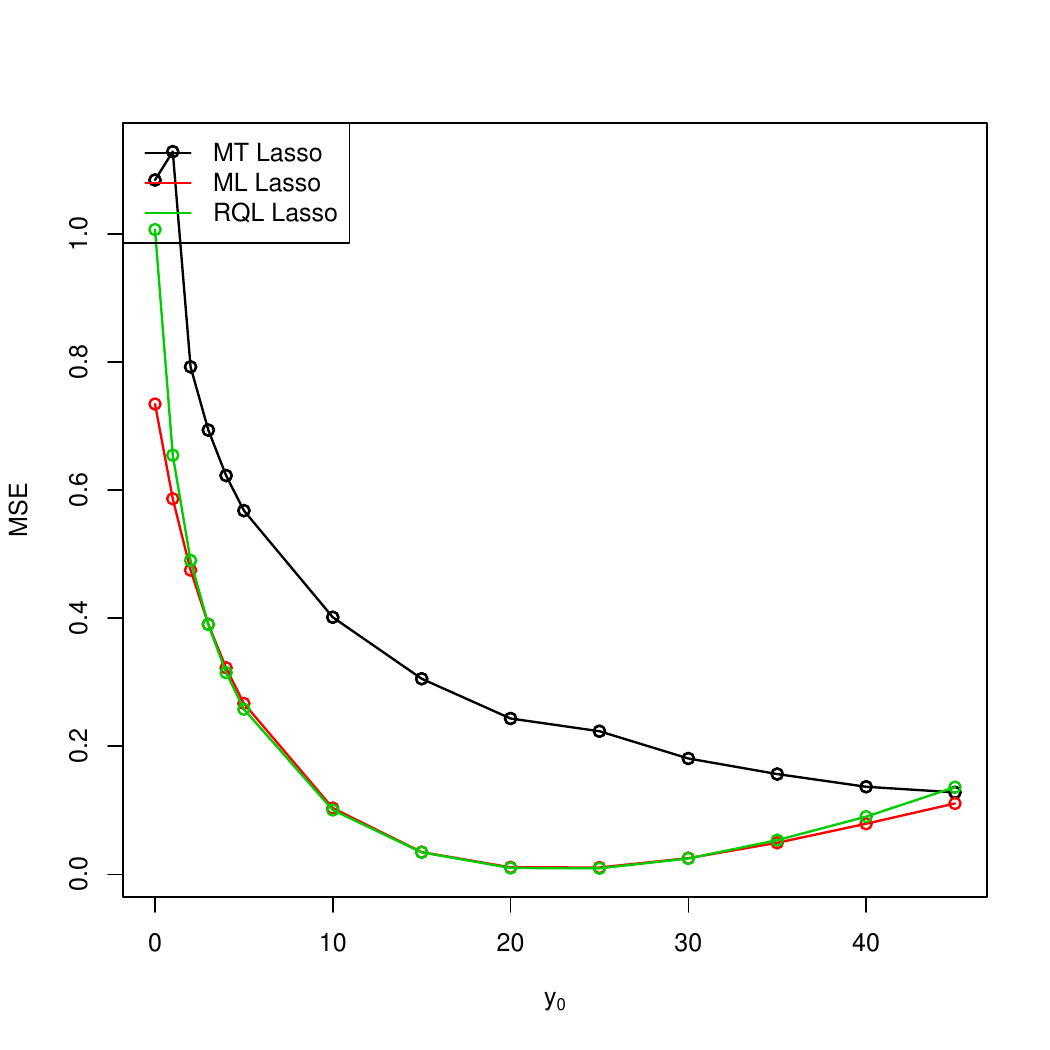} 
			\\
			$n=1000,\, \epsilon=0.05$ & $n=1000,\, \epsilon=0.1$  & $n=1000,\, \epsilon=0.15$ \\
			\includegraphics[scale=0.3]{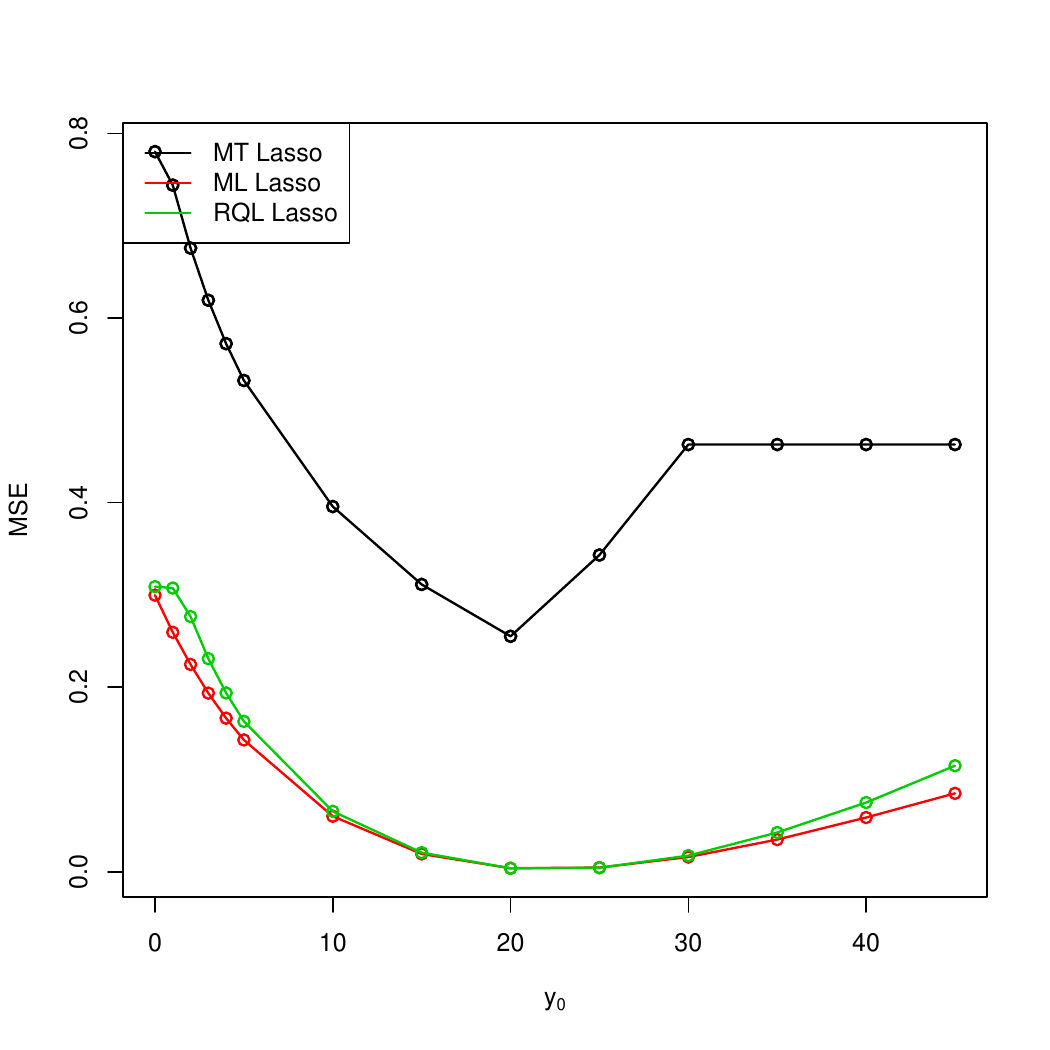}
			&   \includegraphics[scale=0.3]{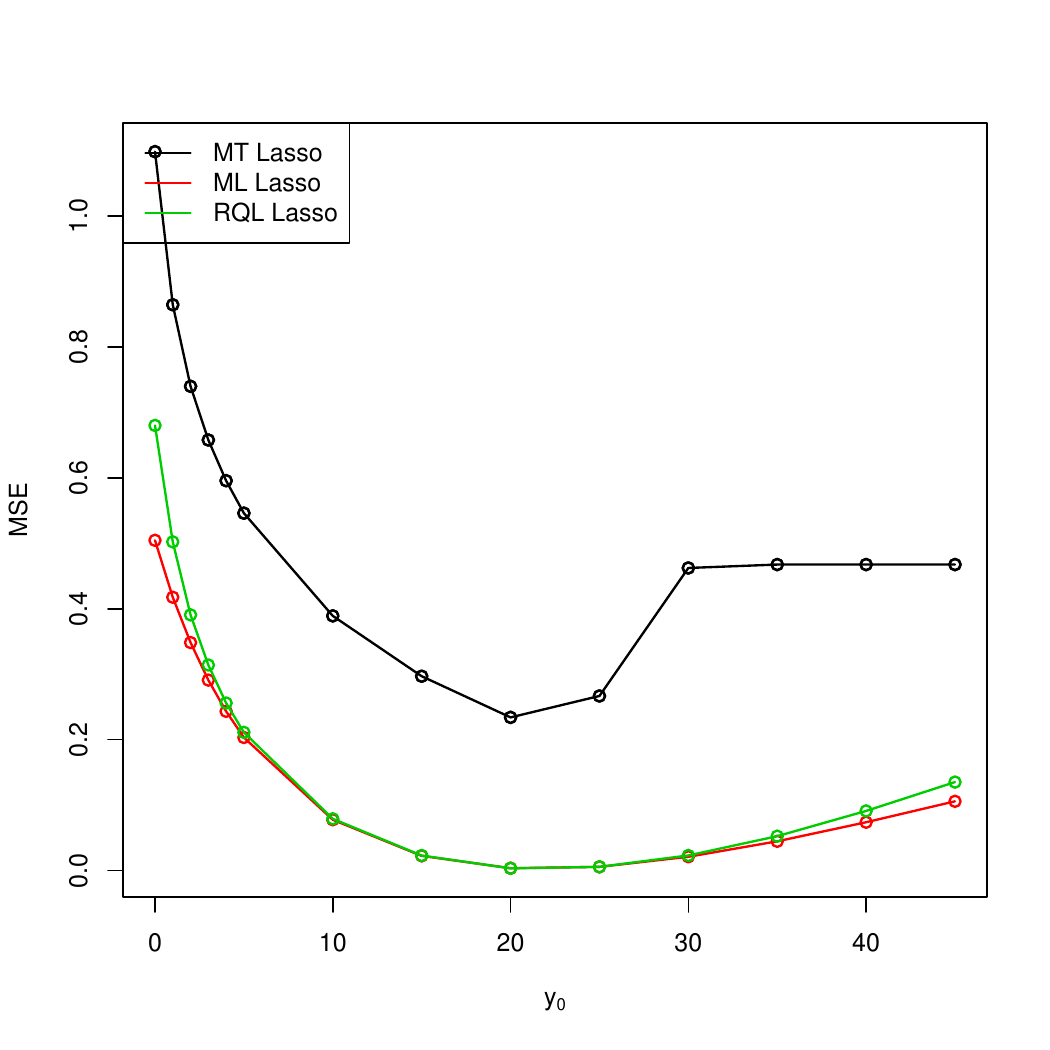} 
			&   \includegraphics[scale=0.3]{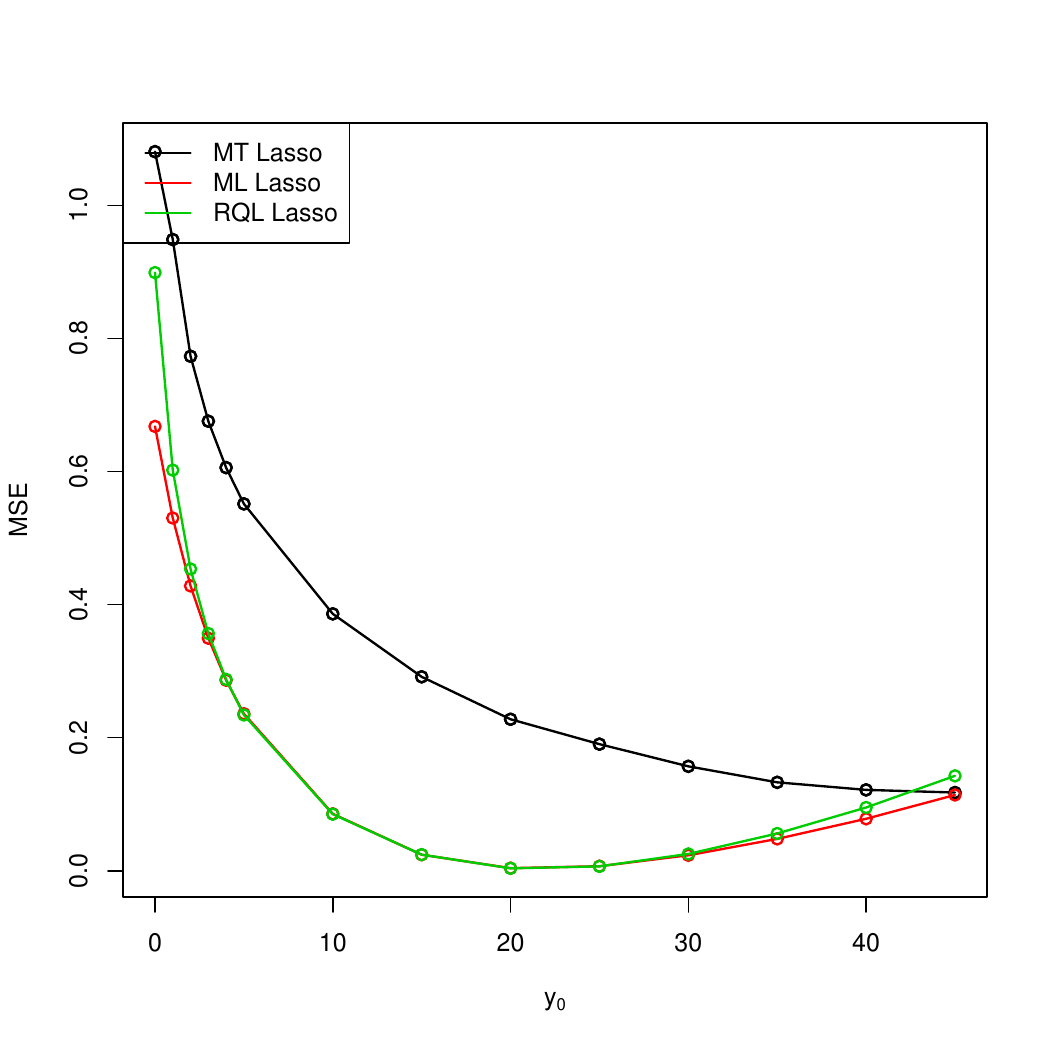}
		\end{tabular}
		\caption{\small \label{fig:msesimu1p50Lasso}  Mean squared errors for model AVY and Lasso methods, $p=50$.}
	\end{center} 
\end{figure}

\begin{figure}[tp]
	\begin{center}
		\footnotesize
		\renewcommand{\arraystretch}{0.22}
		\newcolumntype{M}{>{\centering\arraybackslash}m{\dimexpr.01\linewidth-1\tabcolsep}}
		\newcolumntype{G}{>{\centering\arraybackslash}m{\dimexpr.38\linewidth-1\tabcolsep}}
		\begin{tabular}{GGG}
			$n=100,\, \epsilon=0.05$ & $n=100,\, \epsilon=0.1$  & $n=100,\, \epsilon=0.15$ \\
			\includegraphics[scale=0.29]{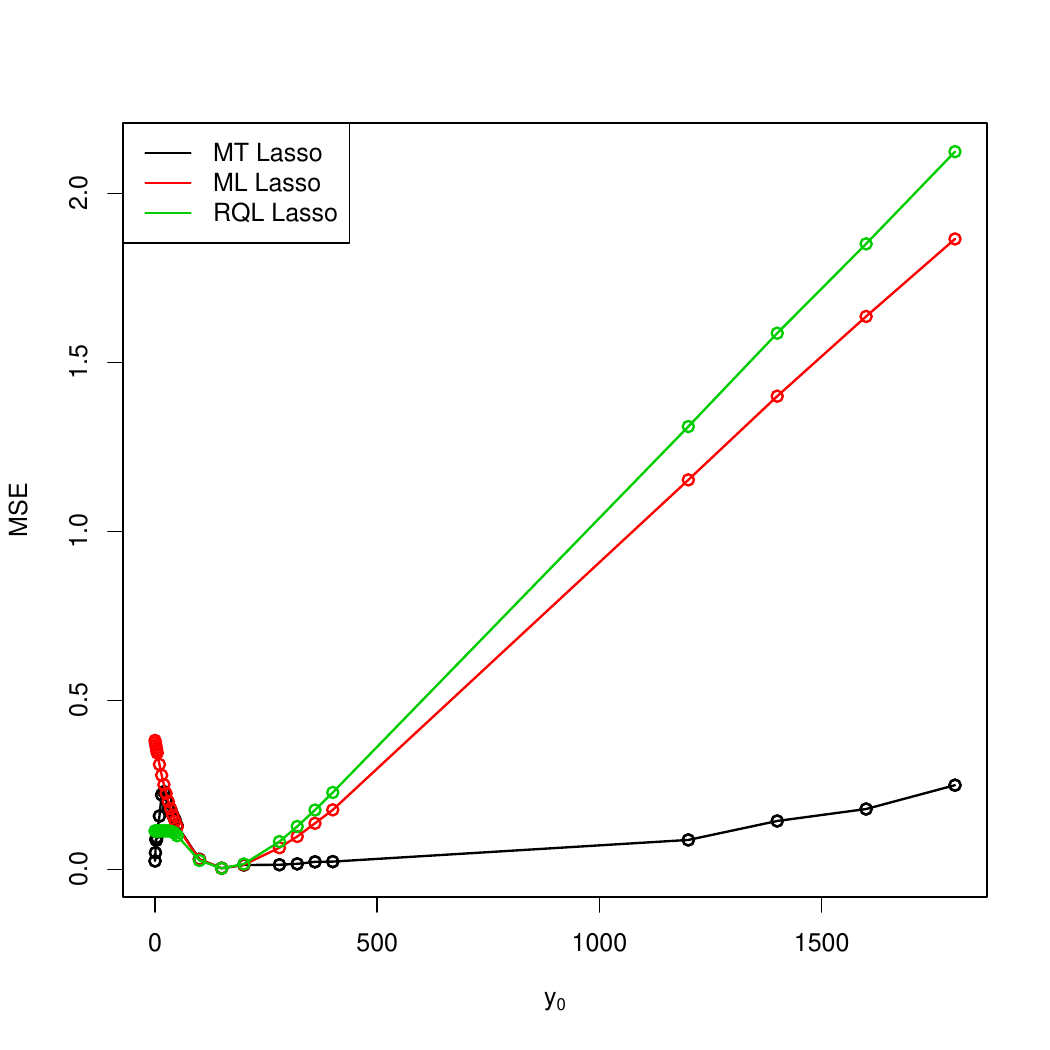} 
			&  \includegraphics[scale=0.29]{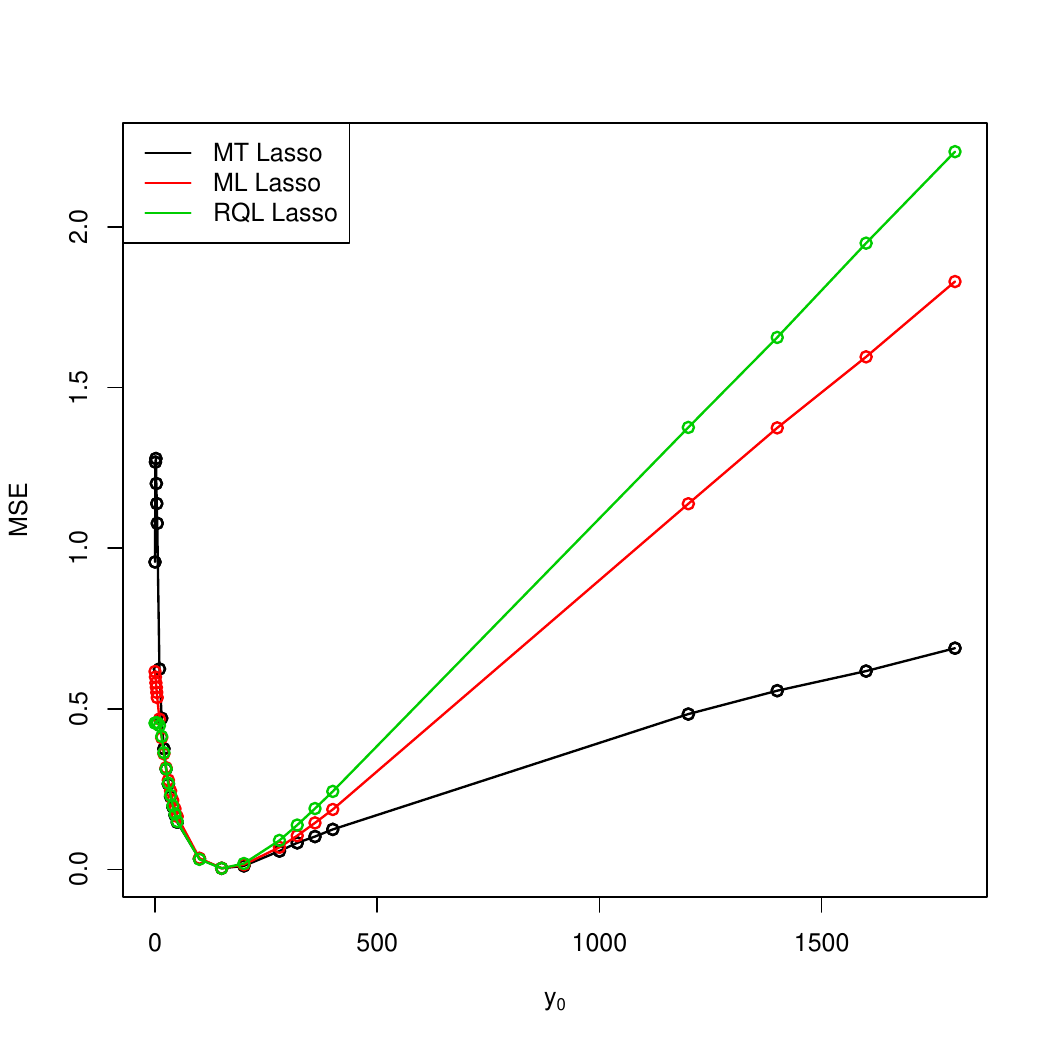}
			&   \includegraphics[scale=0.29]{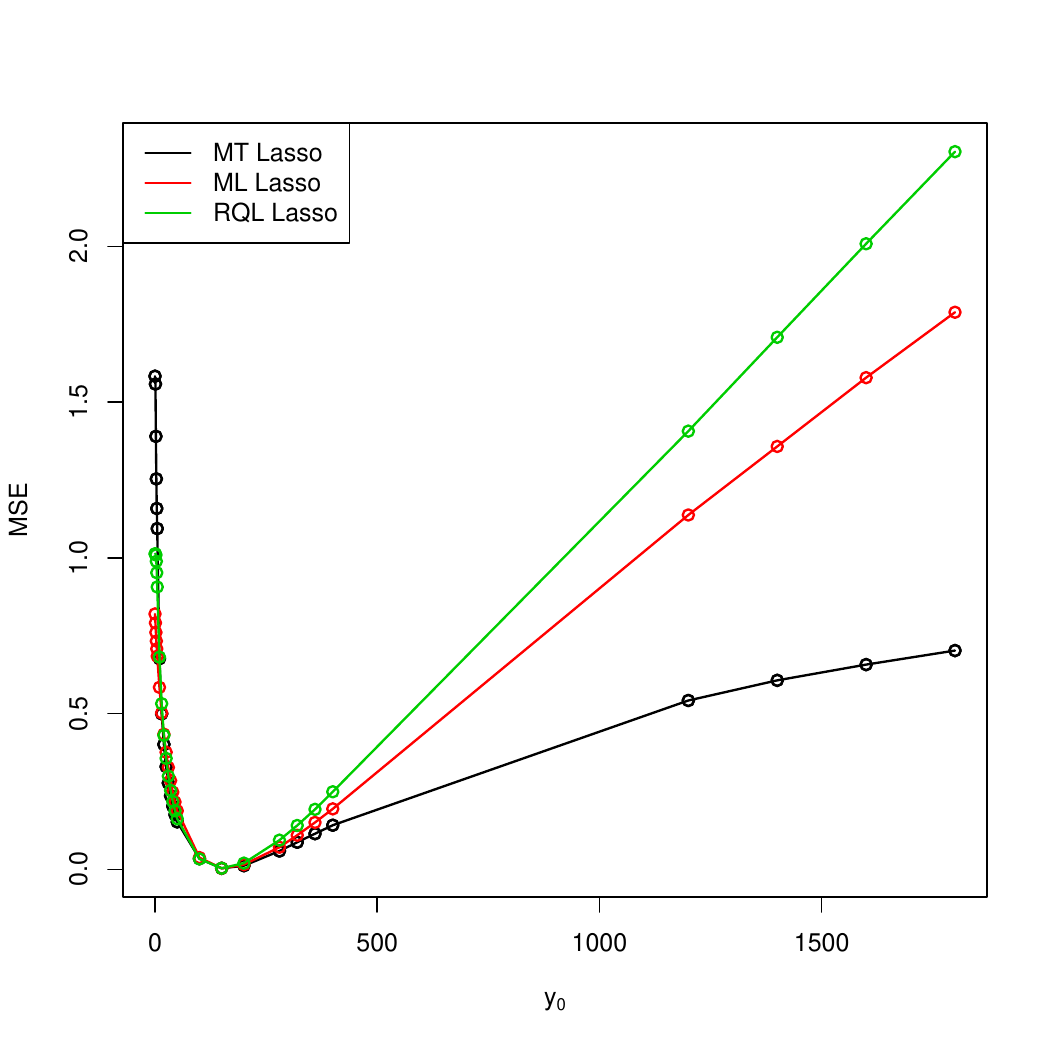}
			\\
			$n=400,\, \epsilon=0.05$ & $n=400,\, \epsilon=0.1$  & $n=400,\, \epsilon=0.15$ \\
			\includegraphics[scale=0.3]{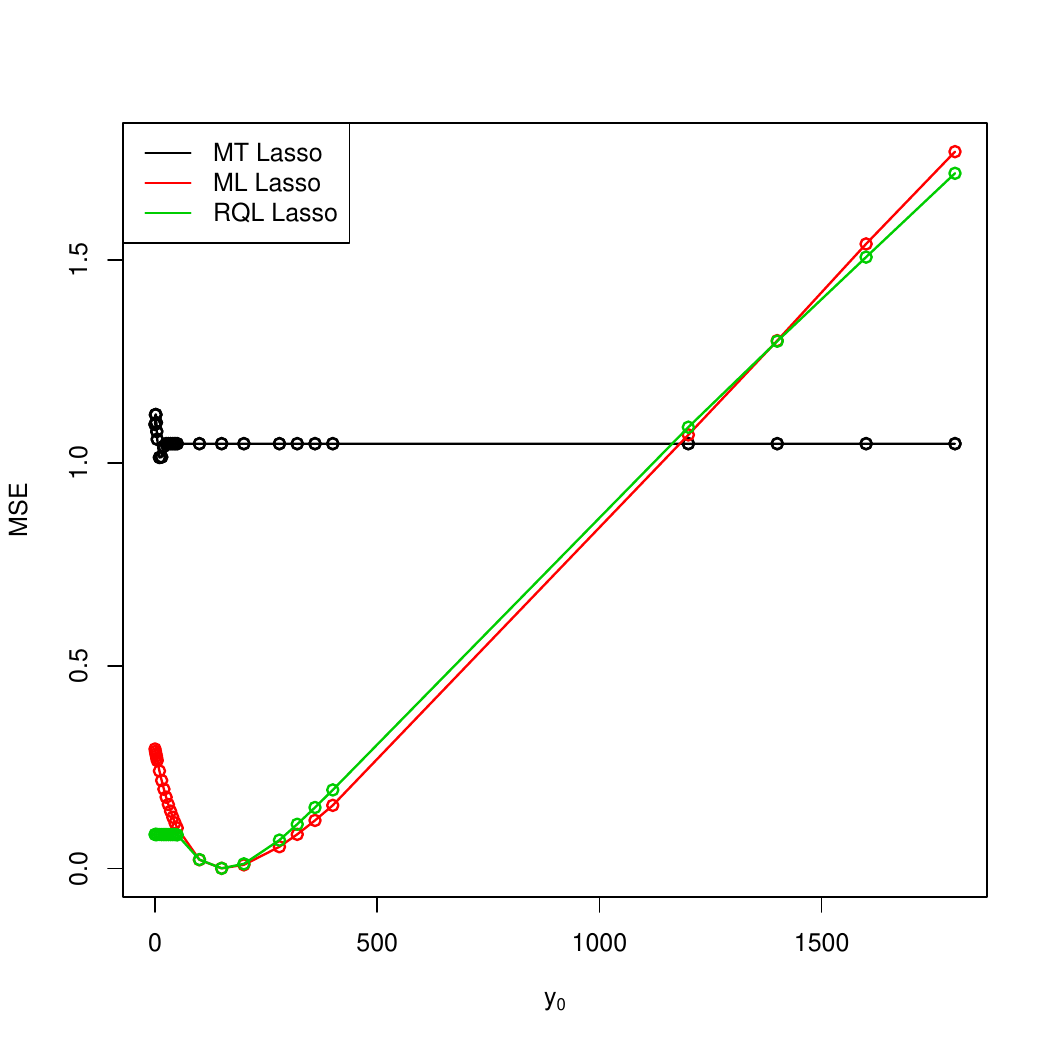}
			&   \includegraphics[scale=0.3]{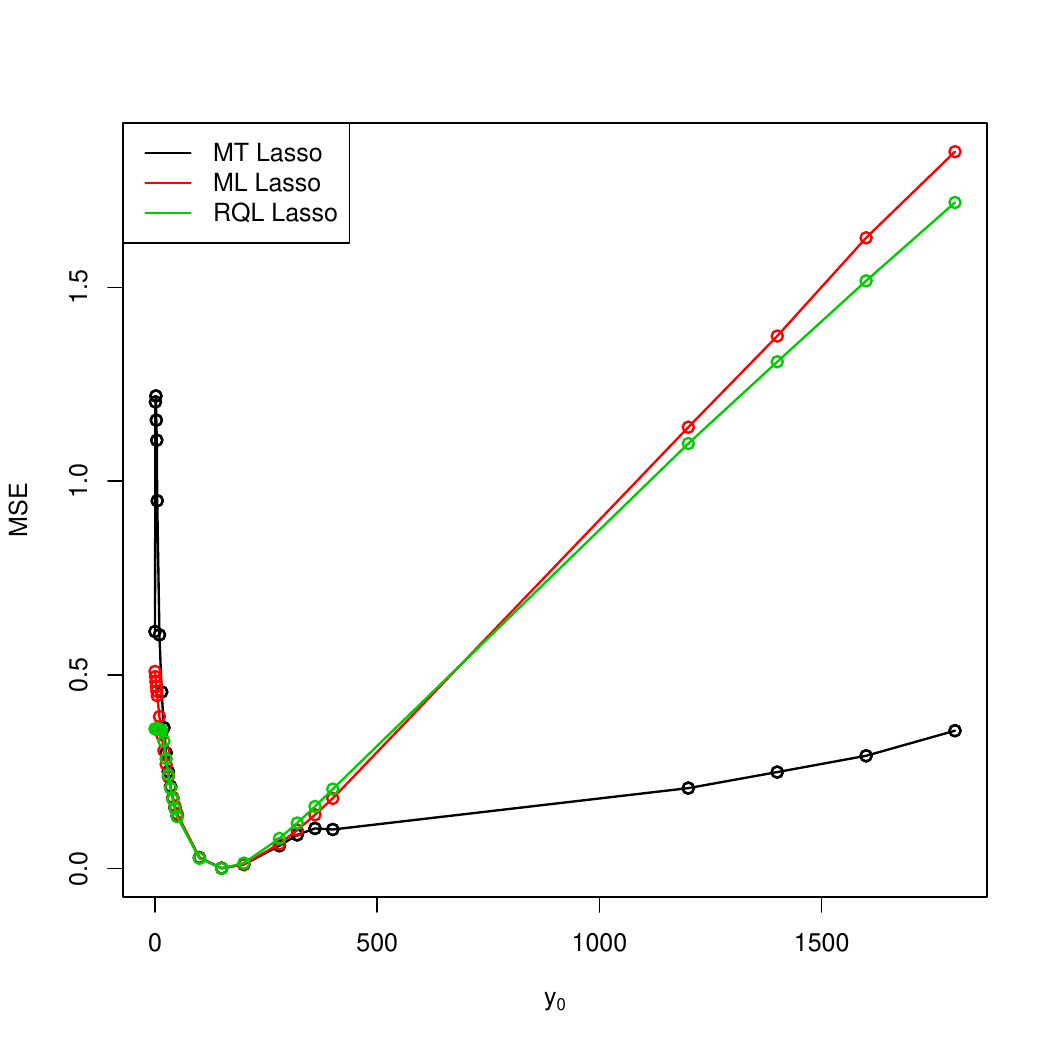} 
			&   \includegraphics[scale=0.3]{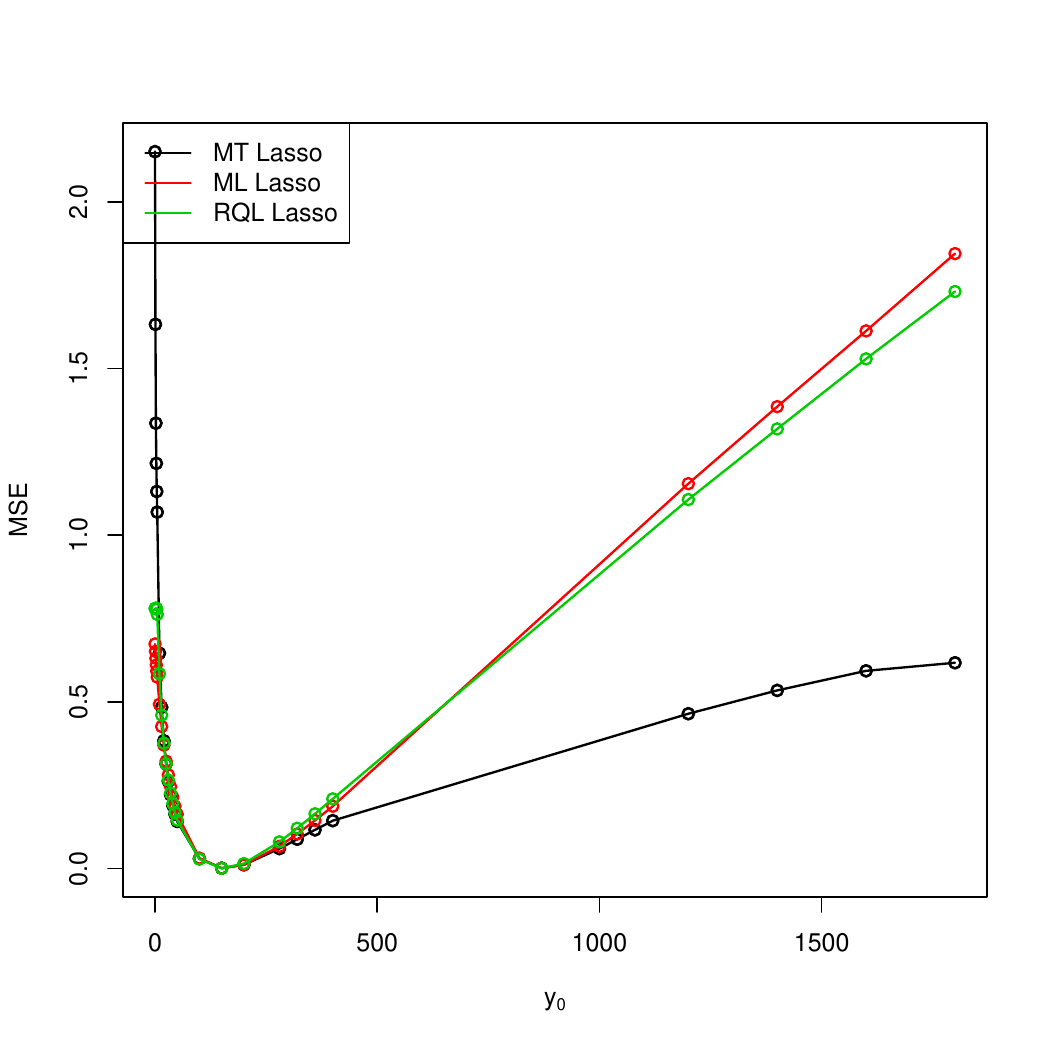} 
			\\
			$n=1000,\, \epsilon=0.05$ & $n=1000,\, \epsilon=0.1$  & $n=1000,\, \epsilon=0.15$ \\
			\includegraphics[scale=0.3]{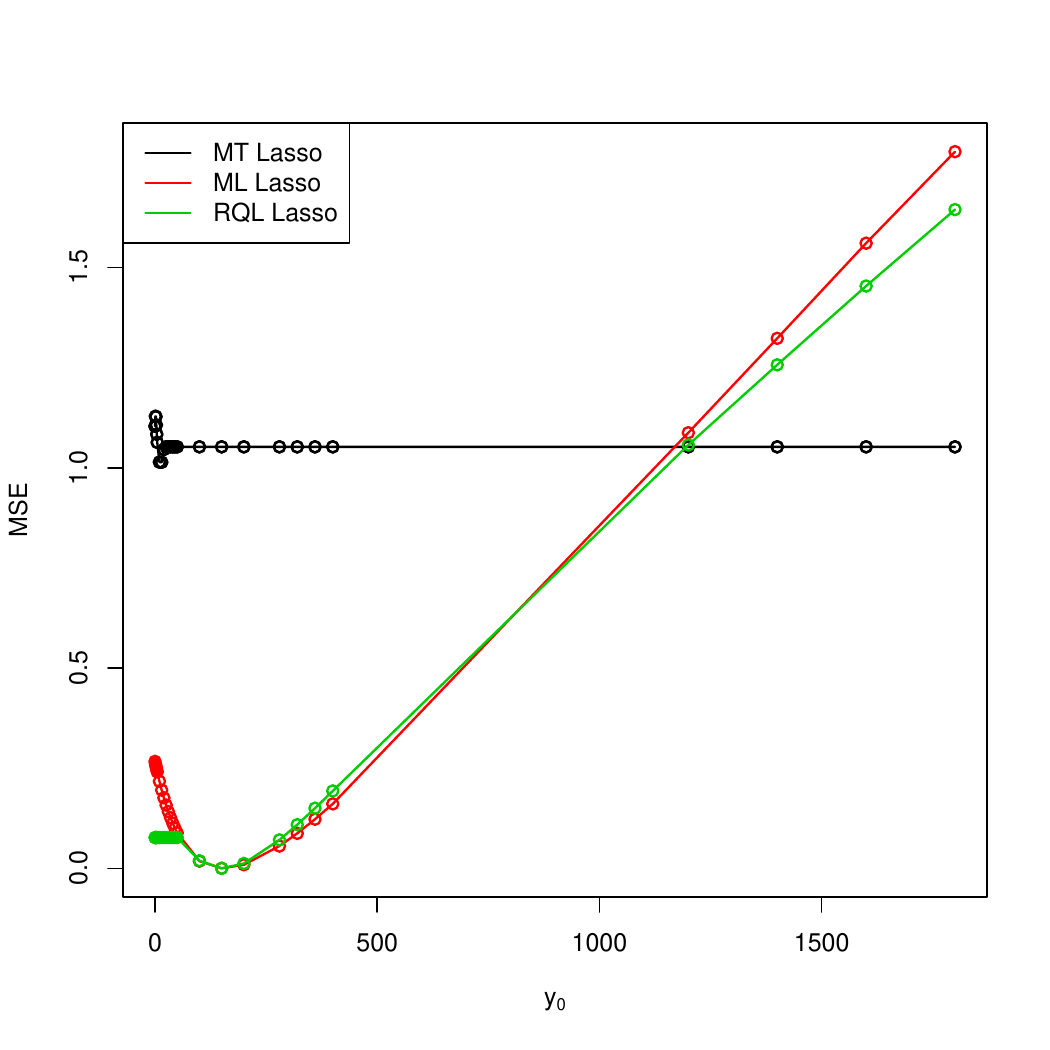}
			&   \includegraphics[scale=0.3]{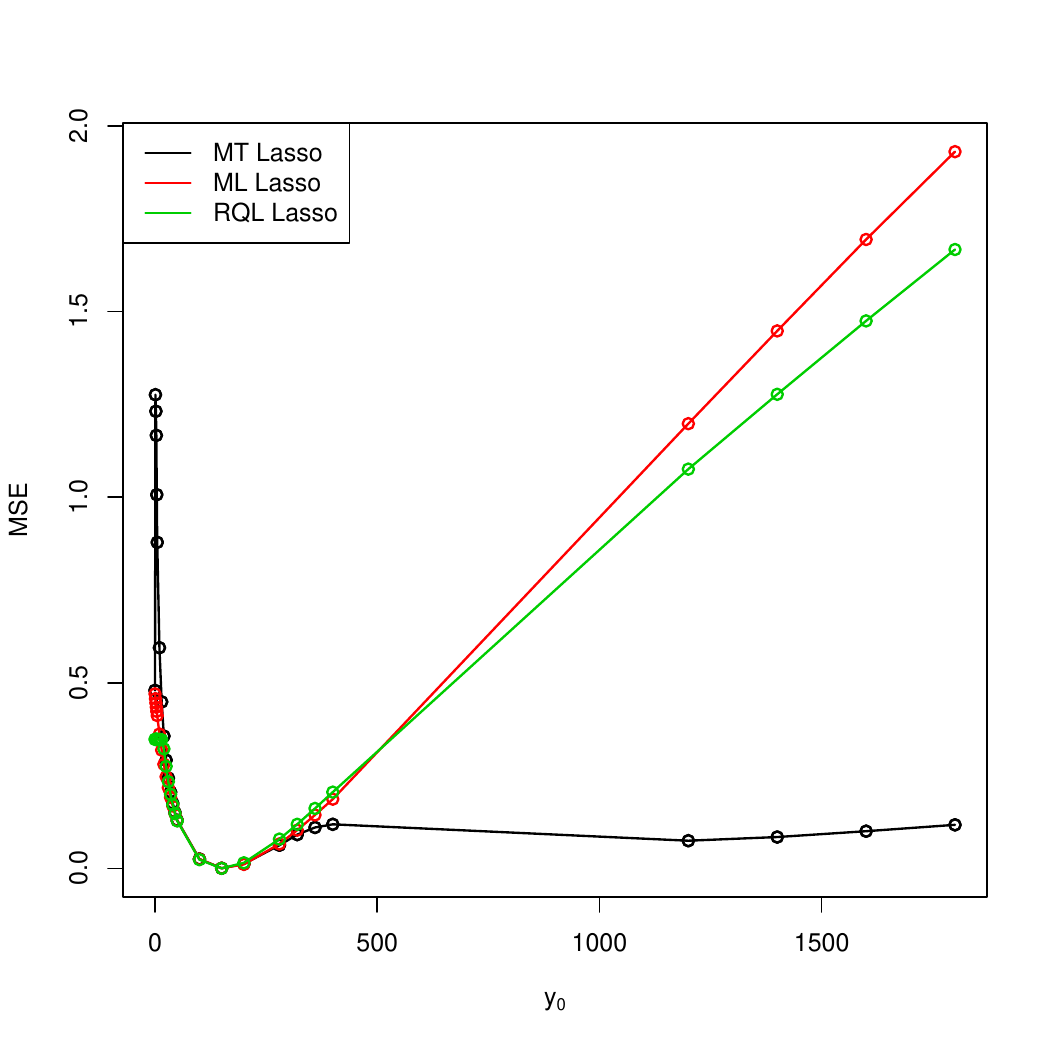} 
			&   \includegraphics[scale=0.3]{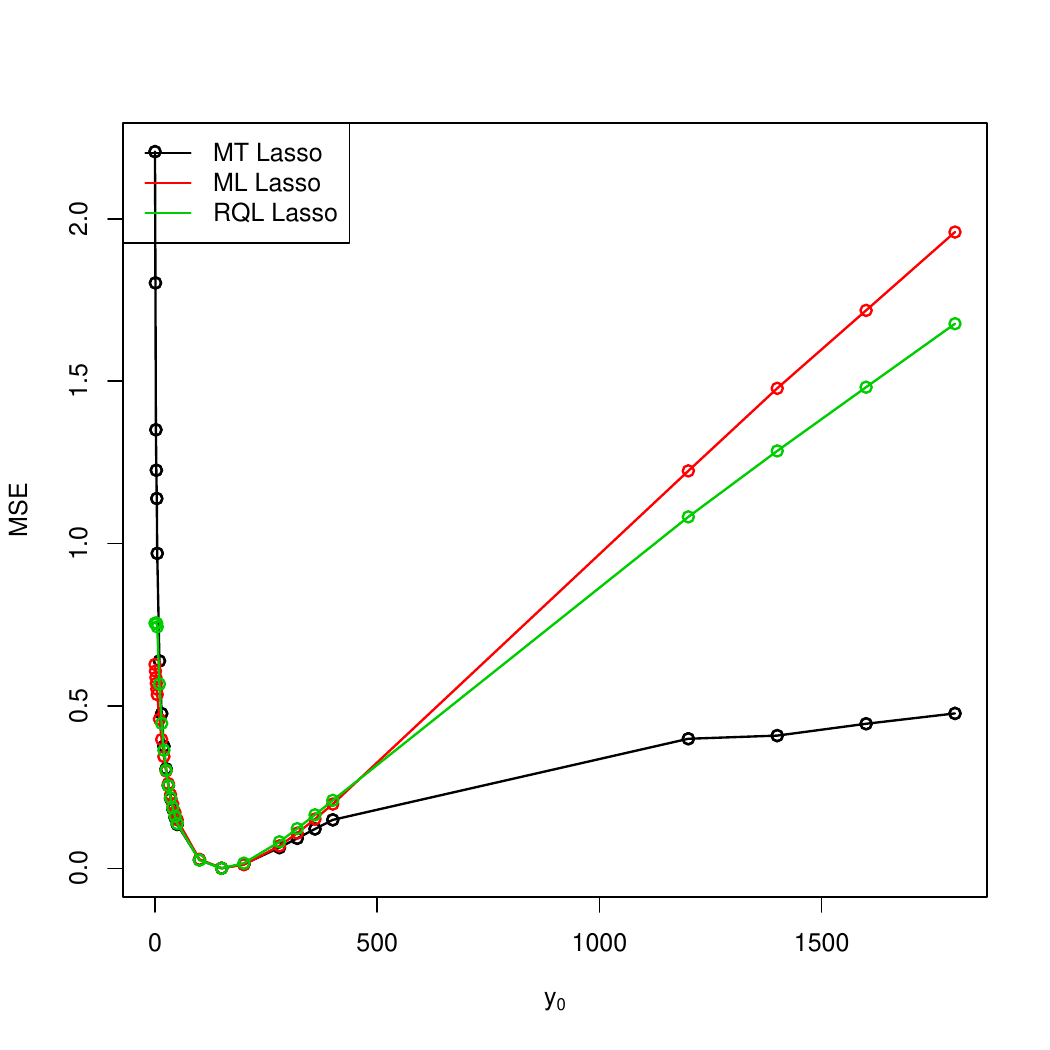}
		\end{tabular}
		\caption{\small \label{fig:msesimu2p10Lasso}  Mean squared errors for model AVY2 and Lasso methods, $p=10$.}
	\end{center} 
\end{figure}

\begin{figure}[tp]
	\begin{center}
		\footnotesize
		\renewcommand{\arraystretch}{0.22}
		\newcolumntype{M}{>{\centering\arraybackslash}m{\dimexpr.01\linewidth-1\tabcolsep}}
		\newcolumntype{G}{>{\centering\arraybackslash}m{\dimexpr.38\linewidth-1\tabcolsep}}
		\begin{tabular}{GGG}
			$n=100,\, \epsilon=0.05$ & $n=100,\, \epsilon=0.1$  & $n=100,\, \epsilon=0.15$ \\
			\includegraphics[scale=0.29]{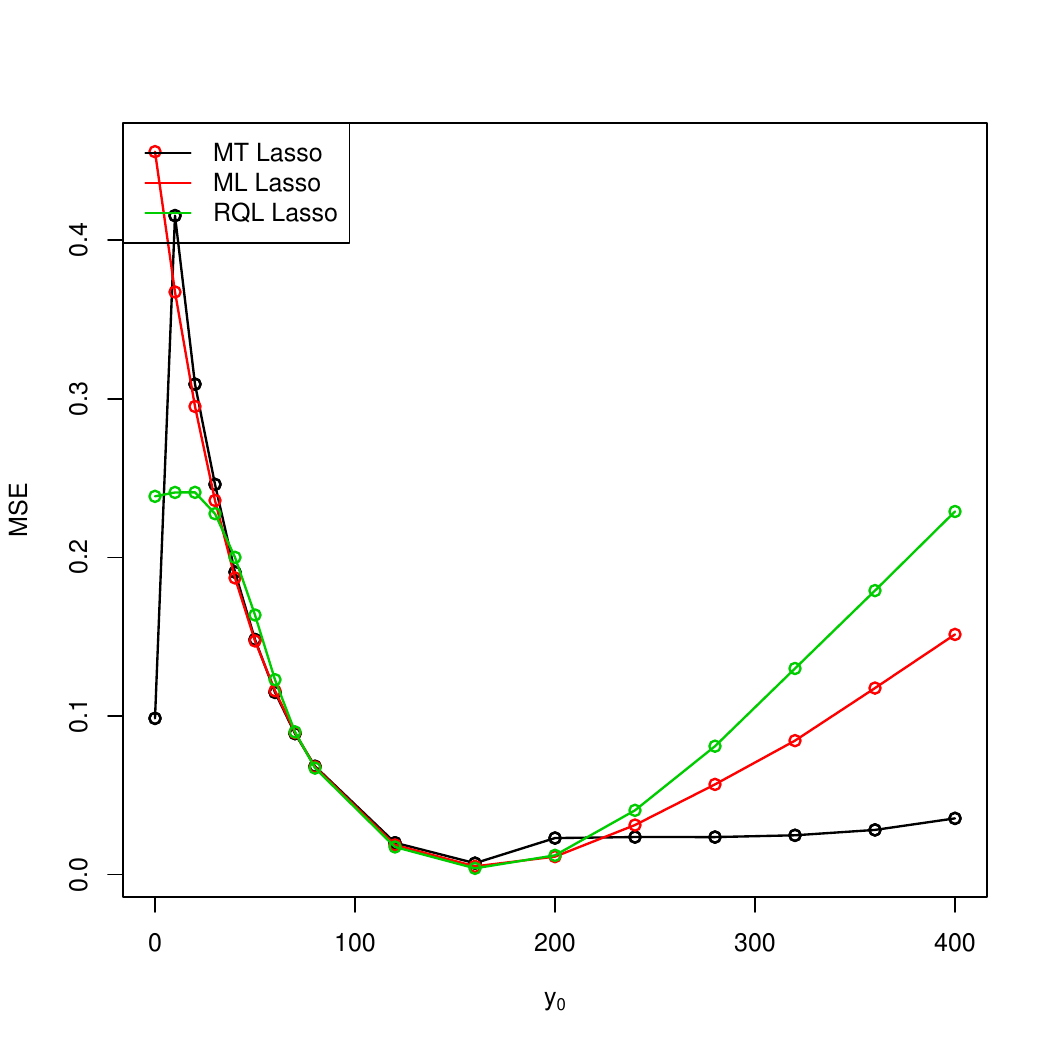} 
			&  \includegraphics[scale=0.29]{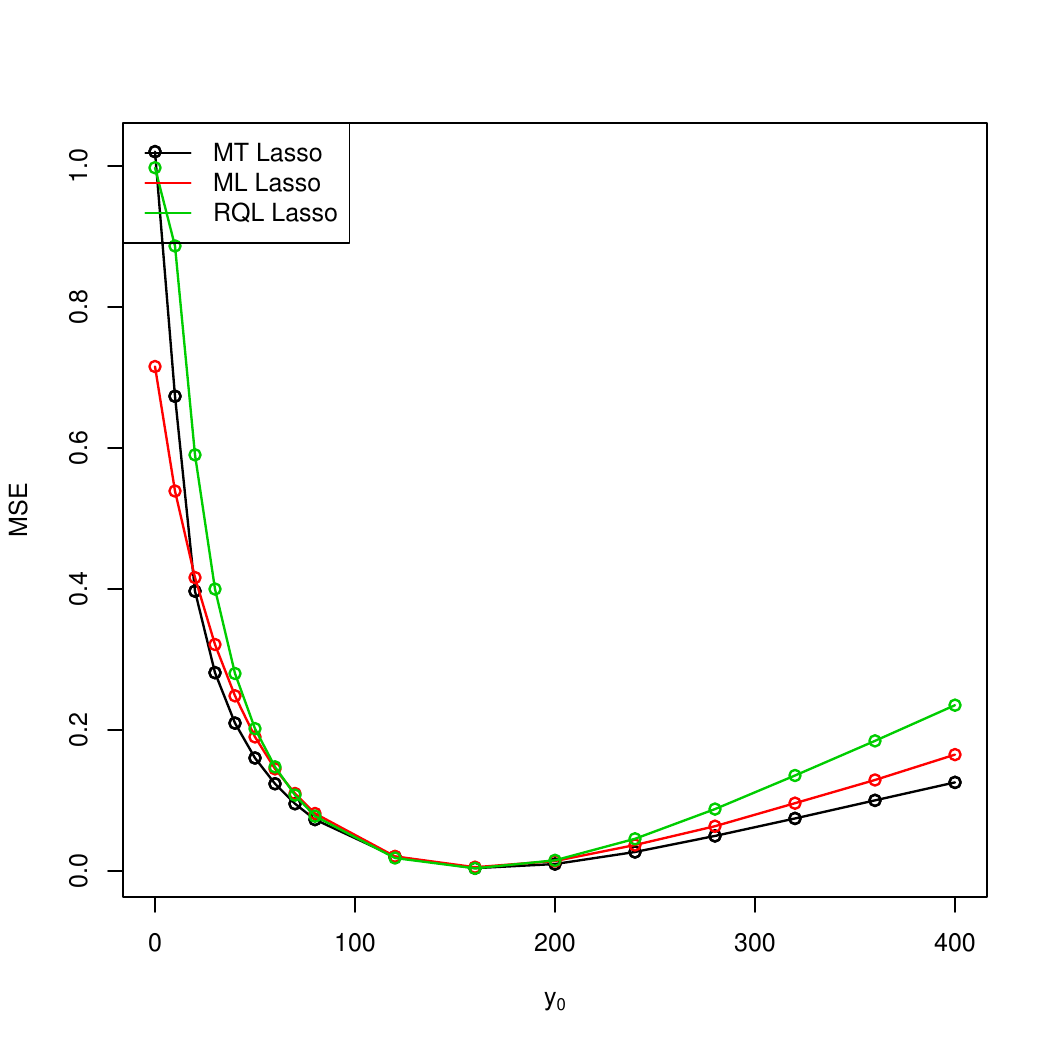}
			&   \includegraphics[scale=0.29]{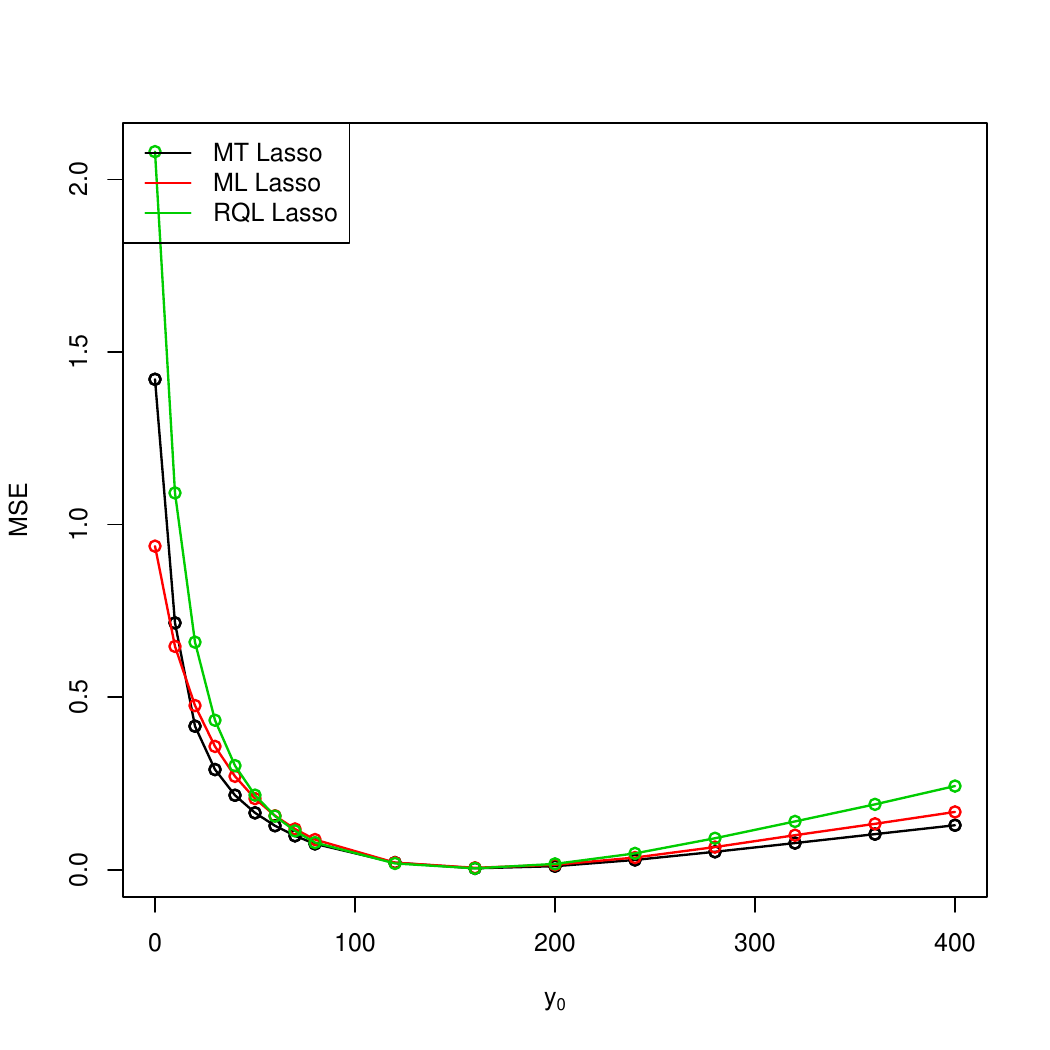}
			\\
			$n=400,\, \epsilon=0.05$ & $n=400,\, \epsilon=0.1$  & $n=400,\, \epsilon=0.15$ \\
			\includegraphics[scale=0.3]{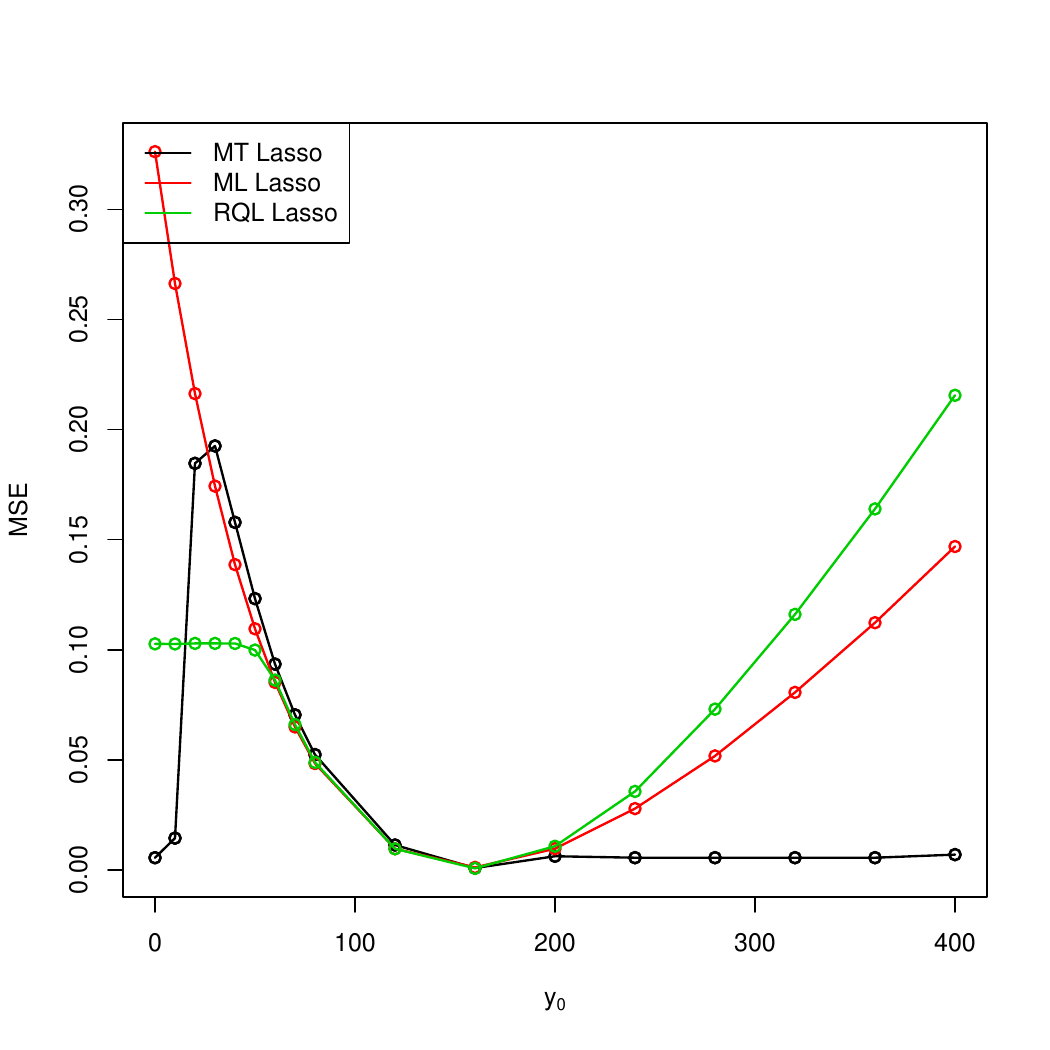}
			&   \includegraphics[scale=0.3]{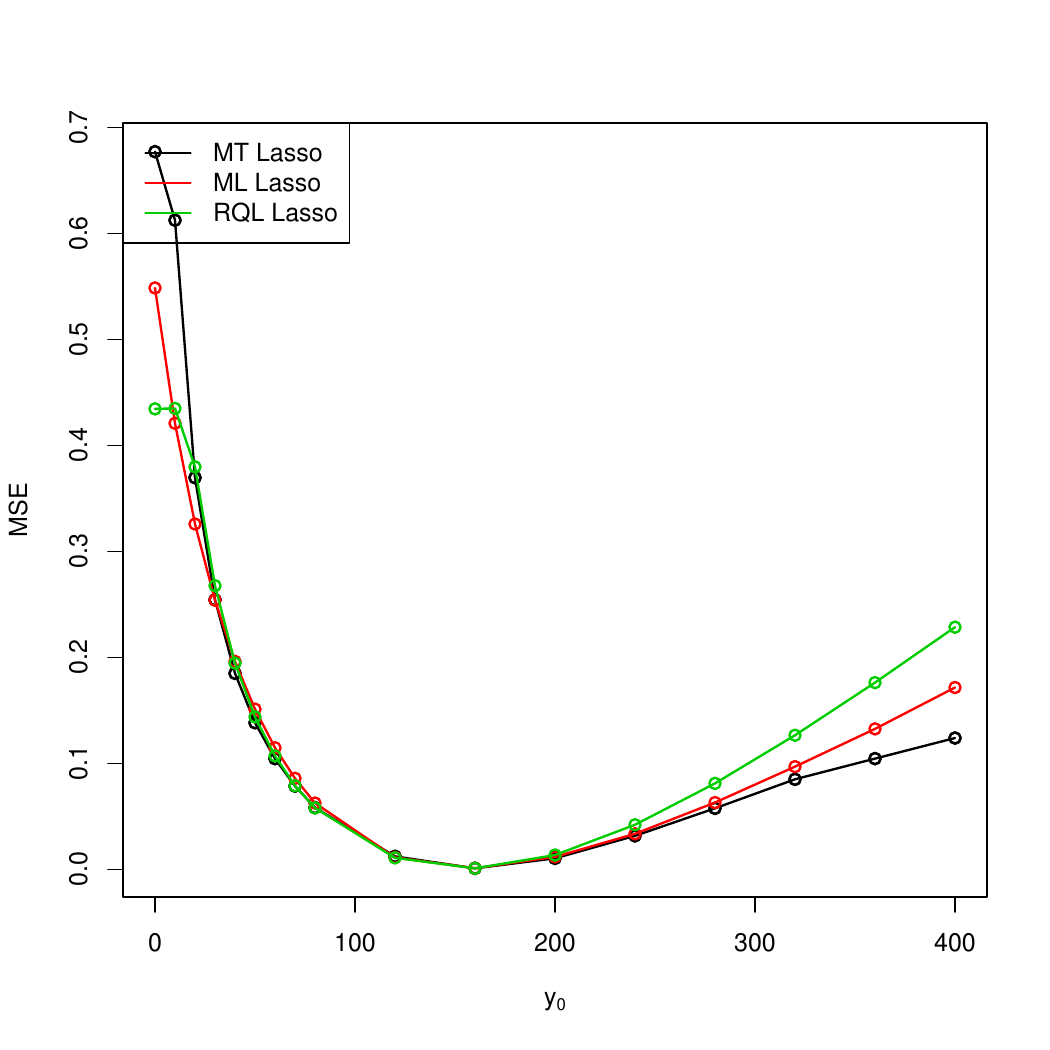} 
			&   \includegraphics[scale=0.3]{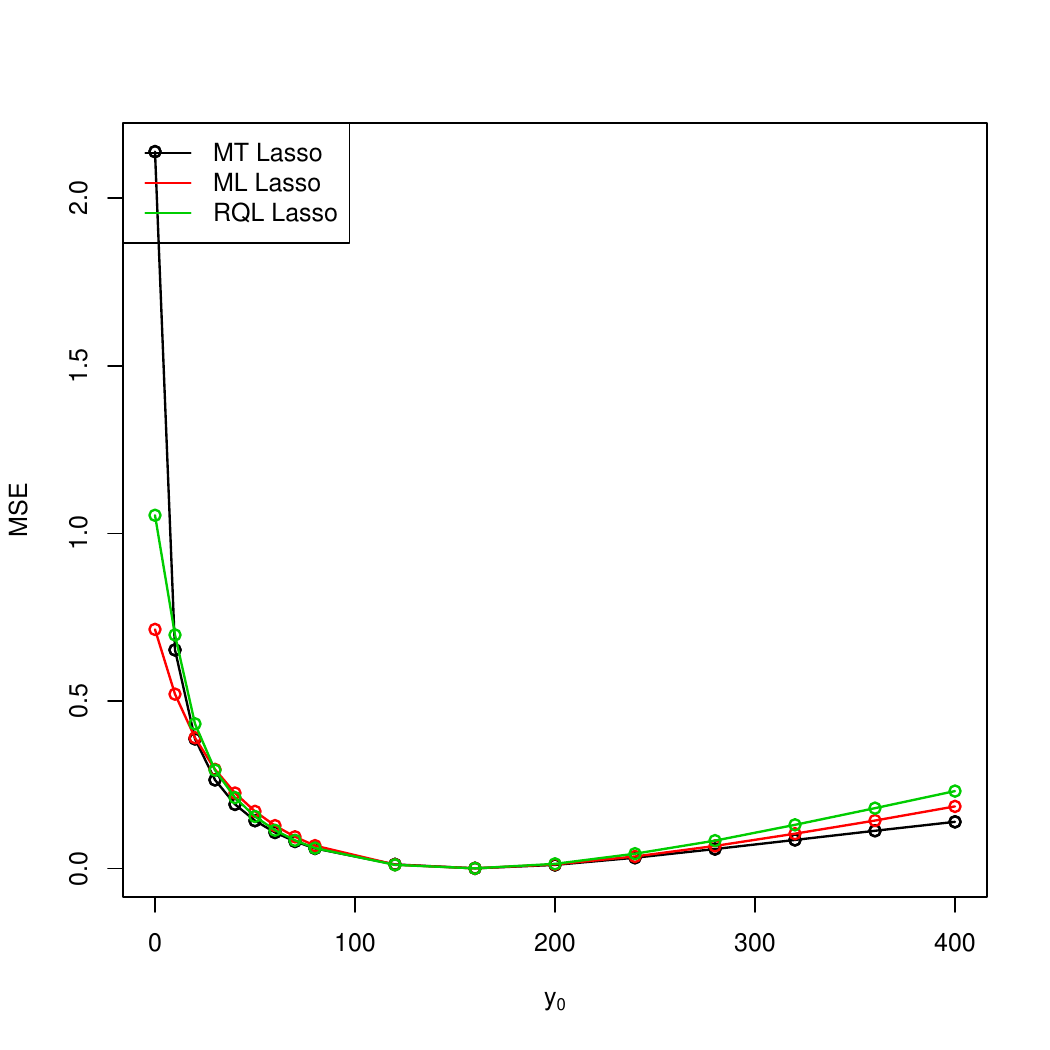} 
			\\
			$n=1000,\, \epsilon=0.05$ & $n=1000,\, \epsilon=0.1$  & $n=1000,\, \epsilon=0.15$ \\
			\includegraphics[scale=0.3]{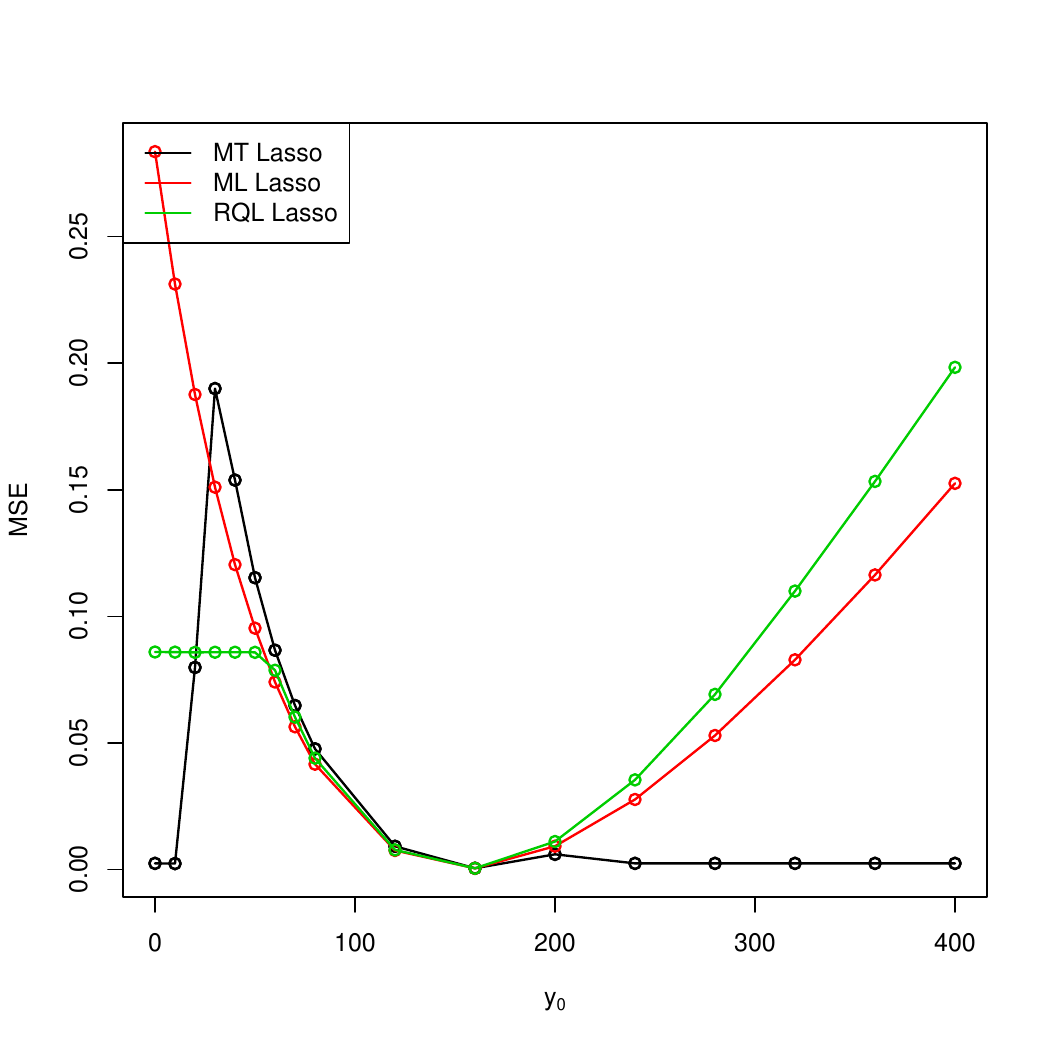}
			&   \includegraphics[scale=0.3]{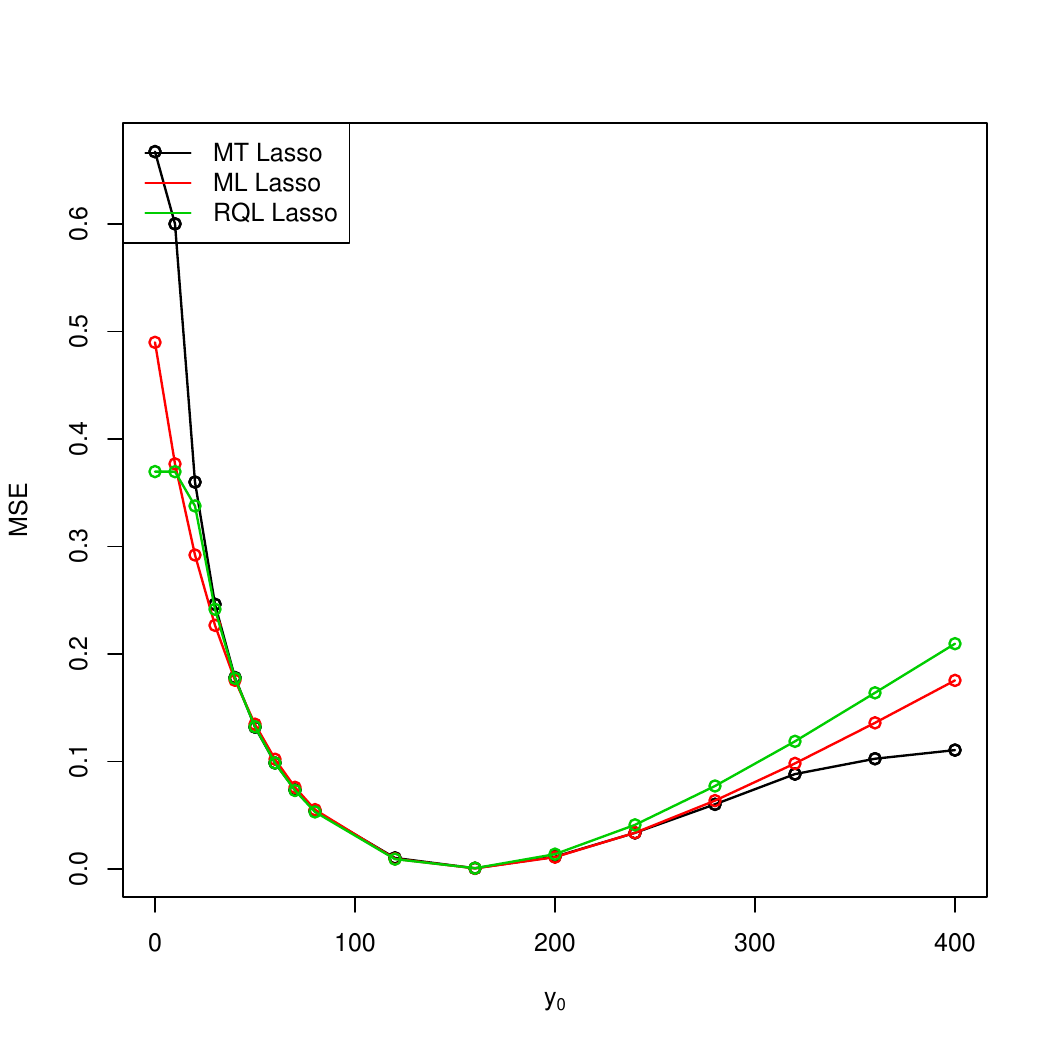} 
			&   \includegraphics[scale=0.3]{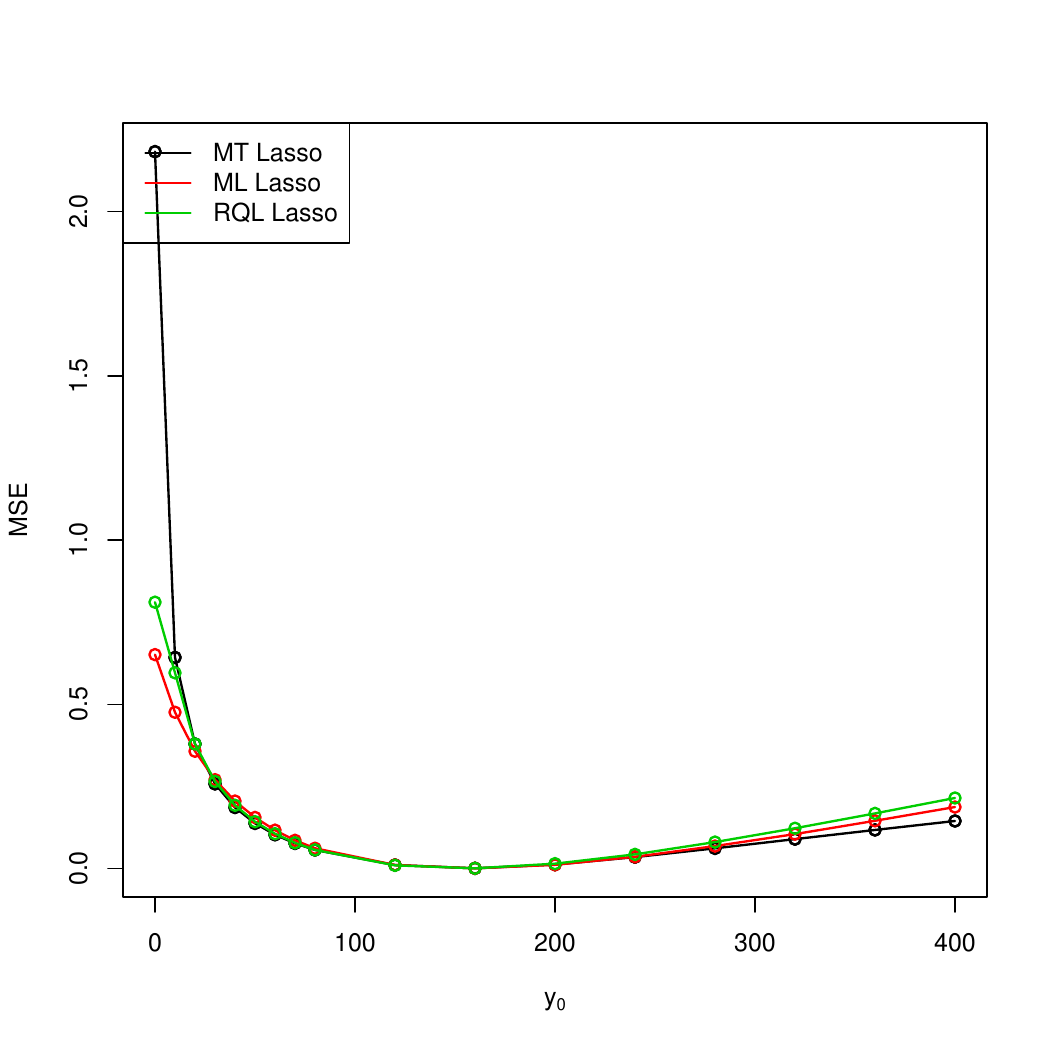}
		\end{tabular}
		\caption{\small \label{fig:msesimu2p50Lasso}  Mean squared errors for model AVY2 and Lasso methods, $p=50$.}
	\end{center} 
\end{figure}

\begin{figure}[tp]
	\begin{center}
		\footnotesize
		\renewcommand{\arraystretch}{0.22}
		\newcolumntype{M}{>{\centering\arraybackslash}m{\dimexpr.01\linewidth-1\tabcolsep}}
		\newcolumntype{G}{>{\centering\arraybackslash}m{\dimexpr.38\linewidth-1\tabcolsep}}
		\begin{tabular}{GGG}
			$n=50,\, \epsilon=0.05$ & $n=100,\, \epsilon=0.1$  & $n=50,\, \epsilon=0.15$ \\
			\includegraphics[scale=0.29]{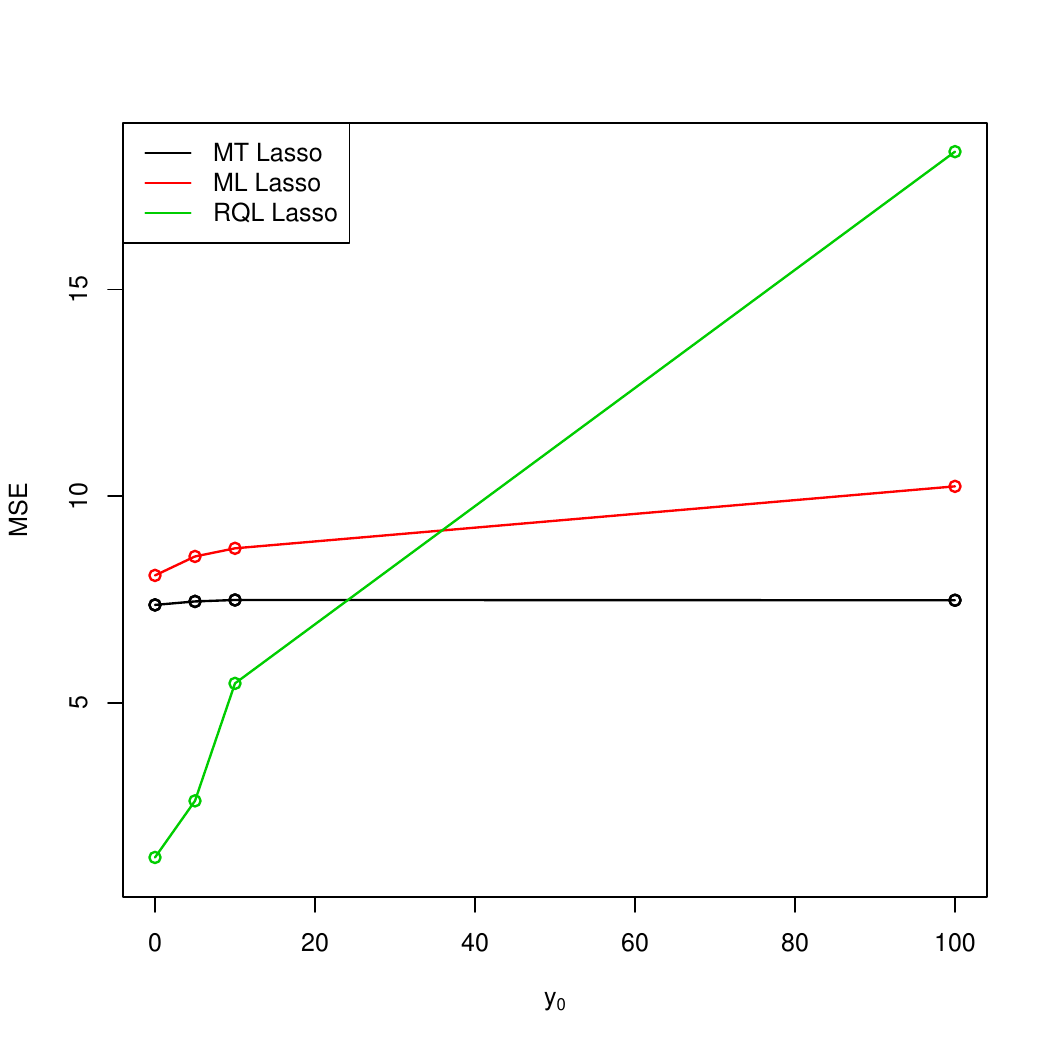} 
			&  \includegraphics[scale=0.29]{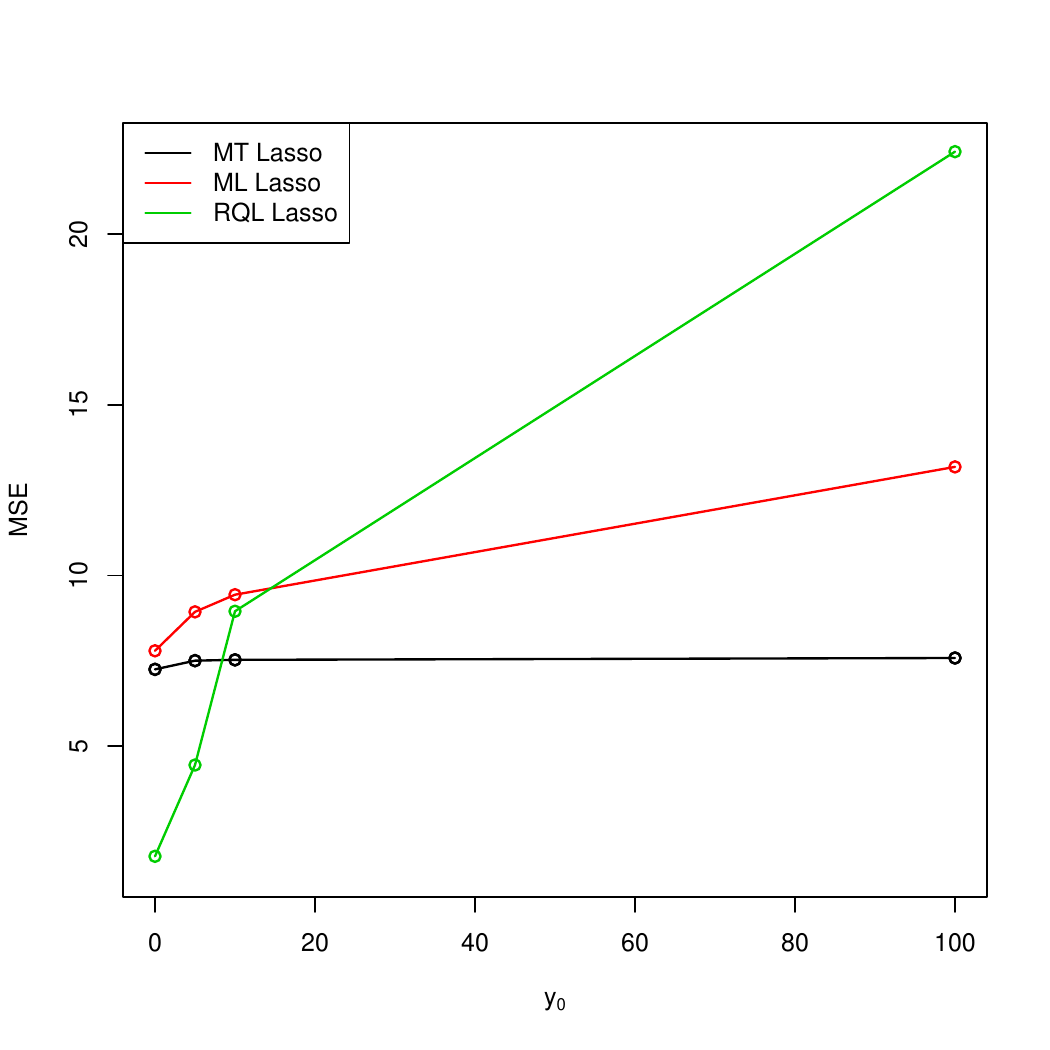}
			&   \includegraphics[scale=0.29]{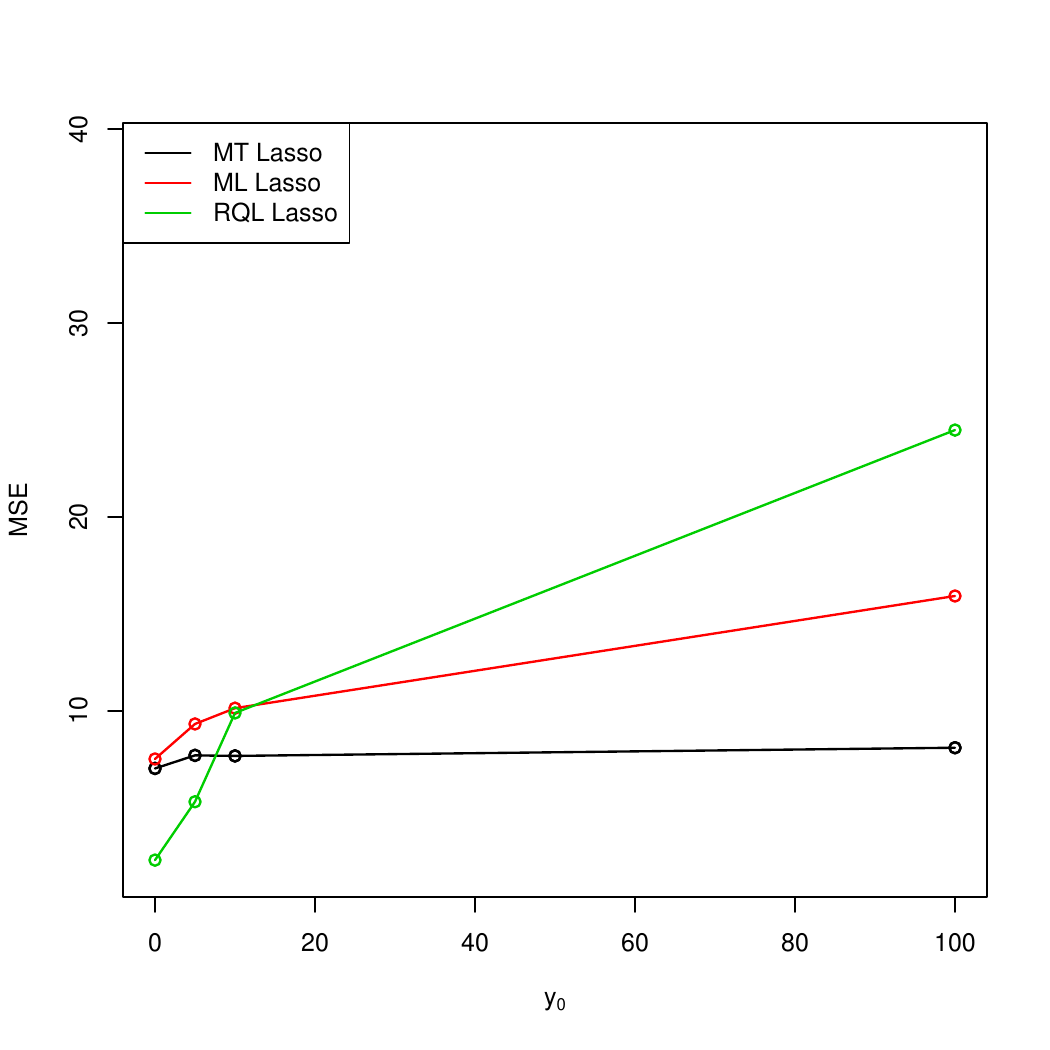}
			\\
			$n=100,\, \epsilon=0.05$ & $n=100,\, \epsilon=0.1$  & $n=100,\, \epsilon=0.15$ \\
			\includegraphics[scale=0.3]{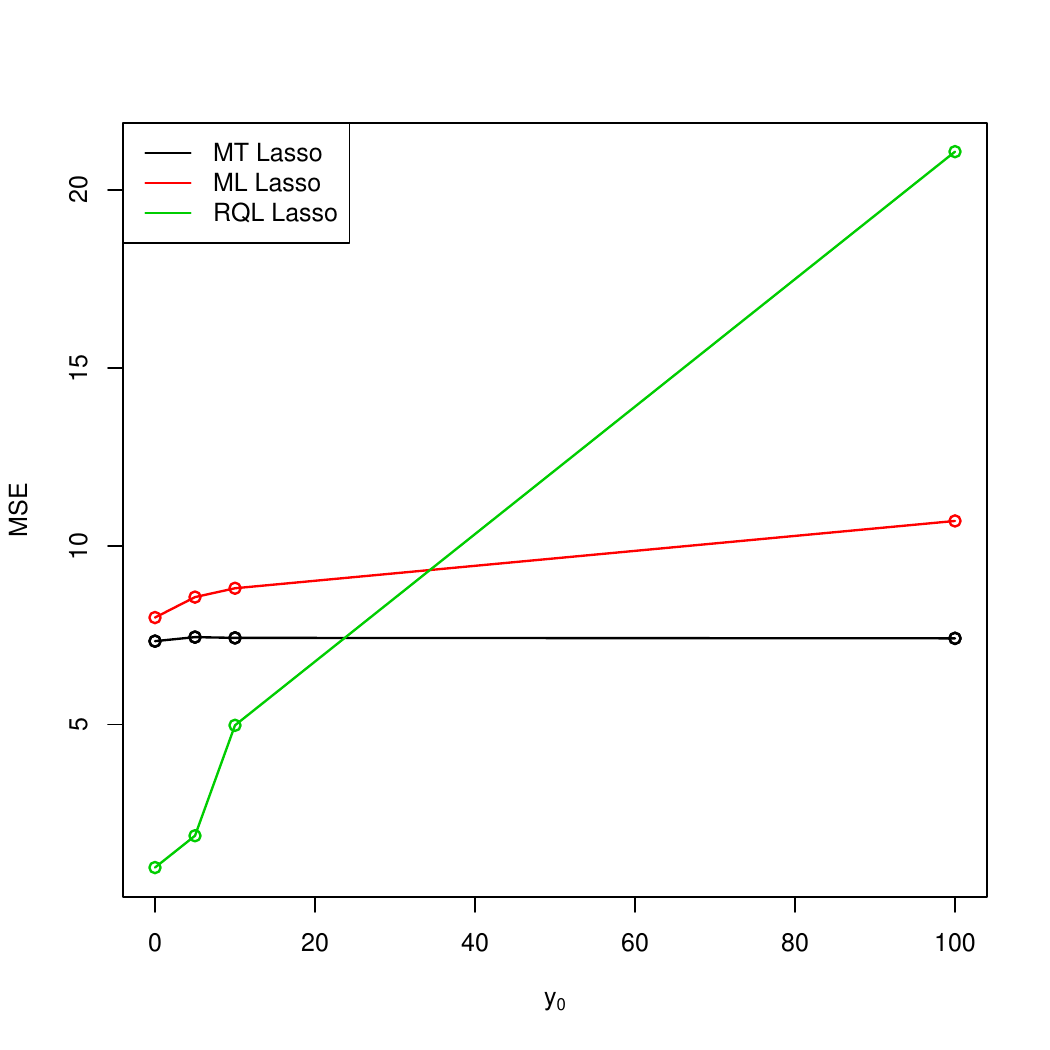}
			&   \includegraphics[scale=0.3]{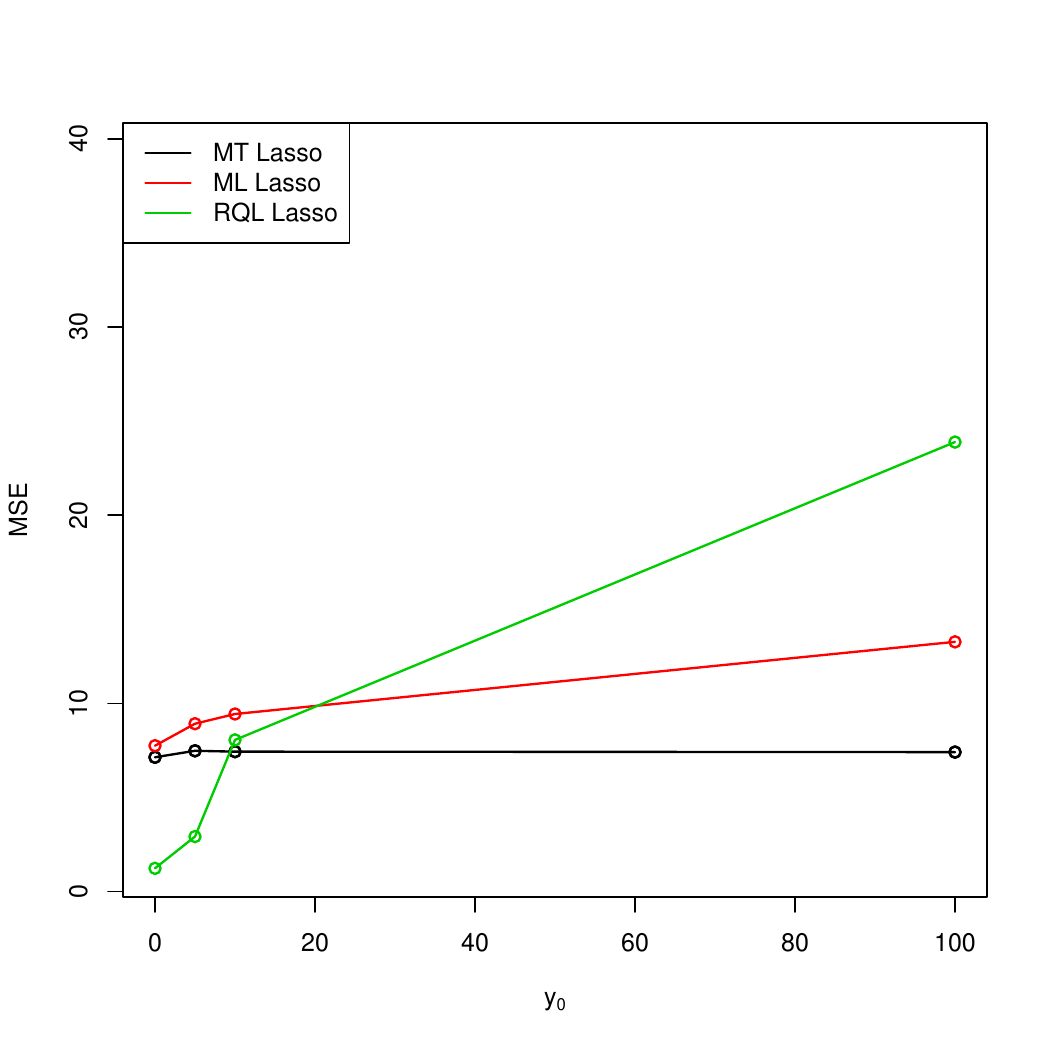} 
			&   \includegraphics[scale=0.3]{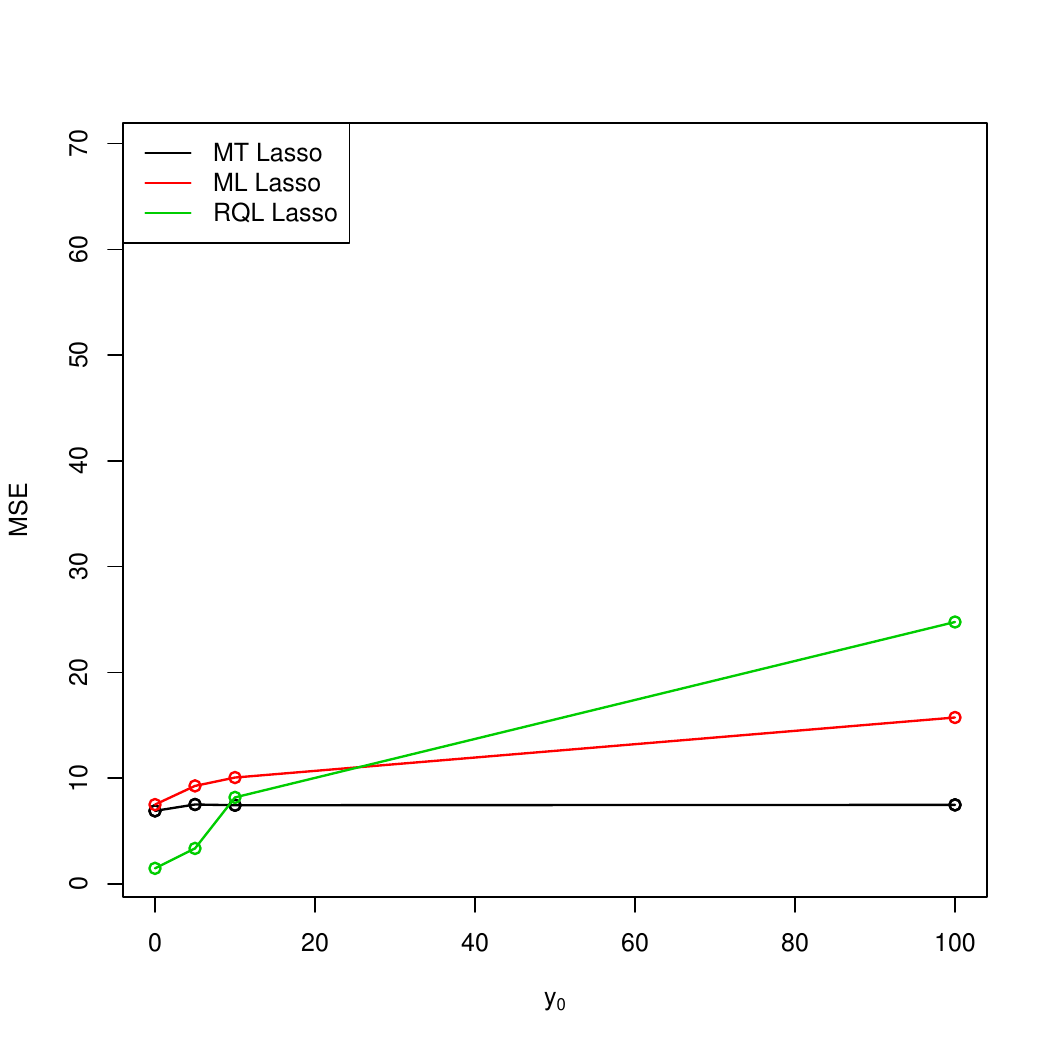} 
			\\
			$n=200,\, \epsilon=0.05$ & $n=200,\, \epsilon=0.1$  & $n=200,\, \epsilon=0.15$ \\
			\includegraphics[scale=0.3]{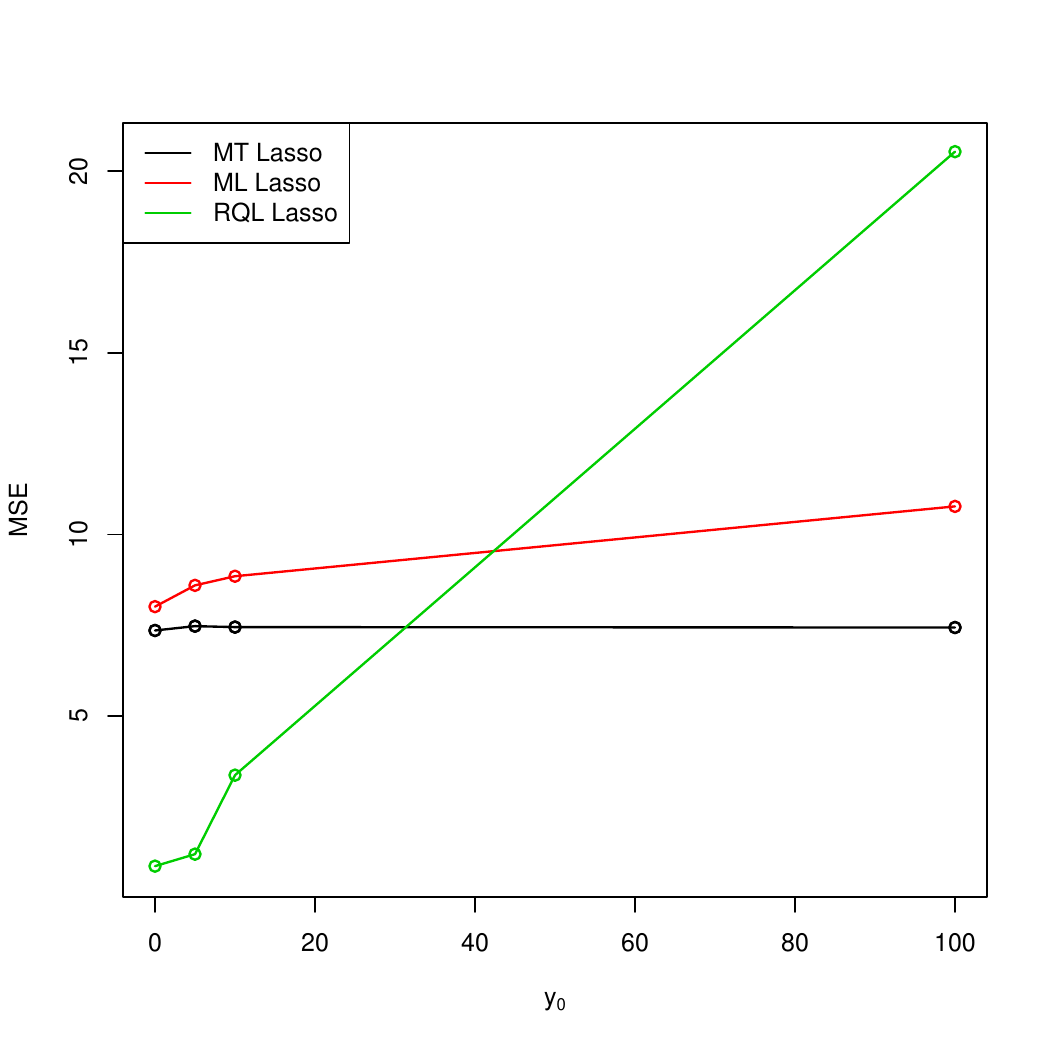}
			&   \includegraphics[scale=0.3]{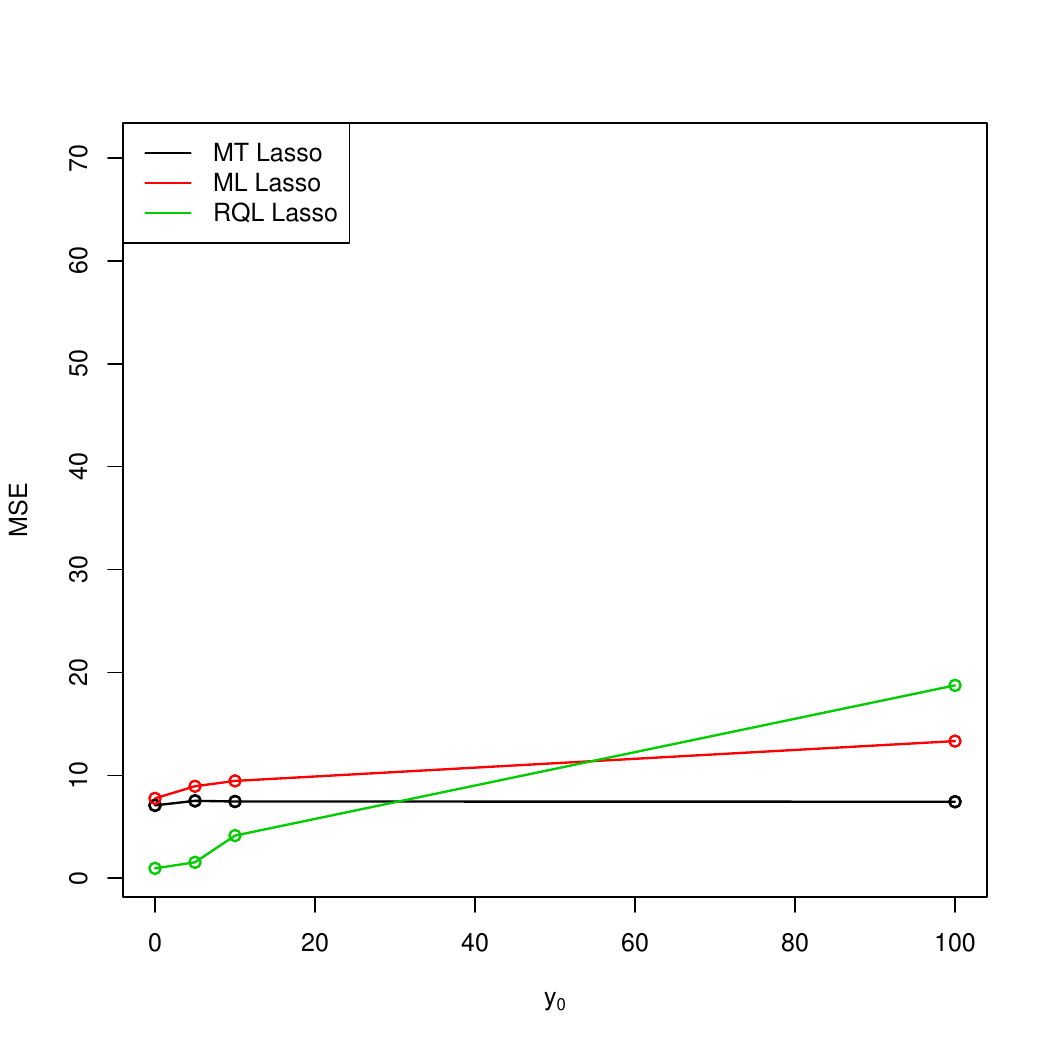} 
			&   \includegraphics[scale=0.3]{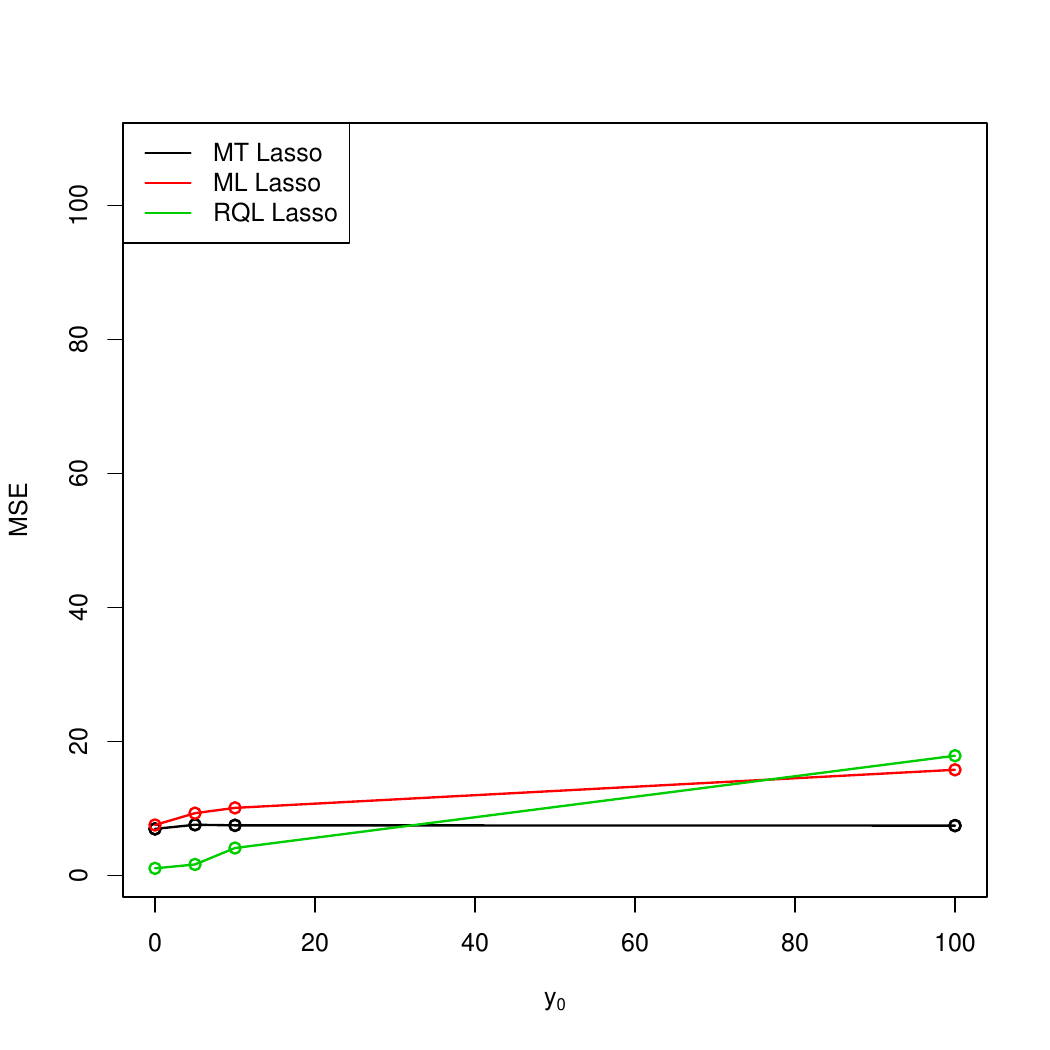}
		\end{tabular}
		\caption{\small \label{fig:msesimuRQLp100Lasso}  Mean squared errors for model AMR and Lasso methods, $p=100$.}
	\end{center} 
\end{figure}

\begin{figure}[tp]
	\begin{center}
		\footnotesize
		\renewcommand{\arraystretch}{0.22}
		\newcolumntype{M}{>{\centering\arraybackslash}m{\dimexpr.01\linewidth-1\tabcolsep}}
		\newcolumntype{G}{>{\centering\arraybackslash}m{\dimexpr.38\linewidth-1\tabcolsep}}
		\begin{tabular}{GGG}
			$n=50,\, \epsilon=0.05$ & $n=50,\, \epsilon=0.1$  & $n=50,\, \epsilon=0.15$ \\
			\includegraphics[scale=0.29]{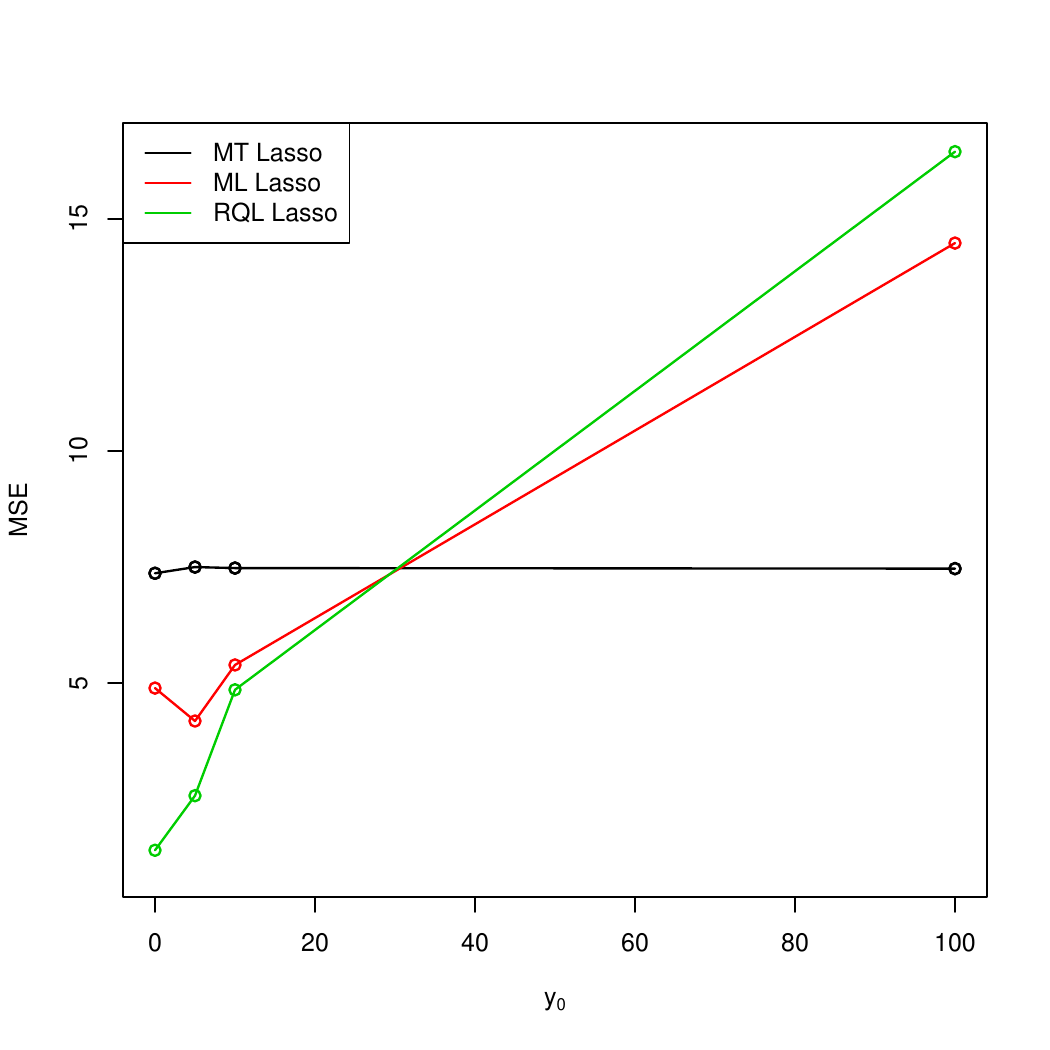} 
			&  \includegraphics[scale=0.29]{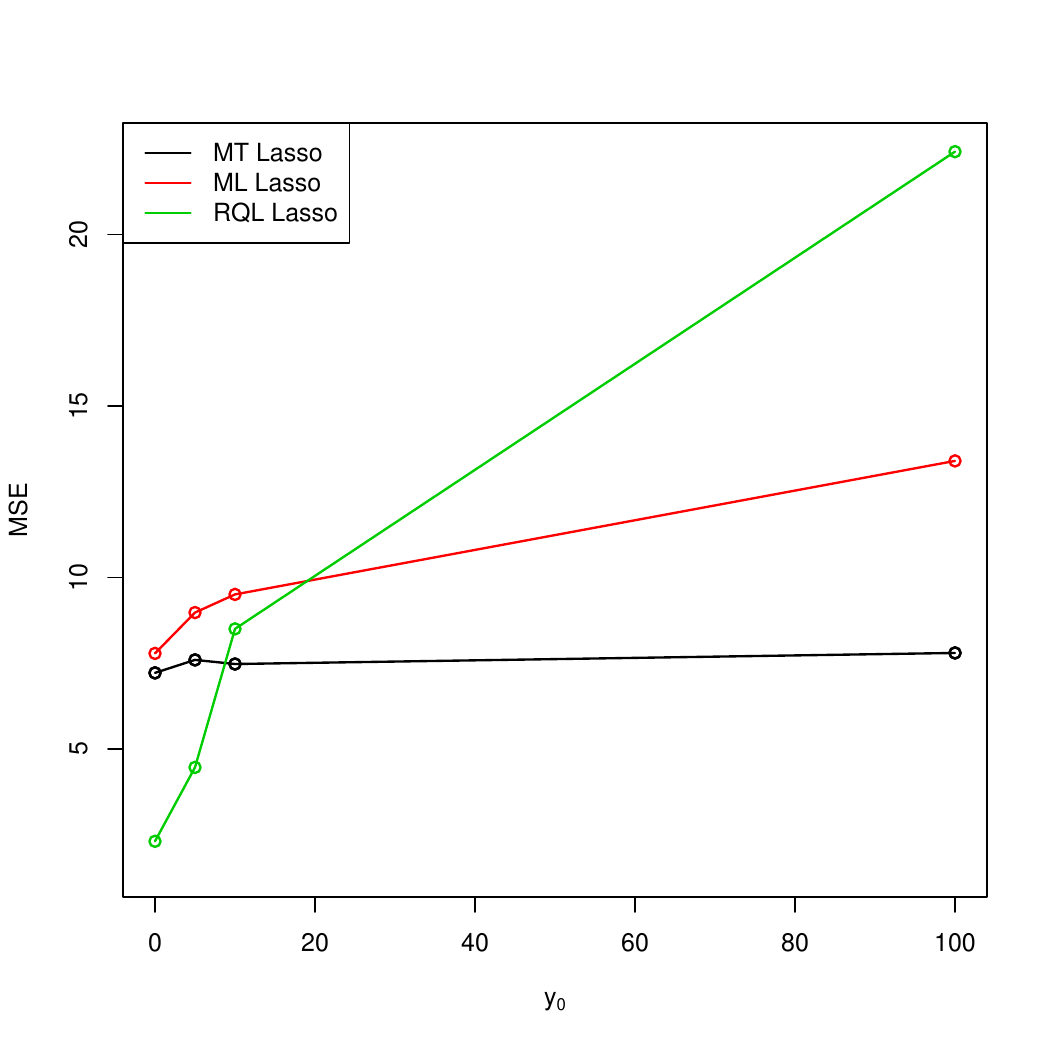}
			&   \includegraphics[scale=0.29]{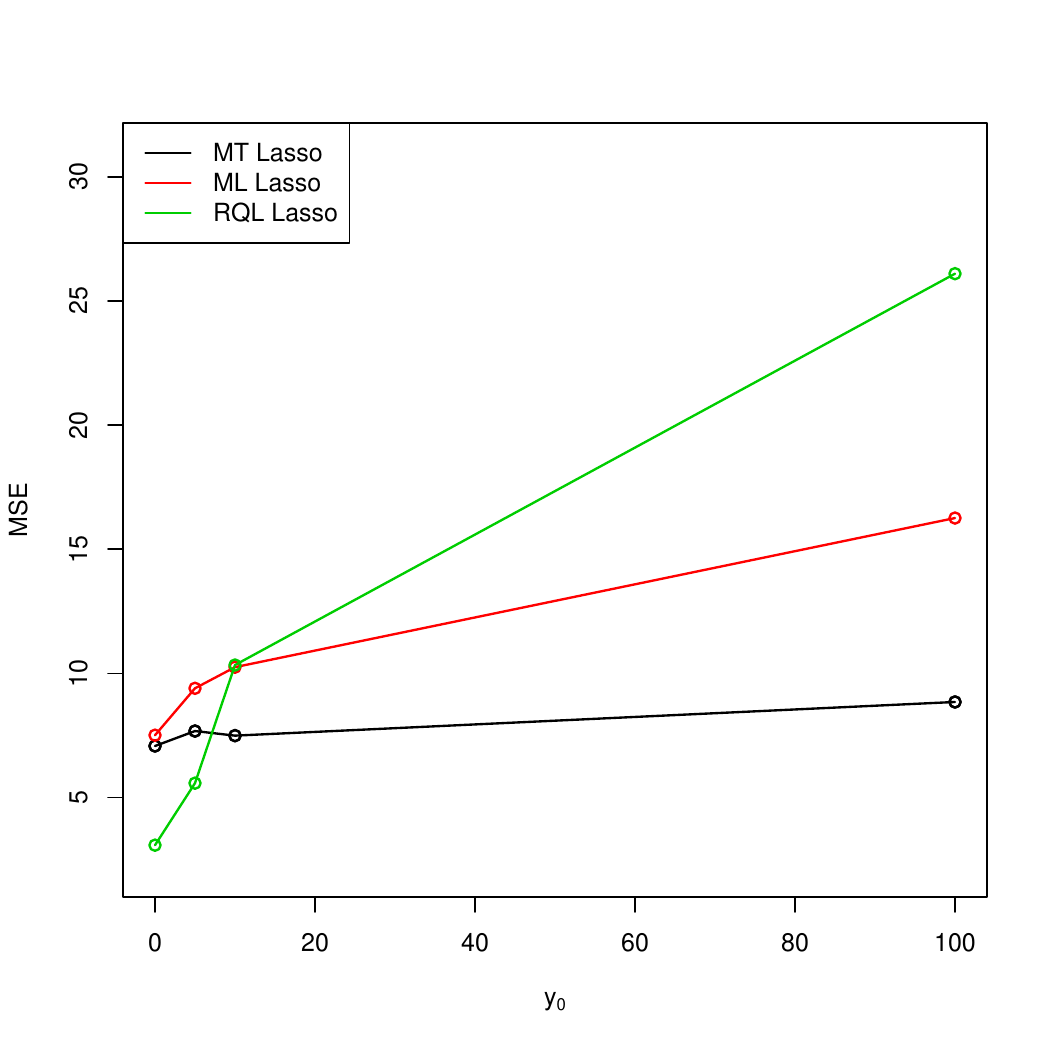}
			\\
			$n=100,\, \epsilon=0.05$ & $n=100,\, \epsilon=0.1$  & $n=100,\, \epsilon=0.15$ \\
			\includegraphics[scale=0.3]{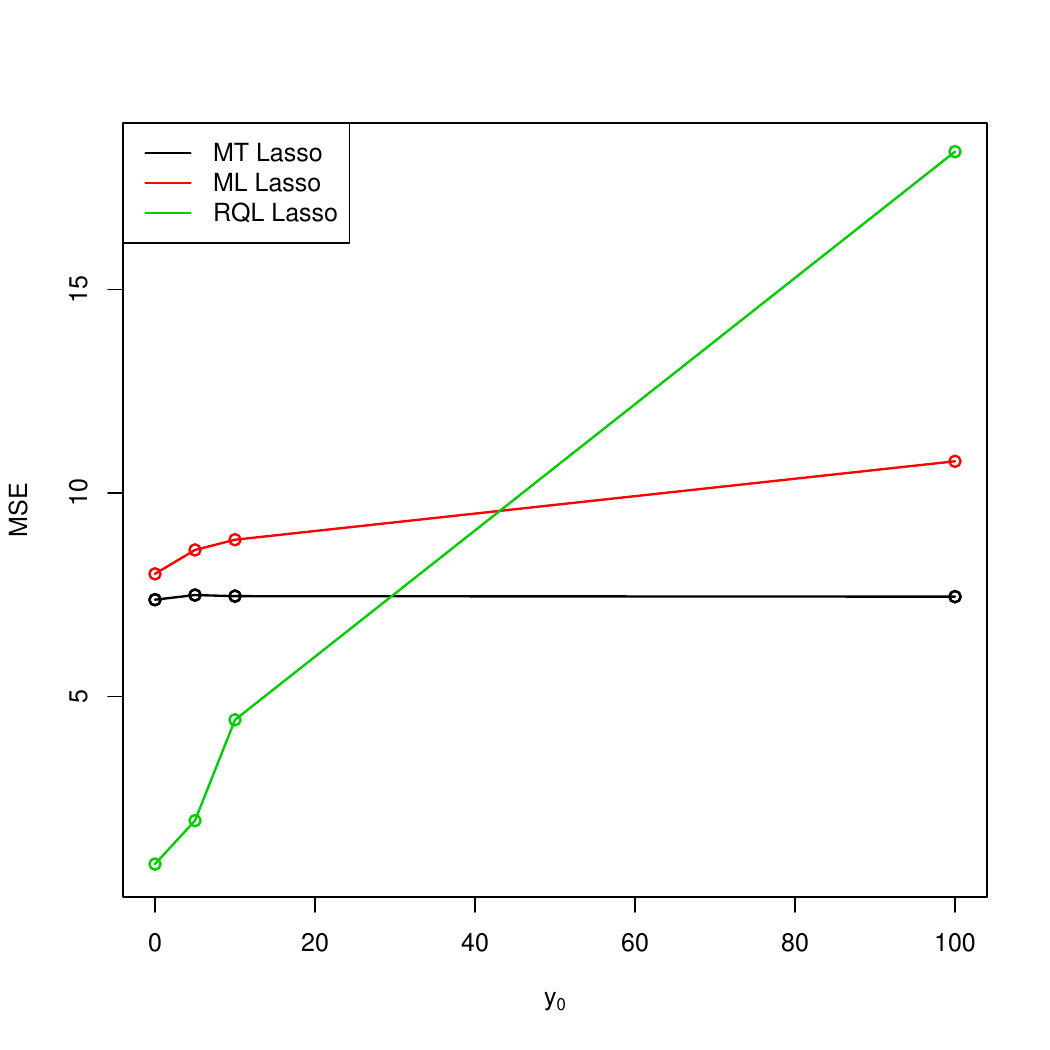}
			&   \includegraphics[scale=0.3]{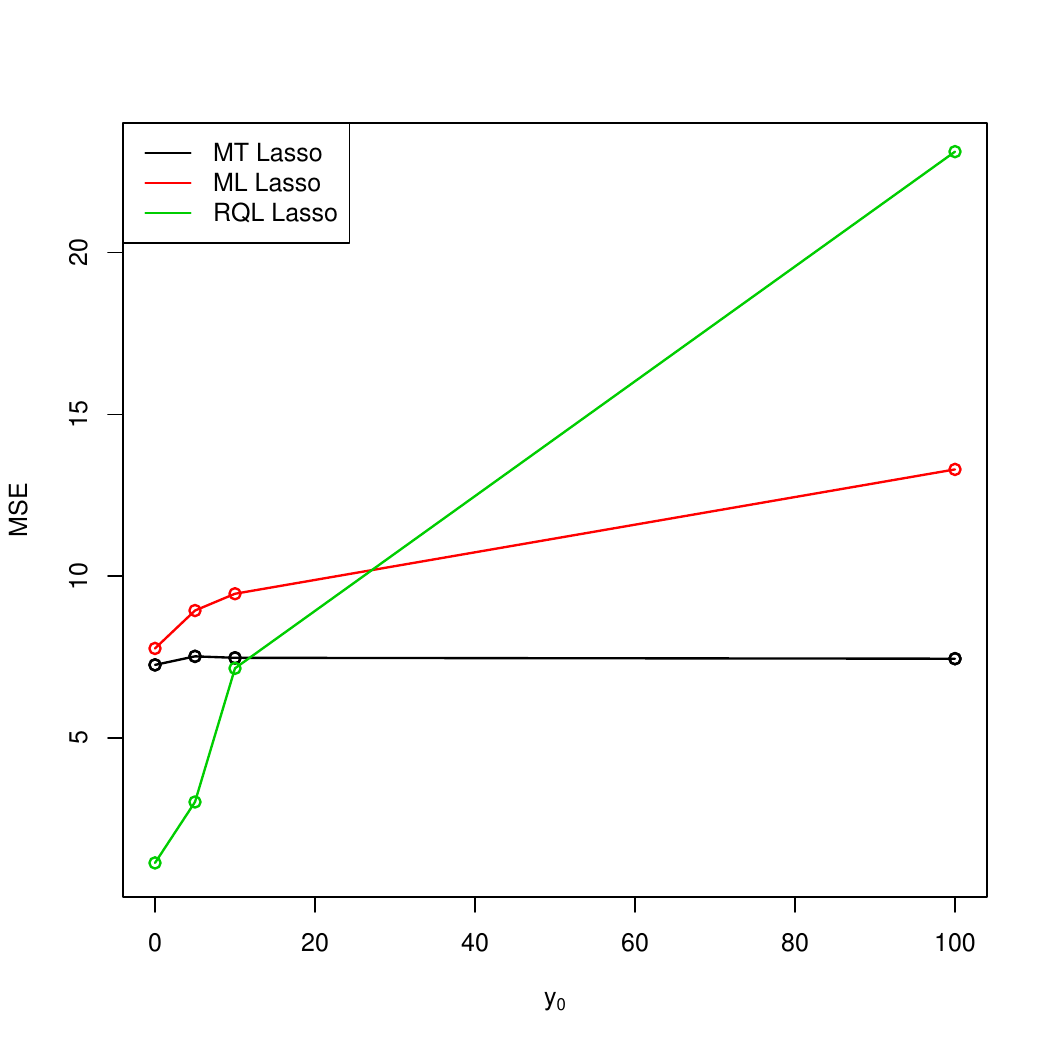} 
			&   \includegraphics[scale=0.3]{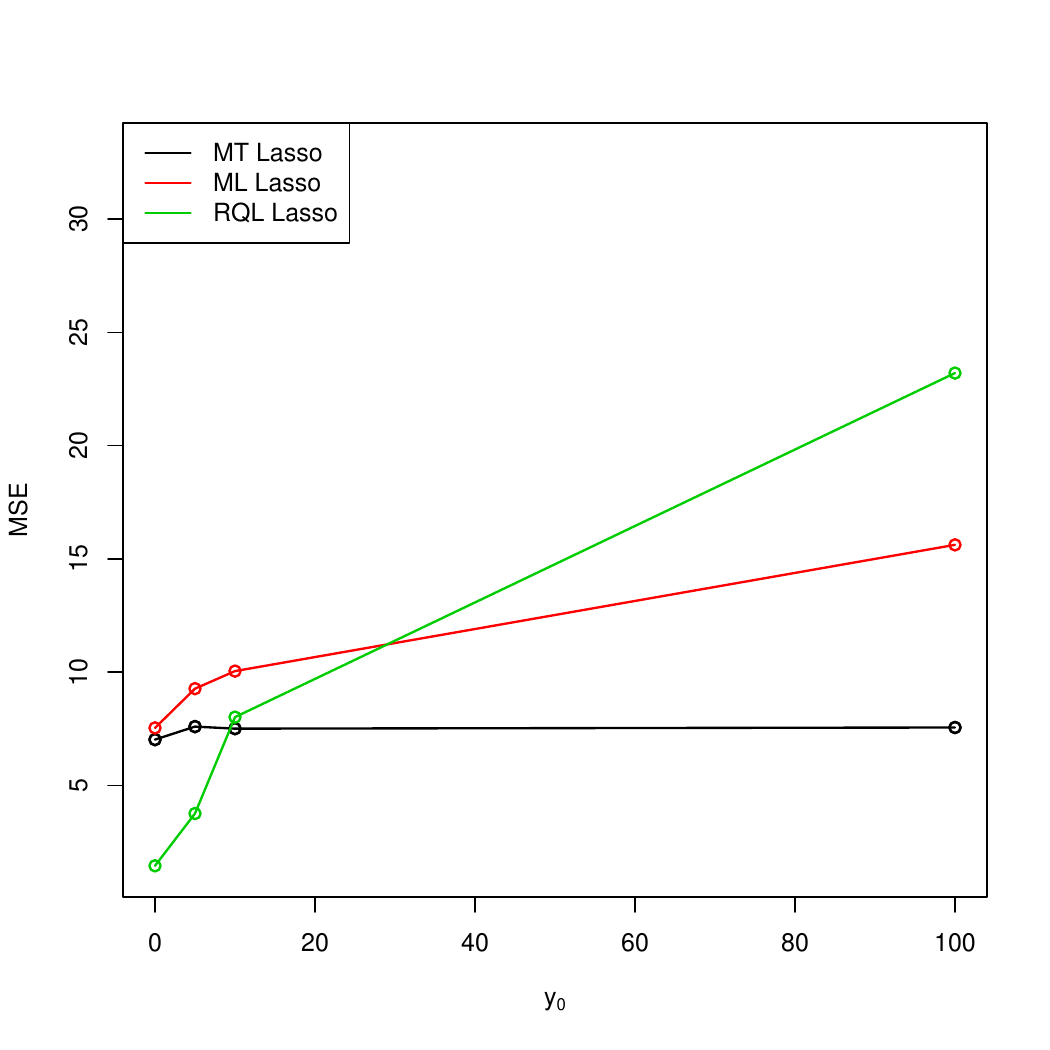} 
			\\
			$n=200,\, \epsilon=0.05$ & $n=200,\, \epsilon=0.1$  & $n=200,\, \epsilon=0.15$ \\
			\includegraphics[scale=0.3]{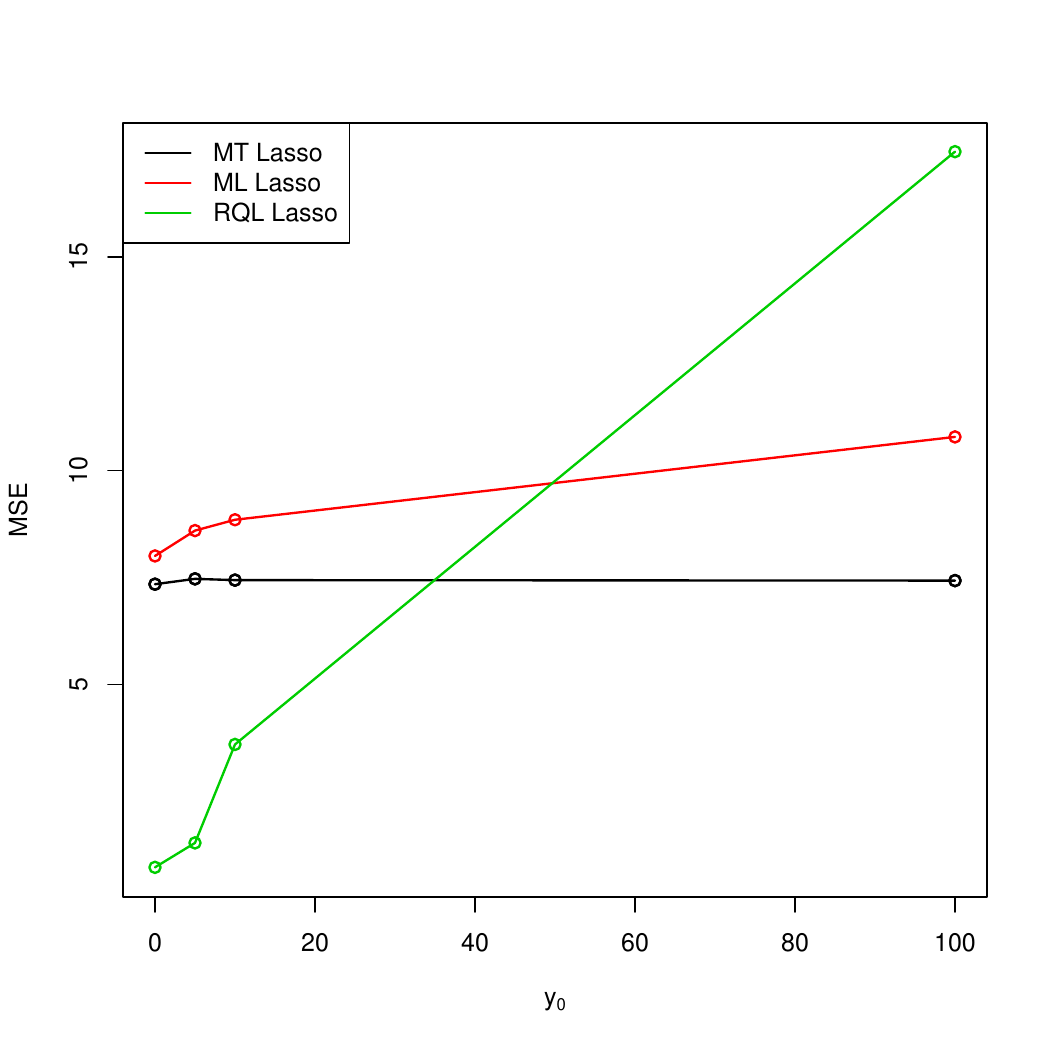}
			&   \includegraphics[scale=0.3]{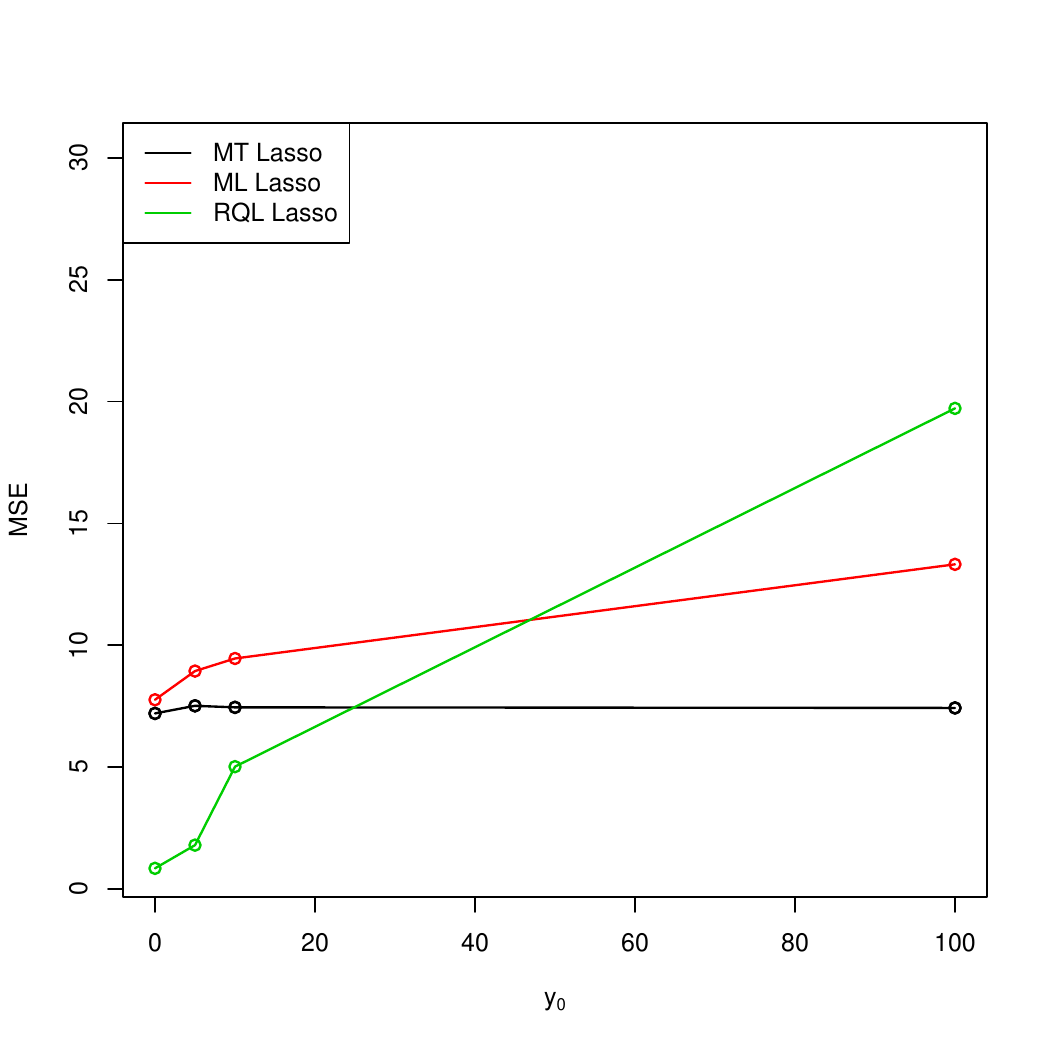} 
			&   \includegraphics[scale=0.3]{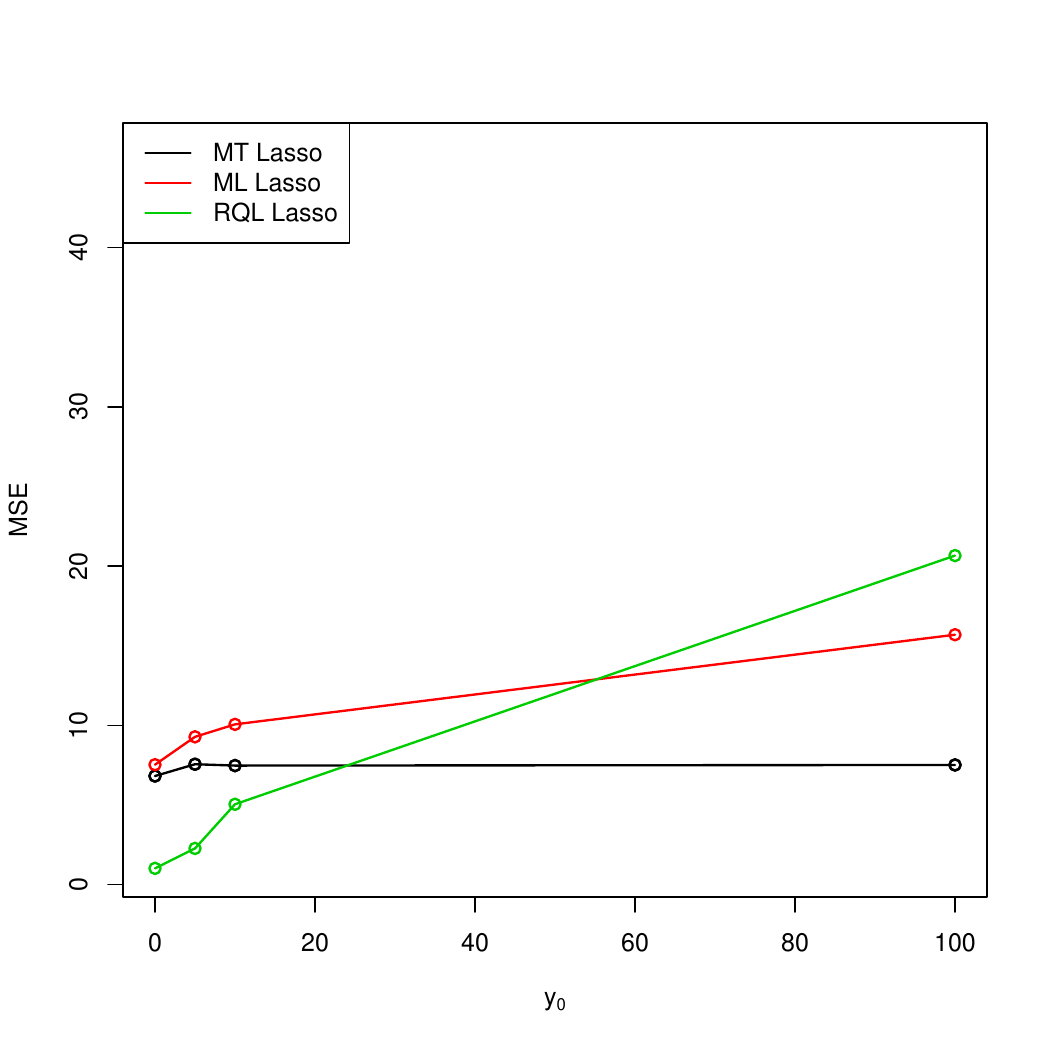}
		\end{tabular}
		\caption{\small \label{fig:msesimuRQLp400Lasso}  Mean squared errors for model AMR and Lasso methods, $p=400$.}
	\end{center} 
\end{figure}

\begin{figure}[tp]
	\begin{center}
		\footnotesize
		\renewcommand{\arraystretch}{0.22}
		\newcolumntype{M}{>{\centering\arraybackslash}m{\dimexpr.01\linewidth-1\tabcolsep}}
		\newcolumntype{G}{>{\centering\arraybackslash}m{\dimexpr.38\linewidth-1\tabcolsep}}
		\begin{tabular}{GGG}
			$n=50,\, \epsilon=0.05$ & $n=50,\, \epsilon=0.1$  & $n=50,\, \epsilon=0.15$ \\
			\includegraphics[scale=0.29]{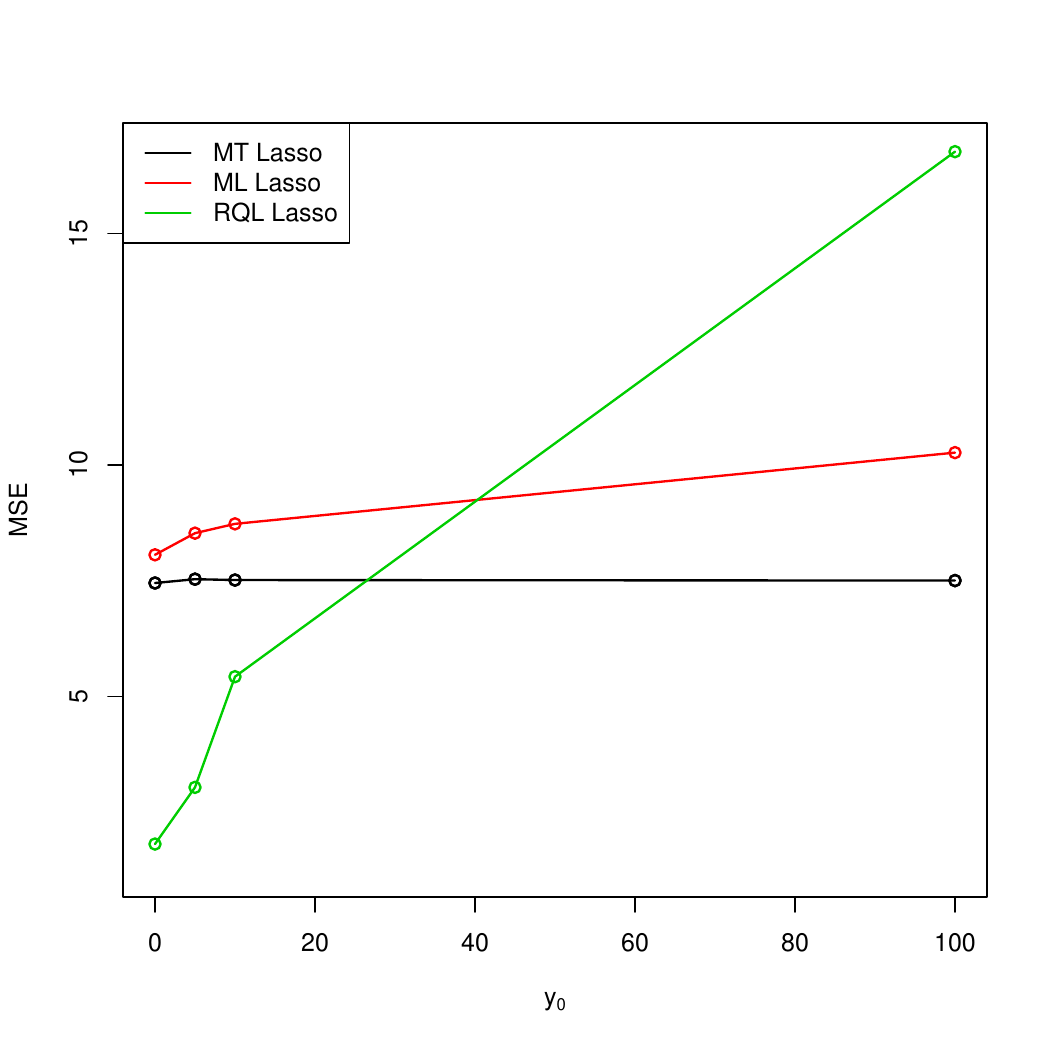} 
			&  \includegraphics[scale=0.29]{simuAMR-mse-lasso-1600-50-0-1}
			&   \includegraphics[scale=0.29]{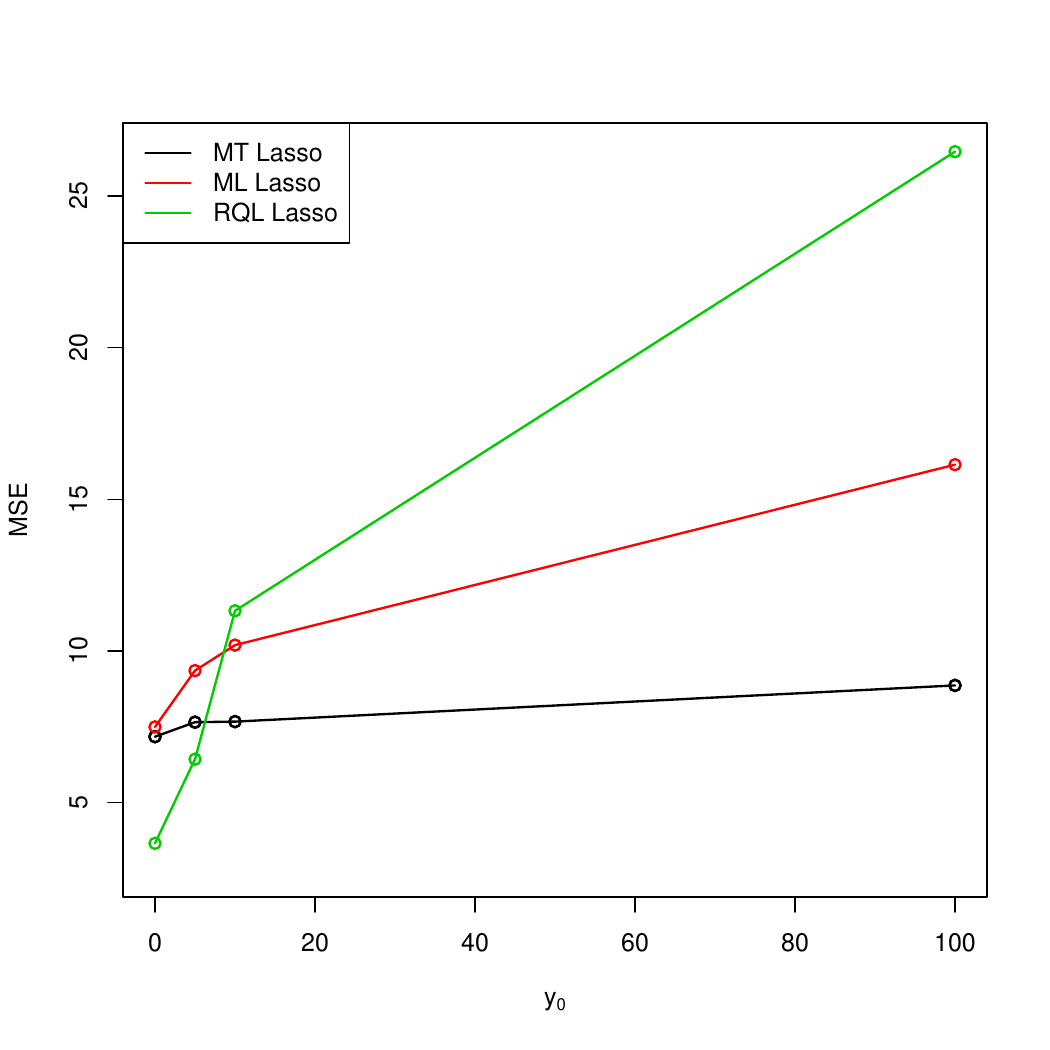}
			\\
			$n=100,\, \epsilon=0.05$ & $n=100,\, \epsilon=0.1$  & $n=100,\, \epsilon=0.15$ \\
			\includegraphics[scale=0.3]{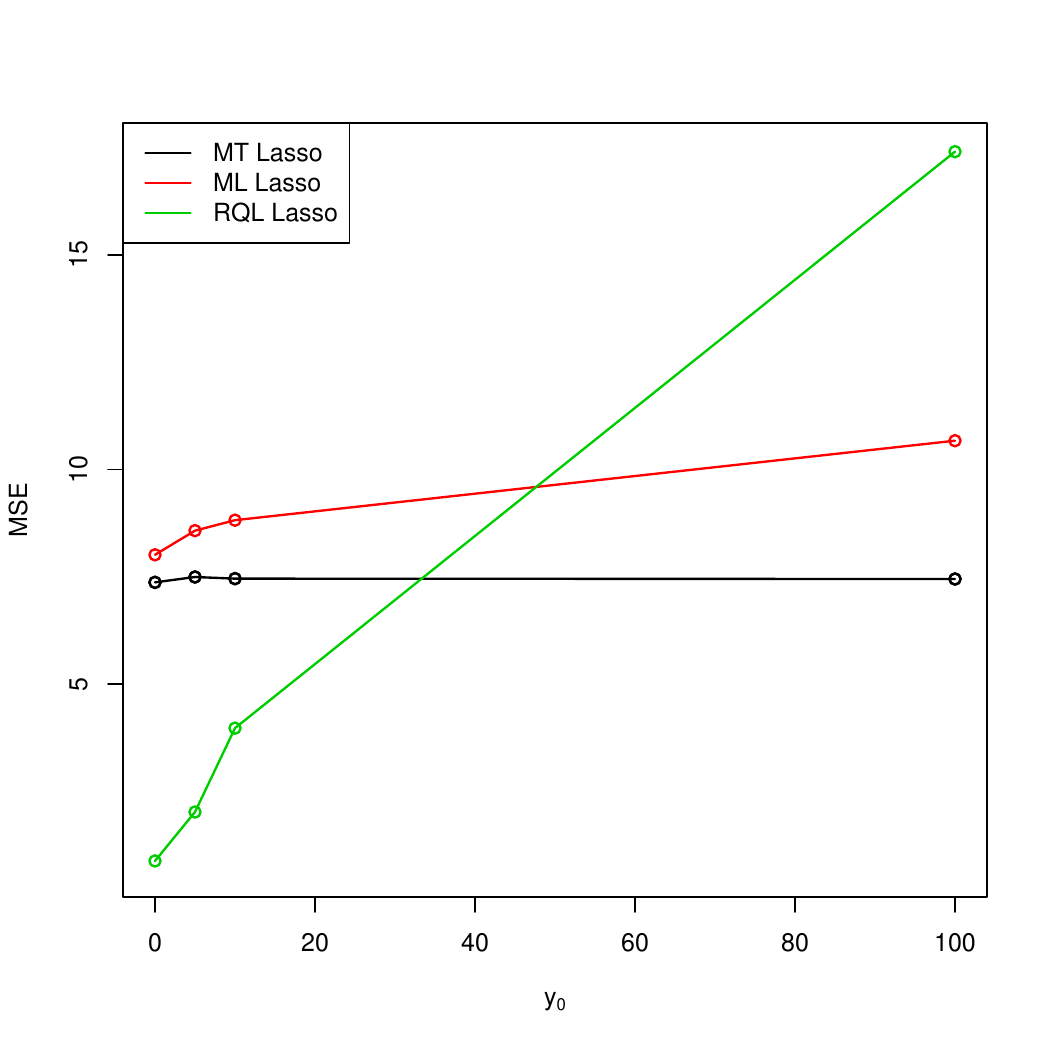}
			&   \includegraphics[scale=0.3]{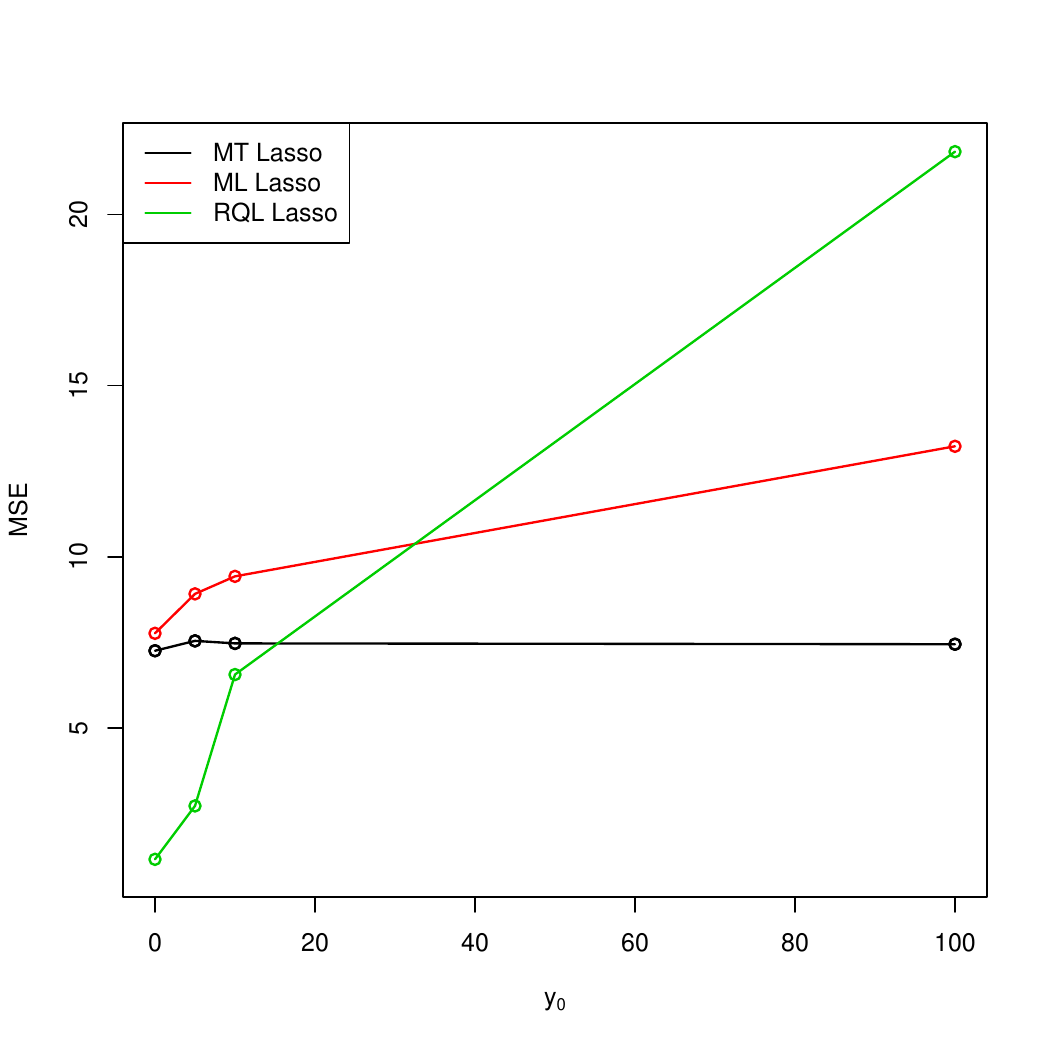} 
			&   \includegraphics[scale=0.3]{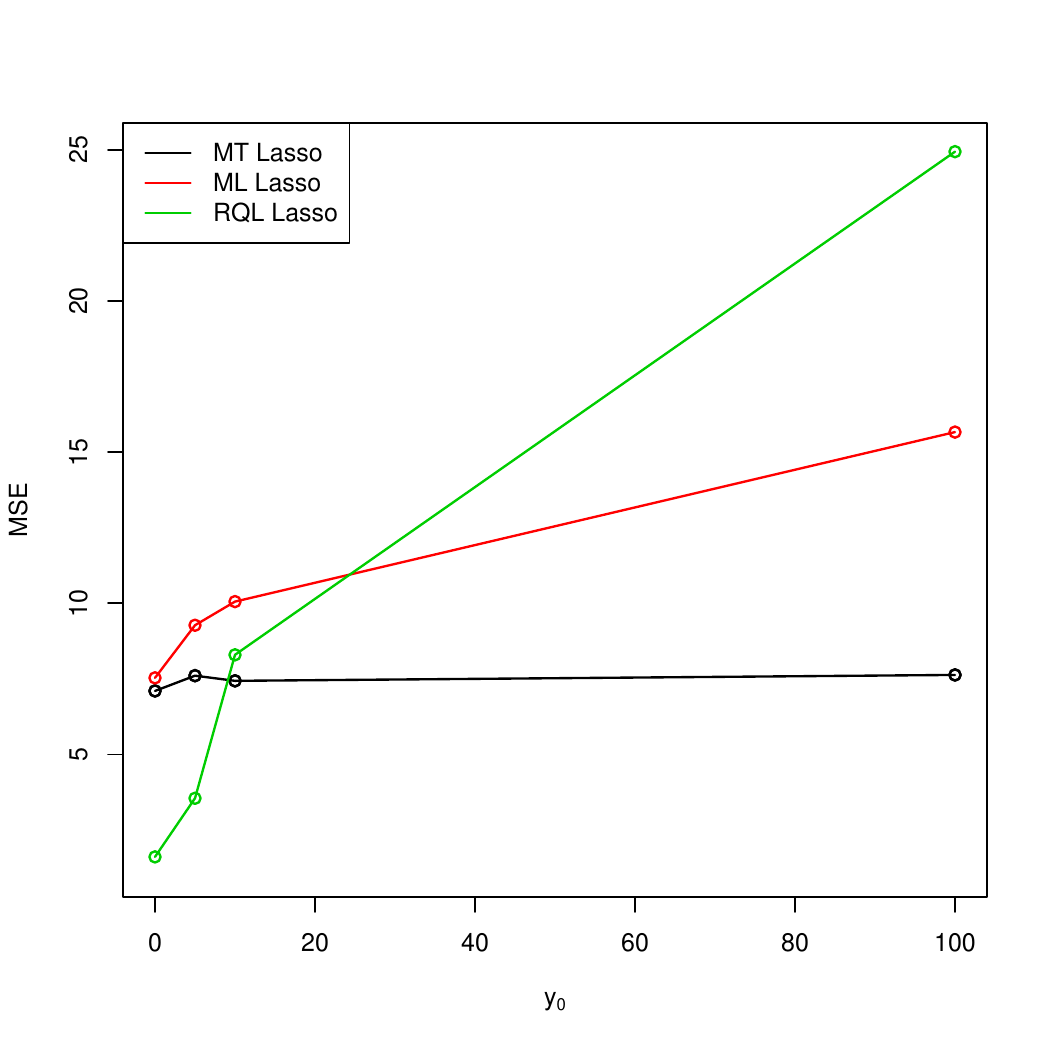} 
			\\
			$n=200,\, \epsilon=0.05$ & $n=200,\, \epsilon=0.1$  & $n=200,\, \epsilon=0.15$ \\
			\includegraphics[scale=0.3]{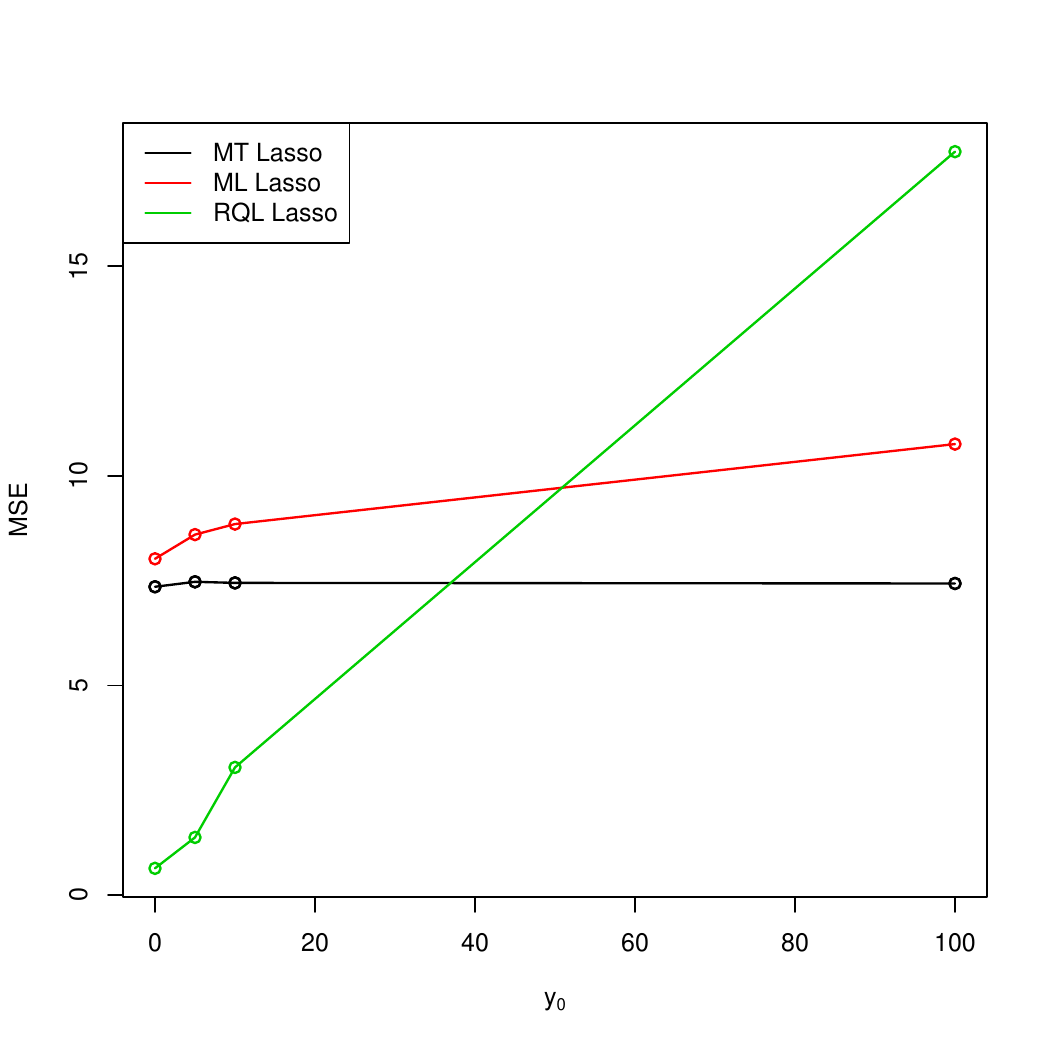}
			&   \includegraphics[scale=0.3]{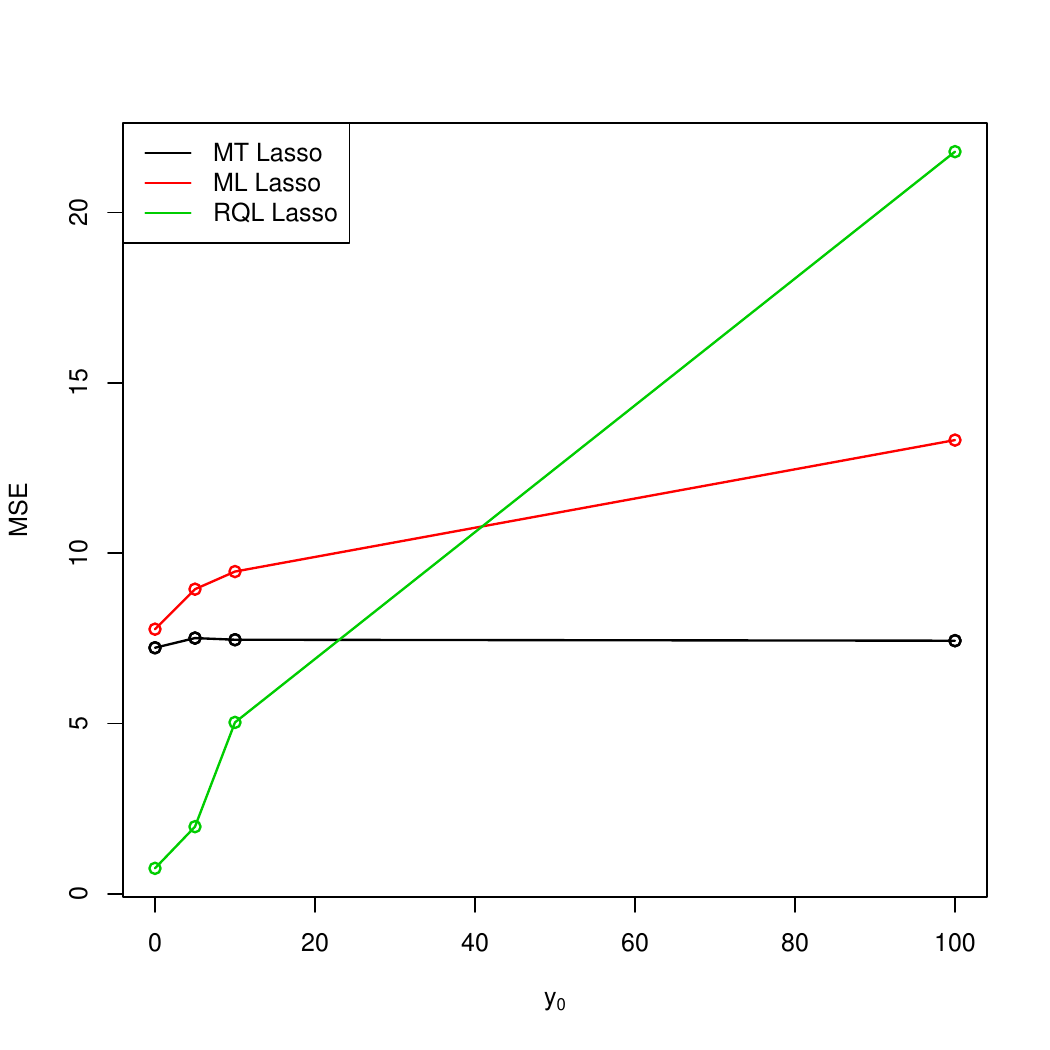} 
			&   \includegraphics[scale=0.3]{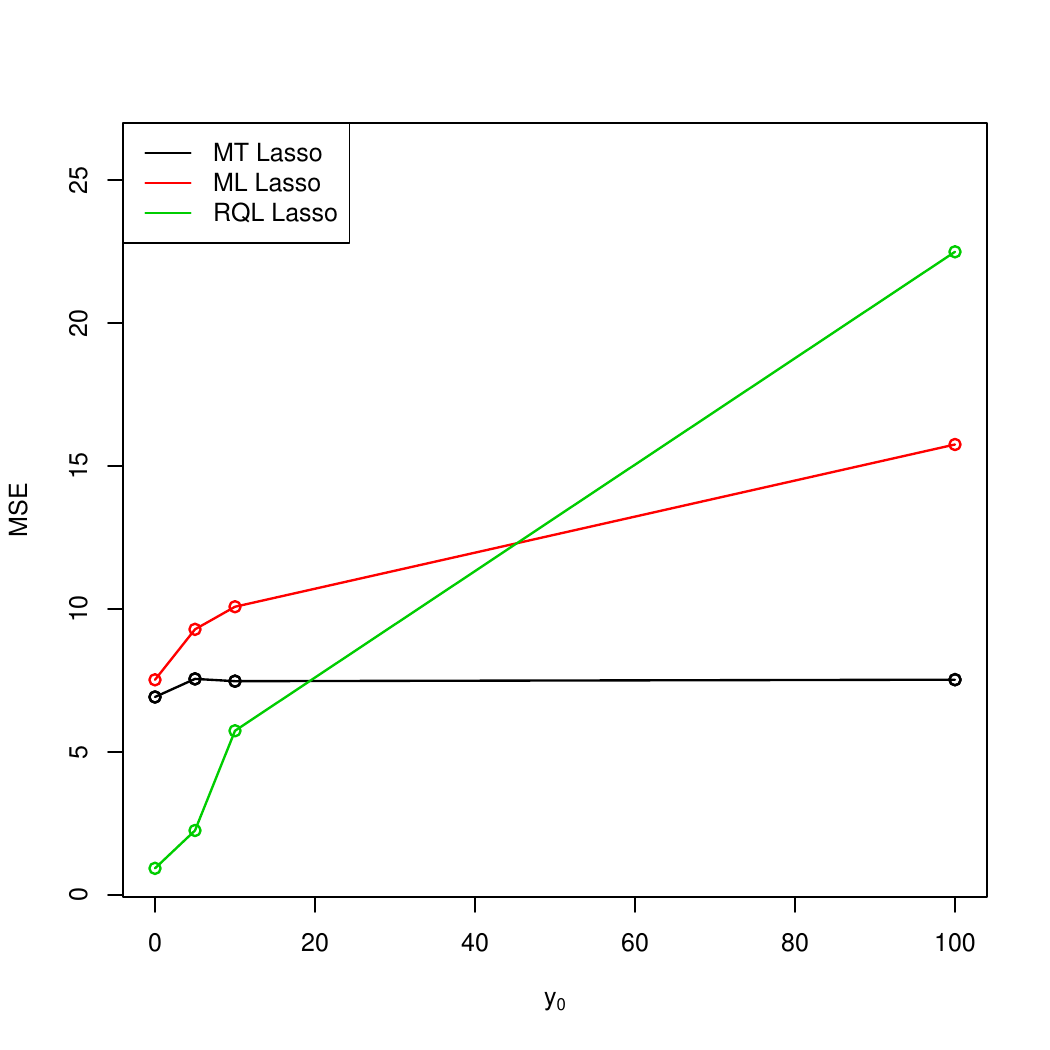}
		\end{tabular}
		\caption{\small \label{fig:msesimuRQLp1600Lasso} Mean squared errors for model AMR and Lasso methods, $p=1600$.}
	\end{center} 
\end{figure}

%%%%%%%%%%%%%%%%%%%%%%  times %%%%%%%%%%%%%%%%%%%%%%%%%%%%%%%%%%%%%%%

\begin{figure}[tp]
	\begin{center}
		\footnotesize
		\renewcommand{\arraystretch}{0.22}
		\newcolumntype{M}{>{\centering\arraybackslash}m{\dimexpr.01\linewidth-1\tabcolsep}}
		\newcolumntype{G}{>{\centering\arraybackslash}m{\dimexpr.38\linewidth-1\tabcolsep}}
		\begin{tabular}{GGG}
			$n=100,\, \epsilon=0.05$ & $n=100,\, \epsilon=0.1$  & $n=100,\, \epsilon=0.15$ \\
			\includegraphics[scale=0.29]{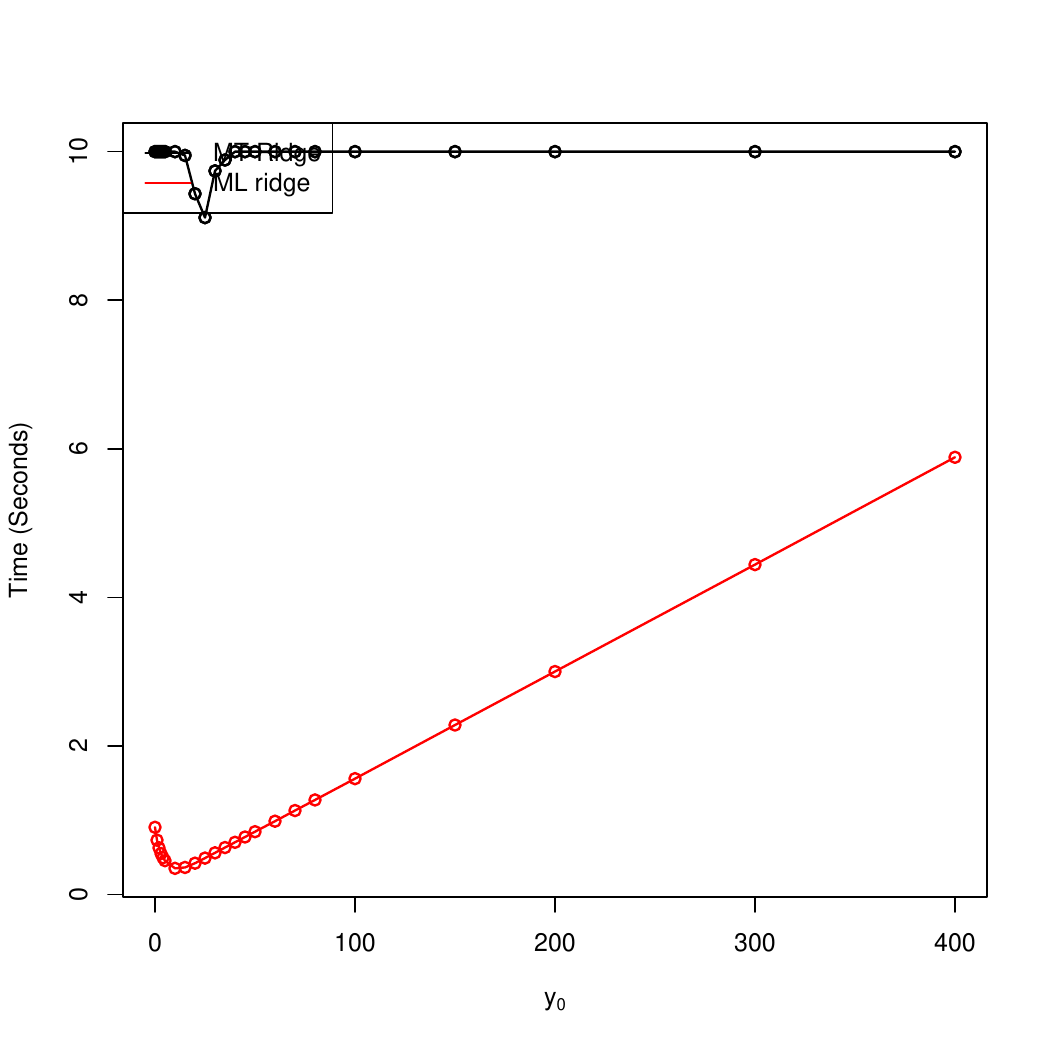} 
			&  \includegraphics[scale=0.29]{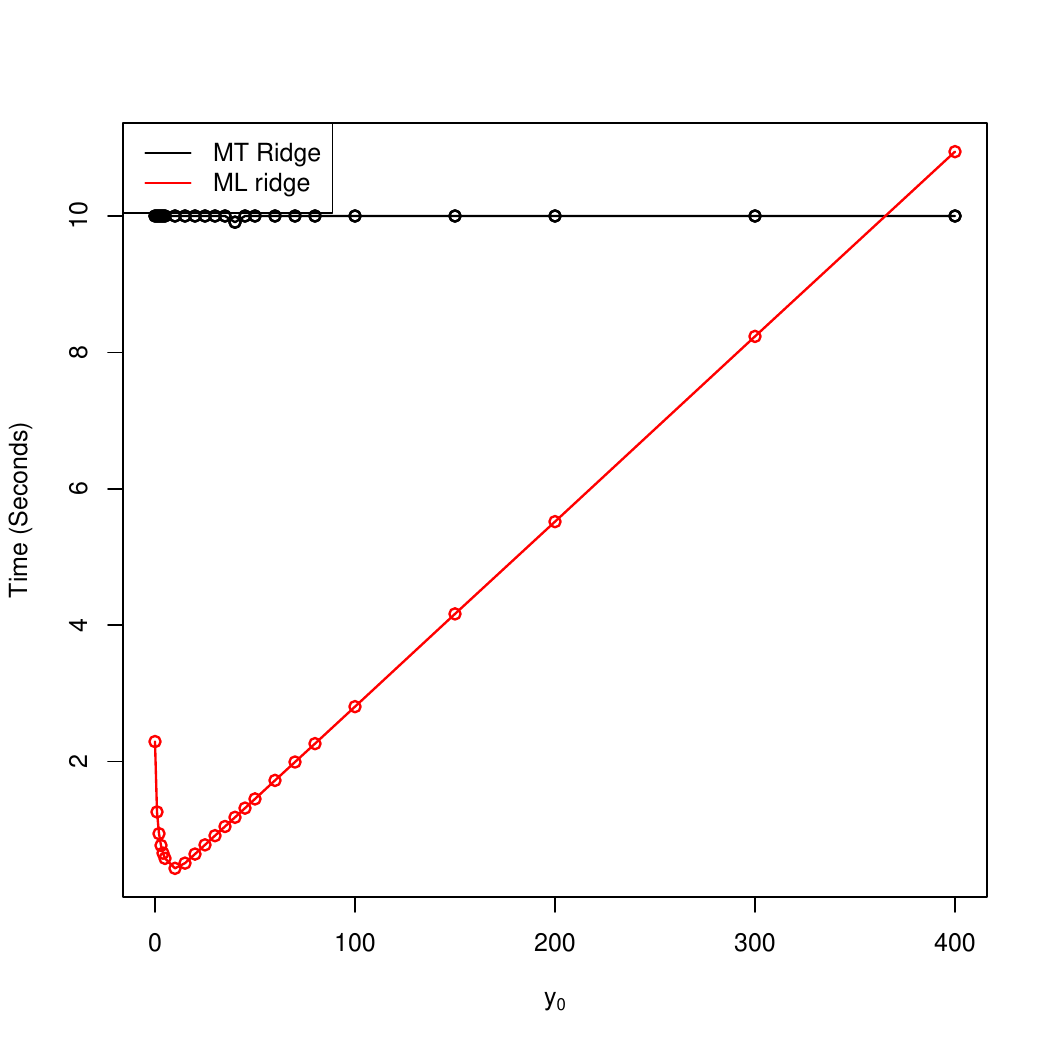}
			&   \includegraphics[scale=0.29]{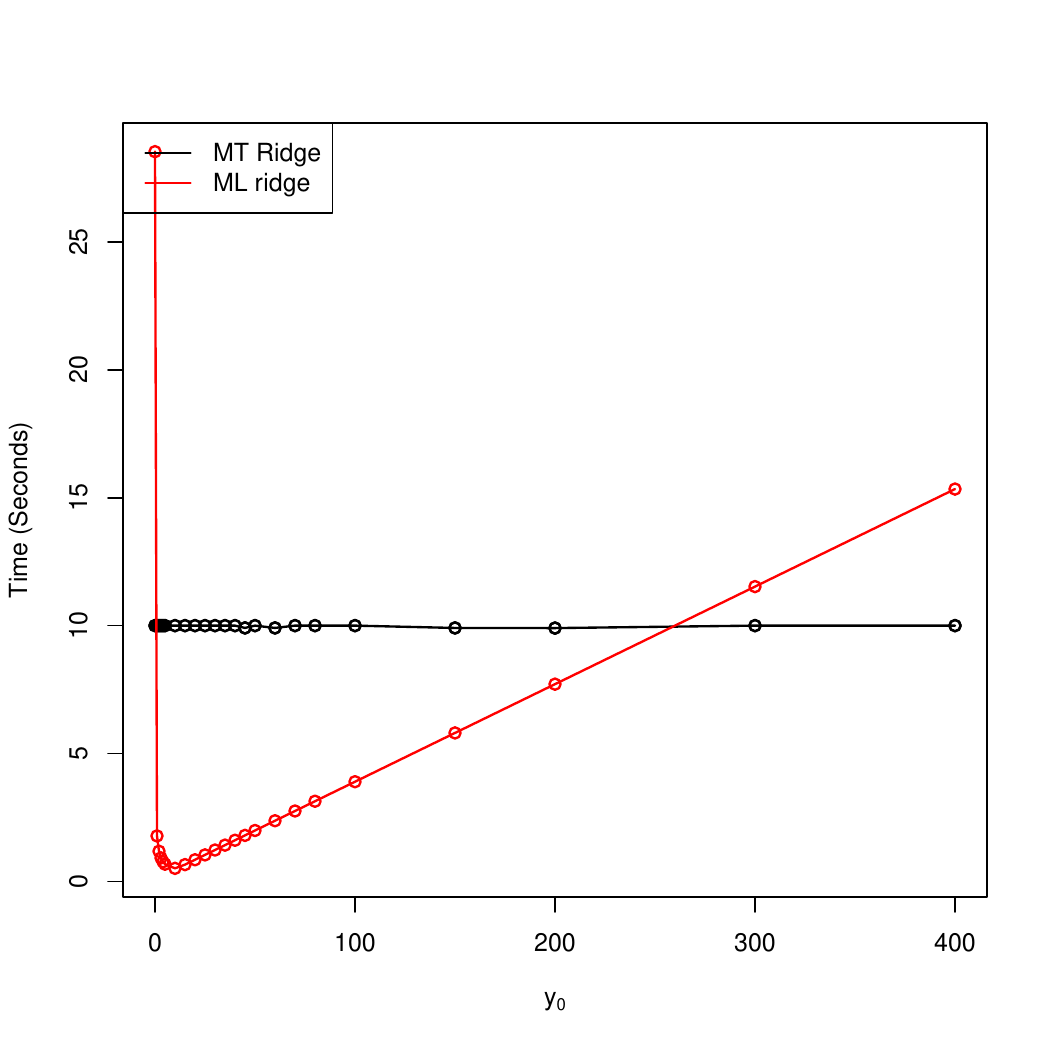}
			\\
			$n=400,\, \epsilon=0.05$ & $n=400,\, \epsilon=0.1$  & $n=400,\, \epsilon=0.15$ \\
			\includegraphics[scale=0.3]{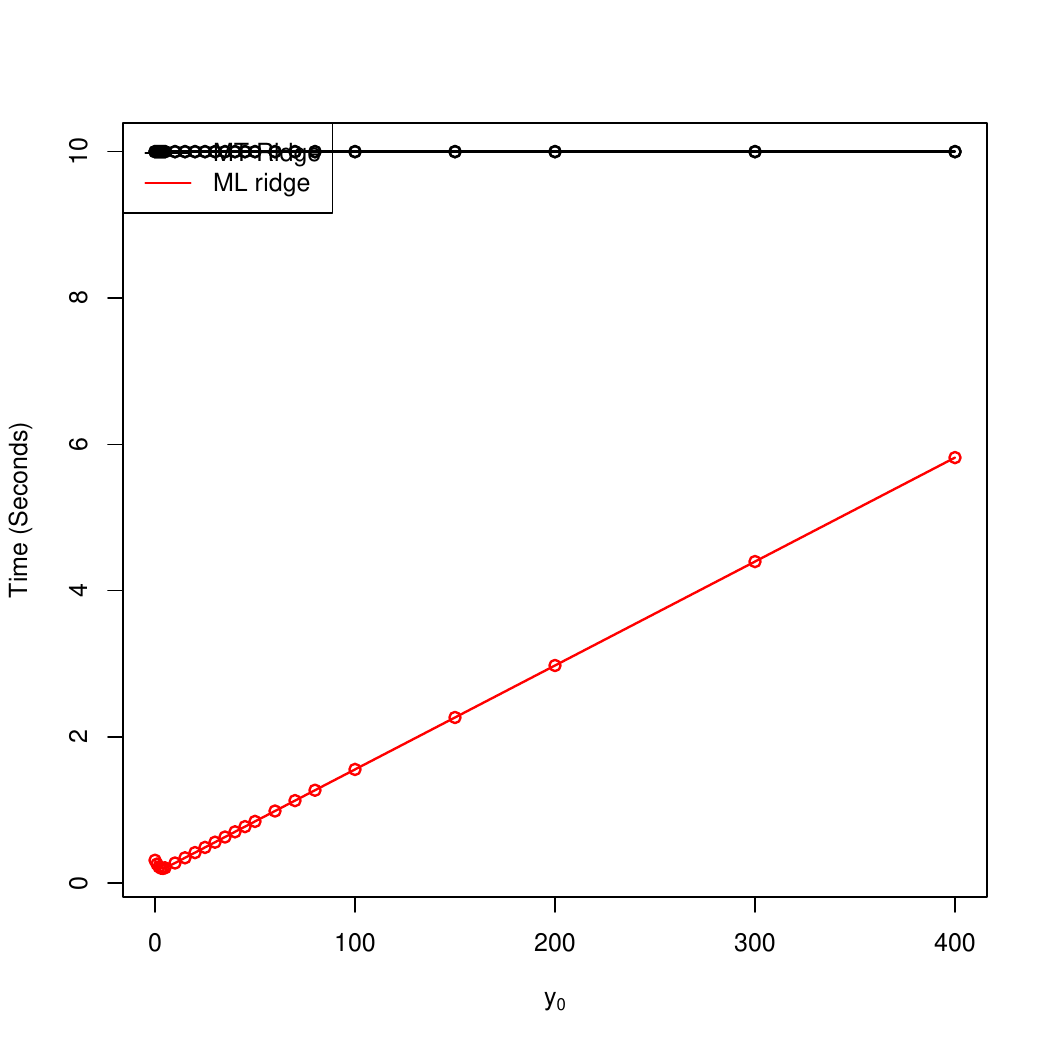}
			&   \includegraphics[scale=0.3]{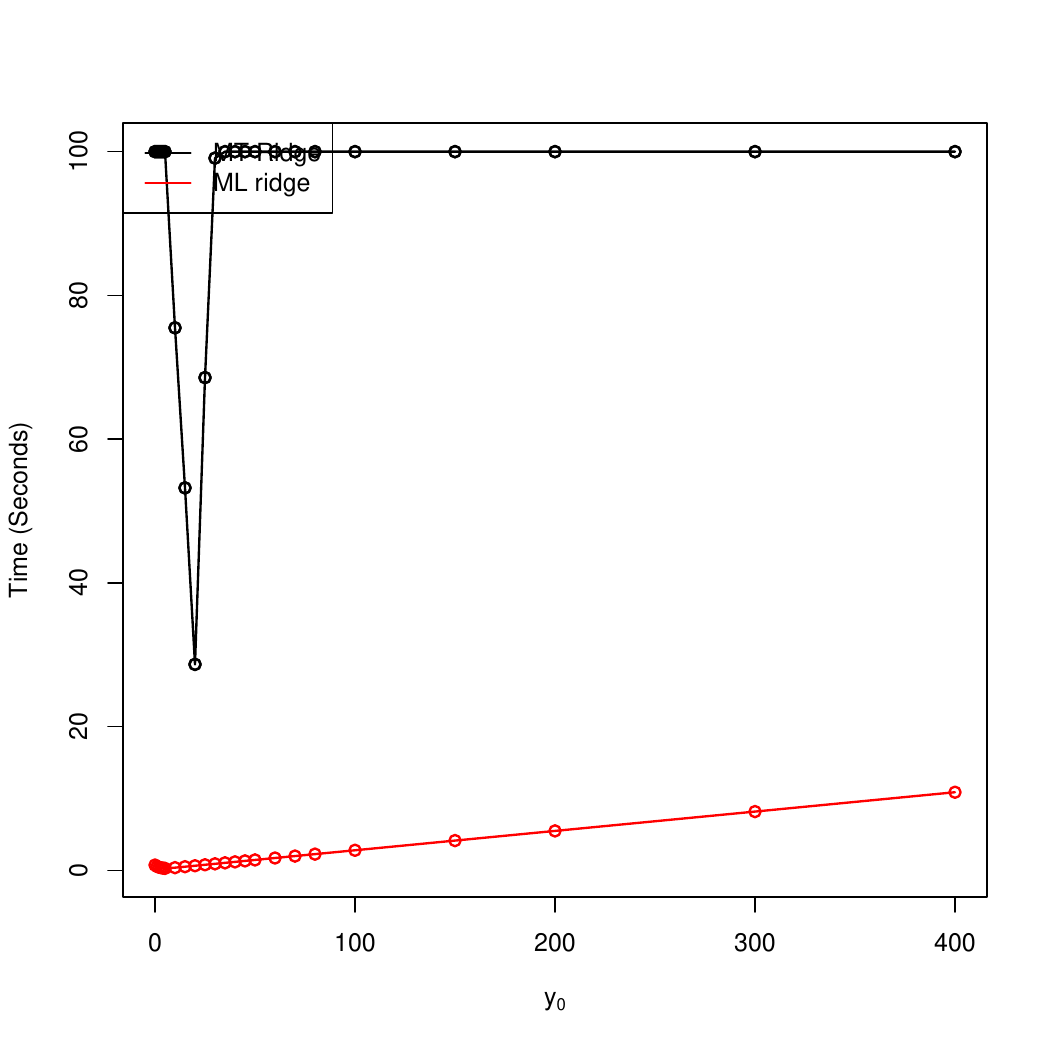} 
			&   \includegraphics[scale=0.3]{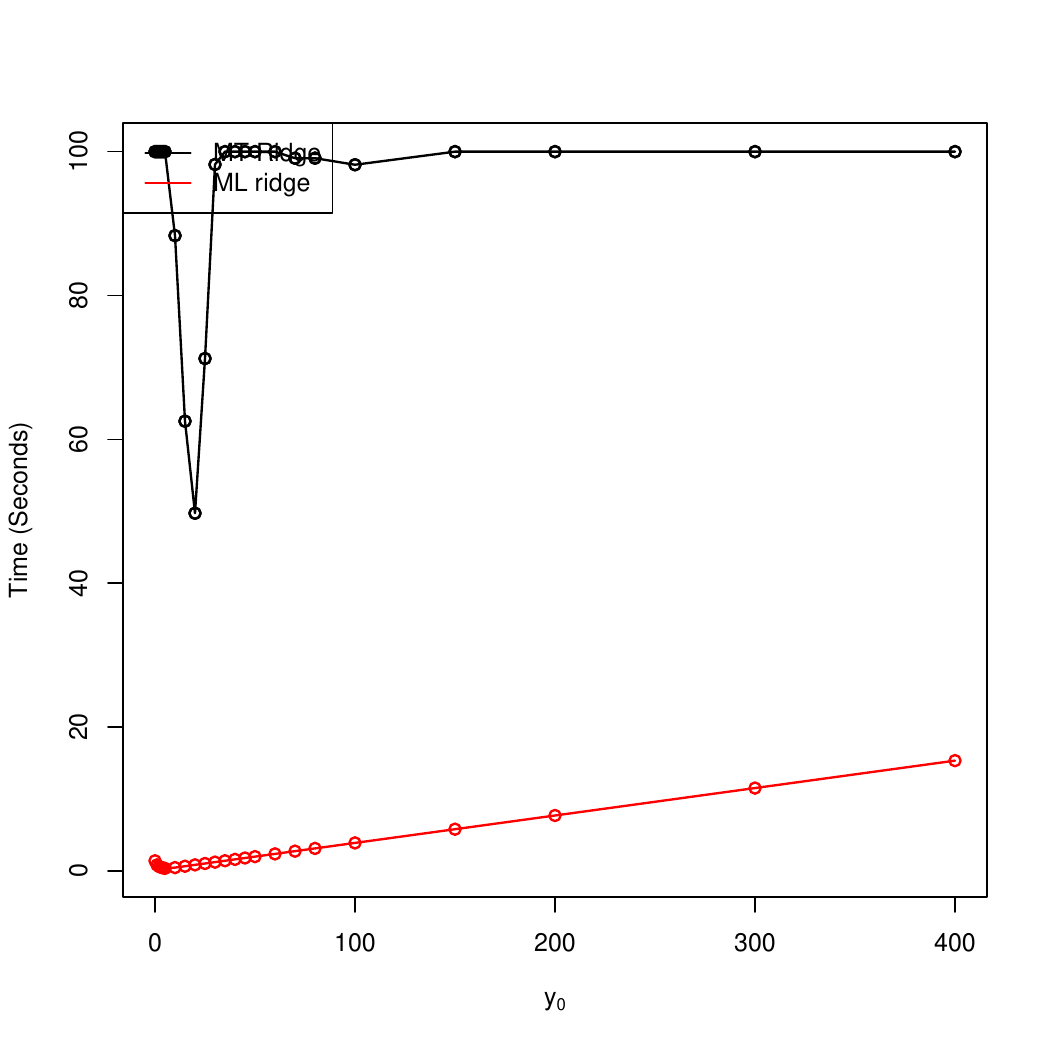} 
			\\
			$n=1000,\, \epsilon=0.05$ & $n=1000,\, \epsilon=0.1$  & $n=1000,\, \epsilon=0.15$ \\
			\includegraphics[scale=0.3]{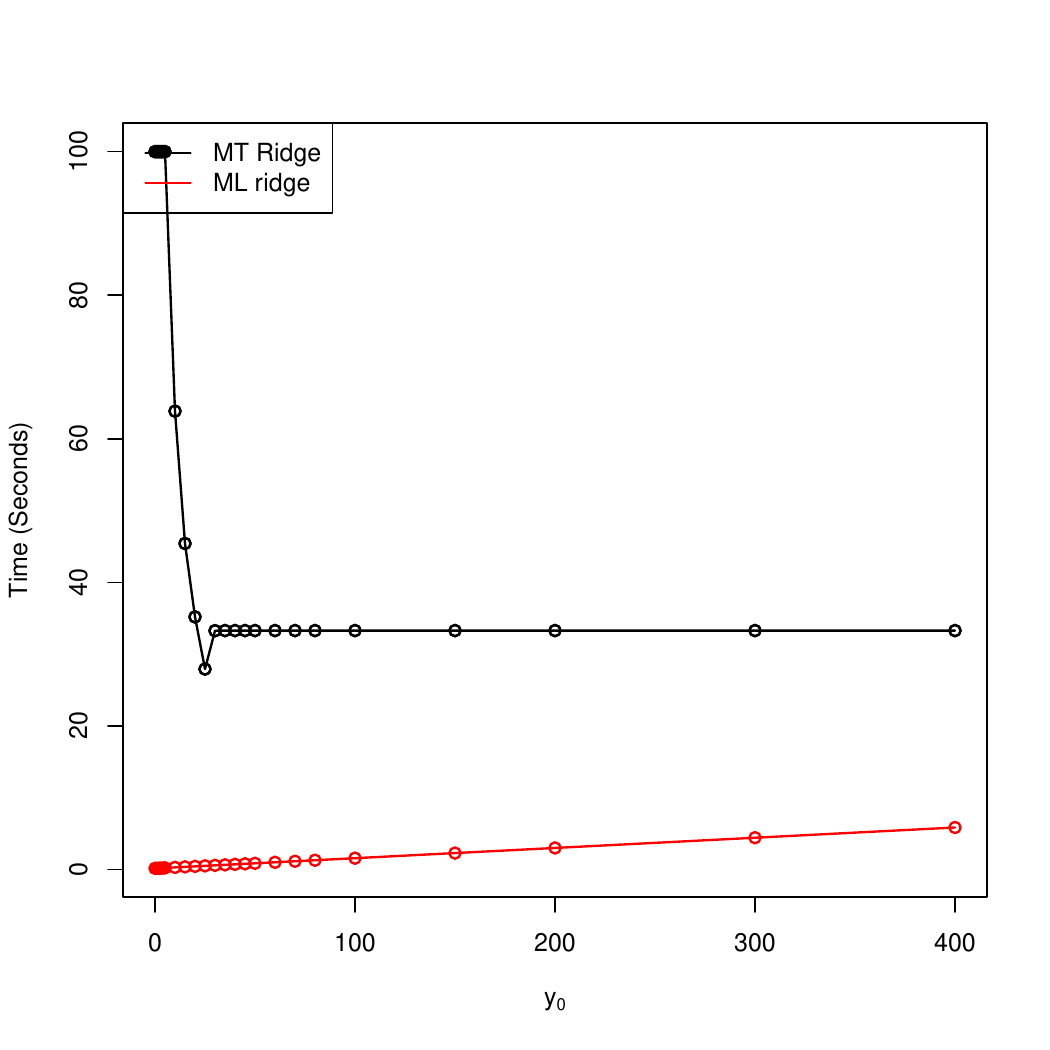}
			&   \includegraphics[scale=0.3]{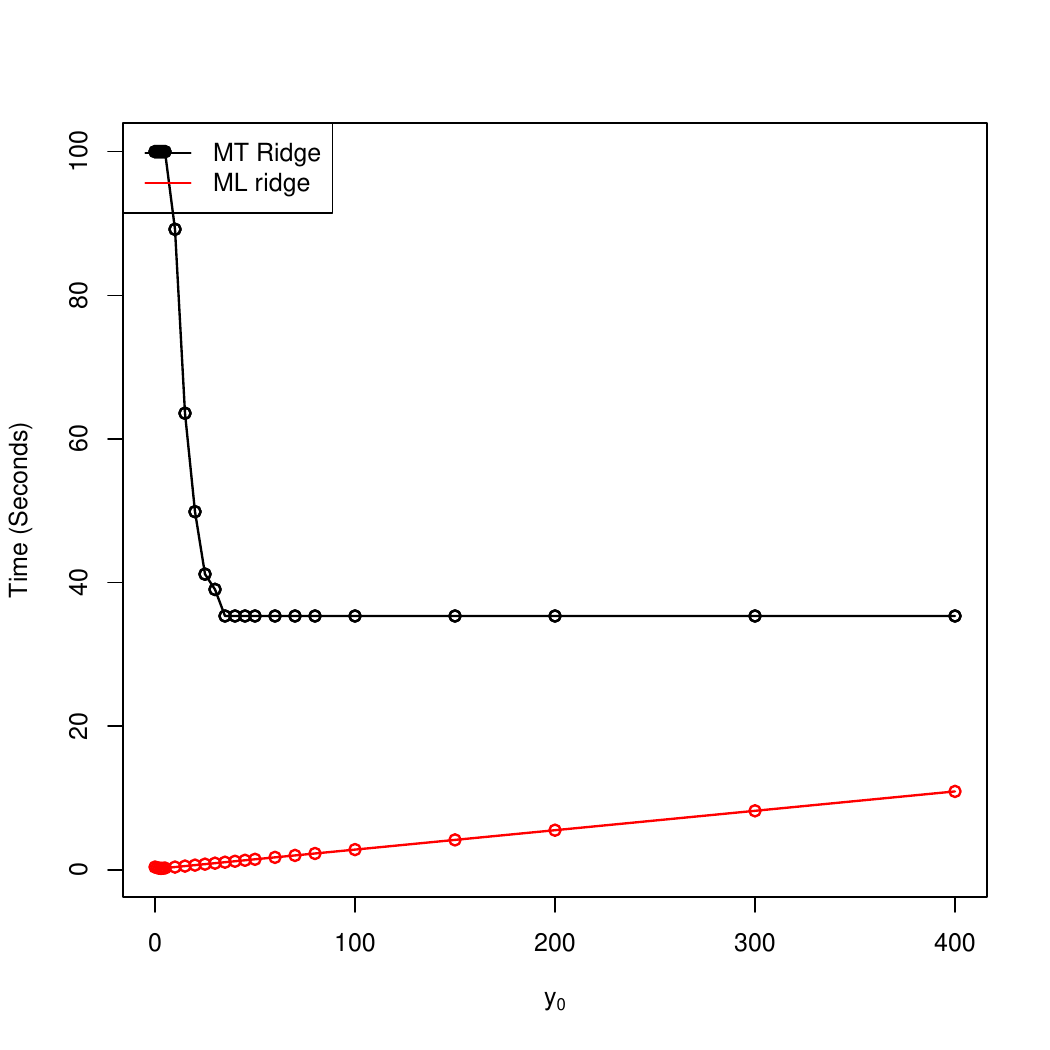} 
			&   \includegraphics[scale=0.3]{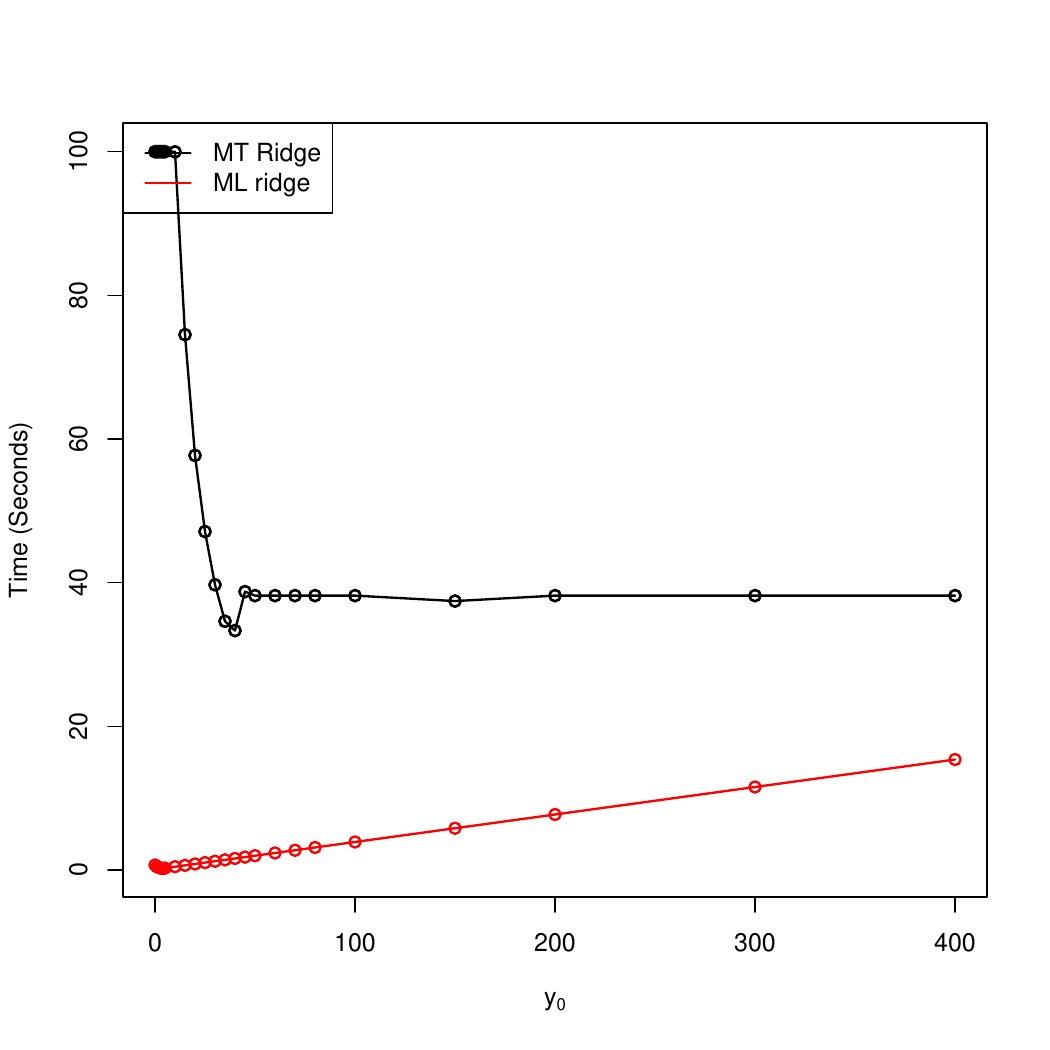}
		\end{tabular}
		\caption{\small \label{fig:timesimu1p10Ridge}  Computation times for model AVY and Ridge methods, $p=10$.}
	\end{center} 
\end{figure}

\begin{figure}[tp]
	\begin{center}
		\footnotesize
		\renewcommand{\arraystretch}{0.22}
		\newcolumntype{M}{>{\centering\arraybackslash}m{\dimexpr.01\linewidth-1\tabcolsep}}
		\newcolumntype{G}{>{\centering\arraybackslash}m{\dimexpr.38\linewidth-1\tabcolsep}}
		\begin{tabular}{GGG}
			$n=100,\, \epsilon=0.05$ & $n=100,\, \epsilon=0.1$  & $n=100,\, \epsilon=0.15$ \\
			\includegraphics[scale=0.29]{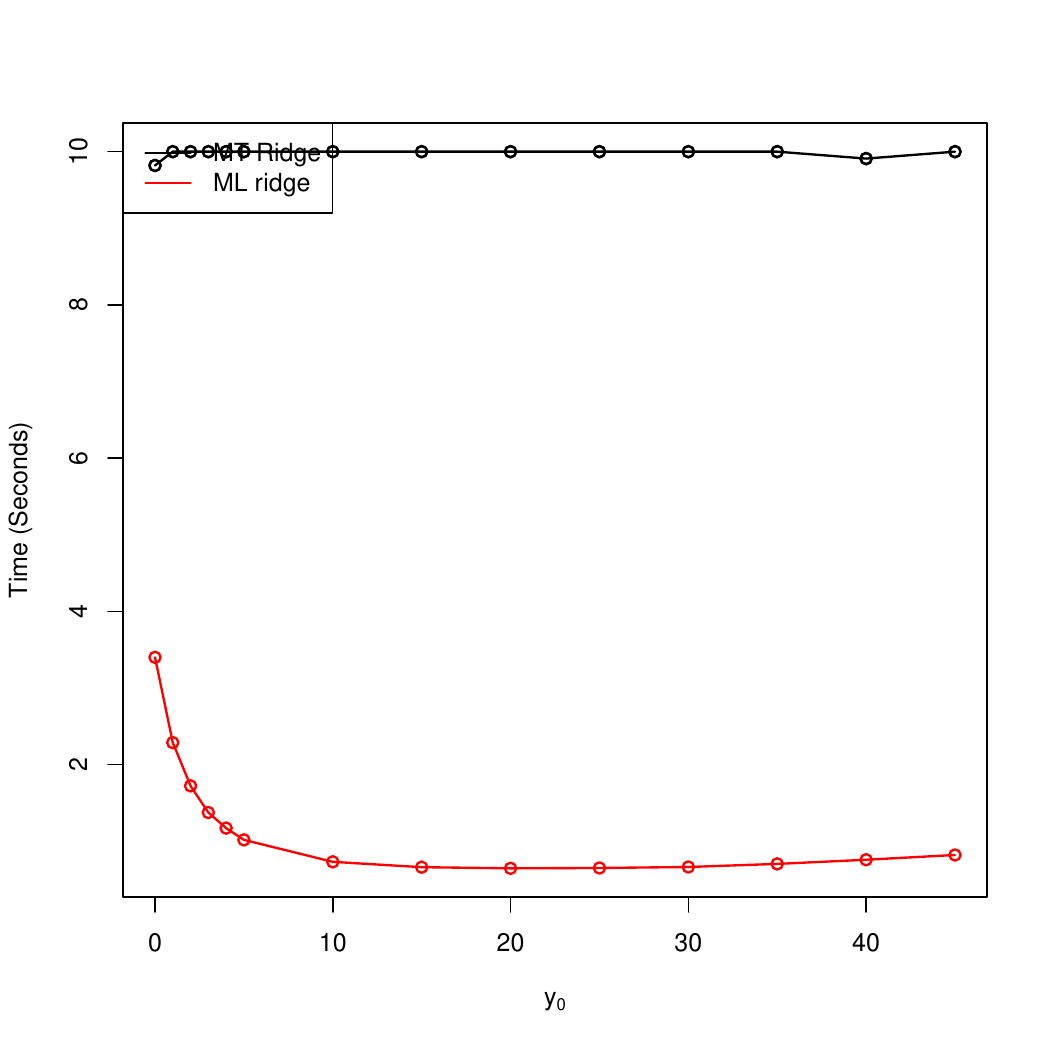} 
			&  \includegraphics[scale=0.29]{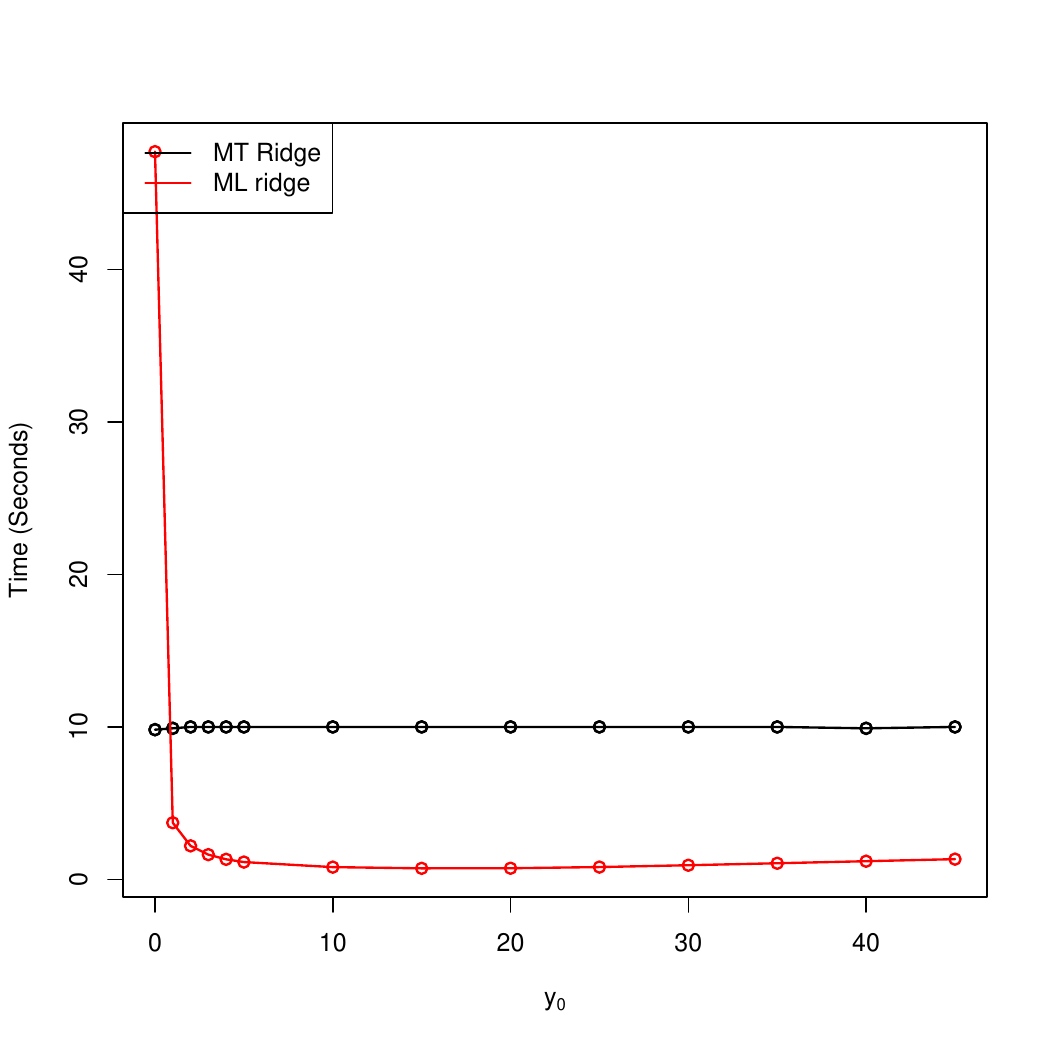}
			&   \includegraphics[scale=0.29]{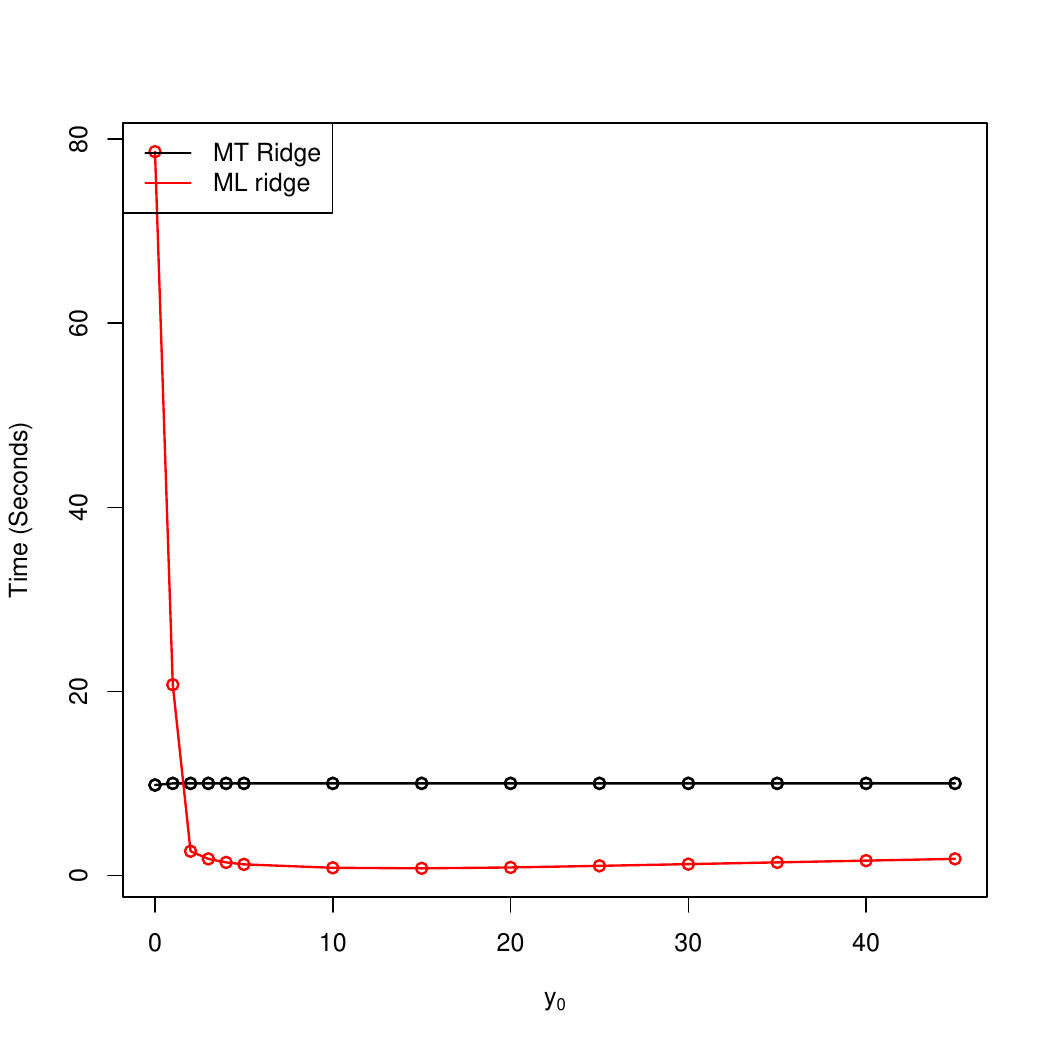}
			\\
			$n=400,\, \epsilon=0.05$ & $n=400,\, \epsilon=0.1$  & $n=400,\, \epsilon=0.15$ \\
			\includegraphics[scale=0.3]{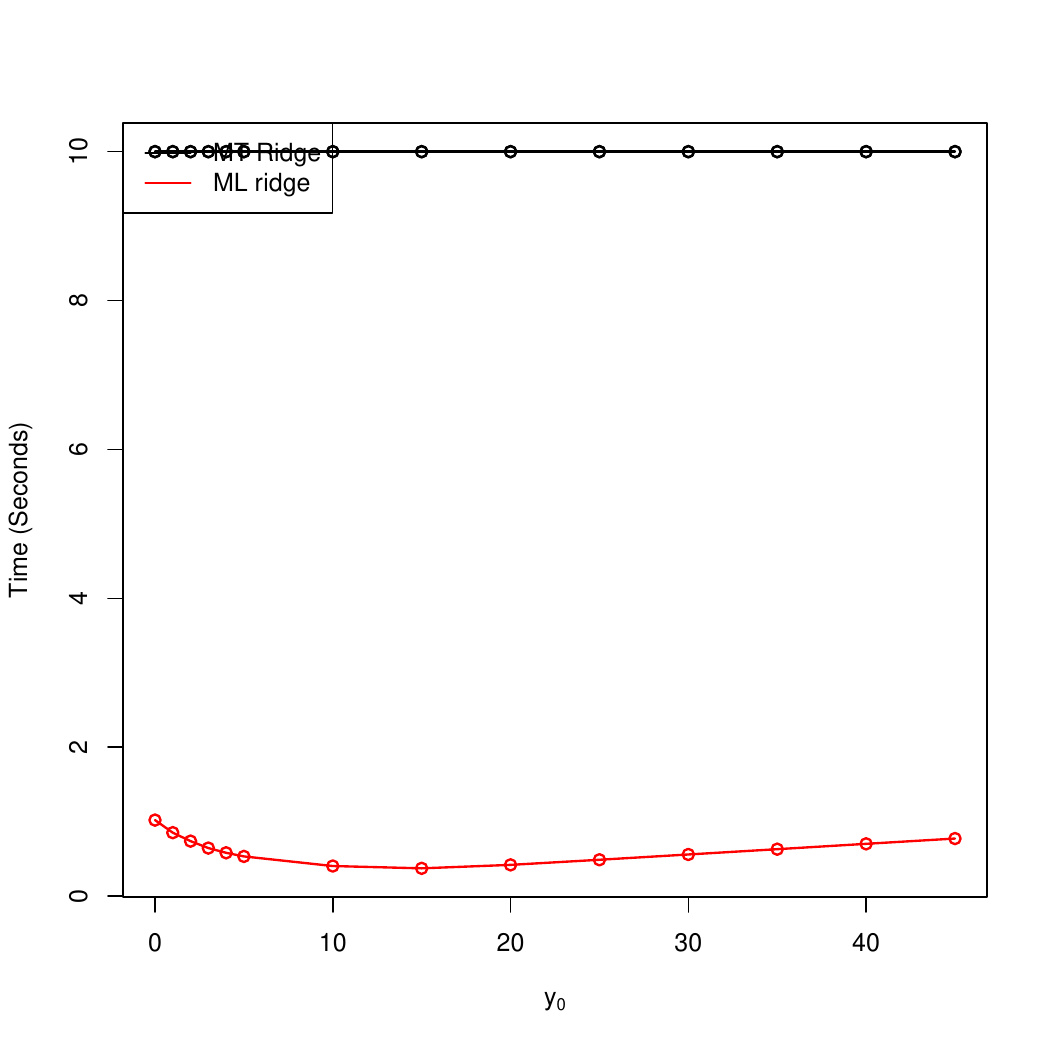}
			&   \includegraphics[scale=0.3]{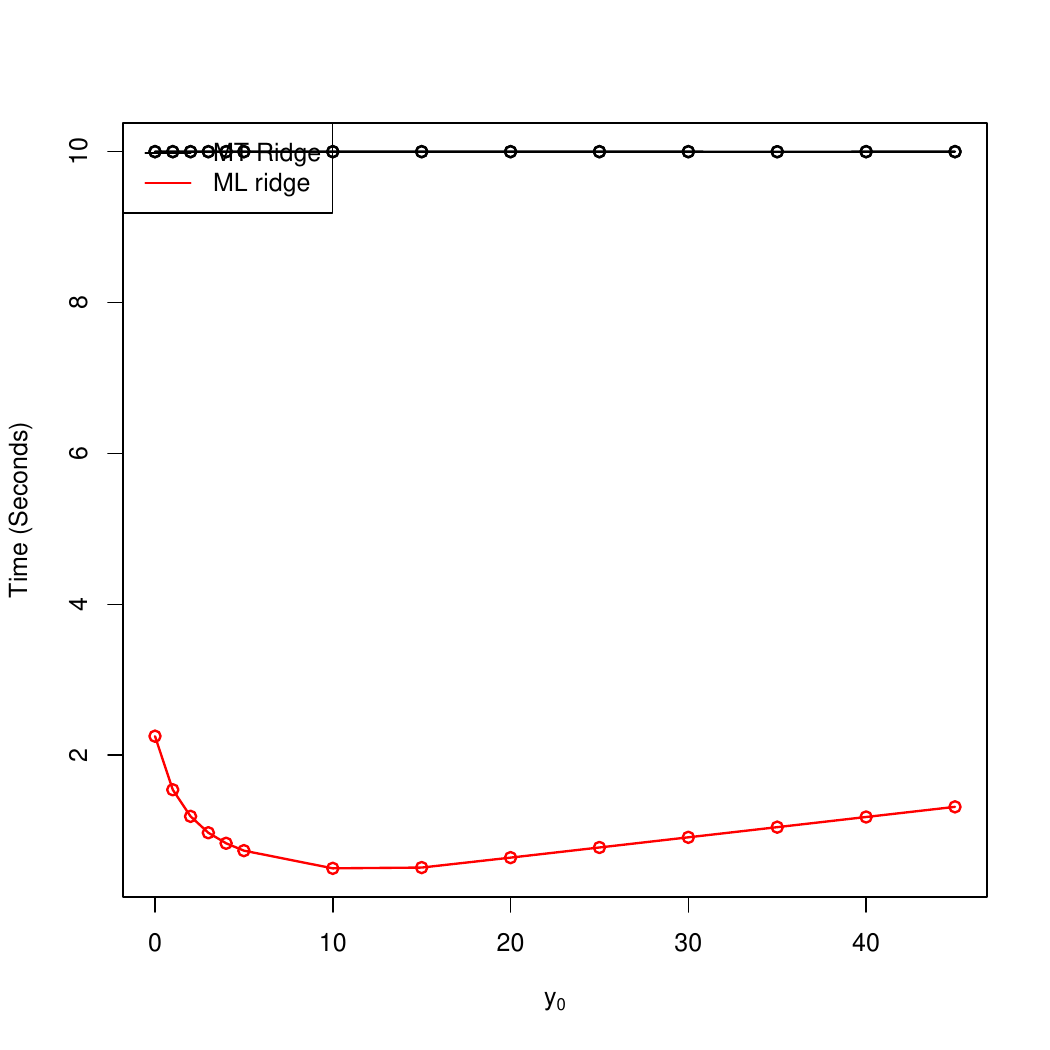} 
			&   \includegraphics[scale=0.3]{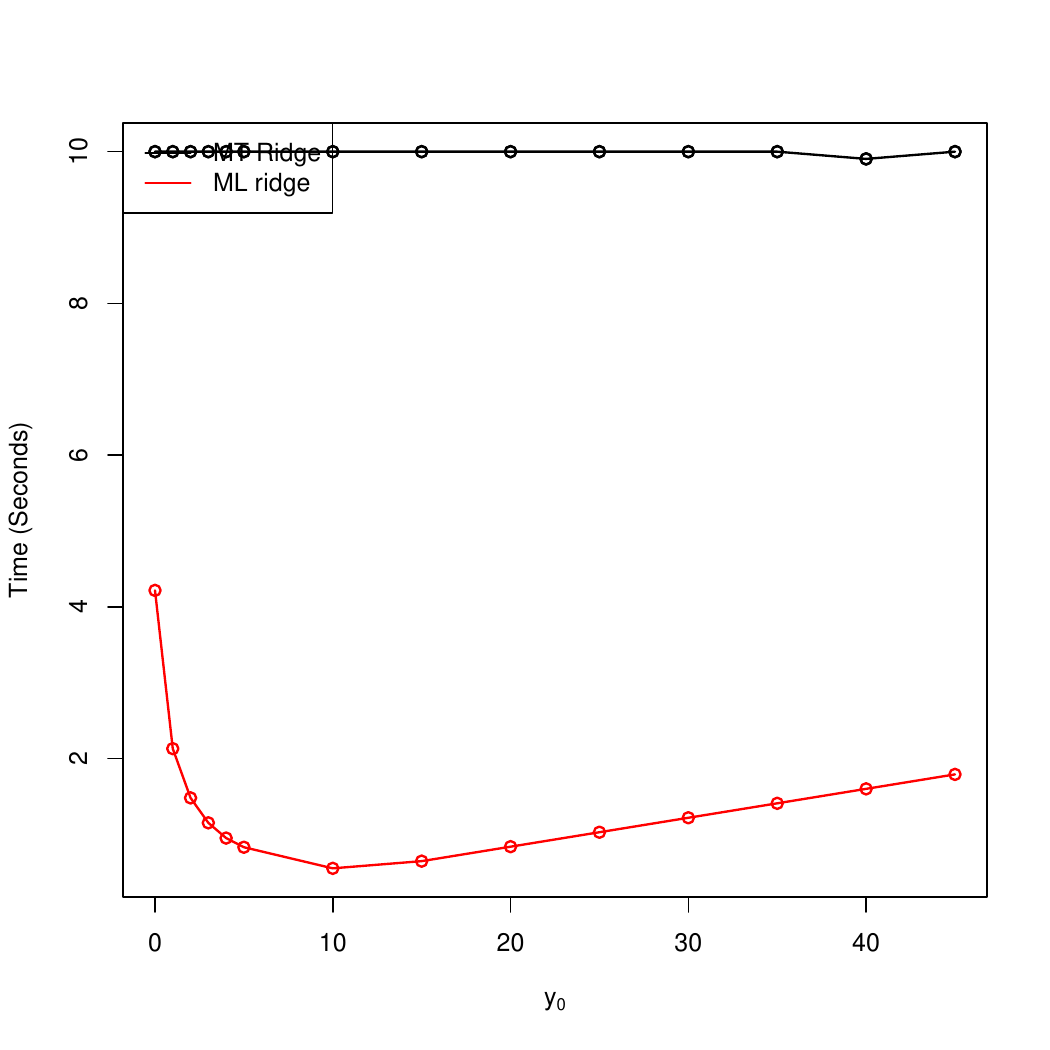} 
			\\
			$n=1000,\, \epsilon=0.05$ & $n=1000,\, \epsilon=0.1$  & $n=1000,\, \epsilon=0.15$ \\
			\includegraphics[scale=0.3]{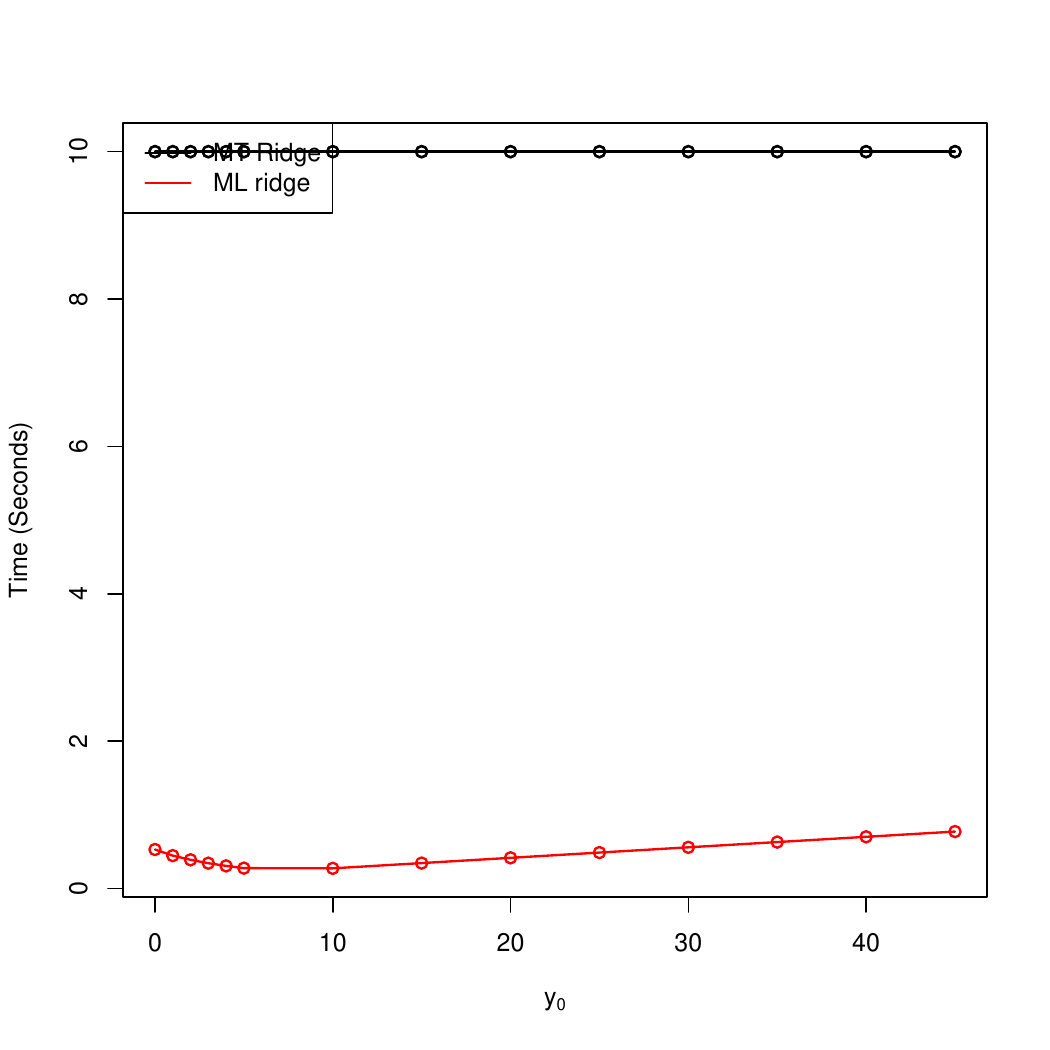}
			&   \includegraphics[scale=0.3]{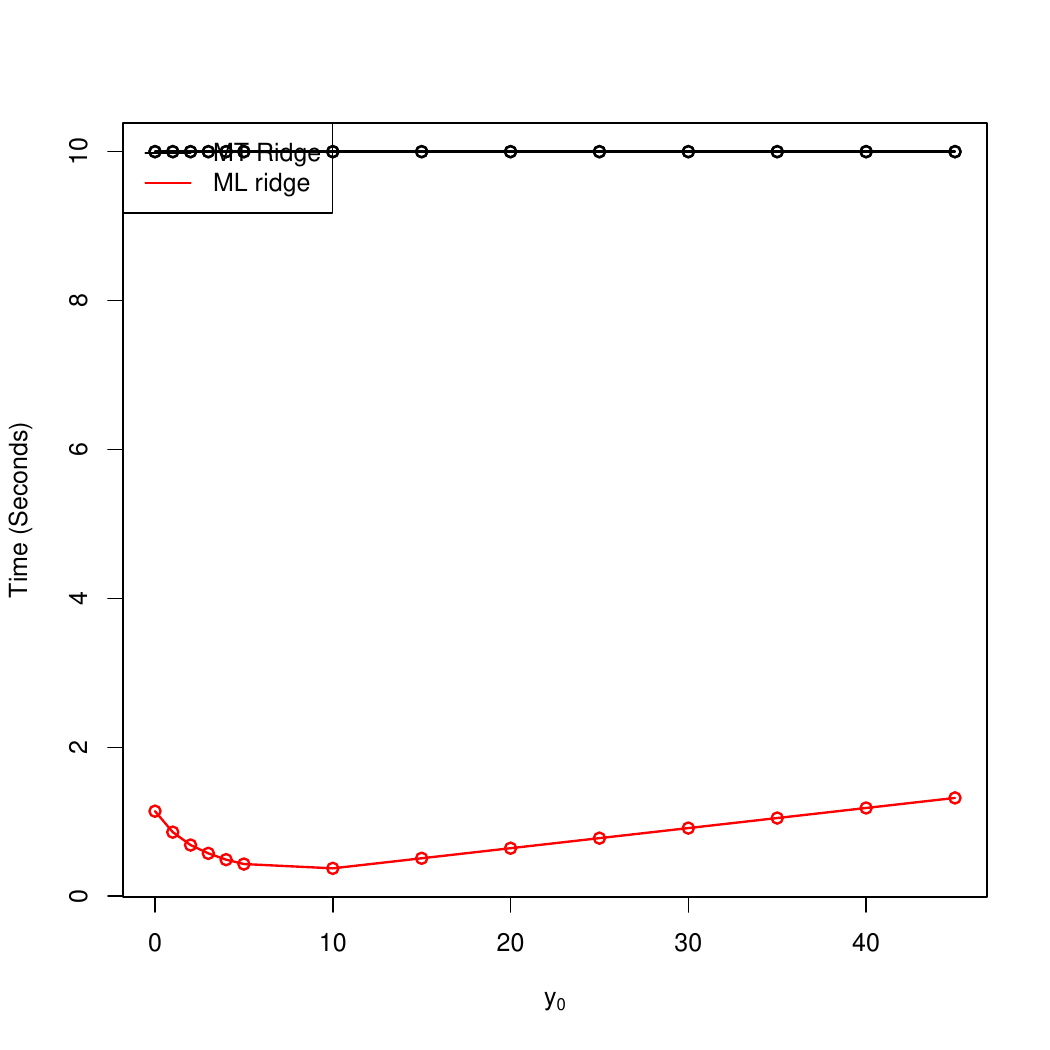} 
			&   \includegraphics[scale=0.3]{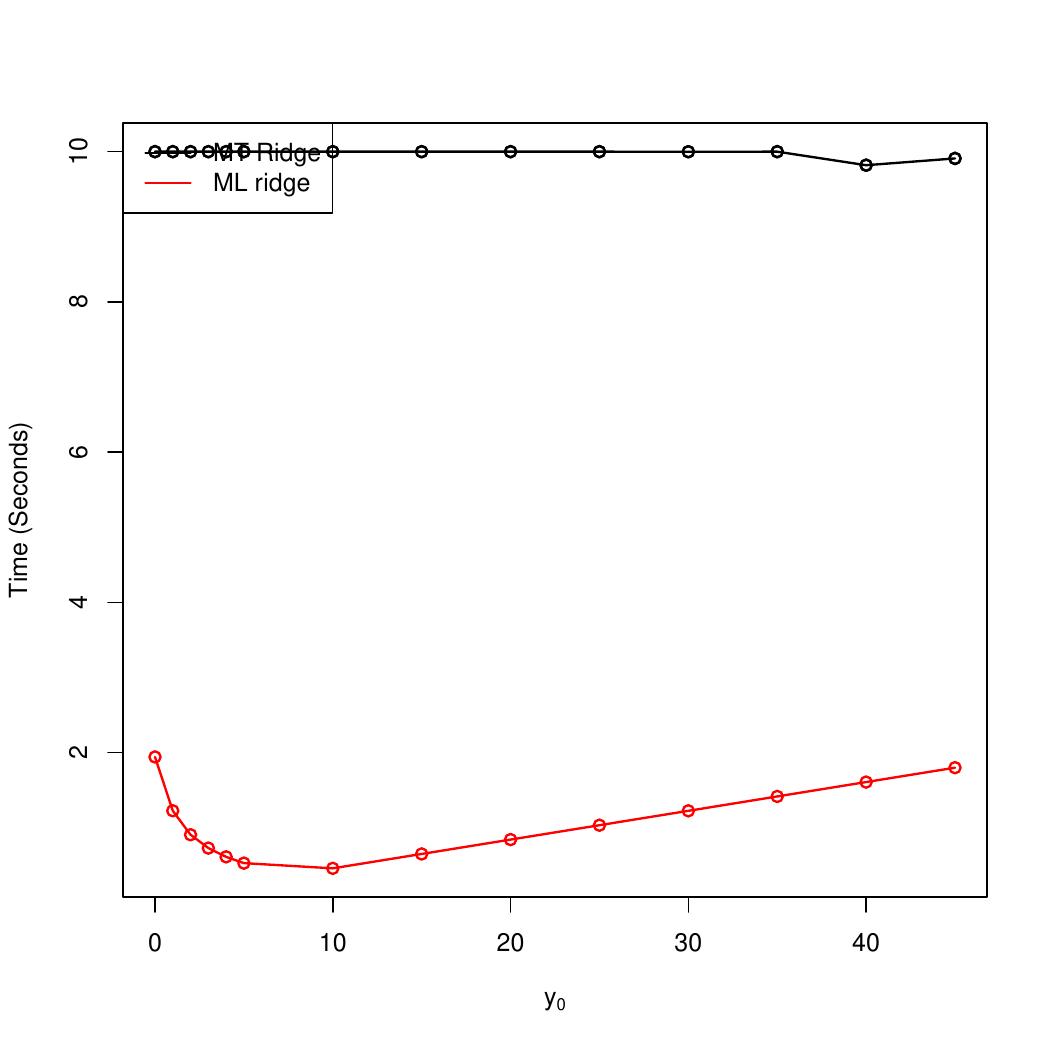}
		\end{tabular}
		\caption{\small \label{fig:timesimu1p50Ridge}  Computation times for model AVY and Ridge methods, $p=50$.}
	\end{center} 
\end{figure}

%%%% Ridge MODEL 2 %%%%%%%%%%%%%%%%%%%%%%%%%%%%%%%%%%%%%%%%%%%%%	 

\begin{figure}[tp]
	\begin{center}
		\footnotesize
		\renewcommand{\arraystretch}{0.22}
		\newcolumntype{M}{>{\centering\arraybackslash}m{\dimexpr.01\linewidth-1\tabcolsep}}
		\newcolumntype{G}{>{\centering\arraybackslash}m{\dimexpr.38\linewidth-1\tabcolsep}}
		\begin{tabular}{GGG}
			$n=100,\, \epsilon=0.05$ & $n=100,\, \epsilon=0.1$  & $n=100,\, \epsilon=0.15$ \\
			\includegraphics[scale=0.29]{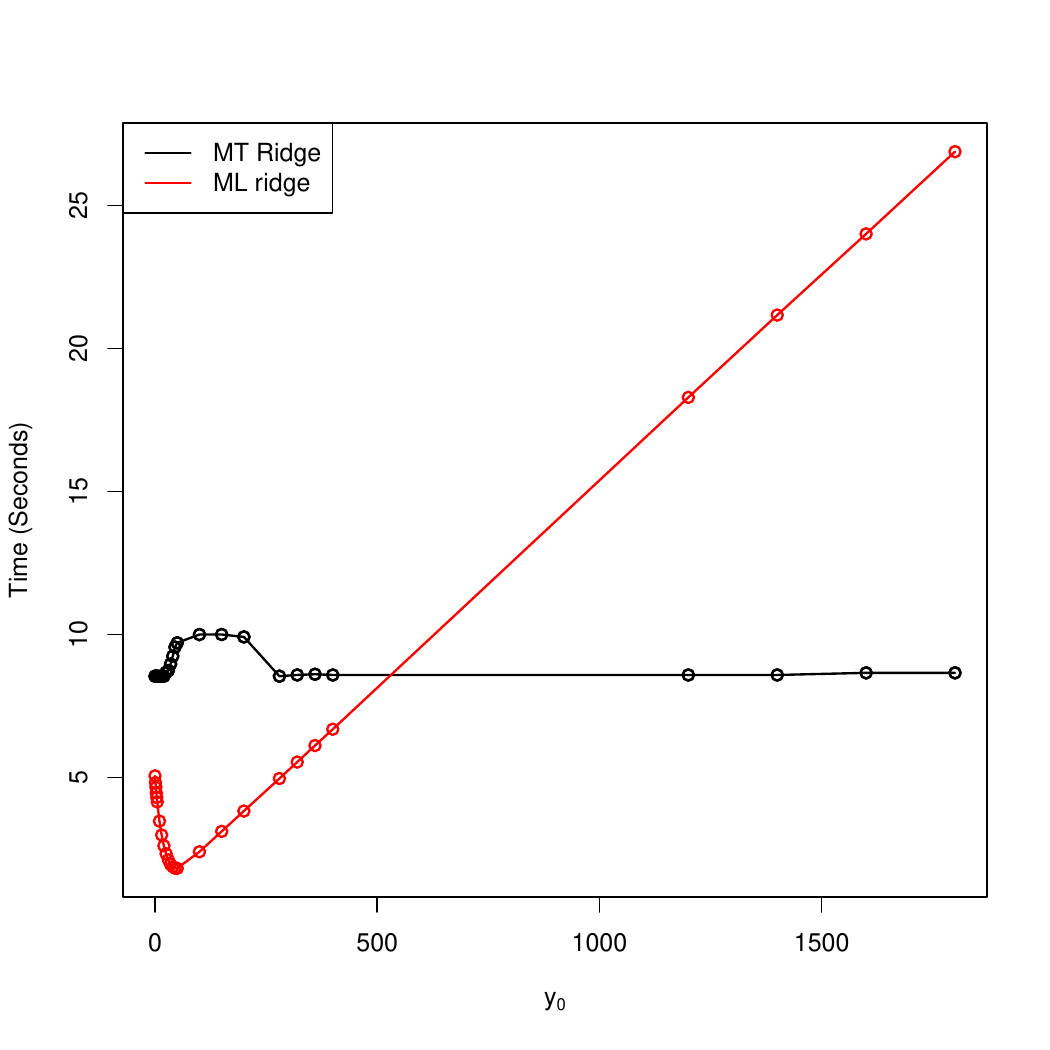} 
			&  \includegraphics[scale=0.29]{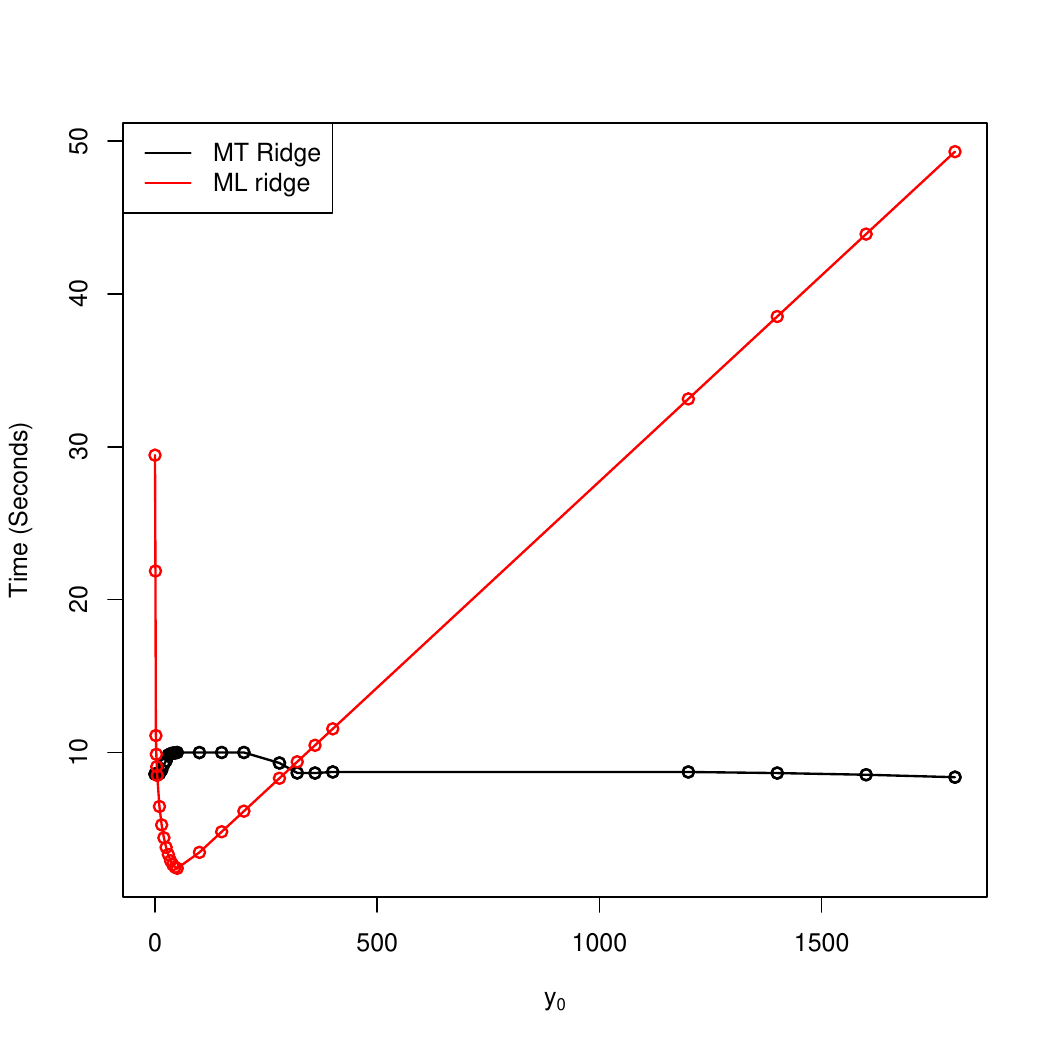}
			&   \includegraphics[scale=0.29]{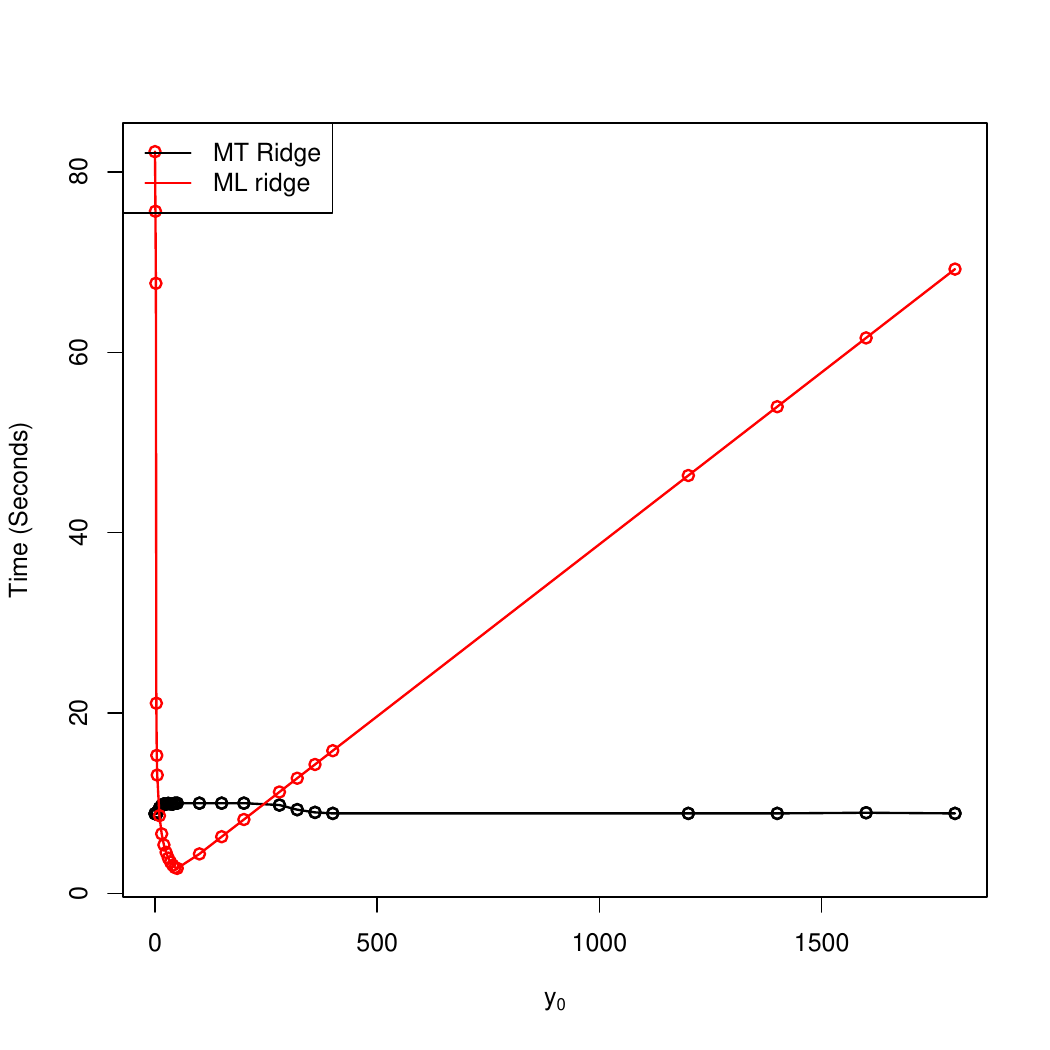}
			\\
			$n=400,\, \epsilon=0.05$ & $n=400,\, \epsilon=0.1$  & $n=400,\, \epsilon=0.15$ \\
			\includegraphics[scale=0.3]{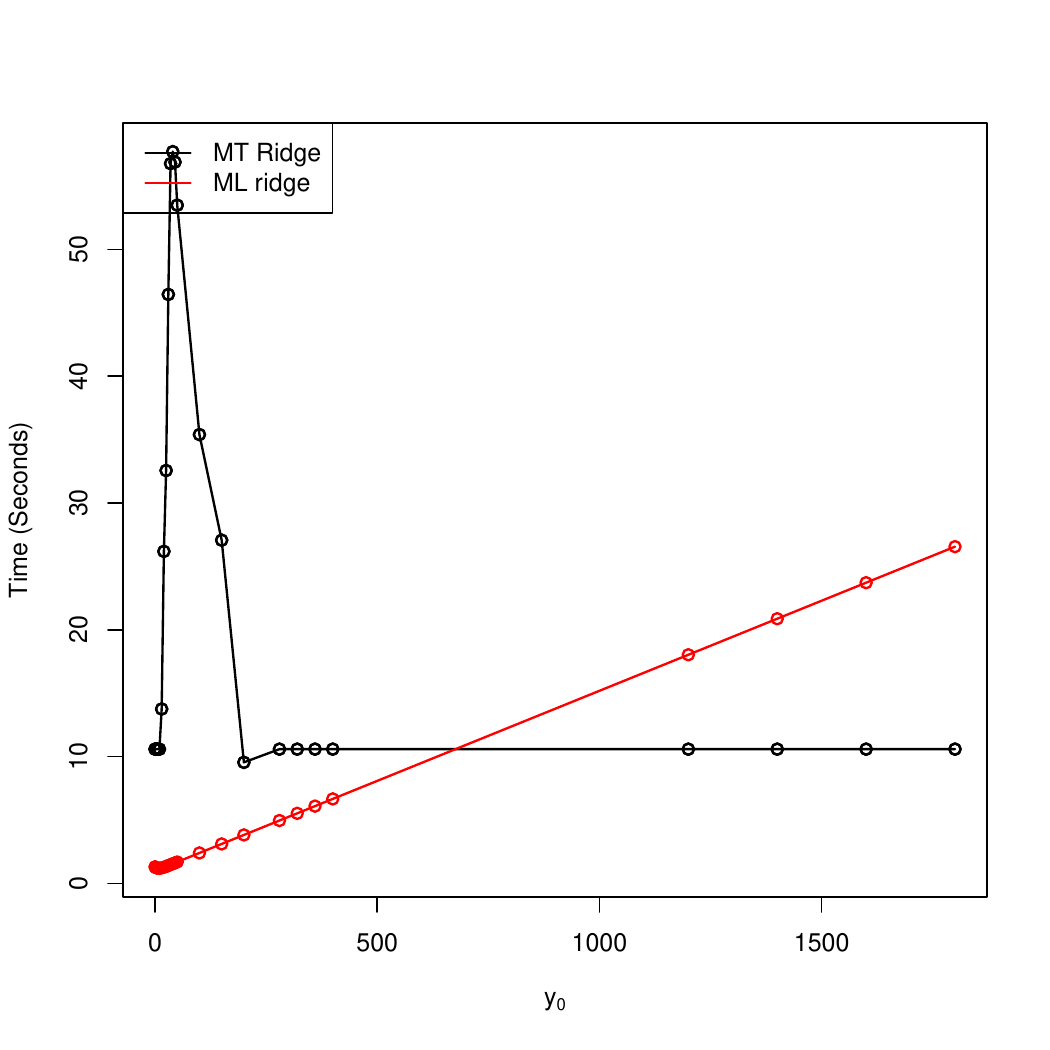}
			&   \includegraphics[scale=0.3]{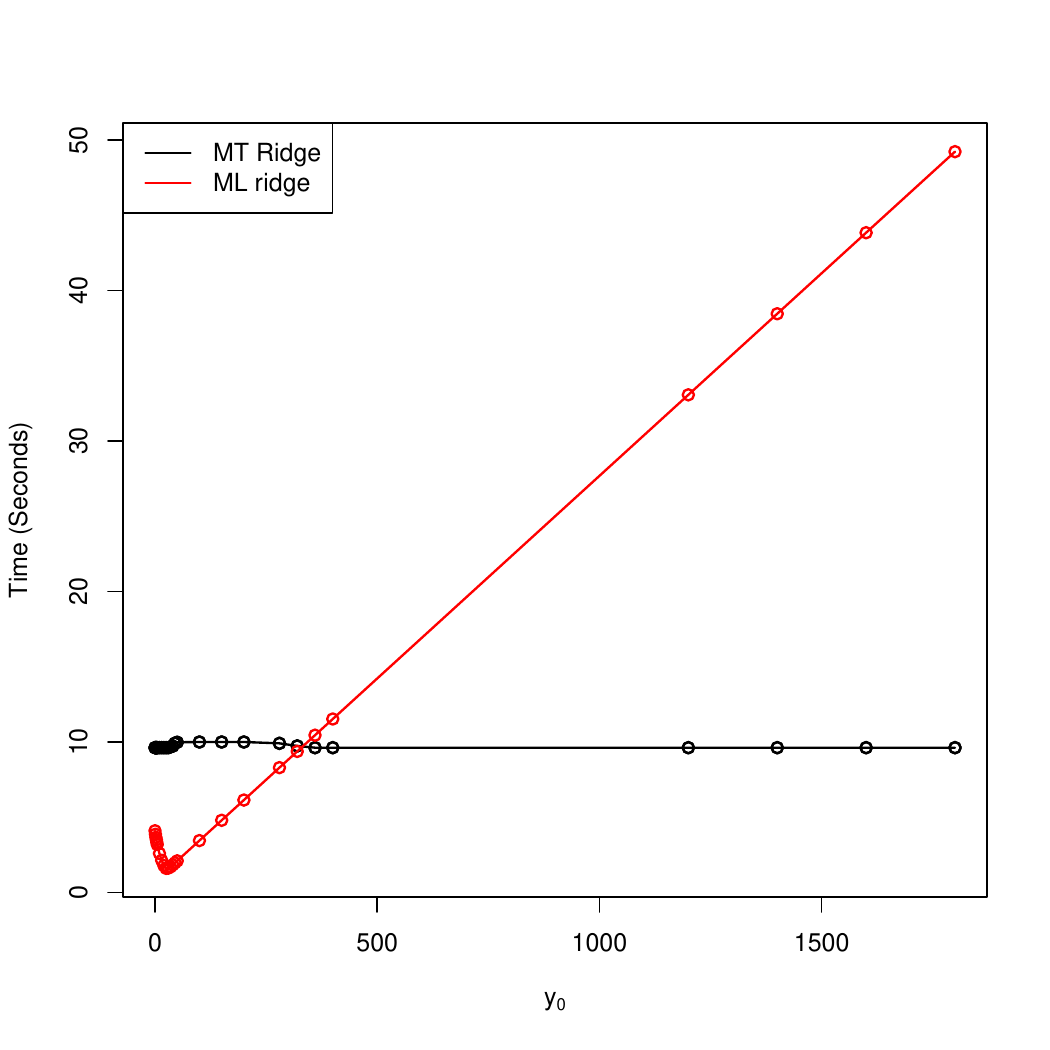} 
			&   \includegraphics[scale=0.3]{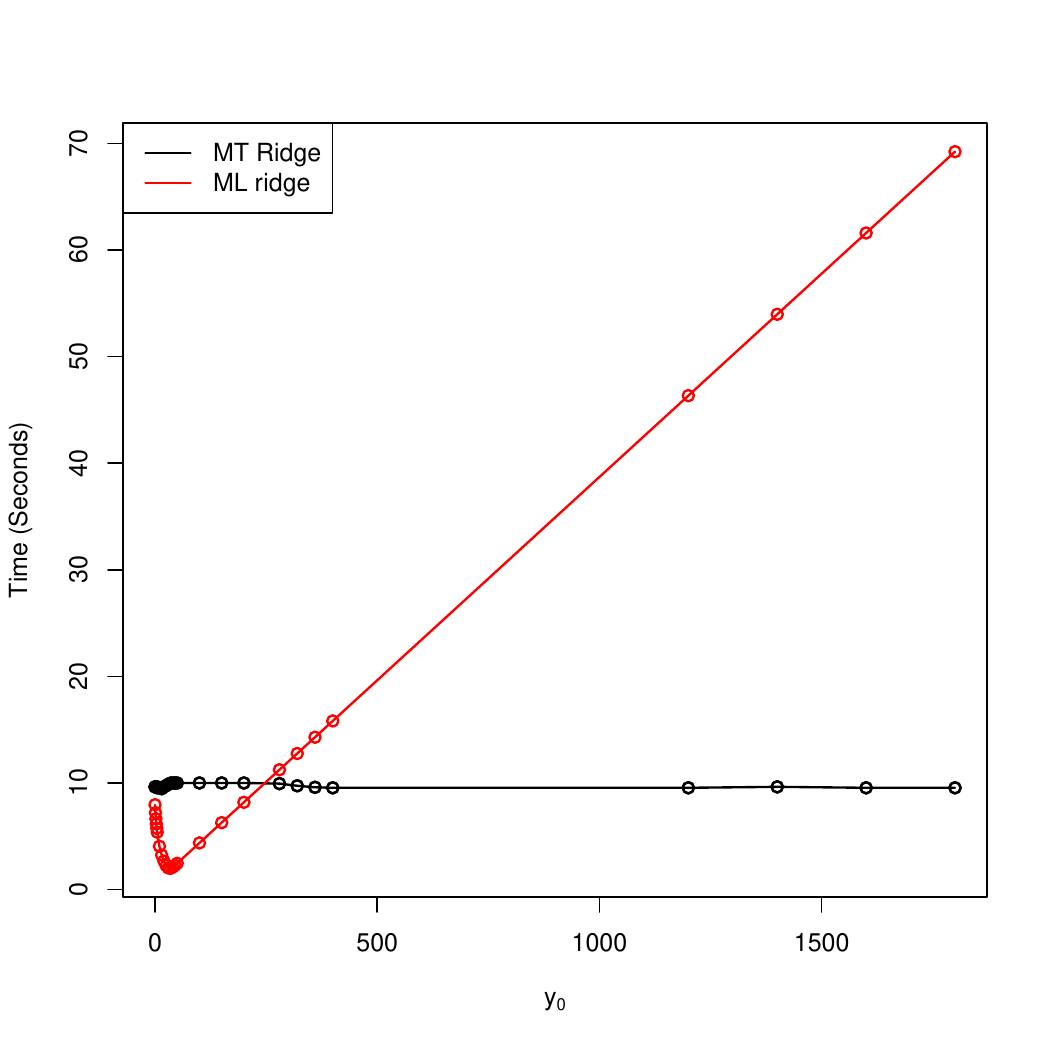} 
			\\
			$n=1000,\, \epsilon=0.05$ & $n=1000,\, \epsilon=0.1$  & $n=1000,\, \epsilon=0.15$ \\
			\includegraphics[scale=0.3]{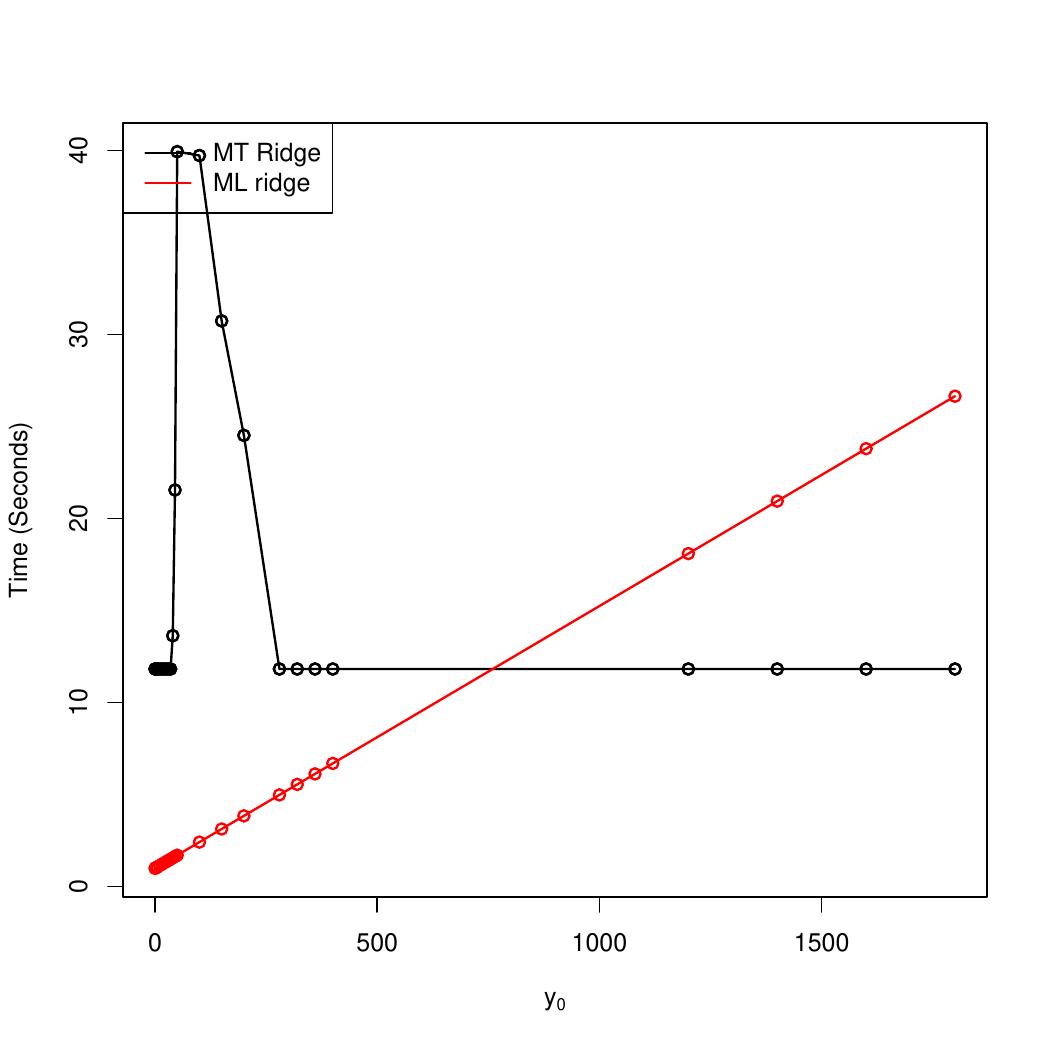}
			&   \includegraphics[scale=0.3]{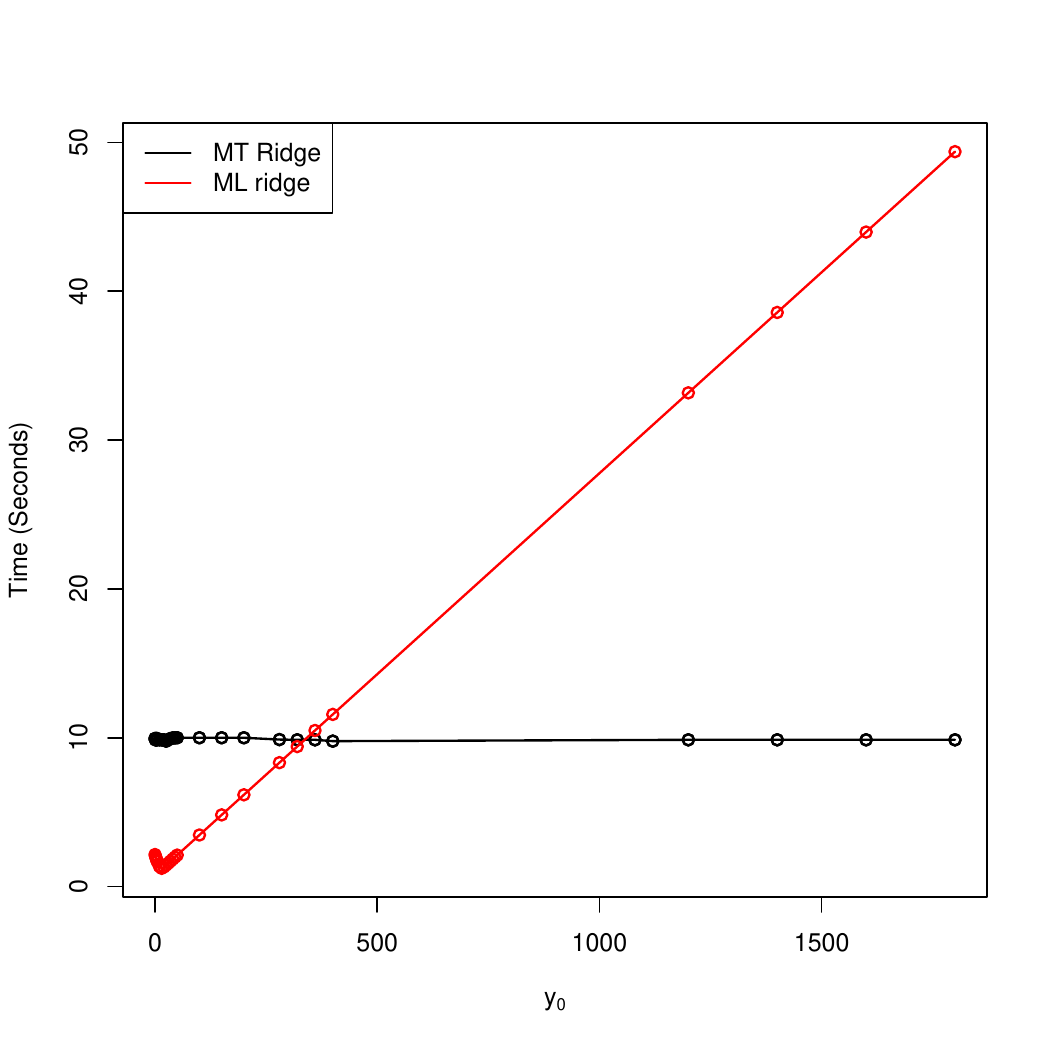} 
			&   \includegraphics[scale=0.3]{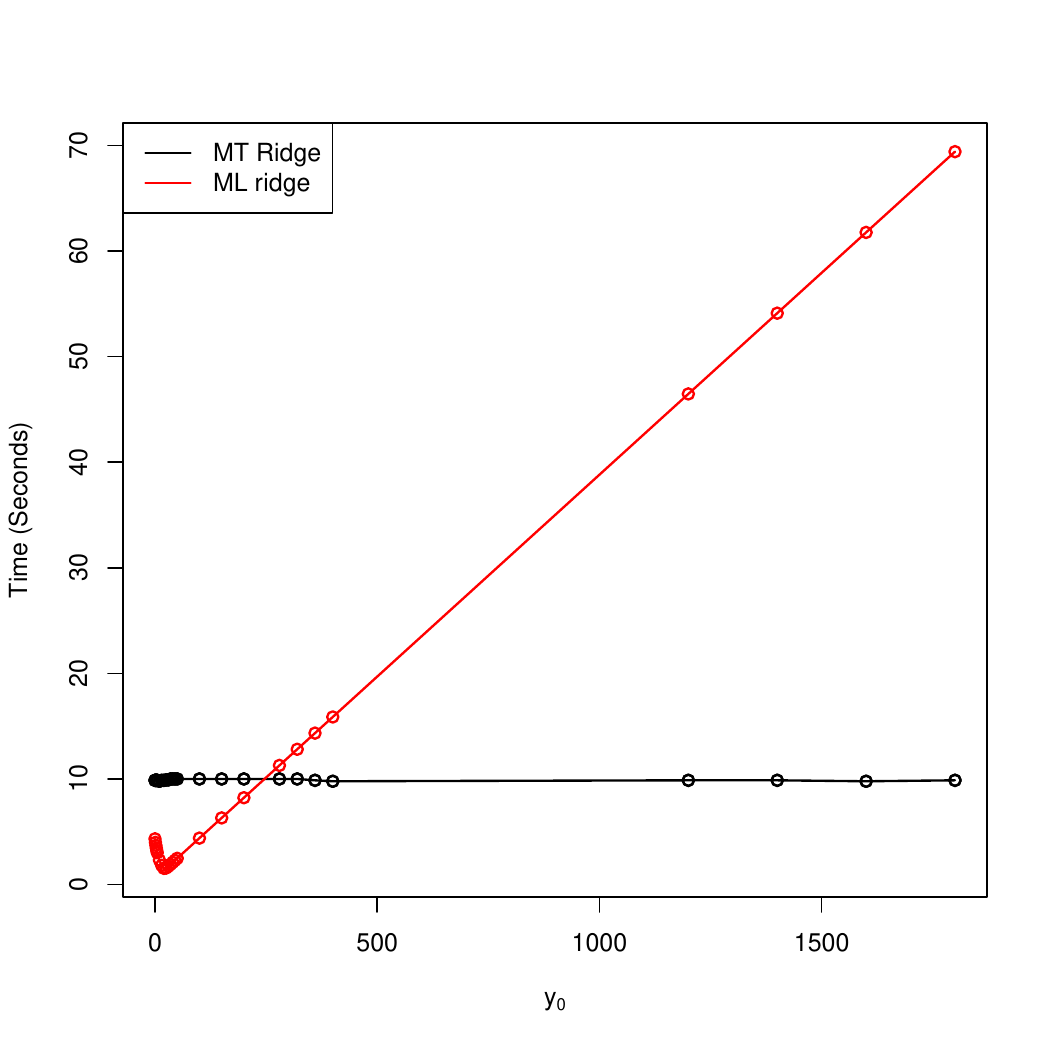}
		\end{tabular}
		\caption{\small \label{fig:timesimu2p10Ridge}  Computation times for model AVY2 and Ridge methods, $p=10$.}
	\end{center} 
\end{figure}

\begin{figure}[tp]
	\begin{center}
		\footnotesize
		\renewcommand{\arraystretch}{0.22}
		\newcolumntype{M}{>{\centering\arraybackslash}m{\dimexpr.01\linewidth-1\tabcolsep}}
		\newcolumntype{G}{>{\centering\arraybackslash}m{\dimexpr.38\linewidth-1\tabcolsep}}
		\begin{tabular}{GGG}
			$n=100,\, \epsilon=0.05$ & $n=100,\, \epsilon=0.1$  & $n=100,\, \epsilon=0.15$ \\
			\includegraphics[scale=0.29]{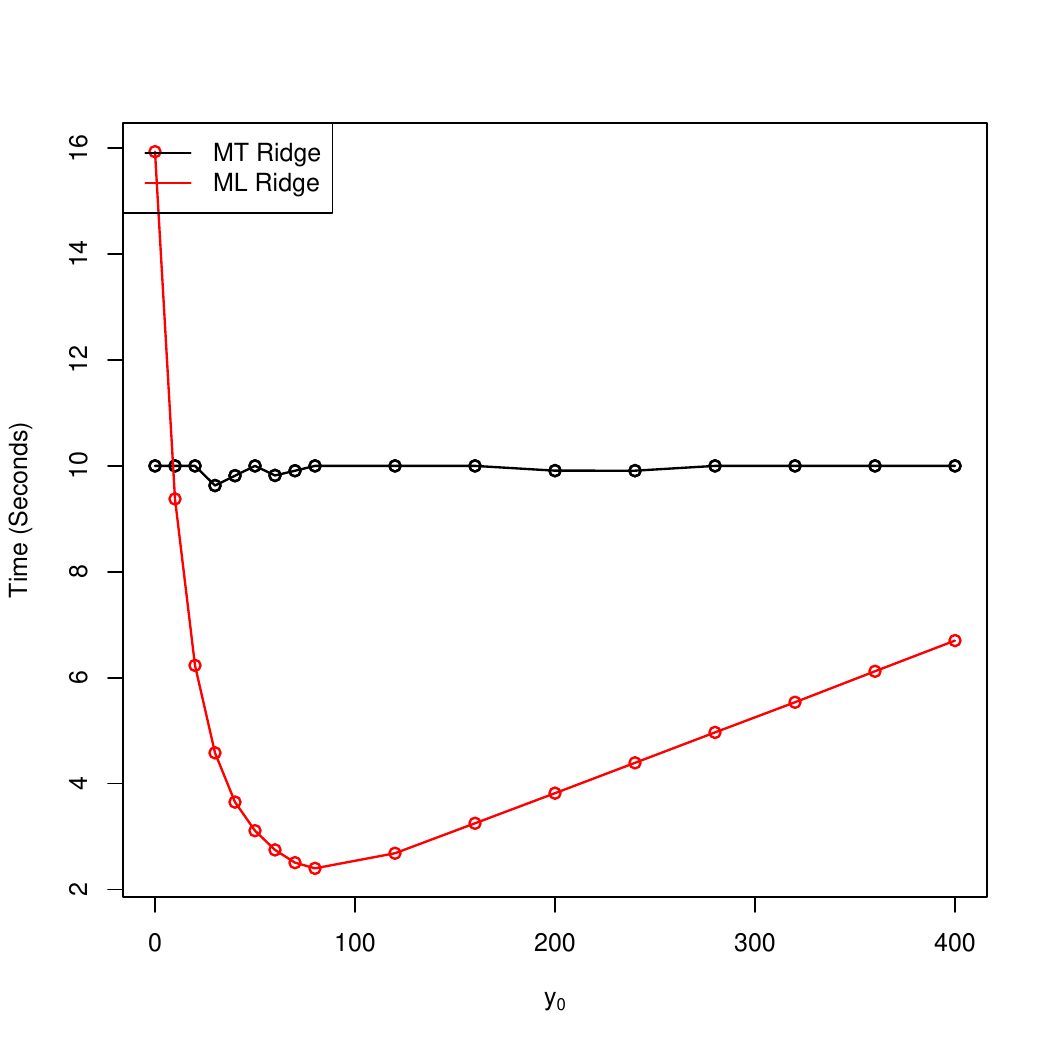} 
			&  \includegraphics[scale=0.29]{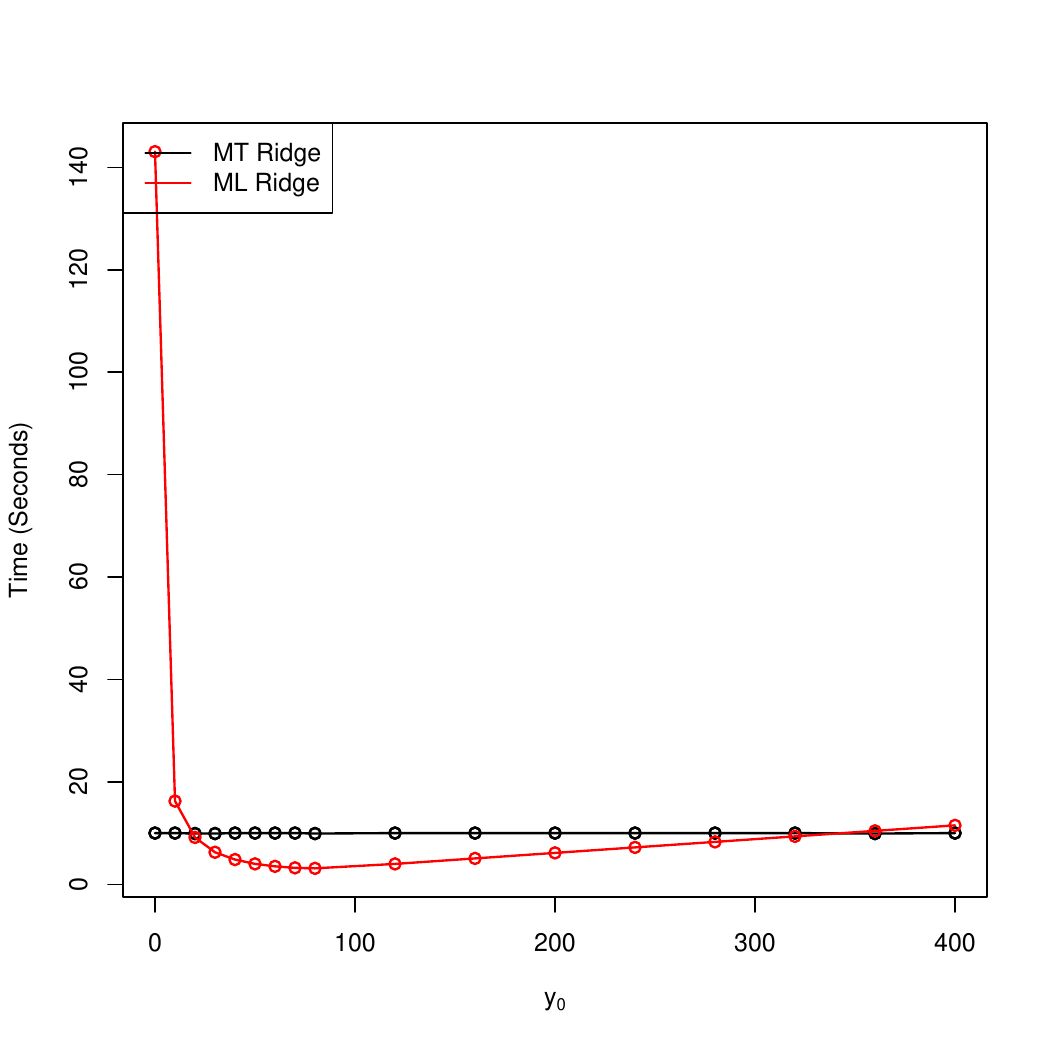}
			&   \includegraphics[scale=0.29]{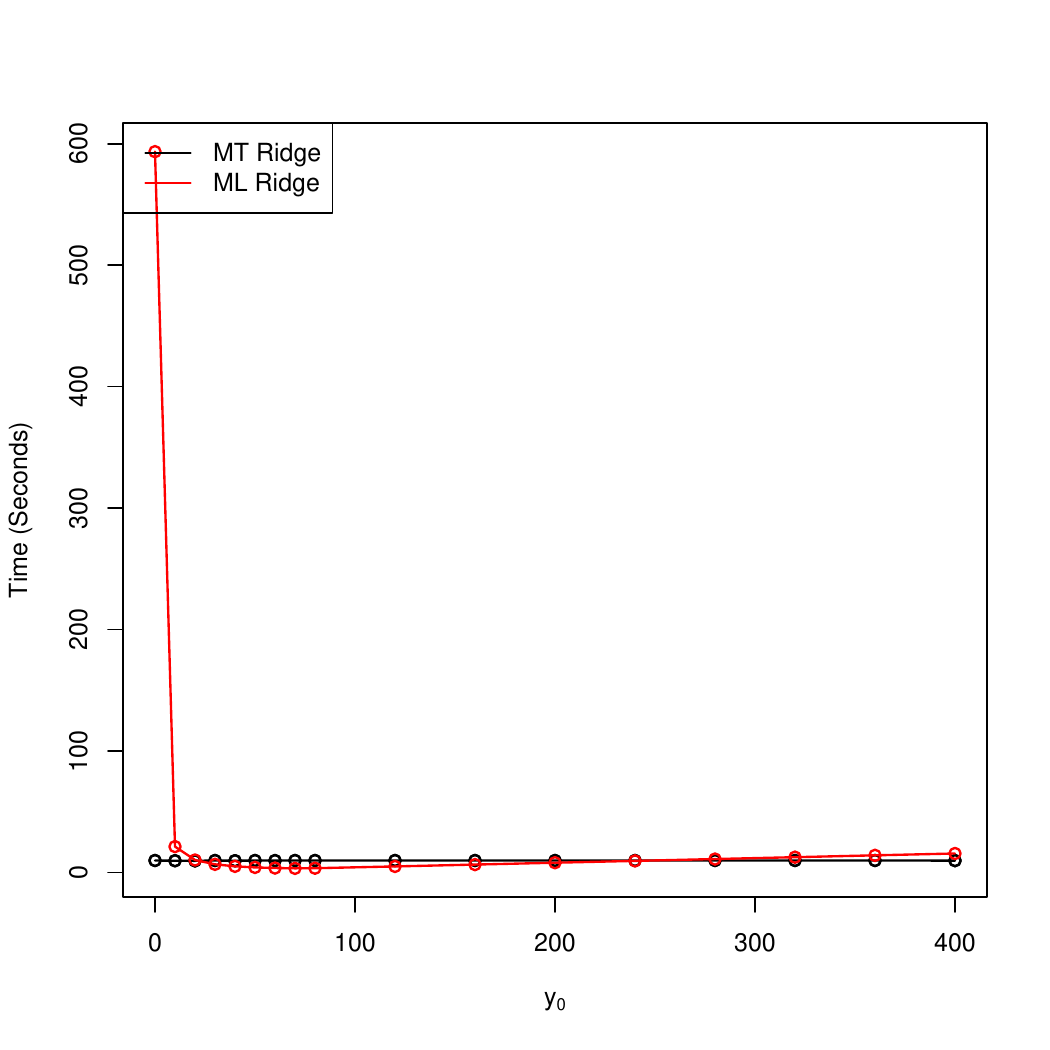}
			\\
			$n=400,\, \epsilon=0.05$ & $n=400,\, \epsilon=0.1$  & $n=400,\, \epsilon=0.15$ \\
			\includegraphics[scale=0.3]{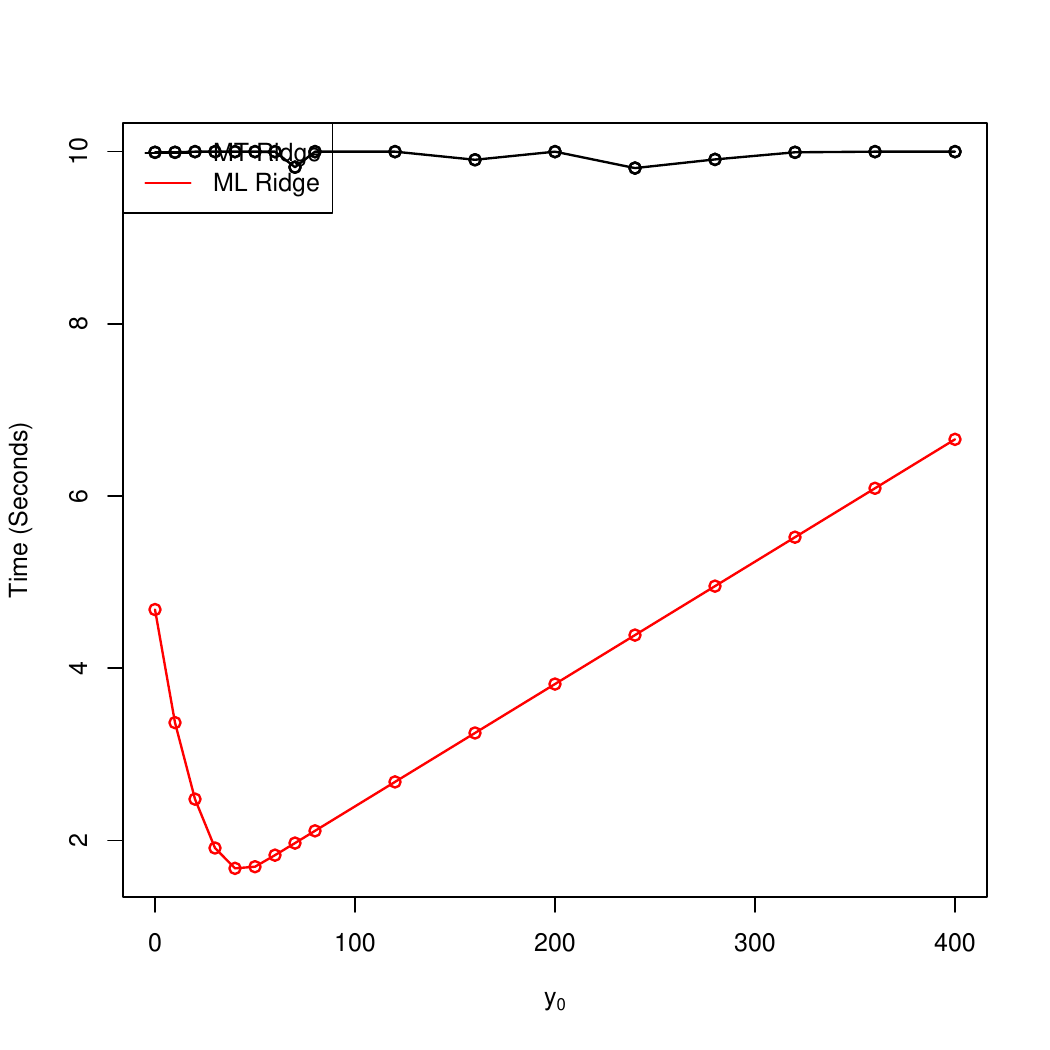}
			&   \includegraphics[scale=0.3]{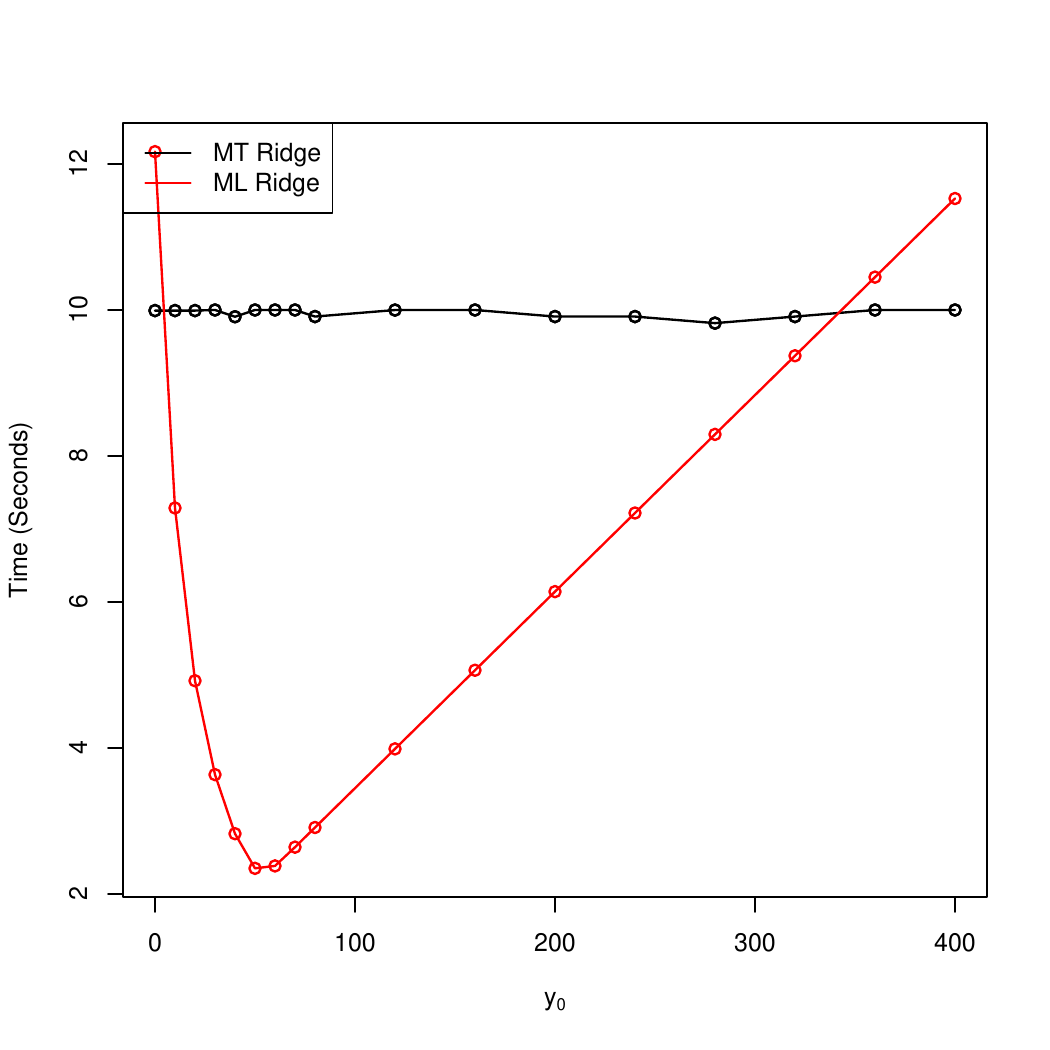} 
			&   \includegraphics[scale=0.3]{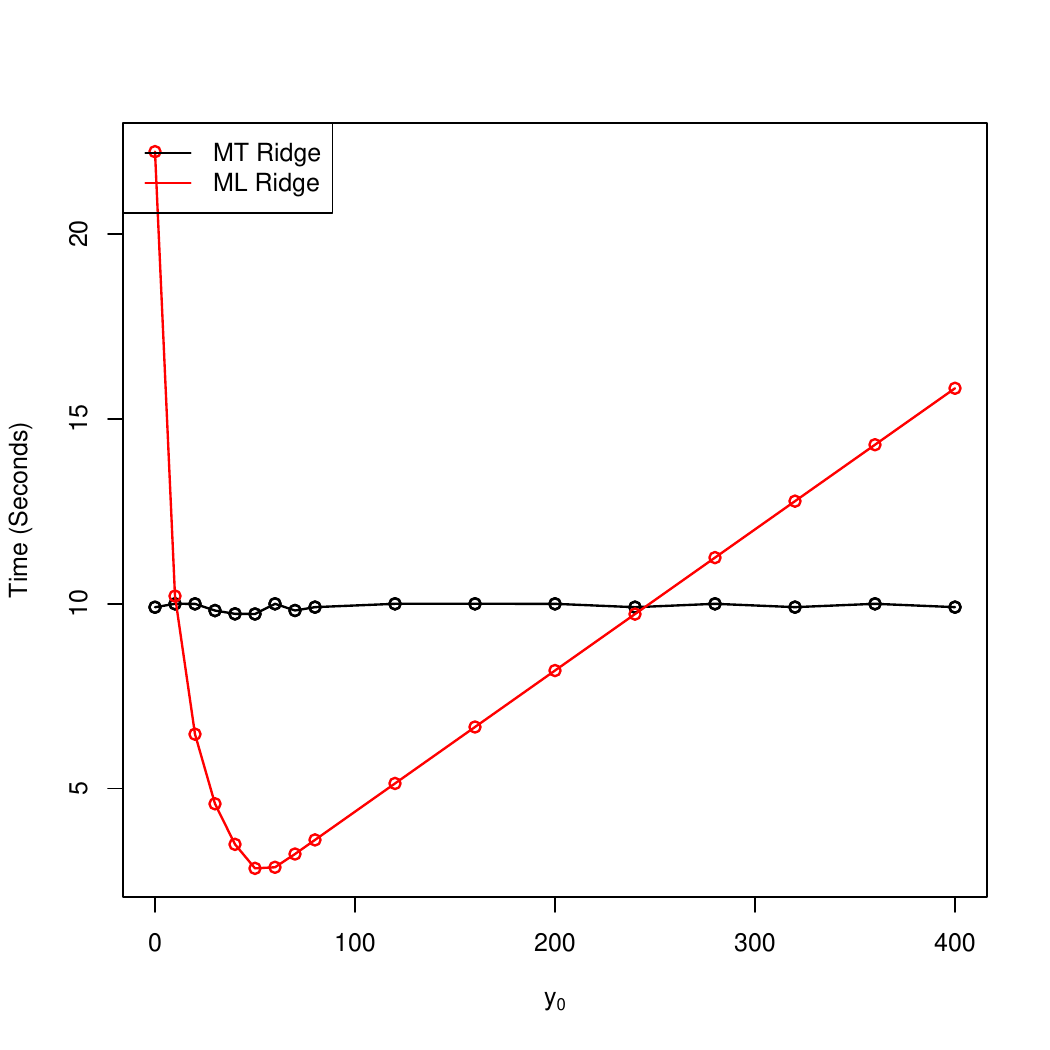} 
			\\
			$n=1000,\, \epsilon=0.05$ & $n=1000,\, \epsilon=0.1$  & $n=1000,\, \epsilon=0.15$ \\
			\includegraphics[scale=0.3]{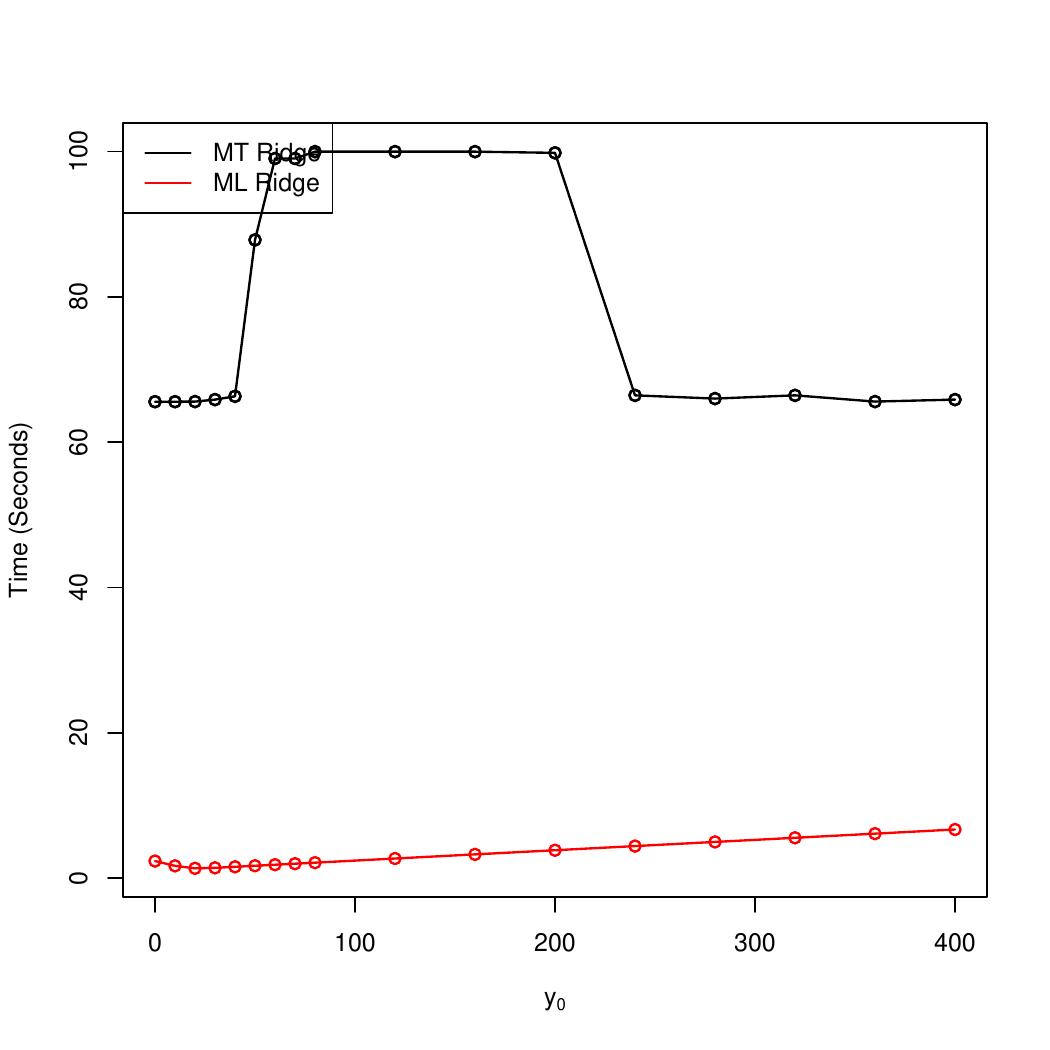}
			&   \includegraphics[scale=0.3]{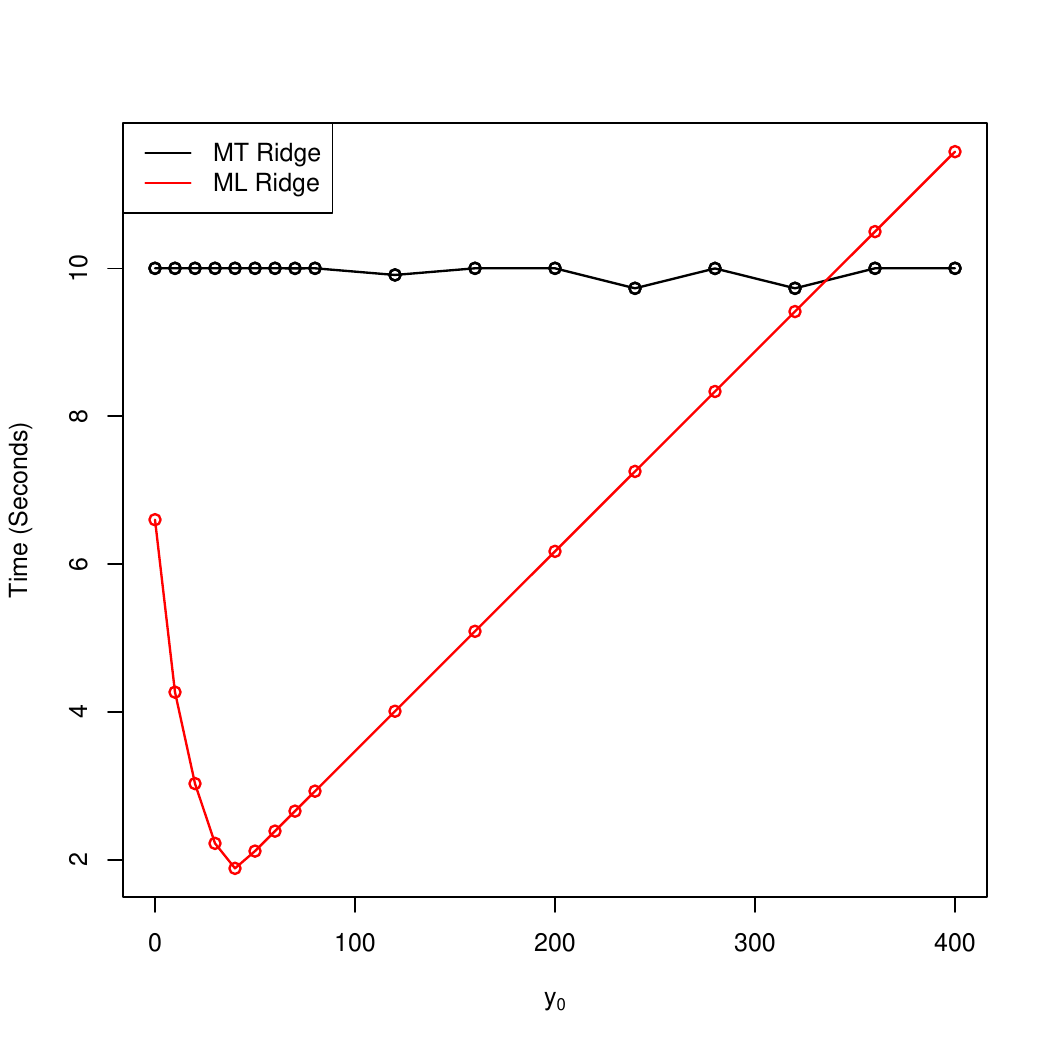} 
			&   \includegraphics[scale=0.3]{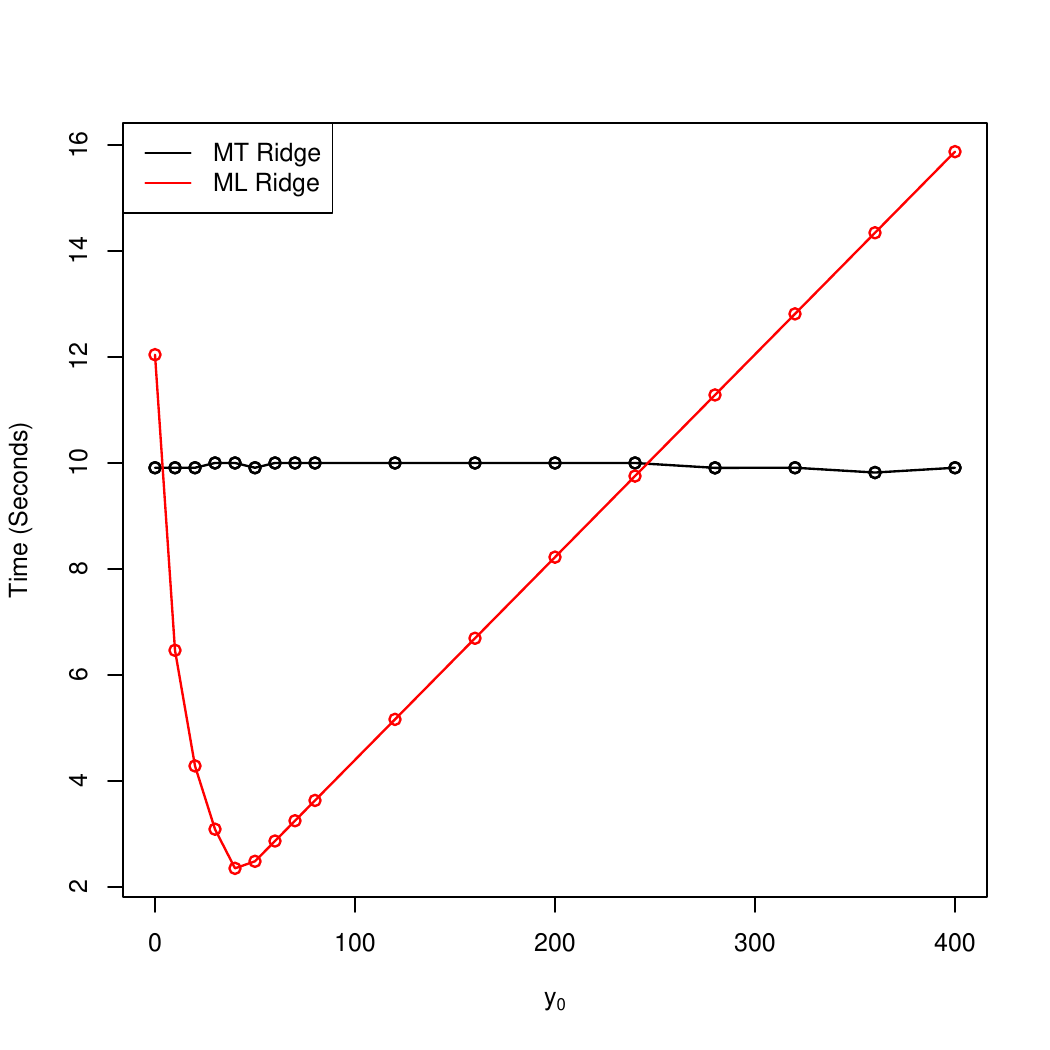}
		\end{tabular}
		\caption{\small \label{fig:timesimu2p50Ridge}  Computation times for AVY2 and Ridge methods, $p=50$.}
	\end{center} 
\end{figure}

%%%%%%%%%%%%%%%%%%%%%%%%%%% Lasso %%%%%%%%%%%%%%%%%%%%%%%%%%%%%%%%%%%%%%%%%%	 

\begin{figure}[tp]
	\begin{center}
		\footnotesize
		\renewcommand{\arraystretch}{0.22}
		\newcolumntype{M}{>{\centering\arraybackslash}m{\dimexpr.01\linewidth-1\tabcolsep}}
		\newcolumntype{G}{>{\centering\arraybackslash}m{\dimexpr.38\linewidth-1\tabcolsep}}
		\begin{tabular}{GGG}
			$n=100,\, \epsilon=0.05$ & $n=100,\, \epsilon=0.1$  & $n=100,\, \epsilon=0.15$ \\
			\includegraphics[scale=0.29]{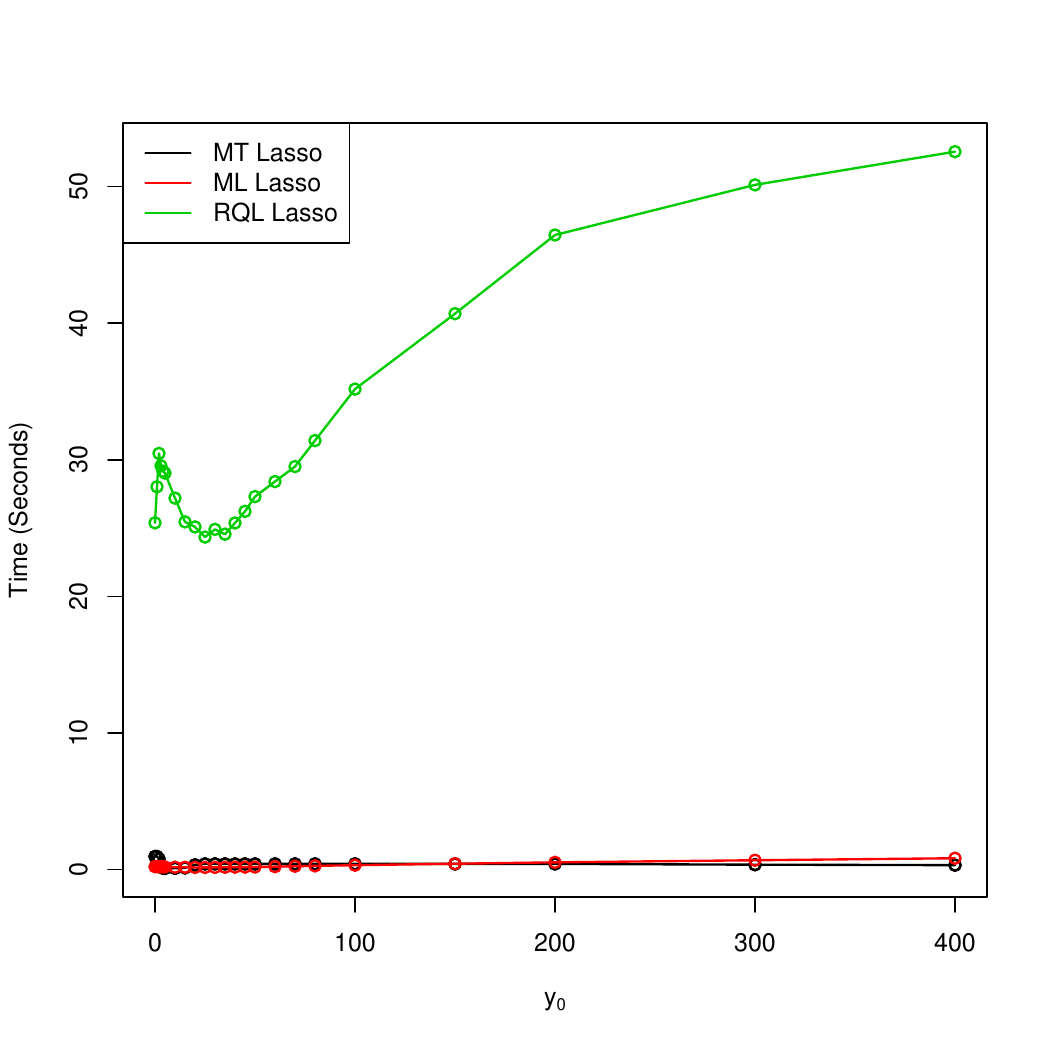} 
			&  \includegraphics[scale=0.29]{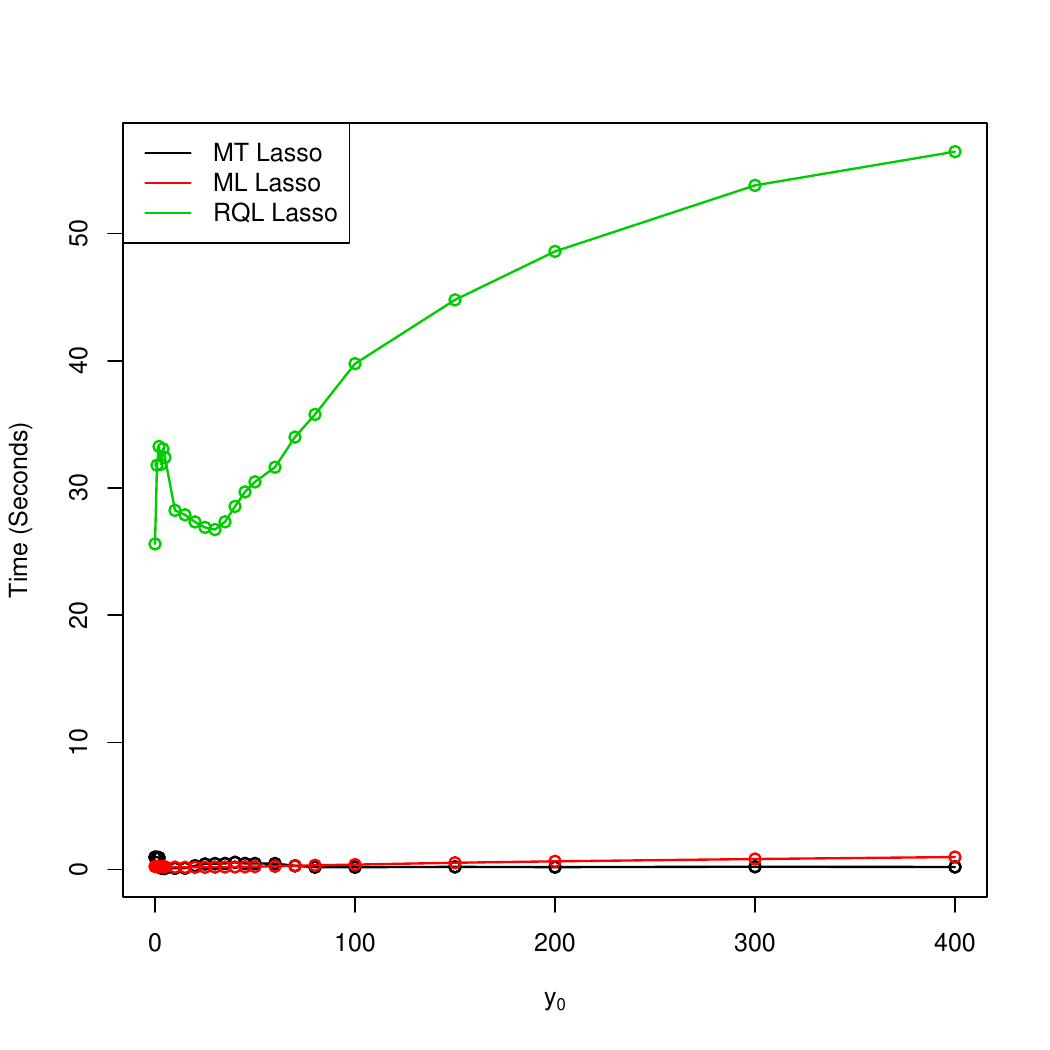}
			&   \includegraphics[scale=0.29]{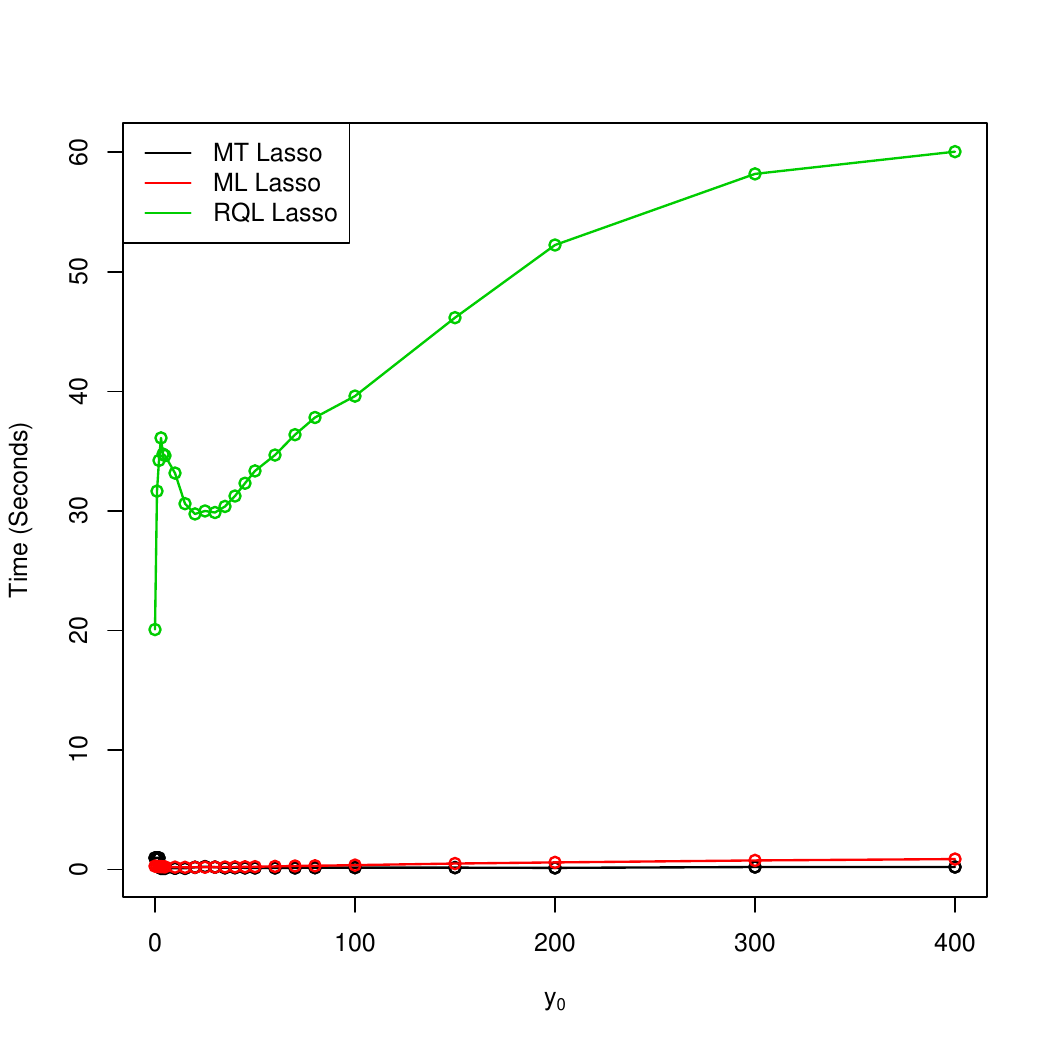}
			\\
			$n=400,\, \epsilon=0.05$ & $n=400,\, \epsilon=0.1$  & $n=400,\, \epsilon=0.15$ \\
			\includegraphics[scale=0.3]{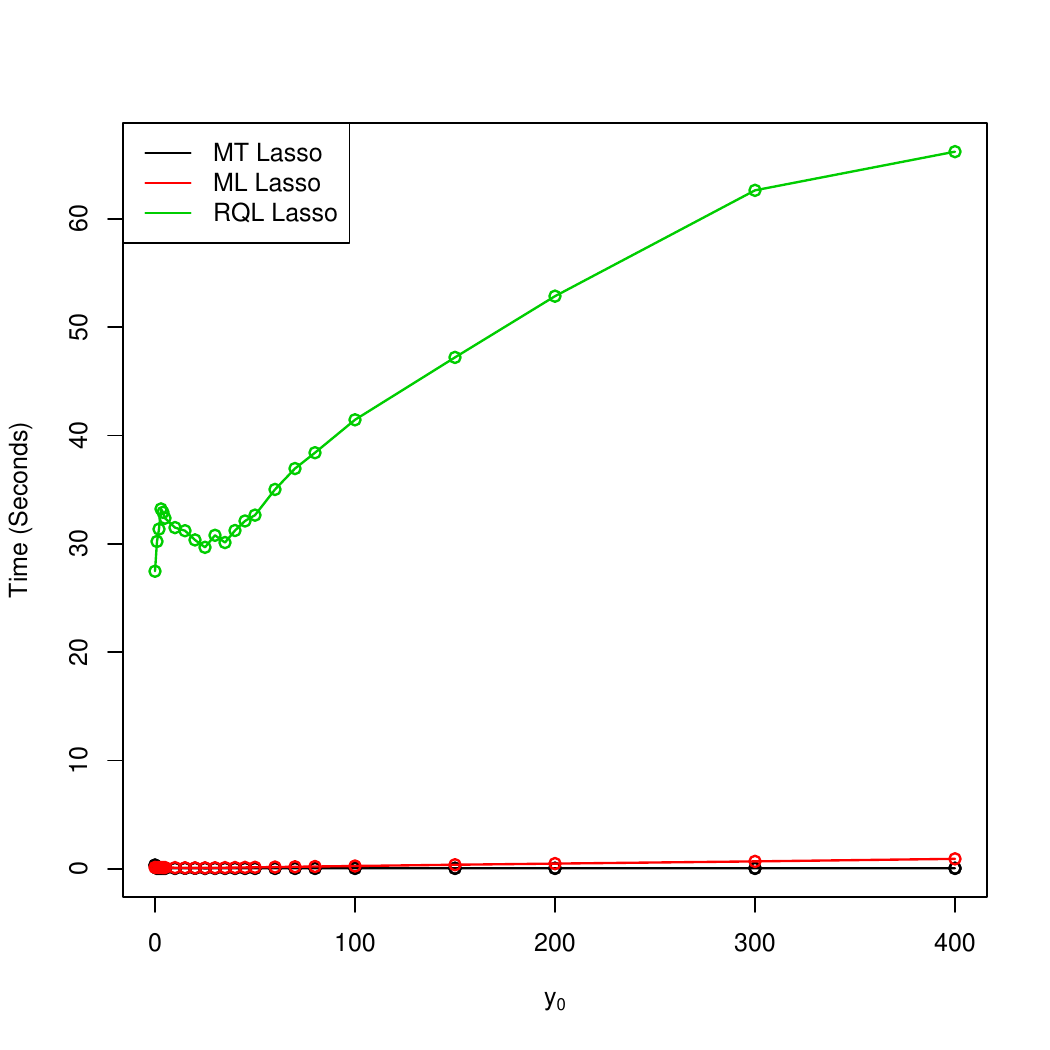}
			&   \includegraphics[scale=0.3]{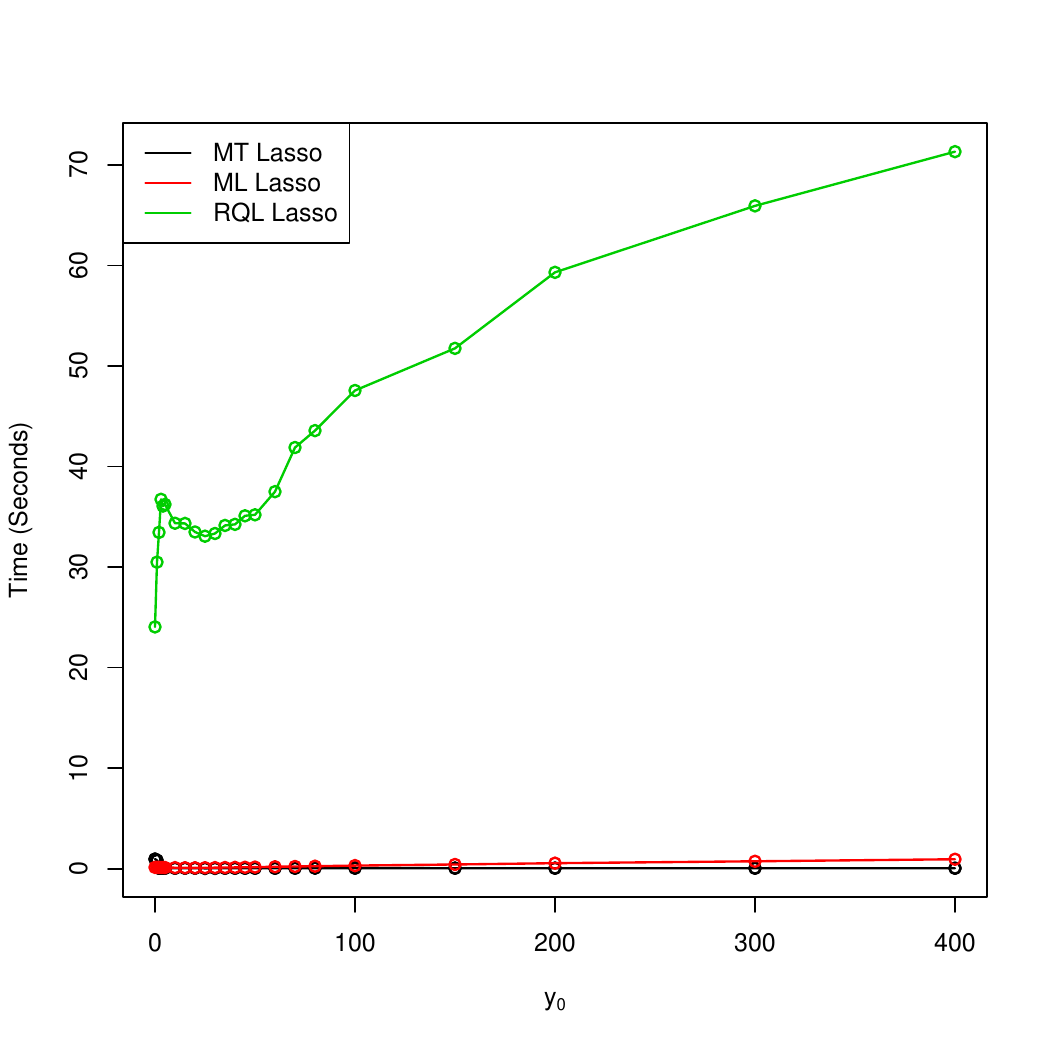} 
			&   \includegraphics[scale=0.3]{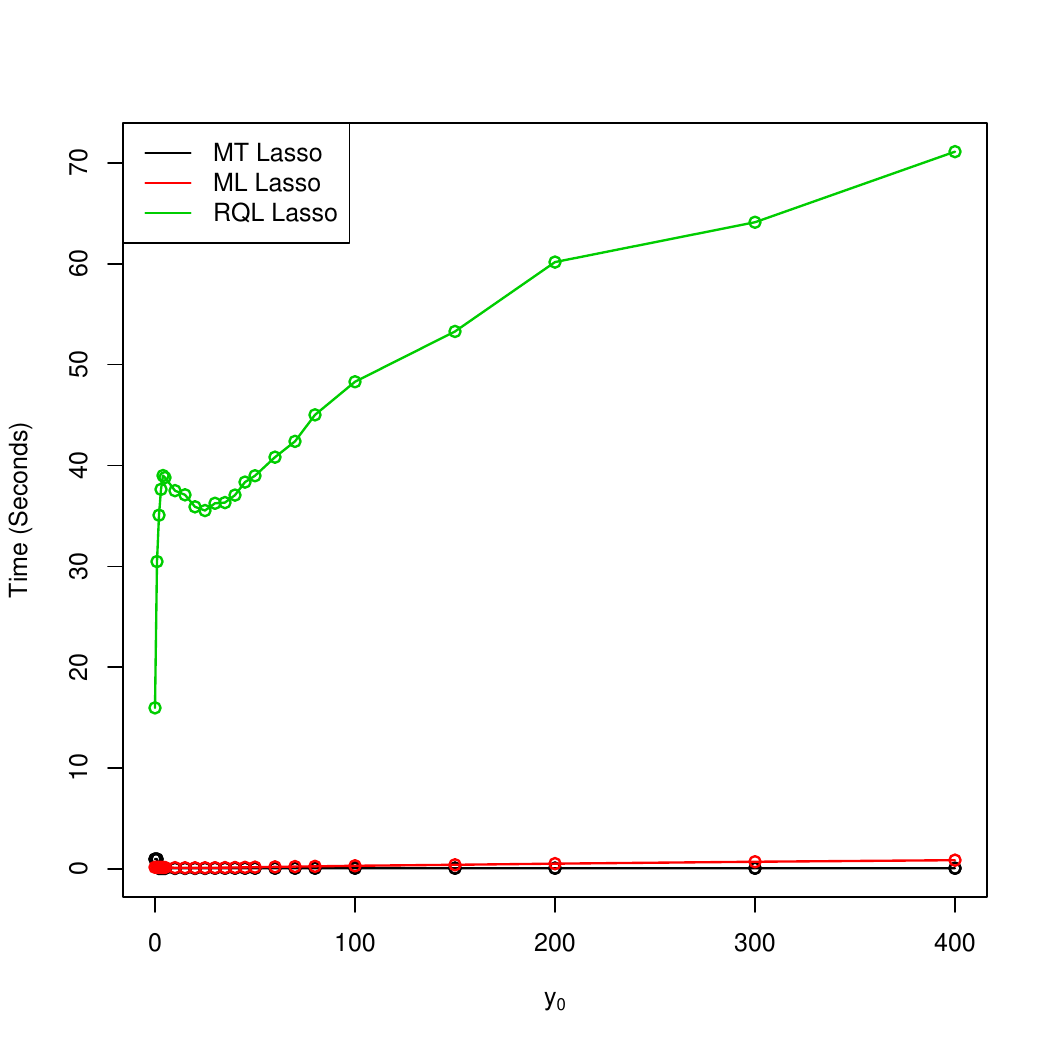} 
			\\
			$n=1000,\, \epsilon=0.05$ & $n=1000,\, \epsilon=0.1$  & $n=1000,\, \epsilon=0.15$ \\
			\includegraphics[scale=0.3]{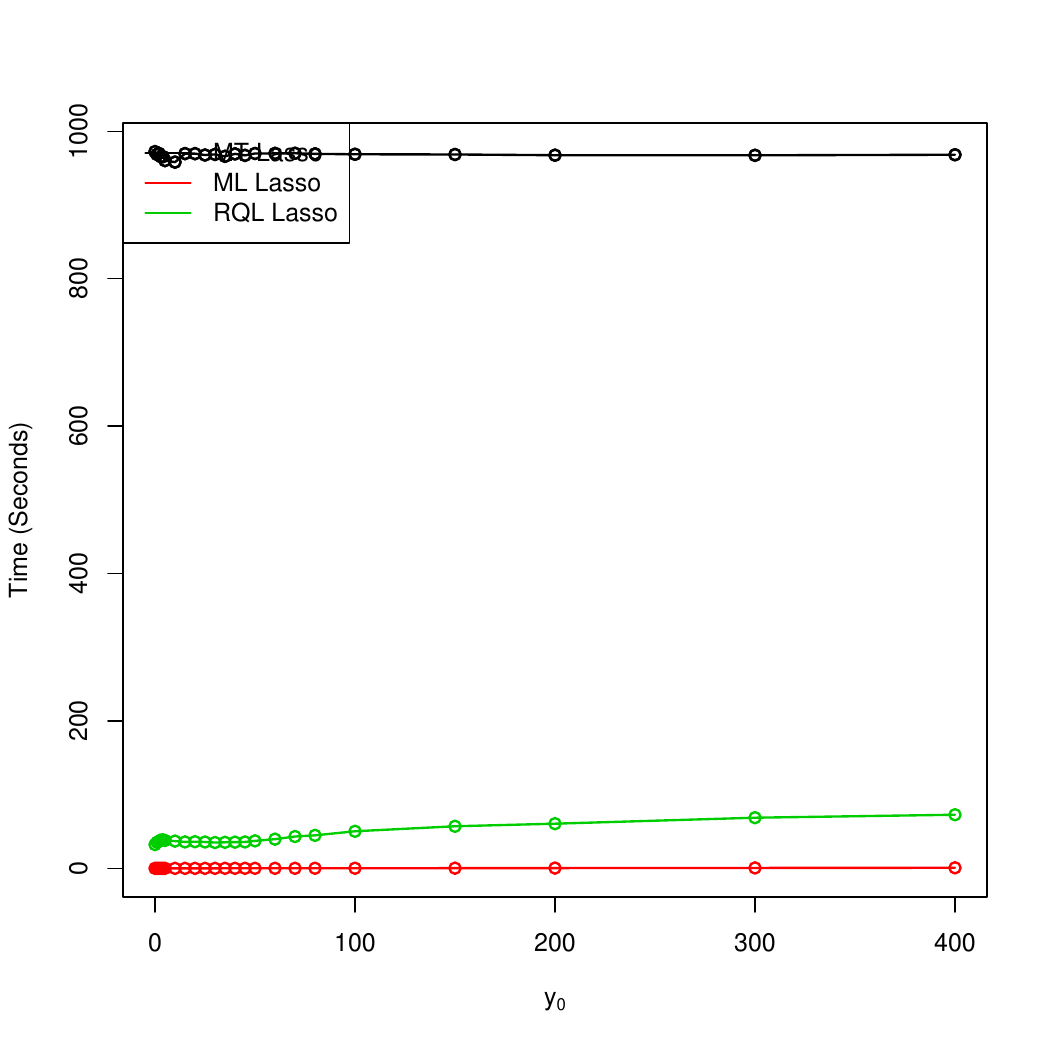}
			&   \includegraphics[scale=0.3]{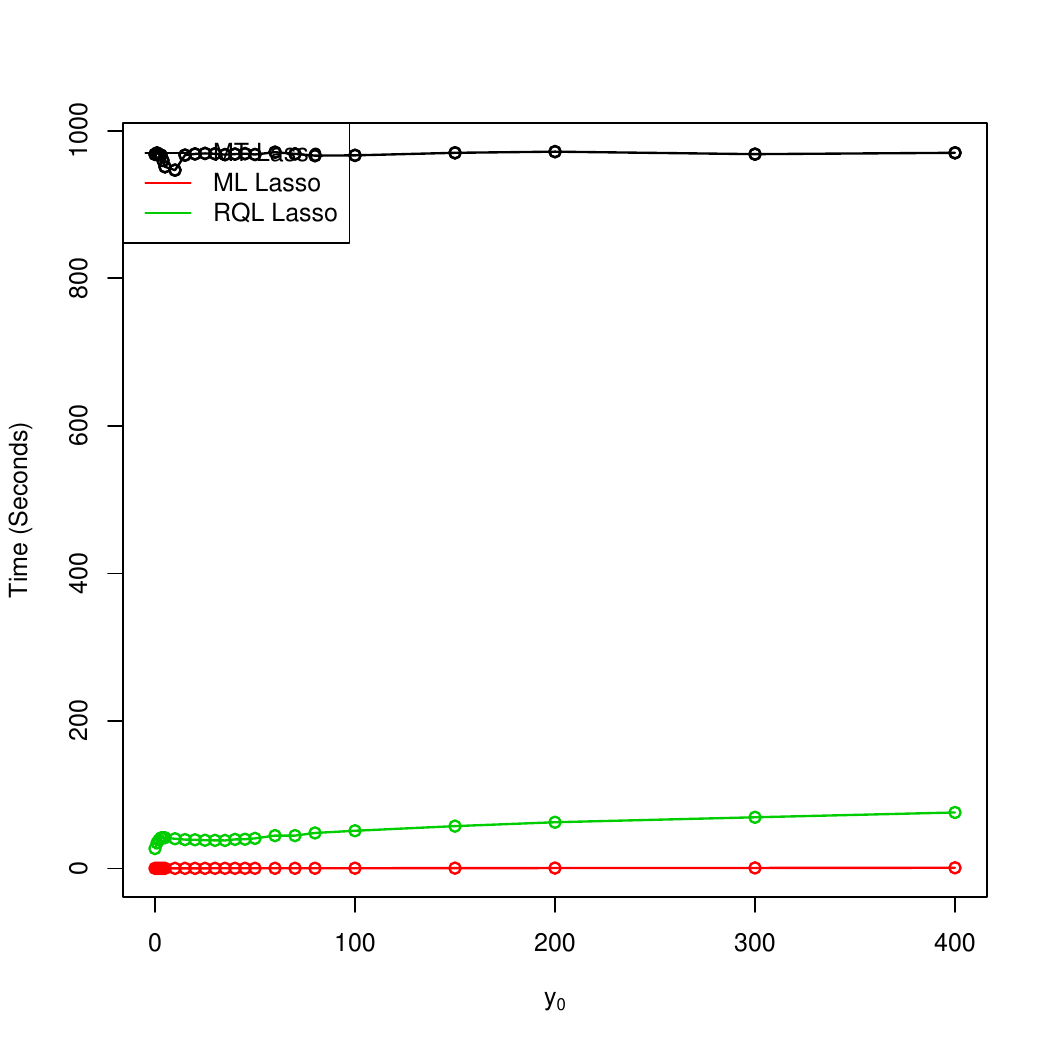} 
			&   \includegraphics[scale=0.3]{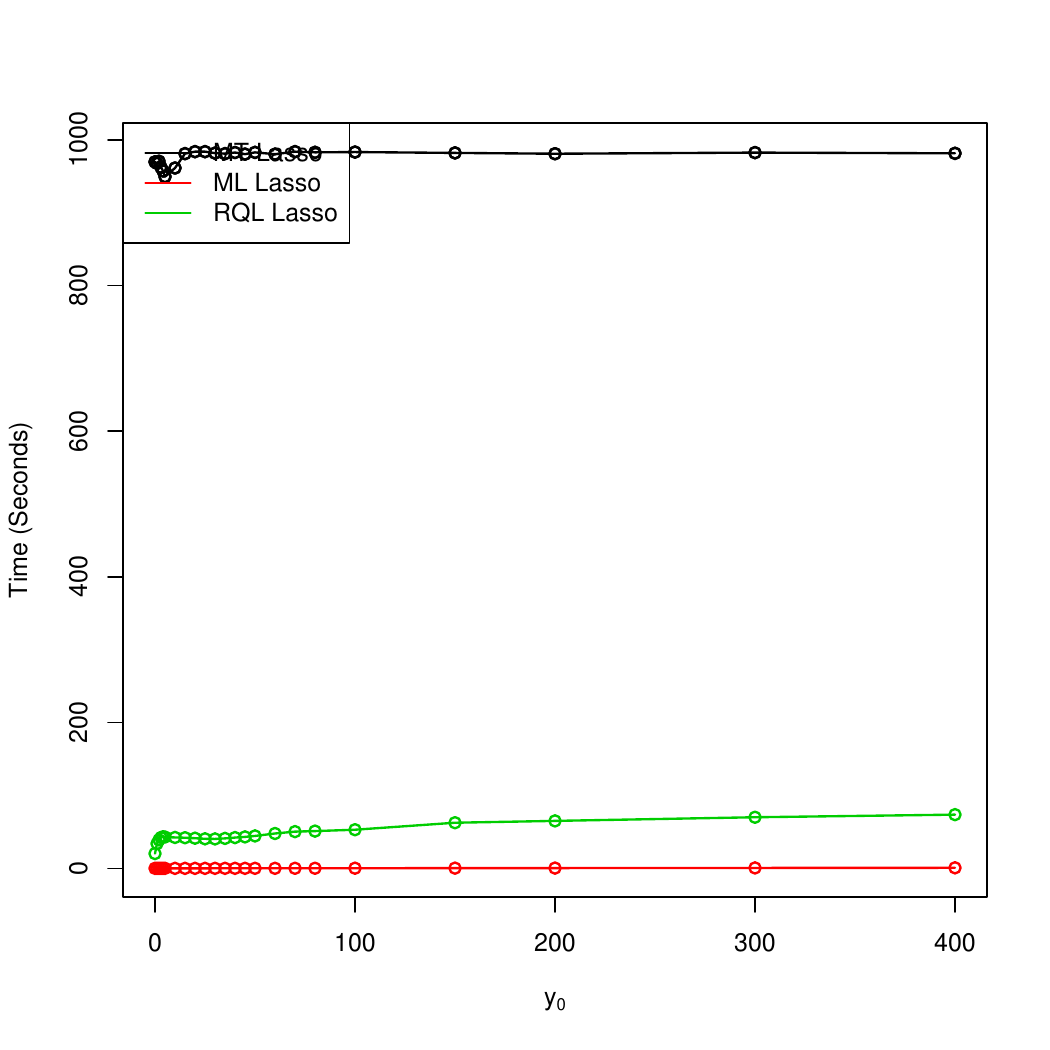}
		\end{tabular}
		\caption{\small \label{fig:timesimu1p10Lasso}  Computation times for model AVY and Lasso methods, $p=10$.}
	\end{center} 
\end{figure}

\begin{figure}[tp]
	\begin{center}
		\footnotesize
		\renewcommand{\arraystretch}{0.22}
		\newcolumntype{M}{>{\centering\arraybackslash}m{\dimexpr.01\linewidth-1\tabcolsep}}
		\newcolumntype{G}{>{\centering\arraybackslash}m{\dimexpr.38\linewidth-1\tabcolsep}}
		\begin{tabular}{GGG}
			$n=100,\, \epsilon=0.05$ & $n=100,\, \epsilon=0.1$  & $n=100,\, \epsilon=0.15$ \\
			\includegraphics[scale=0.29]{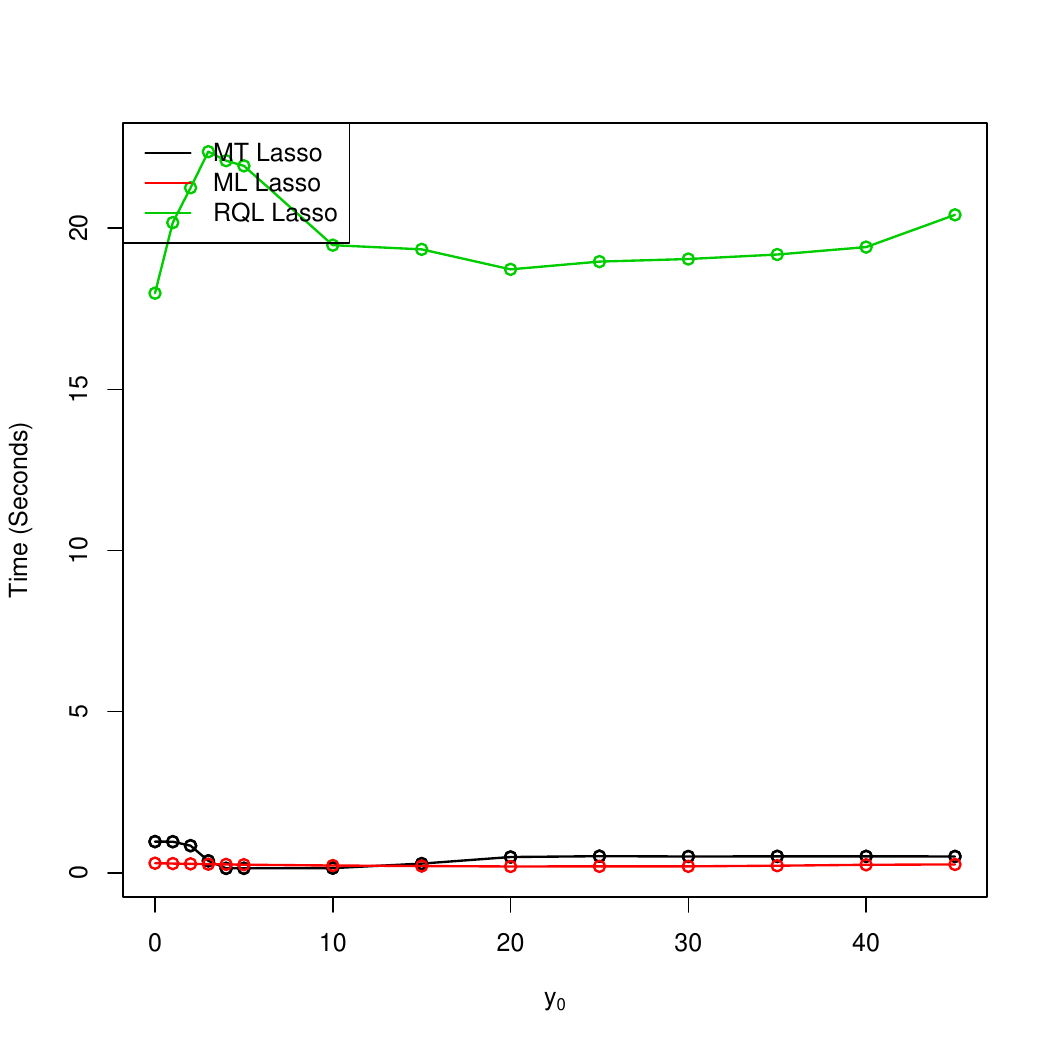} 
			&  \includegraphics[scale=0.29]{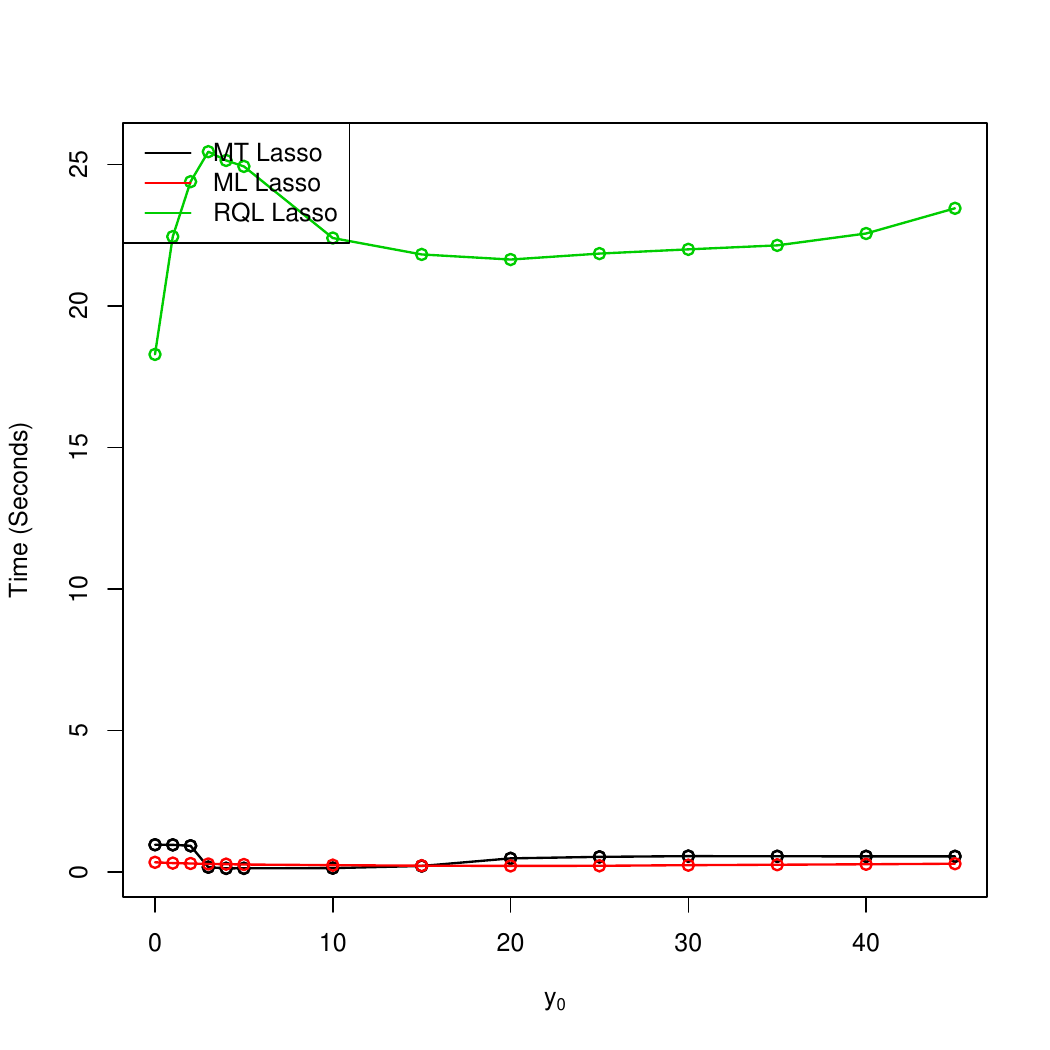}
			&   \includegraphics[scale=0.29]{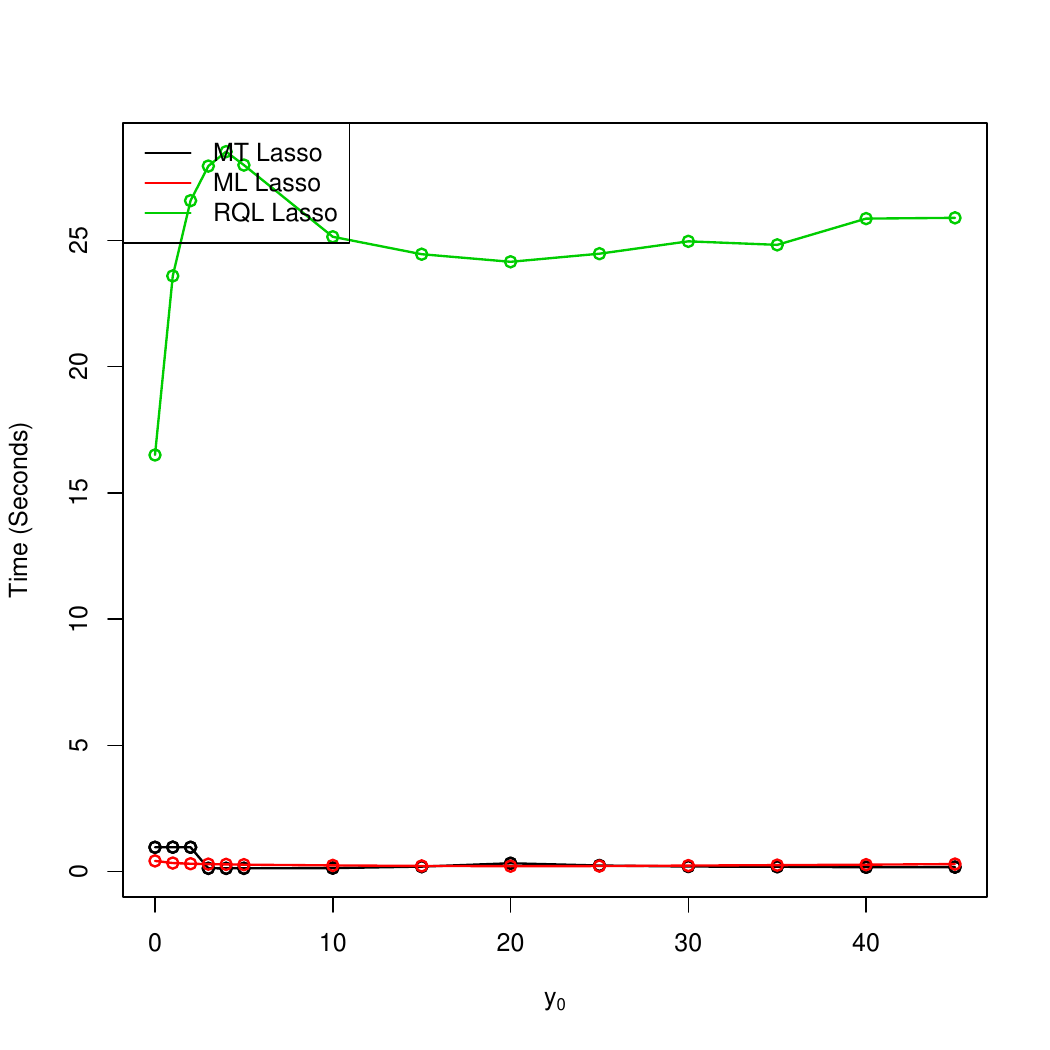}
			\\
			$n=400,\, \epsilon=0.05$ & $n=400,\, \epsilon=0.1$  & $n=400,\, \epsilon=0.15$ \\
			\includegraphics[scale=0.3]{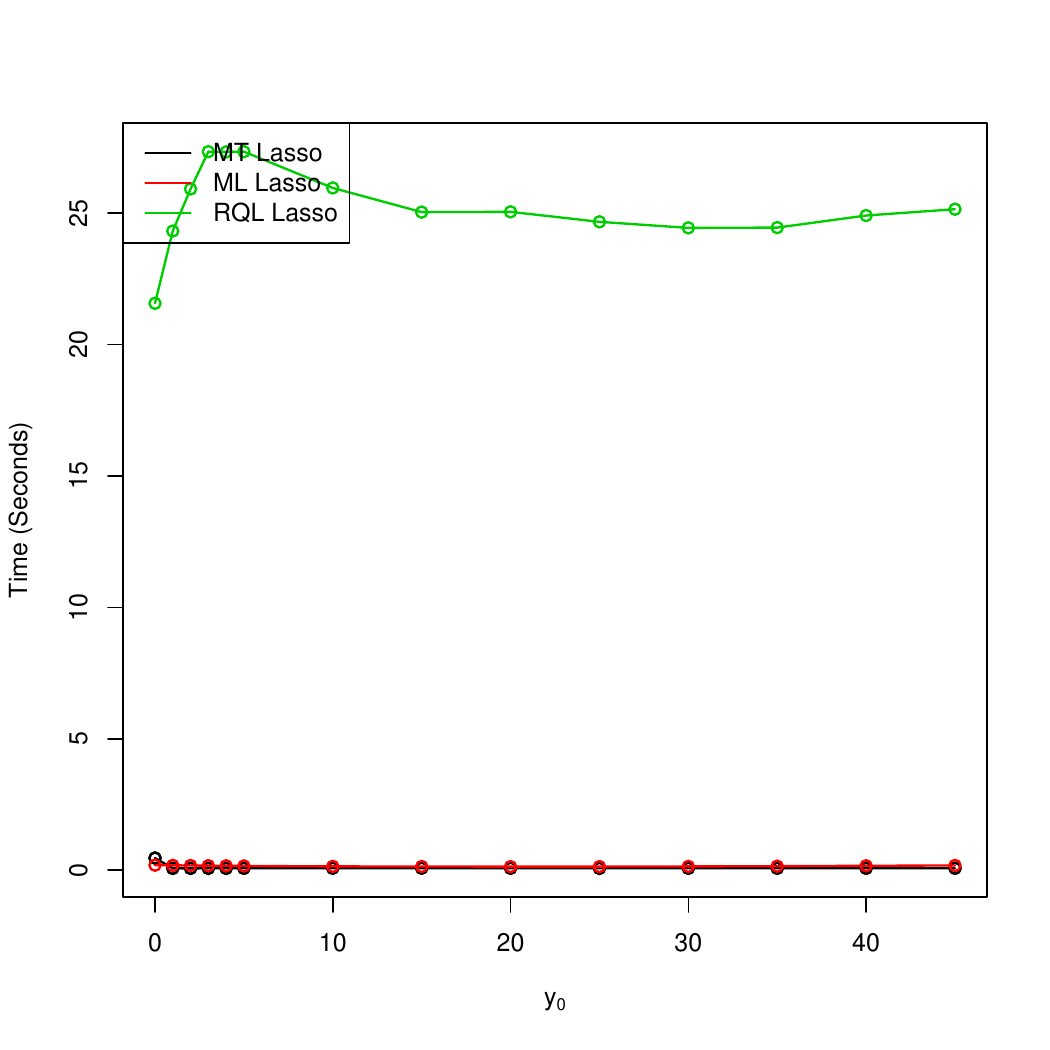}
			&   \includegraphics[scale=0.3]{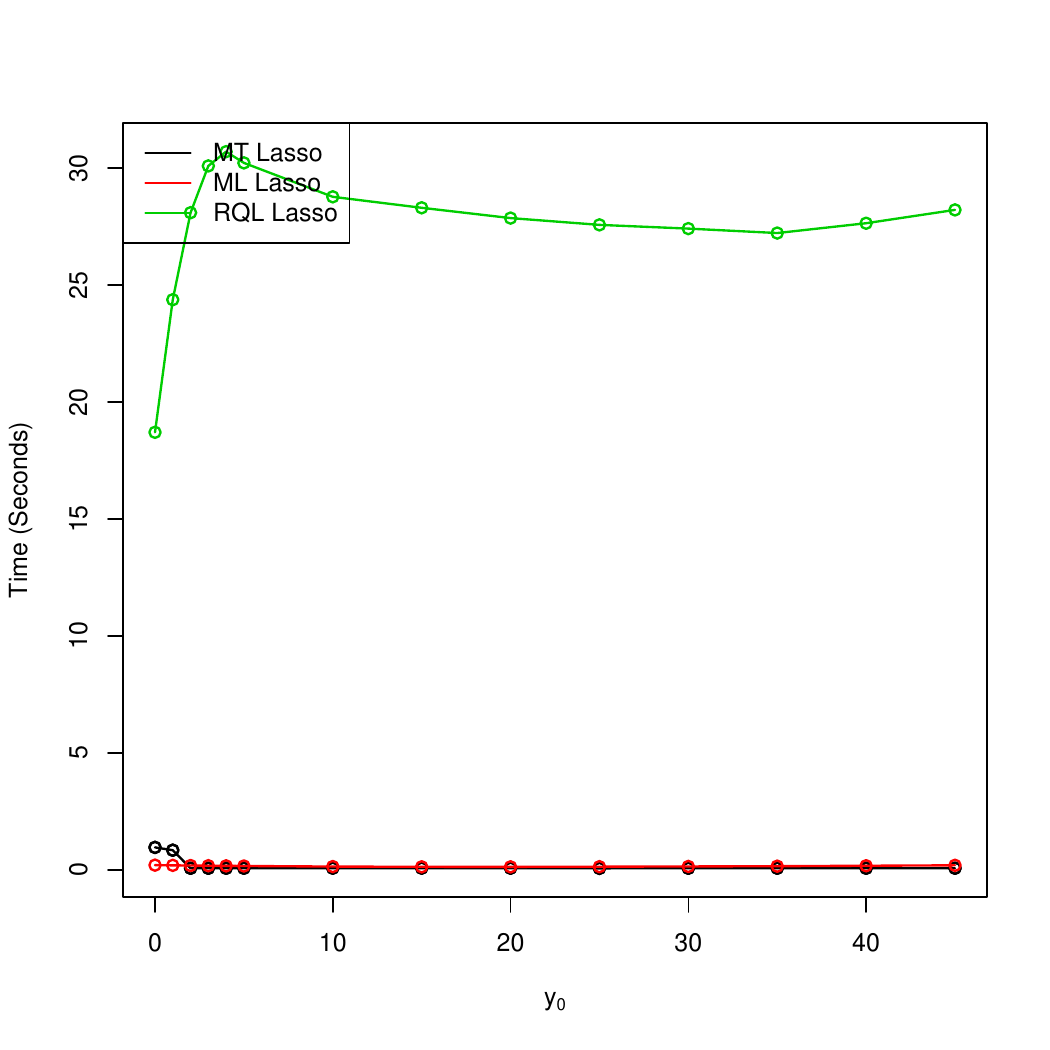} 
			&   \includegraphics[scale=0.3]{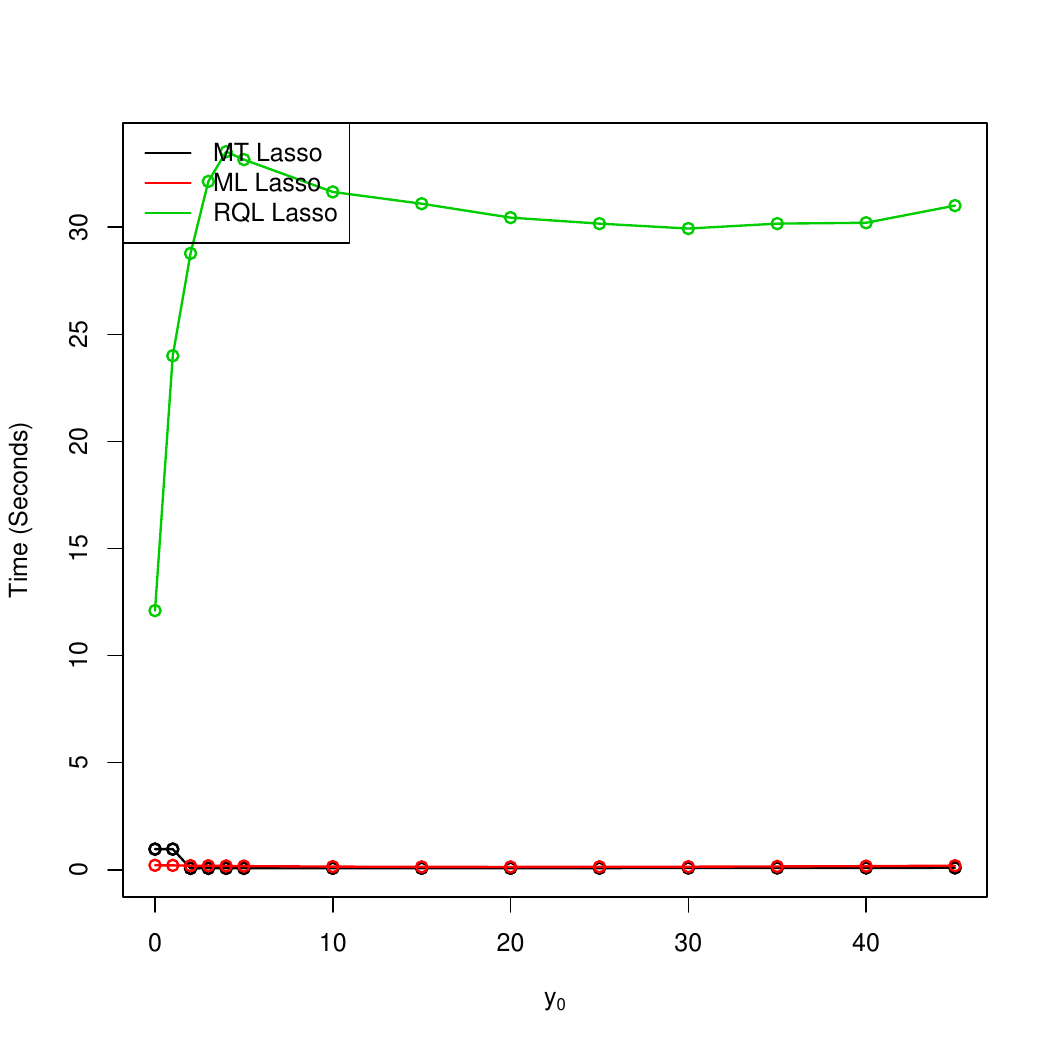} 
			\\
			$n=1000,\, \epsilon=0.05$ & $n=1000,\, \epsilon=0.1$  & $n=1000,\, \epsilon=0.15$ \\
			\includegraphics[scale=0.3]{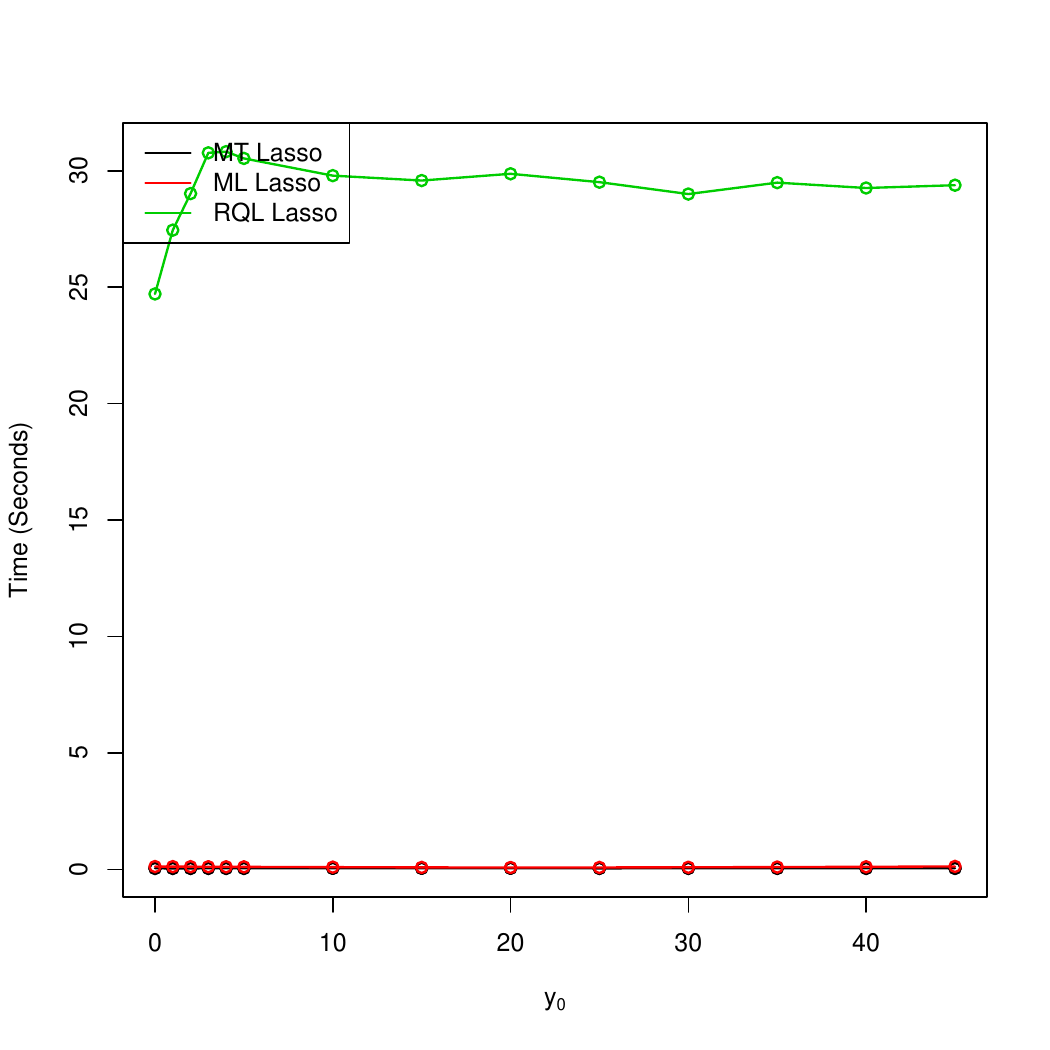}
			&   \includegraphics[scale=0.3]{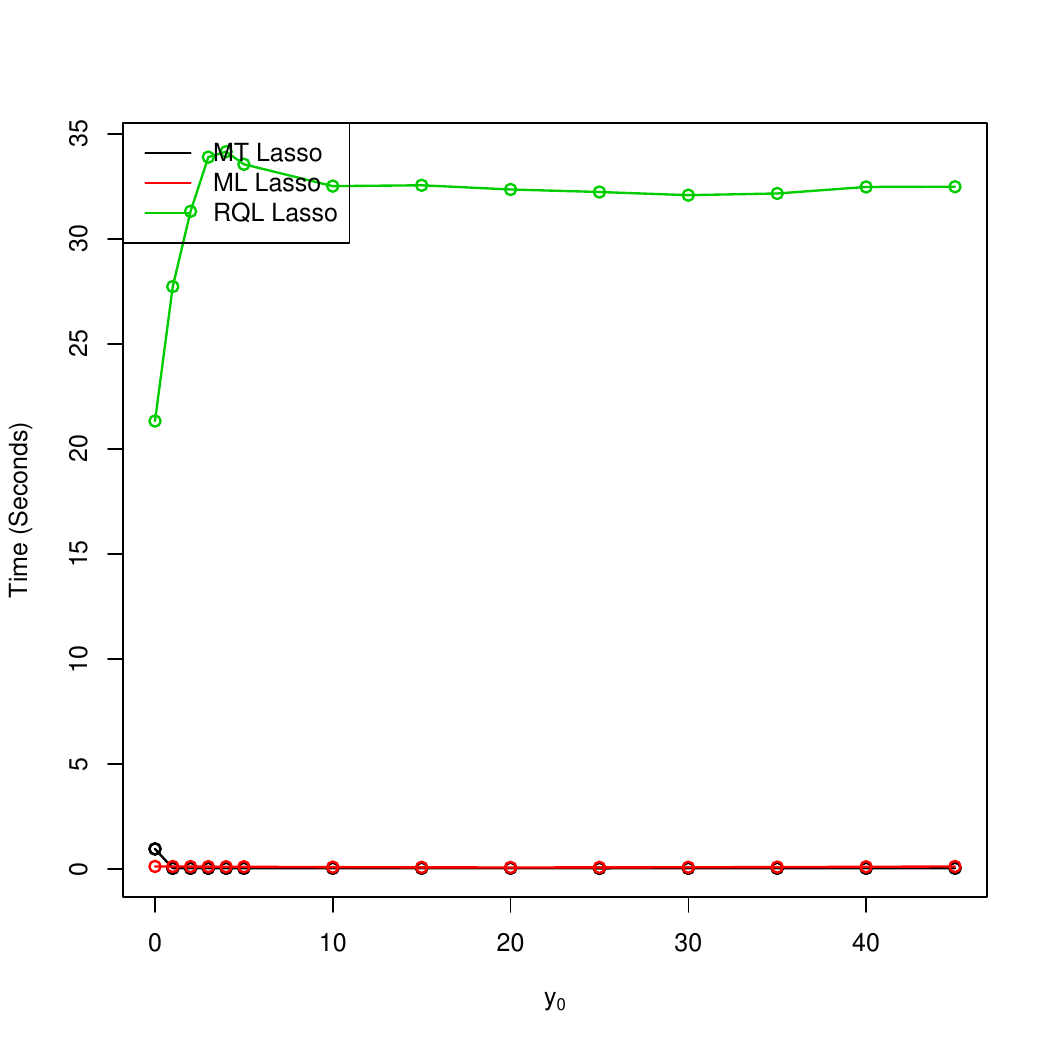} 
			&   \includegraphics[scale=0.3]{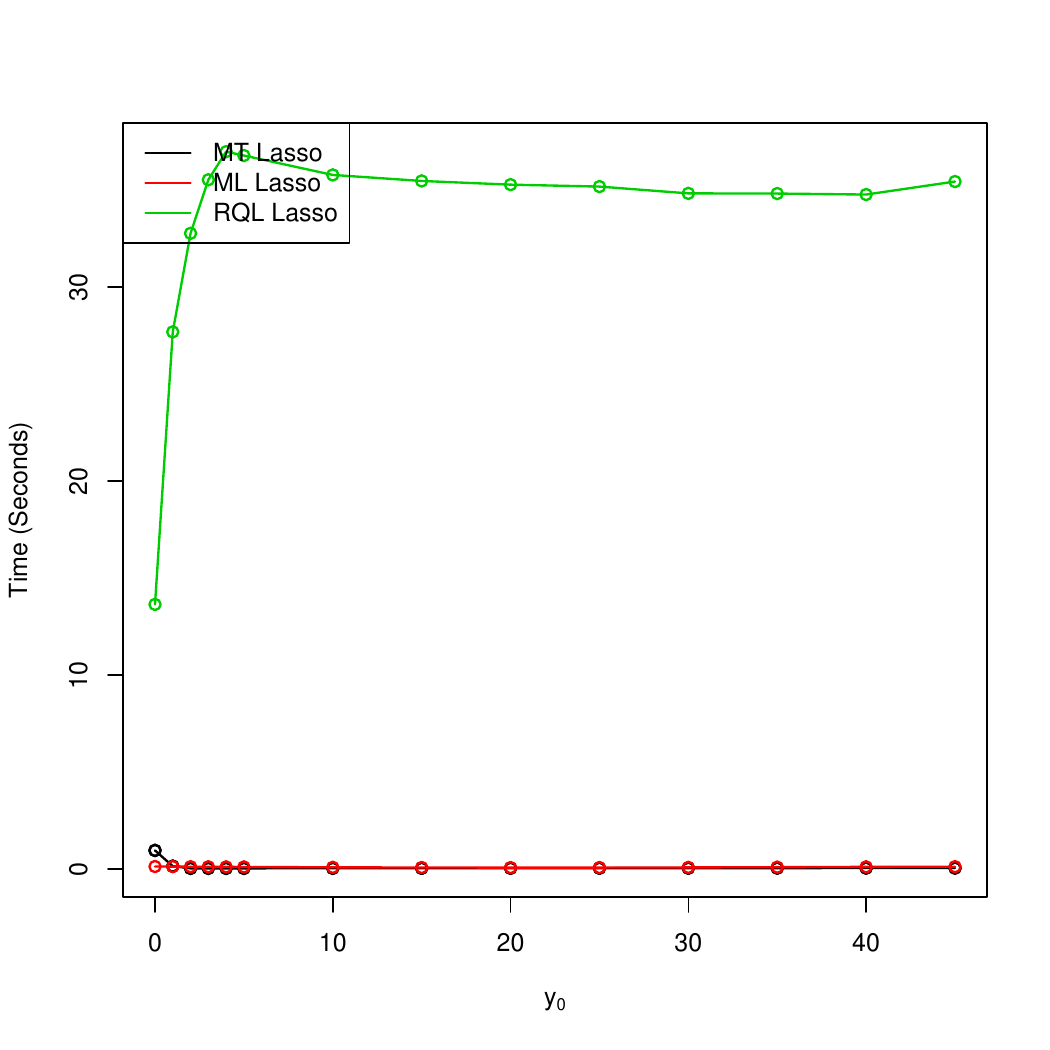}
		\end{tabular}
		\caption{\small \label{fig:timesimu1p50Lasso}  Computation times for model AVY and Lasso methods, $p=50$.}
	\end{center} 
\end{figure}

\begin{figure}[tp]
	\begin{center}
		\footnotesize
		\renewcommand{\arraystretch}{0.22}
		\newcolumntype{M}{>{\centering\arraybackslash}m{\dimexpr.01\linewidth-1\tabcolsep}}
		\newcolumntype{G}{>{\centering\arraybackslash}m{\dimexpr.38\linewidth-1\tabcolsep}}
		\begin{tabular}{GGG}
			$n=100,\, \epsilon=0.05$ & $n=100,\, \epsilon=0.1$  & $n=100,\, \epsilon=0.15$ \\
			\includegraphics[scale=0.29]{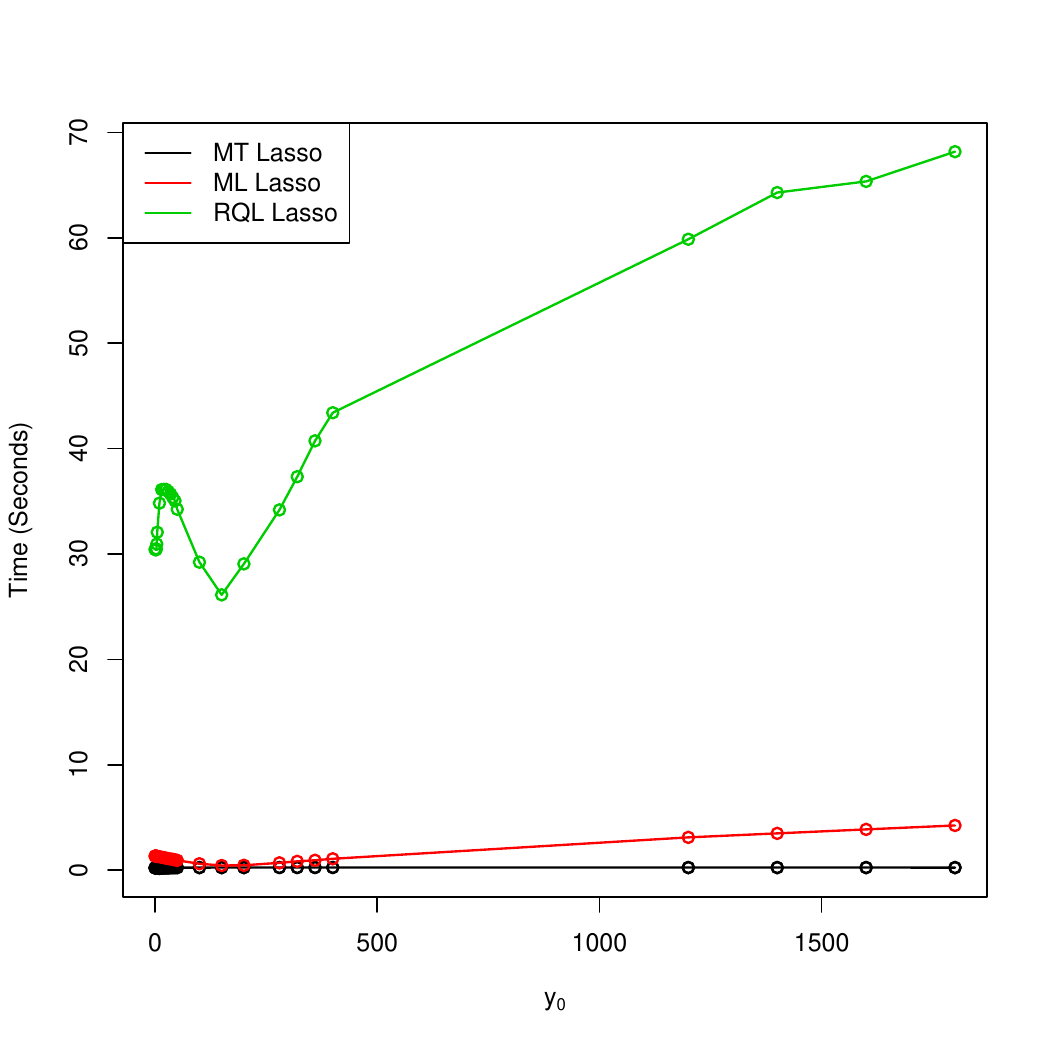} 
			&  \includegraphics[scale=0.29]{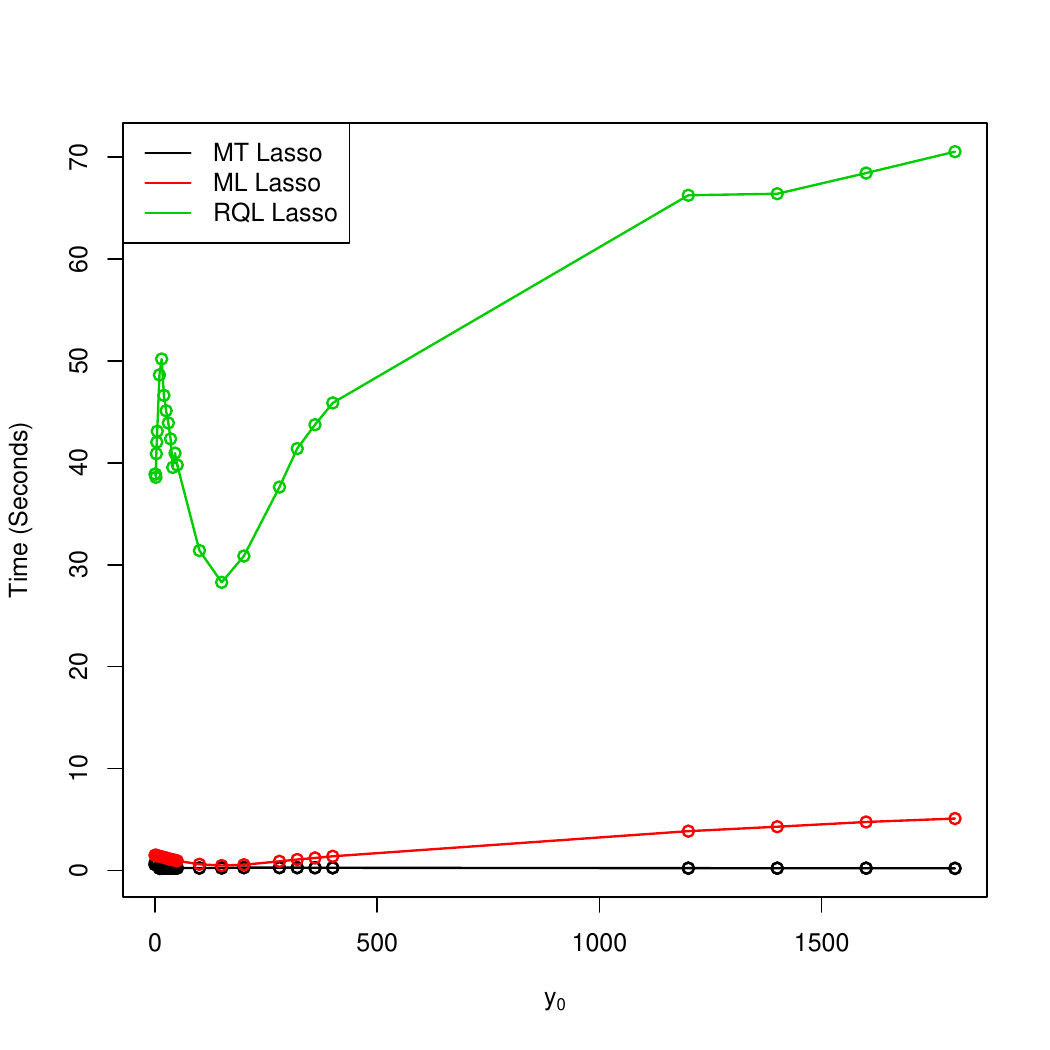}
			&   \includegraphics[scale=0.29]{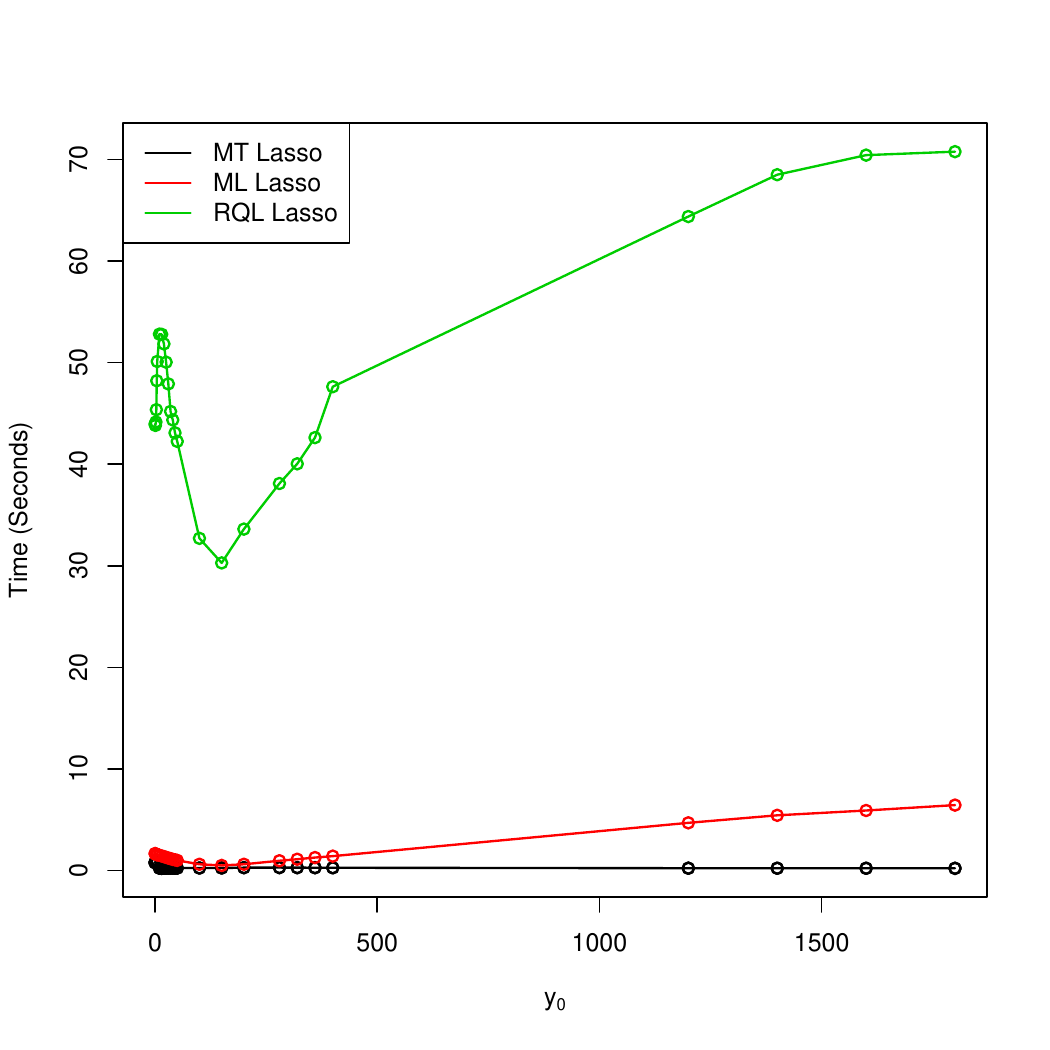}
			\\
			$n=400,\, \epsilon=0.05$ & $n=400,\, \epsilon=0.1$  & $n=400,\, \epsilon=0.15$ \\
			\includegraphics[scale=0.3]{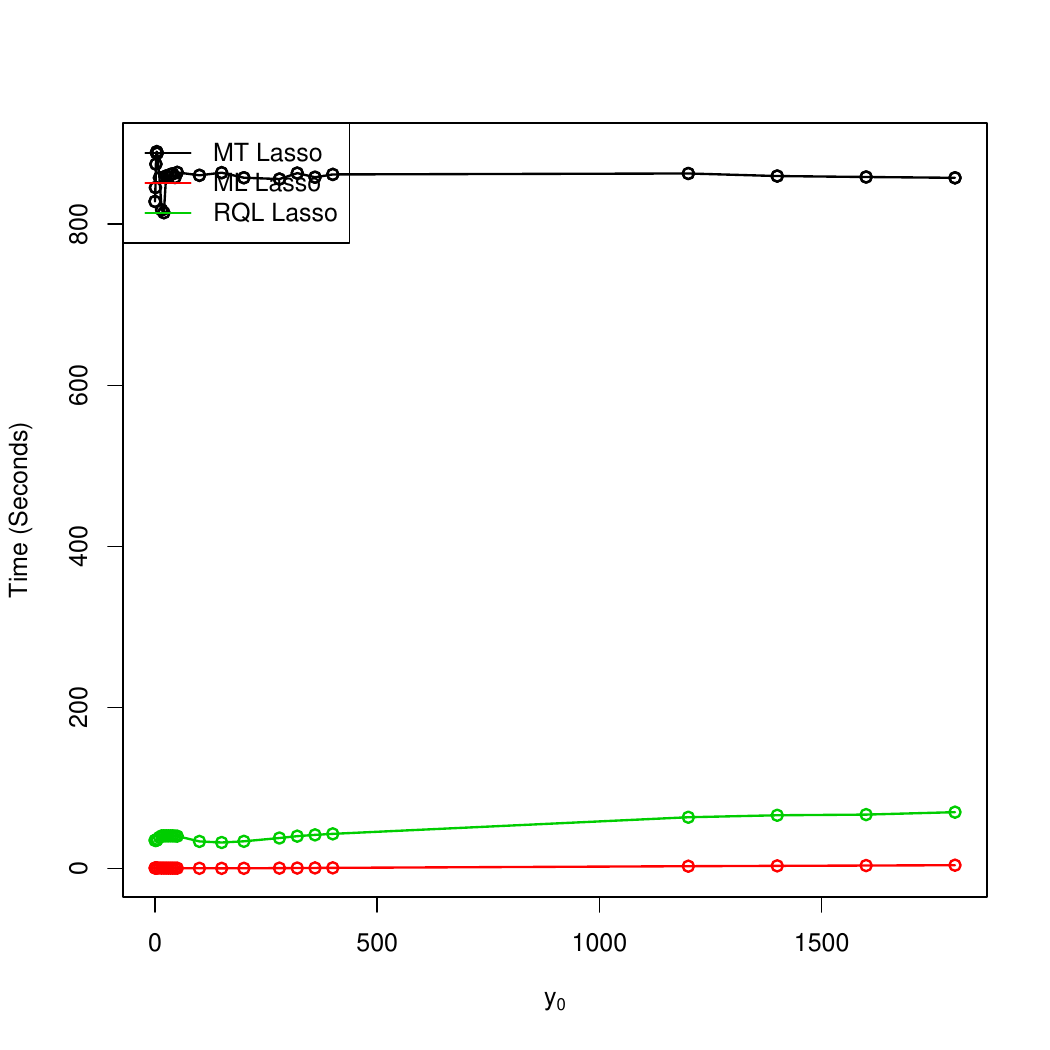}
			&   \includegraphics[scale=0.3]{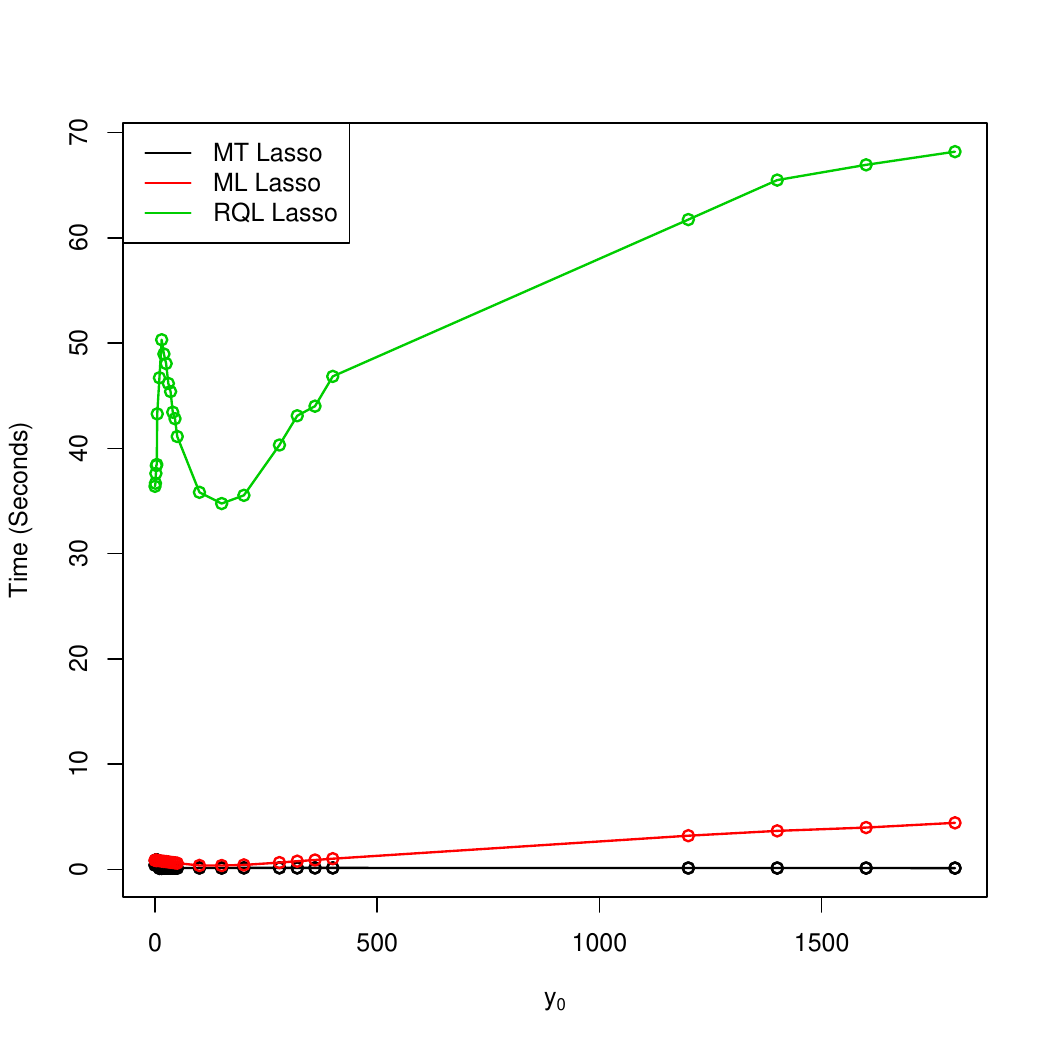} 
			&   \includegraphics[scale=0.3]{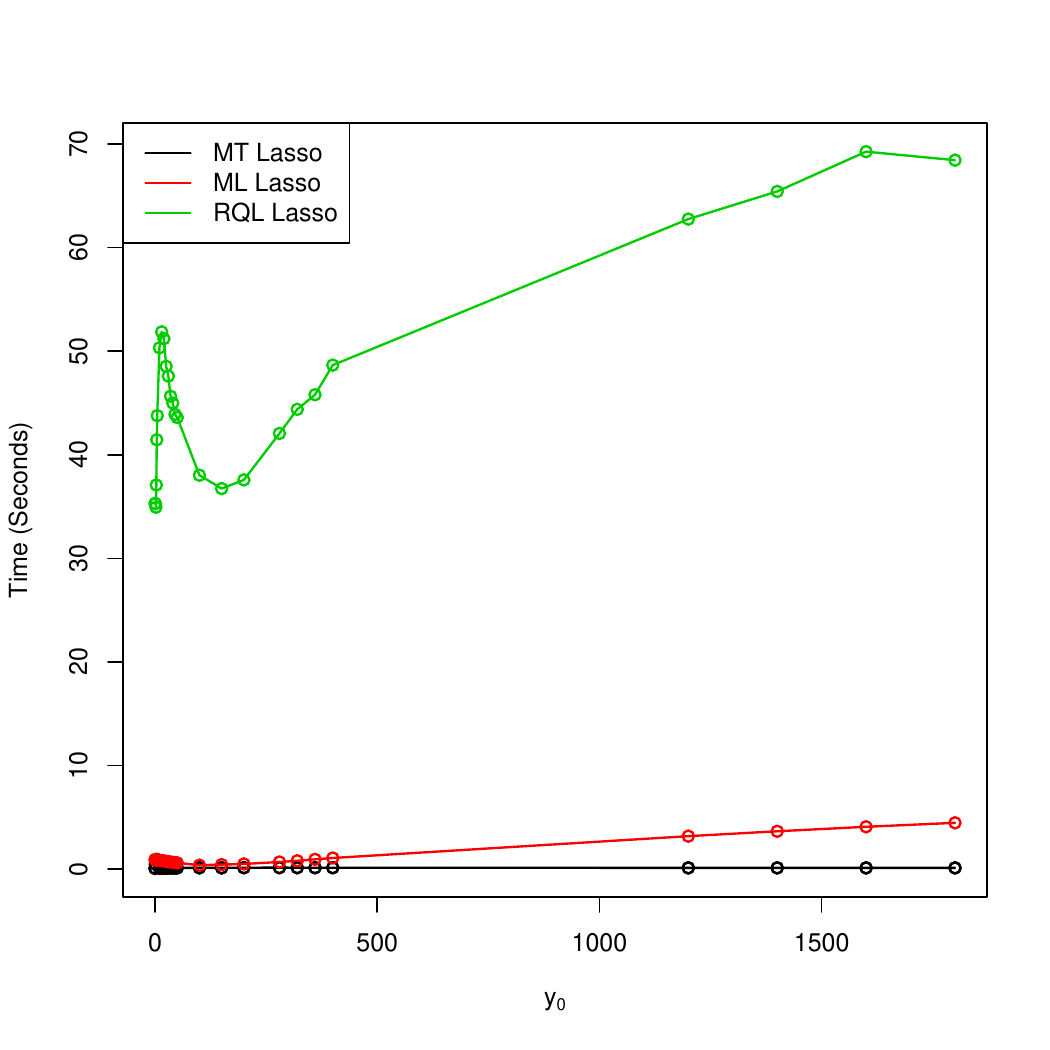} 
			\\
			$n=1000,\, \epsilon=0.05$ & $n=1000,\, \epsilon=0.1$  & $n=1000,\, \epsilon=0.15$ \\
			\includegraphics[scale=0.3]{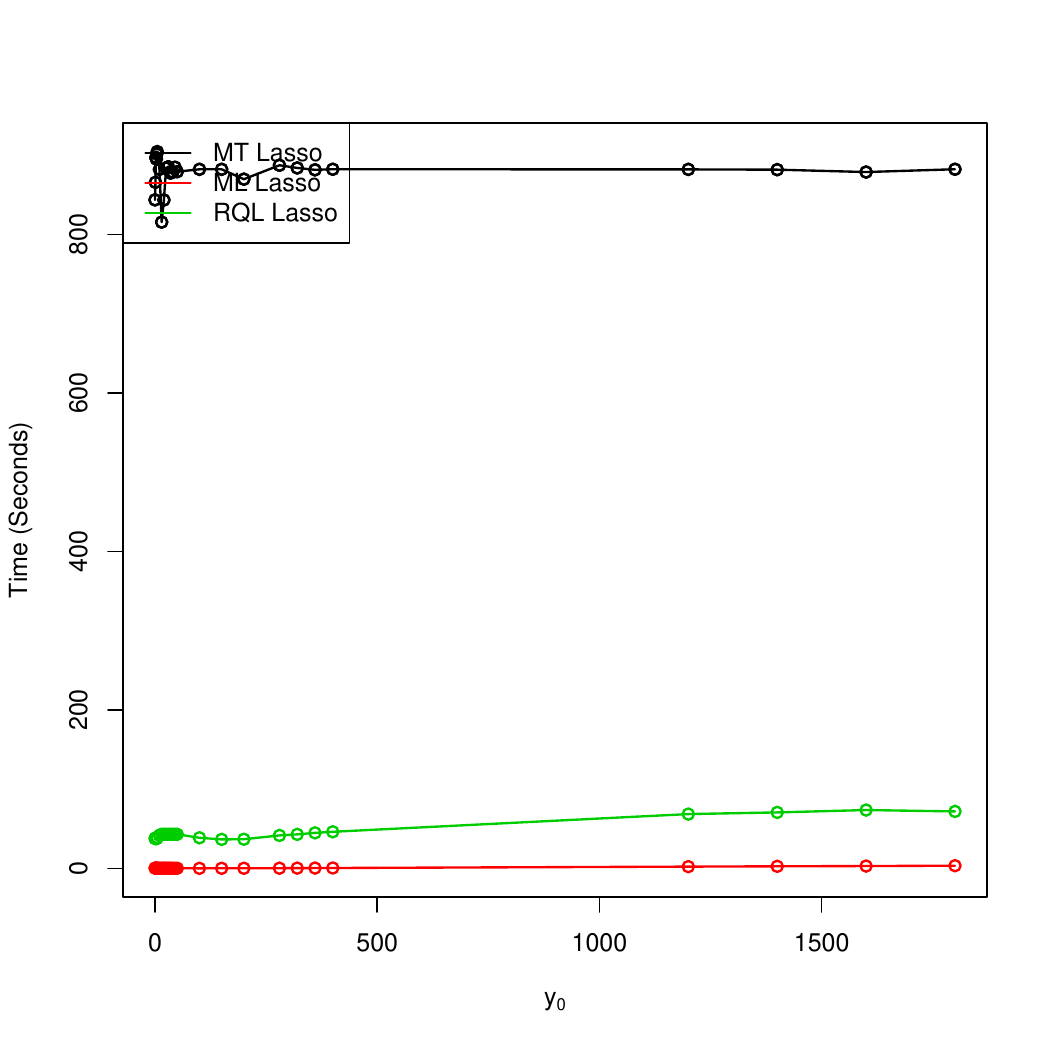}
			&   \includegraphics[scale=0.3]{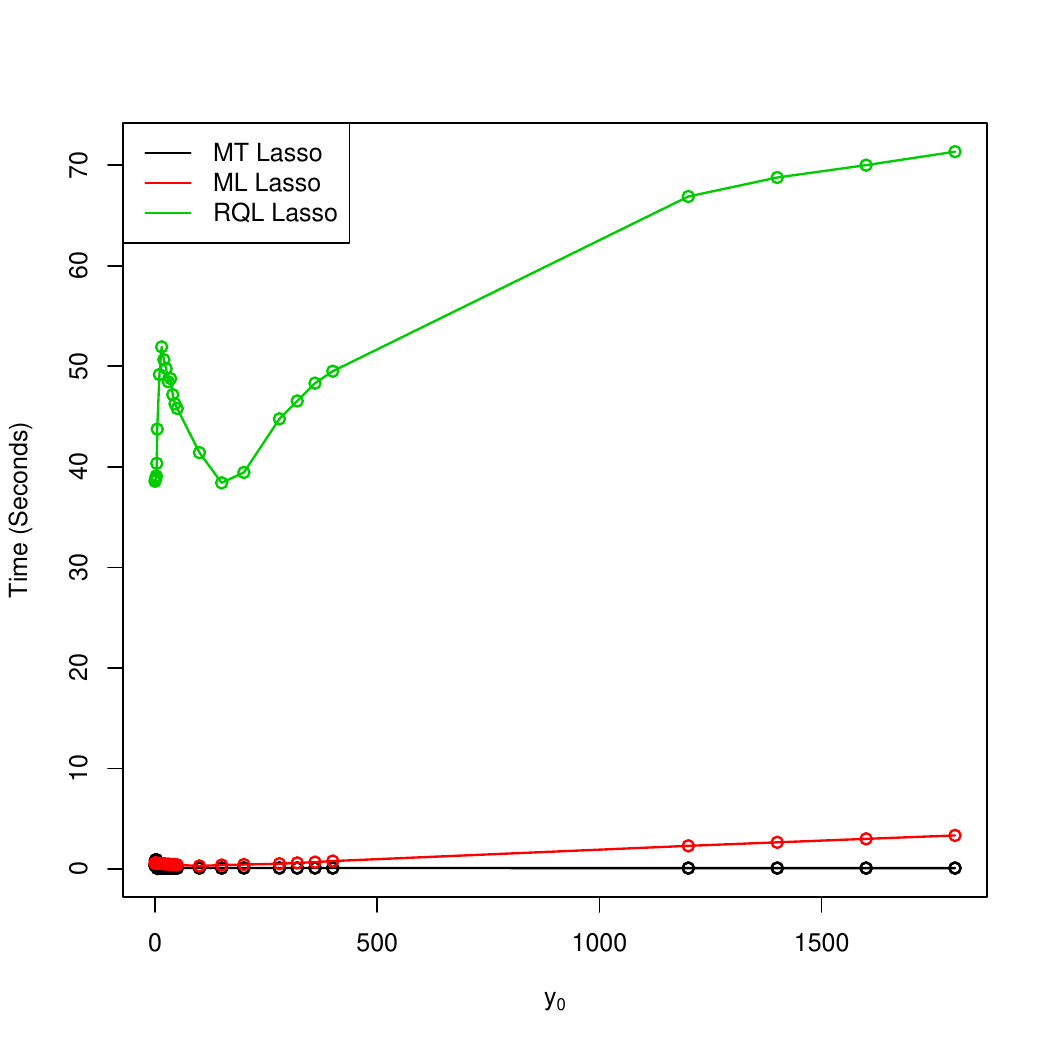} 
			&   \includegraphics[scale=0.3]{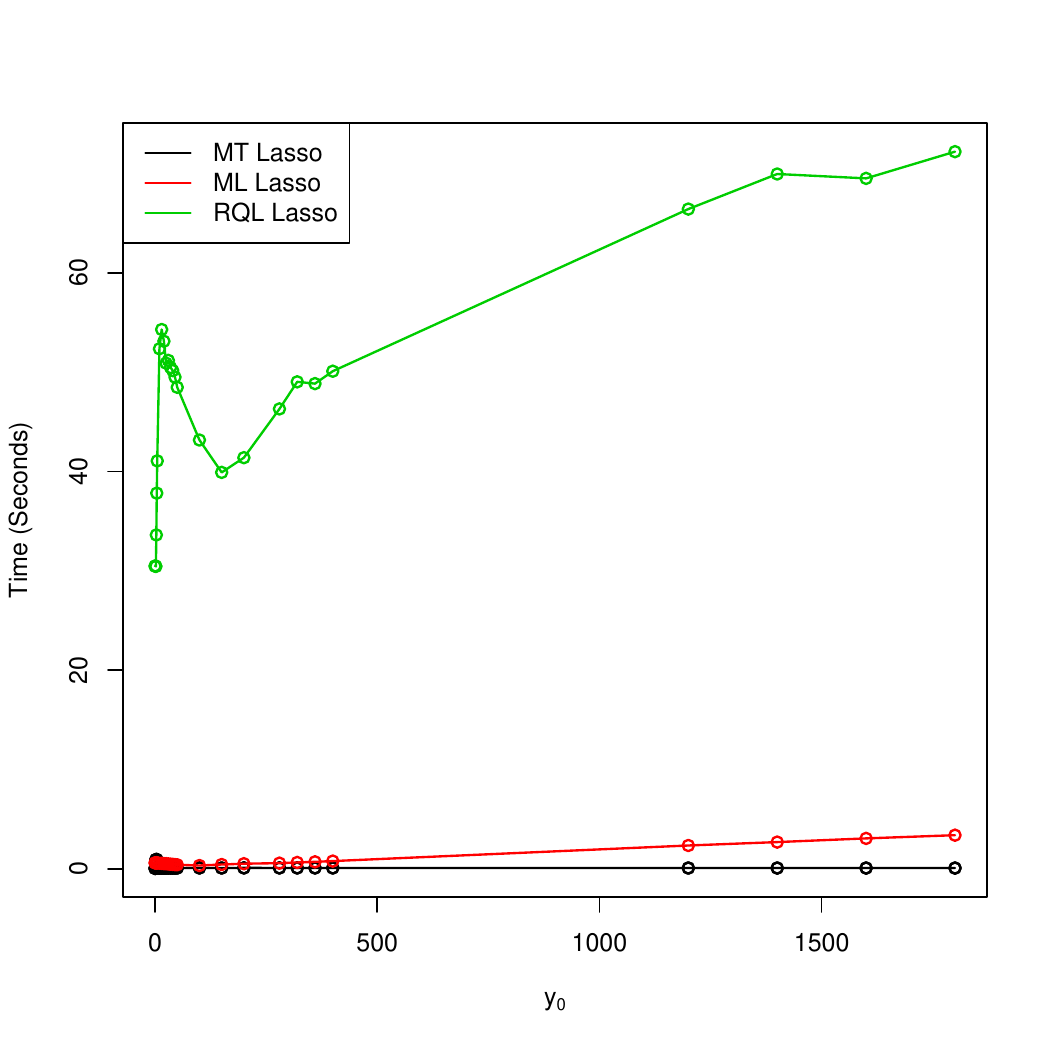}
		\end{tabular}
		\caption{\small \label{fig:timesimu2p10Lasso}  Computation times for model AVY2 and Lasso methods, $p=10$.}
	\end{center} 
\end{figure}

\begin{figure}[tp]
	\begin{center}
		\footnotesize
		\renewcommand{\arraystretch}{0.22}
		\newcolumntype{M}{>{\centering\arraybackslash}m{\dimexpr.01\linewidth-1\tabcolsep}}
		\newcolumntype{G}{>{\centering\arraybackslash}m{\dimexpr.38\linewidth-1\tabcolsep}}
		\begin{tabular}{GGG}
			$n=100,\, \epsilon=0.05$ & $n=100,\, \epsilon=0.1$  & $n=100,\, \epsilon=0.15$ \\
			\includegraphics[scale=0.29]{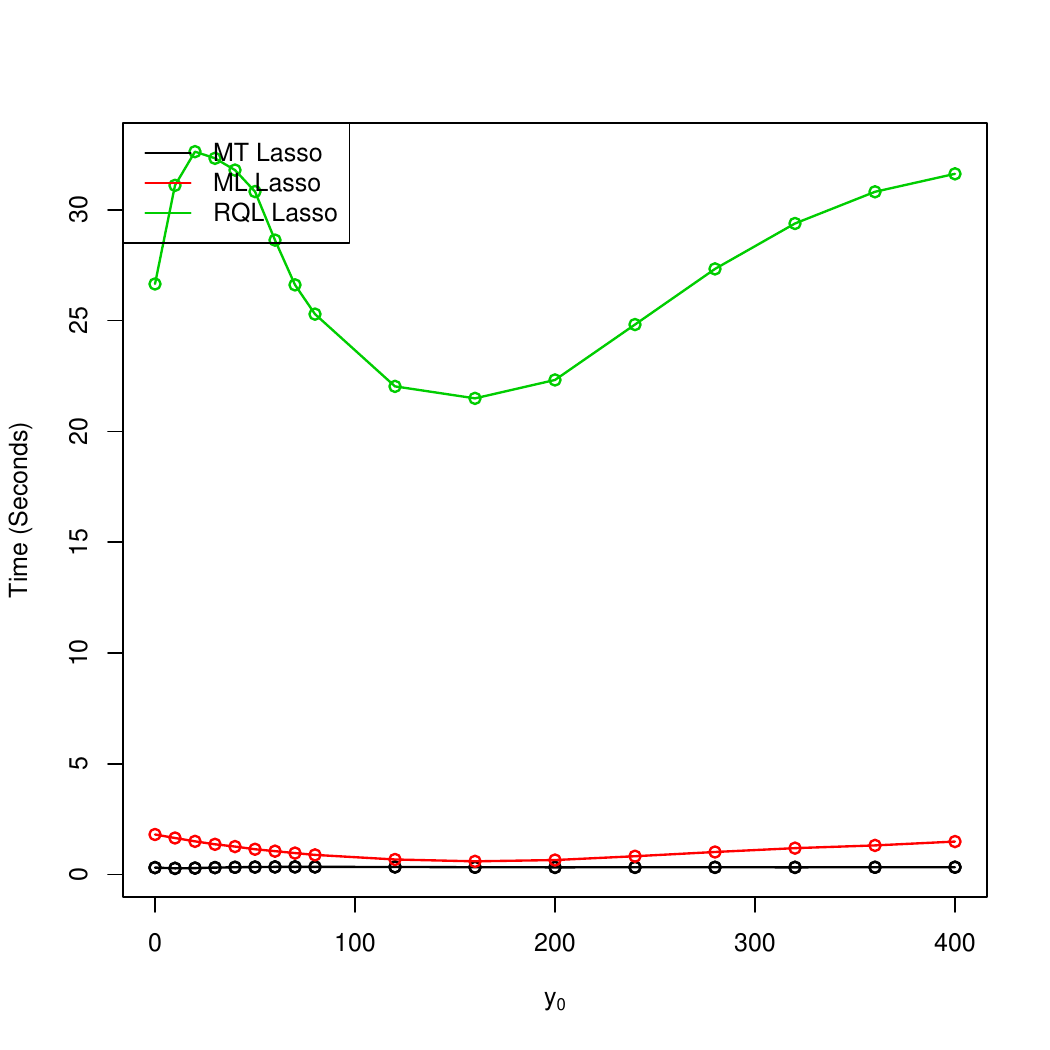} 
			&  \includegraphics[scale=0.29]{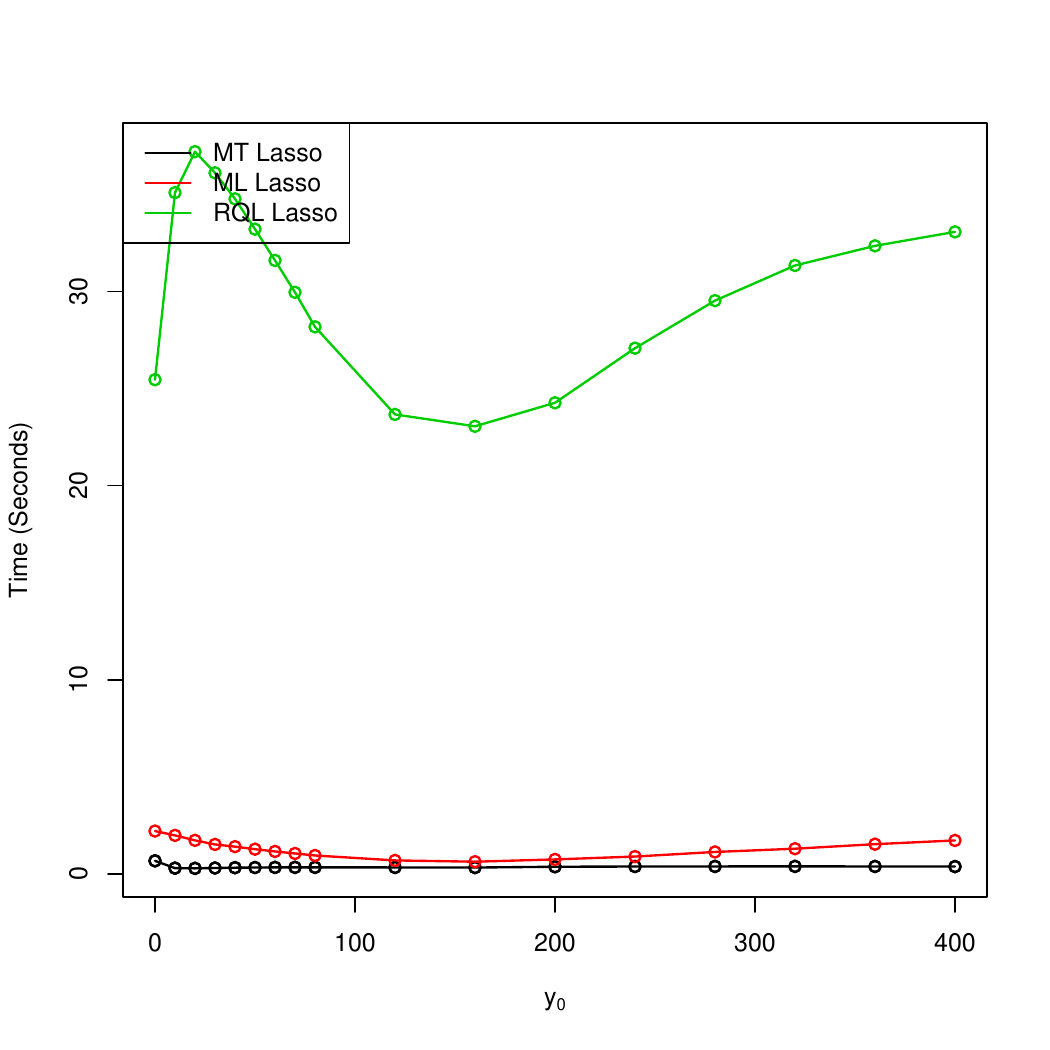}
			&   \includegraphics[scale=0.29]{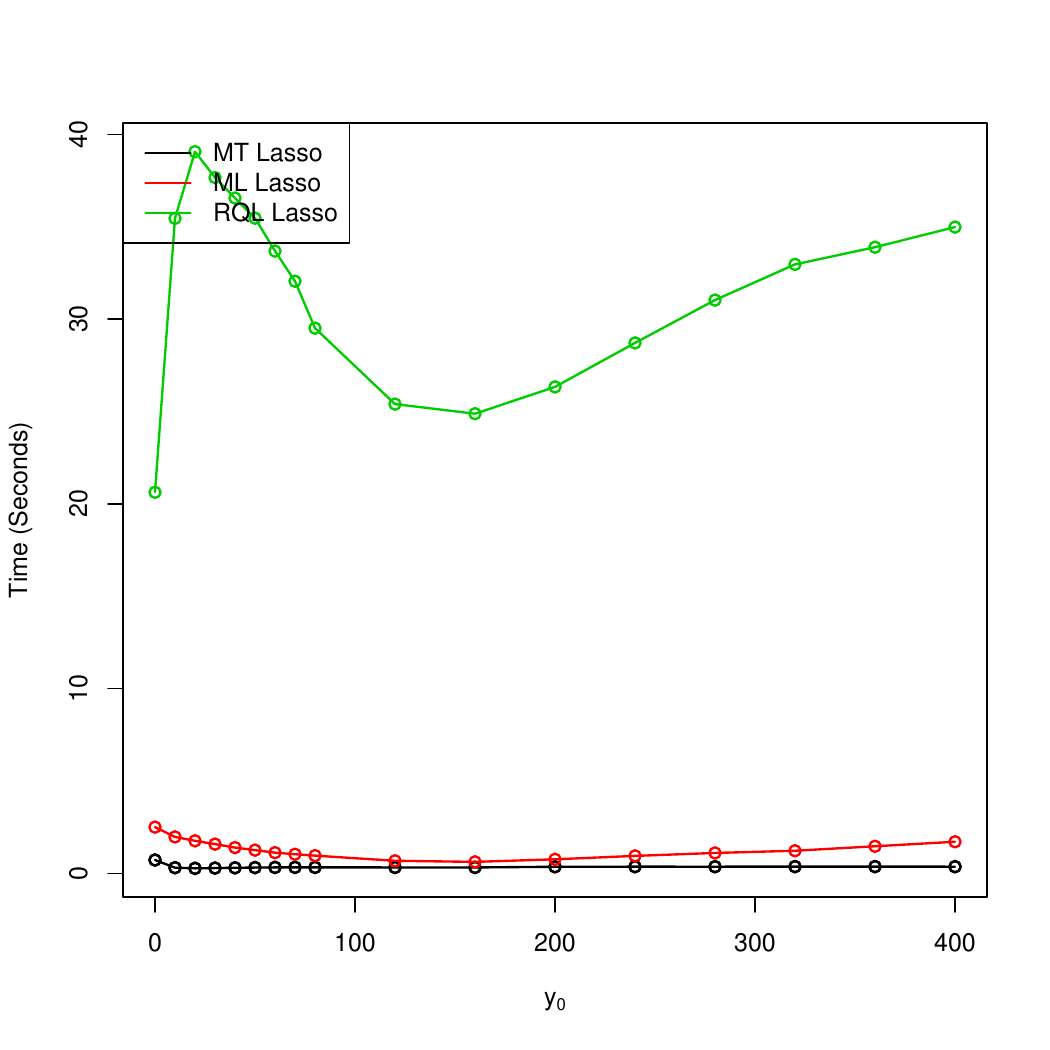}
			\\
			$n=400,\, \epsilon=0.05$ & $n=400,\, \epsilon=0.1$  & $n=400,\, \epsilon=0.15$ \\
			\includegraphics[scale=0.3]{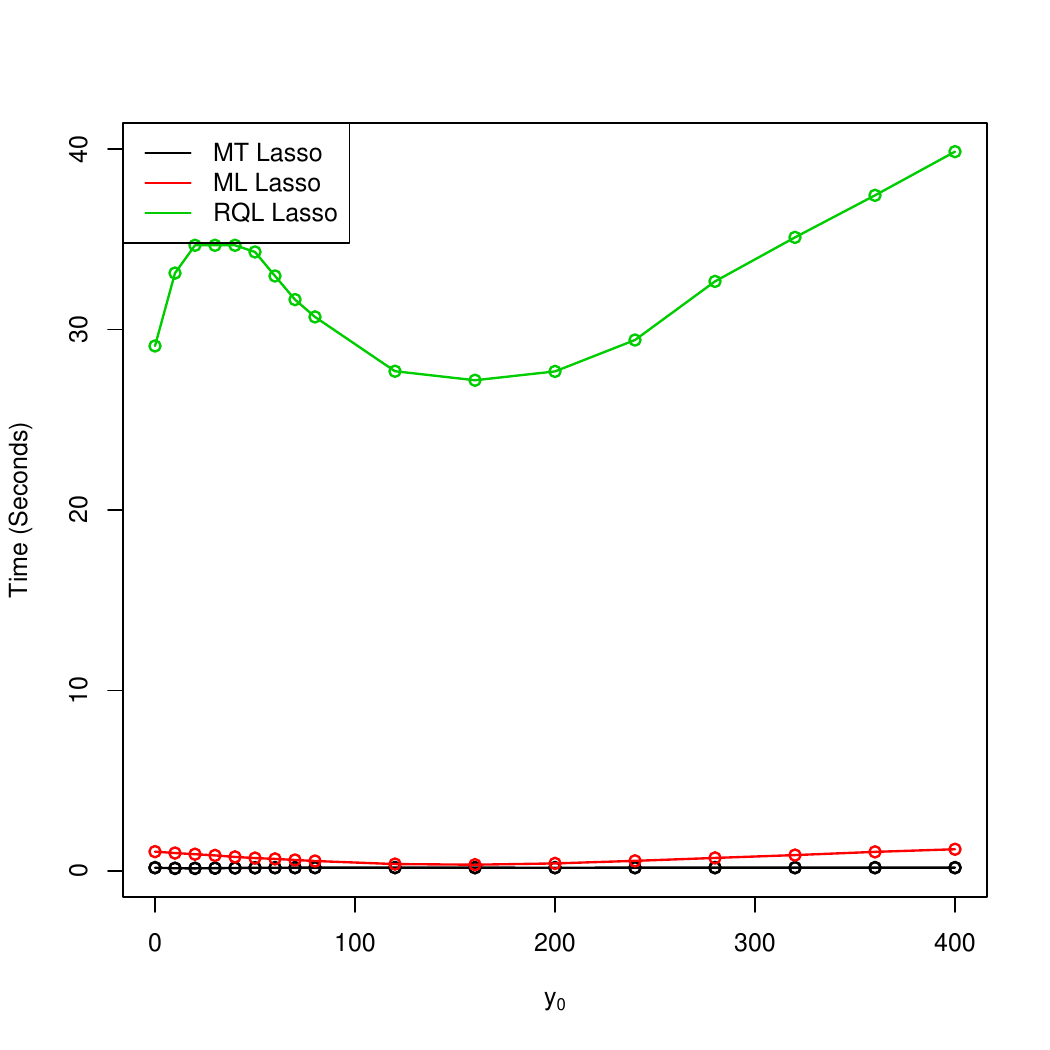}
			&   \includegraphics[scale=0.3]{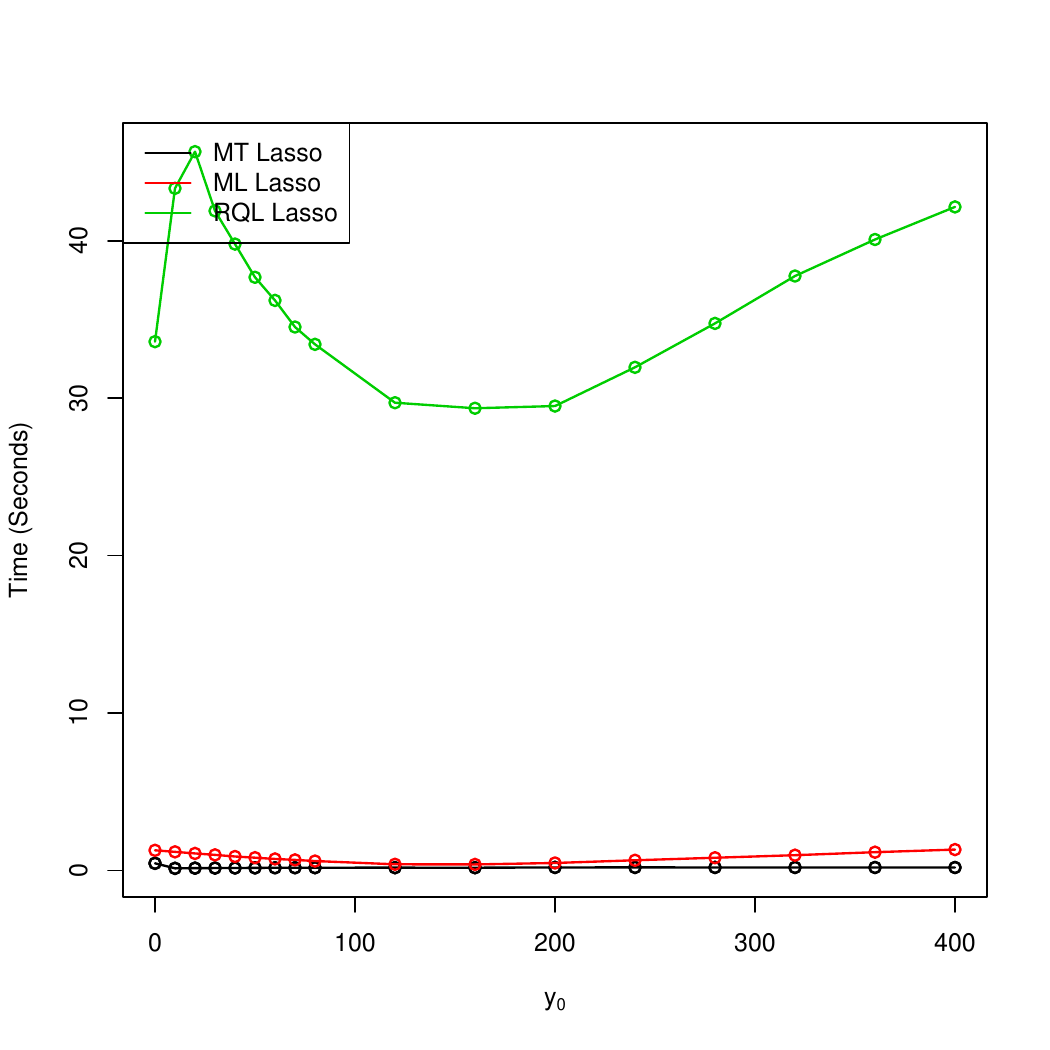} 
			&   \includegraphics[scale=0.3]{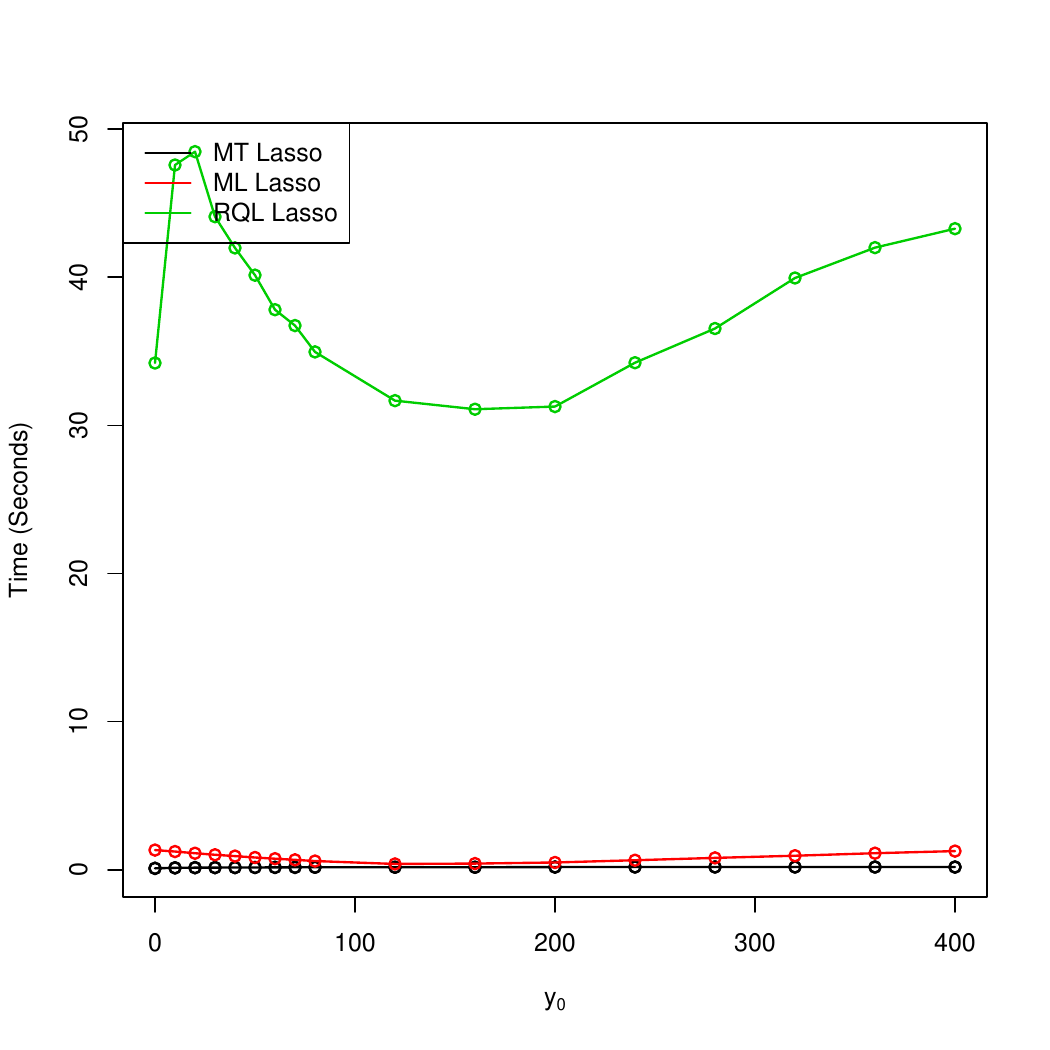} 
			\\
			$n=1000,\, \epsilon=0.05$ & $n=1000,\, \epsilon=0.1$  & $n=1000,\, \epsilon=0.15$ \\
			\includegraphics[scale=0.3]{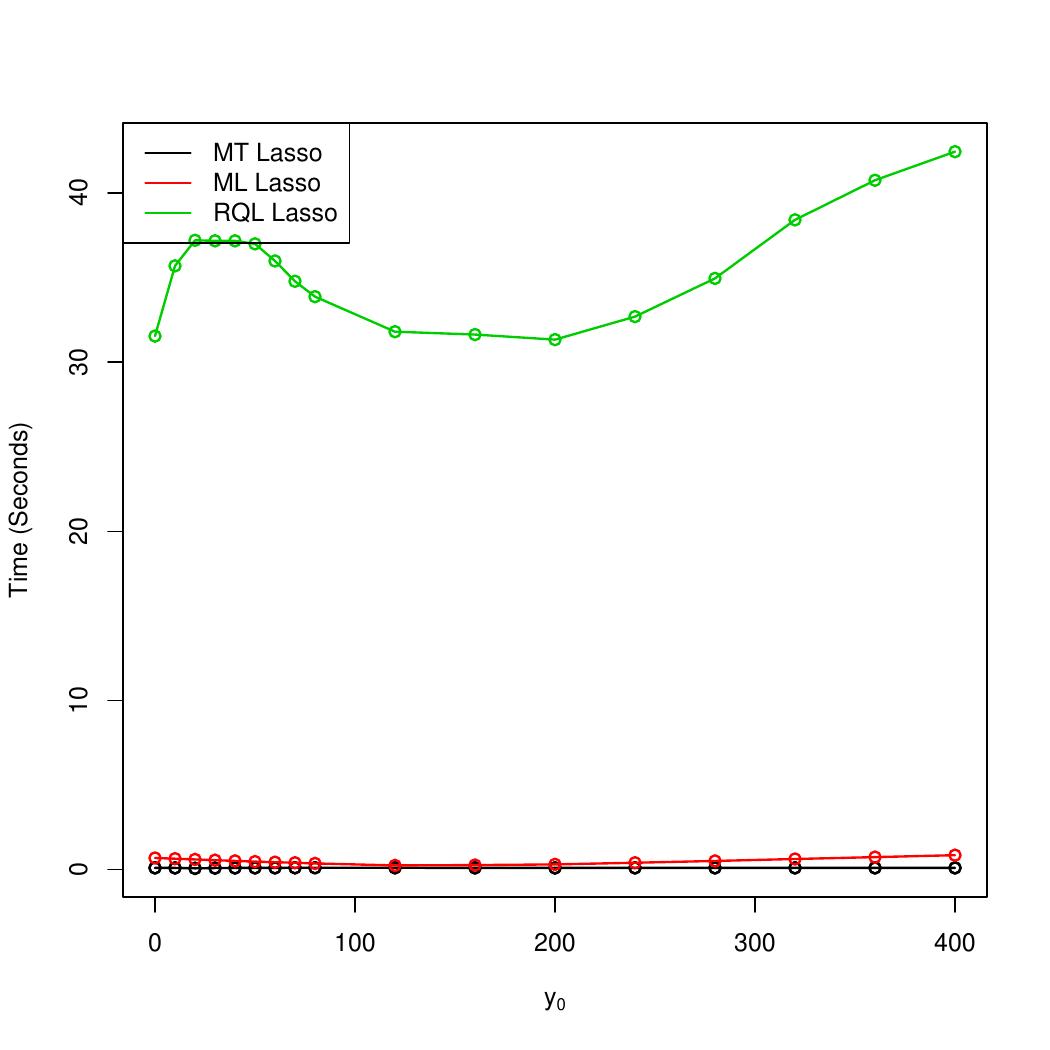}
			&   \includegraphics[scale=0.3]{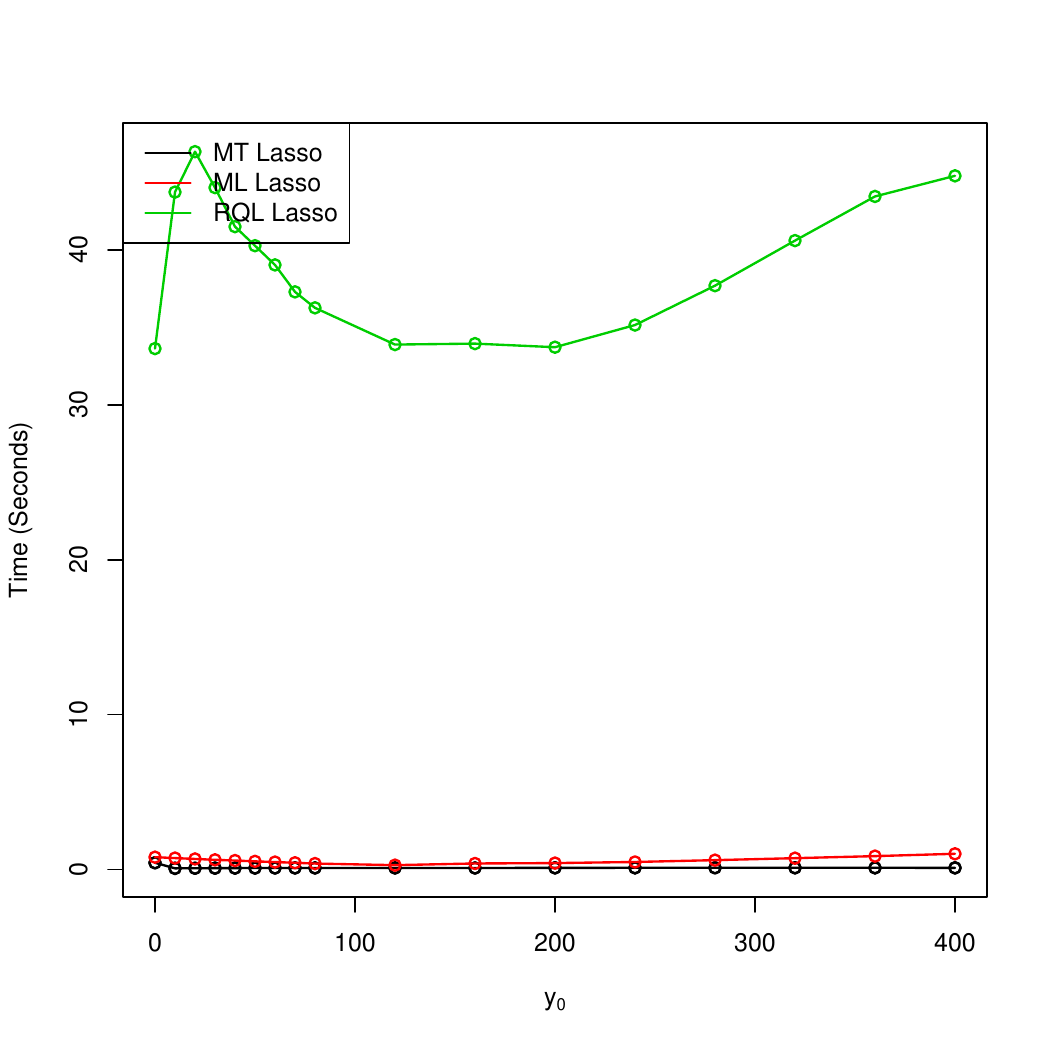} 
			&   \includegraphics[scale=0.3]{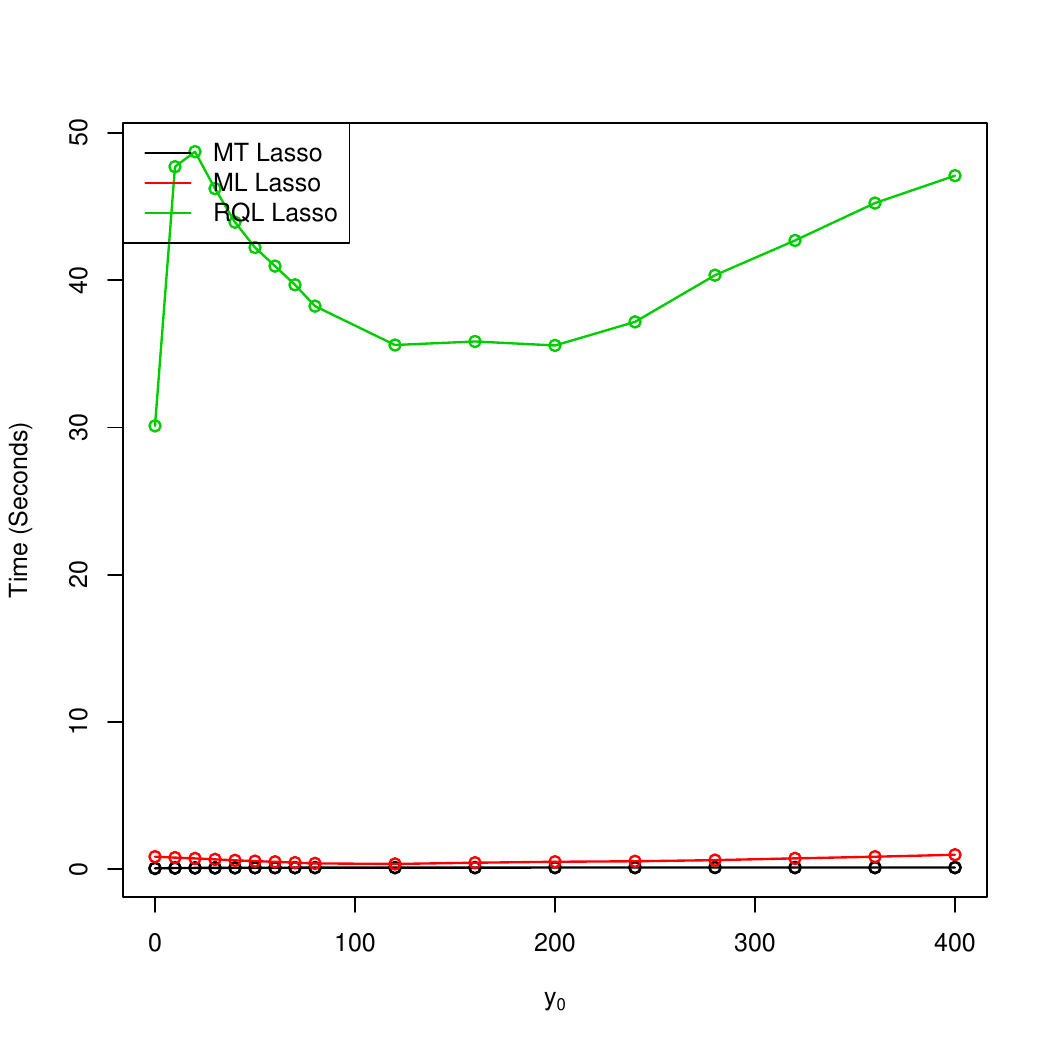}
		\end{tabular}
		\caption{\small \label{fig:timesimu2p50Lasso}  Computation times for model  AVY2 and Lasso methods, $p=50$.}
	\end{center} 
\end{figure}

\begin{figure}[tp]
	\begin{center}
		\footnotesize
		\renewcommand{\arraystretch}{0.22}
		\newcolumntype{M}{>{\centering\arraybackslash}m{\dimexpr.01\linewidth-1\tabcolsep}}
		\newcolumntype{G}{>{\centering\arraybackslash}m{\dimexpr.38\linewidth-1\tabcolsep}}
		\begin{tabular}{GGG}
			$n=50,\, \epsilon=0.05$ & $n=100,\, \epsilon=0.1$  & $n=50,\, \epsilon=0.15$ \\
			\includegraphics[scale=0.29]{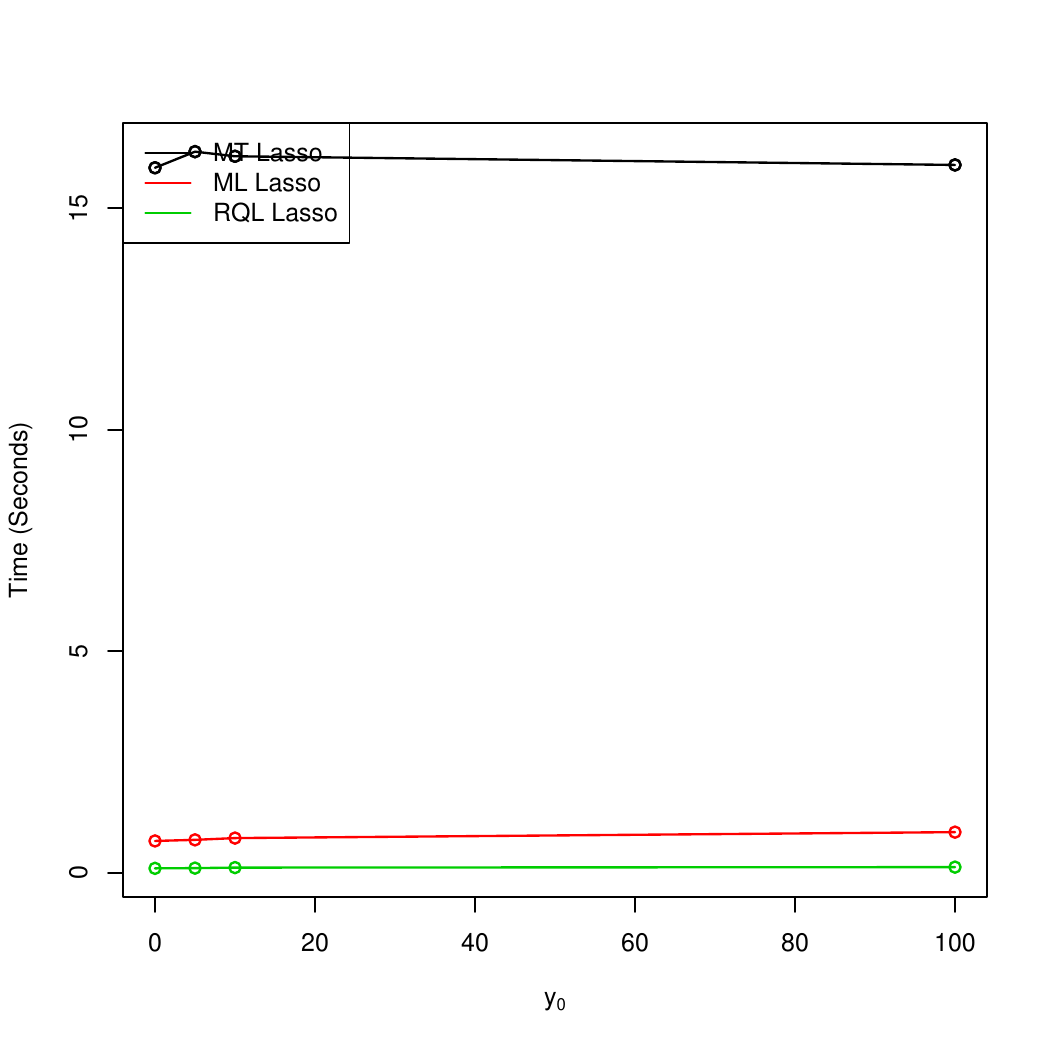} 
			&  \includegraphics[scale=0.29]{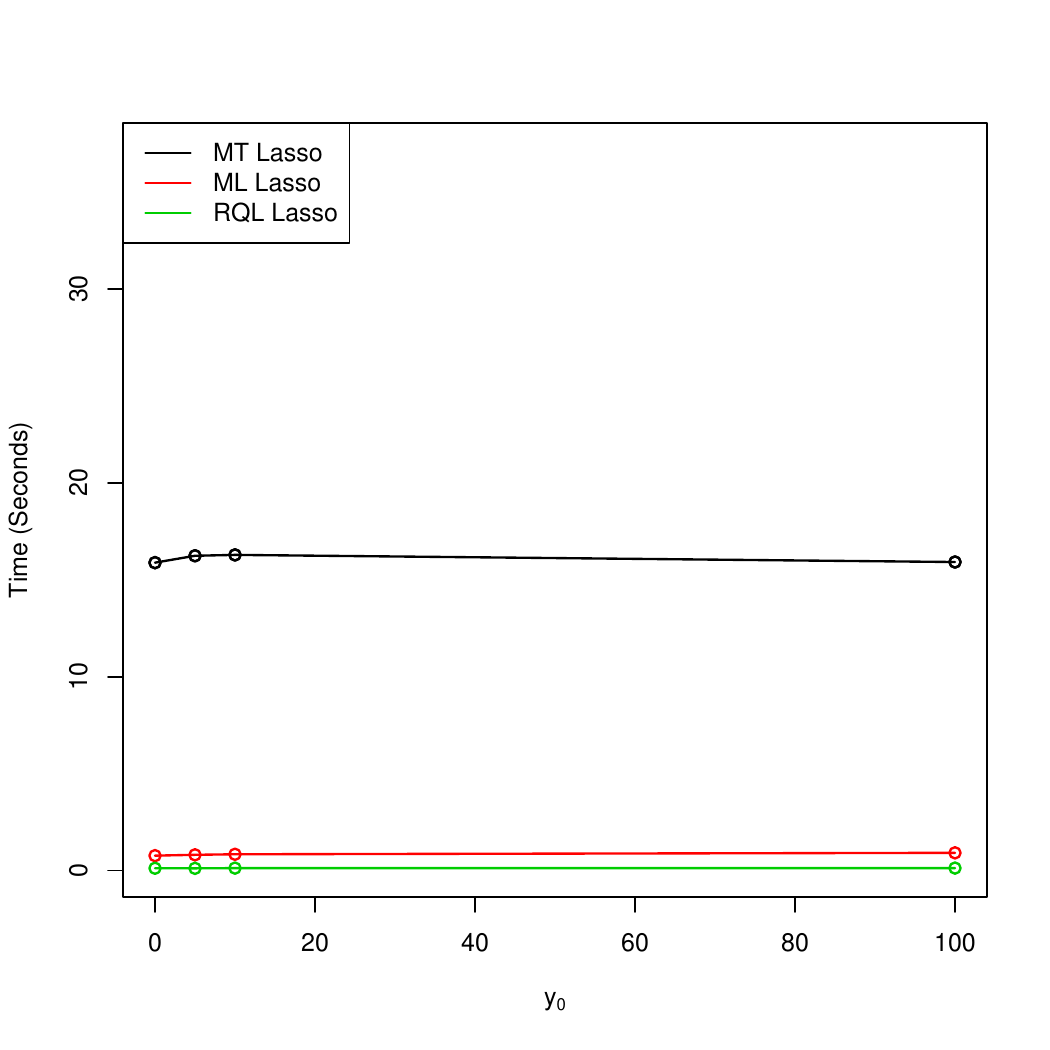}
			&   \includegraphics[scale=0.29]{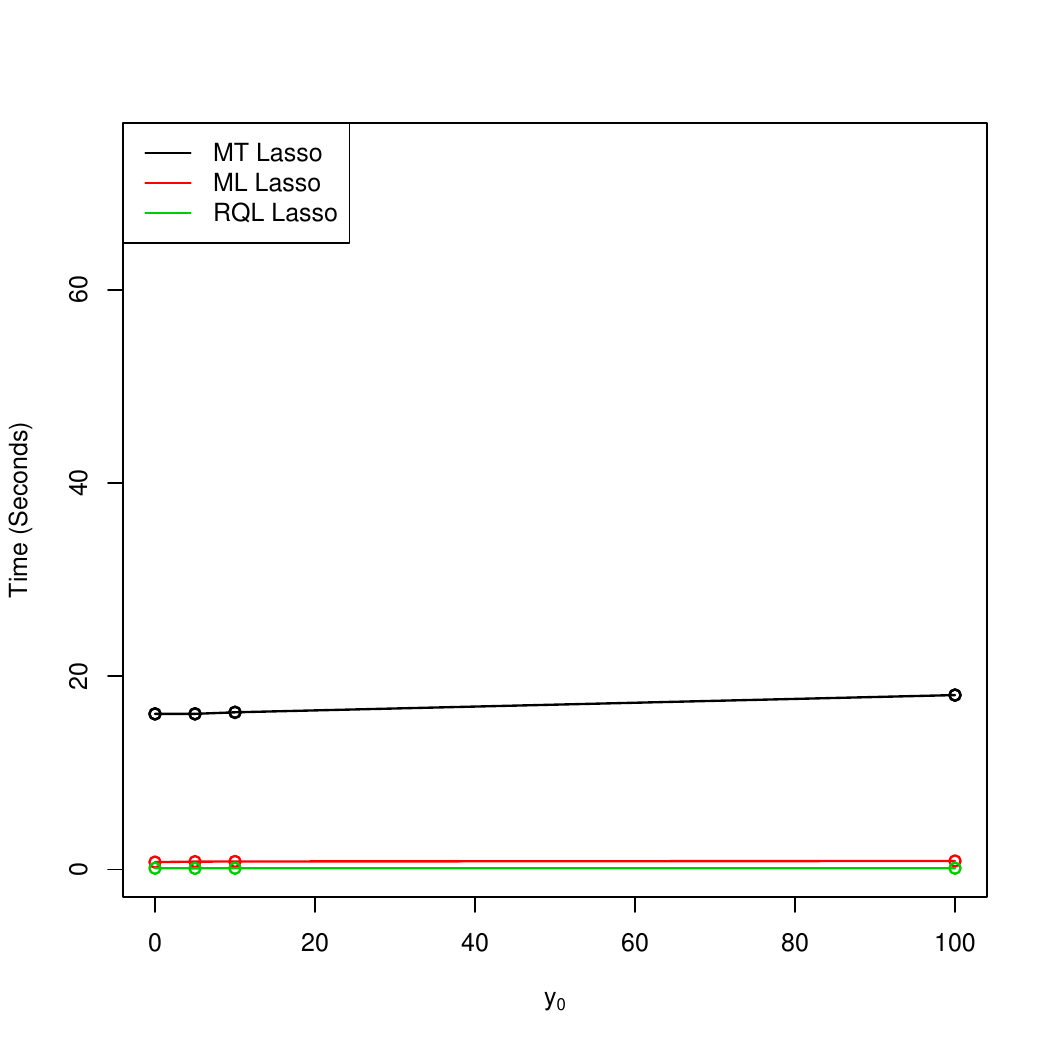}
			\\
			$n=100,\, \epsilon=0.05$ & $n=100,\, \epsilon=0.1$  & $n=100,\, \epsilon=0.15$ \\
			\includegraphics[scale=0.3]{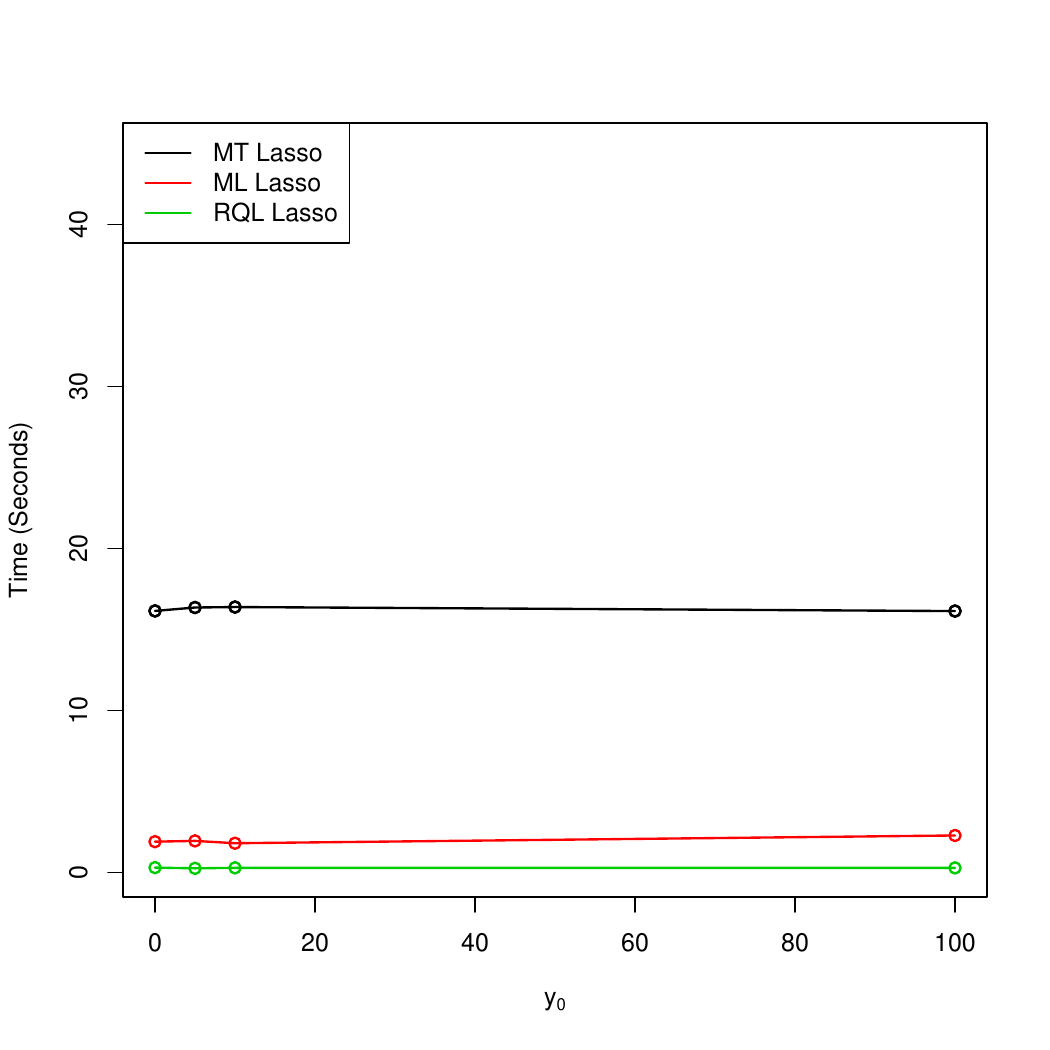}
			&   \includegraphics[scale=0.3]{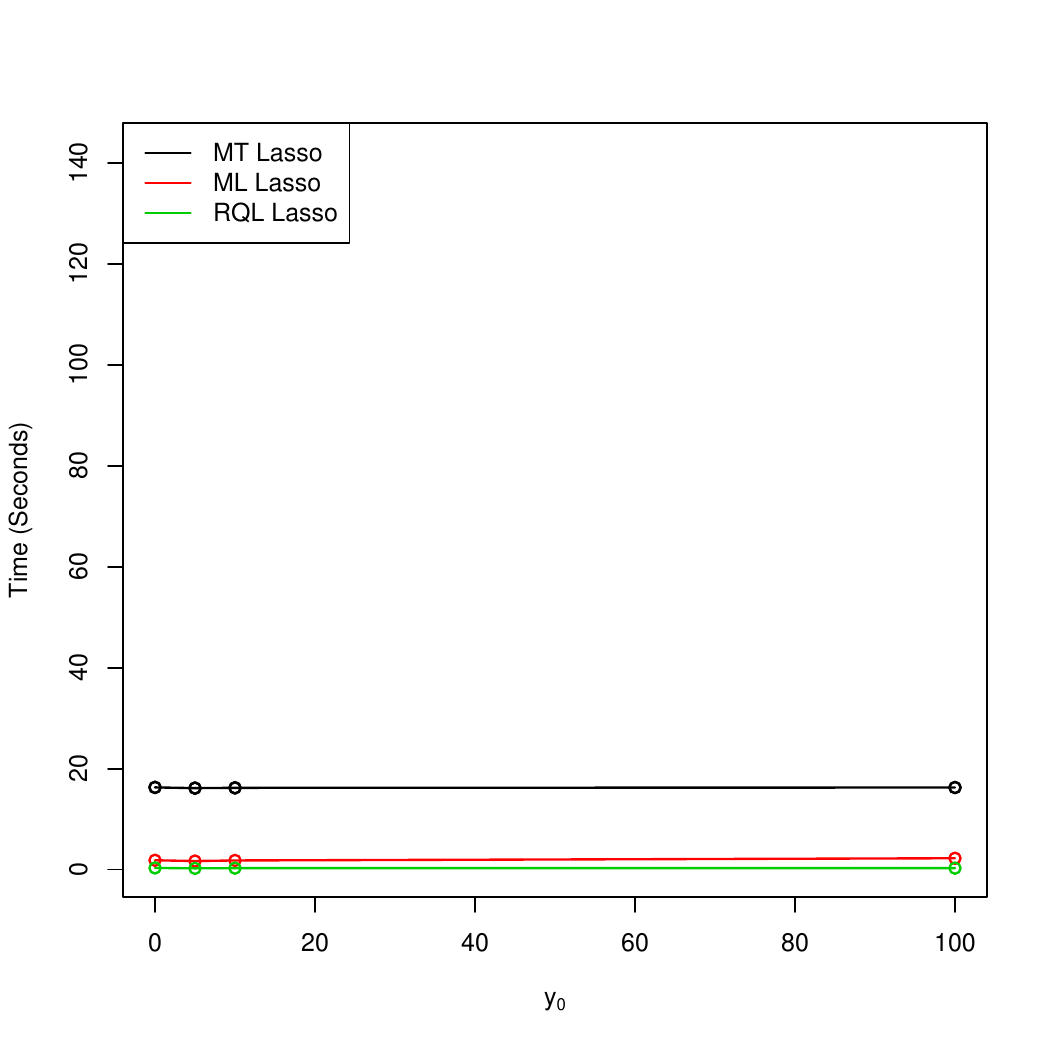} 
			&   \includegraphics[scale=0.3]{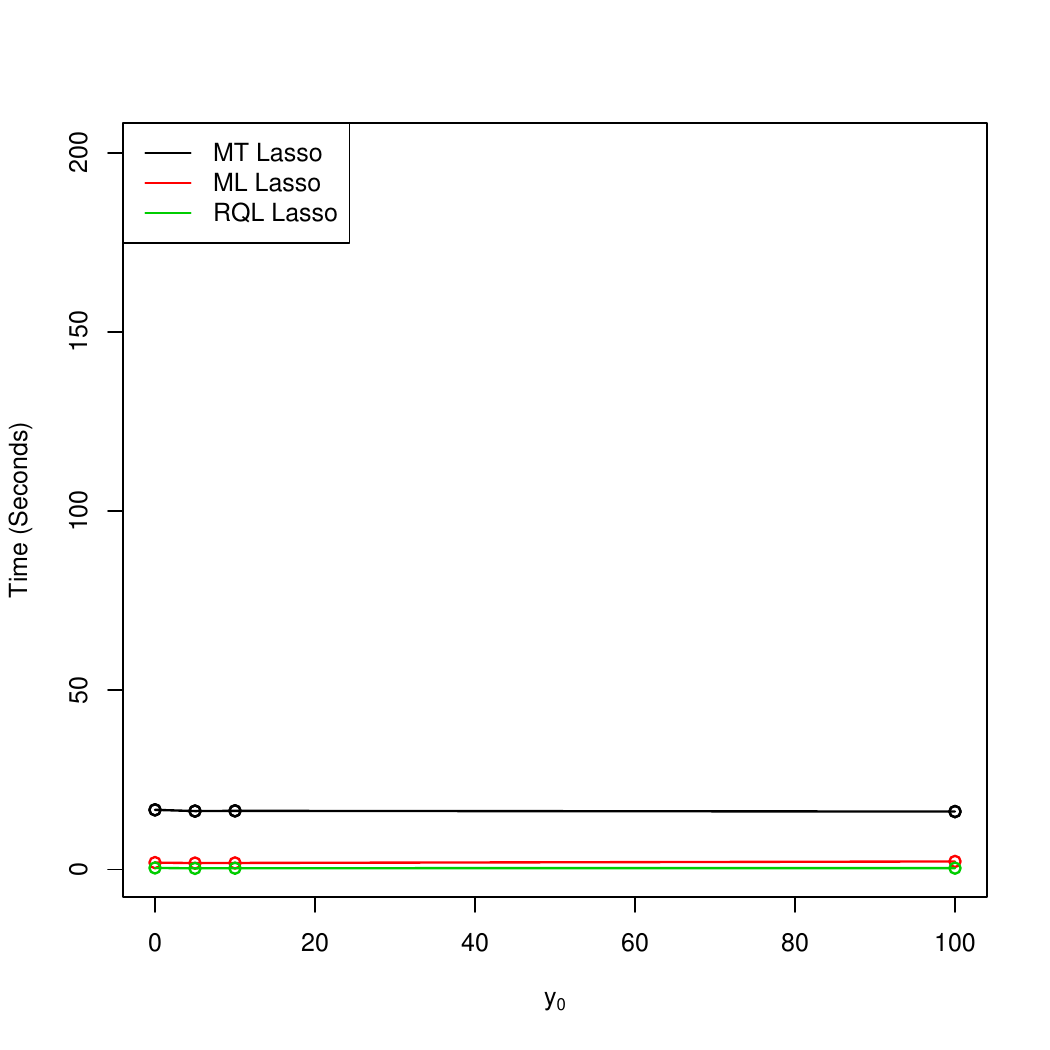} 
			\\
			$n=200,\, \epsilon=0.05$ & $n=200,\, \epsilon=0.1$  & $n=200,\, \epsilon=0.15$ \\
			\includegraphics[scale=0.3]{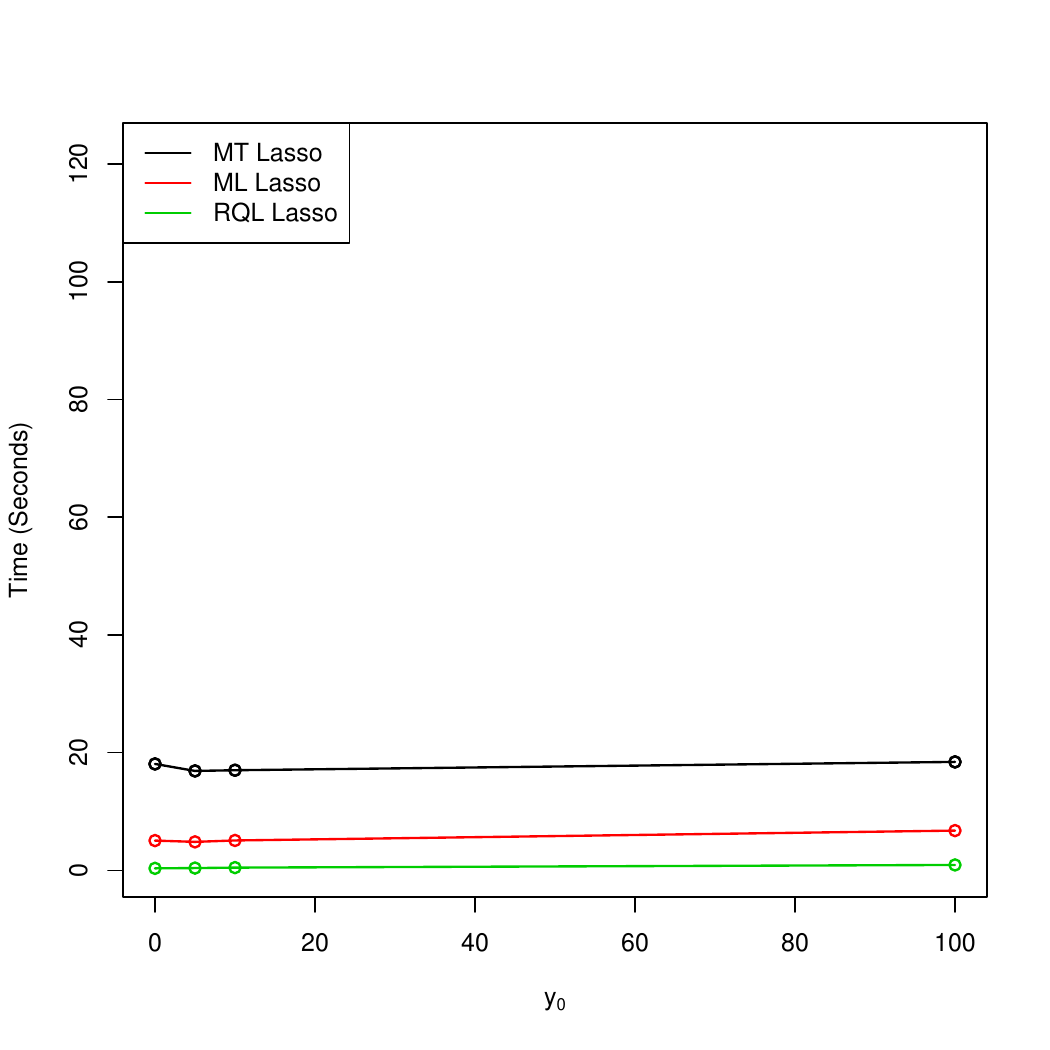}
			&   \includegraphics[scale=0.3]{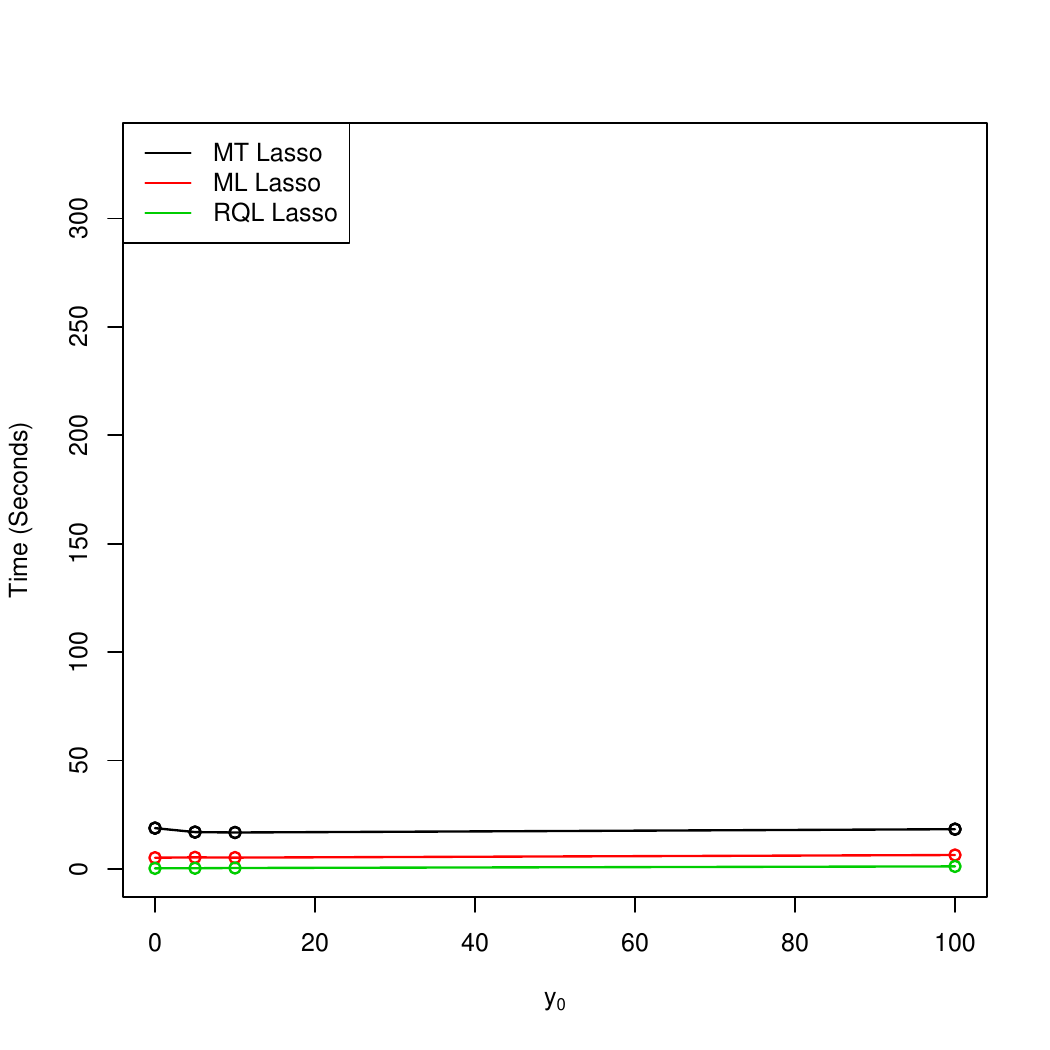} 
			&   \includegraphics[scale=0.3]{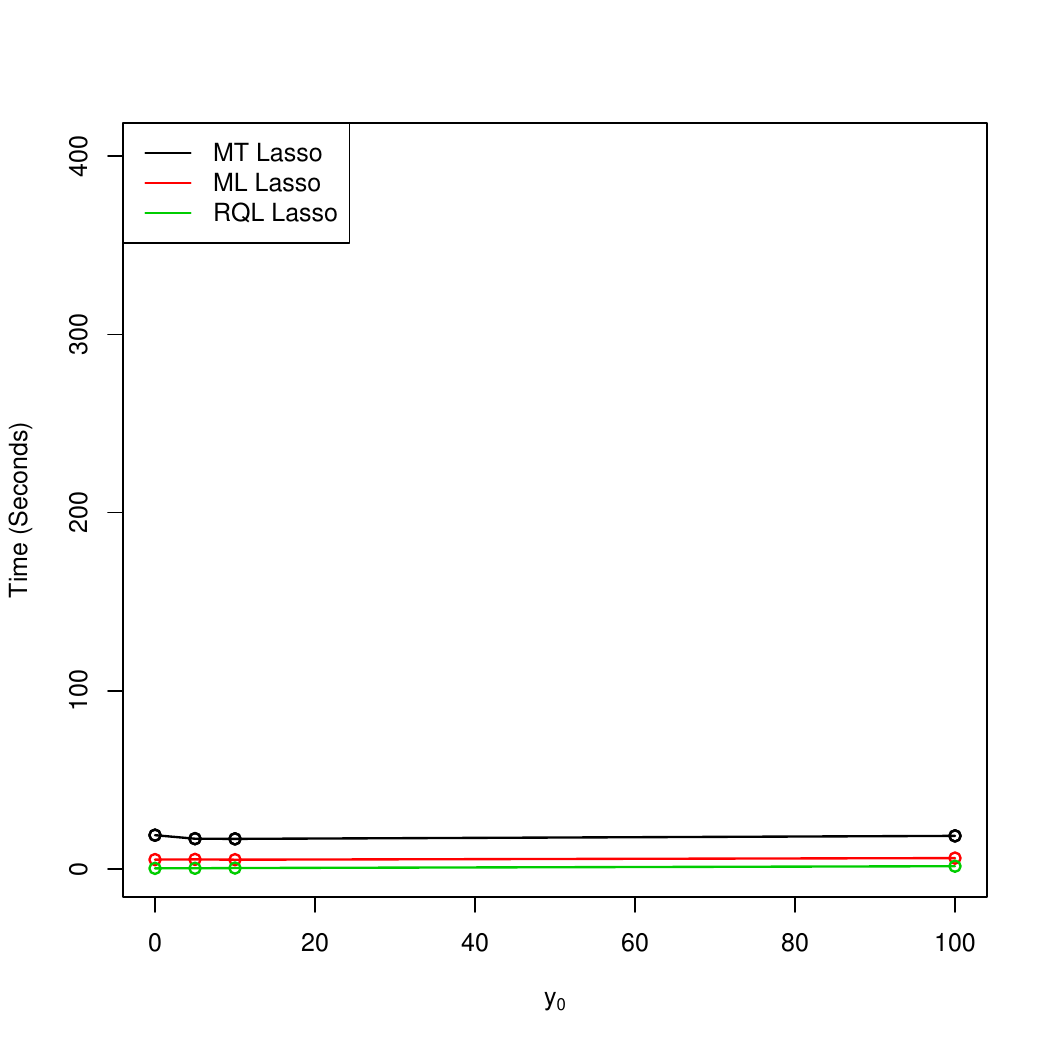}
		\end{tabular}
		\caption{\small \label{fig:timesimuRQLp100Lasso}  Computation times for model AMR and Lasso methods, $p=100$.}
	\end{center} 
\end{figure}

\begin{figure}[tp]
	\begin{center}
		\footnotesize
		\renewcommand{\arraystretch}{0.22}
		\newcolumntype{M}{>{\centering\arraybackslash}m{\dimexpr.01\linewidth-1\tabcolsep}}
		\newcolumntype{G}{>{\centering\arraybackslash}m{\dimexpr.38\linewidth-1\tabcolsep}}
		\begin{tabular}{GGG}
			$n=50,\, \epsilon=0.05$ & $n=50,\, \epsilon=0.1$  & $n=50,\, \epsilon=0.15$ \\
			\includegraphics[scale=0.29]{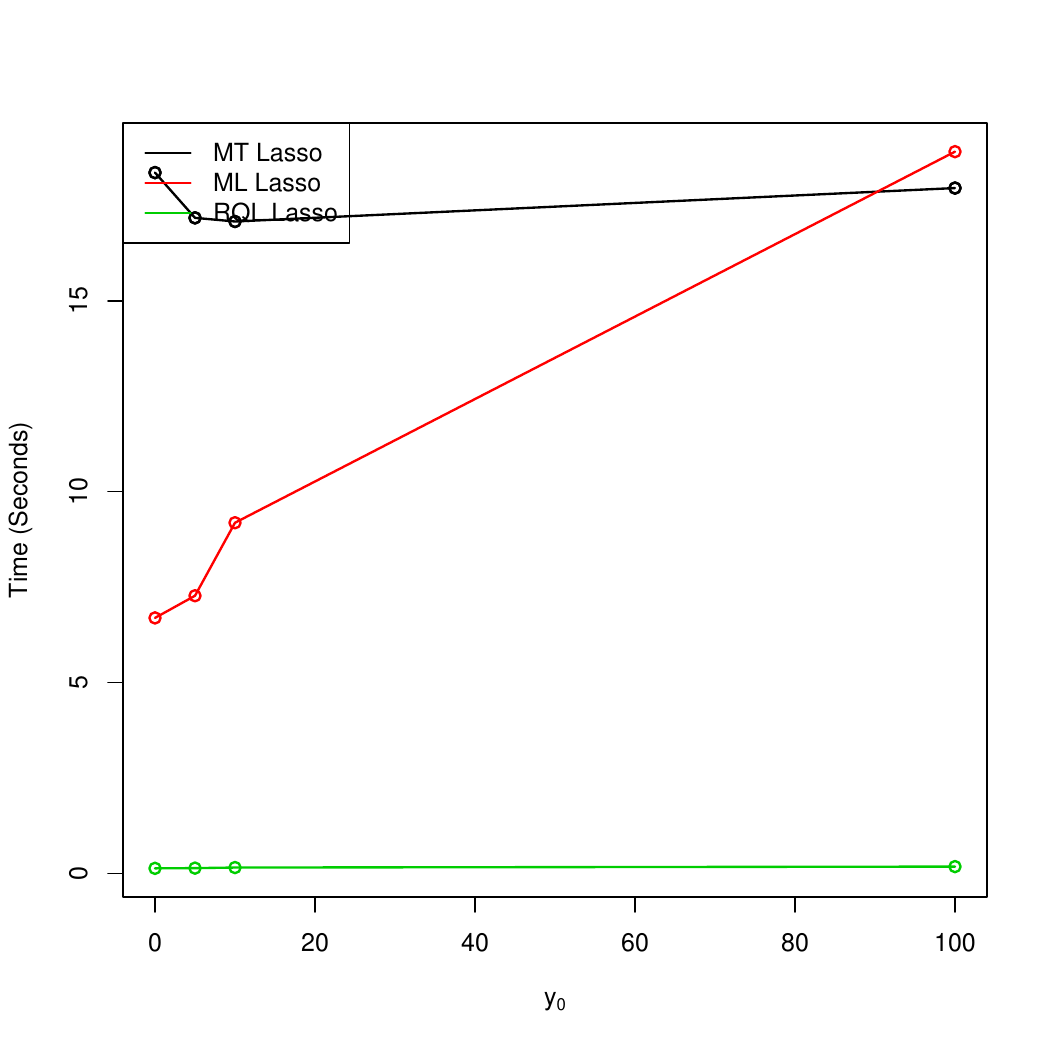} 
			&  \includegraphics[scale=0.29]{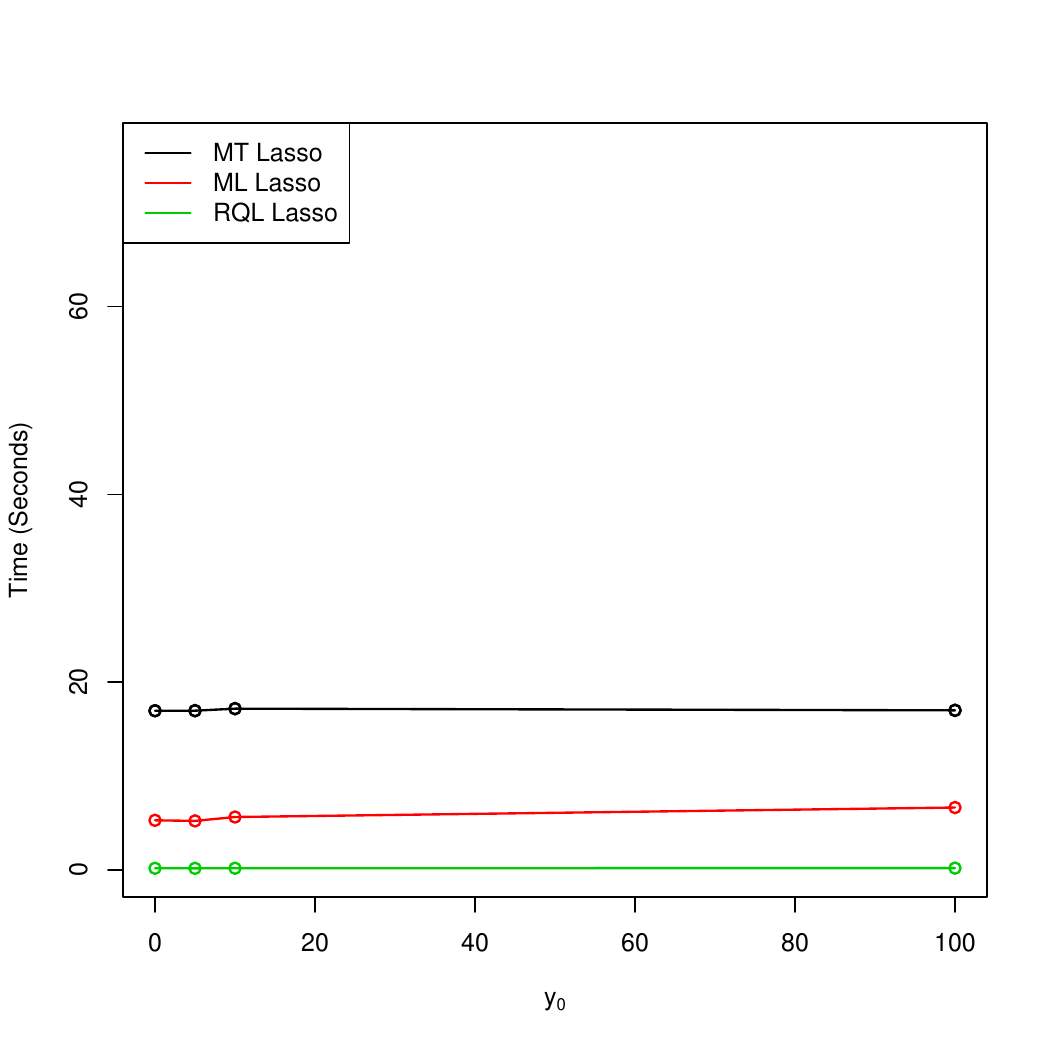}
			&   \includegraphics[scale=0.29]{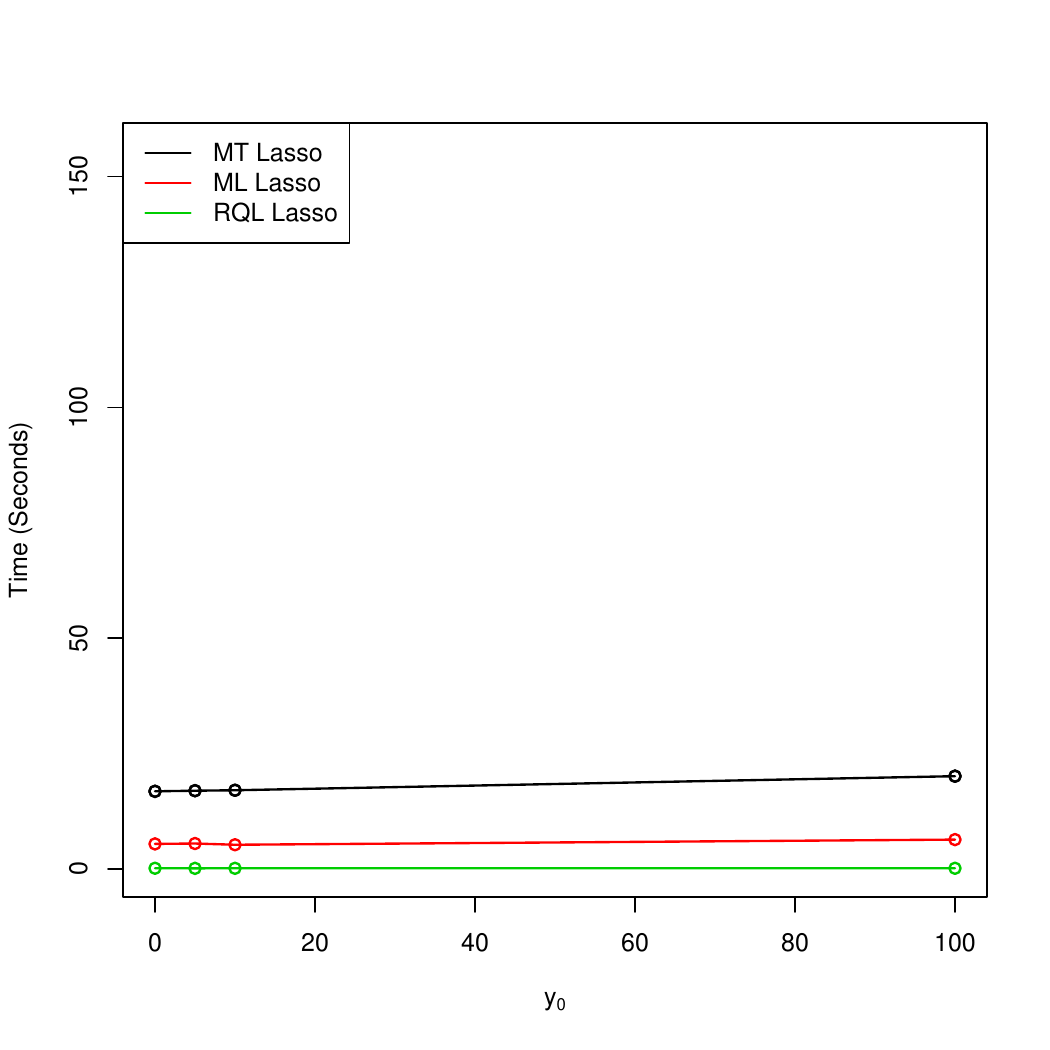}
			\\
			$n=100,\, \epsilon=0.05$ & $n=100,\, \epsilon=0.1$  & $n=100,\, \epsilon=0.15$ \\
			\includegraphics[scale=0.3]{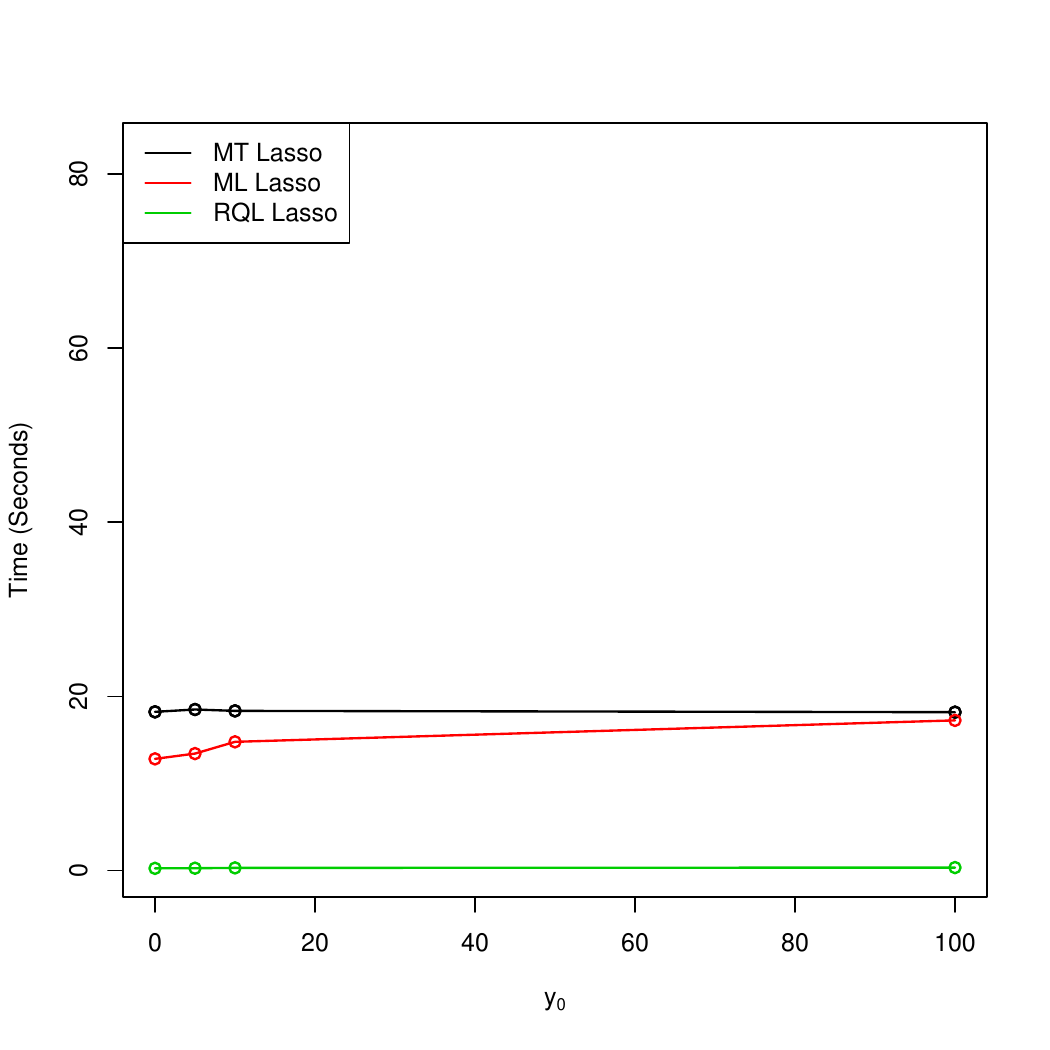}
			&   \includegraphics[scale=0.3]{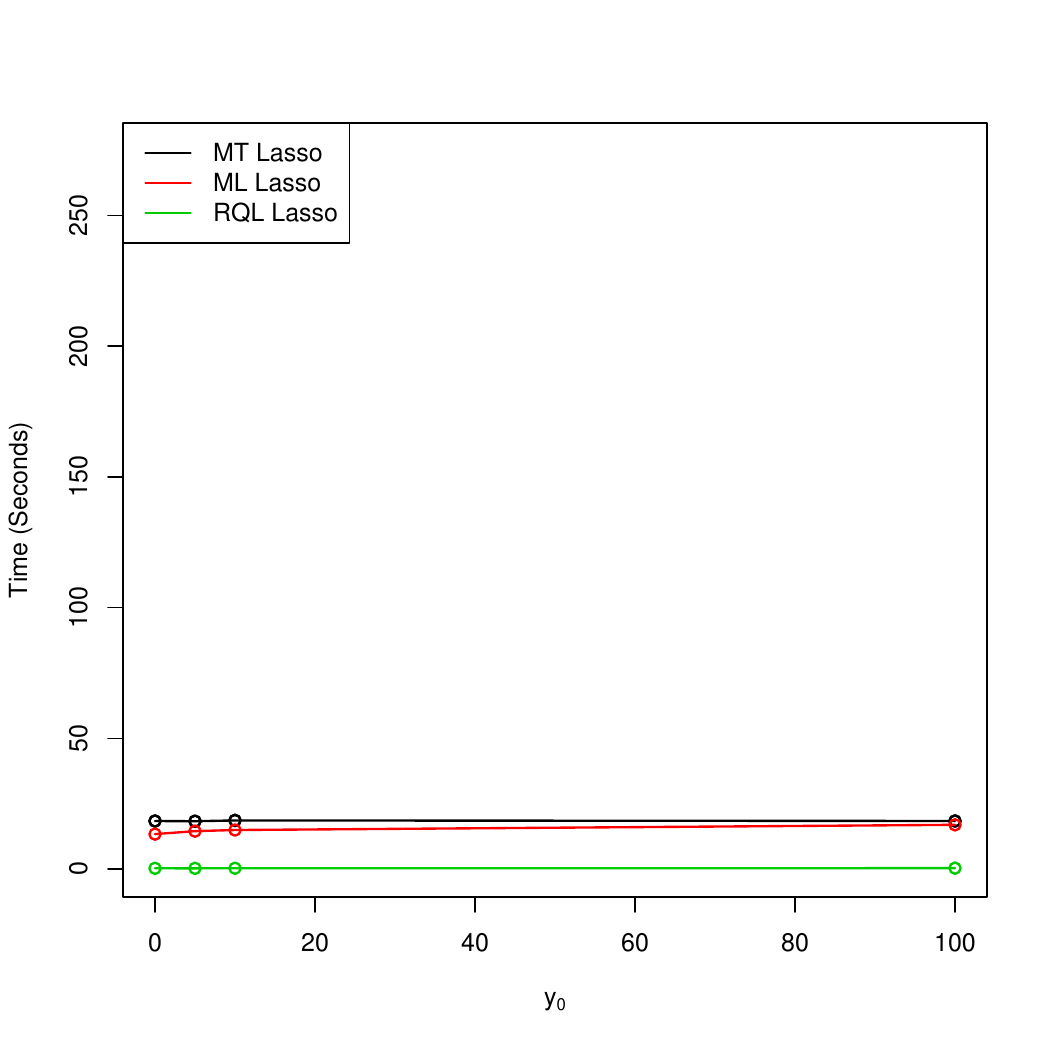} 
			&   \includegraphics[scale=0.3]{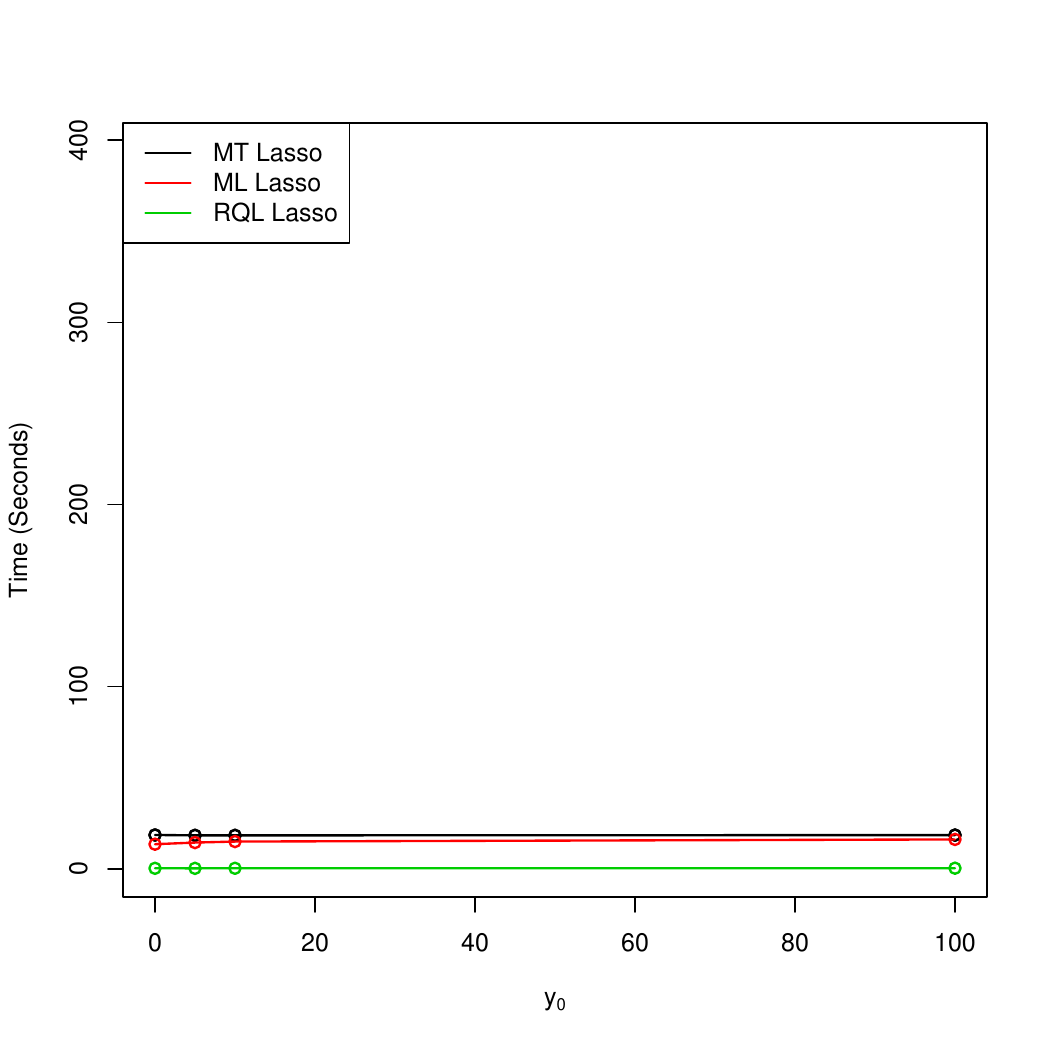} 
			\\
			$n=200,\, \epsilon=0.05$ & $n=200,\, \epsilon=0.1$  & $n=200,\, \epsilon=0.15$ \\
			\includegraphics[scale=0.3]{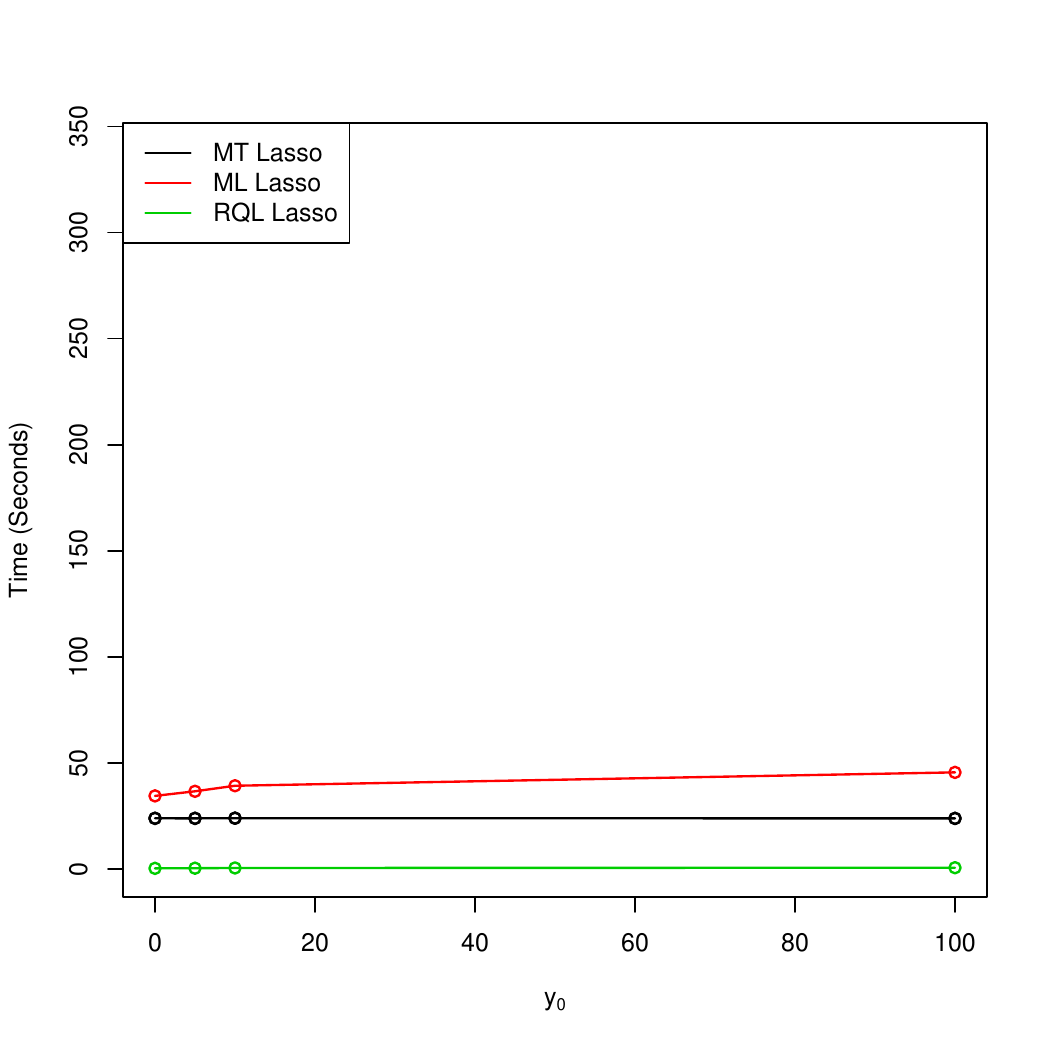}
			&   \includegraphics[scale=0.3]{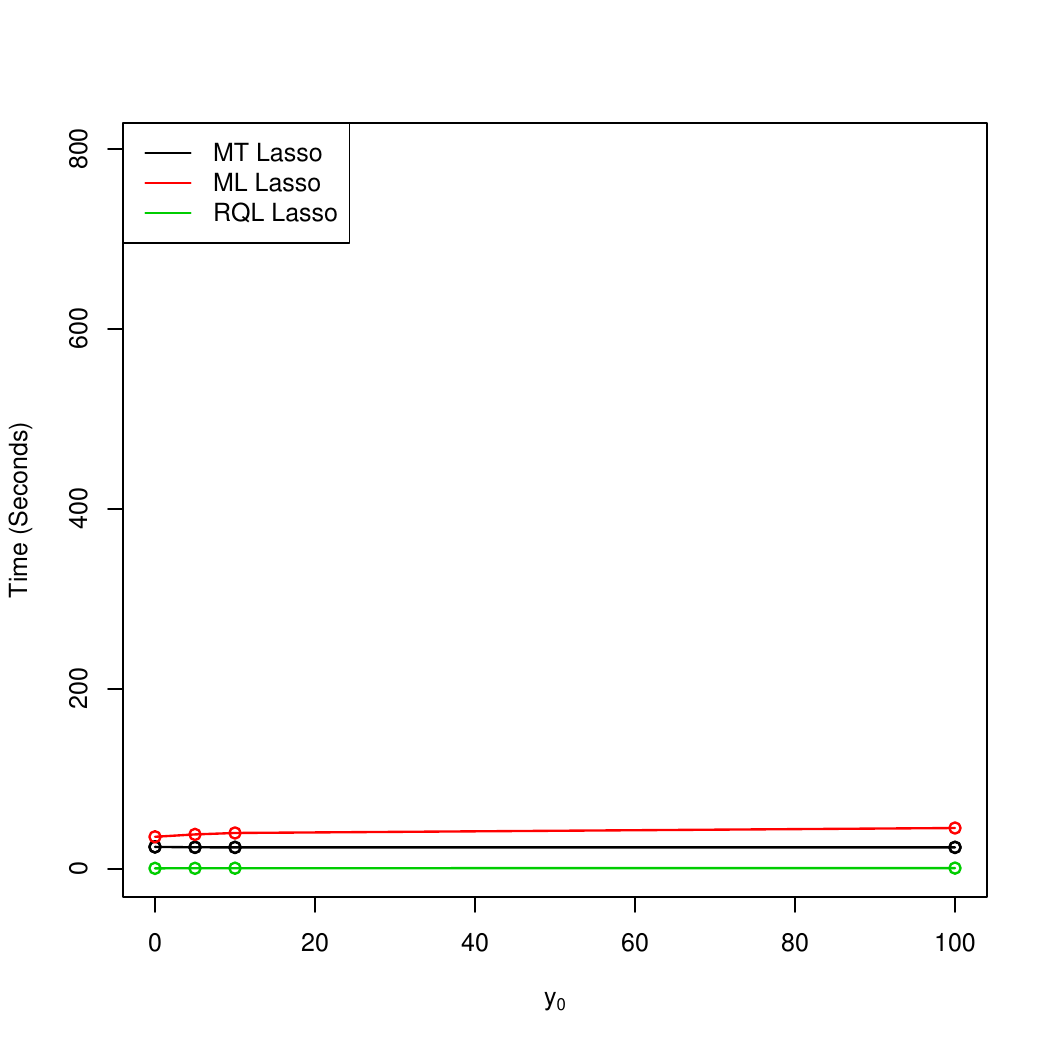} 
			&   \includegraphics[scale=0.3]{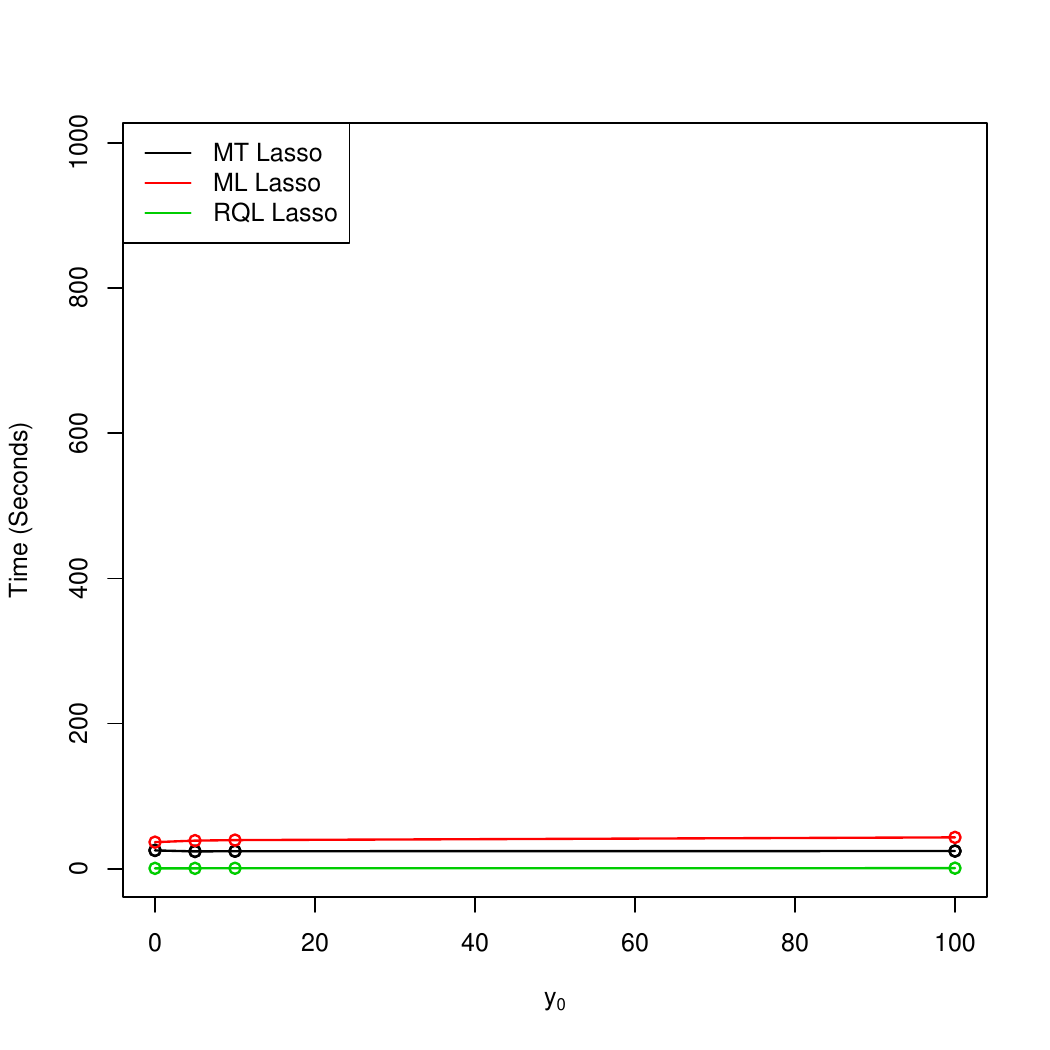}
		\end{tabular}
		\caption{\small \label{fig:timesimuRQLp400Lasso}  Computation times for model AMR and Lasso methods, $p=400$.}
	\end{center} 
\end{figure}

\begin{figure}[tp]
	\begin{center}
		\footnotesize
		\renewcommand{\arraystretch}{0.22}
		\newcolumntype{M}{>{\centering\arraybackslash}m{\dimexpr.01\linewidth-1\tabcolsep}}
		\newcolumntype{G}{>{\centering\arraybackslash}m{\dimexpr.38\linewidth-1\tabcolsep}}
		\begin{tabular}{GGG}
			$n=50,\, \epsilon=0.05$ & $n=50,\, \epsilon=0.1$  & $n=50,\, \epsilon=0.15$ \\
			\includegraphics[scale=0.29]{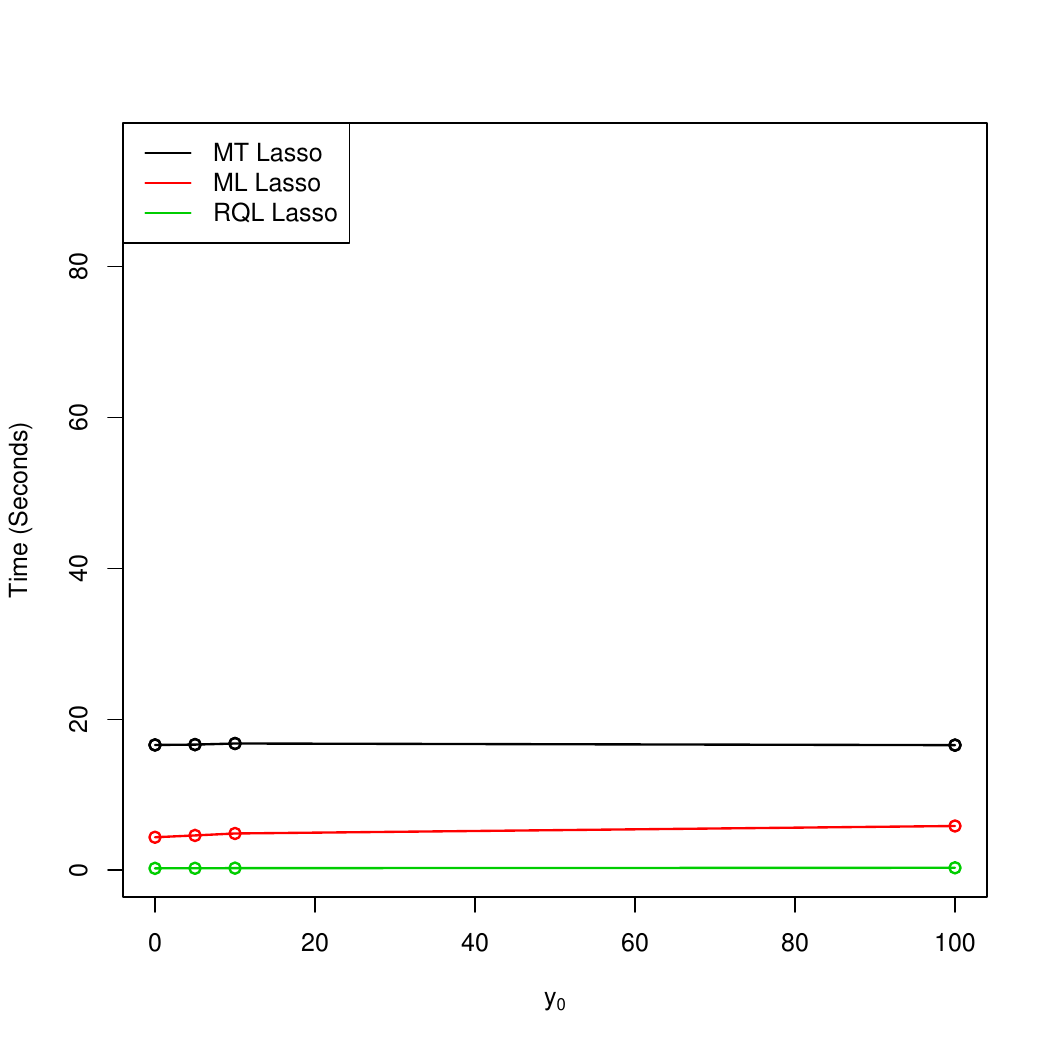} 
			&  \includegraphics[scale=0.29]{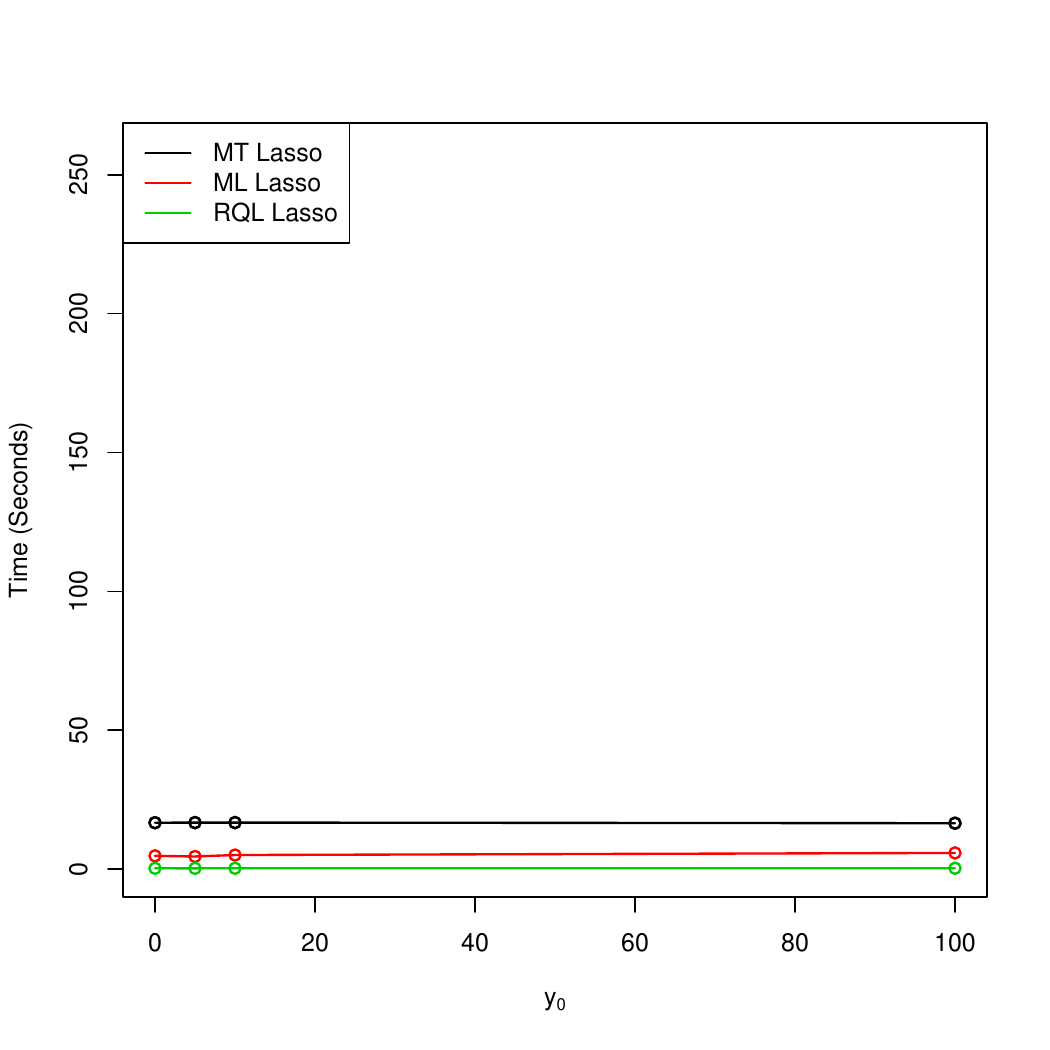}
			&   \includegraphics[scale=0.29]{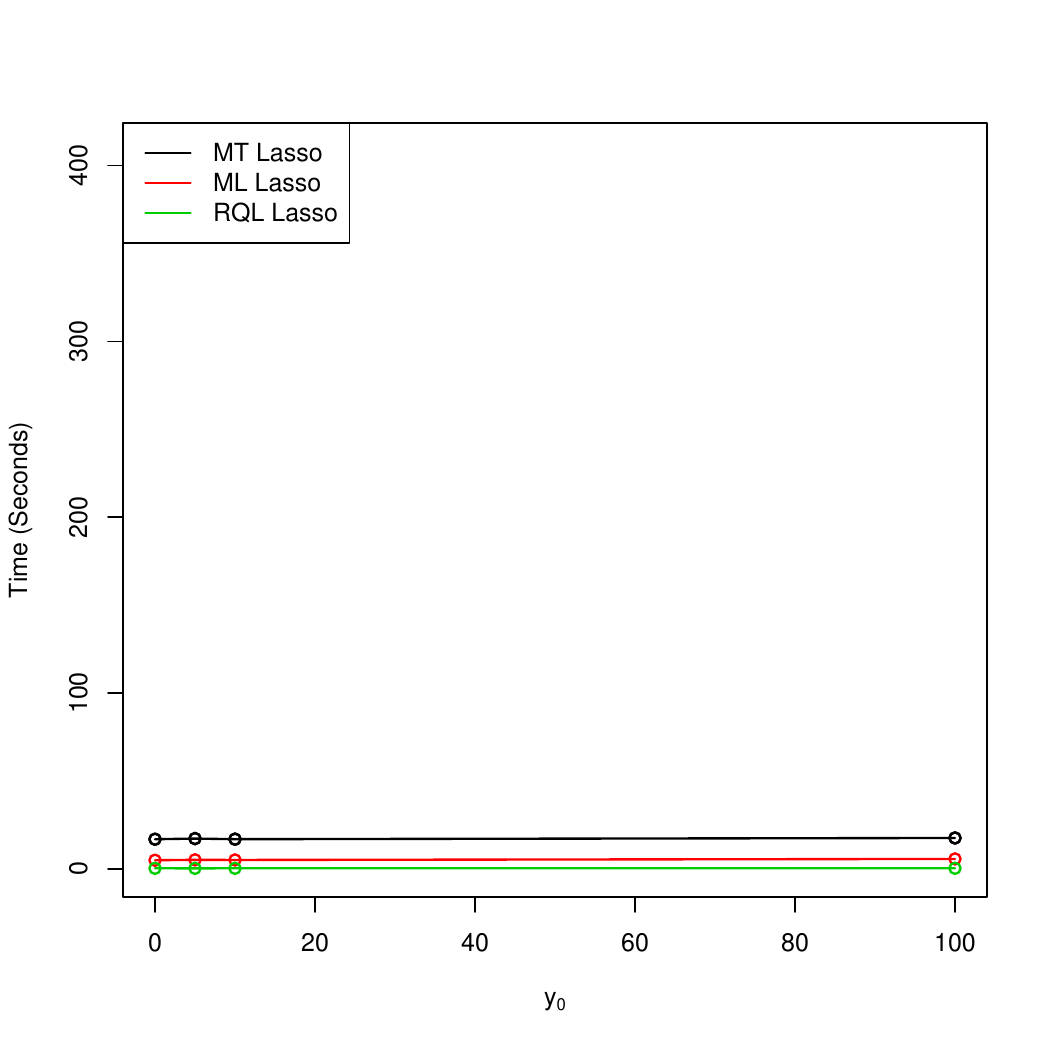}
			\\
			$n=100,\, \epsilon=0.05$ & $n=100,\, \epsilon=0.1$  & $n=100,\, \epsilon=0.15$ \\
			\includegraphics[scale=0.3]{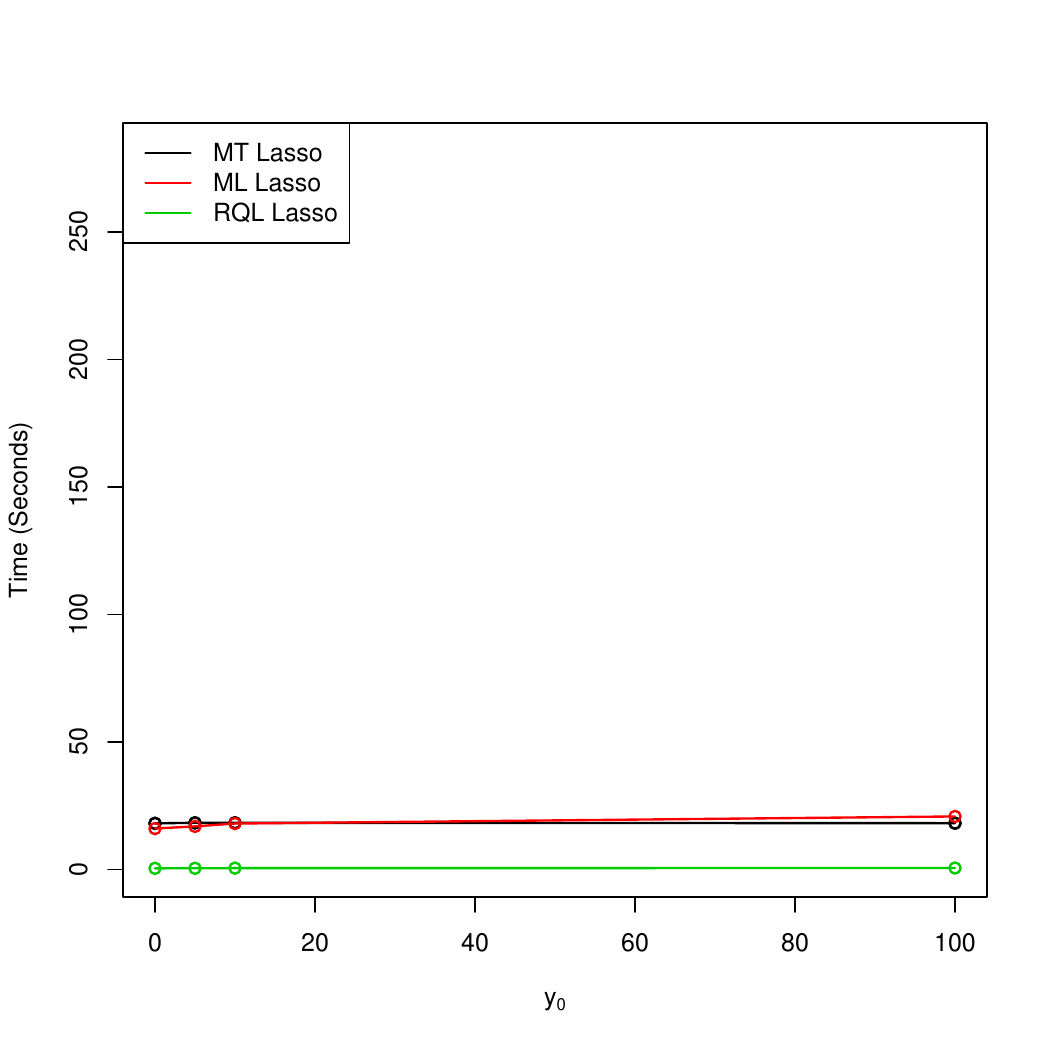}
			&   \includegraphics[scale=0.3]{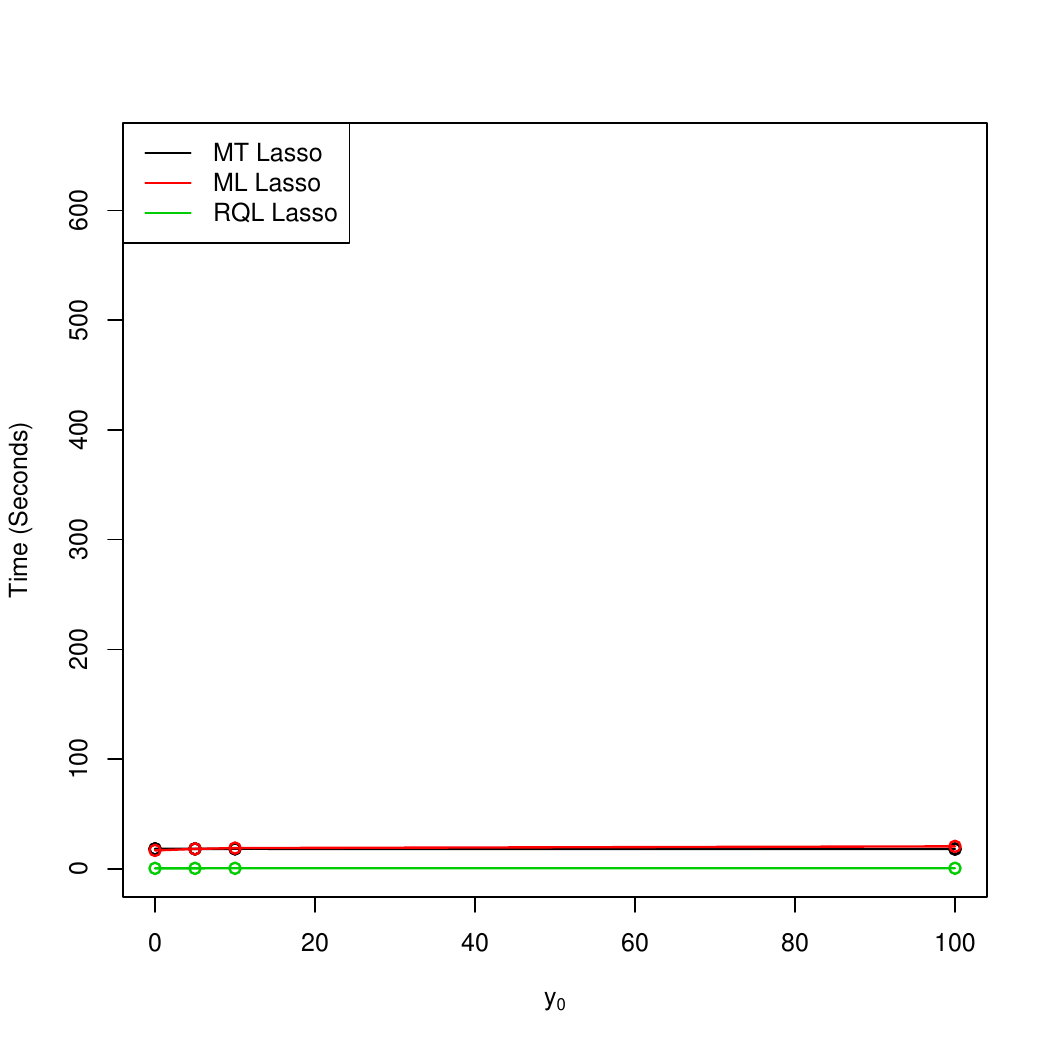} 
			&   \includegraphics[scale=0.3]{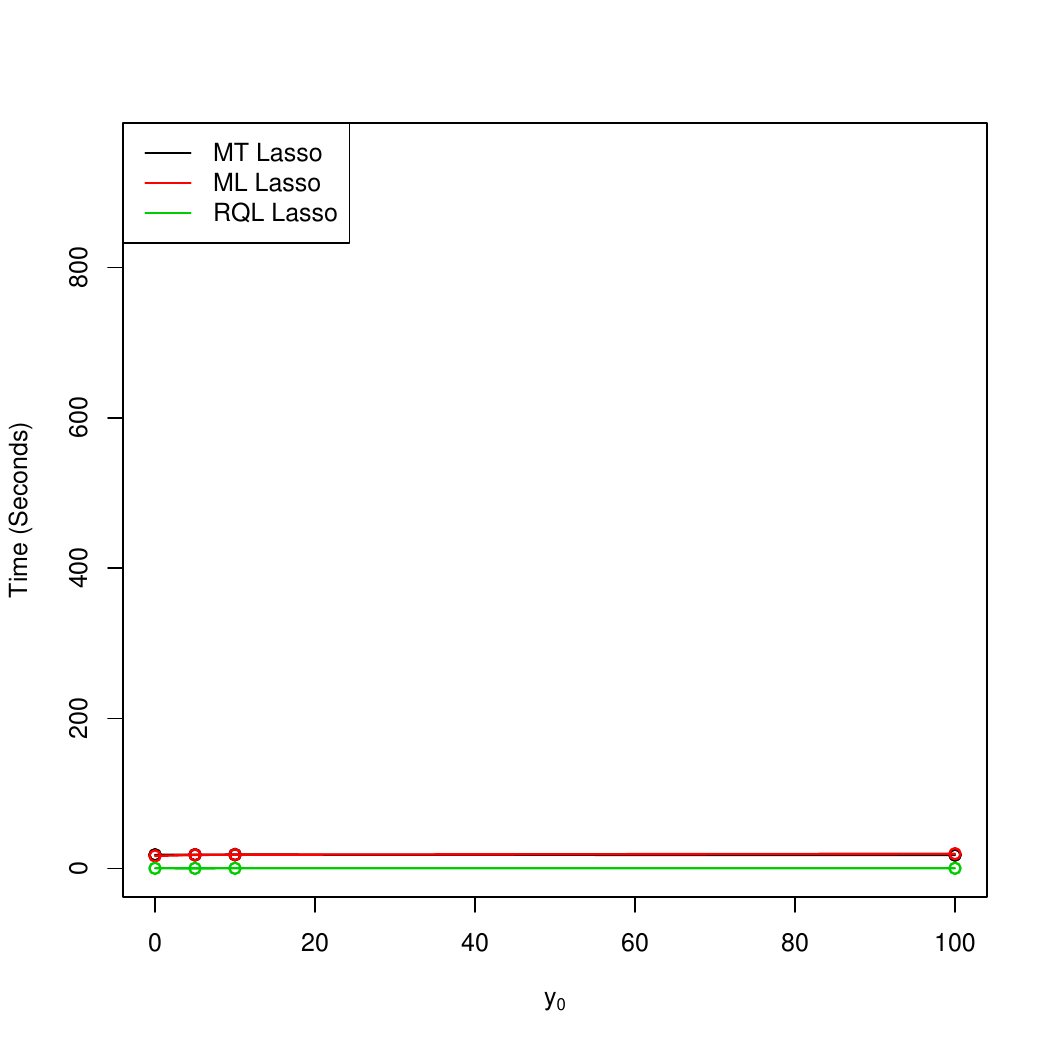} 
			\\
			$n=200,\, \epsilon=0.05$ & $n=200,\, \epsilon=0.1$  & $n=200,\, \epsilon=0.15$ \\
			\includegraphics[scale=0.3]{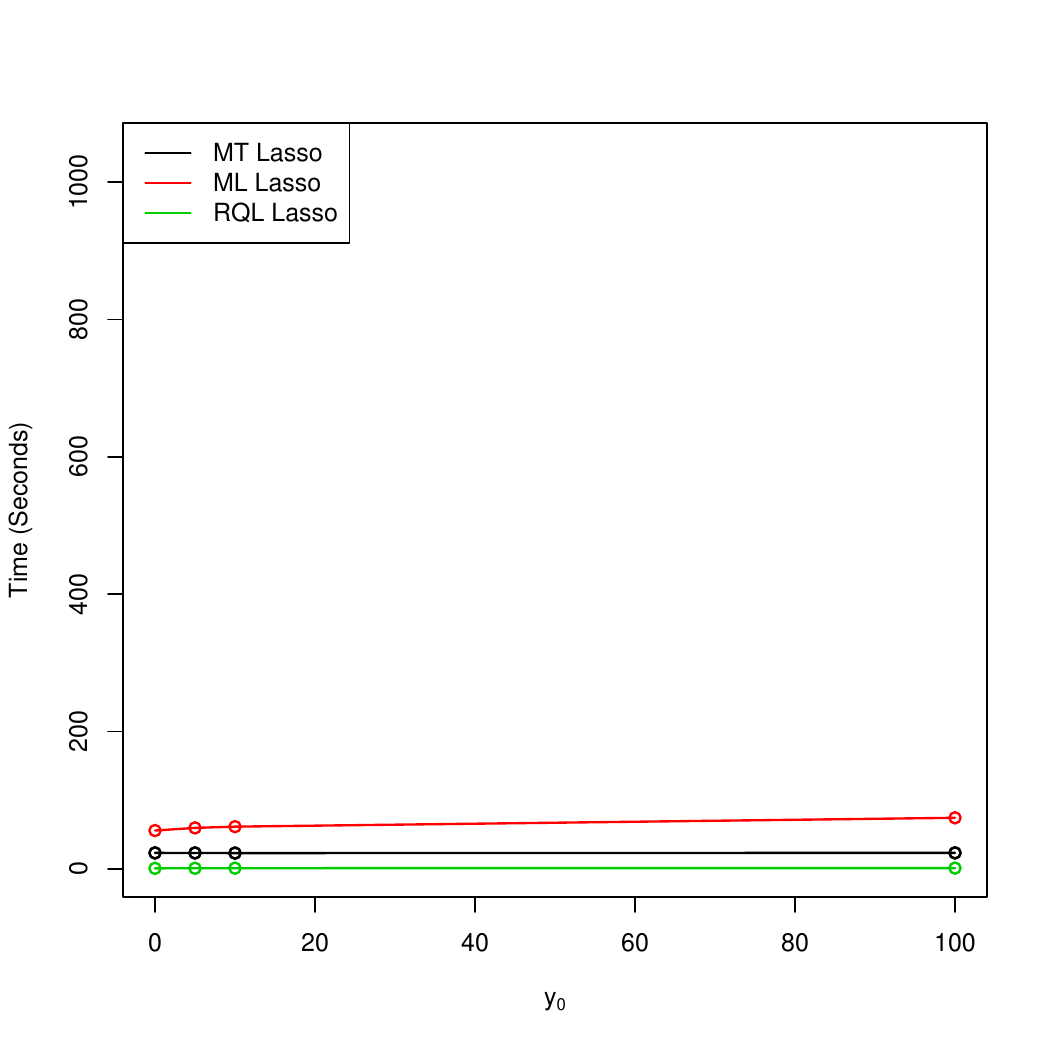}
			&   \includegraphics[scale=0.3]{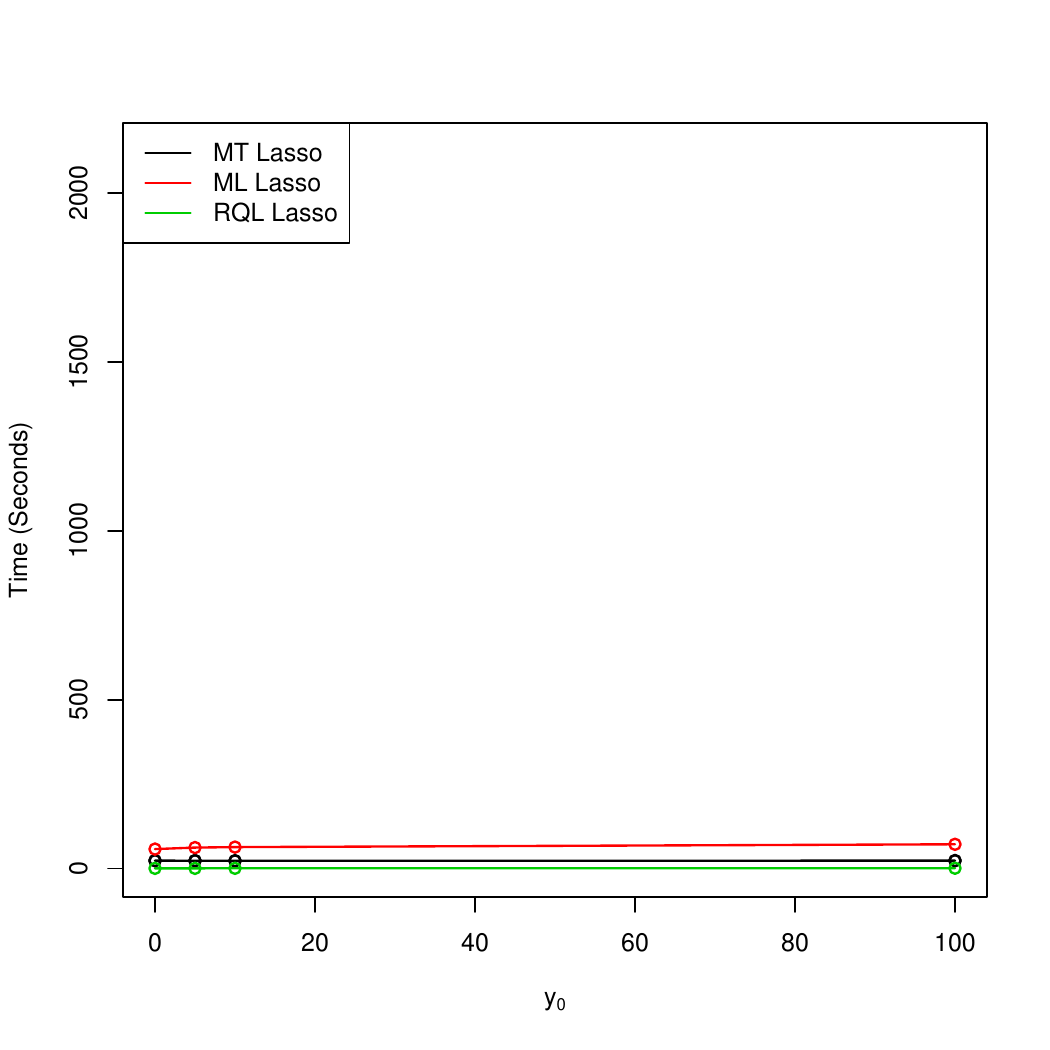} 
			&   \includegraphics[scale=0.3]{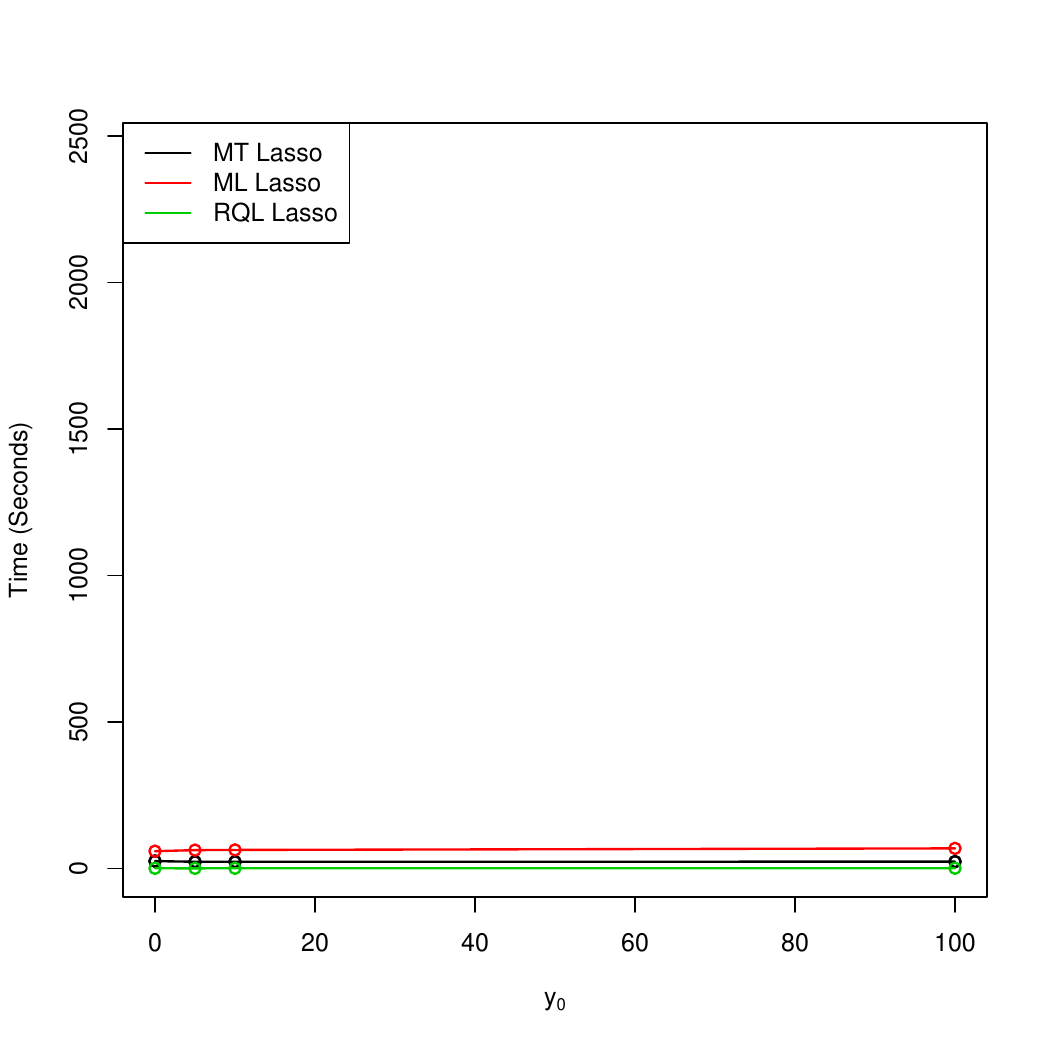}
		\end{tabular}
		\caption{\small \label{fig:timesimuRQLp1600Lasso}  Computation times for model AMR and Lasso methods, $p=1600$.}
	\end{center} 
\end{figure}

%%%%%%%%%%%%%%%%%%%%% Ridge for avella medina setting %%%%%%%%%%

\begin{figure}[tp]
	\begin{center}
		\footnotesize
		\renewcommand{\arraystretch}{0.22}
		\newcolumntype{M}{>{\centering\arraybackslash}m{\dimexpr.01\linewidth-1\tabcolsep}}
		\newcolumntype{G}{>{\centering\arraybackslash}m{\dimexpr.38\linewidth-1\tabcolsep}}
		\begin{tabular}{GGG}
			$n=50,\, \epsilon=0.05$ & $n=100,\, \epsilon=0.1$  & $n=50,\, \epsilon=0.15$ \\
			\includegraphics[scale=0.29]{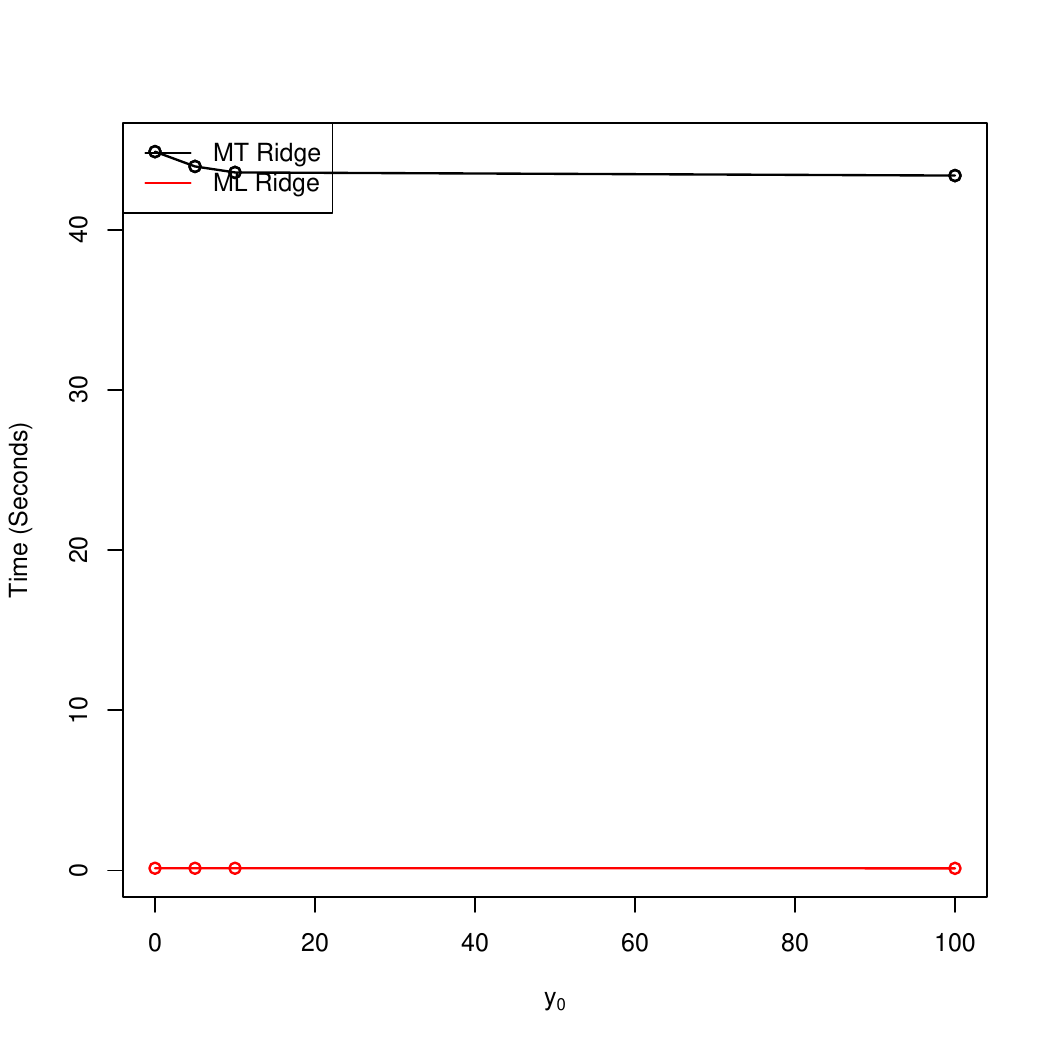} 
			&  \includegraphics[scale=0.29]{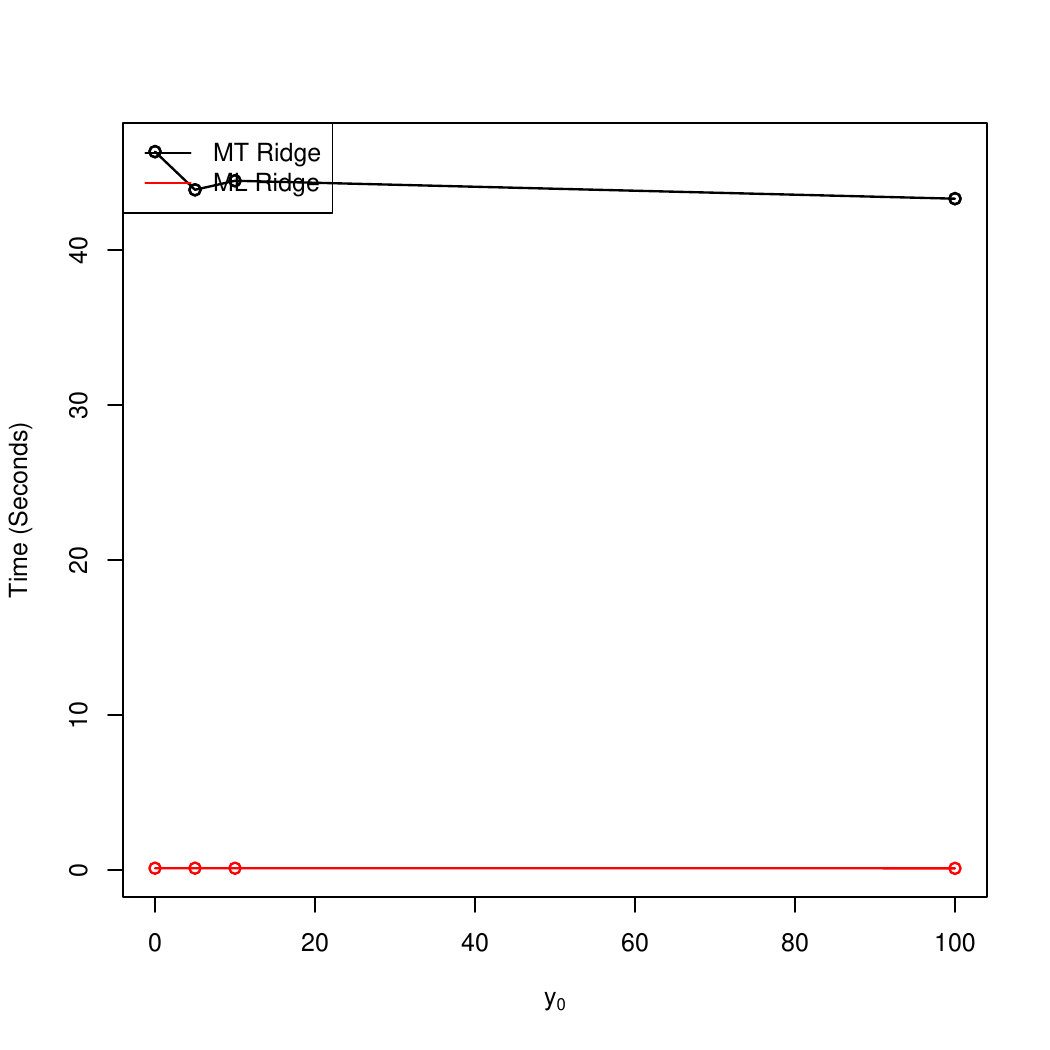}
			&   \includegraphics[scale=0.29]{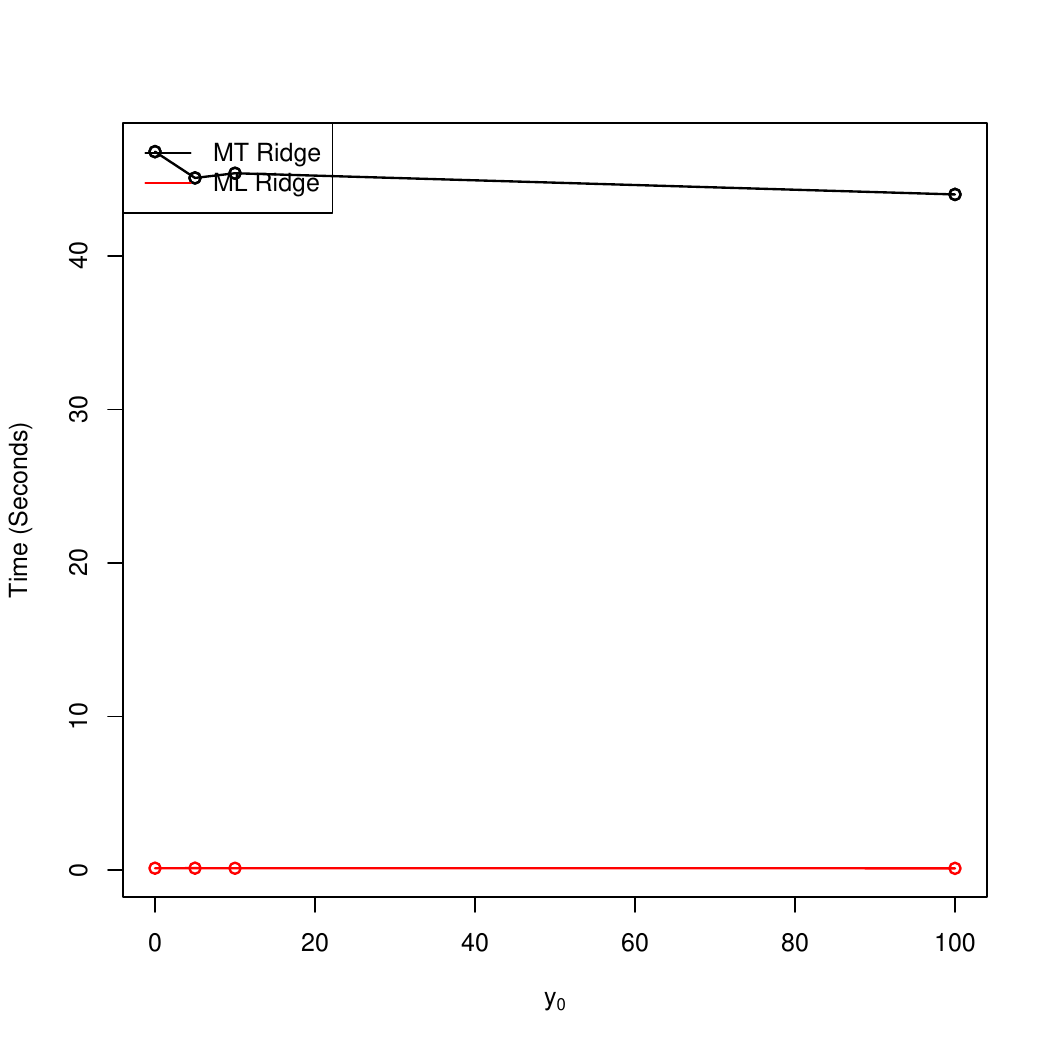}
			\\
			$n=100,\, \epsilon=0.05$ & $n=100,\, \epsilon=0.1$  & $n=100,\, \epsilon=0.15$ \\
			\includegraphics[scale=0.3]{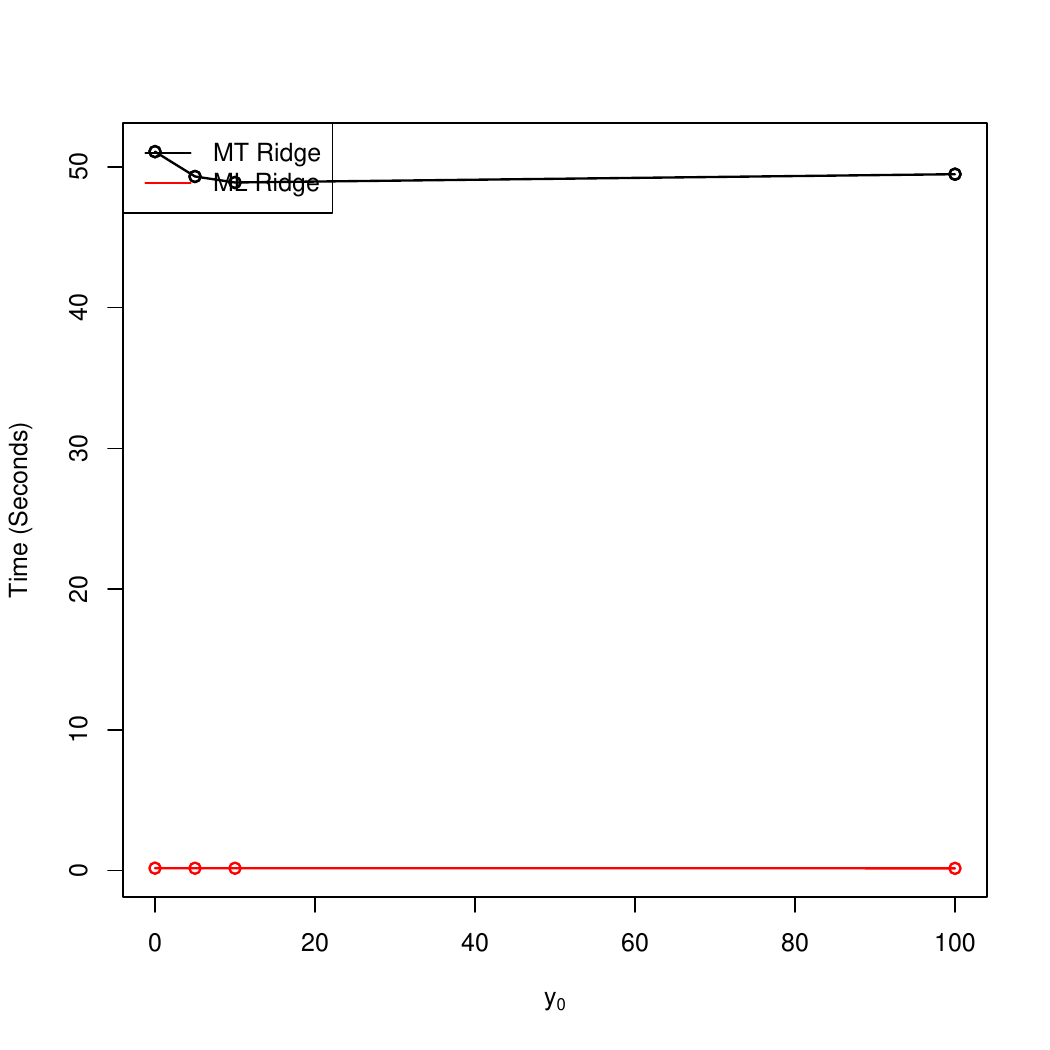}
			&   \includegraphics[scale=0.3]{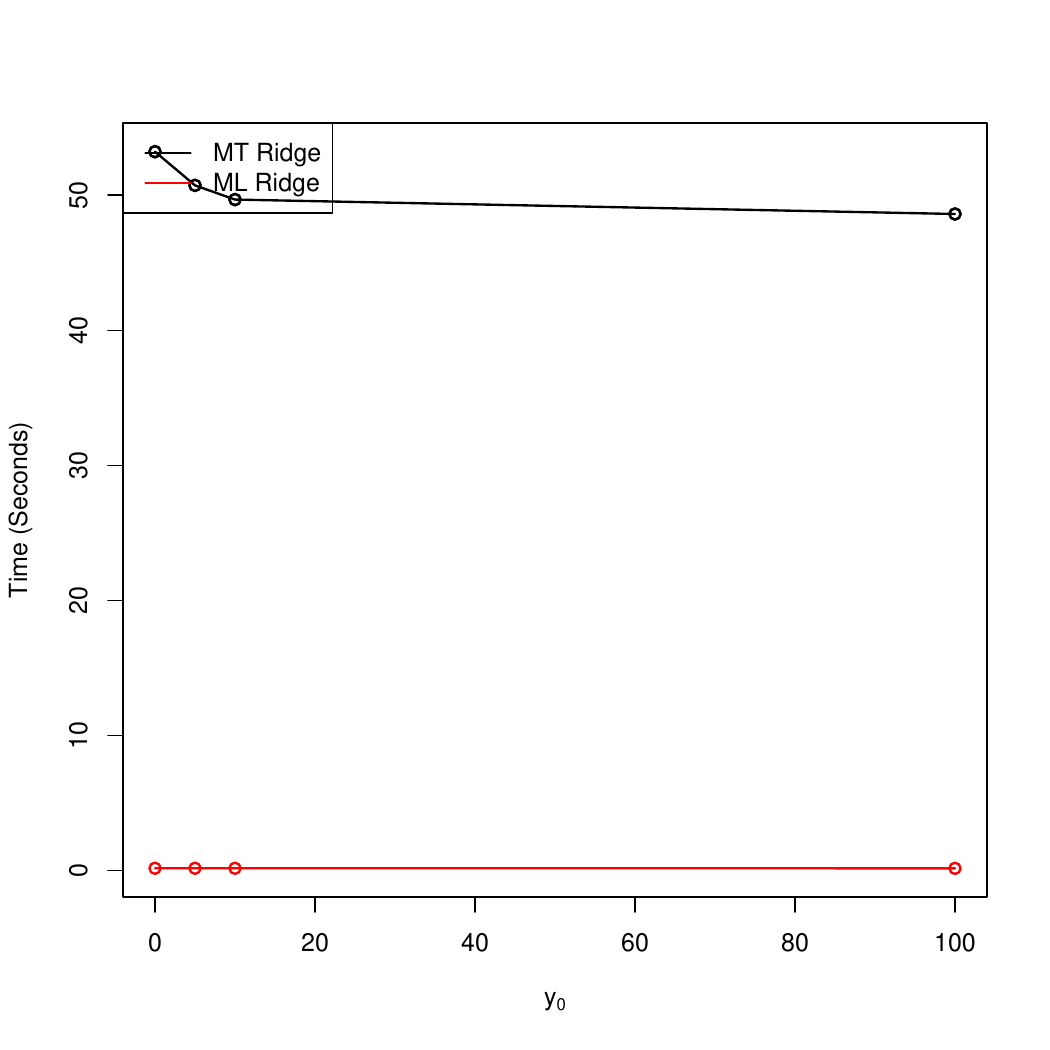} 
			&   \includegraphics[scale=0.3]{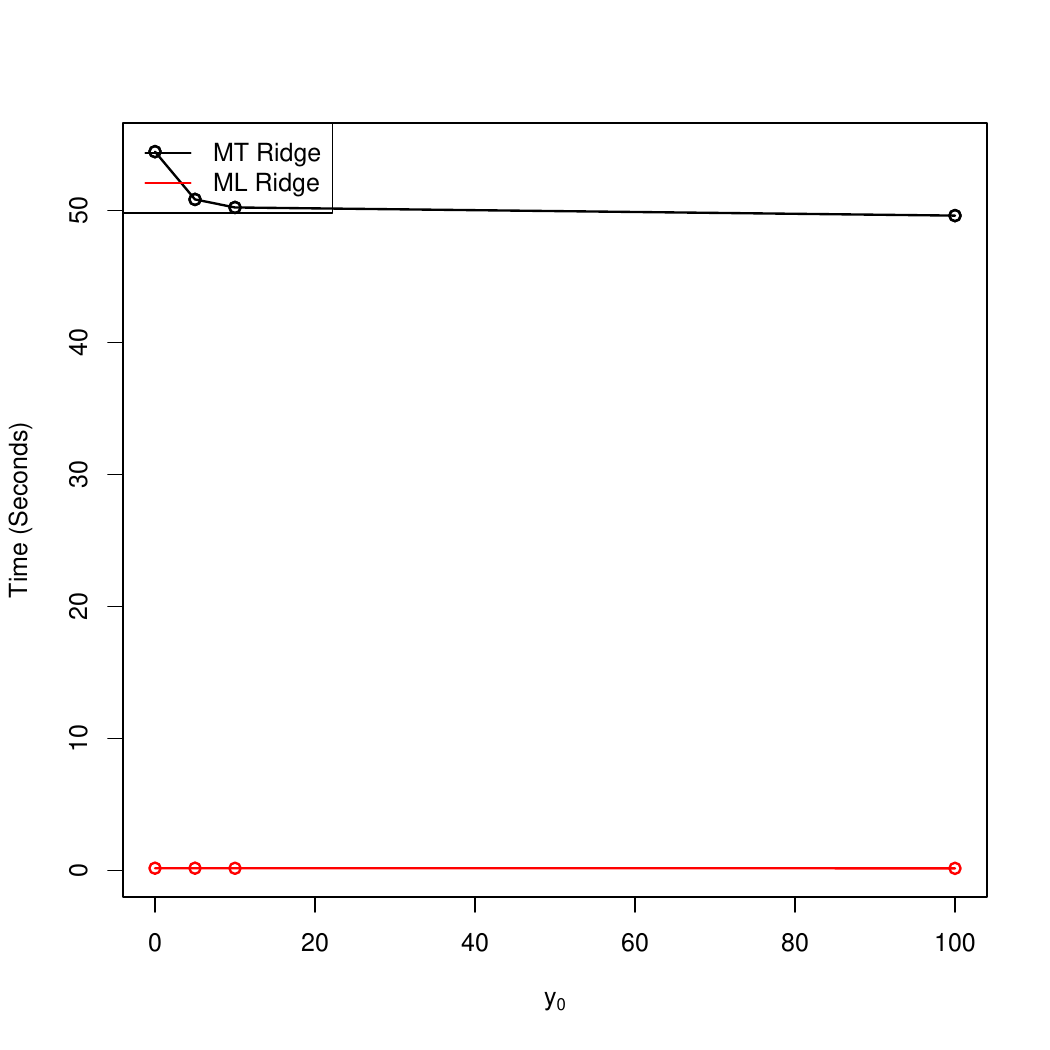} 
			\\
			$n=200,\, \epsilon=0.05$ & $n=200,\, \epsilon=0.1$  & $n=200,\, \epsilon=0.15$ \\
			\includegraphics[scale=0.3]{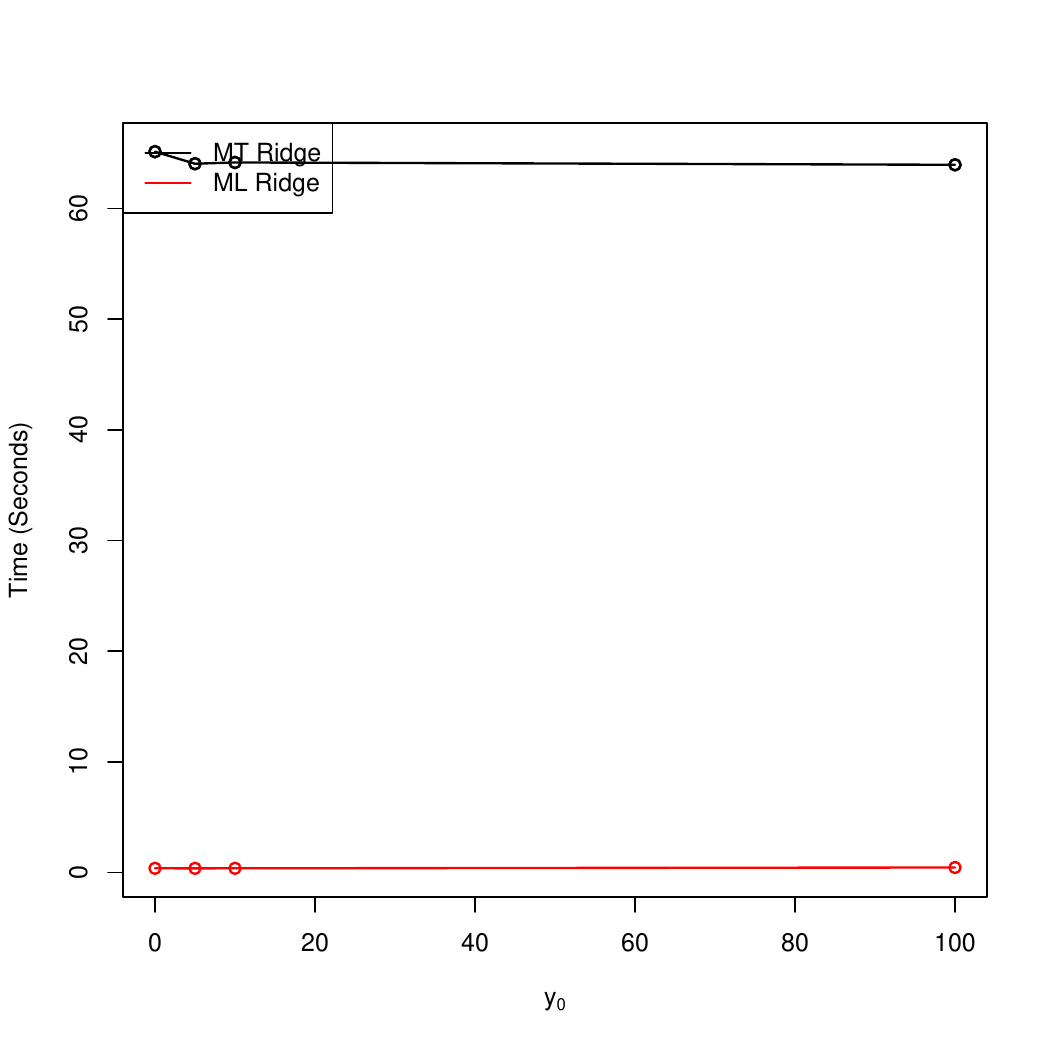}
			&   \includegraphics[scale=0.3]{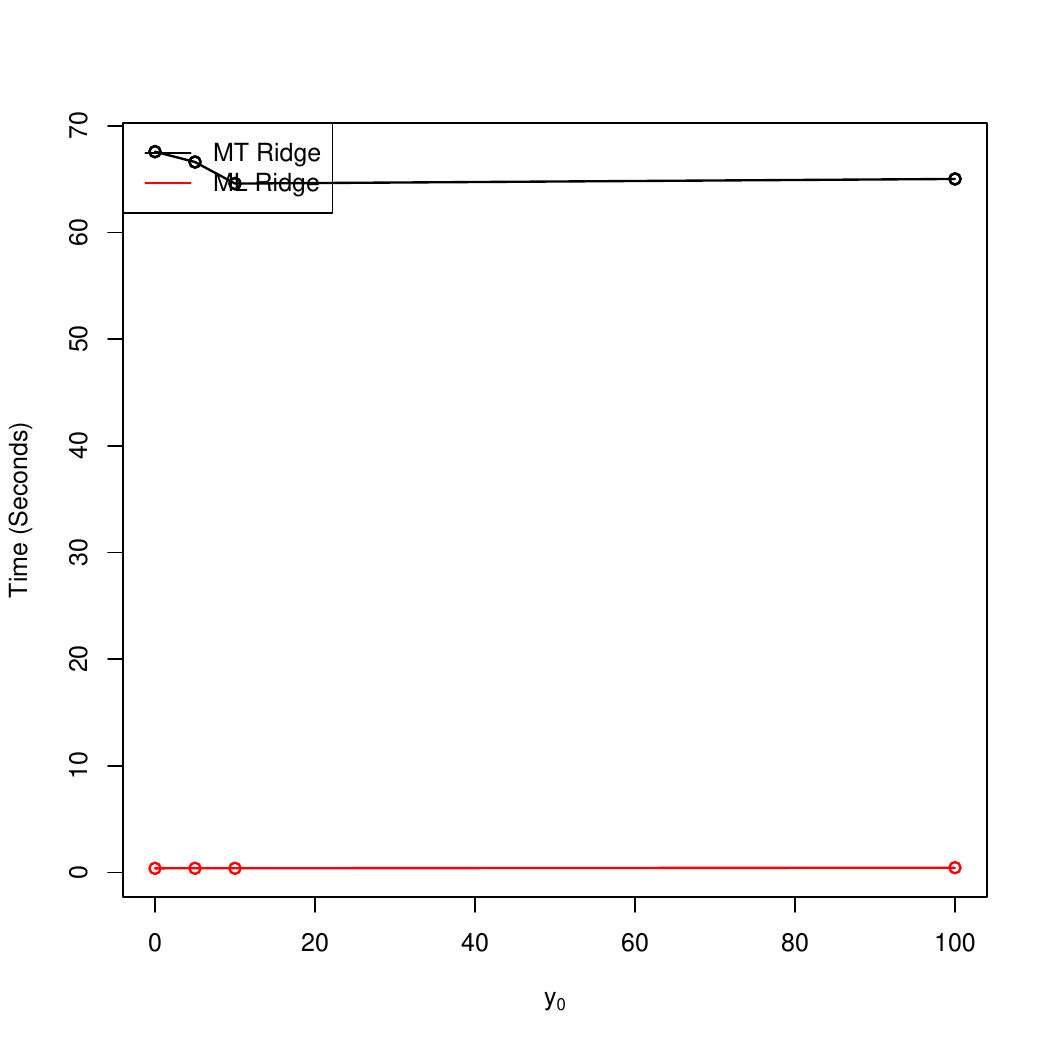} 
			&   \includegraphics[scale=0.3]{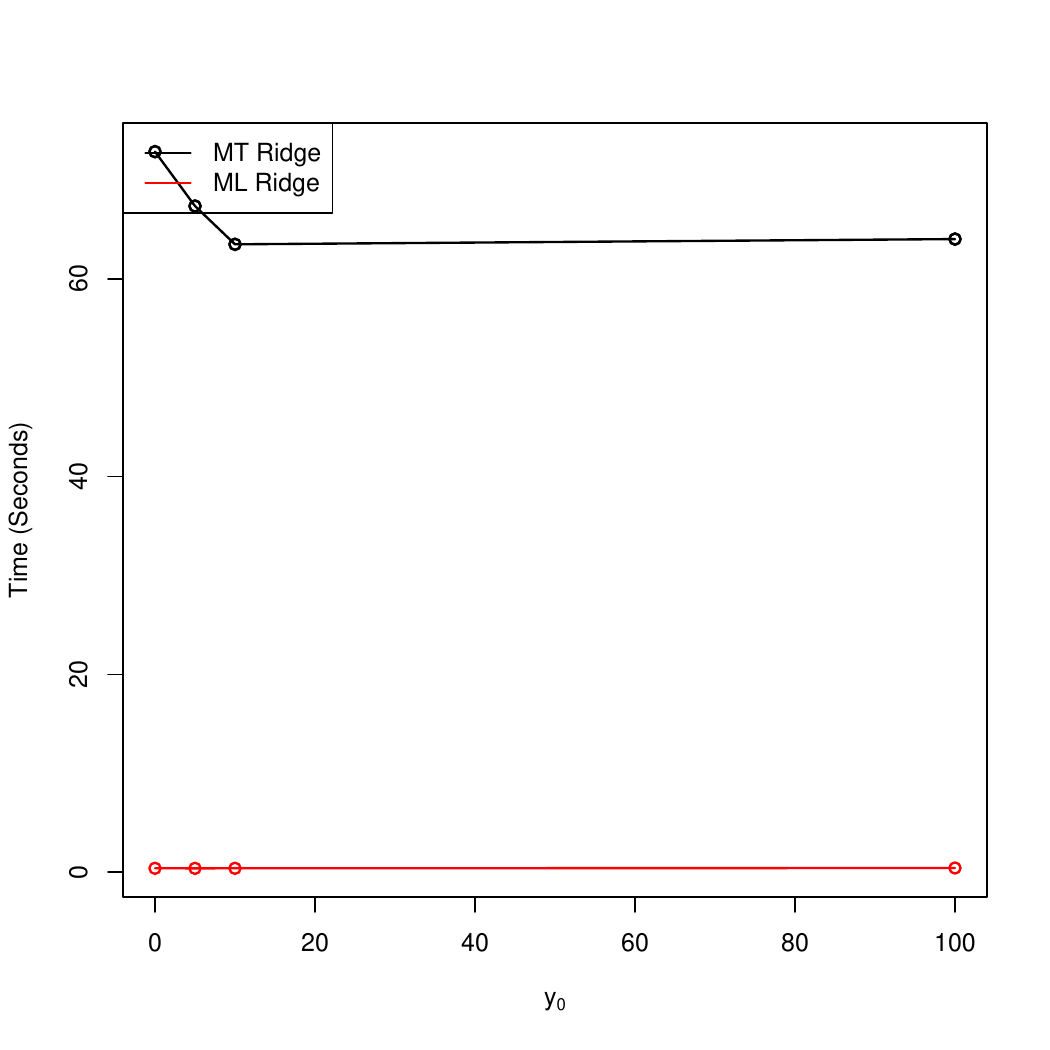}
		\end{tabular}
		\caption{\small \label{fig:timesimuRQLp100Ridge}  Computation times for model AMR and Ridge methods, $p=100$.}
	\end{center} 
\end{figure}

\begin{figure}[tp]
	\begin{center}
		\footnotesize
		\renewcommand{\arraystretch}{0.22}
		\newcolumntype{M}{>{\centering\arraybackslash}m{\dimexpr.01\linewidth-1\tabcolsep}}
		\newcolumntype{G}{>{\centering\arraybackslash}m{\dimexpr.38\linewidth-1\tabcolsep}}
		\begin{tabular}{GGG}
			$n=50,\, \epsilon=0.05$ & $n=50,\, \epsilon=0.1$  & $n=50,\, \epsilon=0.15$ \\
			\includegraphics[scale=0.29]{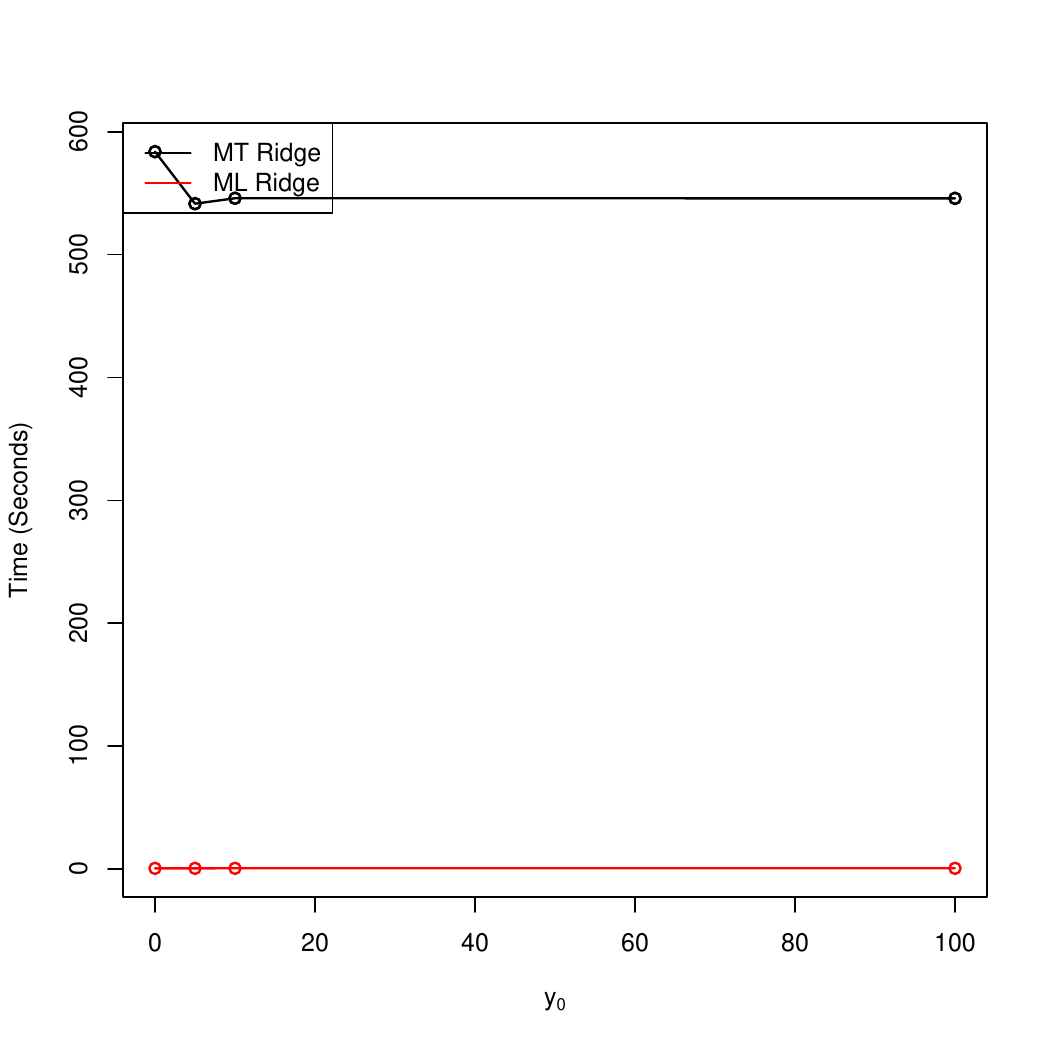} 
			&  \includegraphics[scale=0.29]{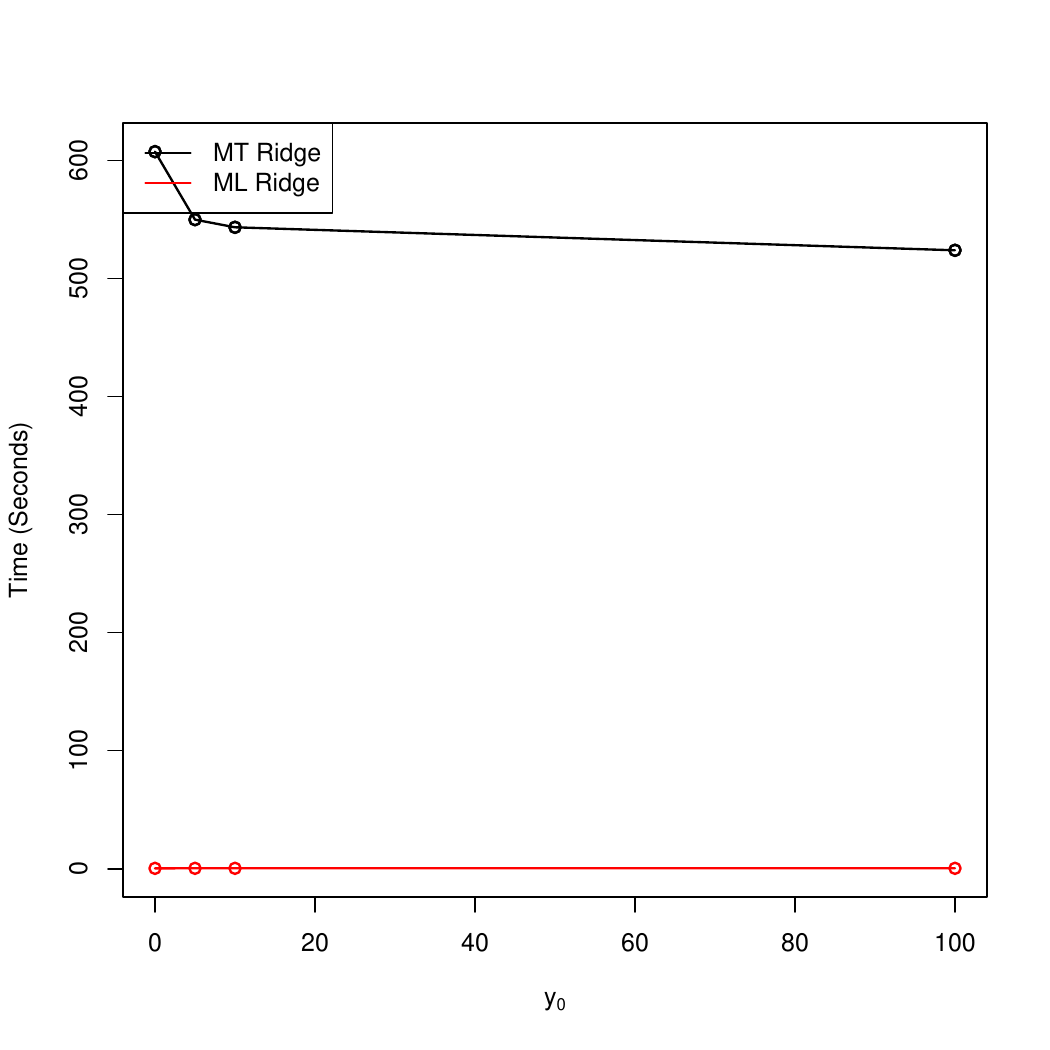}
			&   \includegraphics[scale=0.29]{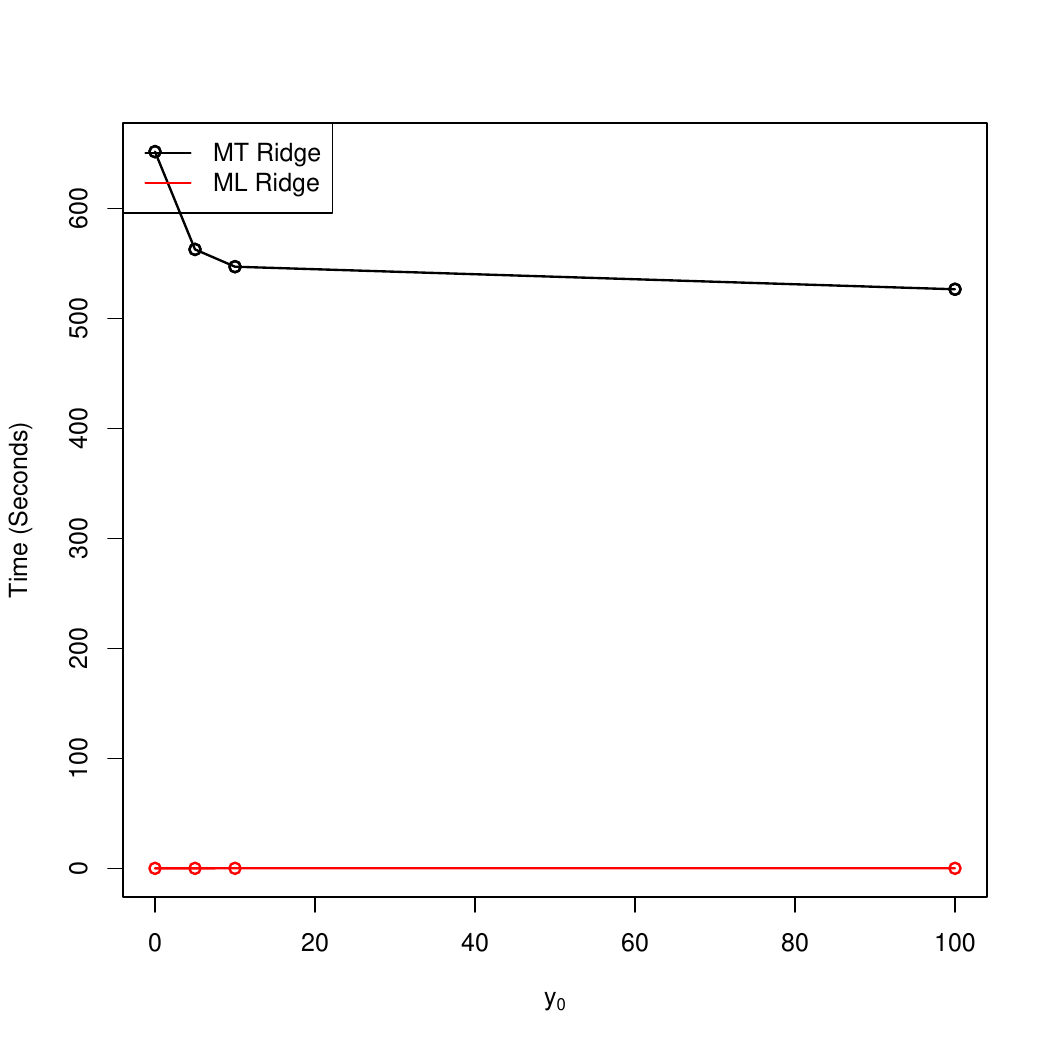}
			\\
			$n=100,\, \epsilon=0.05$ & $n=100,\, \epsilon=0.1$  & $n=100,\, \epsilon=0.15$ \\
			\includegraphics[scale=0.3]{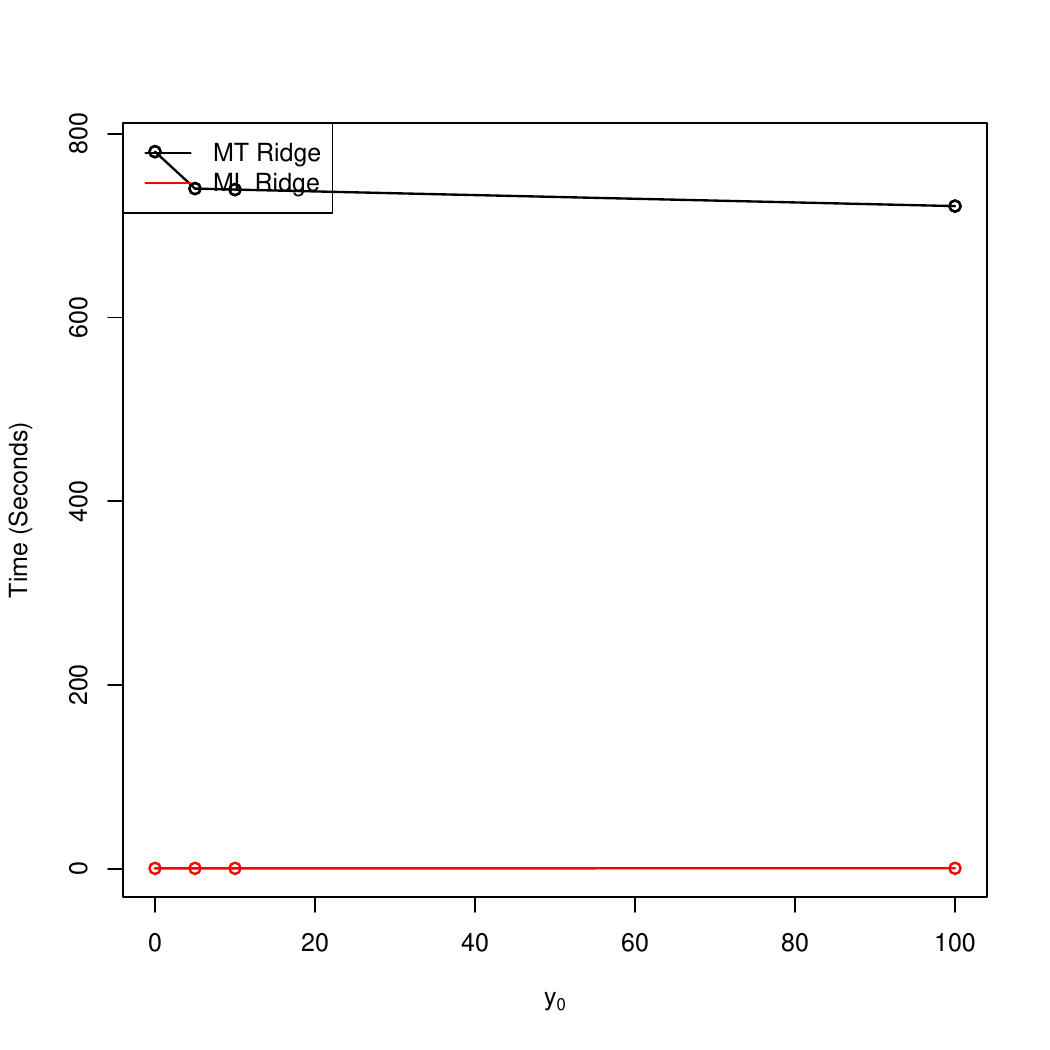}
			&   \includegraphics[scale=0.3]{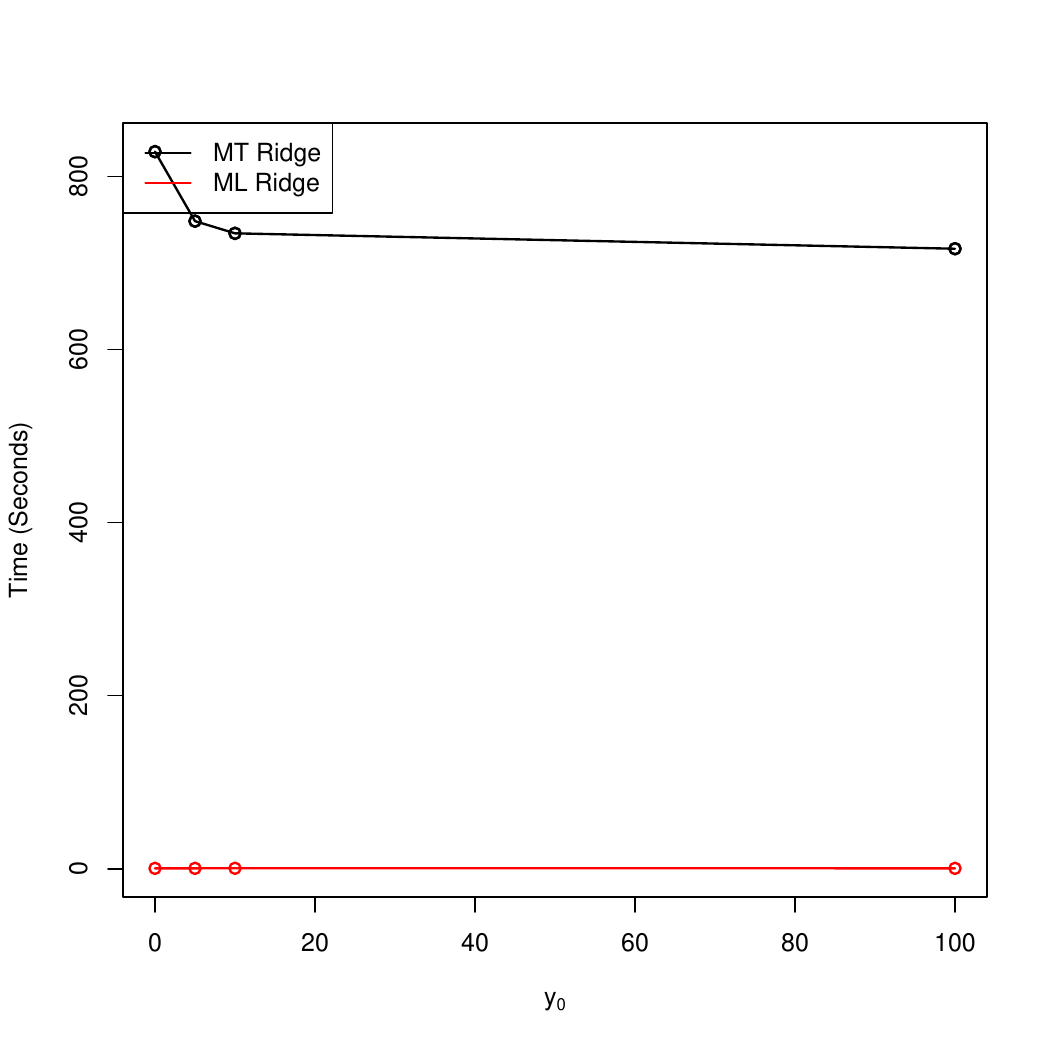} 
			&   \includegraphics[scale=0.3]{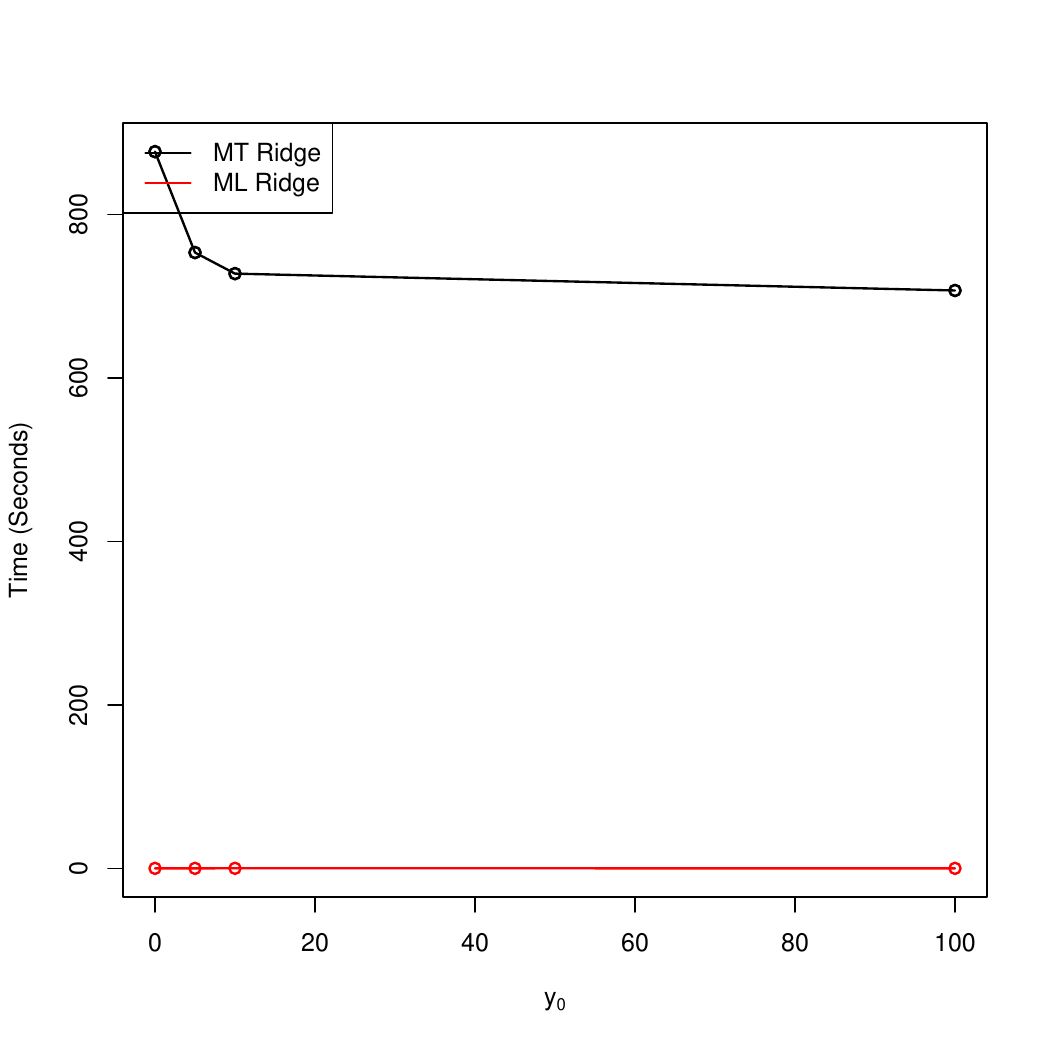} 
			\\
			$n=200,\, \epsilon=0.05$ & $n=200,\, \epsilon=0.1$  & $n=200,\, \epsilon=0.15$ \\
			\includegraphics[scale=0.3]{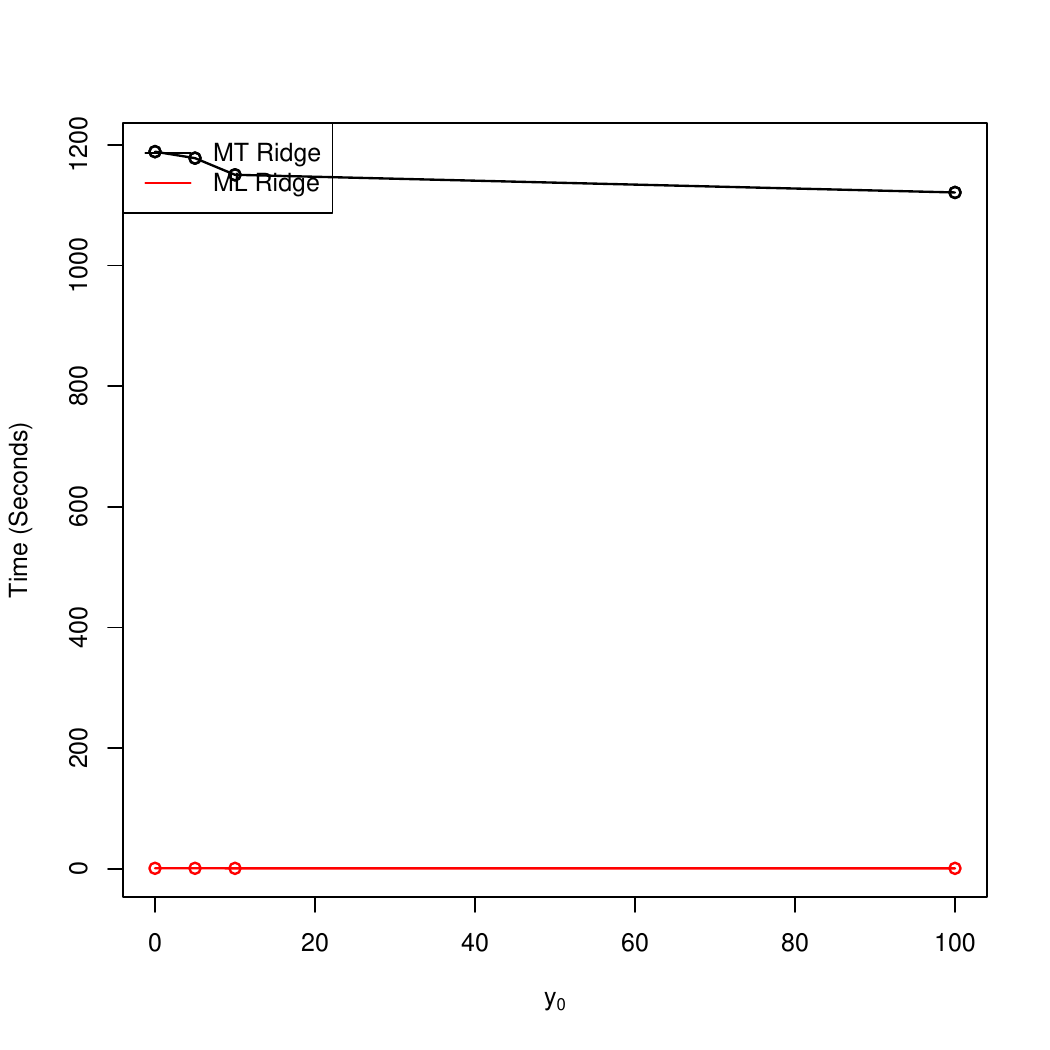}
			&   \includegraphics[scale=0.3]{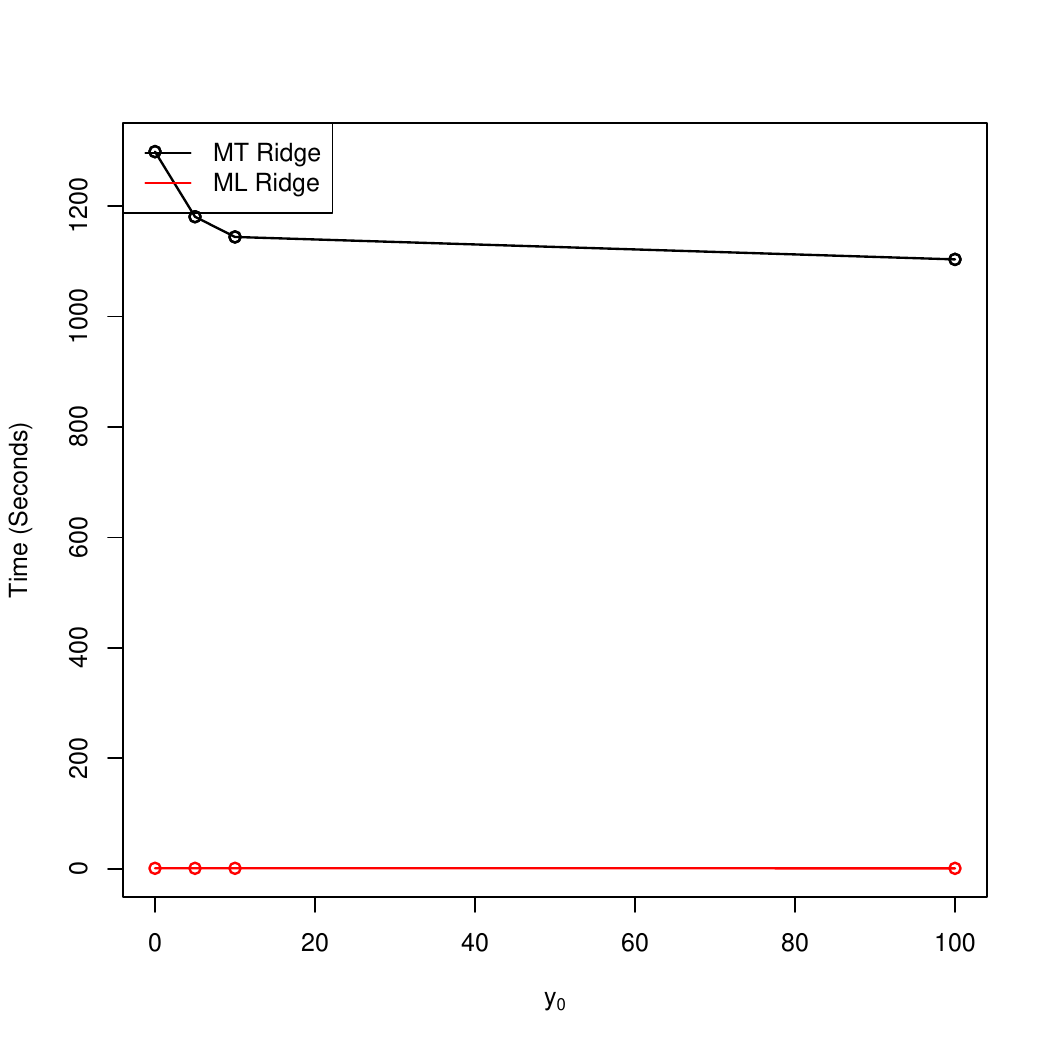} 
			&   \includegraphics[scale=0.3]{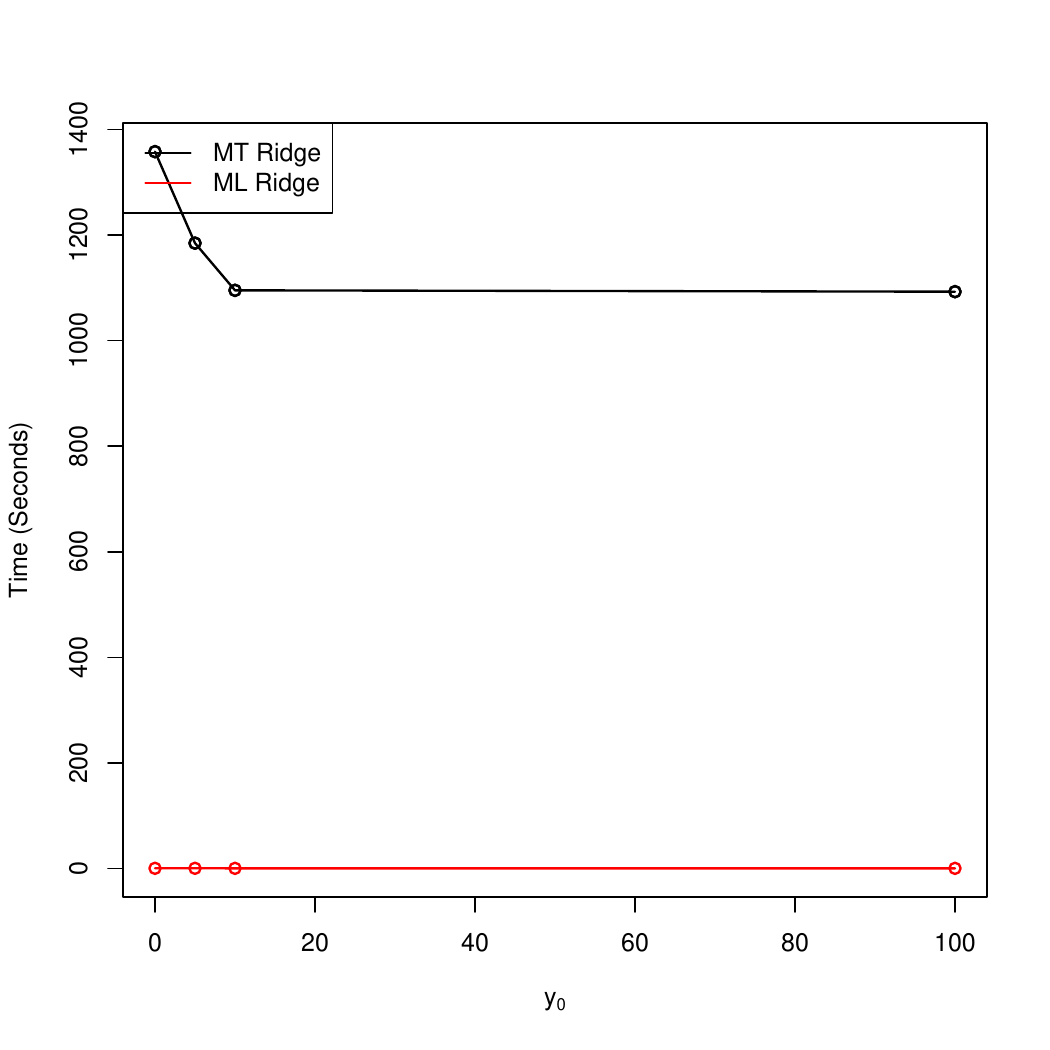}
		\end{tabular}
		\caption{\small \label{fig:timesimuRQLp400Ridge}  Computation times for  model AMR and Ridge methods, $p=400$.}
	\end{center} 
\end{figure}

\begin{figure}[tp]
	\begin{center}
		\footnotesize
		\renewcommand{\arraystretch}{0.22}
		\newcolumntype{M}{>{\centering\arraybackslash}m{\dimexpr.01\linewidth-1\tabcolsep}}
		\newcolumntype{G}{>{\centering\arraybackslash}m{\dimexpr.38\linewidth-1\tabcolsep}}
		\begin{tabular}{GGG}
			$n=50,\, \epsilon=0.05$ & $n=50,\, \epsilon=0.1$  & $n=50,\, \epsilon=0.15$ \\
			\includegraphics[scale=0.29]{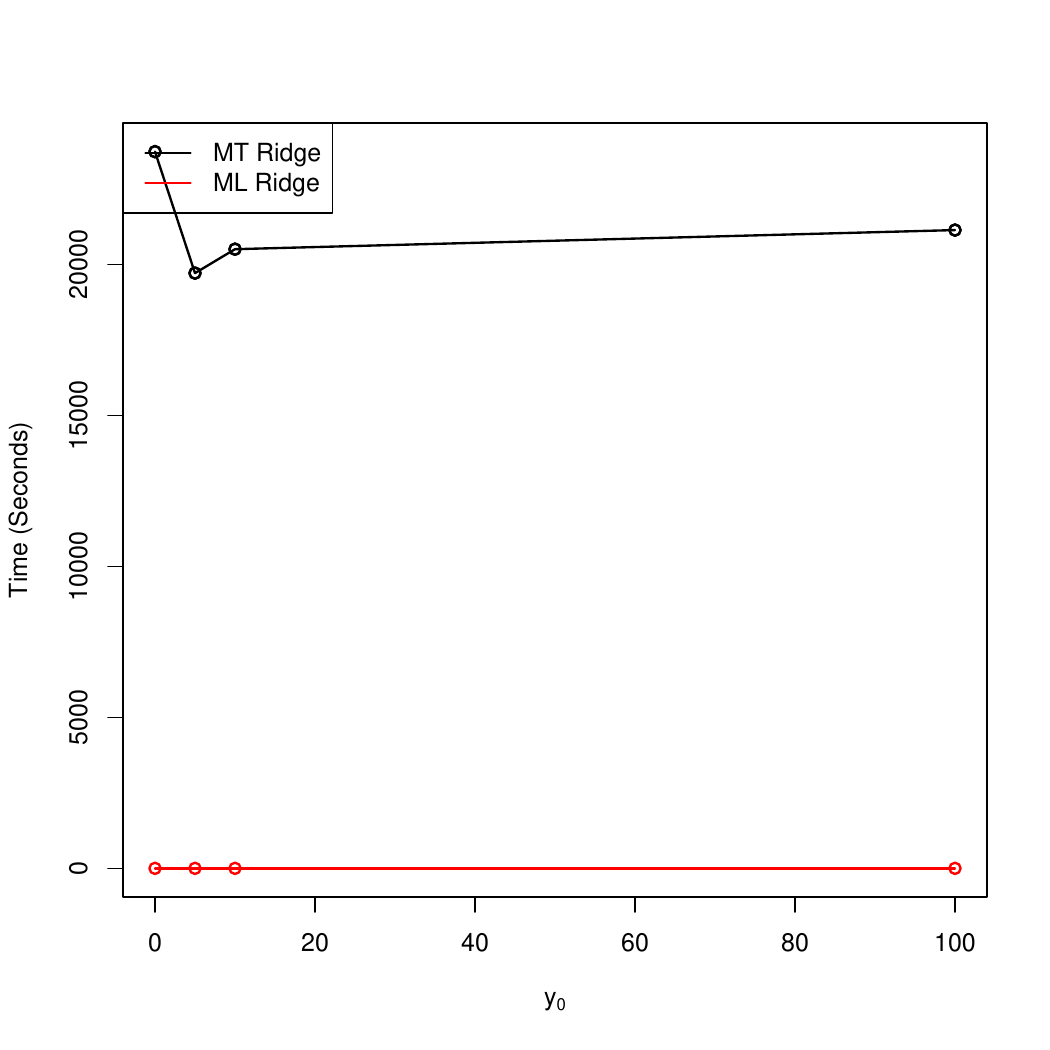} 
			&  \includegraphics[scale=0.29]{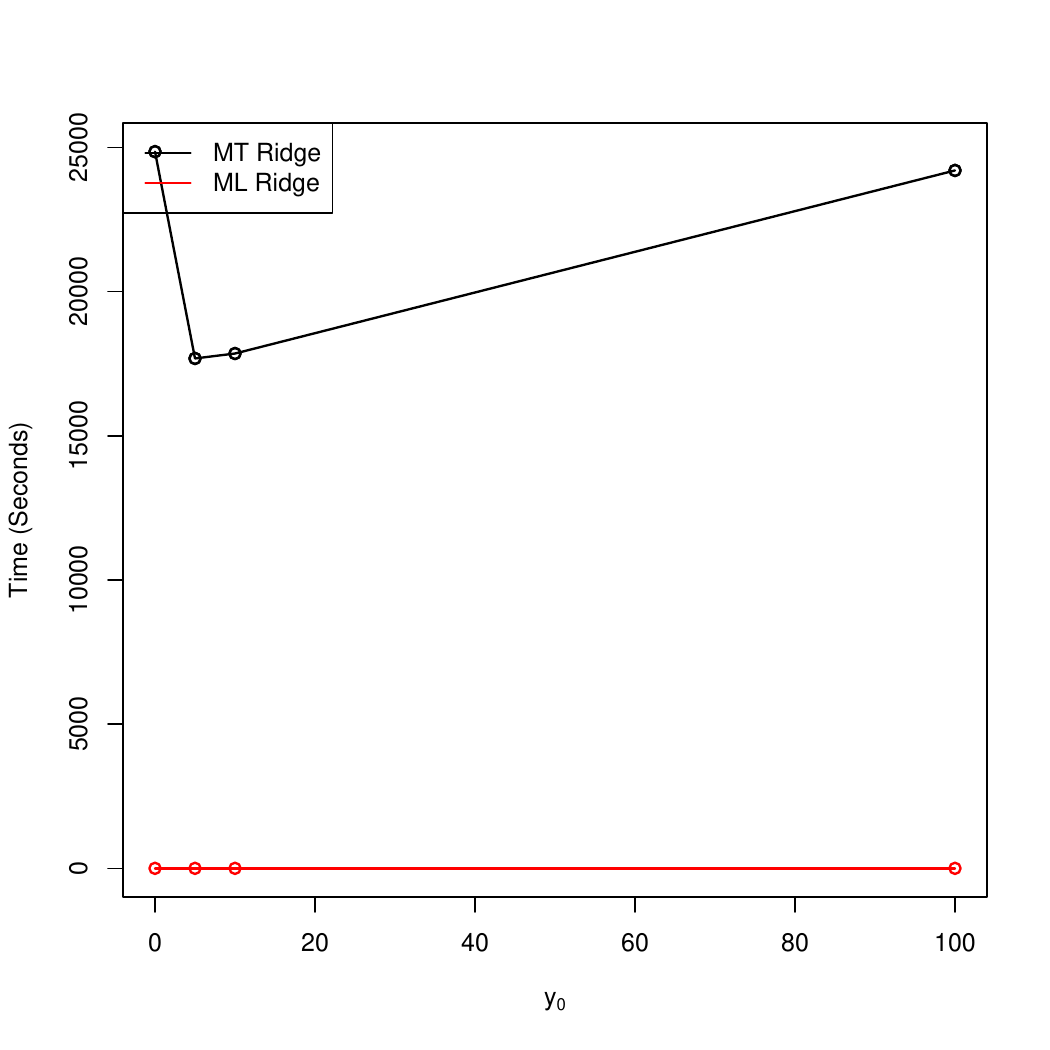}
			&   \includegraphics[scale=0.29]{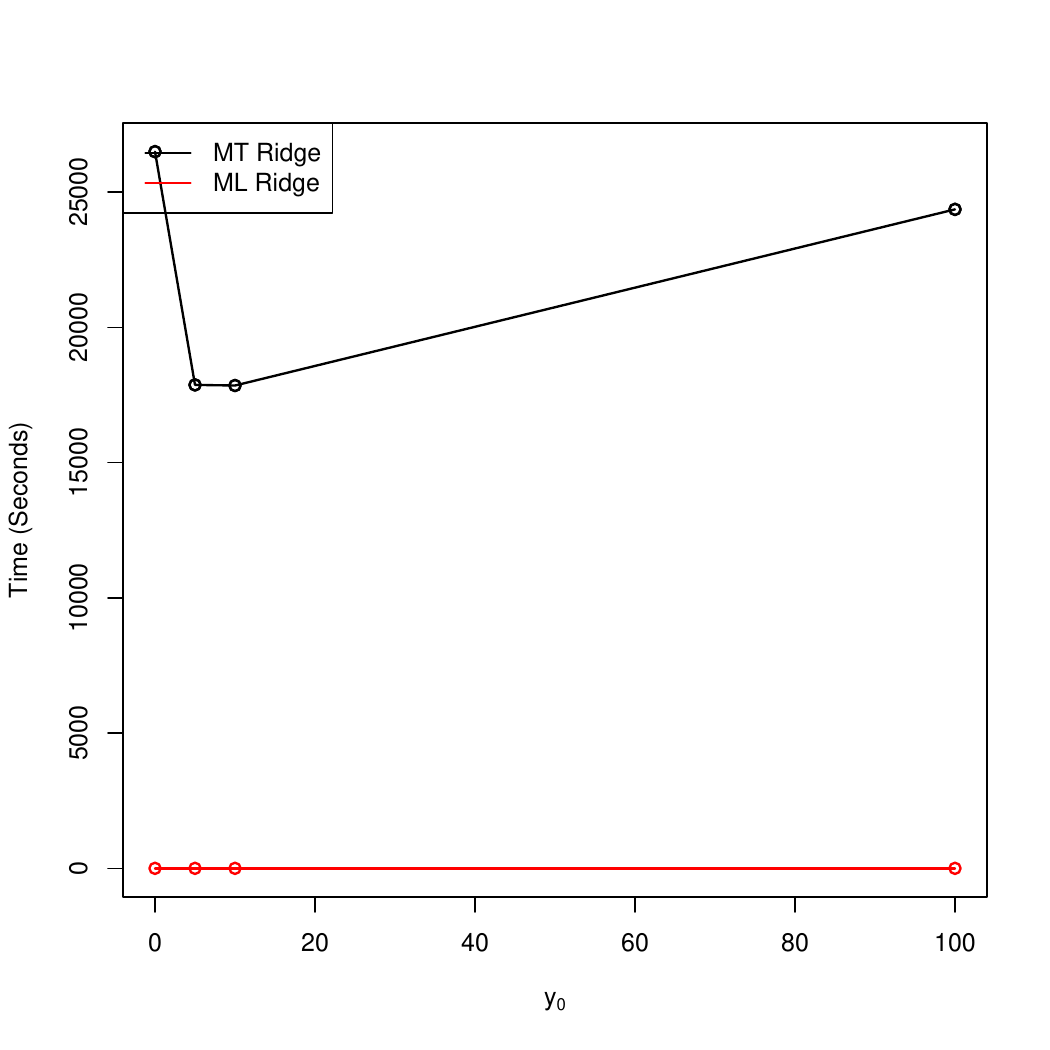}
			\\
			$n=100,\, \epsilon=0.05$ & $n=100,\, \epsilon=0.1$  & $n=100,\, \epsilon=0.15$ \\
			\includegraphics[scale=0.3]{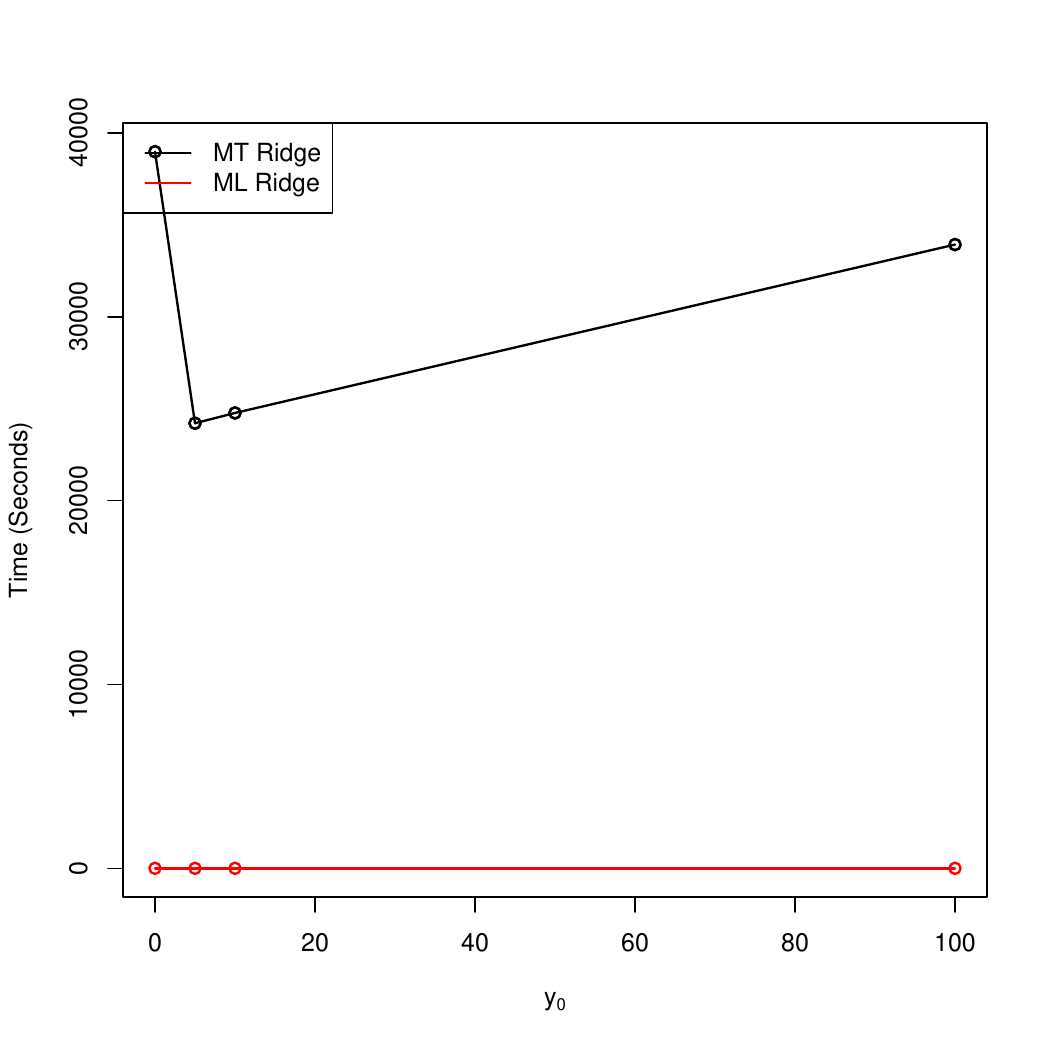}
			&   \includegraphics[scale=0.3]{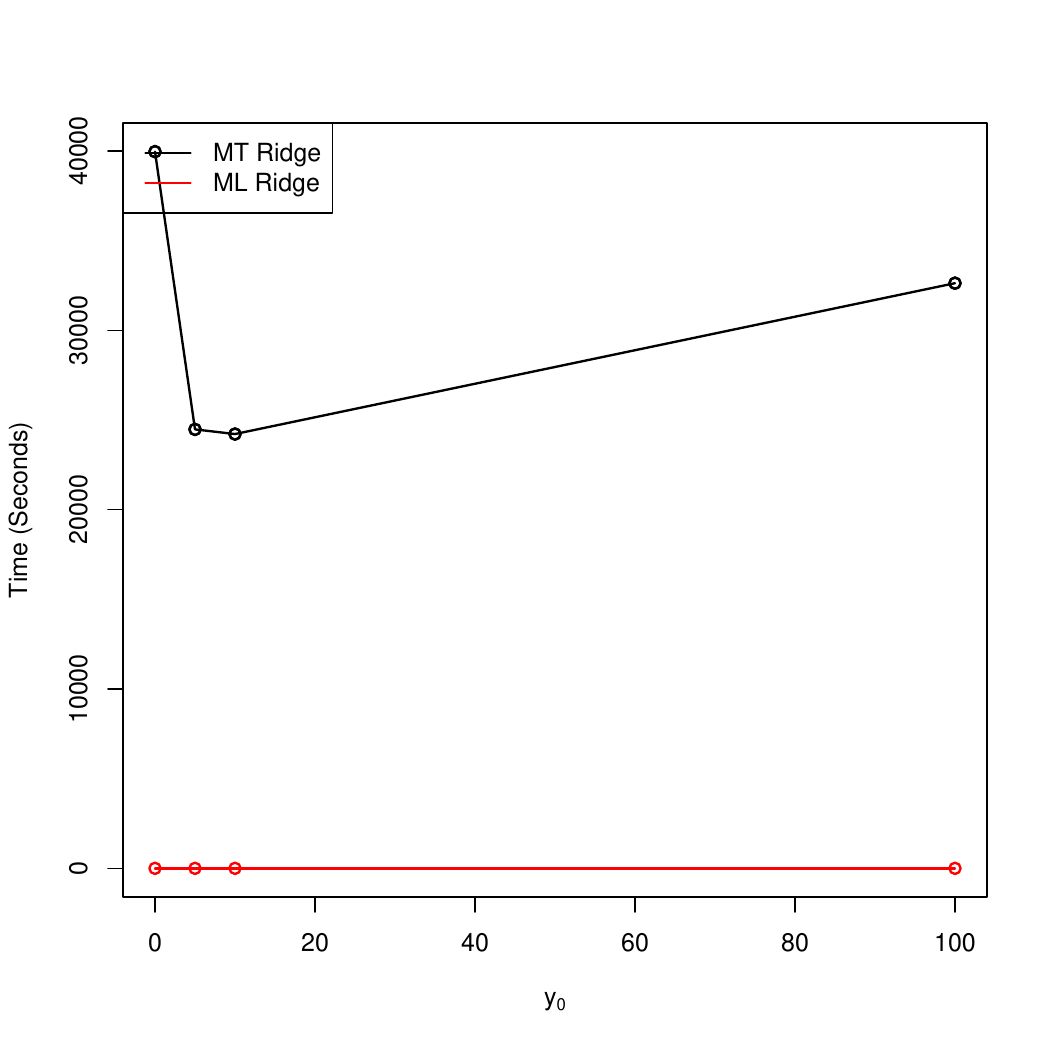} 
			&   \includegraphics[scale=0.3]{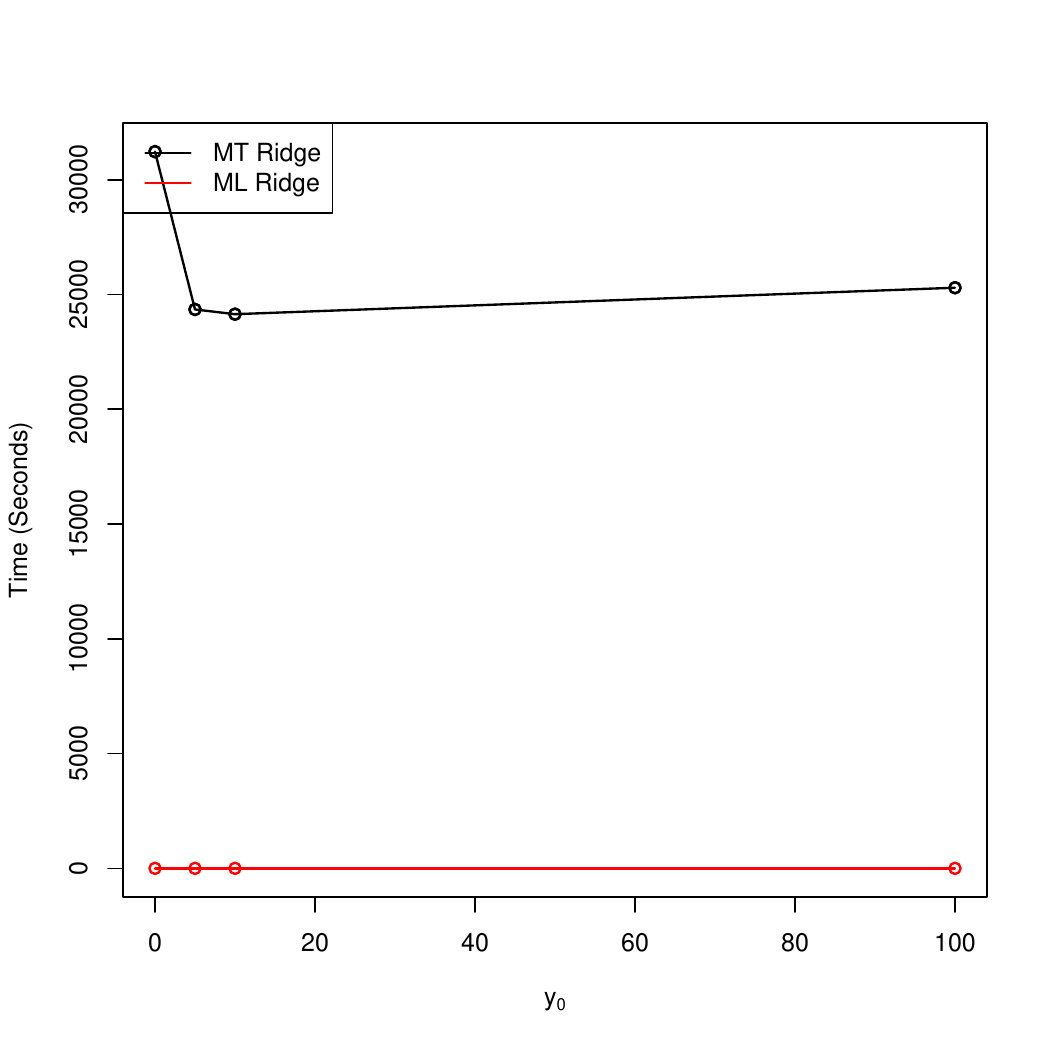} 
			\\
			$n=200,\, \epsilon=0.05$ & $n=200,\, \epsilon=0.1$  & $n=200,\, \epsilon=0.15$ \\
			\includegraphics[scale=0.3]{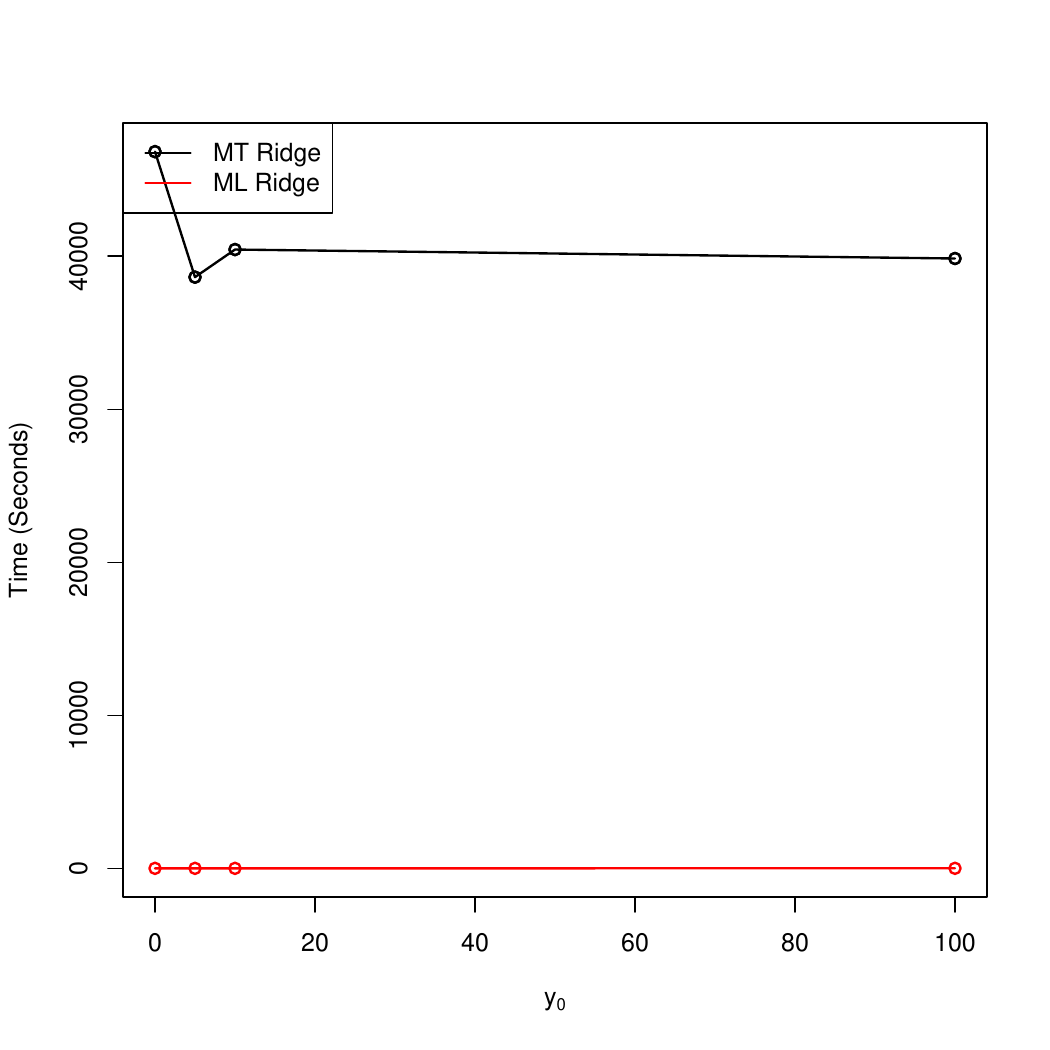}
			&   \includegraphics[scale=0.3]{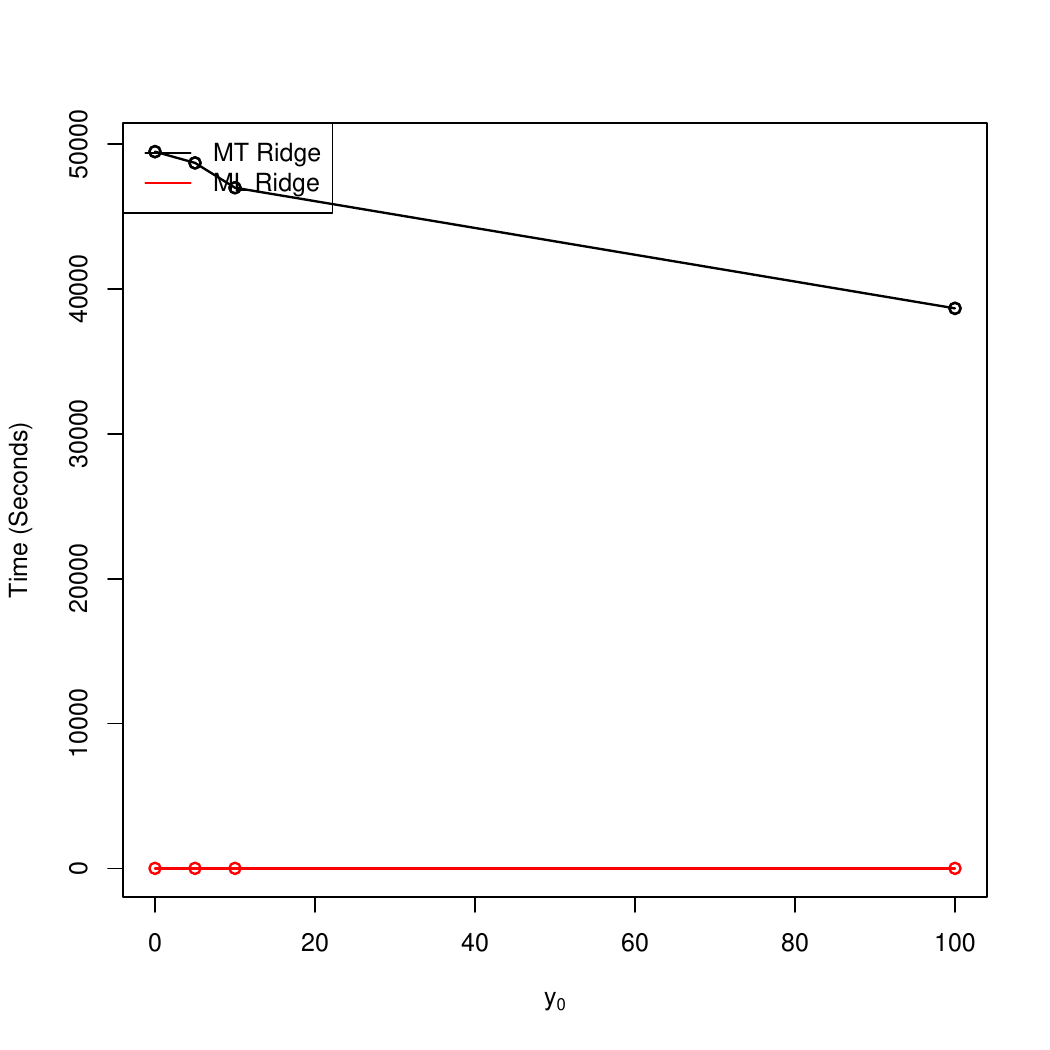} 
			&   \includegraphics[scale=0.3]{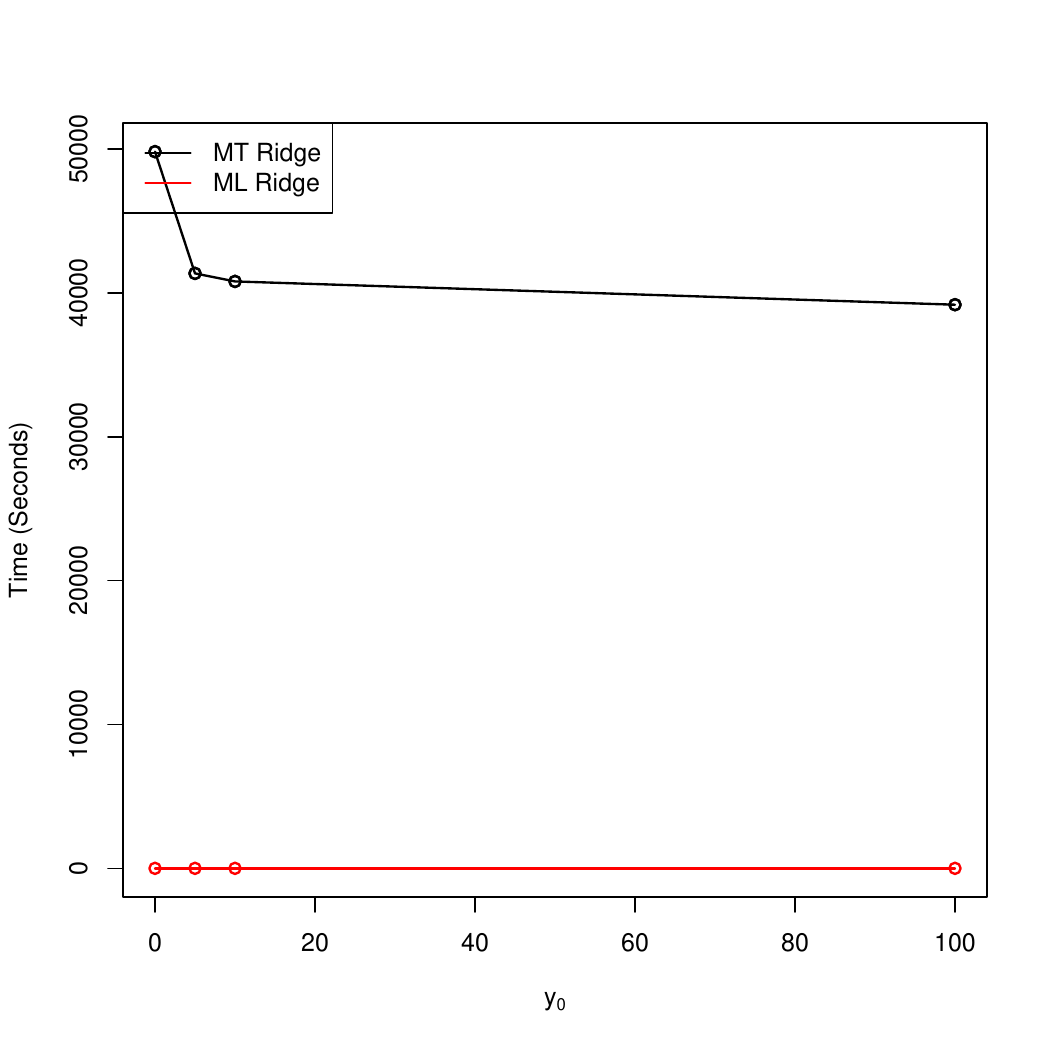}
		\end{tabular}
		\caption{\small \label{fig:timesimuRQLp1600Ridge}  Computation times for  model AMR and Ridge methods, $p=1600$.}
	\end{center} 
\end{figure}

\end{document}